%% file: paper.tex
\newif\ifextended
\extendedtrue

\ifextended
\documentclass[acmsmall,screen,nonacm,10pt]{acmart}
\else
\documentclass[acmsmall,screen,10pt]{acmart}
\fi

\usepackage{ottalt}
\usepackage[utf8]{inputenc}
\usepackage{listings}
\usepackage{leftindex}

\usepackage{tikz}
\usepackage{tikz-cd}
\usetikzlibrary{matrix,arrows,decorations.pathmorphing}

\usepackage{dsfont}
\usepackage{proof}
\usepackage{xcolor}

\usepackage{caption}
\usepackage{subcaption}

\usepackage{scalerel}[2014/03/10]
\usepackage{stackengine}

\newcommand{\I}{\text{id}}
\newcommand{\Id}{\textbf{Id}}
\newcommand{\Ct}{\mathds{C}}
\newcommand{\D}{\mathds{D}}
\newcommand{\Ca}{\mathbf{C}}
\newcommand{\T}{\mathbf{T}}
\newcommand{\E}{\mathbf{End}_{\Ct}}
\newcommand{\ES}{\mathbf{End}_{\Ct}^{\text{s}}}

\newcommand{\EC}{\mathbf{End}_{\Ct}^{\text{cc}}}
\newcommand{\ED}{$\text{DCC}_e$\hspace*{1pt}}

\newcommand{\bc}{\bigcirc}
\newcommand{\Ge}{$\text{GMCC}_{\text{e}}$}

\newtheorem{theorem}{Theorem}[section]
\newtheorem{prop}[theorem]{Proposition}
\newtheorem{lemma}[theorem]{Lemma}

\usepackage{amsmath,amsthm}
\usepackage{stmaryrd}

\ifextended
\hypersetup{
           breaklinks=true,   
           colorlinks=true,   
           pdfusetitle=true}   
\fi

\title{Monadic and Comonadic Aspects of Dependency Analysis}

\author{Pritam Choudhury}
 \affiliation{
   \position{}
   \department{Department of Computer and Information Science}              
   \institution{University of Pennsylvania}                   
   \city{Philadelphia}
   \country{USA}
 }
 \email{pritam@seas.upenn.edu}
 
\acmJournal{PACMPL}

  \input{Rules}

  \renewottcommands[ott]

\allowdisplaybreaks

\newcommand\reallywidetilde[1]{\ThisStyle{%
  \setbox0=\hbox{$\SavedStyle#1$}%
  \stackengine{-.1\LMpt}{$\SavedStyle#1$}{%
    \stretchto{\scaleto{\SavedStyle\mkern.2mu\sim}{.5467\wd0}}{.6\ht0}%
  }{O}{c}{F}{T}{S}%
}}

\def\wdtilde#1{%
  \reallywidetilde{#1}}
  
\ifextended
\else
\setcopyright{rightsretained}
\acmPrice{}
\acmDOI{10.1145/3563335}
\acmYear{2022}
\copyrightyear{2022}
\acmSubmissionID{oopslab22main-p437-p}
\acmJournal{PACMPL}
\acmVolume{6}
\acmNumber{OOPSLA2}
\acmArticle{172}
\acmMonth{10}
\fi



\ifextended
\renewcommand\footnotetextcopyrightpermission[1]{}
\fi

\begin{document}

\bibliographystyle{ACM-Reference-Format}
\citestyle{acmauthoryear} 

\begin{abstract}
\input{Sections/Abstract}
\end{abstract}

\maketitle
\ifextended
\pagestyle{plain}
\fi

\input{Sections/Intro}

\input{Sections/Monads}
\input{Sections/Comonads}
\input{Sections/DependencyGraded}

\input{Sections/Discussion}

\input{Sections/Conclusion}

\begin{acks}
I was supported by the National Science Foundation under Grant Nos. 2006535 and 1703835.\\ I would like to thank Dominic Orchard, Benjamin Bergman, Stephanie Weirich and the anonymous referees for their feedback and suggestions. I dedicate this paper to Master Felice Macera of the Penn TaeKwonDo Club, who has been a mentor par excellence. 
\end{acks}

\newpage
\bibliography{Biblio.bib}

\ifextended
\newpage

\appendix

\input{Sections/Appendix}

\fi

\end{document}

%% file: Rules.tex
\newcommand{\ottdrule}[4][]{{\displaystyle\frac{\begin{array}{l}#2\end{array}}{#3}\quad\ottdrulename{#4}}}

\newcommand{\ottpremise}[1]{ #1 \\}
\newenvironment{ottdefnblock}[3][]{ \framebox{\mbox{#2}} \quad #3 \\[0pt]}{}

\newcommand{\ottnt}[1]{\mathit{#1}}
\newcommand{\ottmv}[1]{\mathit{#1}}
\newcommand{\ottkw}[1]{\mathbf{#1}}
\newcommand{\ottsym}[1]{#1}

\newcommand{\ottdrulename}[1]{\textsc{#1}}

\newcommand{\ottafterlastrule}{\\}

\newcommand{\ottdruleProtXXAlready}[1]{\ottdrule[#1]{%
\ottpremise{ \ell  \sqsubseteq  \ottnt{A} }%
}{
 \ell  \sqsubseteq   T_{  \ell'  } \:  \ottnt{A}  }{%
{\ottdrulename{Prot\_Already}}{}%
}}



%% file: Sections/Abstract.tex
Dependency analysis is vital to several applications in computer science. It lies at the essence of secure information flow analysis, binding-time analysis, etc. Various calculi have been proposed in the literature for analysing individual dependencies. Abadi et. al., by extending Moggi's monadic metalanguage, unified several of these calculi into the Dependency Core Calculus (DCC). DCC has served as a foundational framework for dependency analysis for the last two decades. However, in spite of its success, DCC has its limitations. First, the monadic bind rule of the calculus is nonstandard and relies upon an auxiliary protection judgement. Second, being of a monadic nature, the calculus cannot capture dependency analyses that possess a comonadic nature, for example, the  binding-time calculus, $\lambda^{\circ}$, of Davies. In this paper, we address these limitations by designing an alternative dependency calculus that is inspired by standard ideas from category theory. Our calculus is both monadic and comonadic in nature and subsumes both DCC and $\lambda^{\circ}$. Our construction explains the nonstandard bind rule and the protection judgement of DCC in terms of standard categorical concepts. It also leads to a novel technique for proving correctness of dependency analysis. We use this technique to present alternative proofs of correctness for DCC and $\lambda^{\circ}$.  

%% file: Sections/Intro.tex
\newcommand{\lb}{$\lambda^{\boxempty}$}
\newcommand{\lc}{$\lambda^{\circ}$}

\section{Introduction}

Dependency analysis is the analysis of dependence of an entity upon another. The entities are primarily programs or parts thereof, but they can also be abstract, like security clearance levels in an organization, stages in a compilation process, etc. 

Broadly speaking, an entity depends upon another one if the latter influences the behaviour of the former. On the other hand, an entity is independent of another one if the latter does not interfere in the behaviour of the former. For example, consider the following $\lambda$-terms: $  (   \lambda  \ottmv{x}  .  \ottmv{x}   )   \:   (   \ottsym{2}  +  \ottsym{2}   )  $ and $  (   \lambda  \ottmv{x}  .  \ottsym{4}   )   \:   (   \ottsym{2}  +  \ottsym{2}   )  $. The argument $ (   \ottsym{2}  +  \ottsym{2}   ) $ dictates the normal form of the first term whereas it plays no role in deciding the normal form of the second term. So, we say that the first term depends upon the argument whereas the second one does not. What this means is that in the second term, we can replace the argument $ (   \ottsym{2}  +  \ottsym{2}   ) $ with any other terminating computation, while maintaining the same normal form for the term as a whole. 

The power of dependency analysis comes from this very simple principle: if an entity does not depend upon another one, then variations in the latter should not affect the former. This is the well-known principle of noninterference \citep{goguen}. This principle has far-reaching implications and lies at the heart of several applications in computer science, like secure information flow analysis, binding-time analysis, etc.

In secure information flow analysis \citep{denning1,denning2,smith}, one wishes to guarantee that there is no flow of information from secret data to public variables. Viewed abstractly in terms of security levels, this is equivalent to saying that level `public'  \textit{does not depend} upon level `secret'. In binding-time analysis \citep{hatcliff,gomard,gluck,lambdacirc}, one wishes to guarantee that a given program can be correctly compiled in multiple stages even when each stage can potentially optimize based on inputs received from earlier stages. To ensure correctness of such compilation, it is necessary that an earlier stage \textit{does not depend} upon a later one. Secure information flow analysis and binding-time analysis are examples of dependency analysis. There are many other examples of dependency analysis \citep[etc.]{trust,memory,tip,tang} appearing in the literature.

Over two decades ago, \citet{dcc} showed that several dependency analyses  \citep{slam,tip,hatcliff,volpano,tang} can be seen as instances of a general Dependency Core Calculus (DCC). Their work has served as a foundational framework for dependency analysis in the field of programming languages and has led to extensive research \citep[etc.]{tse-zdancewic,igarashi,ahmed,algehed,persdcc,ddc} on this topic which continues to this day. 

DCC is a simple extension of Moggi's monadic metalanguage \citep{moggi}. The monadic metalanguage is a general calculus for analysing computational effects like nontermination, exceptions, input/output, etc. Moggi showed that computational effects in programming languages can be understood in terms of monads from category theory \citep{maclane}. At first sight, computational effects seem to be quite different from dependencies. So, it comes as a surprise (pointed out by \citet{dcc} themselves) that with just a simple extension, a calculus for analysing computational effects can also analyse dependencies. 

In this regard, \citet{dcc} point out a common feature underlying monads and security levels: just as there is no way of projecting out of a `monad world', there is also no way of projecting out of a `secret world'. Concretely, just as there is no general function of type $  T  \ottnt{A}   \to  \ottnt{A} $ for a monad $T$ and a type $\ottnt{A}$, there is also no non-trivial function from secret data to public variables. So, dependency analysis has a monadic aspect to it.

However, the monadic aspect of dependency analysis might be just half of the story. Everyday experience shows us that security constraints can be enforced not only by restricting outflow but also by restricting inflow. For example, in a world with two levels, $ \mathbf{Public} $ and $ \mathbf{Secret} $, security can be enforced not only by restricting projection out of $ \mathbf{Secret} $ level, but also by restricting injection into $ \mathbf{Public} $ level. These ways are dual to one another. The way of restricting projection, employed in DCC, goes via monads. On the other hand, the way of restricting injection goes via comonads. Similar to a monad $T$ that restricts by disallowing a general function of type $  T  \ottnt{A}   \to  \ottnt{A} $, a comonad $D$ restricts by disallowing a general function of type $ \ottnt{A}  \to    D  \ottnt{A}   $. While the monadic aspect of dependency analysis has received considerable attention in literature, the comonadic aspect has received less attention. 

In this paper, we show that just like the monadic aspect, the comonadic aspect of dependency analysis also has much to offer. Further, we find that the monadic and the comonadic aspects play nicely with one another. We design a language that integrates these two aspects into a single system. This integration helps us unify DCC, a monadic dependency calculus and \lc \citep{lambdacirc}, a comonadic dependency calculus (that is known to be outside the reach of DCC \citep{dcc}). It also leads to a novel general technique for proving correctness of dependency analysis. And above all, it shines light on some of the nuances of dependency analysis. 



In short, we make the following contributions:
\begin{itemize}
\item We present a Graded Monadic Comonadic Calculus, GMCC, and its extension, \Ge{}, and provide meaning-preserving translations from both DCC and \lc{} to \Ge{}.
\item We show that the protection judgement of DCC, when appropriately modified, enables comonadic reasoning in the language. Further, we show that under certain restrictions, DCC, with this modification, is equivalent to GMCC. This equivalence helps explain the nonstandard bind-rule of DCC in terms of standard categorical concepts.
\item GMCC and \Ge{} are general calculi that are sound with respect to a class of categorical models. These  categorical models motivate a novel technique for proving correctness of dependency analyses. We use this technique to provide simple proofs of correctness for both DCC and \lc{}.
\end{itemize}


\ifextended
Note that this paper is an extended version of \citet{gmcc}. The reader can find the proofs of the lemmas and theorems stated in \citet{gmcc} in the appendices of this paper.
\fi

In the next section, we review the basics of dependency analysis and its application in information flow control and staged execution of programs. This section is meant to provide some background to readers who are not very familiar with dependency analysis.

\section{Dependency Analysis in Action} \label{secDepAct}

Consider a database \lstinline|Db| containing demographic information of a city. For the sake of simplicity, let's say the database is represented as a list of tuples with elements of the list corresponding to residents of the city and elements of the tuples corresponding to their demographic information. Further, let's assume that each tuple has only 4 elements for recording name, age, ethnicity and monthly income of a resident (in that order). According to the policies of the city council, name, age and ethnicities of the residents are non-sensitive information that may be shared without any constraint, whereas monthly income of the residents is sensitive information that may be shared only with people who are allowed to handle such information.

Now, consider the following queries:
\begin{enumerate}
\item What fraction of the elderly city residents (age $\geq$ 65 years) are ethnically Caucasian?
\item What is the average monthly income of ethnically Asian residents of the city?
\end{enumerate} 

With the above queries in mind, a programmer seeks the outputs of the following programs, written in \verb|Haskell|-like syntax:
\begin{enumerate}
\item \begin{lstlisting}
      lstEld = filter (\ x -> second x >= 65) Db
      lstEldC = filter (\ x -> third x == "Caucasian") lstEld
      fracEldC = (length lstEldC) / (length lstEld)
      print fracEldC         
      \end{lstlisting}
      
\item \begin{lstlisting}
       lstAsn = filter (\ x -> third x == "Asian") Db
       <@\textcolor{red}{totIncA}@> = foldr (\ x y -> <@\textcolor{red}{fourth x}@> + y) 0 lstAsn
       <@\textcolor{red}{avgIncA}@> = <@\textcolor{red}{totIncA}@> / (length lstAsn)
       print <@\textcolor{red}{avgIncA}@>
       \end{lstlisting}
\end{enumerate} 

Note that functions \lstinline|second|, \lstinline|third| and \lstinline|fourth| access the second, third and fourth elements of a tuple respectively. The question we need to address now is whether the programmer may be allowed to see the outputs of the above programs. Let's suppose this programmer does not have the permission to handle sensitive information. Then, the output of the second program should not be shared with this programmer because  it reveals (at least partially) monthly income data. For example, if there's just a single ethnically Asian resident in the city, then the output of the second program gives away the monthly income of that resident. The output of the first program may, however, be shared freely with this programmer because it does not reveal any sensitive information. 

Next, how do we reach this conclusion about sharing output by \textit{just} analysing the respective programs? In other words, given a program, how do we decide whether or not its output reveals any sensitive information? To answer this question, we need to perform a dependency analysis called information flow analysis: If the output of a program \textit{depends upon} any information deemed sensitive, then the output too needs to be treated as sensitive. Conversely, if the output \textit{does not depend upon} any sensitive information, then the output too is not sensitive. In the second program above, the output, \textcolor{red}{\lstinline|avgIncA|}, is sensitive because it \textit{depends upon} \textcolor{red}{\lstinline|totIncA|}, which in turn \textit{depends upon} \textcolor{red}{\lstinline|fourth x|}, a sensitive piece of information. On the other hand, in the first program, all the data used to compute the output are non-sensitive, rendering the output itself non-sensitive.

To analyse dependency of output upon sensitive information, information flow calculi typically incorporate sensitivity of information in the types themselves. For example, the type of \lstinline|Db| would be \lstinline|[(T L String, T L Int, T L String, T H Int)]|, where \lstinline|T l String| and \lstinline|T l Int| are the types of \verb|l|-security strings and integers respectively, with \verb|l| being \verb|H| or \verb|L|, corresponding to high-security and low-security respectively. In information flow calculi, functions, too, have types that take sensitivity of information into account. For example, the type of `$+$' would be \lstinline|T l Int -> T l Int -> T l Int|. Now, if we write the above programs in such a calculus, we would see that \verb|fracEldC| has type \verb|T L Double| whereas \textcolor{red}{\lstinline|avgIncA|} has type \verb|T H Double|. Any user can access terms of type \verb|T L Double| but terms of type \verb|T H Double| are only accessible to users with high-security clearance. Within the calculus, the noninterference property enforces this restriction. In this way, information flow calculi ensure secure flow of information.

Our next example takes up another form of dependency analysis: binding-time analysis. Binding-time analysis helps in staged execution of programs. Staged execution comes in handy when programs have inputs that are statically known in addition to inputs that are known only at run-time. Computations that depend only upon static inputs may be carried out statically, thereby producing residual programs that can be executed faster at run-time. To find out which computations depend only upon static inputs, we perform a dependency analysis, called binding-time analysis. 

Next, we shall use the same database example to show binding-time analysis in action. However, instead of sensitivity of information, here we focus on availability of information. From the given demographic parameters, name and ethnicity remain constant over time whereas monthly income varies. To account for this fact, the city council mandates that every resident update their monthly income on the last day of each month. Now, to get an accurate answer to the second query, one needs to run the second program on  such days. But if the database contains millions of entries, running this program may burden the computing system on these days. However, observe that some or perhaps most of the work done by this program need not wait for the update in monthly incomes. For example, \verb|lstAsn| can be computed statically beforehand because this computation depends only upon static information. Binding-time analysis identifies such computations, thereby enabling faster execution of programs.

Binding-time calculi typically incorporate information about binding-time in the types themselves. For example, the type of \lstinline|Db| would be \lstinline|[(T Sta String, T Sta Int, T Sta String, T Dyn Int)]|, where \lstinline|T l String| and \lstinline|T l Int| are respectively the types of strings and integers available at time \verb|l|, with \verb|l| being \verb|Dyn| or \verb|Sta|, corresponding to dynamic availability and static availability respectively. In binding-time calculi, functions, too, have types that take binding-time information into account. For example, the type of `$+$' would be \lstinline|T l Int -> T l Int -> T l Int|. If we write the second program in such a calculus, we would see that the computation of \verb|lstAsn| depends upon static data only (To compute \textcolor{red}{\lstinline|totIncA|}, however, we would need dynamic data, \textcolor{red}{\lstinline|fourth x|}). So we could compute \verb|lstAsn| statically and thereafter run \textit{only the residual program} on the last day of every month, thereby reducing the burden on the computing system on such days. To give an estimate, if ethnically Asian residents constitute 5\% of the total population (say), then, compared to the original program, the residual program may run 20x faster.

With this background on dependency analysis, we shall now work towards building our dependency calculi. We shall present two key dependency calculi in this paper: GMCC and its extension, \Ge{}. The calculus GMCC is built up from a graded monadic calculus, GMC, and a graded comonadic calculus, GCC. In the next section, we look at the Graded Monadic Calculus (GMC).

%% file: Sections/Monads.tex
\section{Graded Monadic Calculus} \label{secgmc}

Moggi \citep{moggi} showed that computational effects can be understood in terms of monads. On the other hand, \citet{gifford} showed that side effects can be tracked using effect classes. \citet{wadler} later adapted effect classes to monads but left open the question of a general theory of effects and monads. Eventually, \citet{katsumata} presented such a general theory in terms of graded monads. In this section, we adapt Katsumata's Explicit Subeffecting Calculus to present a simply-typed Graded Monadic Calculus (GMC).

GMC is an extension of the simply-typed $\lambda$-calculus with a grade-annotated monadic type constructor $T_{m}$. The grades, $m$, are drawn from an arbitrary preordered monoid $ \mathcal{M}  = (M, \_\cdot\_  , 1 ,  \leq )$. Recall that a preordered monoid $ \mathcal{M} $ is a monoid $(M ,  \_\cdot\_  , 1)$ along with a preorder $ \leq $ such that the order respects the binary operator, meaning if $ m_{{\mathrm{1}}}   \leq   m'_{{\mathrm{1}}} $ and $ m_{{\mathrm{2}}}   \leq   m'_{{\mathrm{2}}} $, then $  m_{{\mathrm{1}}}  \cdot  m_{{\mathrm{2}}}    \leq    m'_{{\mathrm{1}}}  \cdot  m'_{{\mathrm{2}}}  $. Note that whenever we need to be precise, we use GMC($ \mathcal{M} $) to refer to GMC parametrized by $ \mathcal{M} $. We follow the same convention for other calculi parametrized by algebraic structures. \\ Next, we present the calculus formally. 

\subsection{Grammar and Type System}

The grammar of the calculus appears in Figure \ref{gmc}. In addition to the types and terms of standard $\lambda$-calculus, we have a graded monadic type $ T_{ m } \:  \ottnt{A} $ and terms related to it. The typing rules of the calculus appear in Figure \ref{typGMC}. We omit the typing rules of standard $\lambda$-calculus and consider only the ones related to the graded monadic type.
\begin{figure}[h]
\begin{align*}
 \text{types}, A, B & ::=  \ottkw{Unit}  \: | \:  \mathbf{Void}  \: | \:  \ottnt{A}  \to  \ottnt{B}  \: | \:  \ottnt{A}  \times  \ottnt{B}  \: | \:  \ottnt{A}  +  \ottnt{B}   \: | \:  T_{ m } \:  \ottnt{A}  \\
 \text{terms}, a, b, f, g & ::= \ottmv{x} \: | \:  \lambda  \ottmv{x}  :  \ottnt{A}  .  \ottnt{b}  \: | \:  \ottnt{b}  \:  \ottnt{a}  \\
 & \: | \:  (  \ottnt{a_{{\mathrm{1}}}}  ,  \ottnt{a_{{\mathrm{2}}}}  )  \: | \:  \mathbf{proj}_1 \:  \ottnt{a}  \: | \:  \mathbf{proj}_2 \:  \ottnt{a}  \: | \:  \ottkw{unit}  \\
 & \: | \:  \mathbf{inj}_1 \:  \ottnt{a_{{\mathrm{1}}}}  \: | \:  \mathbf{inj}_2 \:  \ottnt{a_{{\mathrm{2}}}}  \: | \:  \mathbf{case} \:  \ottnt{a}  \: \mathbf{of} \:  \ottnt{b_{{\mathrm{1}}}}  \: ; \:  \ottnt{b_{{\mathrm{2}}}}  \: | \:  \mathbf{abort} \:  \ottnt{a}  \\
 & \: | \:  \ottkw{ret}  \:  \ottnt{a}  \: | \:  \mathbf{lift}^{ m }  \ottnt{f}  \: | \:  \mathbf{join}^{ m_{{\mathrm{1}}} , m_{{\mathrm{2}}} }  \ottnt{a}  \: | \:  \mathbf{up}^{ m_{{\mathrm{1}}} , m_{{\mathrm{2}}} }  \ottnt{a}  \\
 \text{contexts}, \Gamma & ::=  \emptyset  \: | \:  \Gamma  ,  \ottmv{x}  :  \ottnt{A} 
\end{align*}
\caption{Grammar of GMC}
\label{gmc} 
\end{figure}
\begin{figure}
\drules[M]{$ \Gamma  \vdash  \ottnt{a}  :  \ottnt{A} $}{Typing rules}{Return,Fmap,Join,Up} 
\caption{Typing rules of GMC (excerpt)} 
\label{typGMC}
\end{figure}

The \rref{M-Return,M-Fmap,M-Join} are generalizations of the corresponding rules for the ungraded monadic type. Note that the \rref{M-Join} `joins' the grades using the binary operator of the monoid. The \rref{M-Up} relaxes the grade on the monadic type. If $ \mathcal{M} $ is the trivial preordered monoid, then the above rules degenerate to the standard typing rules for monads.\\ Next, we look at the equational theory of the calculus. 

\subsection{Equational Theory} 

Equality over terms of the graded monadic calculus is a congruent equivalence relation generated by the standard $\beta\eta$-equality rules over $\lambda$-terms and the additional rules that appear in Figure \ref{eqGMC}. For presenting the rules, we use the shorthand notation: $ \ottnt{a}  \:  \leftindex^{ m_{{\mathrm{1}}} }{\gg}\!\! =^{ m_{{\mathrm{2}}} }  \ottnt{f}  \triangleq  \mathbf{join}^{ m_{{\mathrm{1}}} , m_{{\mathrm{2}}} }   (    (   \mathbf{lift}^{ m_{{\mathrm{1}}} }  \ottnt{f}   )   \:  \ottnt{a}   )  $ where $a :  T_{ m_{{\mathrm{1}}} } \:  \ottnt{A} $ and $f :  \ottnt{A}  \to   T_{ m_{{\mathrm{2}}} } \:  \ottnt{B}  $. Note that $\_\!\leftindex^{m_{{\mathrm{1}}}}{\gg}\!\! =^{m_{{\mathrm{2}}}}\!\_$ is a graded $\mathbf{bind}$-operator. 
\begin{figure}
\begin{align}
 \mathbf{lift}^{ m }   (   \lambda  \ottmv{x}  .  \ottmv{x}   )   & \equiv  \lambda  \ottmv{x}  .  \ottmv{x}  \label{eq:idn}\\
 \mathbf{lift}^{ m }   (    \lambda  \ottmv{x}  .  \ottnt{g}   \:   (   \ottnt{f}  \:  \ottmv{x}   )    )   & \equiv   \lambda  \ottmv{x}  .   (   \mathbf{lift}^{ m }  \ottnt{g}   )    \:   (    (   \mathbf{lift}^{ m }  \ottnt{f}   )   \:  \ottmv{x}   )   \label{eq:comp}\\
 \mathbf{up}^{ m_{{\mathrm{1}}} , m_{{\mathrm{1}}} }  \ottnt{a}  & \equiv \ottnt{a} \label{eq:refl}\\
 \mathbf{up}^{ m_{{\mathrm{2}}} , m_{{\mathrm{3}}} }   (   \mathbf{up}^{ m_{{\mathrm{1}}} , m_{{\mathrm{2}}} }  \ottnt{a}   )   & \equiv  \mathbf{up}^{ m_{{\mathrm{1}}} , m_{{\mathrm{3}}} }  \ottnt{a}  \label{eq:trans} \\
\hspace*{-25pt}   (   \mathbf{up}^{ m_{{\mathrm{1}}} , m'_{{\mathrm{1}}} }  \ottnt{a}   )   \:  \leftindex^{ m'_{{\mathrm{1}}} }{\gg}\!\! =^{ m_{{\mathrm{2}}} }  \ottnt{f}  & \equiv  \mathbf{up}^{   m_{{\mathrm{1}}}  \cdot  m_{{\mathrm{2}}}   ,   m'_{{\mathrm{1}}}  \cdot  m_{{\mathrm{2}}}   }   (   \ottnt{a}  \:  \leftindex^{ m_{{\mathrm{1}}} }{\gg}\!\! =^{ m_{{\mathrm{2}}} }  \ottnt{f}   )   \label{eq:natl}\\
\hspace*{-25pt}  \ottnt{a}  \:  \leftindex^{ m_{{\mathrm{1}}} }{\gg}\!\! =^{ m'_{{\mathrm{2}}} }   (   \lambda  \ottmv{x}  .   \mathbf{up}^{ m_{{\mathrm{2}}} , m'_{{\mathrm{2}}} }  \ottnt{b}    )   & \equiv  \mathbf{up}^{   m_{{\mathrm{1}}}  \cdot  m_{{\mathrm{2}}}   ,   m_{{\mathrm{1}}}  \cdot  m'_{{\mathrm{2}}}   }   (   \ottnt{a}  \:  \leftindex^{ m_{{\mathrm{1}}} }{\gg}\!\! =^{ m_{{\mathrm{2}}} }   \lambda  \ottmv{x}  .  \ottnt{b}    )   \label{eq:natr}\\
  (   \ottkw{ret}  \:  \ottnt{a}   )   \:  \leftindex^{ \ottsym{1} }{\gg}\!\! =^{ m }  \ottnt{f}  & \equiv  \ottnt{f}  \:  \ottnt{a}  \label{eq:idl}\\
 \ottnt{a}  \:  \leftindex^{ m_{{\mathrm{1}}} }{\gg}\!\! =^{ \ottsym{1} }   (   \lambda  \ottmv{x}  .   \ottkw{ret}  \:  \ottmv{x}    )   & \equiv \ottnt{a} \label{eq:idr}\\
\hspace*{-25pt}   (   \ottnt{a}  \:  \leftindex^{ m_{{\mathrm{1}}} }{\gg}\!\! =^{ m_{{\mathrm{2}}} }  \ottnt{f}   )   \:  \leftindex^{   m_{{\mathrm{1}}}  \cdot  m_{{\mathrm{2}}}   }{\gg}\!\! =^{ m_{{\mathrm{3}}} }  \ottnt{g}  & \equiv  \ottnt{a}  \:  \leftindex^{ m_{{\mathrm{1}}} }{\gg}\!\! =^{   m_{{\mathrm{2}}}  \cdot  m_{{\mathrm{3}}}   }   (   \lambda  \ottmv{x}  .   (    \ottnt{f}  \:  \ottmv{x}   \:  \leftindex^{ m_{{\mathrm{2}}} }{\gg}\!\! =^{ m_{{\mathrm{3}}} }  \ottnt{g}   )    )   \label{eq:assoc}
\end{align}
\caption{Equality rules of GMC (excerpt)}
\label{eqGMC}
\end{figure}

The first two rules correspond to preservation of identity function and composition of functions by $\mathbf{lift}$. The next two rules correspond to reflexivity and transitivity of the order relation. The two rules after that correspond to commutativity of $\mathbf{bind}$ with $\mathbf{up}$. The two subsequent rules correspond to $\mathbf{ret}$ being the left and the right identity of $\mathbf{bind}$. The last rule corresponds to associativity of $\mathbf{bind}$. 

We now want to interpret this calculus in a suitable category. The types of standard $\lambda$-calculus can be interpreted in any bicartesian closed category. To interpret the graded monadic type, we need a graded monad. \citet{fujii} provides a nice account on graded monads and graded comonads. However, for the sake of self-containment, we shall briefly review the theory behind graded monads in the following section.

\subsection{Graded Monads} \label{gradedmonad}

A graded monad is a certain kind of lax monoidal functor \citep{maclane}. A lax monoidal functor from a monoidal category $(M,\otimes_M,1_M)$ to a monoidal category $(N,\otimes_N,1_N)$ is a 3-tuple $(F,F_2,F_0)$ where, 
\begin{itemize}
\item $F$ is a functor from $M$ to $N$
\item For $X , Y \in \text{Obj} (M)$, morphisms \\ $F_2 (X , Y) : F(X) \otimes_N F(Y) \to F(X \otimes_M Y)$ are natural in $X$ and $Y$
\item $F_0 : 1_N \to F (1_M)$
\end{itemize}
such that the diagrams in Figure \ref{laxDiag} commute.
\begin{figure}[h]
\begin{center}
\begin{tikzcd}
 F(1_M) \otimes_N F(X) \arrow{dr}[below,left]{F_2(1_M,X)} & F(X) \arrow{l}[above]{F_0 \otimes_N \I} \arrow{r}{\I \otimes_N F_0} \arrow{d}{\I} & F(X) \otimes_N F(1_M) \arrow{dl}{F_2(X,1_M)} \\
& F(X)
\end{tikzcd}
\end{center}
\begin{center}
\begin{tikzcd}[column sep = 4 em]
F(X) \otimes_N F(Y) \otimes_N F(Z) \arrow{r}{\I \otimes_N F_2(Y,Z)} \arrow{d}[left]{F_2(X,Y) \otimes_N \I} & F(X) \otimes_N F (Y \otimes_M Z) \arrow{d}{F_2(X,Y \otimes_M Z)} \\
F (X \otimes_M Y) \otimes_N F(Z) \arrow{r}{F_2(X \otimes_M Y, Z)} & F (X \otimes_M Y \otimes_M Z) 
\end{tikzcd}
\end{center}
\caption{Commutative diagrams for lax monoidal functor}
\label{laxDiag}
\end{figure}

Note that here we assume $M$ and $N$ to be strict monoidal categories, i.e. $1_M \otimes_M X = X = X \otimes 1_M$ and $(X \otimes_M Y) \otimes_M Z = X \otimes_M (Y \otimes_M Z)$ for any $X, Y, Z \in \text{Obj}(M)$, and similarly for $N$. 

An example of a strict monoidal category is a preordered monoid $\mathcal{M}$. Any preorder $(M,  \leq )$ may be seen as a category, $\Ca({(M, \leq )})$, that has $M$ as its set of objects and a unique morphism   from $m_{{\mathrm{1}}}$ to $m_{{\mathrm{2}}}$ if and only if $ m_{{\mathrm{1}}}   \leq   m_{{\mathrm{2}}} $. Identity morphisms and composition of morphisms are given by reflexivity and transitivity of the order relation. Now, a preordered monoid $\mathcal{M} = (M, \_\cdot\_  , 1 ,  \leq )$ may be seen as a strict monoidal category: $\Ca({\mathcal{M}}) = (\Ca({(M, \leq )}), \_\cdot\_ ,1)$. 

Another example of a strict monoidal category is the category of endofunctors. Given any category $\Ct$, the endofunctors of $\Ct$ form a strict monoidal category, $\E$, with the tensor product given by composition $\_\circ\_$ of functors and the identity object given by the identity functor \Id. 

We use these two monoidal categories to define a graded monad. An $ \mathcal{M} $-graded monad over $\Ct$ is a lax monoidal functor from $\Ca( \mathcal{M} )$ to $\E$. We wish to use an $ \mathcal{M} $-graded monad to interpret  GMC($ \mathcal{M} $). However, such a monad doesn't stand up to the task (Try interpreting \rref{M-Fmap}!). This shortcoming should not come as a surprise because we know that monads, in and of themselves, cannot model the monadic type constructor of Moggi's computational metalanguage \citep{moggi}. For that, they need to be accompanied with tensorial strengths. Here too, we need to add tensorial strengths to graded monads to interpret the graded monadic type constructor. One could define tensorial strength separately after having defined a graded monad first. However, in lieu, one can also just define a strong graded monad in one go using the category of strong endofunctors and strong natural transformations. 

An endofunctor $F$ on a monoidal category $(M,\otimes,1,\alpha,\lambda,\rho)$ is said to be strong \citep{kelly,kock1,kock2} if there exists morphisms $t_{X,Y} : X \otimes F (Y) \to F (X \otimes Y)$, natural in $X$ and $Y$, for $X, Y \in \text{Obj}(M)$ such that the diagrams in Figure \ref{strong} commute.
\begin{figure}[h]
\begin{center}
\begin{tikzcd}
(X \otimes Y) \otimes FZ \arrow{rr}{t_{X\otimes Y,Z}} \arrow{d}[left]{\alpha^{-1}_{X,Y,FZ}} & & F((X \otimes Y) \otimes Z) \arrow{d}{F \alpha^{-1}_{X,Y,Z}} \\
X \otimes (Y \otimes FZ) \arrow{r}{\I \otimes t_{Y,Z}} & X \otimes F(Y \otimes Z) \arrow{r}{t_{X,Y\otimes Z}} & F(X \otimes (Y \otimes Z))
\end{tikzcd}
\end{center}  
\begin{center}
\begin{tikzcd}
1 \otimes FX \arrow{r}{t_{1,X}} \arrow{dr}[below,left]{\lambda_{FX}} & F(1\otimes X) \arrow{d}{F\lambda_X} \\
& FX
\end{tikzcd}
\end{center}
\vspace*{-5pt}
\caption{Commutative diagrams for strong endofunctor}
\label{strong}
\end{figure}

Given strong endofunctors $(F,t^F)$ and $(G,t^G)$, a natural transformation $\alpha : F \to G$ is said to be strong, if for any $X, Y \in \text{Obj}(M$), the diagram in Figure \ref{strongN} commutes.
\begin{figure}
\begin{center}
\begin{tikzcd}
X \otimes F Y \arrow{d}[left]{t^F_{X,Y}} \arrow{r}{\I \otimes \alpha_Y} & X \otimes G Y \arrow{d}{t^G_{X,Y}} \\
F(X \otimes Y) \arrow{r}{\alpha_{X \otimes Y}} & G(X \otimes Y)
\end{tikzcd}
\end{center}
\vspace*{-5pt}
\caption{Commutative diagram for strong natural transformation} 
\label{strongN}
\end{figure}
Strong endofunctors can be defined for arbitrary monoidal categories; however, we just need the ones over cartesian monoidal categories. Given any cartesian category $\Ct$, let $\E^s$ denote the category with objects: strong endofunctors over $(\Ct,\times,\top)$ (where $\top$ is the terminal object) and morphisms: strong natural transformations between them. Like $\E$, category $\E^s$ too is strict monoidal with the monoidal product and the identity object defined in the same way.

Finally, we have the definition of a strong graded monad. Given a preordered monoid $\mathcal{M}$ and a cartesian category $\Ct$, a strong $\mathcal{M}$-graded monad over $\Ct$ is a lax monoidal functor from $\Ca({\mathcal{M}})$ to $\E^s$. Using strong graded monads, we can now provide a categorical model for the graded monadic calculus.

\subsection{Categorical Model}


Let $\Ct$ be any bicartesian closed category. Let $(\T,\mu,\eta)$ be a strong $\mathcal{M}$-graded monad over $\Ct$. Then, the interpretation, $\llbracket \_ \rrbracket$, of types and terms is as follows: The types and terms of standard $\lambda$-calculus are interpreted in the usual way. The graded monadic type and terms related to it are interpreted in Figure \ref{intGMC}.
\begin{figure}[h]
\begin{align*}
 \llbracket   T_{ m } \:  \ottnt{A}   \rrbracket  & =  \mathbf{T} _{m} \llbracket  \ottnt{A}  \rrbracket   &
 \llbracket   \ottkw{ret}  \:  \ottnt{a}   \rrbracket  & = \eta \circ  \llbracket  \ottnt{a}  \rrbracket  \\
 \llbracket   \mathbf{lift}^{ m }  \ottnt{f}   \rrbracket  & = \Lambda (\T_{m}(\Lambda^{-1}  \llbracket  \ottnt{f}  \rrbracket ) \circ t^{\T_{m}}) &
 \llbracket   \mathbf{join}^{ m_{{\mathrm{1}}} , m_{{\mathrm{2}}} }  \ottnt{a}   \rrbracket  & = \mu^{m_{{\mathrm{1}}},m_{{\mathrm{2}}}} \circ  \llbracket  \ottnt{a}  \rrbracket  \\
 \llbracket   \mathbf{up}^{ m_{{\mathrm{1}}} , m_{{\mathrm{2}}} }  \ottnt{a}   \rrbracket  & =  \mathbf{T} ^{ m_{{\mathrm{1}}}   \leq   m_{{\mathrm{2}}} } \circ  \llbracket  \ottnt{a}  \rrbracket  
\end{align*}
\vspace*{-5pt}
\caption{Interpretation of GMC (excerpt)}
\label{intGMC}
\end{figure}

There are a few things to note here: 
\begin{itemize}
\item $\T(m)$, written as $\T_{m}$, is a functor
\item $ \mathbf{T} ({ m_{{\mathrm{1}}}   \leq   m_{{\mathrm{2}}} })$, written as $ \mathbf{T} ^{ m_{{\mathrm{1}}}   \leq   m_{{\mathrm{2}}} }$, is a natural transformation
\item $\eta$ is a natural transformation from $\Id$ to $\T_1$
\item $\mu^{m_{{\mathrm{1}}},m_{{\mathrm{2}}}}$ are morphisms from $\T_{m_{{\mathrm{1}}}} \circ \T_{m_{{\mathrm{2}}}}$ to $\T_{ m_{{\mathrm{1}}}  \cdot  m_{{\mathrm{2}}} }$, and are natural in both $m_{{\mathrm{1}}}$ and $m_{{\mathrm{2}}}$
\item $t^{\T_{m}}$ denotes the strength of $\T_{m}$
\item $\Lambda$ and $\Lambda^{-1}$ denote currying and uncurrying respectively
\end{itemize}

Let us now see why this interpretation satisfies the equational theory of the calculus. Equations (\ref{eq:idn}) and (\ref{eq:comp}) follow because $\T_m$ is a functor, for any $m \in M$. Equations (\ref{eq:refl}) and (\ref{eq:trans}) follow because $\T$ is a functor and as such, preserves identity morphisms and composition of morphisms. Equations (\ref{eq:natl}) and (\ref{eq:natr}) follow because $\mu$ is natural in its first component and its second component respectively. Equations (\ref{eq:idl}) and (\ref{eq:idr}) follow respectively from the left and the right unit laws for graded monad, laws that correspond to the commutative triangles in Figure \ref{laxDiag}. Equation (\ref{eq:assoc}) follows from the associative law for graded monad, the law that corresponds to the commutative square in Figure \ref{laxDiag}. We shall point out that soundness of the equations in the model also depends upon the axioms about strength, shown in Figure \ref{strong}. 

Thus, a bicartesian closed category $\Ct$ with $(\T,\mu,\eta)$, a strong $\mathcal{M}$-graded monad over $\Ct$, is a sound model for GMC($ \mathcal{M} $).

\begin{theorem} \label{gmcsound}
If $ \Gamma  \vdash  \ottnt{a}  :  \ottnt{A} $ in GMC, then $ \llbracket  \ottnt{a}  \rrbracket  \in \text{Hom}_{\Ct} ( \llbracket  \Gamma  \rrbracket ,  \llbracket  \ottnt{A}  \rrbracket )$. Further, if $ \Gamma  \vdash  \ottnt{a_{{\mathrm{1}}}}  :  \ottnt{A} $ and $ \Gamma  \vdash  \ottnt{a_{{\mathrm{2}}}}  :  \ottnt{A} $ such that $\ottnt{a_{{\mathrm{1}}}} \equiv \ottnt{a_{{\mathrm{2}}}}$ in GMC, then $ \llbracket  \ottnt{a_{{\mathrm{1}}}}  \rrbracket  =  \llbracket  \ottnt{a_{{\mathrm{2}}}}  \rrbracket  \in \text{Hom}_{\Ct} ( \llbracket  \Gamma  \rrbracket ,  \llbracket  \ottnt{A}  \rrbracket )$.
\end{theorem}    


\section{DCC and GMC} \label{secdccgmc}

In this section, we look at the relation between DCC and GMC. Since DCC lies at the heart of our paper, we next review the calculus briefly. For simplicity, we first focus on the terminating fragment of the calculus and consider non-termination later in our paper.

\subsection{Dependency Core Calculus} \label{secdcc}

The Dependency Core Calculus is simply-typed $\lambda$-calculus, extended with multiple type constructors, $\mathcal{T}_{\ell}$, which help analyse dependencies. The indices, $\ell$, are elements of an abstract lattice $ \mathcal{L}  = (L,\vee,\wedge,\bot,\top)$. The lattice structure for the calculus is motivated by the lattice model of secure information flow \citep{denning1}. The elements of a lattice model may be thought of as dependency levels, with $ \ell_{{\mathrm{1}}}  \sqsubseteq  \ell_{{\mathrm{2}}} $ meaning $\ell_{{\mathrm{2}}}$ may depend upon $\ell_{{\mathrm{1}}}$ and $\neg( \ell_{{\mathrm{1}}}  \sqsubseteq  \ell_{{\mathrm{2}}} )$ meaning $\ell_{{\mathrm{2}}}$ should not depend upon $\ell_{{\mathrm{1}}}$. (Here, $\sqsubseteq$ is the implied order of the lattice.) For example, public and secret levels may be modelled using a two-point lattice $ \mathcal{L}_2 $: $ \mathbf{Public}  \sqsubset  \mathbf{Secret} $.

DCC uses an auxiliary protection judgement to analyse dependency. The protection judgement, written $ \ell  \sqsubseteq  \ottnt{A} $, and presented in Figure \ref{protect}, can be read as: the terms of type $\ottnt{A}$ may depend upon level $\ell$. 
Another way to read it is: the terms of type $\ottnt{A}$ are at least as secure as level $\ell$. 
With this reading of the protection judgement, the calculus may be said to be correct when it satisfies the following condition: if $ \ell  \sqsubseteq  \ottnt{A} $ and $\neg( \ell  \sqsubseteq  \ell' )$, then the terms of $\ottnt{A}$ are not be visible at $\ell'$. Let us now look at the type system and equational theory of DCC.

\begin{figure}
\drules[Prot]{$ \ell  \sqsubseteq  \ottnt{A} $}{DCC Protect}{Prod,Fun,Monad,Already} 
\caption{Protection rules for DCC} 
\label{protect}
\end{figure}

\subsection{Type System and Equational Theory of DCC}

The typing rules of DCC consist of the ones for standard $\lambda$-calculus along with the introduction and elimination rules for $\mathcal{T}_{\ell}$, shown below.
\drules[DCC]{$ \Gamma  \vdash  \ottnt{a}  :  \ottnt{A} $}{DCC Typing (Excerpt)}{Eta,Bind}
The protection judgement in \rref{DCC-Bind} ensures that $\ottnt{a}$ is visible to $\ottnt{B}$ only if $\ottnt{B}$ has the necessary permission. Note that \rref{DCC-Bind}, unlike a standard monadic bind rule, does not wrap the return type, $\ottnt{B}$, with the constructor $\mathcal{T}_{\ell}$. This difference is significant and we shall see its implications as we go along.  

Now we consider the equational theory of DCC. \citet{dcc} do not explicitly provide an equational theory for DCC. However, they provide an operational semantics for DCC. We describe the equational theory corresponding to the operational semantics they provide. The terms of DCC can be seen as $\lambda$-terms annotated with security labels. If we erase the annotations, we are left with plain $\lambda$-terms. Plain $\lambda$-terms already have an equational theory: the one generated by the standard $\beta\eta$-rules. Using this theory, we define the equational theory of DCC as follows: two DCC terms are equal, if and only if, after erasure, they are equal as $\lambda$-terms, i.e. $ \ottnt{a_{{\mathrm{1}}}}  \simeq  \ottnt{a_{{\mathrm{2}}}}  \triangleq   \lfloor  \ottnt{a_{{\mathrm{1}}}}  \rfloor   \equiv   \lfloor  \ottnt{a_{{\mathrm{2}}}}  \rfloor  $,  where $ \lfloor  \ottnt{a}  \rfloor $ is the plain $\lambda$-term corresponding to the DCC-term $\ottnt{a}$. 

Now we are in a position to explore the relation between DCC and GMC. There are two questions that we would like to address. 
\begin{itemize}
\item Is DCC a graded monadic calculus? In other words, with appropriate restrictions, can we translate GMC to DCC while preserving meaning?
\item Is DCC just a graded monadic calculus? In other words, with appropriate restrictions, can we translate DCC to GMC while preserving meaning? 
\end{itemize}  
We shall see that the answer to the first question is yes, while the answer to the second one is no.

\subsection{Is DCC a Graded Monadic Calculus?}
\label{DCCMonadic}

Both DCC and GMC are calculi parametrized by algebraic structures. To compare the calculi, we first need to relate the parametrizing structures. DCC is parametrized by an arbitrary lattice $ \mathcal{L} $ whereas GMC is parametrized by an arbitrary preordered monoid $ \mathcal{M} $. A preordered monoid is a more general structure because any bounded semilattice may be seen as a preordered monoid. For example, a bounded join-semilattice is a preordered monoid with multiplication, unit and the preorder given by join, $\bot$ and the semilattice order respectively. A point to note here is that in the original formulation of \citet{denning1}, the semantics of secure information flow just constrains the model to be a bounded join-semilattice. However, under the practical assumption of finiteness, such a model collapses to a lattice. Here, we shall compare DCC and GMC over the class of bounded join-semilattices.

Let $ \mathcal{L}  = (L,\vee,\bot)$ be a bounded join-semilattice. Then, the translation, $\overline{\phantom{a}}$, from GMC to DCC, is given in Figure \ref{GMCtoDCC}. This translation preserves typing and meaning.
\begin{theorem} \label{GMCtoDCCTh}
If $ \Gamma  \vdash  \ottnt{a}  :  \ottnt{A} $ in GMC($ \mathcal{L} $), then $  \overline{  \Gamma  }   \vdash   \overline{ \ottnt{a} }   :   \overline{ \ottnt{A} }  $ in DCC($ \mathcal{L} $). Further, if $ \Gamma  \vdash  \ottnt{a_{{\mathrm{1}}}}  :  \ottnt{A} $ and $ \Gamma  \vdash  \ottnt{a_{{\mathrm{2}}}}  :  \ottnt{A} $ such that $\ottnt{a_{{\mathrm{1}}}} \equiv \ottnt{a_{{\mathrm{2}}}}$ in GMC($ \mathcal{L} $), then $ \overline{ \ottnt{a_{{\mathrm{1}}}} }  \simeq  \overline{ \ottnt{a_{{\mathrm{2}}}} } $ in DCC($ \mathcal{L} $).
\end{theorem}
\begin{figure}
\begin{align*}
 \overline{  T_{  \ell  } \:  \ottnt{A}  }  & =  \mathcal{T}_{ \ell } \:   \overline{ \ottnt{A} }   \\
 \overline{  \ottkw{ret}  \:  \ottnt{a}  }  & =  \mathbf{eta} ^{  \bot  }   \overline{ \ottnt{a} }   \\
 \overline{  \mathbf{lift}^{  \ell  }  \ottnt{f}  }  & =  \lambda  \ottmv{x}  :    \mathcal{T}_{ \ell } \:   \overline{ \ottnt{A} }     .   \mathbf{bind} ^{ \ell } \:  \ottmv{y}  =  \ottmv{x}  \: \mathbf{in} \:   \mathbf{eta} ^{ \ell }   (     \overline{ \ottnt{f} }    \:  \ottmv{y}   )     \hspace{3pt} [ \text{Here, }  \ottnt{f}  :   \ottnt{A}  \to  \ottnt{B}   ] \\
 \overline{  \mathbf{join}^{  \ell_{{\mathrm{1}}}  ,  \ell_{{\mathrm{2}}}  }  \ottnt{a}  }  & =  \mathbf{bind} ^{ \ell_{{\mathrm{1}}} } \:  \ottmv{x}  =    \overline{ \ottnt{a} }    \: \mathbf{in} \:   \mathbf{bind} ^{ \ell_{{\mathrm{2}}} } \:  \ottmv{y}  =  \ottmv{x}  \: \mathbf{in} \:   \mathbf{eta} ^{   \ell_{{\mathrm{1}}}  \vee  \ell_{{\mathrm{2}}}   }  \ottmv{y}    \\
 \overline{  \mathbf{up}^{  \ell_{{\mathrm{1}}}  ,  \ell_{{\mathrm{2}}}  }  \ottnt{a}  }  & =  \mathbf{bind} ^{ \ell_{{\mathrm{1}}} } \:  \ottmv{x}  =    \overline{ \ottnt{a} }    \: \mathbf{in} \:   \mathbf{eta} ^{ \ell_{{\mathrm{2}}} }  \ottmv{x}  
\end{align*}
\caption{Translation function from GMC to DCC (excerpt)}
\label{GMCtoDCC}
\end{figure} 



\subsection{Is DCC Just a Graded Monadic Calculus?}
\label{JustMonadic}

Now that GMC can be translated into DCC, can we go the other way around? Let's explore this question. To translate DCC to GMC, we would need to translate the $ \mathbf{bind} $ construct. 
We may attempt a translation for $ \mathbf{bind} $ of DCC using $ \mathbf{bind} $ of GMC. However, note that the signature of $ \mathbf{bind} $ in GMC is: $  T_{  \ell_{{\mathrm{1}}}  } \:  \ottnt{A}   \to     (   \ottnt{A}  \to   T_{  \ell_{{\mathrm{2}}}  } \:  \ottnt{B}    )   \to   T_{    \ell_{{\mathrm{1}}}  \vee  \ell_{{\mathrm{2}}}    } \:  \ottnt{B}    $ whereas that of $ \mathbf{bind} $ in DCC is: $  \mathcal{T}_{ \ell_{{\mathrm{1}}} } \:  \ottnt{A}   \to   (   \ottnt{A}  \to  \ottnt{B}   )   \to \{  \ell_{{\mathrm{1}}}  \sqsubseteq  \ottnt{B}  \} \to B$. For a successful translation, one needs to show that, if $ \ell_{{\mathrm{1}}}  \sqsubseteq  \ottnt{B} $, then there exists a function $\ottnt{j}$ of type $ (    T_{  \ell_{{\mathrm{1}}}  } \:   \underline{ \ottnt{B} }    \to   \underline{ \ottnt{B} }    ) $. In case such a function exists, for $ \ottnt{a}  :   \mathcal{T}_{ \ell_{{\mathrm{1}}} } \:  \ottnt{A}  $ and $ \ottnt{f}  :   \ottnt{A}  \to  \ottnt{B}  $, one can get $  (   \ottnt{j}  \:   (    (   \mathbf{lift}^{  \ell_{{\mathrm{1}}}  }   \underline{ \ottnt{f} }    )   \:   \underline{ \ottnt{a} }    )    )   :   \underline{ \ottnt{B} }  $. (Here, $\underline{\phantom{a}}$ denotes a possible translation of DCC to GMC.) 

We attempt to define $ \ottnt{j}  :    T_{  \ell_{{\mathrm{1}}}  } \:   \underline{ \ottnt{B} }    \to   \underline{ \ottnt{B} }   $ via structural recursion on the judgement $ \ell_{{\mathrm{1}}}  \sqsubseteq  \ottnt{B} $. The interesting cases are \rref{Prot-Monad,Prot-Already}. 
\begin{itemize}
\item \Rref{Prot-Monad}. Here, we have $ \ell  \sqsubseteq   \mathcal{T}_{ \ell' } \:  \ottnt{B}  $ where $ \ell  \sqsubseteq  \ell' $. Need to define $ \ottnt{j}  :    T_{  \ell  } \:   T_{  \ell'  } \:   \underline{ \ottnt{B} }     \to   T_{  \ell'  } \:   \underline{ \ottnt{B} }    $. But, $  \ottmv{x}  :   T_{  \ell  } \:   T_{  \ell'  } \:   \underline{ \ottnt{B} }      \vdash   \mathbf{join}^{  \ell  ,  \ell'  }  \ottmv{x}   :   T_{  \ell'  } \:   \underline{ \ottnt{B} }   $ because $ \ell  \vee  \ell'  = \ell'$. 
\item \Rref{Prot-Already}. Here, we have $ \ell  \sqsubseteq   \mathcal{T}_{ \ell' } \:  \ottnt{B}  $ where $ \ell  \sqsubseteq  \ottnt{B} $. Need to define $ \ottnt{j}  :    T_{  \ell  } \:   T_{  \ell'  } \:   \underline{ \ottnt{B} }     \to   T_{  \ell'  } \:   \underline{ \ottnt{B} }    $. Since $ \ell  \sqsubseteq  \ottnt{B} $, the hypothesis gives us a function $ \ottnt{j_{{\mathrm{0}}}}  :    T_{  \ell  } \:   \underline{ \ottnt{B} }    \to   \underline{ \ottnt{B} }   $. But, now we are stuck! Lifting this function can only give us: $  \mathbf{lift}^{  \ell'  }  \ottnt{j_{{\mathrm{0}}}}   :    T_{  \ell'  } \:   T_{  \ell  } \:   \underline{ \ottnt{B} }     \to   T_{  \ell'  } \:   \underline{ \ottnt{B} }    $, not exactly what we need. 

Here, we could, for instance, add a non-standard flip-rule to GMC like: ``from $ \Gamma  \vdash  \ottnt{a}  :   T_{  \ell_{{\mathrm{1}}}  } \:   T_{  \ell_{{\mathrm{2}}}  } \:  \ottnt{A}   $, derive $ \Gamma  \vdash   \mathbf{flip}^{ \ell_{{\mathrm{1}}} , \ell_{{\mathrm{2}}} }  \ottnt{a}   :   T_{  \ell_{{\mathrm{2}}}  } \:   T_{  \ell_{{\mathrm{1}}}  } \:  \ottnt{A}   $'' and thereafter translate DCC into it. But such an exercise would defeat our whole purpose because then, GMC would no longer be a graded monadic calculus. Note that \citet{persdcc} includes such a rule in his language SDCC, which is shown to be equivalent to (the terminating fragment of) DCC.   
\end{itemize} 

So we see that DCC is not just a graded monadic calculus. The \rref{Prot-Already} makes it something more than that. This rule enables one to flip the modal type constructors. In DCC, from $ \Gamma  \vdash  \ottnt{a}  :   \mathcal{T}_{ \ell_{{\mathrm{1}}} } \:   \mathcal{T}_{ \ell_{{\mathrm{2}}} } \:  \ottnt{A}   $, one can derive $ \Gamma  \vdash   \mathbf{bind} ^{ \ell_{{\mathrm{1}}} } \:  \ottmv{x}  =  \ottnt{a}  \: \mathbf{in} \:   \mathbf{bind} ^{ \ell_{{\mathrm{2}}} } \:  \ottmv{y}  =  \ottmv{x}  \: \mathbf{in} \:   \mathbf{eta} ^{ \ell_{{\mathrm{2}}} }   \mathbf{eta} ^{ \ell_{{\mathrm{1}}} }  \ottmv{y}      :   \mathcal{T}_{ \ell_{{\mathrm{2}}} } \:   \mathcal{T}_{ \ell_{{\mathrm{1}}} } \:  \ottnt{A}   $, using \rref{Prot-Already}. Such a derivation is not possible in a general monadic calculus.

However, if the calculus is also comonadic in addition to being monadic, such a derivation is possible. From $ \mathcal{T}_{ \ell_{{\mathrm{1}}} } \:   \mathcal{T}_{ \ell_{{\mathrm{2}}} } \:  \ottnt{A}  $, using monadic join, we can get $ \mathcal{T}_{   \ell_{{\mathrm{1}}}  \vee  \ell_{{\mathrm{2}}}   } \:  \ottnt{A} $, which is same as $ \mathcal{T}_{   \ell_{{\mathrm{2}}}  \vee  \ell_{{\mathrm{1}}}   } \:  \ottnt{A} $, from which we can get $ \mathcal{T}_{ \ell_{{\mathrm{2}}} } \:   \mathcal{T}_{ \ell_{{\mathrm{1}}} } \:  \ottnt{A}  $, using comonadic fork. So it seems that DCC has some comonadic flavour to it. But is DCC a graded comonadic calculus? In order to address this question, we first need to build the theory of a graded comonadic calculus.

%% file: Sections/Comonads.tex
\section{Graded Comonadic Calculus} \label{secgcc}

Soon after \citet{moggi} showed that computational effects can be understood in terms of monads, \citet{brookes} showed that intensional behaviour of programs, for example, the number of steps necessary for reduction, can be understood in terms of comonads. While monads can model how programs affect the environment, comonads can model how the environment affects programs. Comonads, with necessary extra structure, have been used by \citet{uustalu} and \citet{petricek1}, among others, to model various notions of environment-dependent computation like resource usage of programs, computation on streams, etc. Several calculi \citep{ghica,petricek,brunel} have been developed to provide a general account of such environment-dependent computation. These calculi are parametrized by semiring-like structures, and are modelled using semiring-graded comonads, possibly including additional structure. In this section, we forgo extra structures and present a graded comonadic calculus that is the dual of the graded monadic calculus presented in Section \ref{secgmc}.  

The Graded Comonadic Calculus (GCC), similar to GMC, is parametrized by a preordered monoid $\mathcal{M}$. In addition to the types and terms of standard $\lambda$-calculus, GCC has a graded comonadic type $ D_{ m } \:  \ottnt{A} $ and terms related to it, shown below.
\begin{align*}
\text{types}, A, B & ::= \ldots \: | \:  D_{ m } \:  \ottnt{A}  \\
\text{terms}, a, b, f, g & ::= \ldots \: | \:  \mathbf{extr} \:  \ottnt{a}  \: | \:  \mathbf{lift}^{ m }  \ottnt{f}  \: | \:  \mathbf{fork}^{ m_{{\mathrm{1}}} , m_{{\mathrm{2}}} }  \ottnt{a}  \: | \:  \mathbf{up}^{ m_{{\mathrm{1}}} , m_{{\mathrm{2}}} }  \ottnt{a} 
\end{align*} 
Now we look at the typing rules and the equational theory of the calculus.

\subsection{Type System and Equational Theory}

The typing rules for terms related to the graded comonadic type are presented in Figure \ref{typGCC}.
\begin{figure}[h]
\drules[C]{$ \Gamma  \vdash  \ottnt{a}  :  \ottnt{A} $}{Typing rules}{Extract,Fmap,Fork,Up}
\caption{Typing rules of GCC (excerpt)}
\label{typGCC}
\end{figure} 
The \rref{C-Extract,C-Fmap,C-Fork} are generalizations of the corresponding rules for the ungraded comonadic type. The \rref{C-Up}, like \rref{M-Up}, relaxes the grade on the type. If $\mathcal{M}$ is the trivial preordered monoid, then the above rules degenerate to the standard typing rules for comonads. Note that \rref{C-Fmap,C-Up} are essentially the same as \rref{M-Fmap,M-Up} respectively whereas \rref{C-Extract,C-Fork} are like `inverses' of \rref{M-Return,M-Join} respectively.

The equational theory of the calculus appears in Figure \ref{eqGCC} (we omit the $\beta\eta$-rules of  standard $\lambda$-calculus). For presenting the rules, we use the shorthand notation: $ \ottnt{f}  \:  \leftindex^{ m_{{\mathrm{2}}} }{\ll}\!\! =^{ m_{{\mathrm{1}}} }  \ottnt{a}  \triangleq   (   \mathbf{lift}^{ m_{{\mathrm{1}}} }  \ottnt{f}   )   \:   (   \mathbf{fork}^{ m_{{\mathrm{1}}} , m_{{\mathrm{2}}} }  \ottnt{a}   )  $ where $ \ottnt{a}  :   D_{   m_{{\mathrm{1}}}  \cdot  m_{{\mathrm{2}}}   } \:  \ottnt{A}  $ and $ \ottnt{f}  :     D_{ m_{{\mathrm{2}}} } \:  \ottnt{A}    \to  \ottnt{B}  $. Note that $  \_   \:  \leftindex^{ m_{{\mathrm{2}}} }{\ll}\!\! =^{ m_{{\mathrm{1}}} }   \_  $ is a graded-$\mathbf{extend}$ operator.
\begin{figure}
\hspace*{-1cm}
\begin{align}
 \mathbf{lift}^{ m }   (   \lambda  \ottmv{x}  .  \ottmv{x}   )   & \equiv  \lambda  \ottmv{x}  .  \ottmv{x}  \label{eq:cidn}\\
 \mathbf{lift}^{ m }   (    \lambda  \ottmv{x}  .  \ottnt{g}   \:   (   \ottnt{f}  \:  \ottmv{x}   )    )   & \equiv   \lambda  \ottmv{x}  .   (   \mathbf{lift}^{ m }  \ottnt{g}   )    \:   (    (   \mathbf{lift}^{ m }  \ottnt{f}   )   \:  \ottmv{x}   )   \label{eq:ccomp}\\
 \mathbf{up}^{ m_{{\mathrm{1}}} , m_{{\mathrm{1}}} }  \ottnt{a}  & \equiv \ottnt{a} \\
 \mathbf{up}^{ m_{{\mathrm{2}}} , m_{{\mathrm{3}}} }   (   \mathbf{up}^{ m_{{\mathrm{1}}} , m_{{\mathrm{2}}} }  \ottnt{a}   )   & \equiv  \mathbf{up}^{ m_{{\mathrm{1}}} , m_{{\mathrm{3}}} }  \ottnt{a}  \\
\hspace*{-25pt}  \ottnt{f}  \:  \leftindex^{ m_{{\mathrm{2}}} }{\ll}\!\! =^{ m'_{{\mathrm{1}}} }   (   \mathbf{up}^{   m_{{\mathrm{1}}}  \cdot  m_{{\mathrm{2}}}   ,   m'_{{\mathrm{1}}}  \cdot  m_{{\mathrm{2}}}   }  \ottnt{a}   )   & \equiv  \mathbf{up}^{ m_{{\mathrm{1}}} , m'_{{\mathrm{1}}} }   (   \ottnt{f}  \:  \leftindex^{ m_{{\mathrm{2}}} }{\ll}\!\! =^{ m_{{\mathrm{1}}} }  \ottnt{a}   )   \label{eq:cnatl}\\ 
\hspace*{-25pt}  \ottnt{f}  \:  \leftindex^{ m'_{{\mathrm{2}}} }{\ll}\!\! =^{ m_{{\mathrm{1}}} }   (   \mathbf{up}^{   m_{{\mathrm{1}}}  \cdot  m_{{\mathrm{2}}}   ,   m_{{\mathrm{1}}}  \cdot  m'_{{\mathrm{2}}}   }  \ottnt{a}   )   & \equiv   (    \lambda  \ottmv{x}  .  \ottnt{f}   \:   (   \mathbf{up}^{ m_{{\mathrm{2}}} , m'_{{\mathrm{2}}} }  \ottmv{x}   )    )   \:  \leftindex^{ m_{{\mathrm{2}}} }{\ll}\!\! =^{ m_{{\mathrm{1}}} }  \ottnt{a}  \label{eq:cnatr}\\
 \mathbf{extr} \:   (   \ottnt{f}  \:  \leftindex^{ m_{{\mathrm{2}}} }{\ll}\!\! =^{ \ottsym{1} }  \ottnt{a}   )   & \equiv  \ottnt{f}  \:  \ottnt{a}  \label{eq:cidl}\\
  (   \lambda  \ottmv{x}  .   \mathbf{extr} \:  \ottmv{x}    )   \:  \leftindex^{ \ottsym{1} }{\ll}\!\! =^{ m_{{\mathrm{2}}} }  \ottnt{a}  & \equiv \ottnt{a} \\
\hspace*{-25pt}  \ottnt{g}  \:  \leftindex^{ m_{{\mathrm{2}}} }{\ll}\!\! =^{ m_{{\mathrm{1}}} }   (   \ottnt{f}  \:  \leftindex^{ m_{{\mathrm{3}}} }{\ll}\!\! =^{   m_{{\mathrm{1}}}  \cdot  m_{{\mathrm{2}}}   }  \ottnt{a}   )   & \equiv   (    \lambda  \ottmv{x}  .  \ottnt{g}   \:   (   \ottnt{f}  \:  \leftindex^{ m_{{\mathrm{3}}} }{\ll}\!\! =^{ m_{{\mathrm{2}}} }  \ottmv{x}   )    )   \:  \leftindex^{   m_{{\mathrm{2}}}  \cdot  m_{{\mathrm{3}}}   }{\ll}\!\! =^{ m_{{\mathrm{1}}} }  \ottnt{a}    \label{eq:cassoc}
\end{align}
\caption{Equality rules of GCC (excerpt)}
\label{eqGCC}
\end{figure}

The first four rules are same as their counterparts in GMC. The next two rules correspond to commutativity of $\mathbf{extend}$ with $\mathbf{up}$. The two rules after that correspond to $\mathbf{extr}$ being the left and the right identity of $\mathbf{extend}$. The last rule corresponds to associativity of $\mathbf{extend}$.  

Next, we want to interpret the calculus in a suitable category. Similar to GMC, the types of standard $\lambda$-calculus can be interpreted in any bicartesian closed category. To interpret the graded comonadic type, we need a strong graded comonad, the dual of a strong graded monad we saw earlier. We briefly go over the definition of a strong graded comonad and present the categorical model thereafter.

\subsection{Graded Comonad and Categorical Model}

While a graded monad is a kind of lax monoidal functor, a graded comonad is a kind of oplax monoidal functor.  An oplax monoidal functor from a monoidal category $(M,\otimes_M,1_M)$ to a monoidal category $(N,\otimes_N,1_N)$ is nothing but a lax monoidal functor from $(M^{\text{op}},\otimes_M,1_M)$ to $(N^{\text{op}},\otimes_N,1_N)$. Now, given a preordered monoid $\mathcal{M}$, an $ \mathcal{M} $-graded comonad over a category $\Ct$ is an oplax monoidal functor from $\Ca( \mathcal{M} )$ to $\E$. In order to interpret GCC, we need to add strength to the graded comonad. We follow the same strategy as before and define a strong $\mathcal{M}$-graded comonad over a cartesian category $\Ct$ as an oplax monoidal functor from $\Ca({\mathcal{M}})$ to $\E^s$. We use a strong graded comonad to build the categorical model of GCC. 

Let $ \mathcal{M}  = (M, \_\cdot\_  , 1 ,  \leq )$ be the preordered monoid parametrizing the calculus. Let $\Ct$ be any bicartesian closed category. Let $( \mathbf{D} ,\delta,\epsilon)$ be a strong $\mathcal{M}$-graded comonad over $\Ct$. Then, the interpretation, $\llbracket \_ \rrbracket$, of types and terms of GCC is as given in Figure \ref{intGCC}.
\begin{figure}[h]
\begin{align*}
 \llbracket    D_{ m } \:  \ottnt{A}    \rrbracket  & =  \mathbf{D} _{m} \llbracket  \ottnt{A}  \rrbracket   &
 \llbracket   \mathbf{extr} \:  \ottnt{a}   \rrbracket  & = \epsilon \circ  \llbracket  \ottnt{a}  \rrbracket  \\
 \llbracket   \mathbf{lift}^{ m }  \ottnt{f}   \rrbracket  & = \Lambda ( \mathbf{D} _{m}(\Lambda^{-1}  \llbracket  \ottnt{f}  \rrbracket ) \circ t^{ \mathbf{D} _{m}}) &
 \llbracket   \mathbf{fork}^{ m_{{\mathrm{1}}} , m_{{\mathrm{2}}} }  \ottnt{a}   \rrbracket  & = \delta^{m_{{\mathrm{1}}},m_{{\mathrm{2}}}} \circ  \llbracket  \ottnt{a}  \rrbracket  \\
 \llbracket   \mathbf{up}^{ m_{{\mathrm{1}}} , m_{{\mathrm{2}}} }  \ottnt{a}   \rrbracket  & =  \mathbf{D} ^{ m_{{\mathrm{1}}}   \leq   m_{{\mathrm{2}}} } \circ  \llbracket  \ottnt{a}  \rrbracket  
\end{align*}
\caption{Interpretation of GCC (excerpt)}
\label{intGCC}
\end{figure}

Note that $\epsilon$ is a natural transformation from $ \mathbf{D} _1$ to $\Id$ and $\delta^{m_{{\mathrm{1}}},m_{{\mathrm{2}}}}$ are morphisms from $ \mathbf{D} _{ m_{{\mathrm{1}}}  \cdot  m_{{\mathrm{2}}} }$ to $ \mathbf{D} _{m_{{\mathrm{1}}}} \circ  \mathbf{D} _{m_{{\mathrm{2}}}}$, natural in both $m_{{\mathrm{1}}}$ and $m_{{\mathrm{2}}}$. By reasoning along the lines of Theorem \ref{gmcsound}, we can show that the above interpretation provides a sound model for GCC.
\begin{theorem} \label{gccsound}
If $ \Gamma  \vdash  \ottnt{a}  :  \ottnt{A} $ in GCC, then $ \llbracket  \ottnt{a}  \rrbracket  \in \text{Hom}_{\Ct} ( \llbracket  \Gamma  \rrbracket ,  \llbracket  \ottnt{A}  \rrbracket )$. Further, if $ \Gamma  \vdash  \ottnt{a_{{\mathrm{1}}}}  :  \ottnt{A} $ and $ \Gamma  \vdash  \ottnt{a_{{\mathrm{2}}}}  :  \ottnt{A} $ such that $\ottnt{a_{{\mathrm{1}}}} \equiv \ottnt{a_{{\mathrm{2}}}}$ in GCC, then $ \llbracket  \ottnt{a_{{\mathrm{1}}}}  \rrbracket  =  \llbracket  \ottnt{a_{{\mathrm{2}}}}  \rrbracket  \in \text{Hom}_{\Ct} ( \llbracket  \Gamma  \rrbracket ,  \llbracket  \ottnt{A}  \rrbracket )$.
\end{theorem}  
Now that we have a graded comonadic calculus with us, we can test the comonadic character of DCC. In the next section, we explore the relation between DCC and GCC. In particular, we ask: with appropriate restrictions, can we translate GCC into DCC?

\section{DCC and GCC} \label{secdccgcc}

In Section \ref{JustMonadic}, we saw that DCC is not just a graded monadic calculus. The \rref{Prot-Already} lends a comonadic character to it. But does it make DCC a graded comonadic calculus? In other words, over the class of bounded join-semilattices, can we translate GCC into DCC? We can translate $ D_{  \ell  } \:  \ottnt{A} $ to $ \mathcal{T}_{ \ell } \:  \ottnt{A} $. The constructs $\mathbf{lift}$ and $\mathbf{up}$ can be translated as in Figure \ref{GMCtoDCC}. But to translate $\mathbf{extr}$ and $\mathbf{fork}$, we need to relook at the protection judgement.

\subsection{Protection Judgement, Revisited} 

To translate $\mathbf{extr}$ and $\mathbf{fork}$, we need to be able to construct functions having types $  \mathcal{T}_{  \bot  } \:  \ottnt{A}   \to  \ottnt{A} $ and $   \mathcal{T}_{   \ell_{{\mathrm{1}}}  \vee  \ell_{{\mathrm{2}}}   } \:  \ottnt{A}    \to   \mathcal{T}_{ \ell_{{\mathrm{1}}} } \:   \mathcal{T}_{ \ell_{{\mathrm{2}}} } \:  \ottnt{A}   $ respectively, for an arbitrary type $\ottnt{A}$. However, given the formulation of DCC by \citet{dcc}, such a construction is not possible. This is so because in order to construct a function having type $  \mathcal{T}_{  \bot  } \:  \ottnt{A}   \to  \ottnt{A} $, we need to show that: $  \bot   \sqsubseteq  \ottnt{A} $, for an arbitrary $\ottnt{A}$. The protection rules do not allow such a derivation. Similarly, in order to construct a function having type $   \mathcal{T}_{   \ell_{{\mathrm{1}}}  \vee  \ell_{{\mathrm{2}}}   } \:  \ottnt{A}    \to   \mathcal{T}_{ \ell_{{\mathrm{1}}} } \:   \mathcal{T}_{ \ell_{{\mathrm{2}}} } \:  \ottnt{A}   $, we need to show that: $  \ell_{{\mathrm{1}}}  \vee  \ell_{{\mathrm{2}}}   \sqsubseteq   \mathcal{T}_{ \ell_{{\mathrm{1}}} } \:   \mathcal{T}_{ \ell_{{\mathrm{2}}} } \:  \ottnt{A}   $, for an arbitrary $\ottnt{A}$. Again, such a derivation is not allowed by the protection rules.

However, from a dependency perspective, $  \bot   \sqsubseteq  \ottnt{A} $ and $  \ell_{{\mathrm{1}}}  \vee  \ell_{{\mathrm{2}}}   \sqsubseteq   \mathcal{T}_{ \ell_{{\mathrm{1}}} } \:   \mathcal{T}_{ \ell_{{\mathrm{2}}} } \:  \ottnt{A}   $ are sound judgements. The judgement $  \bot   \sqsubseteq  \ottnt{A} $ is sound because: $ \bot $ is the lowest security level and as such, the terms of any type are at least as secure as $ \bot $. The judgement $  \ell_{{\mathrm{1}}}  \vee  \ell_{{\mathrm{2}}}   \sqsubseteq   \mathcal{T}_{ \ell_{{\mathrm{1}}} } \:   \mathcal{T}_{ \ell_{{\mathrm{2}}} } \:  \ottnt{A}   $ is sound because: $ \mathcal{T}_{ \ell_{{\mathrm{1}}} } \:   \mathcal{T}_{ \ell_{{\mathrm{2}}} } \:  \ottnt{A}  $ is at least as secure as $\ell_{{\mathrm{1}}}$ and $ \mathcal{T}_{ \ell_{{\mathrm{1}}} } \:   \mathcal{T}_{ \ell_{{\mathrm{2}}} } \:  \ottnt{A}  $ is also at least as secure as $\ell_{{\mathrm{2}}}$, so it must be at least as secure as $ \ell_{{\mathrm{1}}}  \vee  \ell_{{\mathrm{2}}} $. This reasoning is supported by the lattice model of \citet{denning1}.  

Now the above judgements are not only sound, but also desirable. Compared to $\ottnt{A}$, the type $ \mathcal{T}_{  \bot  } \:  \ottnt{A} $ offers no extra protection. So, it makes sense to allow programs like the one shown below. 
\begin{center}
$  \ottmv{x}  :   \mathcal{T}_{  \bot  } \:  \ottnt{A}    \vdash   \mathbf{bind} ^{  \bot  } \:  \ottmv{y}  =  \ottmv{x}  \: \mathbf{in} \:  \ottmv{y}   :  \ottnt{A} $ 
\end{center}
Next, the type $ \mathcal{T}_{ \ell_{{\mathrm{1}}} } \:   \mathcal{T}_{ \ell_{{\mathrm{2}}} } \:  \ottnt{A}  $ offers no less protection than the type $ \mathcal{T}_{   \ell_{{\mathrm{1}}}  \vee  \ell_{{\mathrm{2}}}   } \:  \ottnt{A} $. So, programs like:
\begin{center}
$  \ottmv{x}  :   \mathcal{T}_{   \ell_{{\mathrm{1}}}  \vee  \ell_{{\mathrm{2}}}   } \:  \ottnt{A}    \vdash   \mathbf{bind} ^{   \ell_{{\mathrm{1}}}  \vee  \ell_{{\mathrm{2}}}   } \:  \ottmv{y}  =  \ottmv{x}  \: \mathbf{in} \:   \mathbf{eta} ^{ \ell_{{\mathrm{1}}} }   \mathbf{eta} ^{ \ell_{{\mathrm{2}}} }  \ottmv{y}     :   \mathcal{T}_{ \ell_{{\mathrm{1}}} } \:   \mathcal{T}_{ \ell_{{\mathrm{2}}} } \:  \ottnt{A}   $
\end{center}
should be allowed. 

Allowing programs like the one above has some interesting consequences. For example, consider the lattice shown below. 
\begin{center}
\begin{tikzcd}[row sep=tiny,column sep=tiny]
& \top \\
& \ell_{{\mathrm{3}}} \arrow[u,dash] \\
\ell_{{\mathrm{21}}} \arrow[ur,dash] & & \ell_{{\mathrm{22}}} \arrow[ul,dash] \\
\ell_{{\mathrm{11}}} \arrow[u,dash] & & \ell_{{\mathrm{12}}} \arrow[u,dash] \\
& \bot \arrow[ul,dash] \arrow[ur,dash]
\end{tikzcd}
\end{center}
Here, $ \ell_{{\mathrm{11}}}  \vee  \ell_{{\mathrm{12}}}  = \ell_{{\mathrm{3}}}$. On specializing the example program, we have:
\begin{center}
 $  \ottmv{x}  :   \mathcal{T}_{ \ell_{{\mathrm{3}}} } \:  \ottnt{A}    \vdash   \mathbf{bind} ^{ \ell_{{\mathrm{3}}} } \:  \ottmv{y}  =  \ottmv{x}  \: \mathbf{in} \:   \mathbf{eta} ^{ \ell_{{\mathrm{11}}} }   \mathbf{eta} ^{ \ell_{{\mathrm{12}}} }  \ottmv{y}     :   \mathcal{T}_{ \ell_{{\mathrm{11}}} } \:   \mathcal{T}_{ \ell_{{\mathrm{12}}} } \:  \ottnt{A}   $.
\end{center}
This program shows that we can observe $\ell_{{\mathrm{3}}}$-level values in an environment simultaneously protected by $\ell_{{\mathrm{11}}}$ and $\ell_{{\mathrm{12}}}$. Two points are worth noting here.
\begin{itemize}
\item First, $\ell_{{\mathrm{11}}}$ and $\ell_{{\mathrm{12}}}$ together offer much much more protection than either of them individually. Observe that neither $\ell_{{\mathrm{11}}}$ nor $\ell_{{\mathrm{12}}}$ individually offer as much protection as either $\ell_{{\mathrm{21}}}$ or $\ell_{{\mathrm{22}}}$. But, $\ell_{{\mathrm{11}}}$ and $\ell_{{\mathrm{12}}}$ together offer more protection than both $\ell_{{\mathrm{21}}}$ and $\ell_{{\mathrm{22}}}$. Behind this observation, lies a fundamental security principle, the principle that forms the basis of applications like two-factor authentication, two-man rule, etc. It may be phrased in terms of the age-old proverb: the whole is more than just the sum of its parts.
\item Second, $\ell_{{\mathrm{3}}}$ is compromised if $\ell_{{\mathrm{11}}}$ and $\ell_{{\mathrm{12}}}$ are simultaneously compromised. This is so because with a simultaneous access to $\ell_{{\mathrm{11}}}$ and $\ell_{{\mathrm{12}}}$, one has access to $\ell_{{\mathrm{3}}}$ and all levels below it, even $\ell_{{\mathrm{21}}}$ and $\ell_{{\mathrm{22}}}$. It may look counter-intuitive but that just shows the power of simultaneous access and simultaneous protection.   
\end{itemize}
To enable such reasoning within the calculus, we add the following rules to the protection judgement of DCC.
\drules[Prot]{$ \ell  \sqsubseteq  \ottnt{A} $}{Extended Protection Rules}{Minimum,Combine}
As a side note, we shall point out that in the categorical model of DCC given by \citet{dcc}, for an arbitrary $\ottnt{A}$, the interpretations of $ \mathcal{T}_{  \bot  } \:  \ottnt{A} $ and $ \mathcal{T}_{   \ell_{{\mathrm{1}}}  \vee  \ell_{{\mathrm{2}}}   } \:  \ottnt{A} $ are the same as those of $\ottnt{A}$ and $ \mathcal{T}_{ \ell_{{\mathrm{1}}} } \:   \mathcal{T}_{ \ell_{{\mathrm{2}}} } \:  \ottnt{A}  $ respectively. So, DCC extended with the above rules enjoys the same categorical model. For the sake of precision, we shall call DCC extended with these rules \ED{}. The equational theory of \ED{} is defined in the same way as that of DCC. So then, DCC is a proper sub-language of \ED{}. Since DCC is graded monadic, so is \ED{}. Owing to the reasons described above, DCC is not graded comonadic. But \ED{} is graded comonadic, as we see next.

\subsection{\ED{} is a Graded Comonadic Calculus}

Over the class of bounded join-semilattices, we can translate GCC into \ED{}. Let $ \mathcal{L}  = (L,\vee,\bot)$ be a bounded join-semilattice. Then, the translation $\overline{\phantom{a}}$ is given in Figure \ref{GCCtoDCC}. Note the role played by \rref{Prot-Minimum,Prot-Combine} in the translation of $ \mathbf{extr} \:  \ottnt{a} $ and $ \mathbf{fork}^{  \ell_{{\mathrm{1}}}  ,  \ell_{{\mathrm{2}}}  }  \ottnt{a} $ respectively. The next theorem shows that this translation preserves typing and meaning.
\begin{figure}[h]
\begin{align*}
 \overline{   D_{  \ell  } \:  \ottnt{A}   }  & =  \mathcal{T}_{ \ell } \:   \overline{ \ottnt{A} }   \\
 \overline{  \mathbf{extr} \:  \ottnt{a}  }  & =  \mathbf{bind} ^{  \bot  } \:  \ottmv{x}  =   \overline{ \ottnt{a} }   \: \mathbf{in} \:  \ottmv{x}  \\
 \overline{  \mathbf{lift}^{  \ell  }  \ottnt{f}  }  & =  \lambda  \ottmv{x}  :    \mathcal{T}_{ \ell } \:   \overline{ \ottnt{A} }     .   \mathbf{bind} ^{ \ell } \:  \ottmv{y}  =  \ottmv{x}  \: \mathbf{in} \:   \mathbf{eta} ^{ \ell }   (     \overline{ \ottnt{f} }    \:  \ottmv{y}   )     \hspace{3pt} [ \text{Here, }  \ottnt{f}  :   \ottnt{A}  \to  \ottnt{B}   ] \\
 \overline{  \mathbf{fork}^{  \ell_{{\mathrm{1}}}  ,  \ell_{{\mathrm{2}}}  }  \ottnt{a}  }  & =  \mathbf{bind} ^{   \ell_{{\mathrm{1}}}  \vee  \ell_{{\mathrm{2}}}   } \:  \ottmv{x}  =    \overline{ \ottnt{a} }    \: \mathbf{in} \:   \mathbf{eta} ^{ \ell_{{\mathrm{1}}} }   \mathbf{eta} ^{ \ell_{{\mathrm{2}}} }  \ottmv{x}    \\
 \overline{  \mathbf{up}^{  \ell_{{\mathrm{1}}}  ,  \ell_{{\mathrm{2}}}  }  \ottnt{a}  }  & =  \mathbf{bind} ^{ \ell_{{\mathrm{1}}} } \:  \ottmv{x}  =    \overline{ \ottnt{a} }    \: \mathbf{in} \:   \mathbf{eta} ^{ \ell_{{\mathrm{2}}} }  \ottmv{x}  
\end{align*}
\caption{Translation function from GCC to \ED (excerpt)}
\label{GCCtoDCC}
\end{figure}
\begin{theorem} \label{GCCtoDCCe}
If $ \Gamma  \vdash  \ottnt{a}  :  \ottnt{A} $ in GCC($ \mathcal{L} $), then $  \overline{  \Gamma  }   \vdash   \overline{ \ottnt{a} }   :   \overline{ \ottnt{A} }  $ in \ED{}($ \mathcal{L} $). Further, if $ \Gamma  \vdash  \ottnt{a_{{\mathrm{1}}}}  :  \ottnt{A} $ and $ \Gamma  \vdash  \ottnt{a_{{\mathrm{2}}}}  :  \ottnt{A} $ such that $\ottnt{a_{{\mathrm{1}}}} \equiv \ottnt{a_{{\mathrm{2}}}}$ in GCC($ \mathcal{L} $), then $ \overline{ \ottnt{a_{{\mathrm{1}}}} }  \simeq  \overline{ \ottnt{a_{{\mathrm{2}}}} } $ in \ED{}($ \mathcal{L} $).
\end{theorem} 
Earlier, we showed that DCC is a graded monadic calculus by translating GMC into it. However, we couldn't translate DCC into GMC because DCC has some comonadic character to it. Thereafter, we designed a graded comonadic calculus, GCC. Using GCC, we put the comonadic character of DCC to test. We could not translate GCC to DCC, thus showing that DCC is not fully comonadic. Next, we found that such a translation is not possible only because DCC does not allow certain derivations that are both sound and desirable. We extended DCC to \ED{} to allow these derivations and found that we can translate GCC (and GMC) into \ED{}.  Now, can we go the other way around and translate \ED{} into a calculus built up using just GMC and GCC?

%% file: Sections/DependencyGraded.tex
\section{Graded Monadic Comonadic Calculus} \label{secgmcc}

\subsection{The Calculus}

The Graded Monadic Comonadic Calculus (GMCC) combines the Graded Monadic Calculus (GMC) and the Graded Comonadic Calculus (GCC) into a single system. We can view it as an extension of the standard simply-typed $\lambda$-calculus with a graded type constructor $S_{m}$, which behaves both like a graded monadic type constructor, $T_{m}$, and a graded comonadic type constructor, $D_{m}$. The calculus has as terms the union of those of GMC and GCC. The typing rules of the calculus include the rules of GMC and GCC (shown in Figures \ref{typGMC} and \ref{typGCC} respectively) with $S_{m}$ replacing $T_{m}$ and $D_{m}$.

The equational theory of the calculus is generated by the equational theories of GMC and GCC (presented in Figures \ref{eqGMC} and \ref{eqGCC} respectively) along with the following additional rules: $ \mathbf{extr} \:   (   \ottkw{ret}  \:  \ottnt{a}   )   \equiv \ottnt{a}$ and $ \ottkw{ret}  \:   (   \mathbf{extr} \:  \ottnt{a}   )   \equiv \ottnt{a}$ and $ \mathbf{fork}^{ m_{{\mathrm{1}}} , m_{{\mathrm{2}}} }   (   \mathbf{join}^{ m_{{\mathrm{1}}} , m_{{\mathrm{2}}} }  \ottnt{a}   )   \equiv \ottnt{a}$ and $ \mathbf{join}^{ m_{{\mathrm{1}}} , m_{{\mathrm{2}}} }   (   \mathbf{fork}^{ m_{{\mathrm{1}}} , m_{{\mathrm{2}}} }  \ottnt{a}   )   \equiv \ottnt{a}$. These additional rules ensure that the monadic and the comonadic fragments of the calculus behave well with respect to one another. 
\subsection{Categorical Model} \label{GMCCModel}

GMCC enjoys a nice categorical model, as we show next.\\ We interpret the graded type constructor of the calculus as a kind of strong monoidal functor \citep{maclane}. A strong monoidal functor from a monoidal category $(M,\otimes_M,1_M)$ to a monoidal category $(N,\otimes_N,1_N)$ is a lax monoidal functor $(F,F_2,F_0) : (M,\otimes_M,1_M) \to (N,\otimes_N,1_N)$ where $F_0$ and $F_2 (X , Y)$ are invertible for all $X, Y \in \text{Obj}(M)$. Thus, for a strong monoidal functor $S : (M,\otimes_M,1_M) \to (N,\otimes_N,1_N)$, we have: $S (1_M) \cong 1_N$ and $S (X \otimes_M Y) \cong S(X) \otimes_N S(Y)$ for all $X, Y \in \text{Obj}(M)$. Note that if these isomorphisms are identities, then the functor is said to be strict. Further, note that the word `strong' in `strong monoidal functor' and in `strong endofunctor' refer to different properties.

Let $ \mathcal{M}  = (M, \_\cdot\_  , 1 ,  \leq )$ be the preordered monoid parametrizing the calculus. Let $\Ct$ be any bicartesian closed category. Let $\mathbf{S}$ be a strong monoidal functor from $\Ca({\mathcal{M}})$ to $\E^s$. Then $\mathbf{S}$ is both a strong $\mathcal{M}$-graded monad over $\Ct$ and a strong $\mathcal{M}$-graded comonad over $\Ct$. With regard to $\mathbf{S}$, let $\mu,\eta,\delta,\epsilon$ denote the corresponding natural transformations. Then,
\begin{align*}
\epsilon \circ \eta & = \text{id} & \eta \circ \epsilon & = \text{id} \\
\delta \circ \mu & = \text{id} & \mu \circ \delta & = \text{id}
\end{align*} 

We interpret $ S_{ m } \:  \ottnt{A} $ as: $ \llbracket   S_{ m } \:  \ottnt{A}   \rrbracket  = \mathbf{S}_{m}  \llbracket  \ottnt{A}  \rrbracket $. The terms are interpreted as in Figures \ref{intGMC} and \ref{intGCC}.\\ This gives us a sound interpretation of the calculus. 
\begin{theorem} \label{gmccsound}
If $ \Gamma  \vdash  \ottnt{a}  :  \ottnt{A} $ in GMCC, then $ \llbracket  \ottnt{a}  \rrbracket  \in \text{Hom}_{\Ct} ( \llbracket  \Gamma  \rrbracket ,  \llbracket  \ottnt{A}  \rrbracket )$. Further, if $ \Gamma  \vdash  \ottnt{a_{{\mathrm{1}}}}  :  \ottnt{A} $ and $ \Gamma  \vdash  \ottnt{a_{{\mathrm{2}}}}  :  \ottnt{A} $ such that $\ottnt{a_{{\mathrm{1}}}} \equiv \ottnt{a_{{\mathrm{2}}}}$ in GMCC, then $ \llbracket  \ottnt{a_{{\mathrm{1}}}}  \rrbracket  =  \llbracket  \ottnt{a_{{\mathrm{2}}}}  \rrbracket  \in \text{Hom}_{\Ct} ( \llbracket  \Gamma  \rrbracket ,  \llbracket  \ottnt{A}  \rrbracket )$.
\end{theorem} 

The theorem above shows that given a preordered monoid $\mathcal{M}$, any bicartesian closed category $\Ct$ together with a strong monoidal functor from $\Ca({\mathcal{M}})$ to $\E^s$ provides a sound model for GMCC($ \mathcal{M} $). In addition to soundness, GMCC($ \mathcal{M} $) also enjoys completeness with respect to its class of categorical models. Formally, we can show:

\begin{theorem} \label{gmcccomplete}
Given any preordered monoid $ \mathcal{M} $, for typing derivations $ \Gamma  \vdash  \ottnt{a_{{\mathrm{1}}}}  :  \ottnt{A} $ and $ \Gamma  \vdash  \ottnt{a_{{\mathrm{2}}}}  :  \ottnt{A} $ in GMCC($ \mathcal{M} $), if $ \llbracket  \ottnt{a_{{\mathrm{1}}}}  \rrbracket  =  \llbracket  \ottnt{a_{{\mathrm{2}}}}  \rrbracket $ in all models of GMCC($ \mathcal{M} $), then $\ottnt{a_{{\mathrm{1}}}} \equiv \ottnt{a_{{\mathrm{2}}}}$ is derivable in GMCC($ \mathcal{M} $).
\end{theorem}

We use entirely standard term-model techniques \citep{jacob} to prove completeness. First, we construct the classifying category and thereafter, the generic model in the classifying category. The generic model equates only the terms that are equal in the calculus. So, if the interpretations of two GMCC($ \mathcal{M} $)-terms are equal in all models (and therefore in the generic model too), then these terms are equal in the calculus as well. Further, we can also show that:

\begin{theorem}\label{gmccgeneric}
The generic model satisfies the universal property.
\end{theorem}

The above theorem implies that any model of GMCC($ \mathcal{M} $) can be factored through the generic model. We don't explore the consequences of this theorem here, but leave it for future work. 
 
In this section, we have seen that GMCC is sound and complete with respect to its class of categorical models. In the next section, we explore the relation between GMCC and \ED{}.

\section{GMCC and \ED{}} \label{secgmccdcce}

We saw that over the class of bounded join-semilattices, we can translate both GMC and GCC into \ED{}. In fact, over the same class of structures, we can go further and translate GMCC into \ED{} following the translations presented in Figures \ref{GMCtoDCC} and \ref{GCCtoDCC}.
\begin{theorem} \label{DCCComplete}
If $ \Gamma  \vdash  \ottnt{a}  :  \ottnt{A} $ in GMCC($ \mathcal{L} $), then $  \overline{  \Gamma  }   \vdash   \overline{ \ottnt{a} }   :   \overline{ \ottnt{A} }  $ in \ED{}($ \mathcal{L} $). Further, if $ \Gamma  \vdash  \ottnt{a_{{\mathrm{1}}}}  :  \ottnt{A} $ and $ \Gamma  \vdash  \ottnt{a_{{\mathrm{2}}}}  :  \ottnt{A} $ such that $\ottnt{a_{{\mathrm{1}}}} \equiv \ottnt{a_{{\mathrm{2}}}}$ in GMCC($ \mathcal{L} $), then $ \overline{ \ottnt{a_{{\mathrm{1}}}} }  \simeq  \overline{ \ottnt{a_{{\mathrm{2}}}} } $ in \ED{}($ \mathcal{L} $).
\end{theorem}
We now go the other way around and translate \ED{}($ \mathcal{L} $) into GMCC($ \mathcal{L} $).

Note here that \ED{} has a very liberal definition of equality, inherited from DCC. Two \ED{} terms are equal if, after erasure, they are equal as $\lambda$-terms. So $  \mathbf{eta} ^{ \ell }  \ottnt{a}   \simeq  \ottnt{a} $ for any \ED{}-term $\ottnt{a}$. The same is not true in general in GMCC. For example, $ \mathbf{up}^{   \bot   ,  \ell  }   (   \ottkw{ret}  \:  \ottnt{a}   )  $ may not be equal to $\ottnt{a}$. To capture the notion of \ED{}-equality in GMCC, we need to define a similar relation between GMCC terms. For GMCC terms $\ottnt{a_{{\mathrm{1}}}}$ and $\ottnt{a_{{\mathrm{2}}}}$, we say $ \ottnt{a_{{\mathrm{1}}}}  \simeq  \ottnt{a_{{\mathrm{2}}}} $ if and only if $ \lfloor  \ottnt{a_{{\mathrm{1}}}}  \rfloor $ and $ \lfloor  \ottnt{a_{{\mathrm{2}}}}  \rfloor $, the plain $\lambda$-term counterparts of $\ottnt{a_{{\mathrm{1}}}}$ and $\ottnt{a_{{\mathrm{2}}}}$ respectively, are equal as $\lambda$-terms. The erasure operation $\lfloor-\rfloor$  strips away the constructors $\mathbf{ret},\mathbf{join}, \mathbf{extr}, \mathbf{fork}, \mathbf{lift}$ and $\mathbf{up}$ along with the grade annotations. For example, $ \lfloor   \mathbf{join}^{  \ell_{{\mathrm{1}}}  ,  \ell_{{\mathrm{2}}}  }  \ottnt{a}   \rfloor  =  \lfloor  \ottnt{a}  \rfloor $. 

Coming back to the translation from \ED{} to GMCC, it is straightforward for types. The constructor $\mathcal{T}_{\ell}$ gets translated to $S_{\ell}$. The translation for terms requires the following lemma. We use $\underline{\phantom{a}}$ to denote the translation function.
\begin{lemma} \label{protectId}
If $ \ell  \sqsubseteq  \ottnt{B} $, then there exists a term $  \emptyset   \vdash  \ottnt{j}  :    S_{  \ell  } \:    \underline{ \ottnt{B} }     \to   \underline{ \ottnt{B} }   $ such that $ \lfloor  \ottnt{j}  \rfloor  \equiv  \lambda  \ottmv{x}  .  \ottmv{x} $.
\end{lemma}
The above lemma is key to the translation. Note that in Section \ref{JustMonadic}, while trying to translate DCC to GMC, we could not prove a lemma like this one. 
Now, with the above lemma, we can translate \ED{} terms to GMCC terms, as shown below. The function $\ottnt{j}$ used in the translation is as given by Lemma \ref{protectId}.
\[  \underline{  \mathbf{eta} ^{ \ell }  \ottnt{a}  }   =  \mathbf{up}^{   \bot   ,  \ell  }   (   \ottkw{ret}  \:   \underline{ \ottnt{a} }    )   \hspace*{30pt}  \underline{  \mathbf{bind} ^{ \ell } \:  \ottmv{x}  =  \ottnt{a}  \: \mathbf{in} \:  \ottnt{b}  }   =  \ottnt{j}  \:   (    (   \mathbf{lift}^{  \ell  }   (   \lambda  \ottmv{x}  .   \underline{ \ottnt{b} }    )    )   \:   \underline{ \ottnt{a} }    )   \]
This translation preserves typing and meaning. 
\begin{theorem} \label{DCCSound}
If $ \Gamma  \vdash  \ottnt{a}  :  \ottnt{A} $ in \ED{}($ \mathcal{L} $), then $  \underline{  \Gamma  }   \vdash   \underline{ \ottnt{a} }   :   \underline{ \ottnt{A} }  $ in GMCC($ \mathcal{L} $). Further, if $ \Gamma  \vdash  \ottnt{a_{{\mathrm{1}}}}  :  \ottnt{A} $ and $ \Gamma  \vdash  \ottnt{a_{{\mathrm{2}}}}  :  \ottnt{A} $ such that $\ottnt{a_{{\mathrm{1}}}} \simeq \ottnt{a_{{\mathrm{2}}}}$ in \ED{}($ \mathcal{L} $), then $ \underline{ \ottnt{a_{{\mathrm{1}}}} }  \simeq  \underline{ \ottnt{a_{{\mathrm{2}}}} } $ in GMCC($ \mathcal{L} $).
\end{theorem}
Theorems \ref{DCCSound} and \ref{DCCComplete} together show that over the class of bounded join-semilattices, GMCC and \ED{}, seen as dependency calculi, are equivalent. Thus, GMCC is a generalization of the Dependency Core Calculus, DCC. Hence, dependency analysis, at least to the extent DCC is capable of, can be done using just a graded monadic comonadic calculus. This connection of dependency analysis to GMCC is important because:
\begin{itemize}
\item Dependency analysis can now benefit from a wider variety of categorical models. Some of these models may provide simpler proofs of correctness. In fact, using our categorical models, we show, in a straightforward manner, that dependency analyses in \ED{} and \lc{} are correct.
\item More dependency calculi can now be unified under a common framework. As a proof of concept, we show that the binding-time analysing calculus \lc{} of \citet{lambdacirc} can be encoded into \Ge{}, an extension of GMCC. Note that \lc{} can not be translated into DCC \citep{dcc}.
\item The non-standard \textbf{bind}-rule of DCC can be replaced with standard monadic and comonadic typing rules. This finding provides insight into the categorical basis of the \textbf{bind}-rule of DCC. (See Section \ref{modeldcce})
\item GMCC is formed combining GMC and its dual, GCC. GMC can be seen as a restriction of the Explicit Subeffecting Calculus of \citet{katsumata}. The Explicit Subeffecting Calculus is a general system for analysing effects. This clean connection between dependency analysis and effect analysis promises to be a fertile ground for new ideas, especially in the intersection of dependency, effect and coeffect analyses. 
\end{itemize}

Before closing this section, we want to make some remarks on the equivalence between GMCC and \ED{}. From theorems \ref{DCCSound} and \ref{DCCComplete}, we see that GMCC and \ED{} are equivalent upto erasure. We may ask: can this equivalence be made stronger? The answer is yes. In the next section, we shall show that (over the class of bounded join-semilattices) GMCC and \ED{} are semantically equivalent too. We can also show something stronger on the syntactic side as well:

\begin{theorem} \label{GMCCround}
Let $ \Gamma  \vdash  \ottnt{a}  :  \ottnt{A} $ be any derivation in GMCC($ \mathcal{L} $). Then, $  \underline{  \overline{ \ottnt{a} }  }   \equiv  \ottnt{a} $.
\end{theorem}

The above theorem says that any GMCC($ \mathcal{L} $)-term, after a round trip to \ED{}($ \mathcal{L} $), is equal (not just equal upto erasure) to itself. When going the other way, we can prove the following weaker result: if $ \Gamma  \vdash  \ottnt{a}  :  \ottnt{A} $ is a derivation in \ED{}($ \mathcal{L} $), then $ \overline{  \underline{ \ottnt{a} }  }  \simeq a $. Note that we are forced to use equality upto erasure here because we don't have an alternative equational theory for \ED{}. 
 
Next, we develop the categorical semantics for \ED{}. 

\section{\ED{}: Categorical Semantics} \label{nonint}

\citet{dcc} provide a categorical model for DCC and prove noninterference using that model. Several other authors \citep{tse-zdancewic,igarashi,ahmed,algehed} have provided alternative proofs of noninterference for DCC using various techniques, including parametricity. In this section, we present a class of categorical models for \ED{}, in the style of GMCC, and show noninterference for the calculus using these models. We also establish semantic equivalence between \ED{} and GMCC and explain the non-standard \textbf{bind}-rule of DCC in terms of standard category-theoretic concepts.

\subsection{Categorical Models for \ED{}} \label{modeldcce}

Given that GMCC and \ED{}, considered as dependency calculi, are equivalent, we can simply use models of GMCC to interpret \ED{}. However, in this section, we build models for \ED{} from first principles and show them to be computationally adequate with respect to a call-by-value semantics. The original operational semantics of DCC, presented by \citet{dcc}, is somewhat ad hoc from a categorical perspective because according to this semantics, $ \mathbf{eta} ^{ \ell }  \ottnt{a} $ converts to $\ottnt{a}$. If we interpret $\mathbf{eta}$ as the unit of a monad, this conversion would require us to interpret the unit as the identity natural transformation, something that is not very general. On the other hand, a call-by-value semantics or a call-by-name semantics of the calculus \citep{tse-zdancewic} is quite general from a category theoretic perspective. According to a call-by-value semantics, 
$\mathbf{bind}$-expressions convert in the following manner:
\drules[CBV]{$ \vdash  \ottnt{a}  \leadsto  \ottnt{a'} $}{Operational Semantics}{BindLeft,BindBeta} 

Given $\leadsto$, we can define the multistep reduction relation, $\leadsto^{\ast}$, in the usual way. A point to note here is that though we use a call-by-value semantics for \ED{}, we could have used a call-by-name semantics as well. \ED{} being a terminating calculus, the choice of one evaluation strategy over the other does not lead to significant differences in metatheory. We chose call-by-value over call-by-name because the former allows reductions underneath the $\mathbf{eta}$s.



Given a bounded join-semilattice, $ \mathcal{L}  = (L,\vee,\bot)$, and any bicartesian closed category $\Ct$, we interpret the graded type constructor $\mathcal{T}$ of \ED{}($ \mathcal{L} $) as a strong monoidal functor $\mathbf{S}$ from $\Ca( \mathcal{L} )$ to $\ES$. Formally, $ \llbracket   \mathcal{T}_{ \ell } \:  \ottnt{A}   \rrbracket  =  \mathbf{S}_{\ell}  \llbracket  \ottnt{A}  \rrbracket $. Note that since $\vee$ is idempotent, the triple $(\mathbf{S}_{\ell},\mathbf{S}^{  \bot   \sqsubseteq  \ell } \circ \eta,\mu^{\ell,\ell})$, for any $\ell \in L$, is a monad. Further, since $\mu^{\ell,\ell}$ is invertible, such a monad is also idempotent. 

For interpreting terms, we need the following lemma.
\begin{lemma}\label{lemmak}
If $ \ell  \sqsubseteq  \ottnt{A} $, then $\exists$ an isomorphism $k : \mathbf{S}_{\ell}  \llbracket  \ottnt{A}  \rrbracket  \to  \llbracket  \ottnt{A}  \rrbracket $. \\ Further, $k \circ \overline{\eta} = id_{ \llbracket  \ottnt{A}  \rrbracket }$ and $\overline{\eta} \circ k =  \text{id}_{  \mathbf{S}_{  \ell  }   \llbracket  \ottnt{A}  \rrbracket   } $ where $\overline{\eta} \triangleq \mathbf{S}^{  \bot   \sqsubseteq  \ell } \circ \eta :  \llbracket  \ottnt{A}  \rrbracket  \to \mathbf{S}_{\ell}  \llbracket  \ottnt{A}  \rrbracket $.
\end{lemma} 


Using the above lemma, we interpret terms as:
\begin{align*}
 \llbracket   \mathbf{eta} ^{ \ell }  \ottnt{a}   \rrbracket  & = \mathbf{S}^{  \bot   \sqsubseteq  \ell } \circ \eta \circ  \llbracket  \ottnt{a}  \rrbracket  \\
 \llbracket   \mathbf{bind} ^{ \ell } \:  \ottmv{x}  =  \ottnt{a}  \: \mathbf{in} \:  \ottnt{b}   \rrbracket  & = k \circ \mathbf{S}_{\ell}  \llbracket  \ottnt{b}  \rrbracket  \circ \text{t}^{\mathbf{S}_{\ell}} \circ \langle \I ,  \llbracket  \ottnt{a}  \rrbracket  \rangle
\end{align*}
For later reference note that whenever we need to be precise, we use $\llbracket \_ \rrbracket_{(\Ct,\mathbf{S})}$ to refer to the interpretation of \ED($ \mathcal{L} $) in category $\Ct$ using $\mathbf{S}$. The above interpretation of \ED{} is sound, as we see next.
\begin{theorem}\label{dcceWD}
If $ \Gamma  \vdash  \ottnt{a}  :  \ottnt{A} $ in \ED{}, then $ \llbracket  \ottnt{a}  \rrbracket  \in \text{Hom}_{\Ct} ( \llbracket  \Gamma  \rrbracket  ,  \llbracket  \ottnt{A}  \rrbracket )$.
\end{theorem}
\begin{theorem}\label{dcceOE}
If $\Gamma \vdash a : A$ in \ED{} and $ \vdash  \ottnt{a}  \leadsto  \ottnt{a'} $, then $\llbracket a \rrbracket = \llbracket a’ \rrbracket$.
\end{theorem}

Computational adequacy with respect to call-by-value operational semantics follows as a corollary of the above theorem. Below, we assume that the categorical interpretation is injective for ground types. In particular, $ \llbracket   \mathbf{true}   \rrbracket  \neq  \llbracket   \mathbf{false}   \rrbracket $, where $ \mathbf{true}  \triangleq  \mathbf{inj}_1 \:   \ottkw{unit}   :  \mathbf{Bool} $ and $ \mathbf{false}  \triangleq  \mathbf{inj}_2 \:   \ottkw{unit}   :  \mathbf{Bool} $, where $ \mathbf{Bool}  \triangleq   \ottkw{Unit}   +   \ottkw{Unit}  $.

\begin{theorem} \label{dcceCA}
Let the interpretation $\llbracket \_ \rrbracket_{(\Ct,\mathbf{S})}$ be injective for ground types.
\begin{itemize}
\item Let $ \Gamma  \vdash  \ottnt{b}  :   \mathbf{Bool}  $ and $\ottmv{v}$ be a value of type $ \mathbf{Bool} $. If $ \llbracket  \ottnt{b}  \rrbracket _{(\Ct,\mathbf{S})} =  \llbracket  \ottmv{v}  \rrbracket _{(\Ct,\mathbf{S})}$, then $ \vdash  \ottnt{b}  \leadsto^{\ast}  \ottmv{v} $. \vspace*{3pt}
\item Fix some $\ell \in L$. Let $ \Gamma  \vdash  \ottnt{b}  :   \mathcal{T}_{ \ell } \:   \mathbf{Bool}   $ and $\ottmv{v}$ be a value of type $ \mathcal{T}_{ \ell } \:   \mathbf{Bool}  $. Suppose, the morphisms $\overline{\eta}_X \triangleq   \mathbf{S}^{   \bot    \sqsubseteq   \ell  }_{ \ottnt{X} }   \circ   \eta_{ \ottnt{X} }  $ are mono for any $X \in \text{Obj} (\Ct)$.  Now, if $ \llbracket  \ottnt{b}  \rrbracket _{(\Ct,\mathbf{S})} =  \llbracket  \ottmv{v}  \rrbracket _{(\Ct,\mathbf{S})}$, then $ \vdash  \ottnt{b}  \leadsto^{\ast}  \ottmv{v} $. 
\end{itemize}
\end{theorem}

Now that GMCC($ \mathcal{L} $) and \ED{}($ \mathcal{L} $) enjoy the same class of models, we can show the two calculi are exactly equivalent over these models.

\begin{theorem}\label{gmccdcce}
Let $ \mathcal{L} $ be a bounded join-semilattice.
\begin{itemize}
\item If $ \Gamma  \vdash  \ottnt{a}  :  \ottnt{A} $ in GMCC($ \mathcal{L} $), then $ \llbracket   \overline{ \ottnt{a} }   \rrbracket  =  \llbracket  \ottnt{a}  \rrbracket $.
\item If $ \Gamma  \vdash  \ottnt{a}  :  \ottnt{A} $ in \ED{}($ \mathcal{L} $), then $ \llbracket   \underline{ \ottnt{a} }   \rrbracket  =  \llbracket  \ottnt{a}  \rrbracket $.
\end{itemize}
\end{theorem}

The categorical semantics of \ED{} help in understanding the non-standard \textbf{bind}-rule of DCC. By lemma \ref{lemmak}, if $ \ell  \sqsubseteq  \ottnt{A} $, then $ \mathbf{S}_{  \ell  }   \llbracket  \ottnt{A}  \rrbracket   \cong  \llbracket  \ottnt{A}  \rrbracket $. In fact, if $ \ell  \sqsubseteq  \ottnt{A} $, then $ \llbracket  \ottnt{A}  \rrbracket $ is the carrier of an $\mathbf{S}_{\ell}$-algebra. What this means is that the protection judgement is a syntactic mechanism for picking out the appropriate monad algebras. This insight explains the signature of \rref{DCC-Bind}: $  \mathcal{T}_{ \ell } \:  \ottnt{A}   \to   (   \ottnt{A}  \to  \ottnt{B}   )   \to \{  \ell  \sqsubseteq  \ottnt{B}  \} \to B$. If $ \ell  \sqsubseteq  \ottnt{B} $, then $ \llbracket  \ottnt{B}  \rrbracket $ is the carrier of an $\mathbf{S}_{\ell}$-algebra. As such, the return type of the rule can be $\ottnt{B}$, in lieu of $ \mathcal{T}_{ \ell } \:  \ottnt{B} $. The following theorem characterizes the protection judgement in terms of monad algebras.

\begin{theorem} \label{dcceEM}
If $ \ell  \sqsubseteq  \ottnt{A} $ in \ED{}, then $( \llbracket  \ottnt{A}  \rrbracket , k)$ is an $\mathbf{S}_{\ell}$-algebra.\\
Further, if $ \ell  \sqsubseteq  \ottnt{A} $ and $ \ell  \sqsubseteq  \ottnt{B} $, then for any $f \in \text{Hom}_{\Ct} ( \llbracket  \ottnt{A}  \rrbracket  ,  \llbracket  \ottnt{B}  \rrbracket )$, $f$ is an $\mathbf{S}_{\ell}$-algebra morphism.\\
Hence, the full subcategory of $\Ct{}$ with $\text{Obj} := \{  \llbracket  \ottnt{A}  \rrbracket  \: | \:  \ell  \sqsubseteq  \ottnt{A}  \}$ is also a full subcategory of the Eilenberg-Moore category, $\Ct{}^{\mathbf{S}_{\ell}}$.
\end{theorem}

Next, we use these models to prove noninterference for \ED{}.

\subsection{Proof of Noninterference}

Two functors in $\ES$ are crucial to our proof of non-interference. One of them is the identity functor, $\Id$. The other is the terminal functor, denoted by $\ast$, the functor which maps all objects to $\top$, the terminal object of the category and all morphisms to $\langle \rangle$. Now, for every $\ell \in L$, we define a strong monoidal functor $\mathbf{S}^\ell$ from $\Ca( \mathcal{L} )$ to $\ES{}$ as follows.
\begin{align*}
\mathbf{S}^{\ell}(\ell') = \begin{cases} 
                            \Id , & \text{ if }  \ell'  \sqsubseteq  \ell  \hspace{10pt}\\
                            \ast , & \text{ otherwise } \hspace*{10pt}
                            \end{cases}
\mathbf{S}^{\ell}( \ell_{{\mathrm{1}}}  \sqsubseteq  \ell_{{\mathrm{2}}} ) & = \begin{cases}
                            \I , & \text{ if }  \ell_{{\mathrm{2}}}  \sqsubseteq  \ell  \\
                            \langle \rangle , & \text{ otherwise } 
                            \end{cases} 
\end{align*}
The following points are worth noting here.
\begin{itemize}
\item $\mathbf{S}^{\ell}( \bot ) = \Id$ for any $\ell \in L$. Then, for every $\mathbf{S}^{\ell}$, $\eta = \epsilon = \I$ .
\item Fix some $\ell_{{\mathrm{0}}} \in L$. Now, for any $\ell_{{\mathrm{1}}}, \ell_{{\mathrm{2}}} \in L$, we have, $\mathbf{S}^{\ell_{{\mathrm{0}}}}(\ell_{{\mathrm{1}}}) \circ \mathbf{S}^{\ell_{{\mathrm{0}}}}(\ell_{{\mathrm{2}}}) = \mathbf{S}^{\ell_{{\mathrm{0}}}}( \ell_{{\mathrm{1}}}  \vee  \ell_{{\mathrm{2}}} )$. Then, $\mu^{\ell_{{\mathrm{1}}},\ell_{{\mathrm{2}}}} = \delta^{\ell_{{\mathrm{1}}},\ell_{{\mathrm{2}}}} = \I$. Therefore, the $\mathbf{S}^{\ell}$s are all strict monoidal functors.
\end{itemize}
Now, for any bicartesian closed category $\Ct$, any strong monoidal functor from $\Ca( \mathcal{L} )$ to $\ES{}$ provides a computationally adequate interpretation of \ED{}($ \mathcal{L} $), provided the interpretation for  ground types is injective. As such, 
\begin{theorem}\label{dcceSlA}
$\llbracket \_ \rrbracket_{(\mathbf{Set},\mathbf{S}^{\ell})}$, for any $\ell \in L$, is a computationally adequate interpretation of \ED{}($ \mathcal{L} $).
\end{theorem}
Next, we explain the intuition behind these strong monoidal functors. $\mathbf{S}^{\ell}$ keeps untouched all information at levels $\ell'$ where $ \ell'  \sqsubseteq  \ell $ but blacks out all information at every other level. So, $\mathbf{S}^{\ell}$ corresponds to the view of an observer at level $\ell$. We can formalize this intuition. Suppose $\neg ( \ell''  \sqsubseteq  \ell )$, i.e. $\ell$ should not depend upon $\ell''$. Then, if $ \ell''  \sqsubseteq  \ottnt{A} $, the terms of type $A$ should not be visible to an observer at level $\ell$. In other words, if $ \ell''  \sqsubseteq  \ottnt{A} $, $\mathbf{S}^{\ell}$ should black out all information from type $A$. This is indeed the case, as we see next.
\begin{lemma}\label{dcceSlObs}
If $ \ell''  \sqsubseteq  \ottnt{A} $ and $\neg ( \ell''  \sqsubseteq  \ell )$, then $ \llbracket  \ottnt{A}  \rrbracket _{(\Ct, \mathbf{S}^{  \ell  } )} \cong \top$. 
\end{lemma}
The above lemma takes us to our noninterference theorem. Recall the test of correctness for DCC from Section \ref{secdcc}: if $ \ell  \sqsubseteq  \ottnt{A} $ and $\neg( \ell  \sqsubseteq  \ell' )$, then the terms of $\ottnt{A}$ are not be visible at $\ell'$. We prove correctness for \ED{} by formulating this test as:
\begin{theorem} \label{dcceNI}
Let $ \mathcal{L}  = (L,\vee,\bot)$ be the parametrizing semilattice.
\begin{itemize}
\item Suppose $\ell \in L$ such that $\neg ( \ell  \sqsubseteq   \bot  )$. Let $ \ell  \sqsubseteq  \ottnt{A} $. Let $  \emptyset   \vdash  \ottnt{f}  :   \ottnt{A}  \to   \mathbf{Bool}   $ and $  \emptyset   \vdash  \ottnt{a_{{\mathrm{1}}}}  :  \ottnt{A} $ and $  \emptyset   \vdash  \ottnt{a_{{\mathrm{2}}}}  :  \ottnt{A} $. Then, $ \vdash   \ottnt{f}  \:  \ottnt{a_{{\mathrm{1}}}}   \leadsto^{\ast}  \ottmv{v} $ if and only if $ \vdash   \ottnt{f}  \:  \ottnt{a_{{\mathrm{2}}}}   \leadsto^{\ast}  \ottmv{v} $, where $\ottmv{v}$ is a value of type $ \mathbf{Bool} $.

\vspace*{3pt}

\item Suppose $\ell , \ell' \in L$ such that $\neg ( \ell  \sqsubseteq  \ell' )$. Let $ \ell  \sqsubseteq  \ottnt{A} $. Let $  \emptyset   \vdash  \ottnt{f}  :   \ottnt{A}  \to   \mathcal{T}_{ \ell' } \:   \mathbf{Bool}    $ and $  \emptyset   \vdash  \ottnt{a_{{\mathrm{1}}}}  :  \ottnt{A} $ and $  \emptyset   \vdash  \ottnt{a_{{\mathrm{2}}}}  :  \ottnt{A} $. Then, $ \vdash   \ottnt{f}  \:  \ottnt{a_{{\mathrm{1}}}}   \leadsto^{\ast}  \ottmv{v} $ if and only if $ \vdash   \ottnt{f}  \:  \ottnt{a_{{\mathrm{2}}}}   \leadsto^{\ast}  \ottmv{v} $, where $\ottmv{v}$ is a value of type $ \mathcal{T}_{ \ell' } \:   \mathbf{Bool}  $.
\end{itemize}
\end{theorem}
As corollary of the above theorem, we can show that $  \emptyset   \vdash  \ottnt{f}  :    \mathcal{T}_{  \mathbf{Secret}  } \:  \ottnt{A_{{\mathrm{1}}}}   \to   \mathbf{Bool}   $ and $  \emptyset   \vdash  \ottnt{f'}  :    (   \ottnt{A_{{\mathrm{1}}}}  \to   \mathcal{T}_{  \mathbf{Secret}  } \:  \ottnt{A_{{\mathrm{2}}}}    )   \to   \mathbf{Bool}   $, for all types $\ottnt{A_{{\mathrm{1}}}}$ and $\ottnt{A_{{\mathrm{2}}}}$, are constant functions. We can also show that $  \emptyset   \vdash  \ottnt{f''}  :    \mathcal{T}_{ \ell_{{\mathrm{1}}} } \:  \ottnt{A}   \to   \mathcal{T}_{ \ell_{{\mathrm{2}}} } \:   \mathbf{Bool}    $, for any type $\ottnt{A}$, is a constant function, whenever $\neg( \ell_{{\mathrm{1}}}  \sqsubseteq  \ell_{{\mathrm{2}}} )$.

We use the $\mathbf{S}^{\ell}$s to prove the noninterference theorem above. This technique relies on the observation that for two entities $E_1$ and $E_2$, if $E_1$ can be present when $E_2$ is absent, then $E_1$ \textit{does not depend} upon $E_2$. We call this technique the \textit{presence-absence test}. In the next section, we shall use the same technique to prove correctness of binding-time analysis in \lc{}.

\section{Binding-Time Calculus, \lc{}} \label{lcirc}

\lc{} \citep{lambdacirc} is a foundational calculus for binding-time analysis, lying at the heart of state-of-the-art metaprogramming languages like MetaOCaml \citep{meta}. \lc{} is essentially a dependency calculus that ensures early stage computations do not depend upon later stage ones. One might expect that DCC, being a core calculus of dependency, would subsume \lc{}. However, \citet{dcc} noted that \lc{} cannot be translated into DCC. One reason behind this shortcoming is that DCC does not fully utilize the power of comonadic aspect of dependency analysis, as we discussed before. We extended DCC to \ED{} to include the comonadic aspect of dependency analysis. This extension opens up the possibility of \lc{} being translated to \ED{}. In this section, we explore this possibility. We first review the calculus \lc{}, thereafter present a categorical model leading to an alternative proof of correctness and finally show how we can translate \lc{} into our graded monadic comonadic system.

\subsection{The Calculus \lc{}}

\lc{} is simply-typed $\lambda$-calculus extended with a `next time' type constructor, $\bc$. Intuitively, $ \bigcirc  \ottnt{A} $ is the type of terms to be computed upon the `next time'. The calculus models staged computation, with an earlier stage manipulating programs from later stage as data. For a time-ordered normalization, the calculus needs to ensure that computation from an earlier stage \textit{does not depend} upon computation from a later stage. To model such a notion of independence of the past from the future, \citet{lambdacirc} uses temporal logic. 
 In \lc{}, time is discretized as instants or moments, represented by natural numbers. For example, $0$ denotes the present moment, $0'$ denotes the next moment and so on. 
We now look at the calculus formally, as presented by \citet{lambdacirc}.

The grammar of \lc{} appears in Figure \ref{lcgrammar}, typing rules in Figure \ref{lctyping} and the equational theory in Figure \ref{lcequality}. The typing judgement $ \Gamma  \vdash  \ottnt{a}  :^{ n }  \ottnt{A} $ intuitively means that $\ottnt{a}$ is available at time instant $n$, provided the variables in $\Gamma$ are available at their respective time instants. Note that \lc{} does not have sum types; we include them here for the sake of having non-trivial ground types.
\begin{figure}
\begin{align*}
\text{Types}, \ottnt{A} , \ottnt{B} & ::=  K  \: | \:  \ottnt{A}  \to  \ottnt{B}  \: | \:  \bigcirc  \ottnt{A}  \: | \:  \ottkw{Unit}  \: | \:  \ottnt{A_{{\mathrm{1}}}}  +  \ottnt{A_{{\mathrm{2}}}}  \\
\text{Terms}, \ottnt{a}, \ottnt{b} & ::= \ottmv{x} \: | \:  \lambda  \ottmv{x}  :  \ottnt{A}  .  \ottnt{b}  \: | \:  \ottnt{b}  \:  \ottnt{a}  \: | \:  \mathbf{next} \:  \ottnt{a}  \: | \:  \mathbf{prev} \:  \ottnt{a}  \: | \:  \ottkw{unit}  \: | \:  \mathbf{inj}_1 \:  \ottnt{a_{{\mathrm{1}}}}  \: | \:  \mathbf{inj}_2 \:  \ottnt{a_{{\mathrm{2}}}}  \: | \:  \mathbf{case} \:  \ottnt{a}  \: \mathbf{of} \:  \ottnt{b_{{\mathrm{1}}}}  \: ; \:  \ottnt{b_{{\mathrm{2}}}}  \\
\text{Contexts}, \Gamma & ::=  \emptyset  \: | \:  \Gamma  ,   \ottmv{x}  :^{  n  }  \ottnt{A}  
\end{align*}
\caption{Grammar of \lc{}}
\label{lcgrammar}
\end{figure}
\begin{figure}
\drules[LC]{$ \Gamma  \vdash  \ottnt{a}  :^{ n }  \ottnt{A} $}{Typing}{Var,Lam,App,InjOne,Case,Next,Prev}
\caption{Typing Rules of \lc{} (Excerpt)}
\label{lctyping}
\end{figure} 
\begin{figure}
\begin{align*}
  (   \lambda  \ottmv{x}  :  \ottnt{A}  .  \ottnt{b}   )   \:  \ottnt{a}  & \equiv  \ottnt{b}  \{  \ottnt{a}  /  \ottmv{x}  \}  & \ottnt{b} & \equiv   \lambda  \ottmv{x}  :  \ottnt{A}  .  \ottnt{b}   \:  \ottmv{x}  \\
 \mathbf{prev} \:   (   \mathbf{next} \:  \ottnt{a}   )   & \equiv \ottnt{a} & \ottnt{a} & \equiv  \mathbf{next} \:   (   \mathbf{prev} \:  \ottnt{a}   )   \\
\mathbf{case} \: (\mathbf{inj}_i \: a_i)  \: \mathbf{of} \: \ottnt{b_{{\mathrm{1}}}} \: ; \: \ottnt{b_{{\mathrm{2}}}} & \equiv b_i \: a_i &  \ottnt{b}  \:  \ottnt{a}  & \equiv   \mathbf{case} \:  \ottnt{a}  \: \mathbf{of} \:    \lambda  \ottmv{x_{{\mathrm{1}}}}  .  \ottnt{b}   \:   (   \mathbf{inj}_1 \:  \ottmv{x_{{\mathrm{1}}}}   )    \: ; \:   \lambda  \ottmv{x_{{\mathrm{2}}}}  .  \ottnt{b}    \:   (   \mathbf{inj}_2 \:  \ottmv{x_{{\mathrm{2}}}}   )   
\end{align*}
\caption{Equality rules of \lc{} (Excerpt)}
\label{lcequality}
\end{figure}

With this background on \lc{}, let us now build categorical models for the calculus.

\subsection{Categorical Models for \lc{}} \label{seclcCat}

The motivation for categorical models of \lc{}, in the style of GMCC, comes from the observation that \rref{LC-Next,LC-Prev} are like \rref{C-Fork,M-Join} respectively. Here, we can think of $\mathcal{N} = (\mathbb{N},+,0)$ with discrete ordering to be the parametrizing preordered monoid. Then, $\bc$ is like  $S_{ 0' }$, where $S$ is the graded modal type constructor. \lc{} and GMCC($ \mathcal{N} $) share several similarities, but there is a crucial difference between the two calculi. In \lc{}, the types $ \bigcirc^{ n }   (   \ottnt{A}  \to  \ottnt{B}   )  $ and $  \bigcirc^{ n }  \ottnt{A}   \to   \bigcirc^{ n }  \ottnt{B}  $ (where $\bigcirc^n$ is the operator $\bigcirc$ applied $n$ times) are isomorphic whereas in GMCC($ \mathcal{N} $), the types $ S_{  n  } \:   (   \ottnt{A}  \to  \ottnt{B}   )  $ and $  S_{  n  } \:  \ottnt{A}   \to   S_{  n  } \:  \ottnt{B}  $ are not necessarily isomorphic. Owing to this difference, we need to modify our models in order to interpret \lc{}.  More precisely, unlike $S_n$, we cannot model $\bigcirc$ using any strong endofunctor, but require cartesian closed endofunctors. So next, we define the category of cartesian closed endofunctors. 

Let $\Ct$ be a cartesian closed category. An endofunctor $F : \Ct \to \Ct$ is said to be finite-product-preserving if and only if the morphisms $p_{\top} \triangleq \langle \rangle : F(\top) \to \top$ and $p_{X,Y} \triangleq \langle F \pi_1 , F \pi_2 \rangle : F (X \times Y) \to F X \times F Y$, for $X , Y \in \text{Obj}(\Ct)$, have inverses. A finite-product-preserving endofunctor $F : \Ct \to \Ct$ is said to be cartesian closed if and only if the morphisms $q_{X,Y} \triangleq \Lambda (F ( \text{app} ) \circ p^{-1}_{Y^X,X}) : F(Y^X) \to (F Y)^{(F X)}$, for $X , Y \in \text{Obj}(\Ct)$, have inverses. The cartesian closed endofunctors of $\Ct$, with natural transformations as morphisms, form a category, $\EC$. Like $\ES{}$, $\EC{}$ is also a strict monoidal category with the monoidal product and the identity object defined in the same way. We use $\EC{}$ to build models for \lc{}.
 
Let $\Ct$ be any bicartesian closed category. Let $\mathbf{S}$ be a strong monoidal functor from $\Ca(\mathcal{N})$ to $\EC$. Then, the interpretation $ \llbracket   \_   \rrbracket $, or more precisely $ \llbracket   \_   \rrbracket _{(\Ct,\mathbf{S})}$, of types and terms is as follows. The modal operator and contexts are interpreted as: 
\[  \llbracket   \bigcirc  \ottnt{A}   \rrbracket  = \mathbf{S}_{ 0' }  \llbracket  \ottnt{A}  \rrbracket  \hspace{20pt}  \llbracket   \emptyset   \rrbracket  = \top \hspace{20pt}  \llbracket    \Gamma  ,   \ottmv{x}  :^{  n  }  \ottnt{A}     \rrbracket  =  \llbracket  \Gamma  \rrbracket  \times \mathbf{S}_{n}  \llbracket  \ottnt{A}  \rrbracket  \] 
Terms are interpreted as:
\begin{align*}
 \llbracket  \ottmv{x}  \rrbracket  & =  \llbracket  \Gamma_{{\mathrm{1}}}  \rrbracket  \times \mathbf{S}_{n}  \llbracket  \ottnt{A}  \rrbracket  \times  \llbracket  \Gamma_{{\mathrm{2}}}  \rrbracket  \xrightarrow{\pi_i} \mathbf{S}_{n}  \llbracket  \ottnt{A}  \rrbracket \\
 \llbracket   \lambda  \ottmv{x}  :  \ottnt{A}  .  \ottnt{b}   \rrbracket  & =  \llbracket  \Gamma  \rrbracket  \xrightarrow{\Lambda  \llbracket  \ottnt{b}  \rrbracket } (\mathbf{S}_{n}  \llbracket  \ottnt{B}  \rrbracket )^{(\mathbf{S}_{n}  \llbracket  \ottnt{A}  \rrbracket )} \xrightarrow{q^{-1}} \mathbf{S}_{n} ( \llbracket  \ottnt{B}  \rrbracket ^{ \llbracket  \ottnt{A}  \rrbracket }) \\
 \llbracket    \ottnt{b}  \:  \ottnt{a}    \rrbracket  & =  \llbracket  \Gamma  \rrbracket  \xrightarrow{\langle q \circ  \llbracket  \ottnt{b}  \rrbracket  ,  \llbracket  \ottnt{a}  \rrbracket  \rangle} (\mathbf{S}_{n}  \llbracket  \ottnt{B}  \rrbracket )^{(\mathbf{S}_{n}  \llbracket  \ottnt{A}  \rrbracket )} \times \mathbf{S}_{n}  \llbracket  \ottnt{A}  \rrbracket  \xrightarrow{\text{app}} \mathbf{S}_{n}  \llbracket  \ottnt{B}  \rrbracket  \\
 \llbracket   \mathbf{next} \:  \ottnt{a}   \rrbracket  & =  \llbracket  \Gamma  \rrbracket  \xrightarrow{ \llbracket  \ottnt{a}  \rrbracket }\mathbf{S}_{ n  +   0'  }  \llbracket  \ottnt{A}  \rrbracket  \xrightarrow{\delta^{n, 0' }} \mathbf{S}_{n} \mathbf{S}_{ 0' }  \llbracket  \ottnt{A}  \rrbracket  \\
 \llbracket   \mathbf{prev} \:  \ottnt{a}   \rrbracket  & =  \llbracket  \Gamma  \rrbracket  \xrightarrow{ \llbracket  \ottnt{a}  \rrbracket } \mathbf{S}_{n} \mathbf{S}_{ 0' }  \llbracket  \ottnt{A}  \rrbracket  \xrightarrow{\mu^{n, 0' }} \mathbf{S}_{ n  +   0'  }  \llbracket  \ottnt{A}  \rrbracket  \\
 \llbracket   \mathbf{inj}_1 \:  \ottnt{a_{{\mathrm{1}}}}   \rrbracket  & =  \llbracket  \Gamma  \rrbracket  \xrightarrow{ \llbracket  \ottnt{a_{{\mathrm{1}}}}  \rrbracket }  \mathbf{S}_{  n  }   \llbracket  \ottnt{A_{{\mathrm{1}}}}  \rrbracket   \xrightarrow{ \mathbf{S}_{  n  }   i_1  }  \mathbf{S}_{  n  }   (    \llbracket  \ottnt{A_{{\mathrm{1}}}}  \rrbracket   +   \llbracket  \ottnt{A_{{\mathrm{2}}}}  \rrbracket    )  \\
 \llbracket    \mathbf{case} \:  \ottnt{a}  \: \mathbf{of} \:  \ottnt{b_{{\mathrm{1}}}}  \: ; \:  \ottnt{b_{{\mathrm{2}}}}    \rrbracket  & =  \llbracket  \Gamma  \rrbracket  \xrightarrow{ \langle   \langle   \llbracket  \ottnt{b_{{\mathrm{1}}}}  \rrbracket   ,   \llbracket  \ottnt{b_{{\mathrm{2}}}}  \rrbracket   \rangle   ,   \llbracket  \ottnt{a}  \rrbracket   \rangle }   (    \mathbf{S}_{  n  }     \llbracket  \ottnt{B}  \rrbracket  ^{  \llbracket  \ottnt{A_{{\mathrm{1}}}}  \rrbracket  }     \times   \mathbf{S}_{  n  }     \llbracket  \ottnt{B}  \rrbracket  ^{  \llbracket  \ottnt{A_{{\mathrm{2}}}}  \rrbracket  }      )   \times   \mathbf{S}_{  n  }   (    \llbracket  \ottnt{A_{{\mathrm{1}}}}  \rrbracket   +   \llbracket  \ottnt{A_{{\mathrm{2}}}}  \rrbracket    )    \\ & \qquad \; \; \, \xrightarrow{p^{-1} \times  \text{id} }   \mathbf{S}_{  n  }   (      \llbracket  \ottnt{B}  \rrbracket  ^{  \llbracket  \ottnt{A_{{\mathrm{1}}}}  \rrbracket  }    \times     \llbracket  \ottnt{B}  \rrbracket  ^{  \llbracket  \ottnt{A_{{\mathrm{2}}}}  \rrbracket  }     )    \times   \mathbf{S}_{  n  }   (    \llbracket  \ottnt{A_{{\mathrm{1}}}}  \rrbracket   +   \llbracket  \ottnt{A_{{\mathrm{2}}}}  \rrbracket    )    \\ & \qquad \; \; \, \xrightarrow{p^{-1}}  \mathbf{S}_{  n  }   (    (      \llbracket  \ottnt{B}  \rrbracket  ^{  \llbracket  \ottnt{A_{{\mathrm{1}}}}  \rrbracket  }    \times     \llbracket  \ottnt{B}  \rrbracket  ^{  \llbracket  \ottnt{A_{{\mathrm{2}}}}  \rrbracket  }     )   \times   (    \llbracket  \ottnt{A_{{\mathrm{1}}}}  \rrbracket   +   \llbracket  \ottnt{A_{{\mathrm{2}}}}  \rrbracket    )    )   \\ & \qquad \; \; \, \cong  \mathbf{S}_{  n  }   (      \llbracket  \ottnt{B}  \rrbracket  ^{    \llbracket  \ottnt{A_{{\mathrm{1}}}}  \rrbracket   +   \llbracket  \ottnt{A_{{\mathrm{2}}}}  \rrbracket    }    \times   (    \llbracket  \ottnt{A_{{\mathrm{1}}}}  \rrbracket   +   \llbracket  \ottnt{A_{{\mathrm{2}}}}  \rrbracket    )    )   \xrightarrow{ \mathbf{S}_{  n  }   \text{app}  }  \mathbf{S}_{  n  }   \llbracket  \ottnt{B}  \rrbracket  
\end{align*}

This gives us a sound interpretation of \lc{}.
\begin{theorem}\label{lcCat}
If $ \Gamma  \vdash  \ottnt{a}  :^{ n }  \ottnt{A} $ in \lc{}, then $ \llbracket  \ottnt{a}  \rrbracket  \in \text{Hom}_{\Ct} ( \llbracket  \Gamma  \rrbracket , \mathbf{S}_{n}  \llbracket  \ottnt{A}  \rrbracket )$. Further, if $ \Gamma  \vdash  \ottnt{a_{{\mathrm{1}}}}  :^{ n }  \ottnt{A} $ and $ \Gamma  \vdash  \ottnt{a_{{\mathrm{2}}}}  :^{ n }  \ottnt{A} $ such that $ \ottnt{a_{{\mathrm{1}}}}  \equiv  \ottnt{a_{{\mathrm{2}}}} $ in \lc{}, then $ \llbracket  \ottnt{a_{{\mathrm{1}}}}  \rrbracket  =  \llbracket  \ottnt{a_{{\mathrm{2}}}}  \rrbracket  \in \text{Hom}_{\Ct} ( \llbracket  \Gamma  \rrbracket , \mathbf{S}_{n}  \llbracket  \ottnt{A}  \rrbracket )$. 
\end{theorem}

Such an interpretation is also computationally adequate, provided it is injective for ground types. The operational semantics for the calculus is assumed to be the call-by-value semantics induced by the $\beta$-rules in Figure \ref{lcequality}. We denote the multi-step reduction relation corresponding to this operational semantics  by $\longmapsto^{\ast}$. Formally, we can state adequacy as:

\begin{theorem}\label{lcCA}
Let $ \Gamma  \vdash  \ottnt{b}  :^{ \ottsym{0} }   \mathbf{Bool}  $ and $\ottmv{v}$ be a value of type $ \mathbf{Bool} $. If $ \llbracket  \ottnt{b}  \rrbracket  =  \llbracket  \ottmv{v}  \rrbracket $ then $ \vdash  \ottnt{b}  \longmapsto^{\ast}  \ottmv{v} $. 
\end{theorem}

Next, we use these categorical models to show correctness of binding-time analysis in \lc{}.

\subsection{Correctness of Binding-Time Analysis in \lc{}} 

A binding-time analysis is correct if computation from an earlier stage \textit{does not depend} upon computation from a later stage. Here, using the categorical model, we shall show that \lc{} satisfies this condition. We shall use the same \textit{presence-absence test} technique that we used to prove noninterference for \ED{}. The goal is to show that computations at a given stage can proceed when computations from all later stages are blacked out. Towards this end, we present a strong monoidal functor $\mathbf{S}^0$ that keeps untouched all computations at time instant $\ottsym{0}$ but blacks out all computations from every time instant greater than $\ottsym{0}$.
\begin{align*}
\mathbf{S}^0 (n) = \begin{cases}
                  \Id & \text{ if } n = 0 \\
                  \ast & \text{ otherwise}
                  \end{cases}
\end{align*}
Note that $\mathbf{S}^0$ is, in fact, a strict monoidal functor because $\eta = \epsilon = \I$ and $\delta = \mu = \I$. Now, for any bicartesian closed category $\Ct$, any strong monoidal functor from $\Ca(\mathcal{N})$ to $\EC{}$ provides a computationally adequate interpretation of \lc{}, given the interpretation for  ground types is injective. As such, 
\begin{theorem} \label{lcS0}
$\llbracket \_ \rrbracket_{(\mathbf{Set},\mathbf{S}^{\ottsym{0}})}$ is a computationally adequate interpretation of \lc{}.
\end{theorem} 

The existence of $\mathbf{S}^0$ shows that binding-time analysis in \lc{} is correct. We elaborate on this point below. By the above theorem, computations at time instant $0$ can proceed independently of computations from all later stages. Once the computations from time instant $0$ are done, we can move on to computations from time instant $0'$, which now acts like the new $0$. Then, using the same argument, we can show that computations at time instant $0'$ can proceed independently of all later stage computations. Repeating this argument over and over again, we see that we can normalize \lc{}-expressions in a time-ordered manner. Therefore, binding-time analysis in \lc{} is correct. 

We can formalize the above argument into the following noninterference theorem.
\begin{theorem} \label{lcNI}
Let $  \emptyset   \vdash  \ottnt{f}  :^{ \ottsym{0} }    \bigcirc  \ottnt{A}   \to   \mathbf{Bool}   $ and $  \emptyset   \vdash  \ottnt{b_{{\mathrm{1}}}}  :^{ \ottsym{0} }   \bigcirc  \ottnt{A}  $ and $  \emptyset   \vdash  \ottnt{b_{{\mathrm{2}}}}  :^{ \ottsym{0} }   \bigcirc  \ottnt{A}  $. Then, $ \vdash   \ottnt{f}  \:  \ottnt{b_{{\mathrm{1}}}}   \longmapsto^{\ast}  \ottmv{v} $ if and only if $ \vdash   \ottnt{f}  \:  \ottnt{b_{{\mathrm{2}}}}   \longmapsto^{\ast}  \ottmv{v} $, where $\ottmv{v}$ is a value of type $ \mathbf{Bool} $.
\end{theorem}

The presence-absence test provides a simple yet powerful method for proving correctness of dependency analyses. We used it to show correctness of \ED{} and \lc{} in a very straightforward manner. To put it in perspective, \citet{lambdacirc} requires 10 journal pages to establish correctness of \lc{} using syntactic methods whereas our proof of correctness for \lc{} follows almost immediately from the soundness theorem. This shows that the presence-absence test may be a useful tool for establishing correctness of dependency calculi. 

Next, we show how to translate \lc{} into a graded monadic comonadic framework.

\subsection{Can We Translate \lc{} to GMCC?} \label{lcgmcc}

Here, we consider how we might translate \lc{} to GMCC. We can instantiate the parametrizing preordered monoid $\mathcal{M}$ to $\mathcal{N} = (\mathbb{N},+,0)$ and we may translate $\bc$ as: $ \widehat{  \bigcirc  \ottnt{A}  }  =  S_{   0'   } \:   \widehat{ \ottnt{A} }  $. We can translate contexts as:
\[  \widehat{   \emptyset   }  =  \emptyset  \hspace*{20pt}  \widehat{    \Gamma  ,   \ottmv{x}  :^{  n  }  \ottnt{A}     }  =   \widehat{  \Gamma  }   ,   \ottmv{x}  :   S_{  n  } \:   \widehat{ \ottnt{A} }     \]
A typing judgement $ \Gamma  \vdash  \ottnt{a}  :^{ n }  \ottnt{A} $ can then be translated as: $  \widehat{  \Gamma  }   \vdash   \widehat{ \ottnt{a} }   :   S_{  n  } \:   \widehat{ \ottnt{A} }   $. For the modal terms, we have:
\[  \widehat{  \mathbf{next} \:  \ottnt{a}  }  =  \mathbf{fork}^{  n  ,   0'   }   \widehat{ \ottnt{a} }   \hspace*{20pt}  \widehat{  \mathbf{prev} \:  \ottnt{a}  }  =  \mathbf{join}^{  n  ,   0'   }   \widehat{ \ottnt{a} }   \]

This translation works well for the modal constructs; however, we run into a problem when dealing with functions and applications. The problem is that with the above translation, it is not possible to show  typing is preserved in case of functions and applications. The reason behind this problem is the difference between \lc{} and GMCC we referred to earlier: the types $ \bigcirc^{ n }   (   \ottnt{A}  \to  \ottnt{B}   )  $ and $  \bigcirc^{ n }  \ottnt{A}   \to   \bigcirc^{ n }  \ottnt{B}  $ are isomorphic in \lc{} but not necessarily in GMCC.

This difference arises from the fact that in \lc{}, the grades pervade all the typing rules, including the ones for functions and applications, while in GMCC, they are restricted to the monadic and the comonadic typing rules. We could have designed GMCC by permeating the grades along all the typing rules in lieu of restricting them to a fragment of the calculus. We avoided such a design for the sake of simplicity. However, now that we understand the calculus, we can consider the implications of such a design choice. In the next section, we explore this design choice and present \Ge{}, GMCC extended with graded contexts and graded typing judgements.

\section{\Ge{}} \label{secgmcce}

\subsection{The Calculus}

Like GMCC, \Ge{} is parametrized by an arbitrary preordered monoid, $ \mathcal{M} $. The types of the calculus are the same as those of GMCC. With respect to terms, \Ge{} differs from GMCC in having only the following two term-level constructs (in lieu of \textbf{ret}, \textbf{extr}, etc.) for introducing and eliminating the modal type. 
\[ \text{terms}, a, b ::= \ldots (\lambda\text{-calculus terms}) \: | \:  \mathbf{split}^{ m }  \ottnt{a}  \: | \:  \mathbf{merge}^{ m }  \ottnt{a}  \]
\Ge{} also differs from GMCC wrt contexts and typing judgements, both of which are graded in \Ge{}, like in \lc{}. The typing judgement of \Ge{}, $ \Gamma  \vdash  \ottnt{a}  :^{ m }  \ottnt{A} $, intuitively means that $\ottnt{a}$ can be observed at $m$, provided the variables in $\Gamma$ are observable at their respective grades. The typing rules of the calculus appear in Figure \ref{gmccetyp}.  The rules pertaining to standard $\lambda$-calculus terms are as expected. The \rref{E-Split,E-Merge} introduce and eliminate the modal type. These rules are similar to \rref{C-Fork} and \rref{M-Join} respectively. The \rref{E-Up}, akin to \rref{M-Up,C-Up}, implicitly relaxes the grade at which a term is observed.

The equational theory of the calculus is induced by the $\beta\eta$-rules of standard $\lambda$-calculus along with the following $\beta\eta$-rule for modal terms.
\[  \mathbf{merge}^{ m }   (   \mathbf{split}^{ m }  \ottnt{a}   )   \equiv \ottnt{a} \hspace*{25pt} \ottnt{a}  \equiv  \mathbf{split}^{ m }   (   \mathbf{merge}^{ m }  \ottnt{a}   )   \]

\begin{figure}
\drules[E]{$ \Gamma  \vdash  \ottnt{a}  :^{ m }  \ottnt{A} $}{Typing}{Var,Lam,App,Pair,Proj,Inj,Case,Split,Merge,Unit,Up}
\caption{Typing Rules of \Ge{}}
\label{gmccetyp}
\end{figure}

The graded presentation of the calculus leads to some interesting consequences. Unlike GMCC, this calculus  enjoys the following properties.
\begin{prop} \label{propiso}
Let $ \mathcal{M} $ be any preordered monoid. Then, in \Ge{}($ \mathcal{M} $),
\begin{itemize}
\item The types $ S_{ m } \:   \ottkw{Unit}  $ and $ \ottkw{Unit} $ are isomorphic.
\item The types $ S_{ m } \:   (   \ottnt{A_{{\mathrm{1}}}}  \times  \ottnt{A_{{\mathrm{2}}}}   )  $ and $  S_{ m } \:  \ottnt{A_{{\mathrm{1}}}}   \times   S_{ m } \:  \ottnt{A_{{\mathrm{2}}}}  $, for all types $\ottnt{A_{{\mathrm{1}}}}$ and $\ottnt{A_{{\mathrm{2}}}}$, are isomorphic.
\item The types $ S_{ m } \:   (   \ottnt{A}  \to  \ottnt{B}   )  $ and $  S_{ m } \:  \ottnt{A}   \to   S_{ m } \:  \ottnt{B}  $, for all types $\ottnt{A}$ and $\ottnt{B}$, are isomorphic. 
\end{itemize}
\end{prop}

The third property above reminds us of the isomorphism between types $ \bigcirc^{ n }   (   \ottnt{A}  \to  \ottnt{B}   )  $ and $  \bigcirc^{ n }  \ottnt{A}   \to   \bigcirc^{ n }  \ottnt{B}  $ in \lc{}. Recall that we could not translate \lc{} into GMCC because such an isomorphism does not hold in general in GMCC. We shall see, \Ge{}, that satisfies these isomorphisms, can readily capture \lc{}. 

Next, we build categorical models for \Ge{}. The models are similar to those of GMCC; however, as in the case of \lc{}, we need to use cartesian closed endofunctors in lieu of just strong endofunctors. Let $ \mathcal{M} $ be the parametrizing structure. Let $\Ct$ be any bicartesian closed category. Further, let $\mathbf{S}$ be a strong monoidal functor from $\Ca(\mathcal{M})$ to $\EC$. Then, the interpretation $ \llbracket   \_   \rrbracket $ of types, contexts and terms is as follows.
\[  \llbracket   S_{ m } \:  \ottnt{A}   \rrbracket  =  \mathbf{S}_{ m }   \llbracket  \ottnt{A}  \rrbracket   \hspace{20pt}  \llbracket   \emptyset   \rrbracket  = \top \hspace{20pt}  \llbracket    \Gamma  ,   \ottmv{x}  :^{ m }  \ottnt{A}     \rrbracket  =  \llbracket  \Gamma  \rrbracket  \times \mathbf{S}_{m}  \llbracket  \ottnt{A}  \rrbracket  \]
\begin{align*}
 \llbracket   \mathbf{split}^{ m_{{\mathrm{2}}} }  \ottnt{a}   \rrbracket  & =  \llbracket  \Gamma  \rrbracket  \xrightarrow{ \llbracket  \ottnt{a}  \rrbracket }  \mathbf{S}_{  m_{{\mathrm{1}}}  \cdot  m_{{\mathrm{2}}}  }   \llbracket  \ottnt{A}  \rrbracket   \xrightarrow{ \delta^{ m_{{\mathrm{1}}} , m_{{\mathrm{2}}} }_{  \llbracket  \ottnt{A}  \rrbracket  } }  \mathbf{S}_{ m_{{\mathrm{1}}} }   \mathbf{S}_{ m_{{\mathrm{2}}} }   \llbracket  \ottnt{A}  \rrbracket    \\
 \llbracket   \mathbf{merge}^{ m_{{\mathrm{2}}} }  \ottnt{a}   \rrbracket  & =  \llbracket  \Gamma  \rrbracket  \xrightarrow{ \llbracket  \ottnt{a}  \rrbracket }  \mathbf{S}_{ m_{{\mathrm{1}}} }   \mathbf{S}_{ m_{{\mathrm{2}}} }   \llbracket  \ottnt{A}  \rrbracket    \xrightarrow{ \mu^{ m_{{\mathrm{1}}} , m_{{\mathrm{2}}} }_{  \llbracket  \ottnt{A}  \rrbracket  } }  \mathbf{S}_{  m_{{\mathrm{1}}}  \cdot  m_{{\mathrm{2}}}  }   \llbracket  \ottnt{A}  \rrbracket  
\end{align*}

Note that the types and terms of the standard $\lambda$-calculus are interpreted as in the case of \lc{}.\\ This gives us a sound interpretation of the calculus.

\begin{theorem}\label{gmcceSound}
If $ \Gamma  \vdash  \ottnt{a}  :^{ m }  \ottnt{A} $ in \Ge{}, then $ \llbracket  \ottnt{a}  \rrbracket  \in \text{Hom}_{\Ct} ( \llbracket  \Gamma  \rrbracket , \mathbf{S}_{m}  \llbracket  \ottnt{A}  \rrbracket )$. Further, if $ \Gamma  \vdash  \ottnt{a_{{\mathrm{1}}}}  :^{ m }  \ottnt{A} $ and $ \Gamma  \vdash  \ottnt{a_{{\mathrm{2}}}}  :^{ m }  \ottnt{A} $ such that $ \ottnt{a_{{\mathrm{1}}}}  \equiv  \ottnt{a_{{\mathrm{2}}}} $ in \Ge{}, then $ \llbracket  \ottnt{a_{{\mathrm{1}}}}  \rrbracket  =  \llbracket  \ottnt{a_{{\mathrm{2}}}}  \rrbracket  \in \text{Hom}_{\Ct} ( \llbracket  \Gamma  \rrbracket , \mathbf{S}_{m}  \llbracket  \ottnt{A}  \rrbracket )$. 
\end{theorem}

Next, we show that both \ED{} and \lc{} can be translated to \Ge{}.

\subsection{Translations from \ED{} and \lc{} to \Ge{}}

We know that, over the class of bounded join-semilattices, \ED{} is equivalent to GMCC. Given that GMCC is quite close to \Ge{}, we shall use GMCC as the source language for our translation. Let $ \mathcal{L}  = (L,\vee,\bot)$ be an arbitrary join-semilattice. Then, the translation $\widetilde{\phantom{a}}$ from GMCC($ \mathcal{L} $) to \Ge{}($ \mathcal{L} $) is as follows. For types, $ \widetilde{ \ottnt{A} }  = \ottnt{A}$. For terms,
\begin{align*}
 \widetilde{  \ottkw{ret}  \:  \ottnt{a}  }  & =  \mathbf{split}^{   \bot   }   \widetilde{ \ottnt{a} }   &  \widetilde{  \mathbf{extr} \:  \ottnt{a}  }  & =  \mathbf{merge}^{   \bot   }   \widetilde{ \ottnt{a} }   \\
 \widetilde{  \mathbf{join}^{  \ell_{{\mathrm{1}}}  ,  \ell_{{\mathrm{2}}}  }  \ottnt{a}  }  & =  \mathbf{split}^{   \ell_{{\mathrm{1}}}  \vee  \ell_{{\mathrm{2}}}   }   (   \mathbf{merge}^{  \ell_{{\mathrm{2}}}  }   (   \mathbf{merge}^{  \ell_{{\mathrm{1}}}  }   \widetilde{ \ottnt{a} }    )    )   &  \widetilde{  \mathbf{fork}^{  \ell_{{\mathrm{1}}}  ,  \ell_{{\mathrm{2}}}  }  \ottnt{a}  }  & =  \mathbf{split}^{  \ell_{{\mathrm{1}}}  }   (   \mathbf{split}^{  \ell_{{\mathrm{2}}}  }   (   \mathbf{merge}^{   \ell_{{\mathrm{1}}}  \vee  \ell_{{\mathrm{2}}}   }   \widetilde{ \ottnt{a} }    )    )   \\
 \widetilde{  \mathbf{lift}^{  \ell  }  \ottnt{f}  }  & =  \lambda  \ottmv{x}  .   \mathbf{split}^{  \ell  }   (    \widetilde{ \ottnt{f} }   \:   (   \mathbf{merge}^{  \ell  }  \ottmv{x}   )    )    &  \widetilde{  \mathbf{up}^{  \ell_{{\mathrm{1}}}  ,  \ell_{{\mathrm{2}}}  }  \ottnt{a}  }  & =  \mathbf{split}^{  \ell_{{\mathrm{2}}}  }   (   \mathbf{merge}^{  \ell_{{\mathrm{1}}}  }   \widetilde{ \ottnt{a} }    )  
\end{align*}

The standard $\lambda$-calculus terms are translated in the expected manner. This translation is sound, as we see next. Below we use $ \Gamma ^{ m } $ to denote the graded counterpart of $\Gamma$ where every assumption is held at $m$.

\begin{theorem} \label{gmcc2gmcce}
If $ \Gamma  \vdash  \ottnt{a}  :  \ottnt{A} $ in GMCC($ \mathcal{L} $), then $  \Gamma ^{   \bot   }   \vdash   \widetilde{ \ottnt{a} }   :^{   \bot   }  \ottnt{A} $ in \Ge{}($ \mathcal{L} $). Further, if $ \Gamma  \vdash  \ottnt{a_{{\mathrm{1}}}}  :  \ottnt{A} $ and $ \Gamma  \vdash  \ottnt{a_{{\mathrm{2}}}}  :  \ottnt{A} $ such that $ \ottnt{a_{{\mathrm{1}}}}  \equiv  \ottnt{a_{{\mathrm{2}}}} $ in GMCC($ \mathcal{L} $), then $  \widetilde{ \ottnt{a_{{\mathrm{1}}}} }   \equiv   \widetilde{ \ottnt{a_{{\mathrm{2}}}} }  $ in \Ge{}($ \mathcal{L} $).
\end{theorem}

One might wonder here whether we can go the other way around and translate \Ge{}($ \mathcal{L} $) to GMCC($ \mathcal{L} $). Though the two calculi are very similar, such a translation is not possible, owing to Proposition \ref{propiso}. More concretely, observe that the type $ S_{   \mathbf{Secret}   } \:   (    S_{   \mathbf{Secret}   } \:   \mathbf{Bool}    \to   \mathbf{Bool}    )  $ has four distinct terms in \Ge{}($ \mathcal{L}_2 $) but only two distinct terms in GMCC($ \mathcal{L}_2 $) ($ \mathcal{L}_2  \triangleq  \mathbf{Public}  \sqsubset  \mathbf{Secret} $).

Next, we return to our incomplete translation of \lc{} from Section \ref{lcgmcc}. Though we couldn't translate \lc{} to GMCC($ \mathcal{N} $), we can now easily translate it to \Ge{}($ \mathcal{N} $). The translation for the modal types and terms is as follows: 
\[  \widehat{  \bigcirc  \ottnt{A}  }  =  S_{   0'   } \:   \widehat{ \ottnt{A} }   \hspace*{15pt}  \widehat{  \mathbf{next} \:  \ottnt{a}  }  =  \mathbf{split}^{   0'   }   \widehat{ \ottnt{a} }   \hspace*{15pt}  \widehat{  \mathbf{prev} \:  \ottnt{a}  }  =  \mathbf{merge}^{   0'   }   \widehat{ \ottnt{a} }   \]
This translation is sound, as we see next. Below, $ \widehat{  \Gamma  } $ denotes $\Gamma$ with the types of the assumptions translated. 

\begin{theorem}\label{lc2gmcce}
If $ \Gamma  \vdash  \ottnt{a}  :^{ n }  \ottnt{A} $ in \lc{}, then $  \widehat{  \Gamma  }   \vdash   \widehat{ \ottnt{a} }   :^{  n  }   \widehat{ \ottnt{A} }  $ in \Ge{}($ \mathcal{N} $). Further, if $ \Gamma  \vdash  \ottnt{a_{{\mathrm{1}}}}  :^{ n }  \ottnt{A} $ and $ \Gamma  \vdash  \ottnt{a_{{\mathrm{2}}}}  :^{ n }  \ottnt{A} $ such that $ \ottnt{a_{{\mathrm{1}}}}  \equiv  \ottnt{a_{{\mathrm{2}}}} $ in \lc{}, then $  \widehat{ \ottnt{a_{{\mathrm{1}}}} }   \equiv   \widehat{ \ottnt{a_{{\mathrm{2}}}} }  $ in \Ge{}($ \mathcal{N} $).
\end{theorem}

We see that both \ED{} and \lc{} can be soundly translated into \Ge{}. Hence, \Ge{} is more general in its analysis of dependencies. \Ge{} shares some similarities with the sealing calculus \citep{igarashi}, which also subsumes the terminating fragment of DCC. Akin to \textbf{split} and \textbf{merge} in \Ge{}, the sealing calculus uses \textbf{seal} and \textbf{unseal} to introduce and eliminate the modal type. However, the sealing calculus is less general because it works for lattices only whereas \Ge{} works for any preordered monoid.

%% file: Sections/Discussion.tex
\section{Discussions and Related Work} \label{secdiscuss}

\subsection{Nontermination}

Till now, we did not include nontermination in any of our calculi. DCC, as presented by \citet{dcc}, includes nonterminating computations. So here we discuss how we can add nontermination to GMCC and \Ge{}. One of the features of these calculi is that they can be decomposed into a standard $\lambda$-calculus fragment and a modal fragment. The categorical models for the calculi also reflect the separation between the two fragments. The $\lambda$-calculus fragment is modelled by a bicartesian closed category, $\Ct$, whereas the modal fragment is modelled using monoidal functors from the parametrizing monoid to the category of endofunctors over $\Ct$. This separation between the two fragments makes it easier to add nontermination to these calculi. To include nontermination in these calculi, we just need to add the modal fragment on top of an already nonterminating calculus, for example, $\lambda$-calculus with pointed types. The nonterminating calculus can then be modelled by an appropriate category $\D$, for example, $\mathbf{Cpo}$, and the modal fragment by monoidal functors from the parametrizing monoid to the category of endofunctors over $\D$. Owing to the separation between the fragments, we conjecture that the proofs of the theorems would carry over smoothly to the new setting.




\subsection{Algehed's SDCC}

\citet{persdcc} is similar in spirit to our work. \citet{persdcc} designs a calculus, SDCC, that is equivalent to (the terminating fragment of) DCC. SDCC has the same types as DCC. However, it replaces the $\mathbf{bind}$ construct of DCC with four new constructs: $\mu, map, up$ and $com$. The constructs $\mu,map$ and $up$ serve the same purpose as $\mathbf{join},\mathbf{lift}$ and $\mathbf{up}$ respectively.  
The construct $com$ is what separates SDCC from our work. While \citet{persdcc} went with $com$ and designed a calculus equivalent to DCC, we went with $\mathbf{extr}$ and $\mathbf{fork}$ and designed a calculus that is equivalent to an extension of DCC. This choice helped us realize the power of comonadic aspect of dependency analysis.

\subsection{Relational Semantics of DCC, Revisited}

\citet{dcc} present a relational categorical model for DCC. They interpret each of the $\mathcal{T}_{\ell}$s as a separate monad on a base category, $\mathcal{DC}$. However, their interpretation can also be phrased in terms of a strong monoidal functor from $\Ca( \mathcal{L} )$ to $\mathbf{End}^{s}_{\mathcal{DC}}$. In other words, the model given by \citet{dcc} is an instance of the general class of models admitted by \ED{}, and hence DCC. 

In this regard, we want to point out a technical problem with the category $\mathcal{DC}$. Since $\mathcal{DC}$ models simply-typed $\lambda$-calculus, it should be cartesian closed. \citet{dcc} claim that it is so. However, contrary to their claim, the category $\mathcal{DC}$ is not cartesian closed. The exponential object in $\mathcal{DC}$ does not satisfy the universal property, unless the relations in the definition of $\text{Obj}(\mathcal{DC})$ are restricted to reflexive ones only. Note that such modified objects are nothing but classified sets that \citet{kavvos} uses to build a categorical model of DCC based on cohesion. How the cohesion-based models relate to our graded models of DCC is something we would like to explore in future.

%% file: Sections/Conclusion.tex
\section{Conclusion}

Dependency analysis is as much an analysis of dependence as of independence. By controlling the flow of information, it aims to ensure that certain entities are independent of certain other entities while being dependent upon yet other entities. To control flow, it needs to make use of unidirectional devices, devices that allow flow along one direction but block along the other. The two simple yet robust unidirectional devices in programming languages are monads and comonads. Monads allow inflow but block outflow whereas comonads allow outflow but block inflow. When used together, they ensure that flow respects the constraints imposed upon it by a predetermined structure, like a preordered monoid. Such a controlled flow can be used for a variety of purposes like enforcing security constraints, analysing binding-time, etc. 

%% file: Sections/Appendix.tex
\section{Proofs of lemmas/theorems stated in Section \ref{secgmc}}

\begin{theorem}[Theorem \ref{gmcsound}] \label{GMCPrf}
If $ \Gamma  \vdash  \ottnt{a}  :  \ottnt{A} $ in GMC, then $ \llbracket  \ottnt{a}  \rrbracket  \in \text{Hom}_{\Ct} ( \llbracket  \Gamma  \rrbracket ,  \llbracket  \ottnt{A}  \rrbracket )$. Further, if $ \Gamma  \vdash  \ottnt{a_{{\mathrm{1}}}}  :  \ottnt{A} $ and $ \Gamma  \vdash  \ottnt{a_{{\mathrm{2}}}}  :  \ottnt{A} $ such that $\ottnt{a_{{\mathrm{1}}}} \equiv \ottnt{a_{{\mathrm{2}}}}$ in GMC, then $ \llbracket  \ottnt{a_{{\mathrm{1}}}}  \rrbracket  =  \llbracket  \ottnt{a_{{\mathrm{2}}}}  \rrbracket  \in \text{Hom}_{\Ct} ( \llbracket  \Gamma  \rrbracket ,  \llbracket  \ottnt{A}  \rrbracket )$.
\end{theorem}

\begin{proof}
Let $ \Gamma  \vdash  \ottnt{a}  :  \ottnt{A} $. We show $ \llbracket  \ottnt{a}  \rrbracket  \in \text{Hom}_{\Ct} ( \llbracket  \Gamma  \rrbracket ,  \llbracket  \ottnt{A}  \rrbracket )$ by induction on the typing derivation.
\begin{itemize}
\item $\lambda$-calculus. Standard.

\item \Rref{M-Return}. Have: $ \Gamma  \vdash   \ottkw{ret}  \:  \ottnt{a}   :   T_{ \ottsym{1} } \:  \ottnt{A}  $ where $ \Gamma  \vdash  \ottnt{a}  :  \ottnt{A} $.\\
By IH, $ \llbracket  \ottnt{a}  \rrbracket  \in \text{Hom}_{\Ct} ( \llbracket  \Gamma  \rrbracket  ,  \llbracket  \ottnt{A}  \rrbracket )$.\\
Now, $ \llbracket   \ottkw{ret}  \:  \ottnt{a}   \rrbracket  =  \llbracket  \Gamma  \rrbracket  \xrightarrow{ \llbracket  \ottnt{a}  \rrbracket }  \llbracket  \ottnt{A}  \rrbracket  \xrightarrow{\eta_{ \llbracket  \ottnt{A}  \rrbracket }} \T_{\ottsym{1}}  \llbracket  \ottnt{A}  \rrbracket $.\\
Therefore, $ \llbracket   \ottkw{ret}  \:  \ottnt{a}   \rrbracket  \in \text{Hom}_{\Ct} ( \llbracket  \Gamma  \rrbracket  ,  \llbracket   T_{ \ottsym{1} } \:  \ottnt{A}   \rrbracket )$.

\item \Rref{M-Fmap}. Have: $ \Gamma  \vdash   \mathbf{lift}^{ m }  \ottnt{f}   :    T_{ m } \:  \ottnt{A}   \to   T_{ m } \:  \ottnt{B}   $ where $ \Gamma  \vdash  \ottnt{f}  :   \ottnt{A}  \to  \ottnt{B}  $.\\
By IH, $ \llbracket  \Gamma  \rrbracket  \xrightarrow{ \llbracket  \ottnt{f}  \rrbracket }  \llbracket  \ottnt{B}  \rrbracket ^{ \llbracket  \ottnt{A}  \rrbracket }$.\\
Now, $ \llbracket   \mathbf{lift}^{ m }  \ottnt{f}   \rrbracket  = \Lambda\Big(  \llbracket  \Gamma  \rrbracket  \times \T_m  \llbracket  \ottnt{A}  \rrbracket  \xrightarrow{t^{\T_{m}}} \T_m ( \llbracket  \Gamma  \rrbracket  \times  \llbracket  \ottnt{A}  \rrbracket ) \xrightarrow{\T_m (\Lambda^{-1}  \llbracket  \ottnt{f}  \rrbracket )} \T_m  \llbracket  \ottnt{B}  \rrbracket  \Big)$.\\
Therefore, $ \llbracket   \mathbf{lift}^{ m }  \ottnt{f}   \rrbracket  \in \text{Hom}_{\Ct} ( \llbracket  \Gamma  \rrbracket ,  \llbracket     T_{ m } \:  \ottnt{A}   \to   T_{ m } \:  \ottnt{B}     \rrbracket )$.

\item \Rref{M-Join}. Have: $ \Gamma  \vdash   \mathbf{join}^{ m_{{\mathrm{1}}} , m_{{\mathrm{2}}} }  \ottnt{a}   :   T_{  m_{{\mathrm{1}}}  \cdot  m_{{\mathrm{2}}}  } \:  \ottnt{A}  $ where $ \Gamma  \vdash  \ottnt{a}  :   T_{ m_{{\mathrm{1}}} } \:   T_{ m_{{\mathrm{2}}} } \:  \ottnt{A}   $.\\
By IH, $ \llbracket  \Gamma  \rrbracket  \xrightarrow{ \llbracket  \ottnt{a}  \rrbracket } \T_{m_{{\mathrm{1}}}} \T_{m_{{\mathrm{2}}}}  \llbracket  \ottnt{A}  \rrbracket $.\\
Now, $ \llbracket   \mathbf{join}^{ m_{{\mathrm{1}}} , m_{{\mathrm{2}}} }  \ottnt{a}   \rrbracket  =  \llbracket  \Gamma  \rrbracket  \xrightarrow{ \llbracket  \ottnt{a}  \rrbracket } \T_{m_{{\mathrm{1}}}} \T_{m_{{\mathrm{2}}}}  \llbracket  \ottnt{A}  \rrbracket  \xrightarrow{\mu^{m_{{\mathrm{1}}},m_{{\mathrm{2}}}}_{ \llbracket  \ottnt{A}  \rrbracket }} \T_{ m_{{\mathrm{1}}}  \cdot  m_{{\mathrm{2}}} }  \llbracket  \ottnt{A}  \rrbracket $.\\
Therefore, $ \llbracket   \mathbf{join}^{ m_{{\mathrm{1}}} , m_{{\mathrm{2}}} }  \ottnt{a}   \rrbracket  \in \text{Hom}_{\Ct} ( \llbracket  \Gamma  \rrbracket ,  \llbracket   T_{  m_{{\mathrm{1}}}  \cdot  m_{{\mathrm{2}}}  } \:  \ottnt{A}   \rrbracket )$.

\item \Rref{M-Up}. Have: $ \Gamma  \vdash   \mathbf{up}^{ m_{{\mathrm{1}}} , m_{{\mathrm{2}}} }  \ottnt{a}   :   T_{ m_{{\mathrm{2}}} } \:  \ottnt{A}  $ where $ \Gamma  \vdash  \ottnt{a}  :   T_{ m_{{\mathrm{1}}} } \:  \ottnt{A}  $ and $ m_{{\mathrm{1}}}   \leq   m_{{\mathrm{2}}} $.\\
By IH, $ \llbracket  \Gamma  \rrbracket  \xrightarrow{ \llbracket  \ottnt{a}  \rrbracket } \T_{m_{{\mathrm{1}}}}  \llbracket  \ottnt{A}  \rrbracket $.\\
Now, $ \llbracket   \mathbf{up}^{ m_{{\mathrm{1}}} , m_{{\mathrm{2}}} }  \ottnt{a}   \rrbracket  =  \llbracket  \Gamma  \rrbracket  \xrightarrow{ \llbracket  \ottnt{a}  \rrbracket } \T_{m_{{\mathrm{1}}}}  \llbracket  \ottnt{A}  \rrbracket  \xrightarrow{\T^{ m_{{\mathrm{1}}}   \leq   m_{{\mathrm{2}}} }_{ \llbracket  \ottnt{A}  \rrbracket }} \T_{m_{{\mathrm{2}}}}  \llbracket  \ottnt{A}  \rrbracket $.\\
Therefore, $ \llbracket   \mathbf{up}^{ m_{{\mathrm{1}}} , m_{{\mathrm{2}}} }  \ottnt{a}   \rrbracket  \in \text{Hom}_{\Ct} ( \llbracket  \Gamma  \rrbracket ,  \llbracket   T_{ m_{{\mathrm{2}}} } \:  \ottnt{A}   \rrbracket )$.
\end{itemize}

\vspace*{10pt}

Next, we show that if $ \Gamma  \vdash  \ottnt{a_{{\mathrm{1}}}}  :  \ottnt{A} $ and $ \Gamma  \vdash  \ottnt{a_{{\mathrm{2}}}}  :  \ottnt{A} $ such that $\ottnt{a_{{\mathrm{1}}}} \equiv \ottnt{a_{{\mathrm{2}}}}$ in GMC, then $ \llbracket  \ottnt{a_{{\mathrm{1}}}}  \rrbracket  =  \llbracket  \ottnt{a_{{\mathrm{2}}}}  \rrbracket $.\\
By inversion on $\ottnt{a_{{\mathrm{1}}}} \equiv \ottnt{a_{{\mathrm{2}}}}$:

\begin{itemize}

\item $\lambda$-calculus. Standard.

\item $ \mathbf{lift}^{ m }   (   \lambda  \ottmv{x}  .  \ottmv{x}   )   \equiv  \lambda  \ottmv{x}  .  \ottmv{x} $.\\
Given: $ \Gamma  \vdash   \mathbf{lift}^{ m }   (   \lambda  \ottmv{x}  :  \ottnt{A}  .  \ottmv{x}   )    :    T_{ m } \:  \ottnt{A}   \to   T_{ m } \:  \ottnt{A}   $ and $ \Gamma  \vdash   \lambda  \ottmv{x}  :   T_{ m } \:  \ottnt{A}   .  \ottmv{x}   :    T_{ m } \:  \ottnt{A}   \to   T_{ m } \:  \ottnt{A}   $.\\
Now, \begin{align*}
& \hspace*{10pt}  \llbracket   \mathbf{lift}^{ m }   (   \lambda  \ottmv{x}  :  \ottnt{A}  .  \ottmv{x}   )    \rrbracket  \\
& = \Lambda \Big( \T_m (\Lambda^{-1}  \llbracket   \lambda  \ottmv{x}  :  \ottnt{A}  .  \ottmv{x}   \rrbracket ) \circ t^{\T_m}_{ \llbracket  \Gamma  \rrbracket ,  \llbracket  \ottnt{A}  \rrbracket } \Big) \\
& = \Lambda (\T_m \pi_2 \circ t^{\T_m}_{ \llbracket  \Gamma  \rrbracket , \llbracket  \ottnt{A}  \rrbracket }) \\
& = \Lambda \pi_2 \hspace*{2pt} [\text{By naturality and strength}]\\
& =  \llbracket   \lambda  \ottmv{x}  :   T_{ m } \:  \ottnt{A}   .  \ottmv{x}   \rrbracket .
\end{align*}

\item $ \mathbf{lift}^{ m }   (    \lambda  \ottmv{x}  .  \ottnt{g}   \:   (   \ottnt{f}  \:  \ottmv{x}   )    )   \equiv   \lambda  \ottmv{x}  .   (   \mathbf{lift}^{ m }  \ottnt{g}   )    \:   (    (   \mathbf{lift}^{ m }  \ottnt{f}   )   \:  \ottmv{x}   )  $.\\
Given: $ \Gamma  \vdash   \mathbf{lift}^{ m }   (    \lambda  \ottmv{x}  :  \ottnt{A}  .  \ottnt{g}   \:   (   \ottnt{f}  \:  \ottmv{x}   )    )    :    T_{ m } \:  \ottnt{A}   \to   T_{ m } \:  \ottnt{C}   $ and $ \Gamma  \vdash    \lambda  \ottmv{x}  :   T_{ m } \:  \ottnt{A}   .   (   \mathbf{lift}^{ m }  \ottnt{g}   )    \:   (    (   \mathbf{lift}^{ m }  \ottnt{f}   )   \:  \ottmv{x}   )    :    T_{ m } \:  \ottnt{A}   \to   T_{ m } \:  \ottnt{C}   $ where $ \Gamma  \vdash  \ottnt{f}  :   \ottnt{A}  \to  \ottnt{B}  $ and $ \Gamma  \vdash  \ottnt{g}  :   \ottnt{B}  \to  \ottnt{C}  $.\\
Now, 
\begin{align*}
&\hspace*{10pt}  \llbracket   \mathbf{lift}^{ m }   (    \lambda  \ottmv{x}  :  \ottnt{A}  .  \ottnt{g}   \:   (   \ottnt{f}  \:  \ottmv{x}   )    )    \rrbracket  \\
& =  \Lambda \Big(     \mathbf{T}_{ m }   (   \Lambda^{-1}   \llbracket     \lambda  \ottmv{x}  :  \ottnt{A}  .  \ottnt{g}   \:   (   \ottnt{f}  \:  \ottmv{x}   )     \rrbracket    )    \circ   t^{\mathbf{T}_{ m } }_{  \llbracket  \Gamma  \rrbracket  ,   \llbracket  \ottnt{A}  \rrbracket  }     \Big)  \\ 
& =  \Lambda \Big(     \mathbf{T}_{ m }   (    \Lambda^{-1}    \llbracket  \ottnt{g}  \rrbracket     \circ   \langle   \pi_1   ,   \Lambda^{-1}    \llbracket  \ottnt{f}  \rrbracket     \rangle    )    \circ   t^{\mathbf{T}_{ m } }_{  \llbracket  \Gamma  \rrbracket  ,   \llbracket  \ottnt{A}  \rrbracket  }     \Big) .
\end{align*} 
Next,
\begin{align*}
&\hspace*{10pt}  \llbracket     \lambda  \ottmv{x}  :   T_{ m } \:  \ottnt{A}   .   (   \mathbf{lift}^{ m }  \ottnt{g}   )    \:   (    (   \mathbf{lift}^{ m }  \ottnt{f}   )   \:  \ottmv{x}   )     \rrbracket  \\
& =  \Lambda \Big(     \Lambda^{-1}    \llbracket   \mathbf{lift}^{ m }  \ottnt{g}   \rrbracket     \circ   \langle   \pi_1   ,   \Lambda^{-1}    \mathbf{lift}^{ m }  \ottnt{f}     \rangle     \Big)  \\
& =  \Lambda \Big(      \mathbf{T}_{ m }   (   \Lambda^{-1}   \llbracket  \ottnt{g}  \rrbracket    )    \circ    t^{\mathbf{T}_{ m } }_{  \llbracket  \Gamma  \rrbracket  ,   \llbracket  \ottnt{B}  \rrbracket  }     \circ   \langle   \pi_1   ,    \mathbf{T}_{ m }   (   \Lambda^{-1}   \llbracket  \ottnt{f}  \rrbracket    )    \circ   t^{\mathbf{T}_{ m } }_{  \llbracket  \Gamma  \rrbracket  ,   \llbracket  \ottnt{A}  \rrbracket  }    \rangle     \Big) .   
\end{align*}
The morphisms above are equal because the squares below commute:
\begin{figure}[h]
\begin{tikzcd}[row sep = 3.5 em, column sep = 4.5 em]
  \llbracket  \Gamma  \rrbracket   \times   \mathbf{T}_{ m }   \llbracket  \ottnt{A}  \rrbracket    \arrow{d}{ t^{\mathbf{T}_{ m } }_{  \llbracket  \Gamma  \rrbracket  ,   \llbracket  \ottnt{A}  \rrbracket  } } \arrow{r}{ \langle   \pi_1   ,   t^{\mathbf{T}_{ m } }_{  \llbracket  \Gamma  \rrbracket  ,   \llbracket  \ottnt{A}  \rrbracket  }   \rangle } &   \llbracket  \Gamma  \rrbracket   \times   \mathbf{T}_{ m }   (    \llbracket  \Gamma  \rrbracket   \times   \llbracket  \ottnt{A}  \rrbracket    )    \arrow{d}{ t^{\mathbf{T}_{ m } }_{   \llbracket  \Gamma  \rrbracket   ,     \llbracket  \Gamma  \rrbracket   \times   \llbracket  \ottnt{A}  \rrbracket    } } \arrow{r}{   \text{id}_{  \llbracket  \Gamma  \rrbracket  }    \times   \mathbf{T}_{ m }   (   \Lambda^{-1}   \llbracket  \ottnt{f}  \rrbracket    )   } &   \llbracket  \Gamma  \rrbracket   \times   \mathbf{T}_{ m }   \llbracket  \ottnt{B}  \rrbracket    \arrow{d}{ t^{\mathbf{T}_{ m } }_{  \llbracket  \Gamma  \rrbracket  ,   \llbracket  \ottnt{B}  \rrbracket  } }\\
 \mathbf{T}_{ m }   (    \llbracket  \Gamma  \rrbracket   \times   \llbracket  \ottnt{A}  \rrbracket    )   \arrow{r}{ \mathbf{T}_{ m }   \langle   \pi_1   ,   \text{id}_{    \llbracket  \Gamma  \rrbracket   \times   \llbracket  \ottnt{A}  \rrbracket    }   \rangle  } &  \mathbf{T}_{ m }   (    \llbracket  \Gamma  \rrbracket   \times   (    \llbracket  \Gamma  \rrbracket   \times   \llbracket  \ottnt{A}  \rrbracket    )    )   \arrow{r}{ \mathbf{T}_{ m }   (     \text{id}_{  \llbracket  \Gamma  \rrbracket  }    \times    \Lambda^{-1}   \llbracket  \ottnt{f}  \rrbracket      )  } &  \mathbf{T}_{ m }   (    \llbracket  \Gamma  \rrbracket   \times   \llbracket  \ottnt{B}  \rrbracket    )  
\end{tikzcd}
\end{figure}

\vspace*{1pt}

The right square commutes by naturality whereas the left one commutes because the diagram below commutes:
\begin{figure}[h]
\begin{tikzcd}[row sep = 3 em, column sep = 3.8 em]
  \llbracket  \Gamma  \rrbracket   \times   \mathbf{T}_{ m }   \llbracket  \ottnt{A}  \rrbracket    \arrow{rr}{ t^{\mathbf{T}_{ m } }_{  \llbracket  \Gamma  \rrbracket  ,   \llbracket  \ottnt{A}  \rrbracket  } } \arrow{d}{  \langle   \text{id}_{  \llbracket  \Gamma  \rrbracket  }   ,   \text{id}_{  \llbracket  \Gamma  \rrbracket  }   \rangle   \times   \mathbf{T}_{ m }   \text{id}_{  \llbracket  \ottnt{A}  \rrbracket  }   } & &  \mathbf{T}_{ m }   (    \llbracket  \Gamma  \rrbracket   \times   \llbracket  \ottnt{A}  \rrbracket    )   \arrow{d}{ \mathbf{T}_{ m }   (    \langle   \text{id}_{  \llbracket  \Gamma  \rrbracket  }   ,   \text{id}_{  \llbracket  \Gamma  \rrbracket  }   \rangle   \times   \text{id}_{  \llbracket  \ottnt{A}  \rrbracket  }    )  }\\
  (    \llbracket  \Gamma  \rrbracket   \times   \llbracket  \Gamma  \rrbracket    )   \times   \mathbf{T}_{ m }   \llbracket  \ottnt{A}  \rrbracket    \arrow{rr}{ t^{\mathbf{T}_{ m } }_{    \llbracket  \Gamma  \rrbracket   \times   \llbracket  \Gamma  \rrbracket    ,   \llbracket  \ottnt{A}  \rrbracket  } } \arrow{d}{ \alpha^{-1}_{  \llbracket  \Gamma  \rrbracket  ,  \llbracket  \Gamma  \rrbracket  ,  \mathbf{T}_{ m }   \llbracket  \ottnt{A}  \rrbracket   } } & &  \mathbf{T}_{ m }   (    (    \llbracket  \Gamma  \rrbracket   \times   \llbracket  \Gamma  \rrbracket    )   \times   \llbracket  \ottnt{A}  \rrbracket    )   \arrow{d}{ \mathbf{T}_{ m }   \alpha^{-1}_{  \llbracket  \Gamma  \rrbracket  ,  \llbracket  \Gamma  \rrbracket  ,  \llbracket  \ottnt{A}  \rrbracket  }  }\\
  \llbracket  \Gamma  \rrbracket   \times   (    \llbracket  \Gamma  \rrbracket   \times   \mathbf{T}_{ m }   \llbracket  \ottnt{A}  \rrbracket     )   \arrow{r}{  \text{id}_{  \llbracket  \Gamma  \rrbracket  }   \times   t^{\mathbf{T}_{ m } }_{  \llbracket  \Gamma  \rrbracket  ,   \llbracket  \ottnt{A}  \rrbracket  }  } &   \llbracket  \Gamma  \rrbracket   \times   \mathbf{T}_{ m }   (    \llbracket  \Gamma  \rrbracket   \times   \llbracket  \ottnt{A}  \rrbracket    )    \arrow{r}{ t^{\mathbf{T}_{ m } }_{  \llbracket  \Gamma  \rrbracket  ,     \llbracket  \Gamma  \rrbracket   \times   \llbracket  \ottnt{A}  \rrbracket    } } &  \mathbf{T}_{ m }   (    \llbracket  \Gamma  \rrbracket   \times   (    \llbracket  \Gamma  \rrbracket   \times   \llbracket  \ottnt{A}  \rrbracket    )    )    
\end{tikzcd}
\caption{Commutative diagram}
\label{prfcd0}
\end{figure}

\vspace*{1pt}

The square above commutes by naturality whereas the rectangle below commutes by strength.

\item $ \mathbf{up}^{ m_{{\mathrm{1}}} , m_{{\mathrm{1}}} }  \ottnt{a}  \equiv \ottnt{a}$.\\
Given: $ \Gamma  \vdash   \mathbf{up}^{ m_{{\mathrm{1}}} , m_{{\mathrm{1}}} }  \ottnt{a}   :   T_{ m_{{\mathrm{1}}} } \:  \ottnt{A}  $ where $ \Gamma  \vdash  \ottnt{a}  :   T_{ m_{{\mathrm{1}}} } \:  \ottnt{A}  $.\\
Since $\mathbf{T}$ is a functor, $ \mathbf{T}^{ m_{{\mathrm{1}}}  \leq  m_{{\mathrm{1}}} }_{  \llbracket  \ottnt{A}  \rrbracket  }  =  \text{id}_{  \mathbf{T}_{ m_{{\mathrm{1}}} }   \llbracket  \ottnt{A}  \rrbracket   } $.\\
Therefore, $ \llbracket    \mathbf{up}^{ m_{{\mathrm{1}}} , m_{{\mathrm{1}}} }  \ottnt{a}    \rrbracket  =    \mathbf{T}^{ m_{{\mathrm{1}}}  \leq  m_{{\mathrm{1}}} }_{  \llbracket  \ottnt{A}  \rrbracket  }    \circ   \llbracket  \ottnt{a}  \rrbracket   =  \llbracket  \ottnt{a}  \rrbracket $.

\item $ \mathbf{up}^{ m_{{\mathrm{2}}} , m_{{\mathrm{3}}} }   (   \mathbf{up}^{ m_{{\mathrm{1}}} , m_{{\mathrm{2}}} }  \ottnt{a}   )   \equiv  \mathbf{up}^{ m_{{\mathrm{1}}} , m_{{\mathrm{3}}} }  \ottnt{a} $.\\
Given: $ \Gamma  \vdash   \mathbf{up}^{ m_{{\mathrm{2}}} , m_{{\mathrm{3}}} }   (   \mathbf{up}^{ m_{{\mathrm{1}}} , m_{{\mathrm{2}}} }  \ottnt{a}   )    :   T_{ m_{{\mathrm{3}}} } \:  \ottnt{A}  $ and $ \Gamma  \vdash   \mathbf{up}^{ m_{{\mathrm{1}}} , m_{{\mathrm{3}}} }  \ottnt{a}   :   T_{ m_{{\mathrm{3}}} } \:  \ottnt{A}  $ where $ \Gamma  \vdash  \ottnt{a}  :   T_{ m_{{\mathrm{1}}} } \:  \ottnt{A}  $ and $ m_{{\mathrm{1}}}   \leq   m_{{\mathrm{2}}} $ and $ m_{{\mathrm{2}}}   \leq   m_{{\mathrm{3}}} $.\\
Since $\mathbf{T}$ is a functor, $   \mathbf{T}^{ m_{{\mathrm{2}}}  \leq  m_{{\mathrm{3}}} }_{  \llbracket  \ottnt{A}  \rrbracket  }    \circ    \mathbf{T}^{ m_{{\mathrm{1}}}  \leq  m_{{\mathrm{2}}} }_{  \llbracket  \ottnt{A}  \rrbracket  }    =  \mathbf{T}^{ m_{{\mathrm{1}}}  \leq  m_{{\mathrm{3}}} }_{  \llbracket  \ottnt{A}  \rrbracket  } $.\\
Therefore, $ \llbracket    \mathbf{up}^{ m_{{\mathrm{2}}} , m_{{\mathrm{3}}} }   (   \mathbf{up}^{ m_{{\mathrm{1}}} , m_{{\mathrm{2}}} }  \ottnt{a}   )     \rrbracket  =     \mathbf{T}^{ m_{{\mathrm{2}}}  \leq  m_{{\mathrm{3}}} }_{  \llbracket  \ottnt{A}  \rrbracket  }    \circ    \mathbf{T}^{ m_{{\mathrm{1}}}  \leq  m_{{\mathrm{2}}} }_{  \llbracket  \ottnt{A}  \rrbracket  }     \circ   \llbracket  \ottnt{a}  \rrbracket   =    \mathbf{T}^{ m_{{\mathrm{1}}}  \leq  m_{{\mathrm{3}}} }_{  \llbracket  \ottnt{A}  \rrbracket  }    \circ   \llbracket  \ottnt{a}  \rrbracket   =  \llbracket   \mathbf{up}^{ m_{{\mathrm{1}}} , m_{{\mathrm{3}}} }  \ottnt{a}   \rrbracket $. 

\item $  (   \mathbf{up}^{ m_{{\mathrm{1}}} , m'_{{\mathrm{1}}} }  \ottnt{a}   )   \:  \leftindex^{ m'_{{\mathrm{1}}} }{\gg}\!\! =^{ m_{{\mathrm{2}}} }  \ottnt{f}  \equiv  \mathbf{up}^{  m_{{\mathrm{1}}}  \cdot  m_{{\mathrm{2}}}  ,  m'_{{\mathrm{1}}}  \cdot  m_{{\mathrm{2}}}  }   (   \ottnt{a}  \:  \leftindex^{ m_{{\mathrm{1}}} }{\gg}\!\! =^{ m_{{\mathrm{2}}} }  \ottnt{f}   )  $.\\
Given: $ \Gamma  \vdash    (   \mathbf{up}^{ m_{{\mathrm{1}}} , m'_{{\mathrm{1}}} }  \ottnt{a}   )   \:  \leftindex^{ m'_{{\mathrm{1}}} }{\gg}\!\! =^{ m_{{\mathrm{2}}} }  \ottnt{f}   :   T_{  m'_{{\mathrm{1}}}  \cdot  m_{{\mathrm{2}}}  } \:  \ottnt{B}  $ and $ \Gamma  \vdash   \mathbf{up}^{  m_{{\mathrm{1}}}  \cdot  m_{{\mathrm{2}}}  ,  m'_{{\mathrm{1}}}  \cdot  m_{{\mathrm{2}}}  }   (   \ottnt{a}  \:  \leftindex^{ m_{{\mathrm{1}}} }{\gg}\!\! =^{ m_{{\mathrm{2}}} }  \ottnt{f}   )    :   T_{  m'_{{\mathrm{1}}}  \cdot  m_{{\mathrm{2}}}  } \:  \ottnt{B}  $ where $ \Gamma  \vdash  \ottnt{a}  :   T_{ m_{{\mathrm{1}}} } \:  \ottnt{A}  $ and $ \Gamma  \vdash  \ottnt{f}  :   \ottnt{A}  \to   T_{ m_{{\mathrm{2}}} } \:  \ottnt{B}   $ and $ m_{{\mathrm{1}}}   \leq   m'_{{\mathrm{1}}} $.\\
Now,
\begin{align*}
 &\hspace*{10pt}  \llbracket     (   \mathbf{up}^{ m_{{\mathrm{1}}} , m'_{{\mathrm{1}}} }  \ottnt{a}   )   \:  \leftindex^{ m'_{{\mathrm{1}}} }{\gg}\!\! =^{ m_{{\mathrm{2}}} }  \ottnt{f}    \rrbracket  \\
 & =  \llbracket    \mathbf{join}^{ m'_{{\mathrm{1}}} , m_{{\mathrm{2}}} }   (    (   \mathbf{lift}^{ m'_{{\mathrm{1}}} }  \ottnt{f}   )   \:   (   \mathbf{up}^{ m_{{\mathrm{1}}} , m'_{{\mathrm{1}}} }  \ottnt{a}   )    )     \rrbracket  \\
 & =   \mu^{ m'_{{\mathrm{1}}} , m_{{\mathrm{2}}} }_{  \llbracket  \ottnt{B}  \rrbracket  }   \circ   \llbracket     (   \mathbf{lift}^{ m'_{{\mathrm{1}}} }  \ottnt{f}   )   \:   (   \mathbf{up}^{ m_{{\mathrm{1}}} , m'_{{\mathrm{1}}} }  \ottnt{a}   )     \rrbracket  \\
 & =    \mu^{ m'_{{\mathrm{1}}} , m_{{\mathrm{2}}} }_{  \llbracket  \ottnt{B}  \rrbracket  }   \circ   \text{app}    \circ   \langle   \llbracket   \mathbf{lift}^{ m'_{{\mathrm{1}}} }  \ottnt{f}   \rrbracket   ,   \llbracket   \mathbf{up}^{ m_{{\mathrm{1}}} , m'_{{\mathrm{1}}} }  \ottnt{a}   \rrbracket   \rangle   \\
 & =    \mu^{ m'_{{\mathrm{1}}} , m_{{\mathrm{2}}} }_{  \llbracket  \ottnt{B}  \rrbracket  }   \circ   \text{app}    \circ   \langle   \Lambda \Big(     \mathbf{T}_{ m'_{{\mathrm{1}}} }   (   \Lambda^{-1}   \llbracket  \ottnt{f}  \rrbracket    )    \circ   t^{\mathbf{T}_{ m'_{{\mathrm{1}}} } }_{  \llbracket  \Gamma  \rrbracket  ,   \llbracket  \ottnt{A}  \rrbracket  }     \Big)   ,     \mathbf{T}^{ m_{{\mathrm{1}}}  \leq  m'_{{\mathrm{1}}} }_{  \llbracket  \ottnt{A}  \rrbracket  }    \circ   \llbracket  \ottnt{a}  \rrbracket    \rangle   \\
 & =     \mu^{ m'_{{\mathrm{1}}} , m_{{\mathrm{2}}} }_{  \llbracket  \ottnt{B}  \rrbracket  }   \circ   \mathbf{T}_{ m'_{{\mathrm{1}}} }   (   \Lambda^{-1}   \llbracket  \ottnt{f}  \rrbracket    )     \circ   t^{\mathbf{T}_{ m'_{{\mathrm{1}}} } }_{  \llbracket  \Gamma  \rrbracket  ,   \llbracket  \ottnt{A}  \rrbracket  }    \circ   \langle   \text{id}_{  \llbracket  \Gamma  \rrbracket  }   ,     \mathbf{T}^{ m_{{\mathrm{1}}}  \leq  m'_{{\mathrm{1}}} }_{  \llbracket  \ottnt{A}  \rrbracket  }    \circ   \llbracket  \ottnt{a}  \rrbracket    \rangle  \\
 & =      \mu^{ m'_{{\mathrm{1}}} , m_{{\mathrm{2}}} }_{  \llbracket  \ottnt{B}  \rrbracket  }   \circ   \mathbf{T}_{ m'_{{\mathrm{1}}} }   (   \Lambda^{-1}   \llbracket  \ottnt{f}  \rrbracket    )     \circ   t^{\mathbf{T}_{ m'_{{\mathrm{1}}} } }_{  \llbracket  \Gamma  \rrbracket  ,   \llbracket  \ottnt{A}  \rrbracket  }    \circ   (    \text{id}_{  \llbracket  \Gamma  \rrbracket  }   \times   \mathbf{T}^{ m_{{\mathrm{1}}}  \leq  m'_{{\mathrm{1}}} }_{  \llbracket  \ottnt{A}  \rrbracket  }    )    \circ   \langle   \text{id}_{  \llbracket  \Gamma  \rrbracket  }   ,   \llbracket  \ottnt{a}  \rrbracket   \rangle  .
\end{align*}
Next,
\begin{align*}
&\hspace*{10pt}  \llbracket   \mathbf{up}^{  m_{{\mathrm{1}}}  \cdot  m_{{\mathrm{2}}}  ,  m'_{{\mathrm{1}}}  \cdot  m_{{\mathrm{2}}}  }   (   \ottnt{a}  \:  \leftindex^{ m_{{\mathrm{1}}} }{\gg}\!\! =^{ m_{{\mathrm{2}}} }  \ottnt{f}   )    \rrbracket  \\
& =   \mathbf{T}^{  m_{{\mathrm{1}}}  \cdot  m_{{\mathrm{2}}}   \leq   m'_{{\mathrm{1}}}  \cdot  m_{{\mathrm{2}}}  }_{  \llbracket  \ottnt{B}  \rrbracket  }   \circ   \llbracket    \ottnt{a}  \:  \leftindex^{ m_{{\mathrm{1}}} }{\gg}\!\! =^{ m_{{\mathrm{2}}} }  \ottnt{f}    \rrbracket   \\
& =   \mathbf{T}^{  m_{{\mathrm{1}}}  \cdot  m_{{\mathrm{2}}}   \leq   m'_{{\mathrm{1}}}  \cdot  m_{{\mathrm{2}}}  }_{  \llbracket  \ottnt{B}  \rrbracket  }   \circ   \llbracket    \mathbf{join}^{ m_{{\mathrm{1}}} , m_{{\mathrm{2}}} }   (    (   \mathbf{lift}^{ m_{{\mathrm{1}}} }  \ottnt{f}   )   \:  \ottnt{a}   )     \rrbracket  \\
& =    \mathbf{T}^{  m_{{\mathrm{1}}}  \cdot  m_{{\mathrm{2}}}   \leq   m'_{{\mathrm{1}}}  \cdot  m_{{\mathrm{2}}}  }_{  \llbracket  \ottnt{B}  \rrbracket  }   \circ   \mu^{ m_{{\mathrm{1}}} , m_{{\mathrm{2}}} }_{  \llbracket  \ottnt{B}  \rrbracket  }    \circ   \llbracket     (   \mathbf{lift}^{ m_{{\mathrm{1}}} }  \ottnt{f}   )   \:  \ottnt{a}    \rrbracket   \\
& =     \mathbf{T}^{  m_{{\mathrm{1}}}  \cdot  m_{{\mathrm{2}}}   \leq   m'_{{\mathrm{1}}}  \cdot  m_{{\mathrm{2}}}  }_{  \llbracket  \ottnt{B}  \rrbracket  }   \circ   \mu^{ m_{{\mathrm{1}}} , m_{{\mathrm{2}}} }_{  \llbracket  \ottnt{B}  \rrbracket  }    \circ   \text{app}    \circ   \langle   \Lambda \Big(     \mathbf{T}_{ m_{{\mathrm{1}}} }   (   \Lambda^{-1}   \llbracket  \ottnt{f}  \rrbracket    )    \circ   t^{\mathbf{T}_{ m_{{\mathrm{1}}} } }_{  \llbracket  \Gamma  \rrbracket  ,   \llbracket  \ottnt{A}  \rrbracket  }     \Big)   ,   \llbracket  \ottnt{a}  \rrbracket   \rangle  \\
& =      \mathbf{T}^{  m_{{\mathrm{1}}}  \cdot  m_{{\mathrm{2}}}   \leq   m'_{{\mathrm{1}}}  \cdot  m_{{\mathrm{2}}}  }_{  \llbracket  \ottnt{B}  \rrbracket  }   \circ   \mu^{ m_{{\mathrm{1}}} , m_{{\mathrm{2}}} }_{  \llbracket  \ottnt{B}  \rrbracket  }    \circ   \mathbf{T}_{ m_{{\mathrm{1}}} }   (   \Lambda^{-1}   \llbracket  \ottnt{f}  \rrbracket    )     \circ   t^{\mathbf{T}_{ m_{{\mathrm{1}}} } }_{  \llbracket  \Gamma  \rrbracket  ,   \llbracket  \ottnt{A}  \rrbracket  }    \circ   \langle   \text{id}_{  \llbracket  \Gamma  \rrbracket  }   ,   \llbracket  \ottnt{a}  \rrbracket   \rangle  .  
\end{align*}

The morphisms above are equal because the diagram below commutes:

\begin{figure}[h]
\begin{tikzcd}[row sep = 4 em, column sep = 4 em]
  \llbracket  \Gamma  \rrbracket   \times   \mathbf{T}_{ m_{{\mathrm{1}}} }   \llbracket  \ottnt{A}  \rrbracket    \arrow{d}{  \text{id}_{  \llbracket  \Gamma  \rrbracket  }   \times   \mathbf{T}^{ m_{{\mathrm{1}}}  \leq  m'_{{\mathrm{1}}} }_{  \llbracket  \ottnt{A}  \rrbracket  }  } \arrow{r}{ t^{\mathbf{T}_{ m_{{\mathrm{1}}} } }_{  \llbracket  \Gamma  \rrbracket  ,   \llbracket  \ottnt{A}  \rrbracket  } } &  \mathbf{T}_{ m_{{\mathrm{1}}} }   (    \llbracket  \Gamma  \rrbracket   \times   \llbracket  \ottnt{A}  \rrbracket    )   \arrow{d}{ \mathbf{T}^{ m_{{\mathrm{1}}}  \leq  m'_{{\mathrm{1}}} }_{    \llbracket  \Gamma  \rrbracket   \times   \llbracket  \ottnt{A}  \rrbracket    } } \arrow{r}{ \mathbf{T}_{ m_{{\mathrm{1}}} }   (   \Lambda^{-1}   \llbracket  \ottnt{f}  \rrbracket    )  } &  \mathbf{T}_{ m_{{\mathrm{1}}} }    \mathbf{T}_{ m_{{\mathrm{2}}} }   \llbracket  \ottnt{B}  \rrbracket     \arrow{d}{ \mathbf{T}^{ m_{{\mathrm{1}}}  \leq  m'_{{\mathrm{1}}} }_{   \mathbf{T}_{ m_{{\mathrm{2}}} }   \llbracket  \ottnt{B}  \rrbracket    } } \arrow{r}{ \mu^{ m_{{\mathrm{1}}} , m_{{\mathrm{2}}} }_{  \llbracket  \ottnt{B}  \rrbracket  } } &  \mathbf{T}_{  m_{{\mathrm{1}}}  \cdot  m_{{\mathrm{2}}}  }   \llbracket  \ottnt{B}  \rrbracket   \arrow{d}{ \mathbf{T}^{  m_{{\mathrm{1}}}  \cdot  m_{{\mathrm{2}}}   \leq   m'_{{\mathrm{1}}}  \cdot  m_{{\mathrm{2}}}  }_{  \llbracket  \ottnt{B}  \rrbracket  } } \\
  \llbracket  \Gamma  \rrbracket   \times   \mathbf{T}_{ m'_{{\mathrm{1}}} }   \llbracket  \ottnt{A}  \rrbracket    \arrow{r}{ t^{\mathbf{T}_{ m'_{{\mathrm{1}}} } }_{  \llbracket  \Gamma  \rrbracket  ,   \llbracket  \ottnt{A}  \rrbracket  } } &  \mathbf{T}_{ m'_{{\mathrm{1}}} }   (    \llbracket  \Gamma  \rrbracket   \times   \llbracket  \ottnt{A}  \rrbracket    )   \arrow{r}{ \mathbf{T}_{ m'_{{\mathrm{1}}} }   (   \Lambda^{-1}   \llbracket  \ottnt{f}  \rrbracket    )  } &  \mathbf{T}_{ m'_{{\mathrm{1}}} }    \mathbf{T}_{ m_{{\mathrm{2}}} }   \llbracket  \ottnt{B}  \rrbracket     \arrow{r}{ \mu^{ m'_{{\mathrm{1}}} , m_{{\mathrm{2}}} }_{  \llbracket  \ottnt{B}  \rrbracket  } } &  \mathbf{T}_{  m'_{{\mathrm{1}}}  \cdot  m_{{\mathrm{2}}}  }   \llbracket  \ottnt{B}  \rrbracket  
\end{tikzcd}
\end{figure}

The leftmost square commutes because $\mathbf{T}^{ m_{{\mathrm{1}}}   \leq   m'_{{\mathrm{1}}} }$ is a \textit{strong} natural transformation; the middle one commutes because $\mathbf{T}^{ m_{{\mathrm{1}}}   \leq   m'_{{\mathrm{1}}} }$ is a natural transformation; the rightmost one commutes because $\mu$ is natural in its first component.

\item $ \ottnt{a}  \:  \leftindex^{ m_{{\mathrm{1}}} }{\gg}\!\! =^{ m'_{{\mathrm{2}}} }   (   \lambda  \ottmv{x}  .   \mathbf{up}^{ m_{{\mathrm{2}}} , m'_{{\mathrm{2}}} }  \ottnt{b}    )   \equiv  \mathbf{up}^{   m_{{\mathrm{1}}}  \cdot  m_{{\mathrm{2}}}   ,   m_{{\mathrm{1}}}  \cdot  m'_{{\mathrm{2}}}   }   (   \ottnt{a}  \:  \leftindex^{ m_{{\mathrm{1}}} }{\gg}\!\! =^{ m_{{\mathrm{2}}} }   \lambda  \ottmv{x}  .  \ottnt{b}    )  $.\\
Given: $ \Gamma  \vdash   \ottnt{a}  \:  \leftindex^{ m_{{\mathrm{1}}} }{\gg}\!\! =^{ m'_{{\mathrm{2}}} }   (   \lambda  \ottmv{x}  :  \ottnt{A}  .   \mathbf{up}^{ m_{{\mathrm{2}}} , m'_{{\mathrm{2}}} }  \ottnt{b}    )    :   T_{  m_{{\mathrm{1}}}  \cdot  m'_{{\mathrm{2}}}  } \:  \ottnt{B}  $ and $ \Gamma  \vdash   \mathbf{up}^{  m_{{\mathrm{1}}}  \cdot  m_{{\mathrm{2}}}  ,  m_{{\mathrm{1}}}  \cdot  m'_{{\mathrm{2}}}  }   (   \ottnt{a}  \:  \leftindex^{ m_{{\mathrm{1}}} }{\gg}\!\! =^{ m_{{\mathrm{2}}} }   \lambda  \ottmv{x}  :  \ottnt{A}  .  \ottnt{b}    )    :   T_{  m_{{\mathrm{1}}}  \cdot  m'_{{\mathrm{2}}}  } \:  \ottnt{B}  $ where $ \Gamma  \vdash  \ottnt{a}  :   T_{ m_{{\mathrm{1}}} } \:  \ottnt{A}  $ and $  \Gamma  ,   \ottmv{x}  :  \ottnt{A}    \vdash  \ottnt{b}  :   T_{ m_{{\mathrm{2}}} } \:  \ottnt{B}  $.\\
Now,
\begin{align*}
&\hspace*{10pt}  \llbracket    \ottnt{a}  \:  \leftindex^{ m_{{\mathrm{1}}} }{\gg}\!\! =^{ m'_{{\mathrm{2}}} }   (   \lambda  \ottmv{x}  :  \ottnt{A}  .   \mathbf{up}^{ m_{{\mathrm{2}}} , m'_{{\mathrm{2}}} }  \ottnt{b}    )     \rrbracket  \\
& =  \llbracket    \mathbf{join}^{ m_{{\mathrm{1}}} , m'_{{\mathrm{2}}} }   (    (   \mathbf{lift}^{ m_{{\mathrm{1}}} }   (   \lambda  \ottmv{x}  :  \ottnt{A}  .   \mathbf{up}^{ m_{{\mathrm{2}}} , m'_{{\mathrm{2}}} }  \ottnt{b}    )    )   \:  \ottnt{a}   )     \rrbracket  \\
& =   \mu^{ m_{{\mathrm{1}}} , m'_{{\mathrm{2}}} }_{  \llbracket  \ottnt{B}  \rrbracket  }   \circ   \llbracket     (   \mathbf{lift}^{ m_{{\mathrm{1}}} }   (   \lambda  \ottmv{x}  :  \ottnt{A}  .   \mathbf{up}^{ m_{{\mathrm{2}}} , m'_{{\mathrm{2}}} }  \ottnt{b}    )    )   \:  \ottnt{a}    \rrbracket   \\
& =    \mu^{ m_{{\mathrm{1}}} , m'_{{\mathrm{2}}} }_{  \llbracket  \ottnt{B}  \rrbracket  }   \circ   \text{app}    \circ   \langle   \Lambda \Big(     \mathbf{T}_{ m_{{\mathrm{1}}} }   (   \Lambda^{-1}   \llbracket    \lambda  \ottmv{x}  :  \ottnt{A}  .   \mathbf{up}^{ m_{{\mathrm{2}}} , m'_{{\mathrm{2}}} }  \ottnt{b}     \rrbracket    )    \circ   t^{\mathbf{T}_{ m_{{\mathrm{1}}} } }_{  \llbracket  \Gamma  \rrbracket  ,   \llbracket  \ottnt{A}  \rrbracket  }     \Big)   ,   \llbracket  \ottnt{a}  \rrbracket   \rangle   \\
& =     \mu^{ m_{{\mathrm{1}}} , m'_{{\mathrm{2}}} }_{  \llbracket  \ottnt{B}  \rrbracket  }   \circ   \mathbf{T}_{ m_{{\mathrm{1}}} }   (    \mathbf{T}^{ m_{{\mathrm{2}}}  \leq  m'_{{\mathrm{2}}} }_{  \llbracket  \ottnt{B}  \rrbracket  }   \circ   \llbracket  \ottnt{b}  \rrbracket    )     \circ   t^{\mathbf{T}_{ m_{{\mathrm{1}}} } }_{  \llbracket  \Gamma  \rrbracket  ,   \llbracket  \ottnt{A}  \rrbracket  }    \circ   \langle   \text{id}_{  \llbracket  \Gamma  \rrbracket  }   ,   \llbracket  \ottnt{a}  \rrbracket   \rangle   \\
& =      \mu^{ m_{{\mathrm{1}}} , m'_{{\mathrm{2}}} }_{  \llbracket  \ottnt{B}  \rrbracket  }   \circ   \mathbf{T}_{ m_{{\mathrm{1}}} }    \mathbf{T}^{ m_{{\mathrm{2}}}  \leq  m'_{{\mathrm{2}}} }_{  \llbracket  \ottnt{B}  \rrbracket  }      \circ   \mathbf{T}_{ m_{{\mathrm{1}}} }   \llbracket  \ottnt{b}  \rrbracket     \circ   t^{\mathbf{T}_{ m_{{\mathrm{1}}} } }_{  \llbracket  \Gamma  \rrbracket  ,   \llbracket  \ottnt{A}  \rrbracket  }    \circ   \langle   \text{id}_{  \llbracket  \Gamma  \rrbracket  }   ,   \llbracket  \ottnt{a}  \rrbracket   \rangle  .
\end{align*}
Next,
\begin{align*}
&\hspace*{10pt}  \llbracket    \mathbf{up}^{  m_{{\mathrm{1}}}  \cdot  m_{{\mathrm{2}}}  ,  m_{{\mathrm{1}}}  \cdot  m'_{{\mathrm{2}}}  }   (   \ottnt{a}  \:  \leftindex^{ m_{{\mathrm{1}}} }{\gg}\!\! =^{ m_{{\mathrm{2}}} }   \lambda  \ottmv{x}  :  \ottnt{A}  .  \ottnt{b}    )     \rrbracket  \\
& =   \mathbf{T}^{  m_{{\mathrm{1}}}  \cdot  m_{{\mathrm{2}}}   \leq   m_{{\mathrm{1}}}  \cdot  m'_{{\mathrm{2}}}  }_{  \llbracket  \ottnt{B}  \rrbracket  }   \circ   \llbracket    \mathbf{join}^{ m_{{\mathrm{1}}} , m_{{\mathrm{2}}} }   (    (   \mathbf{lift}^{ m_{{\mathrm{1}}} }   (   \lambda  \ottmv{x}  :  \ottnt{A}  .  \ottnt{b}   )    )   \:  \ottnt{a}   )     \rrbracket   \\
& =     \mathbf{T}^{  m_{{\mathrm{1}}}  \cdot  m_{{\mathrm{2}}}   \leq   m_{{\mathrm{1}}}  \cdot  m'_{{\mathrm{2}}}  }_{  \llbracket  \ottnt{B}  \rrbracket  }   \circ   \mu^{ m_{{\mathrm{1}}} , m_{{\mathrm{2}}} }_{  \llbracket  \ottnt{B}  \rrbracket  }    \circ   \text{app}    \circ   \langle   \Lambda \Big(     \mathbf{T}_{ m_{{\mathrm{1}}} }   (   \Lambda^{-1}   \llbracket    \lambda  \ottmv{x}  :  \ottnt{A}  .  \ottnt{b}    \rrbracket    )    \circ   t^{\mathbf{T}_{ m_{{\mathrm{1}}} } }_{  \llbracket  \Gamma  \rrbracket  ,   \llbracket  \ottnt{A}  \rrbracket  }     \Big)   ,   \llbracket  \ottnt{a}  \rrbracket   \rangle   \\
& =      \mathbf{T}^{  m_{{\mathrm{1}}}  \cdot  m_{{\mathrm{2}}}   \leq   m_{{\mathrm{1}}}  \cdot  m'_{{\mathrm{2}}}  }_{  \llbracket  \ottnt{B}  \rrbracket  }   \circ   \mu^{ m_{{\mathrm{1}}} , m_{{\mathrm{2}}} }_{  \llbracket  \ottnt{B}  \rrbracket  }    \circ   \mathbf{T}_{ m_{{\mathrm{1}}} }   \llbracket  \ottnt{b}  \rrbracket     \circ   t^{\mathbf{T}_{ m_{{\mathrm{1}}} } }_{  \llbracket  \Gamma  \rrbracket  ,   \llbracket  \ottnt{A}  \rrbracket  }    \circ   \langle   \text{id}_{  \llbracket  \Gamma  \rrbracket  }   ,   \llbracket  \ottnt{a}  \rrbracket   \rangle  . 
\end{align*}
The two morphisms above are equal because $\mu$ is natural in its second component.

\item $  (   \ottkw{ret}  \:  \ottnt{a}   )   \:  \leftindex^{ \ottsym{1} }{\gg}\!\! =^{ m }  \ottnt{f}  \equiv  \ottnt{f}  \:  \ottnt{a} $.\\
Given: $ \Gamma  \vdash    (   \ottkw{ret}  \:  \ottnt{a}   )   \:  \leftindex^{ \ottsym{1} }{\gg}\!\! =^{ m }  \ottnt{f}   :   T_{ m } \:  \ottnt{B}  $ and $ \Gamma  \vdash   \ottnt{f}  \:  \ottnt{a}   :   T_{ m } \:  \ottnt{B}  $ where $ \Gamma  \vdash  \ottnt{a}  :  \ottnt{A} $ and $ \Gamma  \vdash  \ottnt{f}  :   \ottnt{A}  \to   T_{ m } \:  \ottnt{B}   $.\\
Now,
\begin{align*}
&  \llbracket     (   \ottkw{ret}  \:  \ottnt{a}   )   \:  \leftindex^{ \ottsym{1} }{\gg}\!\! =^{ m }  \ottnt{f}    \rrbracket  \\
= &  \llbracket    \mathbf{join}^{ \ottsym{1} , m }   (    (   \mathbf{lift}^{ \ottsym{1} }  \ottnt{f}   )   \:   (   \ottkw{ret}  \:  \ottnt{a}   )    )     \rrbracket \\
= &    \mu^{ \ottsym{1} , m }_{  \llbracket  \ottnt{B}  \rrbracket  }   \circ   \text{app}    \circ   \langle   \llbracket   \mathbf{lift}^{ \ottsym{1} }  \ottnt{f}   \rrbracket   ,   \llbracket   \ottkw{ret}  \:  \ottnt{a}   \rrbracket   \rangle   \\
= &    \mu^{ \ottsym{1} , m }_{  \llbracket  \ottnt{B}  \rrbracket  }   \circ   \text{app}    \circ   \langle   \Lambda \Big(     \mathbf{T}_{ \ottsym{1} }   (   \Lambda^{-1}   \llbracket  \ottnt{f}  \rrbracket    )    \circ   t^{\mathbf{T}_{ \ottsym{1} } }_{  \llbracket  \Gamma  \rrbracket  ,   \llbracket  \ottnt{A}  \rrbracket  }     \Big)   ,    \eta_{  \llbracket  \ottnt{A}  \rrbracket  }   \circ   \llbracket  \ottnt{a}  \rrbracket    \rangle   \\
= &      \mu^{ \ottsym{1} , m }_{  \llbracket  \ottnt{B}  \rrbracket  }   \circ   \mathbf{T}_{ \ottsym{1} }   (   \Lambda^{-1}   \llbracket  \ottnt{f}  \rrbracket    )     \circ   t^{\mathbf{T}_{ \ottsym{1} } }_{  \llbracket  \Gamma  \rrbracket  ,   \llbracket  \ottnt{A}  \rrbracket  }    \circ   (    \text{id}_{  \llbracket  \Gamma  \rrbracket  }   \times   \eta_{  \llbracket  \ottnt{A}  \rrbracket  }    )    \circ   \langle   \text{id}_{  \llbracket  \Gamma  \rrbracket  }   ,   \llbracket  \ottnt{a}  \rrbracket   \rangle  . 
\end{align*}
And,
\begin{align*}
 \llbracket    \ottnt{f}  \:  \ottnt{a}    \rrbracket  =   \Lambda^{-1}   \llbracket  \ottnt{f}  \rrbracket    \circ   \langle   \text{id}_{  \llbracket  \Gamma  \rrbracket  }   ,   \llbracket  \ottnt{a}  \rrbracket   \rangle  .
\end{align*}
The above two morphisms are equal because the diagram in Figure \ref{prfcd1} commutes. The left square commutes because $\eta$ is a \textit{strong} natural transformation; the square to the right commutes because $\eta$ is a natural transformation; the triangle to the right commutes because $\mathbf{T}$ is a lax monoidal functor. Also, note $ t^{\textbf{Id} }_{  \llbracket  \Gamma  \rrbracket  ,  \llbracket  \ottnt{A}  \rrbracket  }  =  \text{id}_{    \llbracket  \Gamma  \rrbracket   \times   \llbracket  \ottnt{A}  \rrbracket    } $.

\begin{figure}
\begin{tikzcd}[row sep = 3.5 em, column sep = 4 em]
  \llbracket  \Gamma  \rrbracket   \times   \llbracket  \ottnt{A}  \rrbracket   \arrow{d}{  \text{id}_{  \llbracket  \Gamma  \rrbracket  }   \times   \eta_{  \llbracket  \ottnt{A}  \rrbracket  }  } \arrow{r}{ t^{\textbf{Id} }_{  \llbracket  \Gamma  \rrbracket  ,  \llbracket  \ottnt{A}  \rrbracket  } } & 
  \llbracket  \Gamma  \rrbracket   \times   \llbracket  \ottnt{A}  \rrbracket   \arrow{d}{ \eta_{    \llbracket  \Gamma  \rrbracket   \times   \llbracket  \ottnt{A}  \rrbracket    } } \arrow{r}{ \Lambda^{-1}   \llbracket  \ottnt{f}  \rrbracket  } & 
 \mathbf{T}_{ m }   \llbracket  \ottnt{B}  \rrbracket   \arrow{d}{ \eta_{   \mathbf{T}_{ m }   \llbracket  \ottnt{B}  \rrbracket    } } \arrow{rd}{ \text{id}_{   \mathbf{T}_{ m }   \llbracket  \ottnt{B}  \rrbracket    } }  & \\
  \llbracket  \Gamma  \rrbracket   \times   \mathbf{T}_{ \ottsym{1} }   \llbracket  \ottnt{A}  \rrbracket    \arrow{r}{ t^{\mathbf{T}_{ \ottsym{1} } }_{  \llbracket  \Gamma  \rrbracket  ,   \llbracket  \ottnt{A}  \rrbracket  } } &  \mathbf{T}_{ \ottsym{1} }   (    \llbracket  \Gamma  \rrbracket   \times   \llbracket  \ottnt{A}  \rrbracket    )   \arrow{r}{ \mathbf{T}_{ \ottsym{1} }   (   \Lambda^{-1}   \llbracket  \ottnt{f}  \rrbracket    )  } &  \mathbf{T}_{ \ottsym{1} }    \mathbf{T}_{ m }   \llbracket  \ottnt{B}  \rrbracket     \arrow{r}{ \mu^{ \ottsym{1} , m }_{  \llbracket  \ottnt{B}  \rrbracket  } } &  \mathbf{T}_{ m }   \llbracket  \ottnt{B}  \rrbracket  
\end{tikzcd}
\caption{Commutative diagram}
\label{prfcd1}
\end{figure} 

\item $ \ottnt{a}  \:  \leftindex^{ m_{{\mathrm{1}}} }{\gg}\!\! =^{ \ottsym{1} }   (   \lambda  \ottmv{x}  :  \ottnt{A}  .   \ottkw{ret}  \:  \ottmv{x}    )   \equiv \ottnt{a}$.\\
Given: $ \Gamma  \vdash   \ottnt{a}  \:  \leftindex^{ m_{{\mathrm{1}}} }{\gg}\!\! =^{ \ottsym{1} }   (   \lambda  \ottmv{x}  :  \ottnt{A}  .   \ottkw{ret}  \:  \ottmv{x}    )    :   T_{ m_{{\mathrm{1}}} } \:  \ottnt{A}  $ and $ \Gamma  \vdash  \ottnt{a}  :   T_{ m_{{\mathrm{1}}} } \:  \ottnt{A}  $.\\
Now,
\begin{align*}
&  \llbracket    \ottnt{a}  \:  \leftindex^{ m_{{\mathrm{1}}} }{\gg}\!\! =^{ \ottsym{1} }   (   \lambda  \ottmv{x}  :  \ottnt{A}  .   \ottkw{ret}  \:  \ottmv{x}    )     \rrbracket  \\
= &  \llbracket    \mathbf{join}^{ m_{{\mathrm{1}}} , \ottsym{1} }   (    (   \mathbf{lift}^{ m_{{\mathrm{1}}} }   (   \lambda  \ottmv{x}  :  \ottnt{A}  .   \ottkw{ret}  \:  \ottmv{x}    )    )   \:  \ottnt{a}   )     \rrbracket  \\
= &    \mu^{ m_{{\mathrm{1}}} , \ottsym{1} }_{  \llbracket  \ottnt{A}  \rrbracket  }   \circ   \text{app}    \circ   \langle   \Lambda \Big(     \mathbf{T}_{ m_{{\mathrm{1}}} }   (   \Lambda^{-1}   \llbracket    \lambda  \ottmv{x}  :  \ottnt{A}  .   \ottkw{ret}  \:  \ottmv{x}     \rrbracket    )    \circ   t^{\mathbf{T}_{ m_{{\mathrm{1}}} } }_{  \llbracket  \Gamma  \rrbracket  ,   \llbracket  \ottnt{A}  \rrbracket  }     \Big)   ,   \llbracket  \ottnt{a}  \rrbracket   \rangle   \\
= &     \mu^{ m_{{\mathrm{1}}} , \ottsym{1} }_{  \llbracket  \ottnt{A}  \rrbracket  }   \circ   \mathbf{T}_{ m_{{\mathrm{1}}} }   (    \eta_{  \llbracket  \ottnt{A}  \rrbracket  }   \circ   \pi_2    )     \circ   t^{\mathbf{T}_{ m_{{\mathrm{1}}} } }_{  \llbracket  \Gamma  \rrbracket  ,   \llbracket  \ottnt{A}  \rrbracket  }    \circ   \langle   \text{id}_{  \llbracket  \Gamma  \rrbracket  }   ,   \llbracket  \ottnt{a}  \rrbracket   \rangle   \\
= &      \mu^{ m_{{\mathrm{1}}} , \ottsym{1} }_{  \llbracket  \ottnt{A}  \rrbracket  }   \circ   \mathbf{T}_{ m_{{\mathrm{1}}} }   \eta_{  \llbracket  \ottnt{A}  \rrbracket  }     \circ   \mathbf{T}_{ m_{{\mathrm{1}}} }   \pi_2     \circ   t^{\mathbf{T}_{ m_{{\mathrm{1}}} } }_{  \llbracket  \Gamma  \rrbracket  ,   \llbracket  \ottnt{A}  \rrbracket  }    \circ   \langle   \text{id}_{  \llbracket  \Gamma  \rrbracket  }   ,   \llbracket  \ottnt{a}  \rrbracket   \rangle   \\
= &     \mu^{ m_{{\mathrm{1}}} , \ottsym{1} }_{  \llbracket  \ottnt{A}  \rrbracket  }   \circ   \mathbf{T}_{ m_{{\mathrm{1}}} }   \eta_{  \llbracket  \ottnt{A}  \rrbracket  }     \circ   \pi_2    \circ   \langle   \text{id}_{  \llbracket  \Gamma  \rrbracket  }   ,   \llbracket  \ottnt{a}  \rrbracket   \rangle   \hspace*{5pt} \text{[By naturality and strength]}\\
= &    \mu^{ m_{{\mathrm{1}}} , \ottsym{1} }_{  \llbracket  \ottnt{A}  \rrbracket  }   \circ   \mathbf{T}_{ m_{{\mathrm{1}}} }   \eta_{  \llbracket  \ottnt{A}  \rrbracket  }     \circ   \llbracket  \ottnt{a}  \rrbracket   \\
= &  \llbracket  \ottnt{a}  \rrbracket  \hspace*{5pt} [\because \mathbf{T} \text{ is a lax monoidal functor}].
\end{align*}

\item $  (   \ottnt{a}  \:  \leftindex^{ m_{{\mathrm{1}}} }{\gg}\!\! =^{ m_{{\mathrm{2}}} }  \ottnt{f}   )   \:  \leftindex^{  m_{{\mathrm{1}}}  \cdot  m_{{\mathrm{2}}}  }{\gg}\!\! =^{ m_{{\mathrm{3}}} }  \ottnt{g}  \equiv  \ottnt{a}  \:  \leftindex^{ m_{{\mathrm{1}}} }{\gg}\!\! =^{  m_{{\mathrm{2}}}  \cdot  m_{{\mathrm{3}}}  }   (   \lambda  \ottmv{x}  :  \ottnt{A}  .   (    \ottnt{f}  \:  \ottmv{x}   \:  \leftindex^{ m_{{\mathrm{2}}} }{\gg}\!\! =^{ m_{{\mathrm{3}}} }  \ottnt{g}   )    )  $.\\
Given: $ \Gamma  \vdash    (   \ottnt{a}  \:  \leftindex^{ m_{{\mathrm{1}}} }{\gg}\!\! =^{ m_{{\mathrm{2}}} }  \ottnt{f}   )   \:  \leftindex^{  m_{{\mathrm{1}}}  \cdot  m_{{\mathrm{2}}}  }{\gg}\!\! =^{ m_{{\mathrm{3}}} }  \ottnt{g}   :   T_{   m_{{\mathrm{1}}}  \cdot  m_{{\mathrm{2}}}   \cdot  m_{{\mathrm{3}}}  } \:  \ottnt{C}  $ and $ \Gamma  \vdash   \ottnt{a}  \:  \leftindex^{ m_{{\mathrm{1}}} }{\gg}\!\! =^{  m_{{\mathrm{2}}}  \cdot  m_{{\mathrm{3}}}  }   (   \lambda  \ottmv{x}  :  \ottnt{A}  .   (    \ottnt{f}  \:  \ottmv{x}   \:  \leftindex^{ m_{{\mathrm{2}}} }{\gg}\!\! =^{ m_{{\mathrm{3}}} }  \ottnt{g}   )    )    :   T_{   m_{{\mathrm{1}}}  \cdot  m_{{\mathrm{2}}}   \cdot  m_{{\mathrm{3}}}  } \:  \ottnt{C}  $ where $ \Gamma  \vdash  \ottnt{a}  :   T_{ m_{{\mathrm{1}}} } \:  \ottnt{A}  $ and $ \Gamma  \vdash  \ottnt{f}  :   \ottnt{A}  \to   T_{ m_{{\mathrm{2}}} } \:  \ottnt{B}   $ and $ \Gamma  \vdash  \ottnt{g}  :   \ottnt{B}  \to   T_{ m_{{\mathrm{3}}} } \:  \ottnt{C}   $.

\begin{figure}
\hspace*{-1cm}
\begin{tikzcd}[row sep = 3.5 em, column sep = 4.5 em]
  \llbracket  \Gamma  \rrbracket   \times   \mathbf{T}_{ m_{{\mathrm{1}}} }   \mathbf{T}_{ m_{{\mathrm{2}}} }   \llbracket  \ottnt{B}  \rrbracket     \arrow{d}{  \text{id}_{  \llbracket  \Gamma  \rrbracket  }   \times   \mu^{ m_{{\mathrm{1}}} , m_{{\mathrm{2}}} }_{  \llbracket  \ottnt{B}  \rrbracket  }  } \arrow{r}{  \mathbf{T}_{ m_{{\mathrm{1}}} }    t^{\mathbf{T}_{ m_{{\mathrm{2}}} } }_{  \llbracket  \Gamma  \rrbracket  ,   \llbracket  \ottnt{B}  \rrbracket  }     \circ   t^{\mathbf{T}_{ m_{{\mathrm{1}}} } }_{  \llbracket  \Gamma  \rrbracket  ,    \mathbf{T}_{ m_{{\mathrm{2}}} }   \llbracket  \ottnt{B}  \rrbracket    }  } &  \mathbf{T}_{ m_{{\mathrm{1}}} }   \mathbf{T}_{ m_{{\mathrm{2}}} }   (    \llbracket  \Gamma  \rrbracket   \times   \llbracket  \ottnt{B}  \rrbracket    )    \arrow{d}{ \mu^{ m_{{\mathrm{1}}} , m_{{\mathrm{2}}} }_{    \llbracket  \Gamma  \rrbracket   \times   \llbracket  \ottnt{B}  \rrbracket    } } \arrow{r}{ \mathbf{T}_{ m_{{\mathrm{1}}} }    \mathbf{T}_{ m_{{\mathrm{2}}} }   (   \Lambda^{-1}   \llbracket  \ottnt{g}  \rrbracket    )    } &  \mathbf{T}_{ m_{{\mathrm{1}}} }   \mathbf{T}_{ m_{{\mathrm{2}}} }   \mathbf{T}_{ m_{{\mathrm{3}}} }   \llbracket  \ottnt{C}  \rrbracket     \arrow{d}{ \mu^{ m_{{\mathrm{1}}} , m_{{\mathrm{2}}} }_{   \mathbf{T}_{ m_{{\mathrm{3}}} }   \llbracket  \ottnt{C}  \rrbracket    } } \arrow{r}{ \mathbf{T}_{ m_{{\mathrm{1}}} }    \mu^{ m_{{\mathrm{2}}} , m_{{\mathrm{3}}} }_{  \llbracket  \ottnt{C}  \rrbracket  }   } &  \mathbf{T}_{ m_{{\mathrm{1}}} }   \mathbf{T}_{  m_{{\mathrm{2}}}  \cdot  m_{{\mathrm{3}}}  }   \llbracket  \ottnt{C}  \rrbracket    \arrow{d}{ \mu^{ m_{{\mathrm{1}}} ,  m_{{\mathrm{2}}}  \cdot  m_{{\mathrm{3}}}  }_{  \llbracket  \ottnt{C}  \rrbracket  } } \\
  \llbracket  \Gamma  \rrbracket   \times   \mathbf{T}_{  m_{{\mathrm{1}}}  \cdot  m_{{\mathrm{2}}}  }   \llbracket  \ottnt{B}  \rrbracket    \arrow{r}{ t^{\mathbf{T}_{  m_{{\mathrm{1}}}  \cdot  m_{{\mathrm{2}}}  } }_{  \llbracket  \Gamma  \rrbracket  ,   \llbracket  \ottnt{B}  \rrbracket  } } &  \mathbf{T}_{  m_{{\mathrm{1}}}  \cdot  m_{{\mathrm{2}}}  }   (    \llbracket  \Gamma  \rrbracket   \times   \llbracket  \ottnt{B}  \rrbracket    )   \arrow{r}{ \mathbf{T}_{  m_{{\mathrm{1}}}  \cdot  m_{{\mathrm{2}}}  }   (   \Lambda^{-1}   \llbracket  \ottnt{g}  \rrbracket    )  } &  \mathbf{T}_{  m_{{\mathrm{1}}}  \cdot  m_{{\mathrm{2}}}  }   \mathbf{T}_{ m_{{\mathrm{3}}} }   \llbracket  \ottnt{C}  \rrbracket    \arrow{r}{ \mu^{  m_{{\mathrm{1}}}  \cdot  m_{{\mathrm{2}}}  , m_{{\mathrm{3}}} }_{  \llbracket  \ottnt{C}  \rrbracket  } } &  \mathbf{T}_{   m_{{\mathrm{1}}}  \cdot  m_{{\mathrm{2}}}   \cdot  m_{{\mathrm{3}}}  }   \llbracket  \ottnt{C}  \rrbracket  
\end{tikzcd}
\caption{Commutative diagram}
\label{prfcd2}
\end{figure}

Now,
\begin{align*}
&  \llbracket    \ottnt{a}  \:  \leftindex^{ m_{{\mathrm{1}}} }{\gg}\!\! =^{  m_{{\mathrm{2}}}  \cdot  m_{{\mathrm{3}}}  }   (   \lambda  \ottmv{x}  :  \ottnt{A}  .   (    \ottnt{f}  \:  \ottmv{x}   \:  \leftindex^{ m_{{\mathrm{2}}} }{\gg}\!\! =^{ m_{{\mathrm{3}}} }  \ottnt{g}   )    )     \rrbracket  \\
= &  \llbracket    \mathbf{join}^{ m_{{\mathrm{1}}} ,  m_{{\mathrm{2}}}  \cdot  m_{{\mathrm{3}}}  }   (    (   \mathbf{lift}^{ m_{{\mathrm{1}}} }   (   \lambda  \ottmv{x}  :  \ottnt{A}  .   (    \ottnt{f}  \:  \ottmv{x}   \:  \leftindex^{ m_{{\mathrm{2}}} }{\gg}\!\! =^{ m_{{\mathrm{3}}} }  \ottnt{g}   )    )    )   \:  \ottnt{a}   )     \rrbracket  \\
= &    \mu^{ m_{{\mathrm{1}}} ,  m_{{\mathrm{2}}}  \cdot  m_{{\mathrm{3}}}  }_{  \llbracket  \ottnt{C}  \rrbracket  }   \circ   \text{app}    \circ   \langle   \llbracket    \mathbf{lift}^{ m_{{\mathrm{1}}} }   (   \lambda  \ottmv{x}  :  \ottnt{A}  .   (    \ottnt{f}  \:  \ottmv{x}   \:  \leftindex^{ m_{{\mathrm{2}}} }{\gg}\!\! =^{ m_{{\mathrm{3}}} }  \ottnt{g}   )    )     \rrbracket   ,   \llbracket  \ottnt{a}  \rrbracket   \rangle   \\
= &    \mu^{ m_{{\mathrm{1}}} ,  m_{{\mathrm{2}}}  \cdot  m_{{\mathrm{3}}}  }_{  \llbracket  \ottnt{C}  \rrbracket  }   \circ   \text{app}    \circ   \langle   \Lambda \Big(     \mathbf{T}_{ m_{{\mathrm{1}}} }   (   \Lambda^{-1}   \llbracket    \lambda  \ottmv{x}  :  \ottnt{A}  .   (    \ottnt{f}  \:  \ottmv{x}   \:  \leftindex^{ m_{{\mathrm{2}}} }{\gg}\!\! =^{ m_{{\mathrm{3}}} }  \ottnt{g}   )     \rrbracket    )    \circ   t^{\mathbf{T}_{ m_{{\mathrm{1}}} } }_{  \llbracket  \Gamma  \rrbracket  ,   \llbracket  \ottnt{A}  \rrbracket  }     \Big)   ,   \llbracket  \ottnt{a}  \rrbracket   \rangle   \\
= &     \mu^{ m_{{\mathrm{1}}} ,  m_{{\mathrm{2}}}  \cdot  m_{{\mathrm{3}}}  }_{  \llbracket  \ottnt{C}  \rrbracket  }   \circ   \mathbf{T}_{ m_{{\mathrm{1}}} }   (   \Lambda^{-1}   \llbracket    \lambda  \ottmv{x}  :  \ottnt{A}  .   (    \ottnt{f}  \:  \ottmv{x}   \:  \leftindex^{ m_{{\mathrm{2}}} }{\gg}\!\! =^{ m_{{\mathrm{3}}} }  \ottnt{g}   )     \rrbracket    )     \circ   t^{\mathbf{T}_{ m_{{\mathrm{1}}} } }_{  \llbracket  \Gamma  \rrbracket  ,   \llbracket  \ottnt{A}  \rrbracket  }    \circ   \langle   \text{id}_{  \llbracket  \Gamma  \rrbracket  }   ,   \llbracket  \ottnt{a}  \rrbracket   \rangle   \\
= &     \mu^{ m_{{\mathrm{1}}} ,  m_{{\mathrm{2}}}  \cdot  m_{{\mathrm{3}}}  }_{  \llbracket  \ottnt{C}  \rrbracket  }   \circ   \mathbf{T}_{ m_{{\mathrm{1}}} }   \llbracket     \ottnt{f}  \:  \ottmv{x}   \:  \leftindex^{ m_{{\mathrm{2}}} }{\gg}\!\! =^{ m_{{\mathrm{3}}} }  \ottnt{g}    \rrbracket     \circ   t^{\mathbf{T}_{ m_{{\mathrm{1}}} } }_{  \llbracket  \Gamma  \rrbracket  ,   \llbracket  \ottnt{A}  \rrbracket  }    \circ   \langle   \text{id}_{  \llbracket  \Gamma  \rrbracket  }   ,   \llbracket  \ottnt{a}  \rrbracket   \rangle   \\
= &     \mu^{ m_{{\mathrm{1}}} ,  m_{{\mathrm{2}}}  \cdot  m_{{\mathrm{3}}}  }_{  \llbracket  \ottnt{C}  \rrbracket  }   \circ   \mathbf{T}_{ m_{{\mathrm{1}}} }   \llbracket    \mathbf{join}^{ m_{{\mathrm{2}}} , m_{{\mathrm{3}}} }   (    (   \mathbf{lift}^{ m_{{\mathrm{2}}} }  \ottnt{g}   )   \:   (   \ottnt{f}  \:  \ottmv{x}   )    )     \rrbracket     \circ   t^{\mathbf{T}_{ m_{{\mathrm{1}}} } }_{  \llbracket  \Gamma  \rrbracket  ,   \llbracket  \ottnt{A}  \rrbracket  }    \circ   \langle   \text{id}_{  \llbracket  \Gamma  \rrbracket  }   ,   \llbracket  \ottnt{a}  \rrbracket   \rangle   \\
= &     \mu^{ m_{{\mathrm{1}}} ,  m_{{\mathrm{2}}}  \cdot  m_{{\mathrm{3}}}  }_{  \llbracket  \ottnt{C}  \rrbracket  }   \circ   \mathbf{T}_{ m_{{\mathrm{1}}} }   (     \mu^{ m_{{\mathrm{2}}} , m_{{\mathrm{3}}} }_{  \llbracket  \ottnt{C}  \rrbracket  }   \circ   \text{app}    \circ   \langle   \llbracket    \mathbf{lift}^{ m_{{\mathrm{2}}} }  \ottnt{g}    \rrbracket   ,   \llbracket    \ottnt{f}  \:  \ottmv{x}    \rrbracket   \rangle    )     \circ   t^{\mathbf{T}_{ m_{{\mathrm{1}}} } }_{  \llbracket  \Gamma  \rrbracket  ,   \llbracket  \ottnt{A}  \rrbracket  }    \circ   \langle   \text{id}_{  \llbracket  \Gamma  \rrbracket  }   ,   \llbracket  \ottnt{a}  \rrbracket   \rangle   \\
= &      \mu^{ m_{{\mathrm{1}}} ,  m_{{\mathrm{2}}}  \cdot  m_{{\mathrm{3}}}  }_{  \llbracket  \ottnt{C}  \rrbracket  }   \circ   \mathbf{T}_{ m_{{\mathrm{1}}} }    \mu^{ m_{{\mathrm{2}}} , m_{{\mathrm{3}}} }_{  \llbracket  \ottnt{C}  \rrbracket  }      \circ   \mathbf{T}_{ m_{{\mathrm{1}}} }    \mathbf{T}_{ m_{{\mathrm{2}}} }   (    \Lambda^{-1}   \llbracket  \ottnt{g}  \rrbracket    \circ   (    \pi_1   \times   \text{id}    )    )       \circ   \mathbf{T}_{ m_{{\mathrm{1}}} }    t^{\mathbf{T}_{ m_{{\mathrm{2}}} } }_{    \llbracket  \Gamma  \rrbracket   \times   \llbracket  \ottnt{A}  \rrbracket    ,   \llbracket  \ottnt{B}  \rrbracket  }      \circ   \mathbf{T}_{ m_{{\mathrm{1}}} }   \langle   \text{id}   ,   \Lambda^{-1}   \llbracket  \ottnt{f}  \rrbracket    \rangle     \\ & \hspace*{30pt} \circ   t^{\mathbf{T}_{ m_{{\mathrm{1}}} } }_{  \llbracket  \Gamma  \rrbracket  ,   \llbracket  \ottnt{A}  \rrbracket  }   \circ   \langle   \text{id}   ,   \llbracket  \ottnt{a}  \rrbracket   \rangle  \\
= &       \mu^{ m_{{\mathrm{1}}} ,  m_{{\mathrm{2}}}  \cdot  m_{{\mathrm{3}}}  }_{  \llbracket  \ottnt{C}  \rrbracket  }   \circ   \mathbf{T}_{ m_{{\mathrm{1}}} }    \mu^{ m_{{\mathrm{2}}} , m_{{\mathrm{3}}} }_{  \llbracket  \ottnt{C}  \rrbracket  }      \circ   \mathbf{T}_{ m_{{\mathrm{1}}} }    \mathbf{T}_{ m_{{\mathrm{2}}} }   (   \Lambda^{-1}   \llbracket  \ottnt{g}  \rrbracket    )       \circ   \mathbf{T}_{ m_{{\mathrm{1}}} }    t^{\mathbf{T}_{ m_{{\mathrm{2}}} } }_{  \llbracket  \Gamma  \rrbracket  ,   \llbracket  \ottnt{B}  \rrbracket  }      \circ   \mathbf{T}_{ m_{{\mathrm{1}}} }   (    \pi_1   \times   \text{id}    )     \circ   \mathbf{T}_{ m_{{\mathrm{1}}} }   \langle   \text{id}   ,   \Lambda^{-1}   \llbracket  \ottnt{f}  \rrbracket    \rangle    \\ & \hspace*{30pt} \circ   t^{\mathbf{T}_{ m_{{\mathrm{1}}} } }_{  \llbracket  \Gamma  \rrbracket  ,   \llbracket  \ottnt{A}  \rrbracket  }   \circ   \langle   \text{id}   ,   \llbracket  \ottnt{a}  \rrbracket   \rangle   \hspace*{7pt}
 [\text{By naturality of }t^{\mathbf{T}_{m_2}} ].
\end{align*}

Next,
\begin{align*}
&  \llbracket     (   \ottnt{a}  \:  \leftindex^{ m_{{\mathrm{1}}} }{\gg}\!\! =^{ m_{{\mathrm{2}}} }  \ottnt{f}   )   \:  \leftindex^{  m_{{\mathrm{1}}}  \cdot  m_{{\mathrm{2}}}  }{\gg}\!\! =^{ m_{{\mathrm{3}}} }  \ottnt{g}    \rrbracket  \\
= &  \llbracket    \mathbf{join}^{  m_{{\mathrm{1}}}  \cdot  m_{{\mathrm{2}}}  , m_{{\mathrm{3}}} }   (    (   \mathbf{lift}^{  m_{{\mathrm{1}}}  \cdot  m_{{\mathrm{2}}}  }  \ottnt{g}   )   \:   (   \ottnt{a}  \:  \leftindex^{ m_{{\mathrm{1}}} }{\gg}\!\! =^{ m_{{\mathrm{2}}} }  \ottnt{f}   )    )     \rrbracket \\
= &    \mu^{  m_{{\mathrm{1}}}  \cdot  m_{{\mathrm{2}}}  , m_{{\mathrm{3}}} }_{  \llbracket  \ottnt{C}  \rrbracket  }   \circ   \text{app}    \circ   \langle   \llbracket    \mathbf{lift}^{  m_{{\mathrm{1}}}  \cdot  m_{{\mathrm{2}}}  }  \ottnt{g}    \rrbracket   ,   \llbracket    \ottnt{a}  \:  \leftindex^{ m_{{\mathrm{1}}} }{\gg}\!\! =^{ m_{{\mathrm{2}}} }  \ottnt{f}    \rrbracket   \rangle   \\
= &    \mu^{  m_{{\mathrm{1}}}  \cdot  m_{{\mathrm{2}}}  , m_{{\mathrm{3}}} }_{  \llbracket  \ottnt{C}  \rrbracket  }   \circ   \text{app}    \circ   \langle   \Lambda \Big(     \mathbf{T}_{  m_{{\mathrm{1}}}  \cdot  m_{{\mathrm{2}}}  }   (   \Lambda^{-1}   \llbracket  \ottnt{g}  \rrbracket    )    \circ   t^{\mathbf{T}_{  m_{{\mathrm{1}}}  \cdot  m_{{\mathrm{2}}}  } }_{  \llbracket  \Gamma  \rrbracket  ,   \llbracket  \ottnt{B}  \rrbracket  }     \Big)   ,   \llbracket    \ottnt{a}  \:  \leftindex^{ m_{{\mathrm{1}}} }{\gg}\!\! =^{ m_{{\mathrm{2}}} }  \ottnt{f}    \rrbracket   \rangle   \\
= &     \mu^{  m_{{\mathrm{1}}}  \cdot  m_{{\mathrm{2}}}  , m_{{\mathrm{3}}} }_{  \llbracket  \ottnt{C}  \rrbracket  }   \circ   \mathbf{T}_{  m_{{\mathrm{1}}}  \cdot  m_{{\mathrm{2}}}  }   (   \Lambda^{-1}   \llbracket  \ottnt{g}  \rrbracket    )     \circ   t^{\mathbf{T}_{  m_{{\mathrm{1}}}  \cdot  m_{{\mathrm{2}}}  } }_{  \llbracket  \Gamma  \rrbracket  ,   \llbracket  \ottnt{B}  \rrbracket  }    \circ   \langle   \text{id}_{  \llbracket  \Gamma  \rrbracket  }   ,   \llbracket    \mathbf{join}^{ m_{{\mathrm{1}}} , m_{{\mathrm{2}}} }   (    (   \mathbf{lift}^{ m_{{\mathrm{1}}} }  \ottnt{f}   )   \:  \ottnt{a}   )     \rrbracket   \rangle   \\
= &    \mu^{  m_{{\mathrm{1}}}  \cdot  m_{{\mathrm{2}}}  , m_{{\mathrm{3}}} }_{  \llbracket  \ottnt{C}  \rrbracket  }   \circ   \mathbf{T}_{  m_{{\mathrm{1}}}  \cdot  m_{{\mathrm{2}}}  }   (   \Lambda^{-1}   \llbracket  \ottnt{g}  \rrbracket    )     \circ   t^{\mathbf{T}_{  m_{{\mathrm{1}}}  \cdot  m_{{\mathrm{2}}}  } }_{  \llbracket  \Gamma  \rrbracket  ,   \llbracket  \ottnt{B}  \rrbracket  }   \\ & \hspace*{30pt} \circ  \langle   \text{id}_{  \llbracket  \Gamma  \rrbracket  }   ,      \mu^{ m_{{\mathrm{1}}} , m_{{\mathrm{2}}} }_{  \llbracket  \ottnt{B}  \rrbracket  }   \circ   \mathbf{T}_{ m_{{\mathrm{1}}} }   (   \Lambda^{-1}   \llbracket  \ottnt{f}  \rrbracket    )     \circ   t^{\mathbf{T}_{ m_{{\mathrm{1}}} } }_{  \llbracket  \Gamma  \rrbracket  ,   \llbracket  \ottnt{A}  \rrbracket  }    \circ   \langle   \text{id}_{  \llbracket  \Gamma  \rrbracket  }   ,   \llbracket  \ottnt{a}  \rrbracket   \rangle    \rangle  \\
= &     \mu^{  m_{{\mathrm{1}}}  \cdot  m_{{\mathrm{2}}}  , m_{{\mathrm{3}}} }_{  \llbracket  \ottnt{C}  \rrbracket  }   \circ   \mathbf{T}_{  m_{{\mathrm{1}}}  \cdot  m_{{\mathrm{2}}}  }   (   \Lambda^{-1}   \llbracket  \ottnt{g}  \rrbracket    )     \circ   t^{\mathbf{T}_{  m_{{\mathrm{1}}}  \cdot  m_{{\mathrm{2}}}  } }_{  \llbracket  \Gamma  \rrbracket  ,   \llbracket  \ottnt{B}  \rrbracket  }    \circ   (    \text{id}_{  \llbracket  \Gamma  \rrbracket  }   \times   \mu^{ m_{{\mathrm{1}}} , m_{{\mathrm{2}}} }_{  \llbracket  \ottnt{B}  \rrbracket  }    )   \\ & \hspace*{30pt} \circ  \langle   \text{id}_{  \llbracket  \Gamma  \rrbracket  }   ,     \mathbf{T}_{ m_{{\mathrm{1}}} }   (   \Lambda^{-1}   \llbracket  \ottnt{f}  \rrbracket    )    \circ   t^{\mathbf{T}_{ m_{{\mathrm{1}}} } }_{  \llbracket  \Gamma  \rrbracket  ,   \llbracket  \ottnt{A}  \rrbracket  }    \circ   \langle   \text{id}_{  \llbracket  \Gamma  \rrbracket  }   ,   \llbracket  \ottnt{a}  \rrbracket   \rangle    \rangle  \\
= &      \mu^{  m_{{\mathrm{1}}}  \cdot  m_{{\mathrm{2}}}  , m_{{\mathrm{3}}} }_{  \llbracket  \ottnt{C}  \rrbracket  }   \circ   \mathbf{T}_{  m_{{\mathrm{1}}}  \cdot  m_{{\mathrm{2}}}  }   (   \Lambda^{-1}   \llbracket  \ottnt{g}  \rrbracket    )     \circ   t^{\mathbf{T}_{  m_{{\mathrm{1}}}  \cdot  m_{{\mathrm{2}}}  } }_{  \llbracket  \Gamma  \rrbracket  ,   \llbracket  \ottnt{B}  \rrbracket  }    \circ   (    \text{id}   \times   \mu^{ m_{{\mathrm{1}}} , m_{{\mathrm{2}}} }_{  \llbracket  \ottnt{B}  \rrbracket  }    )    \circ   \Big(    \text{id}   \times   (    \mathbf{T}_{ m_{{\mathrm{1}}} }   (   \Lambda^{-1}   \llbracket  \ottnt{f}  \rrbracket    )    \circ   t^{\mathbf{T}_{ m_{{\mathrm{1}}} } }_{  \llbracket  \Gamma  \rrbracket  ,   \llbracket  \ottnt{A}  \rrbracket  }    )    \Big)   \\ & \hspace*{30pt} \circ   \langle   \pi_1   ,   \text{id}   \rangle   \circ   \langle   \text{id}   ,   \llbracket  \ottnt{a}  \rrbracket   \rangle   \\
= &      \mu^{ m_{{\mathrm{1}}} ,  m_{{\mathrm{2}}}  \cdot  m_{{\mathrm{3}}}  }_{  \llbracket  \ottnt{C}  \rrbracket  }   \circ   \mathbf{T}_{ m_{{\mathrm{1}}} }    \mu^{ m_{{\mathrm{2}}} , m_{{\mathrm{3}}} }_{  \llbracket  \ottnt{C}  \rrbracket  }      \circ   \mathbf{T}_{ m_{{\mathrm{1}}} }    \mathbf{T}_{ m_{{\mathrm{2}}} }   (   \Lambda^{-1}   \llbracket  \ottnt{g}  \rrbracket    )       \circ   \mathbf{T}_{ m_{{\mathrm{1}}} }    t^{\mathbf{T}_{ m_{{\mathrm{2}}} } }_{  \llbracket  \Gamma  \rrbracket  ,   \llbracket  \ottnt{B}  \rrbracket  }      \circ   t^{\mathbf{T}_{ m_{{\mathrm{1}}} } }_{  \llbracket  \Gamma  \rrbracket  ,    \mathbf{T}_{ m_{{\mathrm{2}}} }   \llbracket  \ottnt{B}  \rrbracket    }   \\
& \hspace{30pt} \circ    \Big(    \text{id}   \times   (    \mathbf{T}_{ m_{{\mathrm{1}}} }   (   \Lambda^{-1}   \llbracket  \ottnt{f}  \rrbracket    )    \circ   t^{\mathbf{T}_{ m_{{\mathrm{1}}} } }_{  \llbracket  \Gamma  \rrbracket  ,   \llbracket  \ottnt{A}  \rrbracket  }    )    \Big)   \circ   \langle   \pi_1   ,   \text{id}   \rangle    \circ   \langle   \text{id}   ,   \llbracket  \ottnt{a}  \rrbracket   \rangle   \hspace*{5pt} \text{[By Figure \ref{prfcd2}]}
\end{align*}

\vspace*{10pt}

\begin{figure}
\begin{tikzcd}[row sep = 3 em, column sep = 3 em]
  \llbracket  \Gamma  \rrbracket   \times   \mathbf{T}_{ m_{{\mathrm{1}}} }   \llbracket  \ottnt{A}  \rrbracket    \arrow{r}{ t^{\mathbf{T}_{ m_{{\mathrm{1}}} } }_{  \llbracket  \Gamma  \rrbracket  ,   \llbracket  \ottnt{A}  \rrbracket  } } \arrow{d}{ \langle   \pi_1   ,   t^{\mathbf{T}_{ m_{{\mathrm{1}}} } }_{  \llbracket  \Gamma  \rrbracket  ,   \llbracket  \ottnt{A}  \rrbracket  }   \rangle }  &  \mathbf{T}_{ m_{{\mathrm{1}}} }   (    \llbracket  \Gamma  \rrbracket   \times   \llbracket  \ottnt{A}  \rrbracket    )   \arrow{r}{ \mathbf{T}_{ m_{{\mathrm{1}}} }   \langle   \text{id}   ,   \Lambda^{-1}   \llbracket  \ottnt{f}  \rrbracket    \rangle  } \arrow{d}{ \mathbf{T}_{ m_{{\mathrm{1}}} }   \langle   \pi_1   ,   \text{id}   \rangle  } &  \mathbf{T}_{ m_{{\mathrm{1}}} }   (    (    \llbracket  \Gamma  \rrbracket   \times   \llbracket  \ottnt{A}  \rrbracket    )   \times   \mathbf{T}_{ m_{{\mathrm{2}}} }   \llbracket  \ottnt{B}  \rrbracket     )   \arrow{d}{ \mathbf{T}_{ m_{{\mathrm{1}}} }   \langle    \pi_1   \circ   \pi_1    ,   \pi_2   \rangle  } \\
  \llbracket  \Gamma  \rrbracket   \times   \mathbf{T}_{ m_{{\mathrm{1}}} }   (    \llbracket  \Gamma  \rrbracket   \times   \llbracket  \ottnt{A}  \rrbracket    )    \arrow{r}{ t^{\mathbf{T}_{ m_{{\mathrm{1}}} } }_{  \llbracket  \Gamma  \rrbracket  ,     \llbracket  \Gamma  \rrbracket   \times   \llbracket  \ottnt{A}  \rrbracket    } } \arrow{dr}[below left]{  \text{id}   \times   \mathbf{T}_{ m_{{\mathrm{1}}} }   (   \Lambda^{-1}   \llbracket  \ottnt{f}  \rrbracket    )   } &  \mathbf{T}_{ m_{{\mathrm{1}}} }   (    \llbracket  \Gamma  \rrbracket   \times   (    \llbracket  \Gamma  \rrbracket   \times   \llbracket  \ottnt{A}  \rrbracket    )    )   \arrow{r}{ \mathbf{T}_{ m_{{\mathrm{1}}} }   (    \text{id}   \times   \Lambda^{-1}   \llbracket  \ottnt{f}  \rrbracket     )  } &  \mathbf{T}_{ m_{{\mathrm{1}}} }   (    \llbracket  \Gamma  \rrbracket   \times   \mathbf{T}_{ m_{{\mathrm{2}}} }   \llbracket  \ottnt{B}  \rrbracket     )   \\
&   \llbracket  \Gamma  \rrbracket   \times   \mathbf{T}_{ m_{{\mathrm{1}}} }   \mathbf{T}_{ m_{{\mathrm{2}}} }   \llbracket  \ottnt{B}  \rrbracket     \arrow{ur}[below right]{  t^{\mathbf{T}_{ m_{{\mathrm{1}}} } }_{  \llbracket  \Gamma  \rrbracket  ,    \mathbf{T}_{ m_{{\mathrm{2}}} }   \llbracket  \ottnt{B}  \rrbracket    }  } &
\end{tikzcd}
\caption{Commutative diagram}
\label{prfcd3}
\end{figure}

The diagram in Figure \ref{prfcd2} commutes: the leftmost square commutes because $\mu^{m_1,m_2}$ is a \textit{strong} natural transformation; the middle one commutes because $\mu^{m_1,m_2}$ is a natural transformation; the rightmost one commutes because $\mathbf{T}$ is a lax monoidal functor.\\

Next, to show that the above morphisms are equal, we just need to show that:

\begin{center}
$   t^{\mathbf{T}_{ m_{{\mathrm{1}}} } }_{  \llbracket  \Gamma  \rrbracket  ,    \mathbf{T}_{ m_{{\mathrm{2}}} }   \llbracket  \ottnt{B}  \rrbracket    }   \circ   \Big(    \text{id}   \times   (    \mathbf{T}_{ m_{{\mathrm{1}}} }   (   \Lambda^{-1}   \llbracket  \ottnt{f}  \rrbracket    )    \circ   t^{\mathbf{T}_{ m_{{\mathrm{1}}} } }_{  \llbracket  \Gamma  \rrbracket  ,   \llbracket  \ottnt{A}  \rrbracket  }    )    \Big)    \circ   \langle   \pi_1   ,   \text{id}   \rangle   =    \mathbf{T}_{ m_{{\mathrm{1}}} }   (    \pi_1   \times   \text{id}    )    \circ   \mathbf{T}_{ m_{{\mathrm{1}}} }   \langle   \text{id}   ,   \Lambda^{-1}   \llbracket  \ottnt{f}  \rrbracket    \rangle     \circ   t^{\mathbf{T}_{ m_{{\mathrm{1}}} } }_{  \llbracket  \Gamma  \rrbracket  ,   \llbracket  \ottnt{A}  \rrbracket  }  $
\end{center}

These two morphisms are equal by the commutative diagram in Figure \ref{prfcd3}. The diagram in Figure \ref{prfcd3} commutes: the bottom subfigure commutes by naturality of $t^{\mathbf{T}_{m_1}}$; the right one commutes by properties of products; the left one commutes by the commutativity of the diagram in Figure \ref{prfcd0}.
\end{itemize}
\end{proof}

\section{Proofs of lemmas/theorems stated in Section \ref{secdccgmc}}

\begin{theorem}[Theorem \ref{GMCtoDCCTh}] \label{GMC2DCC}
If $ \Gamma  \vdash  \ottnt{a}  :  \ottnt{A} $ in GMC($ \mathcal{L} $), then $  \overline{  \Gamma  }   \vdash   \overline{ \ottnt{a} }   :   \overline{ \ottnt{A} }  $ in DCC($ \mathcal{L} $). Further, if $ \Gamma  \vdash  \ottnt{a_{{\mathrm{1}}}}  :  \ottnt{A} $ and $ \Gamma  \vdash  \ottnt{a_{{\mathrm{2}}}}  :  \ottnt{A} $ such that $\ottnt{a_{{\mathrm{1}}}} \equiv \ottnt{a_{{\mathrm{2}}}}$ in GMC($ \mathcal{L} $), then $ \overline{ \ottnt{a_{{\mathrm{1}}}} }  \simeq  \overline{ \ottnt{a_{{\mathrm{2}}}} } $ in DCC($ \mathcal{L} $).
\end{theorem}

\begin{proof}
By induction on $ \Gamma  \vdash  \ottnt{a}  :  \ottnt{A} $.
\begin{itemize}
\item $\lambda$-calculus. By IH.
\item \Rref{M-Return}. Have: $ \Gamma  \vdash   \ottkw{ret}  \:  \ottnt{a}   :   T_{   \bot   } \:  \ottnt{A}  $ where $ \Gamma  \vdash  \ottnt{a}  :  \ottnt{A} $.\\
By IH, $  \overline{  \Gamma  }   \vdash   \overline{ \ottnt{a} }   :   \overline{ \ottnt{A} }  $. Therefore, $  \overline{  \Gamma  }   \vdash   \mathbf{eta} ^{  \bot  }   \overline{ \ottnt{a} }    :   \mathcal{T}_{  \bot  } \:   \overline{ \ottnt{A} }   $.
\item \Rref{M-Fmap}. Have: $ \Gamma  \vdash   \mathbf{lift}^{  \ell  }  \ottnt{f}   :    T_{  \ell  } \:  \ottnt{A}   \to   T_{  \ell  } \:  \ottnt{B}   $ where $ \Gamma  \vdash  \ottnt{f}  :   \ottnt{A}  \to  \ottnt{B}  $.\\
By IH, $  \overline{  \Gamma  }   \vdash   \overline{ \ottnt{f} }   :    \overline{ \ottnt{A} }   \to   \overline{ \ottnt{B} }   $.\\
Now, $$\infer[\textsc{(Lam)}]{  \overline{  \Gamma  }   \vdash   \lambda  \ottmv{x}  :   \mathcal{T}_{ \ell } \:   \overline{ \ottnt{A} }    .   \mathbf{bind} ^{ \ell } \:  \ottmv{y}  =  \ottmv{x}  \: \mathbf{in} \:   \mathbf{eta} ^{ \ell }   (    \overline{ \ottnt{f} }   \:  \ottmv{y}   )      :    \mathcal{T}_{ \ell } \:   \overline{ \ottnt{A} }    \to   \mathcal{T}_{ \ell } \:   \overline{ \ottnt{B} }    }
          {\infer[\textsc{(Bind)}]{   \overline{  \Gamma  }   ,   \ottmv{x}  :   \mathcal{T}_{ \ell } \:   \overline{ \ottnt{A} }      \vdash   \mathbf{bind} ^{ \ell } \:  \ottmv{y}  =  \ottmv{x}  \: \mathbf{in} \:   \mathbf{eta} ^{ \ell }   (    \overline{ \ottnt{f} }   \:  \ottmv{y}   )     :   \mathcal{T}_{ \ell } \:   \overline{ \ottnt{B} }   }
           {\infer[\textsc{(Var)}]{   \overline{  \Gamma  }   ,   \ottmv{x}  :   \mathcal{T}_{ \ell } \:   \overline{ \ottnt{A} }      \vdash  \ottmv{x}  :   \mathcal{T}_{ \ell } \:   \overline{ \ottnt{A} }   }{} &
            \infer[\textsc{(Eta)}]{   \overline{  \Gamma  }   ,     \ottmv{x}  :   \mathcal{T}_{ \ell } \:   \overline{ \ottnt{A} }     ,   \ottmv{y}  :   \overline{ \ottnt{A} }       \vdash   \mathbf{eta} ^{ \ell }   (    \overline{ \ottnt{f} }   \:  \ottmv{y}   )    :   \mathcal{T}_{ \ell } \:   \overline{ \ottnt{B} }   }
             {\infer[\textsc{(App)}]{   \overline{  \Gamma  }   ,     \ottmv{x}  :   \mathcal{T}_{ \ell } \:   \overline{ \ottnt{A} }     ,   \ottmv{y}  :   \overline{ \ottnt{A} }       \vdash    \overline{ \ottnt{f} }   \:  \ottmv{y}   :   \overline{ \ottnt{B} }  }{} } 
                & \infer[\textsc{(Prot-Monad)}]{ \ell  \sqsubseteq   \mathcal{T}_{ \ell } \:   \overline{ \ottnt{B} }   }{ \ell  \sqsubseteq  \ell } } }
      $$
\item \Rref{M-Join}. Have $ \Gamma  \vdash   \mathbf{join}^{  \ell_{{\mathrm{1}}}  ,  \ell_{{\mathrm{2}}}  }  \ottnt{a}   :   T_{   \ell_{{\mathrm{1}}}  \vee  \ell_{{\mathrm{2}}}   } \:  \ottnt{A}  $ where $ \Gamma  \vdash  \ottnt{a}  :   T_{  \ell_{{\mathrm{1}}}  } \:   T_{  \ell_{{\mathrm{2}}}  } \:  \ottnt{A}   $.\\
By IH, $  \overline{  \Gamma  }   \vdash   \overline{ \ottnt{a} }   :   \mathcal{T}_{ \ell_{{\mathrm{1}}} } \:   \mathcal{T}_{ \ell_{{\mathrm{2}}} } \:   \overline{ \ottnt{A} }    $.\\
Now, $$\infer[\textsc{(Bind)}]{  \overline{  \Gamma  }   \vdash   \mathbf{bind} ^{ \ell_{{\mathrm{1}}} } \:  \ottmv{x}  =   \overline{ \ottnt{a} }   \: \mathbf{in} \:   \mathbf{bind} ^{ \ell_{{\mathrm{2}}} } \:  \ottmv{y}  =  \ottmv{x}  \: \mathbf{in} \:   \mathbf{eta} ^{  \ell_{{\mathrm{1}}}  \vee  \ell_{{\mathrm{2}}}  }  \ottmv{y}     :   \mathcal{T}_{  \ell_{{\mathrm{1}}}  \vee  \ell_{{\mathrm{2}}}  } \:   \overline{ \ottnt{A} }   }
        {\infer[\textsc{(IH)}]{  \overline{  \Gamma  }   \vdash   \overline{ \ottnt{a} }   :   \mathcal{T}_{ \ell_{{\mathrm{1}}} } \:   \mathcal{T}_{ \ell_{{\mathrm{2}}} } \:   \overline{ \ottnt{A} }    }{}
         &
         \infer[\textsc{(Bind)}]{   \overline{  \Gamma  }   ,   \ottmv{x}  :   \mathcal{T}_{ \ell_{{\mathrm{2}}} } \:   \overline{ \ottnt{A} }      \vdash   \mathbf{bind} ^{ \ell_{{\mathrm{2}}} } \:  \ottmv{y}  =  \ottmv{x}  \: \mathbf{in} \:   \mathbf{eta} ^{  \ell_{{\mathrm{1}}}  \vee  \ell_{{\mathrm{2}}}  }  \ottmv{y}    :   \mathcal{T}_{  \ell_{{\mathrm{1}}}  \vee  \ell_{{\mathrm{2}}}  } \:   \overline{ \ottnt{A} }   }
            {\infer[\textsc{(Var)}]{   \overline{  \Gamma  }   ,   \ottmv{x}  :   \mathcal{T}_{ \ell_{{\mathrm{2}}} } \:   \overline{ \ottnt{A} }      \vdash  \ottmv{x}  :   \mathcal{T}_{ \ell_{{\mathrm{2}}} } \:   \overline{ \ottnt{A} }   }{}
             &
             \infer[\textsc{(Eta)}]{   \overline{  \Gamma  }   ,     \ottmv{x}  :   \mathcal{T}_{ \ell_{{\mathrm{2}}} } \:   \overline{ \ottnt{A} }     ,   \ottmv{y}  :   \overline{ \ottnt{A} }       \vdash   \mathbf{eta} ^{  \ell_{{\mathrm{1}}}  \vee  \ell_{{\mathrm{2}}}  }  \ottmv{y}   :   \mathcal{T}_{  \ell_{{\mathrm{1}}}  \vee  \ell_{{\mathrm{2}}}  } \:   \overline{ \ottnt{A} }   }
              {\infer[\textsc{(Var)}]{   \overline{  \Gamma  }   ,     \ottmv{x}  :   \mathcal{T}_{ \ell_{{\mathrm{2}}} } \:   \overline{ \ottnt{A} }     ,   \ottmv{y}  :   \overline{ \ottnt{A} }       \vdash  \ottmv{y}  :   \overline{ \ottnt{A} }  }{} }
             }
         }
      $$
The above derivation uses the judgements $ \ell_{{\mathrm{2}}}  \sqsubseteq   \mathcal{T}_{  \ell_{{\mathrm{1}}}  \vee  \ell_{{\mathrm{2}}}  } \:   \overline{ \ottnt{A} }   $ and $ \ell_{{\mathrm{1}}}  \sqsubseteq   \mathcal{T}_{  \ell_{{\mathrm{1}}}  \vee  \ell_{{\mathrm{2}}}  } \:   \overline{ \ottnt{A} }   $ on the first and the second applications of the bind rule respectively.

\item \Rref{M-Up}. Have: $ \Gamma  \vdash   \mathbf{up}^{  \ell_{{\mathrm{1}}}  ,  \ell_{{\mathrm{2}}}  }  \ottnt{a}   :   T_{  \ell_{{\mathrm{2}}}  } \:  \ottnt{A}  $ where $ \Gamma  \vdash  \ottnt{a}  :   T_{  \ell_{{\mathrm{1}}}  } \:  \ottnt{A}  $ and $ \ell_{{\mathrm{1}}}  \sqsubseteq  \ell_{{\mathrm{2}}} $.\\
By IH, $  \overline{  \Gamma  }   \vdash   \overline{ \ottnt{a} }   :   \mathcal{T}_{ \ell_{{\mathrm{1}}} } \:   \overline{ \ottnt{A} }   $.\\
Now, $$\infer[\textsc{(Bind)}]{  \overline{  \Gamma  }   \vdash   \mathbf{bind} ^{ \ell_{{\mathrm{1}}} } \:  \ottmv{x}  =   \overline{ \ottnt{a} }   \: \mathbf{in} \:   \mathbf{eta} ^{ \ell_{{\mathrm{2}}} }  \ottmv{x}    :   \mathcal{T}_{ \ell_{{\mathrm{2}}} } \:   \overline{ \ottnt{A} }   }
                                {\infer[\textsc{(IH)}]{  \overline{  \Gamma  }   \vdash   \overline{ \ottnt{a} }   :   \mathcal{T}_{ \ell_{{\mathrm{1}}} } \:   \overline{ \ottnt{A} }   }{}
                                 &
                                 \infer[\textsc{(Eta)}]{   \overline{  \Gamma  }   ,   \ottmv{x}  :   \overline{ \ottnt{A} }     \vdash   \mathbf{eta} ^{ \ell_{{\mathrm{2}}} }  \ottmv{x}   :   \mathcal{T}_{ \ell_{{\mathrm{2}}} } \:   \overline{ \ottnt{A} }   }{   \overline{  \Gamma  }   ,   \ottmv{x}  :   \overline{ \ottnt{A} }     \vdash  \ottmv{x}  :   \overline{ \ottnt{A} }  } 
                                 &
                                 \infer[\textsc{(Prot-Monad)}]{ \ell_{{\mathrm{1}}}  \sqsubseteq   \mathcal{T}_{ \ell_{{\mathrm{2}}} } \:   \overline{ \ottnt{A} }   }{ \ell_{{\mathrm{1}}}  \sqsubseteq  \ell_{{\mathrm{2}}} }
                                 }
   $$.
\end{itemize}

Now, for $ \Gamma  \vdash  \ottnt{a_{{\mathrm{1}}}}  :  \ottnt{A} $ and $ \Gamma  \vdash  \ottnt{a_{{\mathrm{2}}}}  :  \ottnt{A} $, if $\ottnt{a_{{\mathrm{1}}}} \equiv \ottnt{a_{{\mathrm{2}}}}$, then $ \lfloor   \overline{ \ottnt{a_{{\mathrm{1}}}} }   \rfloor $ and $ \lfloor   \overline{ \ottnt{a_{{\mathrm{2}}}} }   \rfloor $ are equal as $\lambda$-terms. Hence, $ \overline{ \ottnt{a_{{\mathrm{1}}}} }  \simeq  \overline{ \ottnt{a_{{\mathrm{2}}}} } $ in DCC. 
\end{proof}


\section{Proofs of lemmas/theorems stated in Section \ref{secgcc}}

\begin{theorem}[Theorem \ref{gccsound}]\label{GCCPrf}
If $ \Gamma  \vdash  \ottnt{a}  :  \ottnt{A} $ in GCC, then $ \llbracket  \ottnt{a}  \rrbracket  \in \text{Hom}_{\Ct} ( \llbracket  \Gamma  \rrbracket ,  \llbracket  \ottnt{A}  \rrbracket )$. Further, if $ \Gamma  \vdash  \ottnt{a_{{\mathrm{1}}}}  :  \ottnt{A} $ and $ \Gamma  \vdash  \ottnt{a_{{\mathrm{2}}}}  :  \ottnt{A} $ such that $\ottnt{a_{{\mathrm{1}}}} \equiv \ottnt{a_{{\mathrm{2}}}}$ in GCC, then $ \llbracket  \ottnt{a_{{\mathrm{1}}}}  \rrbracket  =  \llbracket  \ottnt{a_{{\mathrm{2}}}}  \rrbracket  \in \text{Hom}_{\Ct} ( \llbracket  \Gamma  \rrbracket ,  \llbracket  \ottnt{A}  \rrbracket )$.
\end{theorem}

\begin{proof}
Follows from Theorem \ref{GMCPrf} by duality.
\end{proof}


\section{Proofs of lemmas/theorems stated in Section \ref{secdccgcc}}

\begin{theorem}[Theorem \ref{GCCtoDCCe}]\label{GCC2DCCe}
If $ \Gamma  \vdash  \ottnt{a}  :  \ottnt{A} $ in GCC($ \mathcal{L} $), then $  \overline{  \Gamma  }   \vdash   \overline{ \ottnt{a} }   :   \overline{ \ottnt{A} }  $ in \ED{}($ \mathcal{L} $). Further, if $ \Gamma  \vdash  \ottnt{a_{{\mathrm{1}}}}  :  \ottnt{A} $ and $ \Gamma  \vdash  \ottnt{a_{{\mathrm{2}}}}  :  \ottnt{A} $ such that $\ottnt{a_{{\mathrm{1}}}} \equiv \ottnt{a_{{\mathrm{2}}}}$ in GCC($ \mathcal{L} $), then $ \overline{ \ottnt{a_{{\mathrm{1}}}} }  \simeq  \overline{ \ottnt{a_{{\mathrm{2}}}} } $ in \ED{}($ \mathcal{L} $).
\end{theorem} 

\begin{proof}
By induction on $ \Gamma  \vdash  \ottnt{a}  :  \ottnt{A} $. Only the cases $\mathbf{extr}$ and $\mathbf{fork}$ are new; for the other cases, follow the proof of Theorem \ref{GMC2DCC}.
\begin{itemize}

\item \Rref{C-Extract}. Have: $ \Gamma  \vdash   \mathbf{extr} \:  \ottnt{a}   :  \ottnt{A} $ where $ \Gamma  \vdash  \ottnt{a}  :   D_{   \bot   } \:  \ottnt{A}  $.\\
By IH, $  \overline{  \Gamma  }   \vdash   \overline{ \ottnt{a} }   :   \mathcal{T}_{  \bot  } \:   \overline{ \ottnt{A} }   $.\\
Now, $$\infer[\textsc{(Bind)}]{  \overline{  \Gamma  }   \vdash   \mathbf{bind} ^{  \bot  } \:  \ottmv{x}  =   \overline{ \ottnt{a} }   \: \mathbf{in} \:  \ottmv{x}   :   \overline{ \ottnt{A} }  }
        {\infer[\textsc{(IH)}]{  \overline{  \Gamma  }   \vdash   \overline{ \ottnt{a} }   :   \mathcal{T}_{  \bot  } \:   \overline{ \ottnt{A} }   }{}
        &
        \infer[\textsc{(Var)}]{   \overline{  \Gamma  }   ,   \ottmv{x}  :   \overline{ \ottnt{A} }     \vdash  \ottmv{x}  :   \overline{ \ottnt{A} }  }{}
        &
        \infer[\textsc{(Prot-Minimum)}]{  \bot   \sqsubseteq   \overline{ \ottnt{A} }  }{}
        }
     $$

\item \Rref{C-Fork}. Have: $ \Gamma  \vdash   \mathbf{fork}^{  \ell_{{\mathrm{1}}}  ,  \ell_{{\mathrm{2}}}  }  \ottnt{a}   :   D_{  \ell_{{\mathrm{1}}}  } \:    D_{  \ell_{{\mathrm{2}}}  } \:  \ottnt{A}    $ where $ \Gamma  \vdash  \ottnt{a}  :   D_{   \ell_{{\mathrm{1}}}  \vee  \ell_{{\mathrm{2}}}   } \:  \ottnt{A}  $.\\
By IH, $  \overline{  \Gamma  }   \vdash   \overline{ \ottnt{a} }   :   \mathcal{T}_{  \ell_{{\mathrm{1}}}  \vee  \ell_{{\mathrm{2}}}  } \:   \overline{ \ottnt{A} }   $.\\
Now, $$\mkern-18mu \infer[\textsc{(Bind)}]{  \overline{  \Gamma  }   \vdash   \mathbf{bind} ^{  \ell_{{\mathrm{1}}}  \vee  \ell_{{\mathrm{2}}}  } \:  \ottmv{x}  =   \overline{ \ottnt{a} }   \: \mathbf{in} \:   \mathbf{eta} ^{ \ell_{{\mathrm{1}}} }   \mathbf{eta} ^{ \ell_{{\mathrm{2}}} }  \ottmv{x}     :   \mathcal{T}_{ \ell_{{\mathrm{1}}} } \:   \mathcal{T}_{ \ell_{{\mathrm{2}}} } \:   \overline{ \ottnt{A} }    }
          { \infer[\textsc{(IH)}]{  \overline{  \Gamma  }   \vdash   \overline{ \ottnt{a} }   :   \mathcal{T}_{  \ell_{{\mathrm{1}}}  \vee  \ell_{{\mathrm{2}}}  } \:   \overline{ \ottnt{A} }   }{}
          &
            \infer[\textsc{(Eta)}]{   \overline{  \Gamma  }   ,   \ottmv{x}  :   \overline{ \ottnt{A} }     \vdash   \mathbf{eta} ^{ \ell_{{\mathrm{1}}} }   \mathbf{eta} ^{ \ell_{{\mathrm{2}}} }  \ottmv{x}    :   \mathcal{T}_{ \ell_{{\mathrm{1}}} } \:   \mathcal{T}_{ \ell_{{\mathrm{2}}} } \:   \overline{ \ottnt{A} }    }
              {\infer[\textsc{(Eta)}]{   \overline{  \Gamma  }   ,   \ottmv{x}  :   \overline{ \ottnt{A} }     \vdash   \mathbf{eta} ^{ \ell_{{\mathrm{2}}} }  \ottmv{x}   :   \mathcal{T}_{ \ell_{{\mathrm{2}}} } \:   \overline{ \ottnt{A} }   }
               {\infer[\textsc{(Var)}]{   \overline{  \Gamma  }   ,   \ottmv{x}  :   \overline{ \ottnt{A} }     \vdash  \ottmv{x}  :   \overline{ \ottnt{A} }  }{}}
              }
          &
            \infer[\textsc{(Comb)}]{  \ell_{{\mathrm{1}}}  \vee  \ell_{{\mathrm{2}}}   \sqsubseteq   \mathcal{T}_{ \ell_{{\mathrm{1}}} } \:   \mathcal{T}_{ \ell_{{\mathrm{2}}} } \:   \overline{ \ottnt{A} }    }
              {\infer[\textsc{(Mon)}]{ \ell_{{\mathrm{1}}}  \sqsubseteq   \mathcal{T}_{ \ell_{{\mathrm{1}}} } \:   \mathcal{T}_{ \ell_{{\mathrm{2}}} } \:   \overline{ \ottnt{A} }    }{}
              &
              \infer[\textsc{(Aldy)}]{ \ell_{{\mathrm{2}}}  \sqsubseteq   \mathcal{T}_{ \ell_{{\mathrm{1}}} } \:   \mathcal{T}_{ \ell_{{\mathrm{2}}} } \:   \overline{ \ottnt{A} }    }
               {\infer[\textsc{(Mon)}]{ \ell_{{\mathrm{2}}}  \sqsubseteq   \mathcal{T}_{ \ell_{{\mathrm{2}}} } \:   \overline{ \ottnt{A} }   }{}}}
          }
     $$ 
\end{itemize}

Now, for $ \Gamma  \vdash  \ottnt{a_{{\mathrm{1}}}}  :  \ottnt{A} $ and $ \Gamma  \vdash  \ottnt{a_{{\mathrm{2}}}}  :  \ottnt{A} $, if $\ottnt{a_{{\mathrm{1}}}} \equiv \ottnt{a_{{\mathrm{2}}}}$, then $ \lfloor   \overline{ \ottnt{a_{{\mathrm{1}}}} }   \rfloor $ and $ \lfloor   \overline{ \ottnt{a_{{\mathrm{2}}}} }   \rfloor $ are equal as $\lambda$-terms. Hence, $ \overline{ \ottnt{a_{{\mathrm{1}}}} }  \simeq  \overline{ \ottnt{a_{{\mathrm{2}}}} } $ in \ED{}.
\end{proof}


\section{Proofs of lemmas/theorems stated in Section \ref{secgmcc}}

\begin{theorem}[Theorem \ref{gmccsound}]
If $ \Gamma  \vdash  \ottnt{a}  :  \ottnt{A} $ in GMCC, then $ \llbracket  \ottnt{a}  \rrbracket  \in \text{Hom}_{\Ct} ( \llbracket  \Gamma  \rrbracket ,  \llbracket  \ottnt{A}  \rrbracket )$. Further, if $ \Gamma  \vdash  \ottnt{a_{{\mathrm{1}}}}  :  \ottnt{A} $ and $ \Gamma  \vdash  \ottnt{a_{{\mathrm{2}}}}  :  \ottnt{A} $ such that $\ottnt{a_{{\mathrm{1}}}} \equiv \ottnt{a_{{\mathrm{2}}}}$ in GMCC, then $ \llbracket  \ottnt{a_{{\mathrm{1}}}}  \rrbracket  =  \llbracket  \ottnt{a_{{\mathrm{2}}}}  \rrbracket  \in \text{Hom}_{\Ct} ( \llbracket  \Gamma  \rrbracket ,  \llbracket  \ottnt{A}  \rrbracket )$.
\end{theorem}

\begin{proof}
Let $ \Gamma  \vdash  \ottnt{a}  :  \ottnt{A} $. Then, by Theorems \ref{GMCPrf} and \ref{GCCPrf}, $ \llbracket  \ottnt{a}  \rrbracket  \in \text{Hom}_{\Ct} ( \llbracket  \Gamma  \rrbracket ,  \llbracket  \ottnt{A}  \rrbracket )$ because a strong monoidal functor is also a lax and an oplax monoidal functor. For $ \Gamma  \vdash  \ottnt{a_{{\mathrm{1}}}}  :  \ottnt{A} $ and $ \Gamma  \vdash  \ottnt{a_{{\mathrm{2}}}}  :  \ottnt{A} $, if $\ottnt{a_{{\mathrm{1}}}} \equiv \ottnt{a_{{\mathrm{2}}}}$, then $ \llbracket  \ottnt{a_{{\mathrm{1}}}}  \rrbracket  =  \llbracket  \ottnt{a_{{\mathrm{2}}}}  \rrbracket $ by Theorems \ref{GMCPrf} and \ref{GCCPrf} and the equations listed in Section \ref{GMCCModel}. 

\end{proof}


\begin{theorem}[Theorem \ref{gmcccomplete}]
Given any preordered monoid $ \mathcal{M} $, for typing derivations $ \Gamma  \vdash  \ottnt{a_{{\mathrm{1}}}}  :  \ottnt{A} $ and $ \Gamma  \vdash  \ottnt{a_{{\mathrm{2}}}}  :  \ottnt{A} $ in GMCC($ \mathcal{M} $), if $ \llbracket  \ottnt{a_{{\mathrm{1}}}}  \rrbracket  =  \llbracket  \ottnt{a_{{\mathrm{2}}}}  \rrbracket $ in all models of GMCC($ \mathcal{M} $), then $\ottnt{a_{{\mathrm{1}}}} \equiv \ottnt{a_{{\mathrm{2}}}}$ is derivable in GMCC($ \mathcal{M} $).
\end{theorem}

\begin{proof}
We use standard term-model construction for proving this theorem. First, fix the preordered monoid, $ \mathcal{M} $. Next, construct the freely generated bicartesian closed category, $ \mathds{F} $, from the syntax of GMCC($ \mathcal{M} $), as follows:\\
$ \text{Obj}(\mathds{F}) , A, B ::=  \ottkw{Unit}  \: | \:  \mathbf{Void}  \: | \:  \ottnt{A}  \times  \ottnt{B}  \: | \:  \ottnt{A}  +  \ottnt{B}  \: | \:  \ottnt{A}  \to  \ottnt{B}  \: | \:  S_{ m } \:  \ottnt{A} $.\\
$ \text{Hom}_{\mathds{F} } ( \ottnt{A}  ,  \ottnt{B} )  = \{ t \: | \:   \emptyset   \vdash  \ottnt{t}  :   \ottnt{A}  \to  \ottnt{B}   \} / =_{\beta\eta}$.\\
The objects of $ \mathds{F} $ are the types of GMCC($ \mathcal{M} $) while the morphisms are the terms of GMCC($ \mathcal{M} $) quotiented by $\beta\eta$-equivalence. This is the classifying category of GMCC($ \mathcal{M} $).

Now, we define a strong monoidal functor $\mathbb{S}$ from $\Ca( \mathcal{M} )$ to $ \mathbf{End}_{\mathds{F} }^{\text{s} } $.\\
$\mathbb{S}(m) := \mathbb{S}_m$ where $ \mathbb{S}_{ m } \:  \ottnt{A}  :=  S_{ m } \:  \ottnt{A} $ and $\mathbb{S}_m (A \xrightarrow{t} B) :=  S_{ m } \:  \ottnt{A}  \xrightarrow{ \mathbf{lift}^{ m }  \ottnt{t} }  S_{ m } \:  \ottnt{B} $.\\
$\mathbb{S}( m_{{\mathrm{1}}}   \leq   m_{{\mathrm{2}}} ) := \mathbb{S}_{m_{{\mathrm{1}}}} \xrightarrow{\mathbb{S}^{ m_{{\mathrm{1}}}   \leq   m_{{\mathrm{2}}} }} \mathbb{S}_{m_{{\mathrm{2}}}}$ where $\mathbb{S}^{ m_{{\mathrm{1}}}   \leq   m_{{\mathrm{2}}} }_A :=  \lambda  \ottmv{x}  :   S_{ m_{{\mathrm{1}}} } \:  \ottnt{A}   .   \mathbf{up}^{ m_{{\mathrm{1}}} , m_{{\mathrm{2}}} }  \ottmv{x}  $.\\

Need to check that $\mathbb{S}$ is well-defined. In other words, need to show that $\mathbb{S}_m$s are strong endofunctors and $\mathbb{S}^{ m_{{\mathrm{1}}}   \leq   m_{{\mathrm{2}}} }$s are strong natural transformations.

$\mathbb{S}_m$ is a functor because:\\
For $A \in  \text{Obj}(\mathds{F}) $, $ \mathbb{S}_{ m } \:   \text{id}_{ \ottnt{A} }   =  \mathbf{lift}^{ m }   \text{id}_{ \ottnt{A} }   =  \mathbf{lift}^{ m }   (   \lambda  \ottmv{x}  :  \ottnt{A}  .  \ottmv{x}   )   =  \lambda  \ottmv{x}  :   S_{ m } \:  \ottnt{A}   .  \ottmv{x}  =  \text{id}_{  S_{ m } \:  \ottnt{A}  } $.\\
For $f \in  \text{Hom}_{\mathds{F} } ( \ottnt{A}  ,  \ottnt{B} ) $ and $g \in  \text{Hom}_{\mathds{F} } ( \ottnt{B}  ,  \ottnt{C} ) $, \\ $ \mathbb{S}_{ m } \:   (   \ottnt{g}  \circ  \ottnt{f}   )   =  \mathbf{lift}^{ m }   (   \ottnt{g}  \circ  \ottnt{f}   )   =  \mathbf{lift}^{ m }   (    \lambda  \ottmv{x}  .  \ottnt{g}   \:   (   \ottnt{f}  \:  \ottmv{x}   )    )   =   \lambda  \ottmv{x}  .   (   \mathbf{lift}^{ m }  \ottnt{g}   )    \:   (    (   \mathbf{lift}^{ m }  \ottnt{f}   )   \:  \ottmv{x}   )   =   \mathbb{S}_{ m } \:  \ottnt{g}   \circ   \mathbb{S}_{ m } \:  \ottnt{f}  $.\\

Now we define strength of $\mathbb{S}_m$, $t^{\mathbb{S}_m}$.\\
We have, $  \ottmv{x}  :   \ottnt{A}  \times   S_{ m } \:  \ottnt{B}     \vdash   \lambda  \ottmv{y}  .   (   \mathbf{proj}_1 \:  \ottmv{x}   ,  \ottmv{y}  )    :    \ottnt{B}  \to  \ottnt{A}   \times  \ottnt{B}  $.\\
Then, $  \ottmv{x}  :   \ottnt{A}  \times   S_{ m } \:  \ottnt{B}     \vdash   \mathbf{lift}^{ m }   (   \lambda  \ottmv{y}  .   (   \mathbf{proj}_1 \:  \ottmv{x}   ,  \ottmv{y}  )    )    :    S_{ m } \:  \ottnt{B}   \to   S_{ m } \:   (   \ottnt{A}  \times  \ottnt{B}   )    $.\\
And, $  \ottmv{x}  :   \ottnt{A}  \times   S_{ m } \:  \ottnt{B}     \vdash    (   \mathbf{lift}^{ m }   (   \lambda  \ottmv{y}  .   (   \mathbf{proj}_1 \:  \ottmv{x}   ,  \ottmv{y}  )    )    )   \:   (   \mathbf{proj}_2 \:  \ottmv{x}   )    :   S_{ m } \:   (   \ottnt{A}  \times  \ottnt{B}   )   $.\\
So, $  \emptyset   \vdash    \lambda  \ottmv{x}  .   (   \mathbf{lift}^{ m }   (   \lambda  \ottmv{y}  .   (   \mathbf{proj}_1 \:  \ottmv{x}   ,  \ottmv{y}  )    )    )    \:   (   \mathbf{proj}_2 \:  \ottmv{x}   )    :    \ottnt{A}  \times   S_{ m } \:  \ottnt{B}    \to   S_{ m } \:   (   \ottnt{A}  \times  \ottnt{B}   )    $.\\
We define: $ t^{\mathbb{S}_{ m } }_{ \ottnt{A} ,  \ottnt{B} }  \triangleq   \lambda  \ottmv{x}  .   (   \mathbf{lift}^{ m }   (   \lambda  \ottmv{y}  .   (   \mathbf{proj}_1 \:  \ottmv{x}   ,  \ottmv{y}  )    )    )    \:   (   \mathbf{proj}_2 \:  \ottmv{x}   )  $.\\
Check that $ t^{\mathbb{S}_{ m } }_{ \ottnt{A} ,  \ottnt{B} } $ is natural in both $A$ and $B$:\\

The left diagram in Figure \ref{strcd} commutes because:
\begin{align*}
&   t^{\mathbb{S}_{ m } }_{ \ottnt{A'} ,  \ottnt{B} }   \circ   (   \ottnt{f}  \times   \text{id}    )    & & \;\;   \mathbb{S}_{ m } \:   (   \ottnt{f}  \times   \text{id}    )    \circ   t^{\mathbb{S}_{ m } }_{ \ottnt{A} ,  \ottnt{B} }  \\ 
= &   \lambda  \ottmv{x}  .   t^{\mathbb{S}_{ m } }_{ \ottnt{A'} ,  \ottnt{B} }    \:   (    (   \ottnt{f}  \times   \text{id}    )   \:  \ottmv{x}   )   & & =   \lambda  \ottmv{x}  .   (   \mathbf{lift}^{ m }   (   \ottnt{f}  \times   \text{id}    )    )    \:   (    (   \mathbf{lift}^{ m }   (   \lambda  \ottmv{y}  .   (   \mathbf{proj}_1 \:  \ottmv{x}   ,  \ottmv{y}  )    )    )   \:   (   \mathbf{proj}_2 \:  \ottmv{x}   )    )  \\
= &   \lambda  \ottmv{x}  .   t^{\mathbb{S}_{ m } }_{ \ottnt{A'} ,  \ottnt{B} }    \:   (   \ottnt{f}  \:   \mathbf{proj}_1 \:  \ottmv{x}    ,   \mathbf{proj}_2 \:  \ottmv{x}   )   & & =   \lambda  \ottmv{x}  .   (    \lambda  \ottmv{z}  .   (   \mathbf{lift}^{ m }   (   \ottnt{f}  \times   \text{id}    )    )    \:   (    (   \mathbf{lift}^{ m }   (   \lambda  \ottmv{y}  .   (   \mathbf{proj}_1 \:  \ottmv{x}   ,  \ottmv{y}  )    )    )   \:  \ottmv{z}   )    )    \:   (   \mathbf{proj}_2 \:  \ottmv{x}   )   \\
= &   \lambda  \ottmv{x}  .   (   \mathbf{lift}^{ m }   (   \lambda  \ottmv{y}  .   (   \ottnt{f}  \:   \mathbf{proj}_1 \:  \ottmv{x}    ,  \ottmv{y}  )    )    )    \:   (   \mathbf{proj}_2 \:  \ottmv{x}   )   & & =   \lambda  \ottmv{x}  .   (   \mathbf{lift}^{ m }   (    \lambda  \ottmv{z}  .   (   \ottnt{f}  \times   \text{id}    )    \:   (    (   \lambda  \ottmv{y}  .   (   \mathbf{proj}_1 \:  \ottmv{x}   ,  \ottmv{y}  )    )   \:  \ottmv{z}   )    )    )    \:   (   \mathbf{proj}_2 \:  \ottmv{x}   )   \\
& & & =   \lambda  \ottmv{x}  .   (   \mathbf{lift}^{ m }   (   \lambda  \ottmv{z}  .   (   \ottnt{f}  \:   \mathbf{proj}_1 \:  \ottmv{x}    ,  \ottmv{z}  )    )    )    \:   (   \mathbf{proj}_2 \:  \ottmv{x}   )  
\end{align*}

The right diagram in Figure \ref{strcd} commutes because:
\begin{align*}
\mkern-28mu &    t^{\mathbb{S}_{ m } }_{ \ottnt{A} ,  \ottnt{B'} }   \circ   \text{id}    \times   \mathbb{S}_{ m } \:  \ottnt{g}   & &\;\;   \mathbb{S}_{ m } \:   (    \text{id}   \times  \ottnt{g}   )    \circ   t^{\mathbb{S}_{ m } }_{ \ottnt{A} ,  \ottnt{B} }   \\
\mkern-28mu = &   \lambda  \ottmv{x}  .   t^{\mathbb{S}_{ m } }_{ \ottnt{A} ,  \ottnt{B'} }    \:   (   \mathbf{proj}_1 \:  \ottmv{x}   ,    (   \mathbf{lift}^{ m }  \ottnt{g}   )   \:   (   \mathbf{proj}_2 \:  \ottmv{x}   )    )   & & =    \lambda  \ottmv{x}  .   (   \mathbf{lift}^{ m }   (    \text{id}   \times  \ottnt{g}   )    )    \:   (   \mathbf{lift}^{ m }   (   \lambda  \ottmv{y}  .   (   \mathbf{proj}_1 \:  \ottmv{x}   ,  \ottmv{y}  )    )    )    \:   (   \mathbf{proj}_2 \:  \ottmv{x}   )   \\
\mkern-28mu = &   \lambda  \ottmv{x}  .   (   \mathbf{lift}^{ m }   (   \lambda  \ottmv{y}  .   (   \mathbf{proj}_1 \:  \ottmv{x}   ,  \ottmv{y}  )    )    )    \:   (    (   \mathbf{lift}^{ m }  \ottnt{g}   )   \:   (   \mathbf{proj}_2 \:  \ottmv{x}   )    )   & & =   \lambda  \ottmv{x}  .   (     \lambda  \ottmv{z}  .   (    \mathbf{lift}^{ m }   \text{id}    \times  \ottnt{g}   )    \:   (   \mathbf{lift}^{ m }   \lambda  \ottmv{y}  .   (   \mathbf{proj}_1 \:  \ottmv{x}   ,  \ottmv{y}  )     )    \:  \ottmv{z}   )    \:   \mathbf{proj}_2 \:  \ottmv{x}  \\
\mkern-28mu = &   \lambda  \ottmv{x}  .   (    \lambda  \ottmv{z}  .   (   \mathbf{lift}^{ m }   \lambda  \ottmv{y}  .   (   \mathbf{proj}_1 \:  \ottmv{x}   ,  \ottmv{y}  )     )    \:   (    (   \mathbf{lift}^{ m }  \ottnt{g}   )   \:  \ottmv{z}   )    )    \:   \mathbf{proj}_2 \:  \ottmv{x}   & & =   \lambda  \ottmv{x}  .   (   \mathbf{lift}^{ m }   (     \lambda  \ottmv{z}  .   (    \text{id}   \times  \ottnt{g}   )    \:   (   \lambda  \ottmv{y}  .   (   \mathbf{proj}_1 \:  \ottmv{x}   ,  \ottmv{y}  )    )    \:  \ottmv{z}   )    )    \:   (   \mathbf{proj}_2 \:  \ottmv{x}   )   \\
\mkern-28mu = &   \lambda  \ottmv{x}  .   (   \mathbf{lift}^{ m }   (    \lambda  \ottmv{z}  .   (   \lambda  \ottmv{y}  .   (   \mathbf{proj}_1 \:  \ottmv{x}   ,  \ottmv{y}  )    )    \:   (   \ottnt{g}  \:  \ottmv{z}   )    )    )    \:   (   \mathbf{proj}_2 \:  \ottmv{x}   )   & & =   \lambda  \ottmv{x}  .   (   \mathbf{lift}^{ m }   (    \lambda  \ottmv{z}  .   (    \text{id}   \times  \ottnt{g}   )    \:   (   \mathbf{proj}_1 \:  \ottmv{x}   ,  \ottmv{z}  )    )    )    \:   (   \mathbf{proj}_2 \:  \ottmv{x}   )   \\ 
\mkern-28mu = &   \lambda  \ottmv{x}  .   (   \mathbf{lift}^{ m }   (   \lambda  \ottmv{z}  .   (   \mathbf{proj}_1 \:  \ottmv{x}   ,   \ottnt{g}  \:  \ottmv{z}   )    )    )    \:   (   \mathbf{proj}_2 \:  \ottmv{x}   )   & & =   \lambda  \ottmv{x}  .   (   \mathbf{lift}^{ m }   (   \lambda  \ottmv{z}  .   (   \mathbf{proj}_1 \:  \ottmv{x}   ,   \ottnt{g}  \:  \ottmv{z}   )    )    )    \:   (   \mathbf{proj}_2 \:  \ottmv{x}   )  
\end{align*}

\begin{figure}
\centering
\begin{subfigure}{0.4\textwidth}
\begin{tikzcd}
 \ottnt{A}  \times   S_{ m } \:  \ottnt{B}   \arrow{r}{ t^{\mathbb{S}_{ m } }_{ \ottnt{A} ,  \ottnt{B} } } \arrow{d}{ \ottnt{f}  \times   \text{id}  } &  S_{ m } \:   (   \ottnt{A}  \times  \ottnt{B}   )   \arrow{d}{ \mathbb{S}_{ m } \:   (   \ottnt{f}  \times   \text{id}    )  } \\
 \ottnt{A'}  \times   S_{ m } \:  \ottnt{B}   \arrow{r}{ t^{\mathbb{S}_{ m } }_{ \ottnt{A'} ,  \ottnt{B} } } &  S_{ m } \:   (   \ottnt{A'}  \times  \ottnt{B}   )   
\end{tikzcd}
\end{subfigure}
\begin{subfigure}{0.4\textwidth}
\begin{tikzcd}
 \ottnt{A}  \times   S_{ m } \:  \ottnt{B}   \arrow{r}{ t^{\mathbb{S}_{ m } }_{ \ottnt{A} ,  \ottnt{B} } } \arrow{d}{  \text{id}   \times   \mathbb{S}_{ m } \:  \ottnt{g}  } &  S_{ m } \:   (   \ottnt{A}  \times  \ottnt{B}   )   \arrow{d}{ \mathbb{S}_{ m } \:   (    \text{id}   \times  \ottnt{g}   )  } \\
 \ottnt{A}  \times   S_{ m } \:  \ottnt{B'}   \arrow{r}{ t^{\mathbb{S}_{ m } }_{ \ottnt{A} ,  \ottnt{B'} } } &  S_{ m } \:   (   \ottnt{A}  \times  \ottnt{B'}   )  
\end{tikzcd}
\end{subfigure}
\caption{Commutative diagram}
\label{strcd}
\end{figure} 

Now check that $t^{\mathbb{S}_m}$ satisfies the axioms for strength.

\begin{figure}
\centering
\begin{subfigure}{0.35\textwidth}
\begin{tikzcd}
  (    \ottnt{A}  \times  \ottnt{B}    )   \times    \mathbb{S}_{ m } \:  \ottnt{C}    \arrow{rr}{ t^{\mathbb{S}_{ m } }_{  \ottnt{A}  \times  \ottnt{B}  ,  \ottnt{C} } } \arrow{d}{ \alpha^{-1} } & &   \mathbb{S}_{ m } \:   (    (   \ottnt{A}  \times  \ottnt{B}   )   \times  \ottnt{C}   )   \arrow{d}{ \mathbb{S}_{ m } \:   \alpha^{-1}  } \\
  \ottnt{A}   \times   (    \ottnt{B}   \times    \mathbb{S}_{ m } \:  \ottnt{C}     )   \arrow{r}{  \text{id}   \times   t^{\mathbb{S}_{ m } }_{ \ottnt{B} ,  \ottnt{C} }  } &   \ottnt{A}   \times    \mathbb{S}_{ m } \:   (   \ottnt{B}  \times  \ottnt{C}   )     \arrow{r}{ t^{\mathbb{S}_{ m } }_{ \ottnt{A} ,   \ottnt{B}  \times  \ottnt{C}  } } &  \mathbb{S}_{ m } \:   (   \ottnt{A}  \times   (   \ottnt{B}  \times  \ottnt{C}   )    )  
\end{tikzcd}
\end{subfigure}
\hfill
\begin{subfigure}{0.35\textwidth}
\begin{tikzcd}
  \qquad\quad     \top   \times   \mathbb{S}_{ m } \:  \ottnt{A}     \arrow{r}{ t^{\mathbb{S}_{ m } }_{   \top   ,   \ottnt{A}  } } \arrow{dr}{ \pi_2 } &  \mathbb{S}_{ m }   (     \top    \times   \ottnt{A}    )   \arrow{d}{ \mathbb{S}_{ m } \:   \pi_2  }\\
&  \mathbb{S}_{ m } \:  \ottnt{A} 
\end{tikzcd}
\end{subfigure}
\caption{Commutative diagram}
\label{strcd2}
\end{figure} 

The left diagram in Figure \ref{strcd2} commutes because:

\begin{align*}
&    t^{\mathbb{S}_{ m } }_{ \ottnt{A} ,   \ottnt{B}  \times  \ottnt{C}  }   \circ   (    \text{id}   \times   t^{\mathbb{S}_{ m } }_{ \ottnt{B} ,  \ottnt{C} }    )    \circ   \alpha^{-1}   \\
= &    \lambda  \ottmv{x}  .   t^{\mathbb{S}_{ m } }_{ \ottnt{A} ,   \ottnt{B}  \times  \ottnt{C}  }    \:   (    \text{id}   \times   t^{\mathbb{S}_{ m } }_{ \ottnt{B} ,  \ottnt{C} }    )    \:   (   \mathbf{proj}_1 \:   \mathbf{proj}_1 \:  \ottmv{x}    ,   (   \mathbf{proj}_2 \:   \mathbf{proj}_1 \:  \ottmv{x}    ,   \mathbf{proj}_2 \:  \ottmv{x}   )   )   \\
= &   \lambda  \ottmv{x}  .   t^{\mathbb{S}_{ m } }_{ \ottnt{A} ,   \ottnt{B}  \times  \ottnt{C}  }    \:   (   \mathbf{proj}_1 \:   \mathbf{proj}_1 \:  \ottmv{x}    ,    (   \mathbf{lift}^{ m }   (   \lambda  \ottmv{y}  .   (   \mathbf{proj}_2 \:   \mathbf{proj}_1 \:  \ottmv{x}    ,  \ottmv{y}  )    )    )   \:   (   \mathbf{proj}_2 \:  \ottmv{x}   )    )   \\
= &    \lambda  \ottmv{x}  .   (   \mathbf{lift}^{ m }   (   \lambda  \ottmv{z}  .   (   \mathbf{proj}_1 \:   \mathbf{proj}_1 \:  \ottmv{x}    ,  \ottmv{z}  )    )    )    \:   (   \mathbf{lift}^{ m }   (   \lambda  \ottmv{y}  .   (   \mathbf{proj}_2 \:   \mathbf{proj}_1 \:  \ottmv{x}    ,  \ottmv{y}  )    )    )    \:   (   \mathbf{proj}_2 \:  \ottmv{x}   )   \\
= &   \lambda  \ottmv{x}  .   (     \lambda  \ottmv{w}  .   (   \mathbf{lift}^{ m }   (   \lambda  \ottmv{z}  .   (   \mathbf{proj}_1 \:   \mathbf{proj}_1 \:  \ottmv{x}    ,  \ottmv{z}  )    )    )    \:   (   \mathbf{lift}^{ m }   (   \lambda  \ottmv{y}  .   (   \mathbf{proj}_2 \:   \mathbf{proj}_1 \:  \ottmv{x}    ,  \ottmv{y}  )    )    )    \:  \ottmv{w}   )    \:   (   \mathbf{proj}_2 \:  \ottmv{x}   )   \\
= &   \lambda  \ottmv{x}  .   (   \mathbf{lift}^{ m }   (     \lambda  \ottmv{w}  .   (   \lambda  \ottmv{z}  .   (   \mathbf{proj}_1 \:   \mathbf{proj}_1 \:  \ottmv{x}    ,  \ottmv{z}  )    )    \:   (   \lambda  \ottmv{y}  .   (   \mathbf{proj}_2 \:   \mathbf{proj}_1 \:  \ottmv{x}    ,  \ottmv{y}  )    )    \:  \ottmv{w}   )    )    \:   (   \mathbf{proj}_2 \:  \ottmv{x}   )   \\
= &   \lambda  \ottmv{x}  .   (   \mathbf{lift}^{ m }   (    \lambda  \ottmv{w}  .   (   \lambda  \ottmv{z}  .   (   \mathbf{proj}_1 \:   \mathbf{proj}_1 \:  \ottmv{x}    ,  \ottmv{z}  )    )    \:   (   \mathbf{proj}_2 \:   \mathbf{proj}_1 \:  \ottmv{x}    ,  \ottmv{w}  )    )    )    \:   (   \mathbf{proj}_2 \:  \ottmv{x}   )   \\
= &   \lambda  \ottmv{x}  .   (   \mathbf{lift}^{ m }   (   \lambda  \ottmv{w}  .   (   \mathbf{proj}_1 \:   \mathbf{proj}_1 \:  \ottmv{x}    ,   (   \mathbf{proj}_2 \:   \mathbf{proj}_1 \:  \ottmv{x}    ,  \ottmv{w}  )   )    )    )    \:   (   \mathbf{proj}_2 \:  \ottmv{x}   )  
\end{align*}
and
\begin{align*}
&   \mathbb{S}_{ m } \:   \alpha^{-1}    \circ   t^{\mathbb{S}_{ m } }_{  \ottnt{A}  \times  \ottnt{B}  ,  \ottnt{C} }   \\
= &    \lambda  \ottmv{x}  .   \mathbb{S}_{ m } \:   \alpha^{-1}     \:   (   \mathbf{lift}^{ m }   (   \lambda  \ottmv{y}  .   (   \mathbf{proj}_1 \:  \ottmv{x}   ,  \ottmv{y}  )    )    )    \:   (   \mathbf{proj}_2 \:  \ottmv{x}   )   \\
= &    \lambda  \ottmv{x}  .   (   \mathbf{lift}^{ m }   (   \lambda  \ottmv{y}  .   (   \mathbf{proj}_1 \:   \mathbf{proj}_1 \:  \ottmv{y}    ,   (   \mathbf{proj}_2 \:   \mathbf{proj}_1 \:  \ottmv{y}    ,   \mathbf{proj}_2 \:  \ottmv{y}   )   )    )    )    \:   (   \mathbf{lift}^{ m }   (   \lambda  \ottmv{y}  .   (   \mathbf{proj}_1 \:  \ottmv{x}   ,  \ottmv{y}  )    )    )    \:   (   \mathbf{proj}_2 \:  \ottmv{x}   )   \\
= &   \lambda  \ottmv{x}  .   (     \lambda  \ottmv{z}  .   (   \mathbf{lift}^{ m }   (   \lambda  \ottmv{y}  .   (   \mathbf{proj}_1 \:   \mathbf{proj}_1 \:  \ottmv{y}    ,   (   \mathbf{proj}_2 \:   \mathbf{proj}_1 \:  \ottmv{y}    ,   \mathbf{proj}_2 \:  \ottmv{y}   )   )    )    )    \:   (   \mathbf{lift}^{ m }   (   \lambda  \ottmv{y}  .   (   \mathbf{proj}_1 \:  \ottmv{x}   ,  \ottmv{y}  )    )    )    \:  \ottmv{z}   )    \:   (   \mathbf{proj}_2 \:  \ottmv{x}   )   \\
= &   \lambda  \ottmv{x}  .   (   \mathbf{lift}^{ m }   (     \lambda  \ottmv{z}  .   (   \lambda  \ottmv{y}  .   (   \mathbf{proj}_1 \:   \mathbf{proj}_1 \:  \ottmv{y}    ,   (   \mathbf{proj}_2 \:   \mathbf{proj}_1 \:  \ottmv{y}    ,   \mathbf{proj}_2 \:  \ottmv{y}   )   )    )    \:   (   \lambda  \ottmv{y}  .   (   \mathbf{proj}_1 \:  \ottmv{x}   ,  \ottmv{y}  )    )    \:  \ottmv{z}   )    )    \:   (   \mathbf{proj}_2 \:  \ottmv{x}   )   \\
= &   \lambda  \ottmv{x}  .   (   \mathbf{lift}^{ m }   (    \lambda  \ottmv{z}  .   (   \lambda  \ottmv{y}  .   (   \mathbf{proj}_1 \:   \mathbf{proj}_1 \:  \ottmv{y}    ,   (   \mathbf{proj}_2 \:   \mathbf{proj}_1 \:  \ottmv{y}    ,   \mathbf{proj}_2 \:  \ottmv{y}   )   )    )    \:   (   \mathbf{proj}_1 \:  \ottmv{x}   ,  \ottmv{z}  )    )    )    \:   (   \mathbf{proj}_2 \:  \ottmv{x}   )   \\
= &   \lambda  \ottmv{x}  .   (   \mathbf{lift}^{ m }   (   \lambda  \ottmv{z}  .   (   \mathbf{proj}_1 \:   \mathbf{proj}_1 \:  \ottmv{x}    ,   (   \mathbf{proj}_2 \:   \mathbf{proj}_1 \:  \ottmv{x}    ,  \ottmv{z}  )   )    )    )    \:   (   \mathbf{proj}_2 \:  \ottmv{x}   )  .
\end{align*}

The right diagram in Figure \ref{strcd2} commutes because:
\begin{align*}
&  \mathbb{S}_{ m } \:   \pi_2   \circ  t^{\mathbb{S}_{ m } }_{   \top   ,   \ottnt{A}  }  \\
= &    \lambda  \ottmv{x}  .   (   \mathbf{lift}^{ m }   (   \lambda  \ottmv{y}  .   \mathbf{proj}_2 \:  \ottmv{y}    )    )    \:   (   \mathbf{lift}^{ m }   (   \lambda  \ottmv{y}  .   (   \mathbf{proj}_1 \:  \ottmv{x}   ,  \ottmv{y}  )    )    )    \:   (   \mathbf{proj}_2 \:  \ottmv{x}   )   \\
= &   \lambda  \ottmv{x}  .   (     \lambda  \ottmv{z}  .   (   \mathbf{lift}^{ m }   (   \lambda  \ottmv{y}  .   \mathbf{proj}_2 \:  \ottmv{y}    )    )    \:   (   \mathbf{lift}^{ m }   (   \lambda  \ottmv{y}  .   (   \mathbf{proj}_1 \:  \ottmv{x}   ,  \ottmv{y}  )    )    )    \:  \ottmv{z}   )    \:   (   \mathbf{proj}_2 \:  \ottmv{x}   )   \\
= &   \lambda  \ottmv{x}  .   (   \mathbf{lift}^{ m }   (     \lambda  \ottmv{z}  .   (   \lambda  \ottmv{y}  .   \mathbf{proj}_2 \:  \ottmv{y}    )    \:   (   \lambda  \ottmv{y}  .   (   \mathbf{proj}_1 \:  \ottmv{x}   ,  \ottmv{y}  )    )    \:  \ottmv{z}   )    )    \:   (   \mathbf{proj}_2 \:  \ottmv{x}   )   \\
= &   \lambda  \ottmv{x}  .   (   \mathbf{lift}^{ m }   (    \lambda  \ottmv{z}  .   (   \lambda  \ottmv{y}  .   \mathbf{proj}_2 \:  \ottmv{y}    )    \:   (   \mathbf{proj}_1 \:  \ottmv{x}   ,  \ottmv{z}  )    )    )    \:   (   \mathbf{proj}_2 \:  \ottmv{x}   )   \\
= &   \lambda  \ottmv{x}  .   (   \mathbf{lift}^{ m }   (   \lambda  \ottmv{z}  .  \ottmv{z}   )    )    \:   (   \mathbf{proj}_2 \:  \ottmv{x}   )   \\
= &   \lambda  \ottmv{x}  .   (   \lambda  \ottmv{z}  .  \ottmv{z}   )    \:   (   \mathbf{proj}_2 \:  \ottmv{x}   )   \\
= &  \lambda  \ottmv{x}  .   \mathbf{proj}_2 \:  \ottmv{x}   \\
= &  \pi_2 .
\end{align*}
 
So $\mathbb{S}^m$ is a strong endofunctor.\\ Now, we need to check that $\mathbb{S}^{ m_{{\mathrm{1}}}   \leq   m_{{\mathrm{2}}} }$ is a strong natural transformation. To show this, we shall use the following equality: 
\begin{equation} \label{lifteqn}
 \mathbf{lift}^{ m }  \ottnt{f}  =  \lambda  \ottmv{y}  .   (   \ottmv{y}  \:  \leftindex^{ m }{\gg}\!\! =^{ \ottsym{1} }   \lambda  \ottmv{x}  .   \ottkw{ret}  \:   (   \ottnt{f}  \:  \ottmv{x}   )      )  
\end{equation}
\begin{align*}
&  \lambda  \ottmv{y}  .   (   \ottmv{y}  \:  \leftindex^{ m }{\gg}\!\! =^{ \ottsym{1} }   \lambda  \ottmv{x}  .   \ottkw{ret}  \:   (   \ottnt{f}  \:  \ottmv{x}   )      )   \\
= &   \lambda  \ottmv{y}  .   \mathbf{join}^{ m , \ottsym{1} }   (   \mathbf{lift}^{ m }   (   \lambda  \ottmv{x}  .   \ottkw{ret}  \:   (   \ottnt{f}  \:  \ottmv{x}   )     )    )     \:  \ottmv{y}  \\
= &   \lambda  \ottmv{y}  .   \mathbf{join}^{ m , \ottsym{1} }   (   \mathbf{lift}^{ m }   (    \lambda  \ottmv{x}  .   (   \lambda  \ottmv{z}  .   \ottkw{ret}  \:  \ottmv{z}    )    \:   (   \ottnt{f}  \:  \ottmv{x}   )    )    )     \:  \ottmv{y}  \\
= &   \lambda  \ottmv{y}  .   \mathbf{join}^{ m , \ottsym{1} }   (     \lambda  \ottmv{x}  .   (   \mathbf{lift}^{ m }   (   \lambda  \ottmv{z}  .   \ottkw{ret}  \:  \ottmv{z}    )    )    \:   (   \mathbf{lift}^{ m }  \ottnt{f}   )    \:  \ottmv{x}   )     \:  \ottmv{y}  \\
= &    \lambda  \ottmv{y}  .   \mathbf{join}^{ m , \ottsym{1} }   (   \mathbf{lift}^{ m }   (   \lambda  \ottmv{z}  .   \ottkw{ret}  \:  \ottmv{z}    )    )     \:   (   \mathbf{lift}^{ m }  \ottnt{f}   )    \:  \ottmv{y}  \\
= &    \lambda  \ottmv{y}  .   (   \mathbf{lift}^{ m }  \ottnt{f}   )    \:  \ottmv{y}   \:  \leftindex^{ m }{\gg}\!\! =^{ \ottsym{1} }   \lambda  \ottmv{z}  .   \ottkw{ret}  \:  \ottmv{z}    \\
= &   \lambda  \ottmv{y}  .   (   \mathbf{lift}^{ m }  \ottnt{f}   )    \:  \ottmv{y}  \hspace*{3pt} [\text{By Equation (\ref{eq:idr})}]\\
= &  \mathbf{lift}^{ m }  \ottnt{f} 
\end{align*}

As a side note, we can dualize the above argument to show that:
\begin{equation} \label{lifteq2}
 \mathbf{lift}^{ m }  \ottnt{f}  =  \lambda  \ottmv{y}  .   (    (    \lambda  \ottmv{x}  .  \ottnt{f}   \:   (   \mathbf{extr} \:  \ottmv{x}   )    )   \:  \leftindex^{ \ottsym{1} }{\ll}\!\! =^{ m }  \ottmv{y}   )  
\end{equation}

\begin{figure}
\centering
\begin{subfigure}{0.4\textwidth}
\begin{tikzcd}
 S_{ m_{{\mathrm{1}}} } \:  \ottnt{A}  \arrow{r}{ \mathbb{S}^{ m_{{\mathrm{1}}}  \leq  m_{{\mathrm{2}}} }_{ \ottnt{A} } } \arrow{d}{ \mathbb{S}_{ m_{{\mathrm{1}}} } \:  \ottnt{f} } &  S_{ m_{{\mathrm{2}}} } \:  \ottnt{A}  \arrow{d}{ \mathbb{S}_{ m_{{\mathrm{2}}} } \:  \ottnt{f} } \\
 S_{ m_{{\mathrm{1}}} } \:  \ottnt{B}  \arrow{r}{ \mathbb{S}^{ m_{{\mathrm{1}}}  \leq  m_{{\mathrm{2}}} }_{ \ottnt{B} } } &  S_{ m_{{\mathrm{2}}} } \:  \ottnt{B}  
\end{tikzcd}
\end{subfigure}
\begin{subfigure}{0.4\textwidth}
\begin{tikzcd}
 \ottnt{A}  \times   S_{ m_{{\mathrm{1}}} } \:  \ottnt{B}   \arrow{d}{ t^{\mathbb{S}_{ m_{{\mathrm{1}}} } }_{ \ottnt{A} ,  \ottnt{B} } } \arrow{r}{  \text{id}   \times   \mathbb{S}^{ m_{{\mathrm{1}}}  \leq  m_{{\mathrm{2}}} }_{ \ottnt{B} }  } &  \ottnt{A}  \times   S_{ m_{{\mathrm{2}}} } \:  \ottnt{B}   \arrow{d}{ t^{\mathbb{S}_{ m_{{\mathrm{2}}} } }_{ \ottnt{A} ,  \ottnt{B} } } \\
 S_{ m_{{\mathrm{1}}} } \:   (   \ottnt{A}  \times  \ottnt{B}   )   \arrow{r}{ \mathbb{S}^{ m_{{\mathrm{1}}}  \leq  m_{{\mathrm{2}}} }_{   \ottnt{A}  \times  \ottnt{B}   } } &  S_{ m_{{\mathrm{2}}} } \:   (   \ottnt{A}  \times  \ottnt{B}   )  
\end{tikzcd}
\end{subfigure}
\caption{Commutative diagram}
\label{strcd3}
\end{figure} 

The left diagram in Figure \ref{strcd3} commutes because:

\begin{align*}
&   \mathbb{S}_{ m_{{\mathrm{2}}} } \:  \ottnt{f}   \circ   \mathbb{S}^{ m_{{\mathrm{1}}}  \leq  m_{{\mathrm{2}}} }_{ \ottnt{A} }   \\
= &   \lambda  \ottmv{x}  .   (   \mathbf{lift}^{ m_{{\mathrm{2}}} }  \ottnt{f}   )    \:   (   \mathbf{up}^{ m_{{\mathrm{1}}} , m_{{\mathrm{2}}} }  \ottmv{x}   )   \\
= &   \lambda  \ottmv{x}  .   (   \lambda  \ottmv{y}  .   (   \ottmv{y}  \:  \leftindex^{ m_{{\mathrm{2}}} }{\gg}\!\! =^{ \ottsym{1} }   \lambda  \ottmv{z}  .   \ottkw{ret}  \:   (   \ottnt{f}  \:  \ottmv{z}   )      )    )    \:   (   \mathbf{up}^{ m_{{\mathrm{1}}} , m_{{\mathrm{2}}} }  \ottmv{x}   )   \\
= &   \lambda  \ottmv{x}  .   \mathbf{up}^{ m_{{\mathrm{1}}} , m_{{\mathrm{2}}} }  \ottmv{x}    \:  \leftindex^{ m_{{\mathrm{2}}} }{\gg}\!\! =^{ \ottsym{1} }   \lambda  \ottmv{z}  .   \ottkw{ret}  \:   (   \ottnt{f}  \:  \ottmv{z}   )     \\
= &  \lambda  \ottmv{x}  .   \mathbf{up}^{ m_{{\mathrm{1}}} , m_{{\mathrm{2}}} }   (   \ottmv{x}  \:  \leftindex^{ m_{{\mathrm{1}}} }{\gg}\!\! =^{ \ottsym{1} }   \lambda  \ottmv{z}  .   \ottkw{ret}  \:   (   \ottnt{f}  \:  \ottmv{z}   )      )    \hspace{3pt} [\text{By Equation (\ref{eq:natl})}]\\
= &  \lambda  \ottmv{x}  .   \mathbf{up}^{ m_{{\mathrm{1}}} , m_{{\mathrm{2}}} }   (    (   \mathbf{lift}^{ m_{{\mathrm{1}}} }  \ottnt{f}   )   \:  \ottmv{x}   )    \\
= &   \mathbb{S}^{ m_{{\mathrm{1}}}  \leq  m_{{\mathrm{2}}} }_{ \ottnt{B} }   \circ   \mathbb{S}_{ m_{{\mathrm{1}}} } \:  \ottnt{f}  
\end{align*}

The right diagram in Figure \ref{strcd3} commutes because:

\begin{align*}
&    t^{\mathbb{S}_{ m_{{\mathrm{2}}} } }_{ \ottnt{A} ,  \ottnt{B} }   \circ   \text{id}    \times   \mathbb{S}^{ m_{{\mathrm{1}}}  \leq  m_{{\mathrm{2}}} }_{ \ottnt{B} }   \\
= &   \lambda  \ottmv{x}  .   t^{\mathbb{S}_{ m_{{\mathrm{2}}} } }_{ \ottnt{A} ,  \ottnt{B} }    \:   (   \mathbf{proj}_1 \:  \ottmv{x}   ,   \mathbf{up}^{ m_{{\mathrm{1}}} , m_{{\mathrm{2}}} }   \mathbf{proj}_2 \:  \ottmv{x}    )   \\
= &   \lambda  \ottmv{x}  .   (   \mathbf{lift}^{ m_{{\mathrm{2}}} }   (   \lambda  \ottmv{y}  .   (   \mathbf{proj}_1 \:  \ottmv{x}   ,  \ottmv{y}  )    )    )    \:   (   \mathbf{up}^{ m_{{\mathrm{1}}} , m_{{\mathrm{2}}} }   \mathbf{proj}_2 \:  \ottmv{x}    )   \\
= &   \lambda  \ottmv{x}  .   \mathbf{up}^{ m_{{\mathrm{1}}} , m_{{\mathrm{2}}} }   \mathbf{proj}_2 \:  \ottmv{x}     \:  \leftindex^{ m_{{\mathrm{2}}} }{\gg}\!\! =^{ \ottsym{1} }   \lambda  \ottmv{z}  .   \ottkw{ret}  \:   (    (   \lambda  \ottmv{y}  .   (   \mathbf{proj}_1 \:  \ottmv{x}   ,  \ottmv{y}  )    )   \:  \ottmv{z}   )     \\
= &   \lambda  \ottmv{x}  .   \mathbf{up}^{ m_{{\mathrm{1}}} , m_{{\mathrm{2}}} }   \mathbf{proj}_2 \:  \ottmv{x}     \:  \leftindex^{ m_{{\mathrm{2}}} }{\gg}\!\! =^{ \ottsym{1} }   \lambda  \ottmv{z}  .   \ottkw{ret}  \:   (   \mathbf{proj}_1 \:  \ottmv{x}   ,  \ottmv{z}  )     \\
= &  \lambda  \ottmv{x}  .   \mathbf{up}^{ m_{{\mathrm{1}}} , m_{{\mathrm{2}}} }   (    \mathbf{proj}_2 \:  \ottmv{x}   \:  \leftindex^{ m_{{\mathrm{1}}} }{\gg}\!\! =^{ \ottsym{1} }   \lambda  \ottmv{z}  .   \ottkw{ret}  \:   (   \mathbf{proj}_1 \:  \ottmv{x}   ,  \ottmv{z}  )      )    \hspace{3pt} [\text{By Equation (\ref{eq:natl})}] \\
= &  \lambda  \ottmv{x}  .   \mathbf{up}^{ m_{{\mathrm{1}}} , m_{{\mathrm{2}}} }   (    (   \lambda  \ottmv{y}  .   (   \ottmv{y}  \:  \leftindex^{ m_{{\mathrm{1}}} }{\gg}\!\! =^{ \ottsym{1} }   \lambda  \ottmv{z}  .   \ottkw{ret}  \:   (    (   \lambda  \ottmv{w}  .   (   \mathbf{proj}_1 \:  \ottmv{x}   ,  \ottmv{w}  )    )   \:  \ottmv{z}   )      )    )   \:   (   \mathbf{proj}_2 \:  \ottmv{x}   )    )    \\
= &  \lambda  \ottmv{x}  .   \mathbf{up}^{ m_{{\mathrm{1}}} , m_{{\mathrm{2}}} }   (    (   \mathbf{lift}^{ m_{{\mathrm{1}}} }   (   \lambda  \ottmv{w}  .   (   \mathbf{proj}_1 \:  \ottmv{x}   ,  \ottmv{w}  )    )    )   \:   (   \mathbf{proj}_2 \:  \ottmv{x}   )    )    \\
= &   \mathbb{S}^{ m_{{\mathrm{1}}}  \leq  m_{{\mathrm{2}}} }_{   \ottnt{A}  \times  \ottnt{B}   }   \circ   t^{\mathbb{S}_{ m_{{\mathrm{1}}} } }_{ \ottnt{A} ,  \ottnt{B} }  
\end{align*}

This shows that $\mathbb{S}^{ m_{{\mathrm{1}}}   \leq   m_{{\mathrm{2}}} }$ is a strong natural transformation. Hence $\mathbb{S}$ is well-defined.\\ Next we show that $\mathbb{S}$ is indeed a functor.\\
We have: $ \mathbb{S}^{ m  \leq  m }_{ \ottnt{A} }  =  \lambda  \ottmv{x}  .   \mathbf{up}^{ m , m }  \ottmv{x}   =  \lambda  \ottmv{x}  .  \ottmv{x}  =  \text{id}_{  \mathbb{S}_{ m } \:  \ottnt{A}  } $.\\
And, $  \mathbb{S}^{ m_{{\mathrm{2}}}  \leq  m_{{\mathrm{3}}} }_{ \ottnt{A} }   \circ   \mathbb{S}^{ m_{{\mathrm{1}}}  \leq  m_{{\mathrm{2}}} }_{ \ottnt{A} }   =  \lambda  \ottmv{x}  .   \mathbf{up}^{ m_{{\mathrm{2}}} , m_{{\mathrm{3}}} }   (   \mathbf{up}^{ m_{{\mathrm{1}}} , m_{{\mathrm{2}}} }  \ottmv{x}   )    =  \lambda  \ottmv{x}  .   \mathbf{up}^{ m_{{\mathrm{1}}} , m_{{\mathrm{3}}} }  \ottmv{x}   =  \mathbb{S}^{ m_{{\mathrm{1}}}  \leq  m_{{\mathrm{3}}} }_{ \ottnt{A} } $.\\
Hence, $\mathbb{S}$ is a functor. Next we show that $\mathbb{S}$ is strong monoidal.\\

Define: $\eta : \Id{} \to \mathbb{S}_{1}$ and $\mu^{m_{{\mathrm{1}}},m_{{\mathrm{2}}}} : \mathbb{S}_{m_{{\mathrm{1}}}} \circ \mathbb{S}_{m_{{\mathrm{2}}}} \to \mathbb{S}_{ m_{{\mathrm{1}}}  \cdot  m_{{\mathrm{2}}} }$ as:\\
$ \eta_{ \ottnt{A} }  =  \lambda  \ottmv{x}  .   \ottkw{ret}  \:  \ottmv{x}  $ and $ \mu^{ m_{{\mathrm{1}}} , m_{{\mathrm{2}}} }_{ \ottnt{A} }  =  \lambda  \ottmv{x}  .   \mathbf{join}^{ m_{{\mathrm{1}}} , m_{{\mathrm{2}}} }  \ottmv{x}  $.

Need to check that $\eta$ and $\mu^{m_{{\mathrm{1}}},m_{{\mathrm{2}}}}$ are strong natural transformations.

The diagrams in Figure \ref{strcd4} commute because:
\begin{align*}
&   \mathbb{S}_{ \ottsym{1} } \:  \ottnt{f}   \circ   \eta_{ \ottnt{A} }   & &\;\;\;\;   t^{\mathbb{S}_{ \ottsym{1} } }_{ \ottnt{A} ,  \ottnt{B} }   \circ   (    \text{id}   \times   \eta_{ \ottnt{B} }    )   \\
= &   \lambda  \ottmv{x}  .   (   \mathbf{lift}^{ \ottsym{1} }  \ottnt{f}   )    \:   (   \ottkw{ret}  \:  \ottmv{x}   )   & & =   \lambda  \ottmv{x}  .   (   \mathbf{lift}^{ \ottsym{1} }   (   \lambda  \ottmv{y}  .   (   \mathbf{proj}_1 \:  \ottmv{x}   ,  \ottmv{y}  )    )    )    \:   (   \ottkw{ret}  \:   \mathbf{proj}_2 \:  \ottmv{x}    )   \\
= &   \lambda  \ottmv{x}  .   \ottkw{ret}  \:  \ottmv{x}    \:  \leftindex^{ \ottsym{1} }{\gg}\!\! =^{ \ottsym{1} }   \lambda  \ottmv{y}  .   \ottkw{ret}  \:   (   \ottnt{f}  \:  \ottmv{y}   )     & & =   \lambda  \ottmv{x}  .   \ottkw{ret}  \:   \mathbf{proj}_2 \:  \ottmv{x}     \:  \leftindex^{ \ottsym{1} }{\gg}\!\! =^{ \ottsym{1} }   \lambda  \ottmv{z}  .   \ottkw{ret}  \:   (    (   \lambda  \ottmv{y}  .   (   \mathbf{proj}_1 \:  \ottmv{x}   ,  \ottmv{y}  )    )   \:  \ottmv{z}   )     \\
= &   \lambda  \ottmv{x}  .   (   \lambda  \ottmv{y}  .   \ottkw{ret}  \:   (   \ottnt{f}  \:  \ottmv{y}   )     )    \:  \ottmv{x}  & & =   \lambda  \ottmv{x}  .   \ottkw{ret}  \:   \mathbf{proj}_2 \:  \ottmv{x}     \:  \leftindex^{ \ottsym{1} }{\gg}\!\! =^{ \ottsym{1} }   \lambda  \ottmv{z}  .   \ottkw{ret}  \:   (   \mathbf{proj}_1 \:  \ottmv{x}   ,  \ottmv{z}  )     \\
= &  \lambda  \ottmv{x}  .   \ottkw{ret}  \:   (   \ottnt{f}  \:  \ottmv{x}   )    & & =   \lambda  \ottmv{x}  .   (   \lambda  \ottmv{z}  .   \ottkw{ret}  \:   (   \mathbf{proj}_1 \:  \ottmv{x}   ,  \ottmv{z}  )     )    \:   (   \mathbf{proj}_2 \:  \ottmv{x}   )   \\
= &   \eta_{ \ottnt{B} }   \circ  \ottnt{f}  & & =  \lambda  \ottmv{x}  .   \ottkw{ret}  \:   (   \mathbf{proj}_1 \:  \ottmv{x}   ,   \mathbf{proj}_2 \:  \ottmv{x}   )    =  \lambda  \ottmv{x}  .   \ottkw{ret}  \:  \ottmv{x}   =  \eta_{   \ottnt{A}  \times  \ottnt{B}   }  
\end{align*}

This shows that $\eta$ is a strong natural transformation.\\
Next, we show that $\mu^{m_{{\mathrm{1}}},m_{{\mathrm{2}}}}$ is a strong natural transformation.\\ Note that
\begin{equation} \label{joineqn}
 \ottnt{a}  \:  \leftindex^{ m_{{\mathrm{1}}} }{\gg}\!\! =^{ m_{{\mathrm{2}}} }   \lambda  \ottmv{y}  .  \ottmv{y}   =  \mathbf{join}^{ m_{{\mathrm{1}}} , m_{{\mathrm{2}}} }   (    (   \mathbf{lift}^{ m_{{\mathrm{1}}} }   (   \lambda  \ottmv{y}  .  \ottmv{y}   )    )   \:  \ottnt{a}   )   =  \mathbf{join}^{ m_{{\mathrm{1}}} , m_{{\mathrm{2}}} }   (    (   \lambda  \ottmv{y}  .  \ottmv{y}   )   \:  \ottnt{a}   )   =  \mathbf{join}^{ m_{{\mathrm{1}}} , m_{{\mathrm{2}}} }  \ottnt{a} 
\end{equation}

Now, the left diagram in Figure \ref{strcd5} commutes because:
\begin{align*}
&   \mu^{ m_{{\mathrm{1}}} , m_{{\mathrm{2}}} }_{ \ottnt{B} }   \circ   \mathbb{S}_{ m_{{\mathrm{1}}} } \:   \mathbb{S}_{ m_{{\mathrm{2}}} } \:  \ottnt{f}    \\
= &  \lambda  \ottmv{x}  .   \mathbf{join}^{ m_{{\mathrm{1}}} , m_{{\mathrm{2}}} }   (    (   \mathbf{lift}^{ m_{{\mathrm{1}}} }   (   \mathbf{lift}^{ m_{{\mathrm{2}}} }  \ottnt{f}   )    )   \:  \ottmv{x}   )    \\
= &    \lambda  \ottmv{x}  .  \ottmv{x}   \:  \leftindex^{ m_{{\mathrm{1}}} }{\gg}\!\! =^{ m_{{\mathrm{2}}} }   \lambda  \ottmv{y}  .   (   \mathbf{lift}^{ m_{{\mathrm{2}}} }  \ottnt{f}   )     \:  \ottmv{y}  \\
= &   \lambda  \ottmv{x}  .  \ottmv{x}   \:  \leftindex^{ m_{{\mathrm{1}}} }{\gg}\!\! =^{ m_{{\mathrm{2}}} }   \lambda  \ottmv{y}  .   (   \ottmv{y}  \:  \leftindex^{ m_{{\mathrm{2}}} }{\gg}\!\! =^{ \ottsym{1} }   \lambda  \ottmv{z}  .   \ottkw{ret}  \:   (   \ottnt{f}  \:  \ottmv{z}   )      )    \\
= &   \lambda  \ottmv{x}  .   (   \ottmv{x}  \:  \leftindex^{ m_{{\mathrm{1}}} }{\gg}\!\! =^{ m_{{\mathrm{2}}} }   \lambda  \ottmv{y}  .  \ottmv{y}    )    \:  \leftindex^{  m_{{\mathrm{1}}}  \cdot  m_{{\mathrm{2}}}  }{\gg}\!\! =^{ \ottsym{1} }   \lambda  \ottmv{z}  .   \ottkw{ret}  \:   (   \ottnt{f}  \:  \ottmv{z}   )     \hspace*{2pt} [\text{By Equation (\ref{eq:assoc})}]\\
= &   \lambda  \ottmv{x}  .   \mathbf{join}^{ m_{{\mathrm{1}}} , m_{{\mathrm{2}}} }  \ottmv{x}    \:  \leftindex^{  m_{{\mathrm{1}}}  \cdot  m_{{\mathrm{2}}}  }{\gg}\!\! =^{ \ottsym{1} }   \lambda  \ottmv{z}  .   \ottkw{ret}  \:   (   \ottnt{f}  \:  \ottmv{z}   )     \\
= &   \lambda  \ottmv{x}  .   (   \mathbf{lift}^{  m_{{\mathrm{1}}}  \cdot  m_{{\mathrm{2}}}  }  \ottnt{f}   )    \:   (   \mathbf{join}^{ m_{{\mathrm{1}}} , m_{{\mathrm{2}}} }  \ottmv{x}   )   \\
= &   \mathbb{S}_{  m_{{\mathrm{1}}}  \cdot  m_{{\mathrm{2}}}  } \:  \ottnt{f}   \circ   \mu^{ m_{{\mathrm{1}}} , m_{{\mathrm{2}}} }_{ \ottnt{A} }  
\end{align*}

The right diagram in Figure \ref{strcd5} commutes because:

\begin{align*}
&   \mu^{ m_{{\mathrm{1}}} , m_{{\mathrm{2}}} }_{   \ottnt{A}  \times  \ottnt{B}   }   \circ   \mathbb{S}_{ m_{{\mathrm{1}}} } \:   t^{\mathbb{S}_{ m_{{\mathrm{2}}} } }_{ \ottnt{A} ,  \ottnt{B} }    \circ  t^{\mathbb{S}_{ m_{{\mathrm{1}}} } }_{  \ottnt{A}  ,    \mathbb{S}_{ m_{{\mathrm{2}}} } \:  \ottnt{B}   }  \\
= &     \lambda  \ottmv{x}  .   (   \lambda  \ottmv{y}  .   \mathbf{join}^{ m_{{\mathrm{1}}} , m_{{\mathrm{2}}} }  \ottmv{y}    )    \:   (   \mathbf{lift}^{ m_{{\mathrm{1}}} }   t^{\mathbb{S}_{ m_{{\mathrm{2}}} } }_{ \ottnt{A} ,  \ottnt{B} }    )    \:   (    \lambda  \ottmv{z}  .   (   \mathbf{lift}^{ m_{{\mathrm{1}}} }   (   \lambda  \ottmv{y}  .   (   \mathbf{proj}_1 \:  \ottmv{z}   ,  \ottmv{y}  )    )    )    \:   (   \mathbf{proj}_2 \:  \ottmv{z}   )    )    \:  \ottmv{x}  \\
= &     \lambda  \ottmv{x}  .   (   \lambda  \ottmv{y}  .   \mathbf{join}^{ m_{{\mathrm{1}}} , m_{{\mathrm{2}}} }  \ottmv{y}    )    \:   (   \mathbf{lift}^{ m_{{\mathrm{1}}} }   t^{\mathbb{S}_{ m_{{\mathrm{2}}} } }_{ \ottnt{A} ,  \ottnt{B} }    )    \:   (   \mathbf{lift}^{ m_{{\mathrm{1}}} }   (   \lambda  \ottmv{y}  .   (   \mathbf{proj}_1 \:  \ottmv{x}   ,  \ottmv{y}  )    )    )    \:   (   \mathbf{proj}_2 \:  \ottmv{x}   )   \\
= &    \lambda  \ottmv{x}  .   (   \lambda  \ottmv{y}  .   \mathbf{join}^{ m_{{\mathrm{1}}} , m_{{\mathrm{2}}} }  \ottmv{y}    )    \:   (     \lambda  \ottmv{z}  .   (   \mathbf{lift}^{ m_{{\mathrm{1}}} }   t^{\mathbb{S}_{ m_{{\mathrm{2}}} } }_{ \ottnt{A} ,  \ottnt{B} }    )    \:   (   \mathbf{lift}^{ m_{{\mathrm{1}}} }   (   \lambda  \ottmv{y}  .   (   \mathbf{proj}_1 \:  \ottmv{x}   ,  \ottmv{y}  )    )    )    \:  \ottmv{z}   )    \:   (   \mathbf{proj}_2 \:  \ottmv{x}   )   \\
= &    \lambda  \ottmv{x}  .   (   \lambda  \ottmv{y}  .   \mathbf{join}^{ m_{{\mathrm{1}}} , m_{{\mathrm{2}}} }  \ottmv{y}    )    \:   (   \mathbf{lift}^{ m_{{\mathrm{1}}} }   (    \lambda  \ottmv{z}  .   t^{\mathbb{S}_{ m_{{\mathrm{2}}} } }_{ \ottnt{A} ,  \ottnt{B} }    \:   (   \mathbf{proj}_1 \:  \ottmv{x}   ,  \ottmv{z}  )    )    )    \:   (   \mathbf{proj}_2 \:  \ottmv{x}   )   \\
= &    \lambda  \ottmv{x}  .   (   \lambda  \ottmv{y}  .   \mathbf{join}^{ m_{{\mathrm{1}}} , m_{{\mathrm{2}}} }  \ottmv{y}    )    \:   (   \mathbf{lift}^{ m_{{\mathrm{1}}} }   (    \lambda  \ottmv{z}  .   (    \lambda  \ottmv{u}  .   (   \mathbf{lift}^{ m_{{\mathrm{2}}} }   (   \lambda  \ottmv{v}  .   (   \mathbf{proj}_1 \:  \ottmv{u}   ,  \ottmv{v}  )    )    )    \:   (   \mathbf{proj}_2 \:  \ottmv{u}   )    )    \:   (   \mathbf{proj}_1 \:  \ottmv{x}   ,  \ottmv{z}  )    )    )    \:   (   \mathbf{proj}_2 \:  \ottmv{x}   )   \\
= &  \lambda  \ottmv{x}  .   \mathbf{join}^{ m_{{\mathrm{1}}} , m_{{\mathrm{2}}} }   (    (   \mathbf{lift}^{ m_{{\mathrm{1}}} }   (    \lambda  \ottmv{z}  .   (   \mathbf{lift}^{ m_{{\mathrm{2}}} }   (   \lambda  \ottmv{v}  .   (   \mathbf{proj}_1 \:  \ottmv{x}   ,  \ottmv{v}  )    )    )    \:  \ottmv{z}   )    )   \:   (   \mathbf{proj}_2 \:  \ottmv{x}   )    )    \\
= &    \lambda  \ottmv{x}  .   \mathbf{proj}_2 \:  \ottmv{x}    \:  \leftindex^{ m_{{\mathrm{1}}} }{\gg}\!\! =^{ m_{{\mathrm{2}}} }   \lambda  \ottmv{z}  .   (   \mathbf{lift}^{ m_{{\mathrm{2}}} }   (   \lambda  \ottmv{v}  .   (   \mathbf{proj}_1 \:  \ottmv{x}   ,  \ottmv{v}  )    )    )     \:  \ottmv{z}  \\
= &   \lambda  \ottmv{x}  .   \mathbf{proj}_2 \:  \ottmv{x}    \:  \leftindex^{ m_{{\mathrm{1}}} }{\gg}\!\! =^{ m_{{\mathrm{2}}} }   \lambda  \ottmv{z}  .   (   \ottmv{z}  \:  \leftindex^{ m_{{\mathrm{2}}} }{\gg}\!\! =^{ \ottsym{1} }   \lambda  \ottmv{u}  .   \ottkw{ret}  \:   (   \mathbf{proj}_1 \:  \ottmv{x}   ,  \ottmv{u}  )      )    \\
= &   \lambda  \ottmv{x}  .   (    \mathbf{proj}_2 \:  \ottmv{x}   \:  \leftindex^{ m_{{\mathrm{1}}} }{\gg}\!\! =^{ m_{{\mathrm{2}}} }   (   \lambda  \ottmv{y}  .  \ottmv{y}   )    )    \:  \leftindex^{  m_{{\mathrm{1}}}  \cdot  m_{{\mathrm{2}}}  }{\gg}\!\! =^{ \ottsym{1} }   \lambda  \ottmv{u}  .   \ottkw{ret}  \:   (   \mathbf{proj}_1 \:  \ottmv{x}   ,  \ottmv{u}  )     \hspace*{3pt} [\text{By Equation (\ref{eq:assoc})}] \\
= &   \lambda  \ottmv{x}  .   (   \mathbf{join}^{ m_{{\mathrm{1}}} , m_{{\mathrm{2}}} }   (   \mathbf{proj}_2 \:  \ottmv{x}   )    )    \:  \leftindex^{  m_{{\mathrm{1}}}  \cdot  m_{{\mathrm{2}}}  }{\gg}\!\! =^{ \ottsym{1} }   \lambda  \ottmv{u}  .   \ottkw{ret}  \:   (   \mathbf{proj}_1 \:  \ottmv{x}   ,  \ottmv{u}  )     
\end{align*}

and

\begin{align*}
&    t^{\mathbb{S}_{  m_{{\mathrm{1}}}  \cdot  m_{{\mathrm{2}}}  } }_{ \ottnt{A} ,  \ottnt{B} }   \circ   \text{id}    \times   \mu^{ m_{{\mathrm{1}}} , m_{{\mathrm{2}}} }_{ \ottnt{B} }   \\
= &   \lambda  \ottmv{x}  .   (    \lambda  \ottmv{z}  .   (   \mathbf{lift}^{  m_{{\mathrm{1}}}  \cdot  m_{{\mathrm{2}}}  }   (   \lambda  \ottmv{y}  .   (   \mathbf{proj}_1 \:  \ottmv{z}   ,  \ottmv{y}  )    )    )    \:   (   \mathbf{proj}_2 \:  \ottmv{z}   )    )    \:   (   \mathbf{proj}_1 \:  \ottmv{x}   ,   \mathbf{join}^{ m_{{\mathrm{1}}} , m_{{\mathrm{2}}} }   \mathbf{proj}_2 \:  \ottmv{x}    )   \\
= &   \lambda  \ottmv{x}  .   (   \mathbf{lift}^{  m_{{\mathrm{1}}}  \cdot  m_{{\mathrm{2}}}  }   (   \lambda  \ottmv{y}  .   (   \mathbf{proj}_1 \:  \ottmv{x}   ,  \ottmv{y}  )    )    )    \:   (   \mathbf{join}^{ m_{{\mathrm{1}}} , m_{{\mathrm{2}}} }   \mathbf{proj}_2 \:  \ottmv{x}    )   \\
= &   \lambda  \ottmv{x}  .   (   \mathbf{join}^{ m_{{\mathrm{1}}} , m_{{\mathrm{2}}} }   (   \mathbf{proj}_2 \:  \ottmv{x}   )    )    \:  \leftindex^{  m_{{\mathrm{1}}}  \cdot  m_{{\mathrm{2}}}  }{\gg}\!\! =^{ \ottsym{1} }   \lambda  \ottmv{u}  .   \ottkw{ret}  \:   (   \mathbf{proj}_1 \:  \ottmv{x}   ,  \ottmv{u}  )    . 
\end{align*}

\begin{figure}
\centering
\begin{subfigure}{0.4\textwidth}
\begin{tikzcd}
\ottnt{A} \arrow{r}{ \eta_{ \ottnt{A} } } \arrow{d}{\ottnt{f}} &  \mathbb{S}_{ \ottsym{1} } \:  \ottnt{A}  \arrow{d}{ \mathbb{S}_{ \ottsym{1} } \:  \ottnt{f} } \\
\ottnt{B} \arrow{r}{ \eta_{ \ottnt{B} } } &  \mathbb{S}_{ \ottsym{1} } \:  \ottnt{B}  
\end{tikzcd}
\end{subfigure}
\begin{subfigure}{0.4\textwidth}
\begin{tikzcd}
 \ottnt{A}  \times  \ottnt{B}  \arrow{r}{  \text{id}   \times   \eta_{ \ottnt{B} }  } \arrow{d}{ \text{id} } &   \ottnt{A}   \times    \mathbb{S}_{ \ottsym{1} } \:  \ottnt{B}    \arrow{d}{ t^{\mathbb{S}_{ \ottsym{1} } }_{ \ottnt{A} ,  \ottnt{B} } } \\
 \ottnt{A}  \times  \ottnt{B}  \arrow{r}{ \eta_{   \ottnt{A}  \times  \ottnt{B}   } } &  \mathbb{S}_{ \ottsym{1} } \:   (   \ottnt{A}  \times  \ottnt{B}   )   
\end{tikzcd}
\end{subfigure}
\caption{Commutative diagram}
\label{strcd4}
\end{figure} 

\begin{figure}
\centering
\begin{subfigure}{0.4\textwidth}
\begin{tikzcd}
 \mathbb{S}_{ m_{{\mathrm{1}}} } \:   \mathbb{S}_{ m_{{\mathrm{2}}} } \:  \ottnt{A}   \arrow{r}{ \mu^{ m_{{\mathrm{1}}} , m_{{\mathrm{2}}} }_{ \ottnt{A} } } \arrow{d}{ \mathbb{S}_{ m_{{\mathrm{1}}} } \:   \mathbb{S}_{ m_{{\mathrm{2}}} } \:  \ottnt{f}  } &  \mathbb{S}_{  m_{{\mathrm{1}}}  \cdot  m_{{\mathrm{2}}}  } \:  \ottnt{A}  \arrow{d}{ \mathbb{S}_{  m_{{\mathrm{1}}}  \cdot  m_{{\mathrm{2}}}  } \:  \ottnt{f} } \\
 \mathbb{S}_{ m_{{\mathrm{1}}} } \:   \mathbb{S}_{ m_{{\mathrm{2}}} } \:  \ottnt{B}   \arrow{r}{ \mu^{ m_{{\mathrm{1}}} , m_{{\mathrm{2}}} }_{ \ottnt{B} } } &  \mathbb{S}_{  m_{{\mathrm{1}}}  \cdot  m_{{\mathrm{2}}}  } \:  \ottnt{B}  
\end{tikzcd}
\end{subfigure}
\begin{subfigure}{0.4\textwidth}
\begin{tikzcd}
  \ottnt{A}   \times    \mathbb{S}_{ m_{{\mathrm{1}}} } \:   \mathbb{S}_{ m_{{\mathrm{2}}} } \:  \ottnt{B}     \arrow{r}{  \text{id}   \times   \mu^{ m_{{\mathrm{1}}} , m_{{\mathrm{2}}} }_{ \ottnt{B} }  } \arrow{d}[left]{ \mathbb{S}_{ m_{{\mathrm{1}}} } \:   t^{\mathbb{S}_{ m_{{\mathrm{2}}} } }_{ \ottnt{A} ,  \ottnt{B} }   \circ  t^{\mathbb{S}_{ m_{{\mathrm{1}}} } }_{  \ottnt{A}  ,    \mathbb{S}_{ m_{{\mathrm{2}}} } \:  \ottnt{B}   } } &   \ottnt{A}   \times    \mathbb{S}_{  m_{{\mathrm{1}}}  \cdot  m_{{\mathrm{2}}}  } \:  \ottnt{B}    \arrow{d}{ t^{\mathbb{S}_{  m_{{\mathrm{1}}}  \cdot  m_{{\mathrm{2}}}  } }_{ \ottnt{A} ,  \ottnt{B} } } \\
 \mathbb{S}_{ m_{{\mathrm{1}}} } \:   \mathbb{S}_{ m_{{\mathrm{2}}} } \:   (   \ottnt{A}  \times  \ottnt{B}   )    \arrow{r}{ \mu^{ m_{{\mathrm{1}}} , m_{{\mathrm{2}}} }_{   \ottnt{A}  \times  \ottnt{B}   } } &  \mathbb{S}_{  m_{{\mathrm{1}}}  \cdot  m_{{\mathrm{2}}}  } \:   (   \ottnt{A}  \times  \ottnt{B}   )   
\end{tikzcd}
\end{subfigure}
\caption{Commutative diagram}
\label{strcd5}
\end{figure} 

This shows that $\mu^{m_{{\mathrm{1}}},m_{{\mathrm{2}}}}$ is a strong natural transformation.\\
Next, we show that $\mu^{m_{{\mathrm{1}}},m_{{\mathrm{2}}}}$ is natural in both $m_{{\mathrm{1}}}$ and $m_{{\mathrm{2}}}$.

\begin{figure}
\centering
\begin{subfigure}{0.4\textwidth}
\begin{tikzcd}
 \mathbb{S}_{ m_{{\mathrm{1}}} } \:   \mathbb{S}_{ m_{{\mathrm{2}}} } \:  \ottnt{A}   \arrow{r}{ \mu^{ m_{{\mathrm{1}}} , m_{{\mathrm{2}}} }_{ \ottnt{A} } } \arrow{d}[left]{ \mathbb{S}^{ m_{{\mathrm{1}}}  \leq  m'_{{\mathrm{1}}} }_{  \mathbb{S}_{ m_{{\mathrm{2}}} } \:  \ottnt{A}  } } &  \mathbb{S}_{  m_{{\mathrm{1}}}  \cdot  m_{{\mathrm{2}}}  } \:  \ottnt{A}  \arrow{d}{ \mathbb{S}^{  m_{{\mathrm{1}}}  \cdot  m_{{\mathrm{2}}}   \leq   m'_{{\mathrm{1}}}  \cdot  m_{{\mathrm{2}}}  }_{ \ottnt{A} } } \\
 \mathbb{S}_{ m'_{{\mathrm{1}}} } \:   \mathbb{S}_{ m_{{\mathrm{2}}} } \:  \ottnt{A}   \arrow{r}{ \mu^{ m'_{{\mathrm{1}}} , m_{{\mathrm{2}}} }_{ \ottnt{A} } } &  \mathbb{S}_{  m'_{{\mathrm{1}}}  \cdot  m_{{\mathrm{2}}}  } \:  \ottnt{A}  
\end{tikzcd}
\end{subfigure}
\begin{subfigure}{0.4\textwidth}
\begin{tikzcd}
 \mathbb{S}_{ m_{{\mathrm{1}}} } \:   \mathbb{S}_{ m_{{\mathrm{2}}} } \:  \ottnt{A}   \arrow{r}{ \mu^{ m_{{\mathrm{1}}} , m_{{\mathrm{2}}} }_{ \ottnt{A} } } \arrow{d}[left]{ \mathbb{S}_{ m_{{\mathrm{1}}} } \:   \mathbb{S}^{ m_{{\mathrm{2}}}  \leq  m'_{{\mathrm{2}}} }_{ \ottnt{A} }  } &  \mathbb{S}_{  m_{{\mathrm{1}}}  \cdot  m_{{\mathrm{2}}}  } \:  \ottnt{A}  \arrow{d}{ \mathbb{S}^{  m_{{\mathrm{1}}}  \cdot  m_{{\mathrm{2}}}   \leq   m_{{\mathrm{1}}}  \cdot  m'_{{\mathrm{2}}}  }_{ \ottnt{A} } } \\
 \mathbb{S}_{ m_{{\mathrm{1}}} } \:   \mathbb{S}_{ m'_{{\mathrm{2}}} } \:  \ottnt{A}   \arrow{r}{ \mu^{ m_{{\mathrm{1}}} , m'_{{\mathrm{2}}} }_{ \ottnt{A} } } &  \mathbb{S}_{  m_{{\mathrm{1}}}  \cdot  m'_{{\mathrm{2}}}  } \:  \ottnt{A}  
\end{tikzcd}
\end{subfigure}
\caption{Commutative diagram}
\label{strcd6}
\end{figure} 

The diagrams in Figure \ref{strcd6} commute because:

\begin{align*}
&   \mu^{ m'_{{\mathrm{1}}} , m_{{\mathrm{2}}} }_{ \ottnt{A} }   \circ   \mathbb{S}^{ m_{{\mathrm{1}}}  \leq  m'_{{\mathrm{1}}} }_{  \mathbb{S}_{ m_{{\mathrm{2}}} } \:  \ottnt{A}  }   & &\;\;\;\;   \mu^{ m_{{\mathrm{1}}} , m'_{{\mathrm{2}}} }_{ \ottnt{A} }   \circ   \mathbb{S}_{ m_{{\mathrm{1}}} } \:   \mathbb{S}^{ m_{{\mathrm{2}}}  \leq  m'_{{\mathrm{2}}} }_{ \ottnt{A} }    \\
= &  \lambda  \ottmv{x}  .   \mathbf{join}^{ m'_{{\mathrm{1}}} , m_{{\mathrm{2}}} }   (   \mathbf{up}^{ m_{{\mathrm{1}}} , m'_{{\mathrm{1}}} }  \ottmv{x}   )    & & =   \lambda  \ottmv{x}  .   \mathbf{join}^{ m_{{\mathrm{1}}} , m'_{{\mathrm{2}}} }   (   \mathbf{lift}^{ m_{{\mathrm{1}}} }   (   \lambda  \ottmv{y}  .   \mathbf{up}^{ m_{{\mathrm{2}}} , m'_{{\mathrm{2}}} }  \ottmv{y}    )    )     \:  \ottmv{x}  \\
= &   \lambda  \ottmv{x}  .   \mathbf{up}^{ m_{{\mathrm{1}}} , m'_{{\mathrm{1}}} }  \ottmv{x}    \:  \leftindex^{ m'_{{\mathrm{1}}} }{\gg}\!\! =^{ m_{{\mathrm{2}}} }   \lambda  \ottmv{y}  .  \ottmv{y}   & & =   \lambda  \ottmv{x}  .  \ottmv{x}   \:  \leftindex^{ m_{{\mathrm{1}}} }{\gg}\!\! =^{ m'_{{\mathrm{2}}} }   (   \lambda  \ottmv{y}  .   \mathbf{up}^{ m_{{\mathrm{2}}} , m'_{{\mathrm{2}}} }  \ottmv{y}    )   \\
= &  \lambda  \ottmv{x}  .   \mathbf{up}^{  m_{{\mathrm{1}}}  \cdot  m_{{\mathrm{2}}}  ,  m'_{{\mathrm{1}}}  \cdot  m_{{\mathrm{2}}}  }   (   \ottmv{x}  \:  \leftindex^{ m_{{\mathrm{1}}} }{\gg}\!\! =^{ m_{{\mathrm{2}}} }   \lambda  \ottmv{y}  .  \ottmv{y}    )    \hspace*{2pt} [\text{By Eqn. (\ref{eq:natl})}] & & =  \lambda  \ottmv{x}  .   \mathbf{up}^{  m_{{\mathrm{1}}}  \cdot  m_{{\mathrm{2}}}  ,  m_{{\mathrm{1}}}  \cdot  m'_{{\mathrm{2}}}  }   (   \ottmv{x}  \:  \leftindex^{ m_{{\mathrm{1}}} }{\gg}\!\! =^{ m_{{\mathrm{2}}} }   \lambda  \ottmv{y}  .  \ottmv{y}    )    \hspace*{2pt} [\text{By Eqn. (\ref{eq:natr})}]\\
= &  \lambda  \ottmv{x}  .   \mathbf{up}^{  m_{{\mathrm{1}}}  \cdot  m_{{\mathrm{2}}}  ,  m'_{{\mathrm{1}}}  \cdot  m_{{\mathrm{2}}}  }   (   \mathbf{join}^{ m_{{\mathrm{1}}} , m_{{\mathrm{2}}} }  \ottmv{x}   )    & & =  \lambda  \ottmv{x}  .   \mathbf{up}^{  m_{{\mathrm{1}}}  \cdot  m_{{\mathrm{2}}}  ,  m_{{\mathrm{1}}}  \cdot  m'_{{\mathrm{2}}}  }   (   \mathbf{join}^{ m_{{\mathrm{1}}} , m_{{\mathrm{2}}} }  \ottmv{x}   )   \\
= &   \mathbb{S}^{  m_{{\mathrm{1}}}  \cdot  m_{{\mathrm{2}}}   \leq   m'_{{\mathrm{1}}}  \cdot  m_{{\mathrm{2}}}  }_{ \ottnt{A} }   \circ   \mu^{ m_{{\mathrm{1}}} , m_{{\mathrm{2}}} }_{ \ottnt{A} }   & & =   \mathbb{S}^{  m_{{\mathrm{1}}}  \cdot  m_{{\mathrm{2}}}   \leq   m_{{\mathrm{1}}}  \cdot  m'_{{\mathrm{2}}}  }_{ \ottnt{A} }   \circ   \mu^{ m_{{\mathrm{1}}} , m_{{\mathrm{2}}} }_{ \ottnt{A} }  
\end{align*}

This shows that $\mu^{m_{{\mathrm{1}}},m_{{\mathrm{2}}}}$ is natural in both $m_{{\mathrm{1}}}$ and $m_{{\mathrm{2}}}$.\\ Now we show that $\mathbb{S}$ satisfies the axioms for lax monoidal functor.

\begin{figure}
\centering
\begin{subfigure}{0.4\textwidth}
\begin{tikzcd}
 \mathbb{S}_{ \ottsym{1} } \:   \mathbb{S}_{ m } \:  \ottnt{A}   \arrow{dr}[below left]{ \mu^{ \ottsym{1} , m }_{ \ottnt{A} } } &  \mathbb{S}_{ m } \:  \ottnt{A}  \arrow{d}{ \text{id} } \arrow{l}[above]{ \eta_{  \mathbb{S}_{ m } \:  \ottnt{A}  } } \arrow{r}{ \mathbb{S}_{ m } \:   \eta_{ \ottnt{A} }  } &  \mathbb{S}_{ m } \:   \mathbb{S}_{ \ottsym{1} } \:  \ottnt{A}   \arrow{dl}{ \mu^{ m , \ottsym{1} }_{ \ottnt{A} } } \\
&  \mathbb{S}_{ m } \:  \ottnt{A}  &
\end{tikzcd}
\end{subfigure}
\begin{subfigure}{0.4\textwidth}
\begin{tikzcd}
 \mathbb{S}_{ m_{{\mathrm{1}}} } \:   \mathbb{S}_{ m_{{\mathrm{2}}} } \:   \mathbb{S}_{ m_{{\mathrm{3}}} } \:  \ottnt{A}    \arrow{r}{ \mathbb{S}_{ m_{{\mathrm{1}}} } \:   \mu^{ m_{{\mathrm{2}}} , m_{{\mathrm{3}}} }_{ \ottnt{A} }  } \arrow{d}[left]{ \mu^{ m_{{\mathrm{1}}} , m_{{\mathrm{2}}} }_{  \mathbb{S}_{ m_{{\mathrm{3}}} } \:  \ottnt{A}  } } &  \mathbb{S}_{ m_{{\mathrm{1}}} } \:   \mathbb{S}_{  m_{{\mathrm{2}}}  \cdot  m_{{\mathrm{3}}}  } \:  \ottnt{A}   \arrow{d}{ \mu^{ m_{{\mathrm{1}}} ,  m_{{\mathrm{2}}}  \cdot  m_{{\mathrm{3}}}  }_{ \ottnt{A} } } \\
 \mathbb{S}_{  m_{{\mathrm{1}}}  \cdot  m_{{\mathrm{2}}}  } \:   \mathbb{S}_{ m_{{\mathrm{3}}} } \:  \ottnt{A}   \arrow{r}{ \mu^{  m_{{\mathrm{1}}}  \cdot  m_{{\mathrm{2}}}  , m_{{\mathrm{3}}} }_{ \ottnt{A} } } &  \mathbb{S}_{  m_{{\mathrm{1}}}  \cdot    m_{{\mathrm{2}}}  \cdot  m_{{\mathrm{3}}}    } \:  \ottnt{A}  
\end{tikzcd}
\end{subfigure}
\caption{Commutative diagram}
\label{strcd7}
\end{figure} 

The left diagram in Figure \ref{strcd7} commutes because:
\begin{align*}
&   \mu^{ \ottsym{1} , m }_{ \ottnt{A} }   \circ   \eta_{  \mathbb{S}_{ m } \:  \ottnt{A}  }   & & \; \;   \mu^{ m , \ottsym{1} }_{ \ottnt{A} }   \circ   \mathbb{S}_{ m } \:   \eta_{ \ottnt{A} }    \\
= &  \lambda  \ottmv{x}  .   \mathbf{join}^{ \ottsym{1} , m }   (   \ottkw{ret}  \:  \ottmv{x}   )    & & =   \lambda  \ottmv{x}  .   \mathbf{join}^{ m , \ottsym{1} }   (   \mathbf{lift}^{ m }   (   \lambda  \ottmv{y}  .   \ottkw{ret}  \:  \ottmv{y}    )    )     \:  \ottmv{x}  \\ 
= &   \lambda  \ottmv{x}  .   \ottkw{ret}  \:  \ottmv{x}    \:  \leftindex^{ \ottsym{1} }{\gg}\!\! =^{ m }   (   \lambda  \ottmv{y}  .  \ottmv{y}   )   & & =   \lambda  \ottmv{x}  .  \ottmv{x}   \:  \leftindex^{ m }{\gg}\!\! =^{ \ottsym{1} }   \lambda  \ottmv{y}  .   \ottkw{ret}  \:  \ottmv{y}    \\ 
= &   \lambda  \ottmv{x}  .   (   \lambda  \ottmv{y}  .  \ottmv{y}   )    \:  \ottmv{x}  \hspace*{2pt} [\text{By Eqn. (\ref{eq:idl})}] & & =  \lambda  \ottmv{x}  .  \ottmv{x}  \hspace*{2pt} [\text{By Eqn. (\ref{eq:idr})}] \\
= &  \lambda  \ottmv{x}  .  \ottmv{x}  =  \text{id}  & & =  \text{id} 
\end{align*}

The right diagram in Figure \ref{strcd7} commutes because:
\begin{align*}
&   \mu^{ m_{{\mathrm{1}}} ,  m_{{\mathrm{2}}}  \cdot  m_{{\mathrm{3}}}  }_{ \ottnt{A} }   \circ   \mathbb{S}_{ m_{{\mathrm{1}}} } \:   \mu^{ m_{{\mathrm{2}}} , m_{{\mathrm{3}}} }_{ \ottnt{A} }    \\
= &   \lambda  \ottmv{x}  .   \mathbf{join}^{ m_{{\mathrm{1}}} ,  m_{{\mathrm{2}}}  \cdot  m_{{\mathrm{3}}}  }   (   \mathbf{lift}^{ m_{{\mathrm{1}}} }   (   \lambda  \ottmv{y}  .   \mathbf{join}^{ m_{{\mathrm{2}}} , m_{{\mathrm{3}}} }  \ottmv{y}    )    )     \:  \ottmv{x}  \\
= &   \lambda  \ottmv{x}  .  \ottmv{x}   \:  \leftindex^{ m_{{\mathrm{1}}} }{\gg}\!\! =^{  m_{{\mathrm{2}}}  \cdot  m_{{\mathrm{3}}}  }   \lambda  \ottmv{y}  .   \mathbf{join}^{ m_{{\mathrm{2}}} , m_{{\mathrm{3}}} }  \ottmv{y}    \\
= &   \lambda  \ottmv{x}  .  \ottmv{x}   \:  \leftindex^{ m_{{\mathrm{1}}} }{\gg}\!\! =^{  m_{{\mathrm{2}}}  \cdot  m_{{\mathrm{3}}}  }   \lambda  \ottmv{y}  .   (   \ottmv{y}  \:  \leftindex^{ m_{{\mathrm{2}}} }{\gg}\!\! =^{ m_{{\mathrm{3}}} }   \lambda  \ottmv{z}  .  \ottmv{z}    )    \\
= &   \lambda  \ottmv{x}  .   (   \ottmv{x}  \:  \leftindex^{ m_{{\mathrm{1}}} }{\gg}\!\! =^{ m_{{\mathrm{2}}} }   (   \lambda  \ottmv{w}  .  \ottmv{w}   )    )    \:  \leftindex^{  m_{{\mathrm{1}}}  \cdot  m_{{\mathrm{2}}}  }{\gg}\!\! =^{ m_{{\mathrm{3}}} }   \lambda  \ottmv{z}  .  \ottmv{z}   \hspace*{2pt} [\text{By Eqn. (\ref{eq:assoc})}] \\
= &  \lambda  \ottmv{x}  .   \mathbf{join}^{  m_{{\mathrm{1}}}  \cdot  m_{{\mathrm{2}}}  , m_{{\mathrm{3}}} }   (   \ottmv{x}  \:  \leftindex^{ m_{{\mathrm{1}}} }{\gg}\!\! =^{ m_{{\mathrm{2}}} }   \lambda  \ottmv{w}  .  \ottmv{w}    )    \\
= &  \lambda  \ottmv{x}  .   \mathbf{join}^{  m_{{\mathrm{1}}}  \cdot  m_{{\mathrm{2}}}  , m_{{\mathrm{3}}} }   (   \mathbf{join}^{ m_{{\mathrm{1}}} , m_{{\mathrm{2}}} }  \ottmv{x}   )    \\
= &   \mu^{  m_{{\mathrm{1}}}  \cdot  m_{{\mathrm{2}}}  , m_{{\mathrm{3}}} }_{ \ottnt{A} }   \circ   \mu^{ m_{{\mathrm{1}}} , m_{{\mathrm{2}}} }_{  \mathbb{S}_{ m_{{\mathrm{3}}} } \:  \ottnt{A}  }  
\end{align*}

Hence $\mathbb{S}$ is a lax monoidal functor.\\
Now $\mathbb{S}$ is also a strong monoidal functor because $\eta$ and $\mu^{m_{{\mathrm{1}}},m_{{\mathrm{2}}}}$ are isomorphisms.\\
Define: $\epsilon : \mathbb{S}_1 \to \Id{}$ and $\delta^{m_{{\mathrm{1}}},m_{{\mathrm{2}}}} : \mathbb{S}_{ m_{{\mathrm{1}}}  \cdot  m_{{\mathrm{2}}} } \to \mathbb{S}_{m_{{\mathrm{1}}}} \circ \mathbb{S}_{m_{{\mathrm{2}}}}$ as:\\
$\epsilon_A =  \lambda  \ottmv{x}  .   \mathbf{extr} \:  \ottmv{x}  $ and $\delta^{m_{{\mathrm{1}}},m_{{\mathrm{2}}}}_A =  \lambda  \ottmv{x}  .   \mathbf{fork}^{ m_{{\mathrm{1}}} , m_{{\mathrm{2}}} }  \ottmv{x}  $.\\
By dualizing the arguments presented above, we can show that $\epsilon$ and $\delta^{m_{{\mathrm{1}}},m_{{\mathrm{2}}}}$ are strong natural transformations.\\
Further, $\eta$ and $\epsilon$ are inverses of one another because $\eta_A \circ \epsilon_A =  \lambda  \ottmv{x}  .   \ottkw{ret}  \:   (   \mathbf{extr} \:  \ottmv{x}   )    =  \lambda  \ottmv{x}  .  \ottmv{x} $ and $\epsilon_A \circ \eta_A =  \lambda  \ottmv{x}  .   \mathbf{extr} \:   (   \ottkw{ret}  \:  \ottmv{x}   )    =  \lambda  \ottmv{x}  .  \ottmv{x} $.\\
Similarly, $\mu^{m_{{\mathrm{1}}},m_{{\mathrm{2}}}}$ and $\delta^{m_{{\mathrm{1}}},m_{{\mathrm{2}}}}$ are inverses of one another.\\

Hence, $\mathbb{S}$ is a strong monoidal functor.\\
As such, $\llbracket \_ \rrbracket_{( \mathds{F} ,\mathbb{S})}$ provides a sound interpretation of GMCC($ \mathcal{M} $). This is the generic model of GMCC($ \mathcal{M} $).\\

Now we show that if $ \Gamma  \vdash  \ottnt{b}  :  \ottnt{B} $ in GMCC($ \mathcal{M} $), then $\llbracket  \Gamma  \vdash  \ottnt{b}  :  \ottnt{B}  \rrbracket_{( \mathds{F} ,\mathbb{S})} = \lambda x . b \{ \text{proj}_i^n \:  x / x_i \} \in \text{Hom}_{\mathds{F}} (\Pi A_i , B)$ where $\Gamma = x_1 : A_1 , x_2 : A_2 , \ldots , x_n : A_n$ and $\Pi A_i = A_1 \times A_2 \times \ldots \times A_n$. For $n \in \mathbb{N}$ and $1 \leq i \leq n$, $\mathbf{proj}_i^n \: x$ is defined as:
\begin{align*}
\mathbf{proj}_1^1 \: x & = x \\
\mathbf{proj}_i^2 \: x & = \mathbf{proj}_i \: x  \hspace*{3pt} \text{for } i = 1, 2\\
\mathbf{proj}_i^{n+1} \: x & = \mathbf{proj}_i^n \: ( \mathbf{proj}_1 \:  \ottmv{x} ) \hspace*{3pt} \text{for } 2 \leq n \text{ and } 1 \leq i \leq n \\
\mathbf{proj}_{n+1}^{n+1} \: x & =  \mathbf{proj}_2 \:  \ottmv{x}  \hspace{3pt} \text{for } 2 \leq n
\end{align*}

The proof is by induction on $ \Gamma  \vdash  \ottnt{b}  :  \ottnt{B} $. 
\begin{itemize}
\item $\lambda$-calculus. Standard.
\item \Rref{MC-Return}. Have: $ \Gamma  \vdash   \ottkw{ret}  \:  \ottnt{b}   :   S_{ \ottsym{1} } \:  \ottnt{B}  $ where $ \Gamma  \vdash  \ottnt{b}  :  \ottnt{B} $.\\
By IH, $ \llbracket  \Gamma  \vdash  \ottnt{b}  :  \ottnt{B}  \rrbracket  =   \lambda  \ottmv{x}  .  \ottnt{b}   \:   \{ \mathbf{proj}_i^n \:  \ottmv{x}  /  \ottmv{x} _i \}  $. \\
Now, $ \llbracket  \Gamma  \vdash   \ottkw{ret}  \:  \ottnt{b}   :   S_{ \ottsym{1} } \:  \ottnt{B}   \rrbracket  =   \eta_{ \ottnt{B} }   \circ   \llbracket  \Gamma  \vdash  \ottnt{b}  :  \ottnt{B}  \rrbracket   =   \lambda  \ottmv{x}  .   (   \ottkw{ret}  \:  \ottnt{b}   )    \:   \{ \mathbf{proj}_i^n \:  \ottmv{x}  /  \ottmv{x} _i \}  $.
\item \Rref{MC-Extract}. Similar to \rref{MC-Return}.
\item \Rref{MC-Join}. Have: $ \Gamma  \vdash   \mathbf{join}^{ m_{{\mathrm{1}}} , m_{{\mathrm{2}}} }  \ottnt{b}   :   S_{  m_{{\mathrm{1}}}  \cdot  m_{{\mathrm{2}}}  } \:  \ottnt{B}  $ where $ \Gamma  \vdash  \ottnt{b}  :   S_{ m_{{\mathrm{1}}} } \:   S_{ m_{{\mathrm{2}}} } \:  \ottnt{B}   $.\\
By IH, $ \llbracket  \Gamma  \vdash  \ottnt{b}  :  \ottnt{B}  \rrbracket  =   \lambda  \ottmv{x}  .  \ottnt{b}   \:   \{ \mathbf{proj}_i^n \:  \ottmv{x}  /  \ottmv{x} _i \}  $. \\
Now, $ \llbracket  \Gamma  \vdash   \mathbf{join}^{ m_{{\mathrm{1}}} , m_{{\mathrm{2}}} }  \ottnt{b}   :   S_{  m_{{\mathrm{1}}}  \cdot  m_{{\mathrm{2}}}  } \:  \ottnt{B}   \rrbracket  =   \mu^{ m_{{\mathrm{1}}} , m_{{\mathrm{2}}} }_{ \ottnt{B} }   \circ   \llbracket  \Gamma  \vdash  \ottnt{b}  :  \ottnt{B}  \rrbracket   =   \lambda  \ottmv{x}  .   (   \mathbf{join}^{ m_{{\mathrm{1}}} , m_{{\mathrm{2}}} }  \ottnt{b}   )    \:   \{ \mathbf{proj}_i^n \:  \ottmv{x}  /  \ottmv{x} _i \}  $.
\item \Rref{MC-Fork}. Similar to \rref{MC-Join}.
\item \Rref{MC-Fmap}. Have: $ \Gamma  \vdash   \mathbf{lift}^{ m }  \ottnt{f}   :    S_{ m } \:  \ottnt{B}   \to   S_{ m } \:  \ottnt{C}   $ where $ \Gamma  \vdash  \ottnt{f}  :   \ottnt{B}  \to  \ottnt{C}  $.\\
By IH, $ \llbracket  \Gamma  \vdash  \ottnt{f}  :   \ottnt{B}  \to  \ottnt{C}   \rrbracket  =   \lambda  \ottmv{x}  .  \ottnt{f}   \:   \{ \mathbf{proj}_i^n \:  \ottmv{x}  /  \ottmv{x} _i \}  $.\\
Now, \begin{align*}
&  \llbracket  \Gamma  \vdash   \mathbf{lift}^{ m }  \ottnt{f}   :    S_{ m } \:  \ottnt{B}   \to   S_{ m } \:  \ottnt{C}    \rrbracket  \\
= &  \Lambda \Big(     \mathbb{S}_{ m } \:   (   \Lambda^{-1}   \llbracket  \Gamma  \vdash  \ottnt{f}  :   \ottnt{B}  \to  \ottnt{C}   \rrbracket    )    \circ   t^{\mathbb{S}_{ m } }_{ \Gamma ,  \ottnt{B} }     \Big)  \\
= &   \lambda  \ottmv{y}  .   \lambda  \ottmv{z}  .   (    \mathbb{S}_{ m } \:   (   \Lambda^{-1}   \llbracket  \Gamma  \vdash  \ottnt{f}  :   \ottnt{B}  \to  \ottnt{C}   \rrbracket    )    \circ   t^{\mathbb{S}_{ m } }_{ \Gamma ,  \ottnt{B} }    )     \:   (  \ottmv{y}  ,  \ottmv{z}  )   \\
= &    \lambda  \ottmv{y}  .   \lambda  \ottmv{z}  .   (   \mathbf{lift}^{ m }   (   \Lambda^{-1}   \llbracket  \Gamma  \vdash  \ottnt{f}  :   \ottnt{B}  \to  \ottnt{C}   \rrbracket    )    )     \:   (   \mathbf{lift}^{ m }   (   \lambda  \ottmv{w}  .   (  \ottmv{y}  ,  \ottmv{w}  )    )    )    \:  \ottmv{z}  \\
= &  \lambda  \ottmv{y}  .   \mathbf{lift}^{ m }   (    \lambda  \ottmv{z}  .   (   \Lambda^{-1}   \llbracket  \Gamma  \vdash  \ottnt{f}  :   \ottnt{B}  \to  \ottnt{C}   \rrbracket    )    \:   (  \ottmv{y}  ,  \ottmv{z}  )    )    \\
= &  \lambda  \ottmv{y}  .   \mathbf{lift}^{ m }   (    \lambda  \ottmv{z}  .   (    \lambda  \ottmv{u}  .   (    \llbracket  \Gamma  \vdash  \ottnt{f}  :   \ottnt{B}  \to  \ottnt{C}   \rrbracket   \:   \mathbf{proj}_1 \:  \ottmv{u}    )    \:   (   \mathbf{proj}_2 \:  \ottmv{u}   )    )    \:   (  \ottmv{y}  ,  \ottmv{z}  )    )    \\
= &  \lambda  \ottmv{y}  .   \mathbf{lift}^{ m }   (    \lambda  \ottmv{z}  .   (    \llbracket  \Gamma  \vdash  \ottnt{f}  :   \ottnt{B}  \to  \ottnt{C}   \rrbracket   \:  \ottmv{y}   )    \:  \ottmv{z}   )    \\
= &  \lambda  \ottmv{y}  .   \mathbf{lift}^{ m }   (    \llbracket  \Gamma  \vdash  \ottnt{f}  :   \ottnt{B}  \to  \ottnt{C}   \rrbracket   \:  \ottmv{y}   )    \\
= &  \lambda  \ottmv{x}  .   \mathbf{lift}^{ m }   (   \ottnt{f}  \:   \{ \mathbf{proj}_i^n \:  \ottmv{x}  /  \ottmv{x} _i \}    )    \\
= &   \lambda  \ottmv{x}  .   (   \mathbf{lift}^{ m }  \ottnt{f}   )    \:   \{ \mathbf{proj}_i^n \:  \ottmv{x}  /  \ottmv{x} _i \}  
\end{align*}
\item \Rref{MC-Up}. Have: $ \Gamma  \vdash   \mathbf{up}^{ m_{{\mathrm{1}}} , m_{{\mathrm{2}}} }  \ottnt{b}   :   S_{ m_{{\mathrm{2}}} } \:  \ottnt{B}  $ where $ \Gamma  \vdash  \ottnt{b}  :   S_{ m_{{\mathrm{1}}} } \:  \ottnt{B}  $ and $ m_{{\mathrm{1}}}   \leq   m_{{\mathrm{2}}} $.\\
By IH, $ \llbracket  \Gamma  \vdash  \ottnt{b}  :  \ottnt{B}  \rrbracket  =   \lambda  \ottmv{x}  .  \ottnt{b}   \:   \{ \mathbf{proj}_i^n \:  \ottmv{x}  /  \ottmv{x} _i \}  $. \\
Now, $ \llbracket  \Gamma  \vdash   \mathbf{up}^{ m_{{\mathrm{1}}} , m_{{\mathrm{2}}} }  \ottnt{b}   :   S_{ m_{{\mathrm{2}}} } \:  \ottnt{B}   \rrbracket  =   \mathbb{S}^{ m_{{\mathrm{1}}}  \leq  m_{{\mathrm{2}}} }_{ \ottnt{B} }   \circ   \llbracket  \Gamma  \vdash  \ottnt{b}  :  \ottnt{B}  \rrbracket   =   \lambda  \ottmv{x}  .   (   \mathbf{up}^{ m_{{\mathrm{1}}} , m_{{\mathrm{2}}} }  \ottnt{b}   )    \:   \{ \mathbf{proj}_i^n \:  \ottmv{x}  /  \ottmv{x} _i \}  $.
\end{itemize} 

Hence, $\llbracket  \Gamma  \vdash  \ottnt{b}  :  \ottnt{B}  \rrbracket_{( \mathds{F} ,\mathbb{S})} = \lambda x . b \{ \text{proj}_i^n \:  x / x_i \}$.\\ Now, let $ \Gamma  \vdash  \ottnt{b_{{\mathrm{1}}}}  :  \ottnt{B} $ and $ \Gamma  \vdash  \ottnt{b_{{\mathrm{2}}}}  :  \ottnt{B} $ such that $ \llbracket  \ottnt{b_{{\mathrm{1}}}}  \rrbracket  =  \llbracket  \ottnt{b_{{\mathrm{2}}}}  \rrbracket $ in all models of GMCC($ \mathcal{M} $). Then $ \llbracket  \ottnt{b_{{\mathrm{1}}}}  \rrbracket _{( \mathds{F} ,\mathbb{S})} =  \llbracket  \ottnt{b_{{\mathrm{2}}}}  \rrbracket _{( \mathds{F} ,\mathbb{S})}$. Hence, $  \lambda  \ottmv{x}  .  \ottnt{b_{{\mathrm{1}}}}   \:   \{ \mathbf{proj}_i^n \:  \ottmv{x}  /  \ottmv{x} _i \}   \equiv   \lambda  \ottmv{x}  .  \ottnt{b_{{\mathrm{2}}}}   \:   \{ \mathbf{proj}_i^n \:  \ottmv{x}  /  \ottmv{x} _i \}  $. Therefore, $\ottnt{b_{{\mathrm{1}}}} \equiv \ottnt{b_{{\mathrm{2}}}}$.\\
This completes the proof.
\end{proof}


\begin{theorem}[Theorem \ref{gmccgeneric}]
The generic model satisfies the universal property.
\end{theorem}

\begin{proof}

Let the parametrizing monoid be $ \mathcal{M} $. For GMCC($ \mathcal{M} $), let $ \mathit{Cl} $ be the classifying category and let $ \mathit{G} $ be the generic model. 

First, we show how to define models of GMCC($ \mathcal{M} $) in other categories using models of GMCC($ \mathcal{M} $) in $ \mathit{Cl} $. Let $ \mathds{D} $ be any bicartesian closed category and let $F$ be a finite-product-preserving functor from $ \mathit{Cl} $ to $ \mathds{D} $. Now, given any model $W$ of GMCC($ \mathcal{M} $) in $ \mathit{Cl} $, we can define a model $F(W)$ of GMCC($ \mathcal{M} $) in $ \mathds{D} $ as follows:
\begin{align*}
 \llbracket  \ottnt{A}  \rrbracket _{F(W)} & \triangleq F ( \llbracket  \ottnt{A}  \rrbracket _W) \\
 \llbracket  \Gamma  \vdash  \ottnt{a}  :  \ottnt{A}  \rrbracket _{F(W)} & \triangleq \Pi \llbracket A_i \rrbracket_{F(W)} (= \Pi F (\llbracket A_i \rrbracket_W)) \xrightarrow{p^{-1}} F (\Pi \llbracket A_i \rrbracket_W) \xrightarrow{F( \llbracket  \ottnt{a}  \rrbracket _W)} F ( \llbracket  \ottnt{A}  \rrbracket _W) 
\end{align*}
where $\Gamma = x_1 : A_1 , x_2 : A_2 , \ldots , x_n : A_n$ and $p$ is a witness of the finite-product-preserving property of $F$.\\ Next, we check that $F(W)$ is indeed a model of GMCC($ \mathcal{M} $).

If $ \Gamma  \vdash  \ottnt{a}  :  \ottnt{A} $ in GMCC($ \mathcal{M} $), then $ \llbracket  \Gamma  \vdash  \ottnt{a}  :  \ottnt{A}  \rrbracket _W \in \text{Hom}_{ \mathit{Cl} } ( \llbracket  \Gamma  \rrbracket _W,  \llbracket  \ottnt{A}  \rrbracket _W)$.\\ Therefore, $ \llbracket  \Gamma  \vdash  \ottnt{a}  :  \ottnt{A}  \rrbracket _{F(W)} \in \text{Hom}_{ \mathds{D} } ( \llbracket  \Gamma  \rrbracket _{F(W)} ,  \llbracket  \ottnt{A}  \rrbracket _{F(W)})$. Hence, $F(W)$ is a well-defined model of GMCC($ \mathcal{M} $). 

Now, say $ \Gamma  \vdash  \ottnt{a_{{\mathrm{1}}}}  :  \ottnt{A} $ and $ \Gamma  \vdash  \ottnt{a_{{\mathrm{2}}}}  :  \ottnt{A} $ such that $ \ottnt{a_{{\mathrm{1}}}}  \equiv  \ottnt{a_{{\mathrm{2}}}} $ in GMCC($ \mathcal{M} $). \\
Then, $ \llbracket  \Gamma  \vdash  \ottnt{a_{{\mathrm{1}}}}  :  \ottnt{A}  \rrbracket _W =  \llbracket  \Gamma  \vdash  \ottnt{a_{{\mathrm{2}}}}  :  \ottnt{A}  \rrbracket _W \in \text{Hom}_{ \mathit{Cl} } ( \llbracket  \Gamma  \rrbracket _W ,  \llbracket  \ottnt{A}  \rrbracket _W)$. \\
So, $F ( \llbracket  \Gamma  \vdash  \ottnt{a_{{\mathrm{1}}}}  :  \ottnt{A}  \rrbracket _W) \circ q = F ( \llbracket  \Gamma  \vdash  \ottnt{a_{{\mathrm{2}}}}  :  \ottnt{A}  \rrbracket _W) \circ q \in \text{Hom}_{ \mathds{D} } ( \llbracket  \Gamma  \rrbracket _{F(W)} ,  \llbracket  \ottnt{A}  \rrbracket _{F(W)} )$. \\
Or, $ \llbracket  \Gamma  \vdash  \ottnt{a_{{\mathrm{1}}}}  :  \ottnt{A}  \rrbracket _{F(W)} =  \llbracket  \Gamma  \vdash  \ottnt{a_{{\mathrm{2}}}}  :  \ottnt{A}  \rrbracket _{F(W)} \in \text{Hom}_{ \mathds{D} } ( \llbracket  \Gamma  \rrbracket _{F(W)} ,  \llbracket  \ottnt{A}  \rrbracket _{F(W)})$.\\
Hence, $F(W)$ is a sound model of GMCC($ \mathcal{M} $).

Next, we show the universal property: for any given model $W'$ of GMCC($ \mathcal{M} $) in any bicartesian closed category $ \mathds{D} $, there exists a \textit{unique} finite-product-preserving functor, $F$, from $ \mathit{Cl} $ to $ \mathds{D} $ such that $F( \mathit{G} ) = W'$.

Given $W'$, we define a functor, $F$, from $ \mathit{Cl} $ to $ \mathds{D} $, as follows:
\begin{align*}
F (A) & \triangleq  \llbracket  \ottnt{A}  \rrbracket _{W'} \\
F (A \xrightarrow{a} B) & \triangleq  \llbracket   \ottmv{x}  :  \ottnt{A}   \vdash    \ottnt{a}  \:  \ottmv{x}    :  \ottnt{B}  \rrbracket _{W'}
\end{align*} 
Since $A \xrightarrow{a} B$ in $ \mathit{Cl} $, we know that $  \emptyset   \vdash  \ottnt{a}  :   \ottnt{A}  \to  \ottnt{B}  $ in GMCC($ \mathcal{M} $).\\
Therefore, $  \ottmv{x}  :  \ottnt{A}   \vdash   \ottnt{a}  \:  \ottmv{x}   :  \ottnt{B} $ in GMCC($ \mathcal{M} $). So, $ \llbracket   \ottmv{x}  :  \ottnt{A}   \vdash    \ottnt{a}  \:  \ottmv{x}    :  \ottnt{B}  \rrbracket _{W'} \in \text{Hom}_{ \mathds{D} } ( \llbracket  \ottnt{A}  \rrbracket _{W'} ,  \llbracket  \ottnt{B}  \rrbracket _{W'} )$.\\
Here, we need to check that the definition of $F$ respects the equivalence relation on morphisms in $ \mathit{Cl} $. In other words, we need to check that if $  \emptyset   \vdash  \ottnt{a_{{\mathrm{1}}}}  :   \ottnt{A}  \to  \ottnt{B}  $ and $  \emptyset   \vdash  \ottnt{a_{{\mathrm{2}}}}  :   \ottnt{A}  \to  \ottnt{B}  $ such that $ \ottnt{a_{{\mathrm{1}}}}  \equiv  \ottnt{a_{{\mathrm{2}}}} $ in GMCC($ \mathcal{M} $), then $ \llbracket   \ottmv{x}  :  \ottnt{A}   \vdash    \ottnt{a_{{\mathrm{1}}}}  \:  \ottmv{x}    :  \ottnt{B}  \rrbracket _{W'} =  \llbracket   \ottmv{x}  :  \ottnt{A}   \vdash    \ottnt{a_{{\mathrm{2}}}}  \:  \ottmv{x}    :  \ottnt{B}  \rrbracket _{W'}$. But this is true because $W'$ is a sound model of GMCC($ \mathcal{M} $).\\
So $F$ is well-defined.

Next, we check that $F$ is indeed a functor.\\
\[ F (A \xrightarrow{\text{id}} A) =  \llbracket   \ottmv{x}  :  \ottnt{A}   \vdash  \ottmv{x}  :  \ottnt{A}  \rrbracket _{W'} =  \llbracket  \ottnt{A}  \rrbracket _{W'} \xrightarrow{\text{id}}  \llbracket  \ottnt{A}  \rrbracket _{W'} \]
\begin{align*}
& F (A \xrightarrow{f} B \xrightarrow{g} C) \\
= &  \llbracket   \ottmv{x}  :  \ottnt{A}   \vdash     \ottnt{g}  \:  \ottnt{f}   \:  \ottmv{x}    :  \ottnt{C}  \rrbracket _{W'} \\
= &  \text{app}  \circ  \langle  \llbracket   \ottmv{x}  :  \ottnt{A}   \vdash  \ottnt{g}  :   \ottnt{B}  \to  \ottnt{C}   \rrbracket _{W'} ,  \llbracket   \ottmv{x}  :  \ottnt{A}   \vdash    \ottnt{f}  \:  \ottmv{x}    :  \ottnt{B}  \rrbracket _{W'} \rangle \\
= &  \text{app}  \circ  \langle  \llbracket   \ottmv{x}  :  \ottnt{A}   \vdash  \ottnt{g}  :   \ottnt{B}  \to  \ottnt{C}   \rrbracket _{W'} ,  \text{app}  \circ \langle  \llbracket   \ottmv{x}  :  \ottnt{A}   \vdash   \ottnt{f}   :   \ottnt{A}  \to  \ottnt{B}   \rrbracket _{W'} ,  \llbracket   \ottmv{x}  :  \ottnt{A}   \vdash  \ottmv{x}  :  \ottnt{A}  \rrbracket _{W'} \rangle \rangle \\
= &  \text{app}  \circ  \langle  \llbracket   \emptyset   \vdash  \ottnt{g}  :   \ottnt{B}  \to  \ottnt{C}   \rrbracket _{W'} \circ \langle \rangle ,  \text{app}  \circ \langle  \llbracket   \emptyset   \vdash   \ottnt{f}   :   \ottnt{A}  \to  \ottnt{B}   \rrbracket _{W'} \circ \langle \rangle ,  \text{id}  \rangle \rangle \\
= &  \text{app}  \circ \langle  \llbracket   \emptyset   \vdash  \ottnt{g}  :   \ottnt{B}  \to  \ottnt{C}   \rrbracket _{W'} \circ \langle \rangle ,  \text{id}  \rangle \circ  \text{app}  \circ \langle  \llbracket   \emptyset   \vdash  \ottnt{f}  :   \ottnt{A}  \to  \ottnt{B}   \rrbracket _{W'} \circ \langle \rangle ,  \text{id}  \rangle \\
= &  \llbracket   \ottmv{y}  :  \ottnt{B}   \vdash    \ottnt{g}  \:  \ottmv{y}    :  \ottnt{C}  \rrbracket _{W'} \circ  \llbracket   \ottmv{x}  :  \ottnt{A}   \vdash    \ottnt{f}  \:  \ottmv{x}    :  \ottnt{B}  \rrbracket _{W'} \\
= & F (B \xrightarrow{g} C) \circ F (A \xrightarrow{f} B)
\end{align*}

Observe that $F(\ottnt{A_{{\mathrm{1}}}} \times \ottnt{A_{{\mathrm{2}}}}) =  \llbracket   \ottnt{A_{{\mathrm{1}}}}  \times  \ottnt{A_{{\mathrm{2}}}}   \rrbracket _{W'} =  \llbracket  \ottnt{A_{{\mathrm{1}}}}  \rrbracket _{W'} \times  \llbracket  \ottnt{A_{{\mathrm{2}}}}  \rrbracket _{W'} = F \ottnt{A_{{\mathrm{1}}}} \times F \ottnt{A_{{\mathrm{2}}}}$.\\
Hence, $F$ is a finite-product-preserving functor. 

Now, we show that $F( \mathit{G} ) = W'$.\\
$ \llbracket  \ottnt{A}  \rrbracket _{F( \mathit{G} )} = F  ( \llbracket  \ottnt{A}  \rrbracket _{ \mathit{G} }) = F (A) =  \llbracket  \ottnt{A}  \rrbracket _{W'}.$\\
Further, \begin{align*}
 &  \llbracket  \Gamma  \vdash  \ottnt{a}  :  \ottnt{A}  \rrbracket _{F( \mathit{G} )} \\
= & F ( \llbracket  \Gamma  \vdash  \ottnt{a}  :  \ottnt{A}  \rrbracket _{ \mathit{G} }) \\ 
= & F (\lambda x . a \{ \text{proj}_i^n \: x / x_i \}) \hspace*{6pt} [ \text{Here}, \Gamma = x_1 : A_1 , x_2 : A_2, \ldots, x_n : A_n ]\\
= & \llbracket x : \Pi A_i \vdash a \{ \text{proj}_i^n \: x / x_i \} : A \rrbracket_{W'} \\
= &  \llbracket  \Gamma  \vdash  \ottnt{a}  :  \ottnt{A}  \rrbracket _{W'}.
\end{align*}
Therefore, $F( \mathit{G} ) = W'$.

Next, we show the uniqueness property. Let $F'$ be a finite-product-preserving functor from $ \mathit{Cl} $ to $ \mathds{D} $ such that $F'( \mathit{G} ) = W'$. Need to show that $F' = F$.

Now, $F'(A) = F' ( \llbracket  \ottnt{A}  \rrbracket _{ \mathit{G} }) =  \llbracket  \ottnt{A}  \rrbracket _{W'} = F (A)$.\\
Also, $F'(A \xrightarrow{a} B) = F' ( \llbracket   \ottmv{x}  :  \ottnt{A}   \vdash    \ottnt{a}  \:  \ottmv{x}    :  \ottnt{B}  \rrbracket _{ \mathit{G} }) =  \llbracket   \ottmv{x}  :  \ottnt{A}   \vdash    \ottnt{a}  \:  \ottmv{x}    :  \ottnt{B}  \rrbracket _{W'} = F(A \xrightarrow{a} B)$.\\
Therefore, $F' = F$.\\
This completes the proof.
\end{proof}


\section{Proofs of lemmas/theorems stated in Section \ref{secgmccdcce}}

\begin{theorem}[Theorem \ref{DCCComplete}]
If $ \Gamma  \vdash  \ottnt{a}  :  \ottnt{A} $ in GMCC($ \mathcal{L} $), then $  \overline{  \Gamma  }   \vdash   \overline{ \ottnt{a} }   :   \overline{ \ottnt{A} }  $ in \ED{}($ \mathcal{L} $). Further, if $ \Gamma  \vdash  \ottnt{a_{{\mathrm{1}}}}  :  \ottnt{A} $ and $ \Gamma  \vdash  \ottnt{a_{{\mathrm{2}}}}  :  \ottnt{A} $ such that $\ottnt{a_{{\mathrm{1}}}} \equiv \ottnt{a_{{\mathrm{2}}}}$ in GMCC($ \mathcal{L} $), then $ \overline{ \ottnt{a_{{\mathrm{1}}}} }  \simeq  \overline{ \ottnt{a_{{\mathrm{2}}}} } $ in \ED{}($ \mathcal{L} $).
\end{theorem}

\begin{proof}
Follows from Theorems \ref{GMC2DCC} and \ref{GCC2DCCe}.
\end{proof}


\begin{lemma}[Lemma \ref{protectId}]\label{ProtT}
If $ \ell  \sqsubseteq  \ottnt{B} $, then there exists a term $  \emptyset   \vdash   j^{ \ell }_{  \underline{ \ottnt{B} }  }   :    S_{  \ell  } \:    \underline{ \ottnt{B} }     \to   \underline{ \ottnt{B} }   $ such that $ \lfloor    j^{ \ell }_{  \underline{ \ottnt{B} }  }    \rfloor  \equiv  \lambda  \ottmv{x}  .  \ottmv{x} $.
\end{lemma}

\begin{proof}
By induction on $ \ell  \sqsubseteq  \ottnt{B} $.\\
\begin{itemize}

\item \Rref{Prot-Prod}. Have: $ \ell  \sqsubseteq   \ottnt{A}  \times  \ottnt{B}  $ where $ \ell  \sqsubseteq  \ottnt{A} $ and $ \ell  \sqsubseteq  \ottnt{B} $.\\
By IH, $\exists$ $  \emptyset   \vdash   j^{ \ell }_{  \underline{ \ottnt{A} }  }   :    S_{  \ell  } \:   \underline{ \ottnt{A} }    \to   \underline{ \ottnt{A} }   $ and $  \emptyset   \vdash   j^{ \ell }_{  \underline{ \ottnt{B} }  }   :    S_{  \ell  } \:   \underline{ \ottnt{B} }    \to   \underline{ \ottnt{B} }   $ such that $ \lfloor    j^{ \ell }_{  \underline{ \ottnt{A} }  }    \rfloor  \equiv  \lambda  \ottmv{x}  .  \ottmv{x} $ and $ \lfloor    j^{ \ell }_{  \underline{ \ottnt{B} }  }    \rfloor  \equiv  \lambda  \ottmv{x}  .  \ottmv{x} $.\\
Now,
$$\infer[\textsc{(Lam)}]{  \emptyset   \vdash   \lambda  \ottmv{z}  .   (     j^{ \ell }_{  \underline{ \ottnt{A} }  }    \:   (    (   \mathbf{lift}^{  \ell  }   (   \lambda  \ottmv{y}  .   \mathbf{proj}_1 \:  \ottmv{y}    )    )   \:  \ottmv{z}   )    ,     j^{ \ell }_{  \underline{ \ottnt{B} }  }    \:   (    (   \mathbf{lift}^{  \ell  }   (   \lambda  \ottmv{y}  .   \mathbf{proj}_2 \:  \ottmv{y}    )    )   \:  \ottmv{z}   )    )    :     S_{  \ell  } \:   (    \underline{ \ottnt{A} }   \times   \underline{ \ottnt{B} }    )    \to   \underline{ \ottnt{A} }    \times   \underline{ \ottnt{B} }   }
           {\infer[\textsc{(Pair)}]{  \ottmv{z}  :   S_{  \ell  } \:   (    \underline{ \ottnt{A} }   \times   \underline{ \ottnt{B} }    )     \vdash   (     j^{ \ell }_{  \underline{ \ottnt{A} }  }    \:   (    (   \mathbf{lift}^{  \ell  }   (   \lambda  \ottmv{y}  .   \mathbf{proj}_1 \:  \ottmv{y}    )    )   \:  \ottmv{z}   )    ,     j^{ \ell }_{  \underline{ \ottnt{B} }  }    \:   (    (   \mathbf{lift}^{  \ell  }   (   \lambda  \ottmv{y}  .   \mathbf{proj}_2 \:  \ottmv{y}    )    )   \:  \ottmv{z}   )    )   :    \underline{ \ottnt{A} }   \times   \underline{ \ottnt{B} }   }
                 {\infer[\textsc{(IH)}]{  \ottmv{z}  :   S_{  \ell  } \:   (    \underline{ \ottnt{A} }   \times   \underline{ \ottnt{B} }    )     \vdash     j^{ \ell }_{  \underline{ \ottnt{A} }  }    \:   (    (   \mathbf{lift}^{  \ell  }   (   \lambda  \ottmv{y}  .   \mathbf{proj}_1 \:  \ottmv{y}    )    )   \:  \ottmv{z}   )    :   \underline{ \ottnt{A} }  }
                    {\infer[\textsc{(App)}]{  \ottmv{z}  :   S_{  \ell  } \:   (    \underline{ \ottnt{A} }   \times   \underline{ \ottnt{B} }    )     \vdash    (   \mathbf{lift}^{  \ell  }   (   \lambda  \ottmv{y}  .   \mathbf{proj}_1 \:  \ottmv{y}    )    )   \:  \ottmv{z}   :   S_{  \ell  } \:   \underline{ \ottnt{A} }   }
                       {\infer[\textsc{(Fmap)}]{  \emptyset   \vdash   \mathbf{lift}^{  \ell  }   (   \lambda  \ottmv{y}  .   \mathbf{proj}_1 \:  \ottmv{y}    )    :    S_{  \ell  } \:   (    \underline{ \ottnt{A} }   \times   \underline{ \ottnt{B} }    )    \to   S_{  \ell  } \:   \underline{ \ottnt{A} }    }
                         {\infer[]{  \emptyset   \vdash   \lambda  \ottmv{y}  .   \mathbf{proj}_1 \:  \ottmv{y}    :     \underline{ \ottnt{A} }   \times   \underline{ \ottnt{B} }    \to   \underline{ \ottnt{A} }   }{}}
                       }
                    }
                 }  
              }
$$

Note that we omit the derivation of $  \ottmv{z}  :   S_{  \ell  } \:   (    \underline{ \ottnt{A} }   \times   \underline{ \ottnt{B} }    )     \vdash     j^{ \ell }_{  \underline{ \ottnt{B} }  }    \:   (    (   \mathbf{lift}^{  \ell  }   (   \lambda  \ottmv{y}  .   \mathbf{proj}_2 \:  \ottmv{y}    )    )   \:  \ottmv{z}   )    :   \underline{ \ottnt{B} }  $, which is similar to that of $  \ottmv{z}  :   S_{  \ell  } \:   (    \underline{ \ottnt{A} }   \times   \underline{ \ottnt{B} }    )     \vdash     j^{ \ell }_{  \underline{ \ottnt{A} }  }    \:   (    (   \mathbf{lift}^{  \ell  }   (   \lambda  \ottmv{y}  .   \mathbf{proj}_1 \:  \ottmv{y}    )    )   \:  \ottmv{z}   )    :   \underline{ \ottnt{A} }  $.\\

Therefore, $ j^{ \ell }_{    \underline{ \ottnt{A} }   \times   \underline{ \ottnt{B} }    }  =  \lambda  \ottmv{z}  .   (     j^{ \ell }_{  \underline{ \ottnt{A} }  }    \:   (    (   \mathbf{lift}^{  \ell  }   (   \lambda  \ottmv{y}  .   \mathbf{proj}_1 \:  \ottmv{y}    )    )   \:  \ottmv{z}   )    ,     j^{ \ell }_{  \underline{ \ottnt{B} }  }    \:   (    (   \mathbf{lift}^{  \ell  }   (   \lambda  \ottmv{y}  .   \mathbf{proj}_2 \:  \ottmv{y}    )    )   \:  \ottmv{z}   )    )  $ and $ \lfloor    j^{ \ell }_{    \underline{ \ottnt{A} }   \times   \underline{ \ottnt{B} }    }    \rfloor  \equiv  \lambda  \ottmv{x}  .  \ottmv{x} $.

\item \Rref{Prot-Fun}. Have: $ \ell  \sqsubseteq   \ottnt{A}  \to  \ottnt{B}  $ where $ \ell  \sqsubseteq  \ottnt{B} $.\\
By IH, $\exists$ $  \emptyset   \vdash   j^{ \ell }_{  \underline{ \ottnt{B} }  }   :    S_{  \ell  } \:   \underline{ \ottnt{B} }    \to   \underline{ \ottnt{B} }   $ such that $ \lfloor    j^{ \ell }_{  \underline{ \ottnt{B} }  }    \rfloor  \equiv  \lambda  \ottmv{x}  .  \ottmv{x} $.\\
Now,
$$
\infer[\textsc{(Lam)}]{  \emptyset   \vdash    \lambda  \ottmv{z}  .   \lambda  \ottmv{y}  .   j^{ \ell }_{  \underline{ \ottnt{B} }  }     \:   (    (   \mathbf{lift}^{  \ell  }   (    \lambda  \ottmv{x}  .  \ottmv{x}   \:  \ottmv{y}   )    )   \:  \ottmv{z}   )    :    S_{  \ell  } \:   (    \underline{ \ottnt{A} }   \to   \underline{ \ottnt{B} }    )    \to   (    \underline{ \ottnt{A} }   \to   \underline{ \ottnt{B} }    )   }
      {\infer[\textsc{(IH)}]{   \ottmv{z}  :   S_{  \ell  } \:   (    \underline{ \ottnt{A} }   \to   \underline{ \ottnt{B} }    )     ,   \ottmv{y}  :   \underline{ \ottnt{A} }     \vdash    j^{ \ell }_{  \underline{ \ottnt{B} }  }   \:   (    (   \mathbf{lift}^{  \ell  }   (    \lambda  \ottmv{x}  .  \ottmv{x}   \:  \ottmv{y}   )    )   \:  \ottmv{z}   )    :   \underline{ \ottnt{B} }  }
        {\infer[\textsc{(App)}]{   \ottmv{z}  :   S_{  \ell  } \:   (    \underline{ \ottnt{A} }   \to   \underline{ \ottnt{B} }    )     ,   \ottmv{y}  :   \underline{ \ottnt{A} }     \vdash    (   \mathbf{lift}^{  \ell  }   (    \lambda  \ottmv{x}  .  \ottmv{x}   \:  \ottmv{y}   )    )   \:  \ottmv{z}   :   S_{  \ell  } \:   \underline{ \ottnt{B} }   }
           {\infer[\textsc{(Fmap)}]{  \ottmv{y}  :   \underline{ \ottnt{A} }    \vdash   \mathbf{lift}^{  \ell  }   (    \lambda  \ottmv{x}  .  \ottmv{x}   \:  \ottmv{y}   )    :    S_{  \ell  } \:   (    \underline{ \ottnt{A} }   \to   \underline{ \ottnt{B} }    )    \to   S_{  \ell  } \:   \underline{ \ottnt{B} }    }
              {\infer[]{  \ottmv{y}  :   \underline{ \ottnt{A} }    \vdash    \lambda  \ottmv{x}  .  \ottmv{x}   \:  \ottmv{y}   :    (    \underline{ \ottnt{A} }   \to   \underline{ \ottnt{B} }    )   \to   \underline{ \ottnt{B} }   }{}
              }
           }
        }
      }
$$

Therefore, $ j^{ \ell }_{    \underline{ \ottnt{A} }   \to   \underline{ \ottnt{B} }    }  =   \lambda  \ottmv{z}  .   \lambda  \ottmv{y}  .   j^{ \ell }_{  \underline{ \ottnt{B} }  }     \:   (    (   \mathbf{lift}^{  \ell  }   (    \lambda  \ottmv{x}  .  \ottmv{x}   \:  \ottmv{y}   )    )   \:  \ottmv{z}   )  $ and $ \lfloor    j^{ \ell }_{    \underline{ \ottnt{A} }   \to   \underline{ \ottnt{B} }    }    \rfloor  \equiv  \lambda  \ottmv{x}  .  \ottmv{x} $.

\item \Rref{Prot-Monad}. Have: $ \ell_{{\mathrm{1}}}  \sqsubseteq   T_{  \ell_{{\mathrm{2}}}  } \:  \ottnt{A}  $ where $ \ell_{{\mathrm{1}}}  \sqsubseteq  \ell_{{\mathrm{2}}} $.\\
Now, $  \emptyset   \vdash   \lambda  \ottmv{x}  .   \mathbf{join}^{  \ell_{{\mathrm{1}}}  ,  \ell_{{\mathrm{2}}}  }  \ottmv{x}    :    S_{  \ell_{{\mathrm{1}}}  } \:   S_{  \ell_{{\mathrm{2}}}  } \:   \underline{ \ottnt{B} }     \to   S_{  \ell_{{\mathrm{2}}}  } \:   \underline{ \ottnt{B} }    $.\\
Therefore, $ j^{ \ell_{{\mathrm{1}}} }_{   S_{  \ell_{{\mathrm{2}}}  } \:   \underline{ \ottnt{B} }    }  =  \lambda  \ottmv{x}  .   \mathbf{join}^{  \ell_{{\mathrm{1}}}  ,  \ell_{{\mathrm{2}}}  }  \ottmv{x}  $ and $  \lfloor    j^{ \ell_{{\mathrm{1}}} }_{   S_{  \ell_{{\mathrm{2}}}  } \:   \underline{ \ottnt{B} }    }    \rfloor  \equiv  \lambda  \ottmv{x}  .  \ottmv{x} $.

\item \Rref{Prot-Already}. Have: $ \ell  \sqsubseteq   T_{  \ell'  } \:  \ottnt{A}  $ where $ \ell  \sqsubseteq  \ottnt{A} $.\\
By IH, $\exists$ $  \emptyset   \vdash   j^{ \ell }_{  \underline{ \ottnt{A} }  }   :    S_{  \ell  } \:   \underline{ \ottnt{A} }    \to   \underline{ \ottnt{A} }   $ such that $ \lfloor    j^{ \ell }_{  \underline{ \ottnt{A} }  }    \rfloor  \equiv  \lambda  \ottmv{x}  .  \ottmv{x} $.\\
Now,
$$
\infer[\textsc{(Lam)}]{  \emptyset   \vdash    \lambda  \ottmv{x}  .   (   \mathbf{lift}^{  \ell'  }   j^{ \ell }_{  \underline{ \ottnt{A} }  }    )    \:   (   \mathbf{fork}^{  \ell'  ,  \ell  }   (   \mathbf{join}^{  \ell  ,  \ell'  }  \ottmv{x}   )    )    :    S_{  \ell  } \:   S_{  \ell'  } \:   \underline{ \ottnt{A} }     \to   S_{  \ell'  } \:   \underline{ \ottnt{A} }    }
  {\infer[\textsc{(App)}]{  \ottmv{x}  :   S_{  \ell  } \:   S_{  \ell'  } \:   \underline{ \ottnt{A} }      \vdash    (   \mathbf{lift}^{  \ell'  }   j^{ \ell }_{  \underline{ \ottnt{A} }  }    )   \:   (   \mathbf{fork}^{  \ell'  ,  \ell  }   (   \mathbf{join}^{  \ell  ,  \ell'  }  \ottmv{x}   )    )    :   S_{  \ell'  } \:   \underline{ \ottnt{A} }   }
     {\infer[\textsc{(Fmap)}]{  \emptyset   \vdash   \mathbf{lift}^{  \ell'  }   j^{ \ell }_{  \underline{ \ottnt{A} }  }    :    S_{  \ell'  } \:   S_{  \ell  } \:   \underline{ \ottnt{A} }     \to   S_{  \ell'  } \:   \underline{ \ottnt{A} }    }
      {\infer[\textsc{(IH)}]{  \emptyset   \vdash   j^{ \ell }_{  \underline{ \ottnt{A} }  }   :    S_{  \ell  } \:   \underline{ \ottnt{A} }    \to   \underline{ \ottnt{A} }   }{}}
        &
      \infer[\textsc{(Fork)}]{  \ottmv{x}  :   S_{  \ell  } \:   S_{  \ell'  } \:   \underline{ \ottnt{A} }      \vdash   \mathbf{fork}^{  \ell'  ,  \ell  }   (   \mathbf{join}^{  \ell  ,  \ell'  }  \ottmv{x}   )    :   S_{  \ell'  } \:   S_{  \ell  } \:   \underline{ \ottnt{A} }    }
       {\infer[\textsc{(Join)}]{  \ottmv{x}  :   S_{  \ell  } \:   S_{  \ell'  } \:   \underline{ \ottnt{A} }      \vdash   \mathbf{join}^{  \ell  ,  \ell'  }  \ottmv{x}   :   S_{   \ell  \vee  \ell'   } \:   \underline{ \ottnt{A} }   }{}}
     }
  }
$$
Notice the flip going from $\mathbf{join}$ to $\mathbf{fork}$ in the above derivation!\\

Therefore, $ j^{ \ell }_{   S_{  \ell'  } \:   \underline{ \ottnt{A} }    }  =   \lambda  \ottmv{x}  .   (   \mathbf{lift}^{  \ell'  }   j^{ \ell }_{  \underline{ \ottnt{A} }  }    )    \:   (   \mathbf{fork}^{  \ell'  ,  \ell  }   (   \mathbf{join}^{  \ell  ,  \ell'  }  \ottmv{x}   )    )  $ and $ \lfloor    j^{ \ell }_{   S_{  \ell'  } \:   \underline{ \ottnt{A} }    }    \rfloor  \equiv  \lambda  \ottmv{x}  .  \ottmv{x} $. 

\item \Rref{Prot-Minimum}. Have: $  \bot   \sqsubseteq  \ottnt{A} $.\\
Now, $  \emptyset   \vdash   \lambda  \ottmv{x}  .   \mathbf{extr} \:  \ottmv{x}    :    S_{   \bot   } \:   \underline{ \ottnt{A} }    \to   \underline{ \ottnt{A} }   $.\\
Therefore, $ j^{  \bot  }_{  \underline{ \ottnt{A} }  }  =  \lambda  \ottmv{x}  .   \mathbf{extr} \:  \ottmv{x}  $ and $ \lfloor    j^{  \bot  }_{  \underline{ \ottnt{A} }  }    \rfloor  \equiv  \lambda  \ottmv{x}  .  \ottmv{x} $.

\item \Rref{Prot-Combine}. Have: $ \ell  \sqsubseteq  \ottnt{A} $ where $ \ell_{{\mathrm{1}}}  \sqsubseteq  \ottnt{A} $ and $ \ell_{{\mathrm{2}}}  \sqsubseteq  \ottnt{A} $ and $ \ell  \sqsubseteq   \ell_{{\mathrm{1}}}  \vee  \ell_{{\mathrm{2}}}  $.\\
By IH, $\exists$ $  \emptyset   \vdash   j^{ \ell_{{\mathrm{1}}} }_{  \underline{ \ottnt{A} }  }   :    S_{  \ell_{{\mathrm{1}}}  } \:   \underline{ \ottnt{A} }    \to   \underline{ \ottnt{A} }   $ and $  \emptyset   \vdash   j^{ \ell_{{\mathrm{2}}} }_{  \underline{ \ottnt{A} }  }   :    S_{  \ell_{{\mathrm{2}}}  } \:   \underline{ \ottnt{A} }    \to   \underline{ \ottnt{A} }   $ such that $ \lfloor    j^{ \ell_{{\mathrm{1}}} }_{  \underline{ \ottnt{A} }  }    \rfloor  \equiv  \lambda  \ottmv{x}  .  \ottmv{x} $ and $ \lfloor    j^{ \ell_{{\mathrm{2}}} }_{  \underline{ \ottnt{A} }  }    \rfloor  \equiv  \lambda  \ottmv{x}  .  \ottmv{x} $.\\
Now,
$$
\infer[\textsc{(Lam)}]{  \emptyset   \vdash    \lambda  \ottmv{x}  .   j^{ \ell_{{\mathrm{1}}} }_{  \underline{ \ottnt{A} }  }    \:   (    (   \mathbf{lift}^{  \ell_{{\mathrm{1}}}  }   j^{ \ell_{{\mathrm{2}}} }_{  \underline{ \ottnt{A} }  }    )   \:   (   \mathbf{fork}^{  \ell_{{\mathrm{1}}}  ,  \ell_{{\mathrm{2}}}  }   (   \mathbf{up}^{  \ell  ,   \ell_{{\mathrm{1}}}  \vee  \ell_{{\mathrm{2}}}   }  \ottmv{x}   )    )    )    :    S_{  \ell  } \:   \underline{ \ottnt{A} }    \to   \underline{ \ottnt{A} }   }
 {\infer[\textsc{(IH)}]{  \ottmv{x}  :   S_{  \ell  } \:   \underline{ \ottnt{A} }     \vdash    j^{ \ell_{{\mathrm{1}}} }_{  \underline{ \ottnt{A} }  }   \:   (    (   \mathbf{lift}^{  \ell_{{\mathrm{1}}}  }   j^{ \ell_{{\mathrm{2}}} }_{  \underline{ \ottnt{A} }  }    )   \:   (   \mathbf{fork}^{  \ell_{{\mathrm{1}}}  ,  \ell_{{\mathrm{2}}}  }   (   \mathbf{up}^{  \ell  ,   \ell_{{\mathrm{1}}}  \vee  \ell_{{\mathrm{2}}}   }  \ottmv{x}   )    )    )    :   \underline{ \ottnt{A} }  }
   {\infer[\textsc{(App)}]{  \ottmv{x}  :   S_{  \ell  } \:   \underline{ \ottnt{A} }     \vdash    (   \mathbf{lift}^{  \ell_{{\mathrm{1}}}  }   j^{ \ell_{{\mathrm{2}}} }_{  \underline{ \ottnt{A} }  }    )   \:   (   \mathbf{fork}^{  \ell_{{\mathrm{1}}}  ,  \ell_{{\mathrm{2}}}  }   (   \mathbf{up}^{  \ell  ,   \ell_{{\mathrm{1}}}  \vee  \ell_{{\mathrm{2}}}   }  \ottmv{x}   )    )    :   S_{  \ell_{{\mathrm{1}}}  } \:   \underline{ \ottnt{A} }   }
     {\infer[\textsc{(Fmap)}]{  \emptyset   \vdash   \mathbf{lift}^{  \ell_{{\mathrm{1}}}  }   j^{ \ell_{{\mathrm{2}}} }_{  \underline{ \ottnt{A} }  }    :    S_{  \ell_{{\mathrm{1}}}  } \:   S_{  \ell_{{\mathrm{2}}}  } \:   \underline{ \ottnt{A} }     \to   S_{  \ell_{{\mathrm{1}}}  } \:   \underline{ \ottnt{A} }    }
       {\infer[\textsc{(IH)}]{  \emptyset   \vdash   j^{ \ell_{{\mathrm{2}}} }_{  \underline{ \ottnt{A} }  }   :    S_{  \ell_{{\mathrm{2}}}  } \:   \underline{ \ottnt{A} }    \to   \underline{ \ottnt{A} }   }{}}
       &
      \infer[\textsc{(Fork)}]{  \ottmv{x}  :   S_{  \ell  } \:   \underline{ \ottnt{A} }     \vdash   \mathbf{fork}^{  \ell_{{\mathrm{1}}}  ,  \ell_{{\mathrm{2}}}  }   (   \mathbf{up}^{  \ell  ,   \ell_{{\mathrm{1}}}  \vee  \ell_{{\mathrm{2}}}   }  \ottmv{x}   )    :   S_{  \ell_{{\mathrm{1}}}  } \:   S_{  \ell_{{\mathrm{2}}}  } \:   \underline{ \ottnt{A} }    }
        {\infer[\textsc{(Up)}]{  \ottmv{x}  :   S_{  \ell  } \:   \underline{ \ottnt{A} }     \vdash   \mathbf{up}^{  \ell  ,   \ell_{{\mathrm{1}}}  \vee  \ell_{{\mathrm{2}}}   }  \ottmv{x}   :   S_{   \ell_{{\mathrm{1}}}  \vee  \ell_{{\mathrm{2}}}   } \:   \underline{ \ottnt{A} }   }{}}
     }
   }
 }
$$

Therefore, $ j^{ \ell }_{  \underline{ \ottnt{A} }  }  =   \lambda  \ottmv{x}  .   j^{ \ell_{{\mathrm{1}}} }_{  \underline{ \ottnt{A} }  }    \:   (    (   \mathbf{lift}^{  \ell_{{\mathrm{1}}}  }   j^{ \ell_{{\mathrm{2}}} }_{  \underline{ \ottnt{A} }  }    )   \:   (   \mathbf{fork}^{  \ell_{{\mathrm{1}}}  ,  \ell_{{\mathrm{2}}}  }   (   \mathbf{up}^{  \ell  ,   \ell_{{\mathrm{1}}}  \vee  \ell_{{\mathrm{2}}}   }  \ottmv{x}   )    )    )  $ and $ \lfloor    j^{ \ell }_{  \underline{ \ottnt{A} }  }    \rfloor  \equiv  \lambda  \ottmv{x}  .  \ottmv{x} $.

\end{itemize}
\end{proof}


\begin{theorem}[Theorem \ref{DCCSound}]
If $ \Gamma  \vdash  \ottnt{a}  :  \ottnt{A} $ in \ED{}($ \mathcal{L} $), then $  \underline{  \Gamma  }   \vdash   \underline{ \ottnt{a} }   :   \underline{ \ottnt{A} }  $ in GMCC($ \mathcal{L} $). Further, if $ \Gamma  \vdash  \ottnt{a_{{\mathrm{1}}}}  :  \ottnt{A} $ and $ \Gamma  \vdash  \ottnt{a_{{\mathrm{2}}}}  :  \ottnt{A} $ such that $\ottnt{a_{{\mathrm{1}}}} \simeq \ottnt{a_{{\mathrm{2}}}}$ in \ED{}($ \mathcal{L} $), then $ \underline{ \ottnt{a_{{\mathrm{1}}}} }  \simeq  \underline{ \ottnt{a_{{\mathrm{2}}}} } $ in GMCC($ \mathcal{L} $).
\end{theorem}

\begin{proof}

By induction on $ \Gamma  \vdash  \ottnt{a}  :  \ottnt{A} $.

\begin{itemize}
\item $\lambda$-calculus. By IH.

\item \Rref{DCC-Eta}. Have: $ \Gamma  \vdash   \mathbf{eta} ^{ \ell }  \ottnt{a}   :   \mathcal{T}_{ \ell } \:  \ottnt{A}  $ where $ \Gamma  \vdash  \ottnt{a}  :  \ottnt{A} $.\\
By IH, $  \underline{  \Gamma  }   \vdash   \underline{ \ottnt{a} }   :   \underline{ \ottnt{A} }  $. \\ Next, $ \underline{   \mathbf{eta} ^{ \ell }  \ottnt{a}   }  =  \mathbf{up}^{   \bot   ,  \ell  }   (   \ottkw{ret}  \:   \underline{ \ottnt{a} }    )  $. \\
Now, $  \underline{  \Gamma  }   \vdash   \ottkw{ret}  \:   \underline{ \ottnt{a} }    :   S_{   \bot   } \:   \underline{ \ottnt{A} }   $.\\ So, $  \underline{  \Gamma  }   \vdash   \mathbf{up}^{   \bot   ,  \ell  }   (   \ottkw{ret}  \:   \underline{ \ottnt{a} }    )    :   S_{  \ell  } \:   \underline{ \ottnt{A} }   $.\\

\item \Rref{DCC-Bind}. Have: $ \Gamma  \vdash   \mathbf{bind} ^{ \ell } \:  \ottmv{x}  =  \ottnt{a}  \: \mathbf{in} \:  \ottnt{b}   :  \ottnt{B} $ where $ \Gamma  \vdash  \ottnt{a}  :   \mathcal{T}_{ \ell } \:  \ottnt{A}  $ and $  \Gamma  ,   \ottmv{x}  :  \ottnt{A}    \vdash  \ottnt{b}  :  \ottnt{B} $ and $ \ell  \sqsubseteq  \ottnt{B} $.\\
By IH, $  \underline{  \Gamma  }   \vdash   \underline{ \ottnt{a} }   :   S_{  \ell  } \:   \underline{ \ottnt{A} }   $ and $   \underline{  \Gamma  }   ,   \ottmv{x}  :   \underline{ \ottnt{A} }     \vdash   \underline{ \ottnt{b} }   :   \underline{ \ottnt{B} }  $.\\
Now, since $ \ell  \sqsubseteq  \ottnt{B} $, by Lemma \ref{ProtT}, $\exists$ $  \emptyset   \vdash   j^{ \ell }_{  \underline{ \ottnt{B} }  }   :    S_{  \ell  } \:   \underline{ \ottnt{B} }    \to   \underline{ \ottnt{B} }   $ such that $ \lfloor    j^{ \ell }_{  \underline{ \ottnt{B} }  }    \rfloor  \equiv  \lambda  \ottmv{x}  .  \ottmv{x} $.\\
Next, $  \underline{  \Gamma  }   \vdash   \lambda  \ottmv{x}  .   \underline{ \ottnt{b} }    :    \underline{ \ottnt{A} }   \to   \underline{ \ottnt{B} }   $.\\ So, $  \underline{  \Gamma  }   \vdash   \mathbf{lift}^{  \ell  }   (   \lambda  \ottmv{x}  .   \underline{ \ottnt{b} }    )    :    S_{  \ell  } \:   \underline{ \ottnt{A} }    \to   S_{  \ell  } \:   \underline{ \ottnt{B} }    $.\\
As such, $  \underline{  \Gamma  }   \vdash    (   \mathbf{lift}^{  \ell  }   (   \lambda  \ottmv{x}  .   \underline{ \ottnt{b} }    )    )   \:   \underline{ \ottnt{a} }    :   S_{  \ell  } \:   \underline{ \ottnt{B} }   $.\\ Then, $  \underline{  \Gamma  }   \vdash    j^{ \ell }_{  \underline{ \ottnt{B} }  }   \:   (    (   \mathbf{lift}^{  \ell  }   (   \lambda  \ottmv{x}  .   \underline{ \ottnt{b} }    )    )   \:   \underline{ \ottnt{a} }    )    :   \underline{ \ottnt{B} }  $.   

\end{itemize} 

Now, we can show that, if $ \Gamma  \vdash  \ottnt{a}  :  \ottnt{A} $ in \ED{}($ \mathcal{L} $), then $ \lfloor   \underline{ \ottnt{a} }   \rfloor  \equiv  \lfloor  \ottnt{a}  \rfloor $. The proof is by straightforward induction on the typing judgement. For the $\mathbf{eta}$-case, note that $ \lfloor    \mathbf{up}^{   \bot   ,  \ell  }   (   \ottkw{ret}  \:   \underline{ \ottnt{a} }    )     \rfloor  =  \lfloor    \underline{ \ottnt{a} }    \rfloor $. For the $\mathbf{bind}$-case, note that
$ \lfloor     j^{ \ell }_{  \underline{ \ottnt{B} }  }   \:   (    (   \mathbf{lift}^{  \ell  }   (   \lambda  \ottmv{x}  .   \underline{ \ottnt{b} }    )    )   \:   \underline{ \ottnt{a} }    )     \rfloor  =   \lfloor    j^{ \ell }_{  \underline{ \ottnt{B} }  }    \rfloor   \:   (    (   \lambda  \ottmv{x}  .   \lfloor   \underline{ \ottnt{b} }   \rfloor    )   \:    \lfloor   \underline{ \ottnt{a} }   \rfloor     )   \equiv   (   \lambda  \ottmv{y}  .  \ottmv{y}   )   \:   (    (   \lambda  \ottmv{x}  .   \lfloor   \underline{ \ottnt{b} }   \rfloor    )   \:    \lfloor   \underline{ \ottnt{a} }   \rfloor     )   \equiv   \lfloor   \underline{ \ottnt{b} }   \rfloor   \{   \lfloor   \underline{ \ottnt{a} }   \rfloor   /  \ottmv{x}  \} $.\\

So, for $ \Gamma  \vdash  \ottnt{a_{{\mathrm{1}}}}  :  \ottnt{A} $ and $ \Gamma  \vdash  \ottnt{a_{{\mathrm{2}}}}  :  \ottnt{A} $, if $\ottnt{a_{{\mathrm{1}}}} \simeq \ottnt{a_{{\mathrm{2}}}} $ in \ED{}($ \mathcal{L} $), then we have:
\[  \lfloor   \underline{ \ottnt{a_{{\mathrm{1}}}} }   \rfloor  \equiv  \lfloor  \ottnt{a_{{\mathrm{1}}}}  \rfloor  \equiv  \lfloor  \ottnt{a_{{\mathrm{2}}}}  \rfloor  \equiv  \lfloor   \underline{ \ottnt{a_{{\mathrm{2}}}} }   \rfloor .\]
Hence, $ \underline{ \ottnt{a_{{\mathrm{1}}}} }  \simeq  \underline{ \ottnt{a_{{\mathrm{2}}}} } $ in GMCC($ \mathcal{L} $).

\end{proof}


\begin{theorem}[Theorem \ref{GMCCround}]
Let $ \Gamma  \vdash  \ottnt{a}  :  \ottnt{A} $ be any derivation in GMCC($ \mathcal{L} $). Then, $  \underline{  \overline{ \ottnt{a} }  }   \equiv  \ottnt{a} $.
\end{theorem}

\begin{proof}
By induction on $ \Gamma  \vdash  \ottnt{a}  :  \ottnt{A} $. Note that $ \underline{  \overline{ \ottnt{A} }  }  = \ottnt{A}$.

\begin{itemize}

\item $\lambda$-calculus. By IH.

\item \Rref{MC-Return}. Have: $ \Gamma  \vdash   \ottkw{ret}  \:  \ottnt{a}   :   S_{ \ottsym{1} } \:  \ottnt{A}  $ where $ \Gamma  \vdash  \ottnt{a}  :  \ottnt{A} $.\\
By IH, $  \underline{  \overline{ \ottnt{a} }  }   \equiv  \ottnt{a} $. Now,
\[  \underline{  \overline{   \ottkw{ret}  \:  \ottnt{a}   }  }  =  \underline{   \mathbf{eta} ^{  \bot  }   \overline{ \ottnt{a} }    }  =  \mathbf{up}^{   \bot   ,   \bot   }   (   \ottkw{ret}  \:   \underline{  \overline{ \ottnt{a} }  }    )   \equiv  \ottkw{ret}  \:   \underline{  \overline{ \ottnt{a} }  }   \equiv  \ottkw{ret}  \:  \ottnt{a} . \]

\item \Rref{MC-Fmap}. Have: $ \Gamma  \vdash   \mathbf{lift}^{  \ell  }  \ottnt{f}   :    S_{  \ell  } \:  \ottnt{A}   \to   S_{  \ell  } \:  \ottnt{B}   $ where $ \Gamma  \vdash  \ottnt{f}  :   \ottnt{A}  \to  \ottnt{B}  $.\\
By IH, $  \underline{  \overline{ \ottnt{f} }  }   \equiv  \ottnt{f} $. Now,
\begin{align*}
&  \underline{  \overline{   \mathbf{lift}^{  \ell  }  \ottnt{f}   }  }  \\
= &  \underline{   \lambda  \ottmv{x}  :   \mathcal{T}_{ \ell } \:   \overline{ \ottnt{A} }    .   \mathbf{bind} ^{ \ell } \:  \ottmv{y}  =  \ottmv{x}  \: \mathbf{in} \:   \mathbf{eta} ^{ \ell }   (     \overline{ \ottnt{f} }    \:  \ottmv{y}   )      }  \\
= &   \lambda  \ottmv{x}  .    j^{ \ell }_{   S_{  \ell  } \:  \ottnt{B}   }     \:   (    (   \mathbf{lift}^{  \ell  }   (   \lambda  \ottmv{y}  .   \underline{   \mathbf{eta} ^{ \ell }   (     \overline{ \ottnt{f} }    \:  \ottmv{y}   )    }    )    )   \:  \ottmv{x}   )   \\
= &  \lambda  \ottmv{x}  .   \mathbf{join}^{  \ell  ,  \ell  }   (    (   \mathbf{lift}^{  \ell  }   (   \lambda  \ottmv{y}  .   \mathbf{up}^{   \bot   ,  \ell  }   (   \ottkw{ret}  \:   (     \underline{  \overline{ \ottnt{f} }  }    \:  \ottmv{y}   )    )     )    )   \:  \ottmv{x}   )    \\
\equiv &   \lambda  \ottmv{x}  .  \ottmv{x}   \:  \leftindex^{  \ell  }{\gg}\!\! =^{  \ell  }   \lambda  \ottmv{y}  .   \mathbf{up}^{   \bot   ,  \ell  }   (    \ottkw{ret}  \:  \ottnt{f}   \:  \ottmv{y}   )     \\
\equiv &  \lambda  \ottmv{x}  .   \mathbf{up}^{  \ell  ,  \ell  }   (   \ottmv{x}  \:  \leftindex^{  \ell  }{\gg}\!\! =^{   \bot   }   \lambda  \ottmv{y}  .   \ottkw{ret}  \:   (   \ottnt{f}  \:  \ottmv{y}   )      )    \hspace*{3pt} [\text{By Equation } (\ref{eq:natr})] \\
\equiv &   \lambda  \ottmv{x}  .  \ottmv{x}   \:  \leftindex^{  \ell  }{\gg}\!\! =^{   \bot   }   \lambda  \ottmv{y}  .   \ottkw{ret}  \:   (   \ottnt{f}  \:  \ottmv{y}   )     \\
\equiv &  \mathbf{lift}^{  \ell  }  \ottnt{f}  \hspace*{3pt} [\text{By Equation } (\ref{lifteqn})]
\end{align*} 

\item \Rref{MC-Join}. Have $ \Gamma  \vdash   \mathbf{join}^{  \ell_{{\mathrm{1}}}  ,  \ell_{{\mathrm{2}}}  }  \ottnt{a}   :   S_{   \ell_{{\mathrm{1}}}  \vee  \ell_{{\mathrm{2}}}   } \:  \ottnt{A}  $ where $ \Gamma  \vdash  \ottnt{a}  :   S_{  \ell_{{\mathrm{1}}}  } \:   S_{  \ell_{{\mathrm{2}}}  } \:  \ottnt{A}   $.\\
By IH, $  \underline{  \overline{ \ottnt{a} }  }   \equiv  \ottnt{a} $. Now,
\begin{align*}
&  \underline{  \overline{   \mathbf{join}^{  \ell_{{\mathrm{1}}}  ,  \ell_{{\mathrm{2}}}  }  \ottnt{a}   }  }  \\
= &  \underline{   \mathbf{bind} ^{ \ell_{{\mathrm{1}}} } \:  \ottmv{x}  =   \overline{ \ottnt{a} }   \: \mathbf{in} \:   \mathbf{bind} ^{ \ell_{{\mathrm{2}}} } \:  \ottmv{y}  =  \ottmv{x}  \: \mathbf{in} \:   \mathbf{eta} ^{  \ell_{{\mathrm{1}}}  \vee  \ell_{{\mathrm{2}}}  }  \ottmv{y}     }  \\
= &    j^{ \ell_{{\mathrm{1}}} }_{   S_{   \ell_{{\mathrm{1}}}  \vee  \ell_{{\mathrm{2}}}   } \:  \ottnt{A}   }    \:   (    (   \mathbf{lift}^{  \ell_{{\mathrm{1}}}  }   (   \lambda  \ottmv{x}  .   \underline{   \mathbf{bind} ^{ \ell_{{\mathrm{2}}} } \:  \ottmv{y}  =  \ottmv{x}  \: \mathbf{in} \:   \mathbf{eta} ^{  \ell_{{\mathrm{1}}}  \vee  \ell_{{\mathrm{2}}}  }  \ottmv{y}    }    )    )   \:    \underline{  \overline{ \ottnt{a} }  }     )   \\
\equiv &  \mathbf{join}^{  \ell_{{\mathrm{1}}}  ,   \ell_{{\mathrm{1}}}  \vee  \ell_{{\mathrm{2}}}   }   (    (   \mathbf{lift}^{  \ell_{{\mathrm{1}}}  }   (    \lambda  \ottmv{x}  .    j^{ \ell_{{\mathrm{2}}} }_{   S_{   \ell_{{\mathrm{1}}}  \vee  \ell_{{\mathrm{2}}}   } \:  \ottnt{A}   }     \:   (    (   \mathbf{lift}^{  \ell_{{\mathrm{2}}}  }   (   \lambda  \ottmv{y}  .   \underline{   \mathbf{eta} ^{  \ell_{{\mathrm{1}}}  \vee  \ell_{{\mathrm{2}}}  }  \ottmv{y}   }    )    )   \:  \ottmv{x}   )    )    )   \:  \ottnt{a}   )   \\
= &  \mathbf{join}^{  \ell_{{\mathrm{1}}}  ,   \ell_{{\mathrm{1}}}  \vee  \ell_{{\mathrm{2}}}   }   (    (   \mathbf{lift}^{  \ell_{{\mathrm{1}}}  }   (   \lambda  \ottmv{x}  .   \mathbf{join}^{  \ell_{{\mathrm{2}}}  ,   \ell_{{\mathrm{1}}}  \vee  \ell_{{\mathrm{2}}}   }   (    (   \mathbf{lift}^{  \ell_{{\mathrm{2}}}  }   (   \lambda  \ottmv{y}  .   \mathbf{up}^{   \bot   ,   \ell_{{\mathrm{1}}}  \vee  \ell_{{\mathrm{2}}}   }   (   \ottkw{ret}  \:  \ottmv{y}   )     )    )   \:  \ottmv{x}   )     )    )   \:  \ottnt{a}   )   \\
= &  \mathbf{join}^{  \ell_{{\mathrm{1}}}  ,   \ell_{{\mathrm{1}}}  \vee  \ell_{{\mathrm{2}}}   }   (    (   \mathbf{lift}^{  \ell_{{\mathrm{1}}}  }   (   \lambda  \ottmv{x}  .   (   \ottmv{x}  \:  \leftindex^{  \ell_{{\mathrm{2}}}  }{\gg}\!\! =^{   \ell_{{\mathrm{1}}}  \vee  \ell_{{\mathrm{2}}}   }   (   \lambda  \ottmv{y}  .   \mathbf{up}^{   \bot   ,   \ell_{{\mathrm{1}}}  \vee  \ell_{{\mathrm{2}}}   }   (   \ottkw{ret}  \:  \ottmv{y}   )     )    )    )    )   \:  \ottnt{a}   )   \\
\equiv &  \mathbf{join}^{  \ell_{{\mathrm{1}}}  ,   \ell_{{\mathrm{1}}}  \vee  \ell_{{\mathrm{2}}}   }   (    (   \mathbf{lift}^{  \ell_{{\mathrm{1}}}  }   (   \lambda  \ottmv{x}  .   \mathbf{up}^{  \ell_{{\mathrm{2}}}  ,   \ell_{{\mathrm{1}}}  \vee  \ell_{{\mathrm{2}}}   }   (   \ottmv{x}  \:  \leftindex^{  \ell_{{\mathrm{2}}}  }{\gg}\!\! =^{   \bot   }   \lambda  \ottmv{y}  .   \ottkw{ret}  \:  \ottmv{y}     )     )    )   \:  \ottnt{a}   )   \hspace*{3pt} [\text{By Equation } (\ref{eq:natr})] \\
\equiv &  \mathbf{join}^{  \ell_{{\mathrm{1}}}  ,   \ell_{{\mathrm{1}}}  \vee  \ell_{{\mathrm{2}}}   }   (    (   \mathbf{lift}^{  \ell_{{\mathrm{1}}}  }   (   \lambda  \ottmv{x}  .   \mathbf{up}^{  \ell_{{\mathrm{2}}}  ,   \ell_{{\mathrm{1}}}  \vee  \ell_{{\mathrm{2}}}   }  \ottmv{x}    )    )   \:  \ottnt{a}   )   \hspace*{3pt} [\text{By Equation } (\ref{eq:idr})]\\
= &  \ottnt{a}  \:  \leftindex^{  \ell_{{\mathrm{1}}}  }{\gg}\!\! =^{   \ell_{{\mathrm{1}}}  \vee  \ell_{{\mathrm{2}}}   }   \lambda  \ottmv{x}  .   \mathbf{up}^{  \ell_{{\mathrm{2}}}  ,   \ell_{{\mathrm{1}}}  \vee  \ell_{{\mathrm{2}}}   }  \ottmv{x}    \\
\equiv &  \mathbf{up}^{   \ell_{{\mathrm{1}}}  \vee  \ell_{{\mathrm{2}}}   ,   \ell_{{\mathrm{1}}}  \vee  \ell_{{\mathrm{2}}}   }   (   \ottnt{a}  \:  \leftindex^{  \ell_{{\mathrm{1}}}  }{\gg}\!\! =^{  \ell_{{\mathrm{2}}}  }   \lambda  \ottmv{x}  .  \ottmv{x}    )   \hspace{3pt} [\text{By Equation } (\ref{eq:natr})]\\
\equiv &  \ottnt{a}  \:  \leftindex^{  \ell_{{\mathrm{1}}}  }{\gg}\!\! =^{  \ell_{{\mathrm{2}}}  }   \lambda  \ottmv{x}  .  \ottmv{x}   \equiv  \mathbf{join}^{  \ell_{{\mathrm{1}}}  ,  \ell_{{\mathrm{2}}}  }  \ottnt{a}  \hspace*{3pt} [\text{By Equation } (\ref{joineqn})]
\end{align*}

\item \Rref{MC-Up}. Have: $ \Gamma  \vdash   \mathbf{up}^{  \ell_{{\mathrm{1}}}  ,  \ell_{{\mathrm{2}}}  }  \ottnt{a}   :   S_{  \ell_{{\mathrm{2}}}  } \:  \ottnt{A}  $ where $ \Gamma  \vdash  \ottnt{a}  :   S_{  \ell_{{\mathrm{1}}}  } \:  \ottnt{A}  $ and $ \ell_{{\mathrm{1}}}  \sqsubseteq  \ell_{{\mathrm{2}}} $.\\
By IH, $  \underline{  \overline{ \ottnt{a} }  }   \equiv  \ottnt{a} $. Now,
\begin{align*}
&  \underline{  \overline{   \mathbf{up}^{  \ell_{{\mathrm{1}}}  ,  \ell_{{\mathrm{2}}}  }  \ottnt{a}   }  }  \\
= &  \underline{   \mathbf{bind} ^{ \ell_{{\mathrm{1}}} } \:  \ottmv{x}  =   \overline{ \ottnt{a} }   \: \mathbf{in} \:   \mathbf{eta} ^{ \ell_{{\mathrm{2}}} }  \ottmv{x}    }  \\
= &    j^{ \ell_{{\mathrm{1}}} }_{   S_{  \ell_{{\mathrm{2}}}  } \:  \ottnt{A}   }    \:   (    (   \mathbf{lift}^{  \ell_{{\mathrm{1}}}  }   (   \lambda  \ottmv{x}  .   \underline{   \mathbf{eta} ^{ \ell_{{\mathrm{2}}} }  \ottmv{x}   }    )    )   \:   \underline{  \overline{ \ottnt{a} }  }    )   \\
= &  \mathbf{join}^{  \ell_{{\mathrm{1}}}  ,  \ell_{{\mathrm{2}}}  }   (    (   \mathbf{lift}^{  \ell_{{\mathrm{1}}}  }   (   \lambda  \ottmv{x}  .   \mathbf{up}^{   \bot   ,  \ell_{{\mathrm{2}}}  }   (   \ottkw{ret}  \:  \ottmv{x}   )     )    )   \:   \underline{  \overline{ \ottnt{a} }  }    )   \\
\equiv &  \ottnt{a}  \:  \leftindex^{  \ell_{{\mathrm{1}}}  }{\gg}\!\! =^{  \ell_{{\mathrm{2}}}  }   (   \lambda  \ottmv{x}  .   \mathbf{up}^{   \bot   ,  \ell_{{\mathrm{2}}}  }   (   \ottkw{ret}  \:  \ottmv{x}   )     )   \\
\equiv &  \mathbf{up}^{  \ell_{{\mathrm{1}}}  ,  \ell_{{\mathrm{2}}}  }   (   \ottnt{a}  \:  \leftindex^{  \ell_{{\mathrm{1}}}  }{\gg}\!\! =^{   \bot   }   \lambda  \ottmv{x}  .   \ottkw{ret}  \:  \ottmv{x}     )   \hspace{3pt} [\text{By Equation } (\ref{eq:natr})]\\
\equiv &  \mathbf{up}^{  \ell_{{\mathrm{1}}}  ,  \ell_{{\mathrm{2}}}  }  \ottnt{a}  \hspace*{3pt} [\text{By Equation } (\ref{eq:idr})]
\end{align*}

\item \Rref{MC-Extract}. Have: $ \Gamma  \vdash   \mathbf{extr} \:  \ottnt{a}   :  \ottnt{A} $ where $ \Gamma  \vdash  \ottnt{a}  :   S_{ \ottsym{1} } \:  \ottnt{A}  $.\\
By IH, $  \underline{  \overline{ \ottnt{a} }  }   \equiv  \ottnt{a} $.\\
Now, $ \underline{  \overline{  \mathbf{extr} \:  \ottnt{a}  }  }  =  \underline{   \mathbf{bind} ^{  \bot  } \:  \ottmv{x}  =   \overline{ \ottnt{a} }   \: \mathbf{in} \:  \ottmv{x}   }  =    j^{  \bot  }_{ \ottnt{A} }    \:   (    (   \mathbf{lift}^{   \bot   }   (   \lambda  \ottmv{x}  .  \ottmv{x}   )    )   \:   \underline{  \overline{ \ottnt{a} }  }    )   \equiv    j^{  \bot  }_{ \ottnt{A} }    \:  \ottnt{a}  =  \mathbf{extr} \:  \ottnt{a} $.

\item \Rref{MC-Fork}. Have: $ \Gamma  \vdash   \mathbf{fork}^{  \ell_{{\mathrm{1}}}  ,  \ell_{{\mathrm{2}}}  }  \ottnt{a}   :   S_{  \ell_{{\mathrm{1}}}  } \:   S_{  \ell_{{\mathrm{2}}}  } \:  \ottnt{A}   $ where $ \Gamma  \vdash  \ottnt{a}  :   S_{   \ell_{{\mathrm{1}}}  \vee  \ell_{{\mathrm{2}}}   } \:  \ottnt{A}  $.\\
By IH, $  \underline{  \overline{ \ottnt{a} }  }   \equiv  \ottnt{a} $. Now,
\begin{align*}
&  \underline{  \overline{   \mathbf{fork}^{  \ell_{{\mathrm{1}}}  ,  \ell_{{\mathrm{2}}}  }  \ottnt{a}   }  }  \\
= &  \underline{   \mathbf{bind} ^{  \ell_{{\mathrm{1}}}  \vee  \ell_{{\mathrm{2}}}  } \:  \ottmv{x}  =   \overline{ \ottnt{a} }   \: \mathbf{in} \:   \mathbf{eta} ^{ \ell_{{\mathrm{1}}} }   \mathbf{eta} ^{ \ell_{{\mathrm{2}}} }  \ottmv{x}     }  \\
= &    j^{  \ell_{{\mathrm{1}}}  \vee  \ell_{{\mathrm{2}}}  }_{   S_{  \ell_{{\mathrm{1}}}  } \:   S_{  \ell_{{\mathrm{2}}}  } \:  \ottnt{A}    }    \:   (    (   \mathbf{lift}^{   \ell_{{\mathrm{1}}}  \vee  \ell_{{\mathrm{2}}}   }   (   \lambda  \ottmv{x}  .   \underline{   \mathbf{eta} ^{ \ell_{{\mathrm{1}}} }   \mathbf{eta} ^{ \ell_{{\mathrm{2}}} }  \ottmv{x}    }    )    )   \:    \underline{  \overline{ \ottnt{a} }  }     )   \\
\equiv &    j^{  \ell_{{\mathrm{1}}}  \vee  \ell_{{\mathrm{2}}}  }_{   S_{  \ell_{{\mathrm{1}}}  } \:   S_{  \ell_{{\mathrm{2}}}  } \:  \ottnt{A}    }    \:   (    (   \mathbf{lift}^{   \ell_{{\mathrm{1}}}  \vee  \ell_{{\mathrm{2}}}   }   (   \lambda  \ottmv{x}  .   \mathbf{up}^{   \bot   ,  \ell_{{\mathrm{1}}}  }   \ottkw{ret}  \:   \mathbf{up}^{   \bot   ,  \ell_{{\mathrm{2}}}  }   \ottkw{ret}  \:  \ottmv{x}       )    )   \:  \ottnt{a}   )   \\
\triangleq &    j^{  \ell_{{\mathrm{1}}}  \vee  \ell_{{\mathrm{2}}}  }_{   S_{  \ell_{{\mathrm{1}}}  } \:   S_{  \ell_{{\mathrm{2}}}  } \:  \ottnt{A}    }    \:  \ottnt{t_{{\mathrm{0}}}}  \hspace*{3pt} [\text{Say, } \ottnt{t_{{\mathrm{0}}}} \triangleq   (   \mathbf{lift}^{   \ell_{{\mathrm{1}}}  \vee  \ell_{{\mathrm{2}}}   }   (   \lambda  \ottmv{x}  .   \mathbf{up}^{   \bot   ,  \ell_{{\mathrm{1}}}  }   \ottkw{ret}  \:   \mathbf{up}^{   \bot   ,  \ell_{{\mathrm{2}}}  }   \ottkw{ret}  \:  \ottmv{x}       )    )   \:  \ottnt{a}  ] \\
= &    j^{ \ell_{{\mathrm{1}}} }_{   S_{  \ell_{{\mathrm{1}}}  } \:   S_{  \ell_{{\mathrm{2}}}  } \:  \ottnt{A}    }    \:   (    (   \mathbf{lift}^{  \ell_{{\mathrm{1}}}  }    j^{ \ell_{{\mathrm{2}}} }_{   S_{  \ell_{{\mathrm{1}}}  } \:   S_{  \ell_{{\mathrm{2}}}  } \:  \ottnt{A}    }     )   \:   (   \mathbf{fork}^{  \ell_{{\mathrm{1}}}  ,  \ell_{{\mathrm{2}}}  }  \ottnt{t_{{\mathrm{0}}}}   )    )   \\
= &    j^{ \ell_{{\mathrm{1}}} }_{   S_{  \ell_{{\mathrm{1}}}  } \:   S_{  \ell_{{\mathrm{2}}}  } \:  \ottnt{A}    }    \:   (    (   \mathbf{lift}^{  \ell_{{\mathrm{1}}}  }   (    \lambda  \ottmv{z}  .   (   \mathbf{lift}^{  \ell_{{\mathrm{1}}}  }   j^{ \ell_{{\mathrm{2}}} }_{  S_{  \ell_{{\mathrm{2}}}  } \:  \ottnt{A}  }    )    \:   (   \mathbf{fork}^{  \ell_{{\mathrm{1}}}  ,  \ell_{{\mathrm{2}}}  }   (   \mathbf{join}^{  \ell_{{\mathrm{2}}}  ,  \ell_{{\mathrm{1}}}  }  \ottmv{z}   )    )    )    )   \:   (   \mathbf{fork}^{  \ell_{{\mathrm{1}}}  ,  \ell_{{\mathrm{2}}}  }  \ottnt{t_{{\mathrm{0}}}}   )    )  \\
\equiv &    j^{ \ell_{{\mathrm{1}}} }_{   S_{  \ell_{{\mathrm{1}}}  } \:   S_{  \ell_{{\mathrm{2}}}  } \:  \ottnt{A}    }    \:   (    (     \lambda  \ottmv{z}  .   (   \mathbf{lift}^{  \ell_{{\mathrm{1}}}  }   (   \mathbf{lift}^{  \ell_{{\mathrm{1}}}  }   j^{ \ell_{{\mathrm{2}}} }_{  S_{  \ell_{{\mathrm{2}}}  } \:  \ottnt{A}  }    )    )    \:   (   \mathbf{lift}^{  \ell_{{\mathrm{1}}}  }   (   \lambda  \ottmv{y}  .   \mathbf{fork}^{  \ell_{{\mathrm{1}}}  ,  \ell_{{\mathrm{2}}}  }   (   \mathbf{join}^{  \ell_{{\mathrm{2}}}  ,  \ell_{{\mathrm{1}}}  }  \ottmv{y}   )     )    )    \:  \ottmv{z}   )   \:   (   \mathbf{fork}^{  \ell_{{\mathrm{1}}}  ,  \ell_{{\mathrm{2}}}  }  \ottnt{t_{{\mathrm{0}}}}   )    )   \hspace*{3pt} [\text{By Eqn. } (\ref{eq:ccomp})] \\
\equiv &    j^{ \ell_{{\mathrm{1}}} }_{   S_{  \ell_{{\mathrm{1}}}  } \:   S_{  \ell_{{\mathrm{2}}}  } \:  \ottnt{A}    }    \:   (    (   \mathbf{lift}^{  \ell_{{\mathrm{1}}}  }   (   \mathbf{lift}^{  \ell_{{\mathrm{1}}}  }   j^{ \ell_{{\mathrm{2}}} }_{  S_{  \ell_{{\mathrm{2}}}  } \:  \ottnt{A}  }    )    )   \:   (    (   \mathbf{lift}^{  \ell_{{\mathrm{1}}}  }   (   \lambda  \ottmv{y}  .   \mathbf{fork}^{  \ell_{{\mathrm{1}}}  ,  \ell_{{\mathrm{2}}}  }   (   \mathbf{join}^{  \ell_{{\mathrm{2}}}  ,  \ell_{{\mathrm{1}}}  }  \ottmv{y}   )     )    )   \:   (   \mathbf{fork}^{  \ell_{{\mathrm{1}}}  ,  \ell_{{\mathrm{2}}}  }  \ottnt{t_{{\mathrm{0}}}}   )    )    )   \\
= &    j^{ \ell_{{\mathrm{1}}} }_{   S_{  \ell_{{\mathrm{1}}}  } \:   S_{  \ell_{{\mathrm{2}}}  } \:  \ottnt{A}    }    \:   (    (   \mathbf{lift}^{  \ell_{{\mathrm{1}}}  }   (   \mathbf{lift}^{  \ell_{{\mathrm{1}}}  }   j^{ \ell_{{\mathrm{2}}} }_{  S_{  \ell_{{\mathrm{2}}}  } \:  \ottnt{A}  }    )    )   \:   (    (   \lambda  \ottmv{y}  .   \mathbf{fork}^{  \ell_{{\mathrm{1}}}  ,  \ell_{{\mathrm{2}}}  }   (   \mathbf{join}^{  \ell_{{\mathrm{2}}}  ,  \ell_{{\mathrm{1}}}  }  \ottmv{y}   )     )   \:  \leftindex^{  \ell_{{\mathrm{2}}}  }{\ll}\!\! =^{  \ell_{{\mathrm{1}}}  }  \ottnt{t_{{\mathrm{0}}}}   )    )   \\
\triangleq &    j^{ \ell_{{\mathrm{1}}} }_{   S_{  \ell_{{\mathrm{1}}}  } \:   S_{  \ell_{{\mathrm{2}}}  } \:  \ottnt{A}    }    \:   (    (   \mathbf{lift}^{  \ell_{{\mathrm{1}}}  }   (   \mathbf{lift}^{  \ell_{{\mathrm{1}}}  }   j^{ \ell_{{\mathrm{2}}} }_{  S_{  \ell_{{\mathrm{2}}}  } \:  \ottnt{A}  }    )    )   \:   (    (   \lambda  \ottmv{y}  .   \mathbf{fork}^{  \ell_{{\mathrm{1}}}  ,  \ell_{{\mathrm{2}}}  }   (   \mathbf{join}^{  \ell_{{\mathrm{2}}}  ,  \ell_{{\mathrm{1}}}  }  \ottmv{y}   )     )   \:  \leftindex^{  \ell_{{\mathrm{2}}}  }{\ll}\!\! =^{  \ell_{{\mathrm{1}}}  }   (    (   \mathbf{lift}^{   \ell_{{\mathrm{1}}}  \vee  \ell_{{\mathrm{2}}}   }  \ottnt{f}   )   \:  \ottnt{a}   )    )    )   \\ & [\text{Say, } \ottnt{f} \triangleq  \lambda  \ottmv{x}  .   \mathbf{up}^{   \bot   ,  \ell_{{\mathrm{1}}}  }   \ottkw{ret}  \:   \mathbf{up}^{   \bot   ,  \ell_{{\mathrm{2}}}  }   \ottkw{ret}  \:  \ottmv{x}      ] \\
\equiv &    j^{ \ell_{{\mathrm{1}}} }_{   S_{  \ell_{{\mathrm{1}}}  } \:   S_{  \ell_{{\mathrm{2}}}  } \:  \ottnt{A}    }    \:   (    (   \mathbf{lift}^{  \ell_{{\mathrm{1}}}  }   (   \mathbf{lift}^{  \ell_{{\mathrm{1}}}  }   j^{ \ell_{{\mathrm{2}}} }_{  S_{  \ell_{{\mathrm{2}}}  } \:  \ottnt{A}  }    )    )   \:   (    (   \lambda  \ottmv{y}  .   \mathbf{fork}^{  \ell_{{\mathrm{1}}}  ,  \ell_{{\mathrm{2}}}  }   (   \mathbf{join}^{  \ell_{{\mathrm{2}}}  ,  \ell_{{\mathrm{1}}}  }  \ottmv{y}   )     )   \:  \leftindex^{  \ell_{{\mathrm{2}}}  }{\ll}\!\! =^{  \ell_{{\mathrm{1}}}  }   (    (    \lambda  \ottmv{x}  .  \ottnt{f}   \:   (   \mathbf{extr} \:  \ottmv{x}   )    )   \:  \leftindex^{   \bot   }{\ll}\!\! =^{   \ell_{{\mathrm{1}}}  \vee  \ell_{{\mathrm{2}}}   }  \ottnt{a}   )    )    )   \\ & [\text{By Eqn. } (\ref{lifteq2})] \\
\equiv &    j^{ \ell_{{\mathrm{1}}} }_{   S_{  \ell_{{\mathrm{1}}}  } \:   S_{  \ell_{{\mathrm{2}}}  } \:  \ottnt{A}    }    \:   (    (   \mathbf{lift}^{  \ell_{{\mathrm{1}}}  }   (   \mathbf{lift}^{  \ell_{{\mathrm{1}}}  }   j^{ \ell_{{\mathrm{2}}} }_{  S_{  \ell_{{\mathrm{2}}}  } \:  \ottnt{A}  }    )    )   \:   (    (    \lambda  \ottmv{z}  .   (   \lambda  \ottmv{y}  .   \mathbf{fork}^{  \ell_{{\mathrm{1}}}  ,  \ell_{{\mathrm{2}}}  }   (   \mathbf{join}^{  \ell_{{\mathrm{2}}}  ,  \ell_{{\mathrm{1}}}  }  \ottmv{y}   )     )    \:   (    (    \lambda  \ottmv{x}  .  \ottnt{f}   \:   (   \mathbf{extr} \:  \ottmv{x}   )    )   \:  \leftindex^{   \bot   }{\ll}\!\! =^{  \ell_{{\mathrm{2}}}  }  \ottmv{z}   )    )   \:  \leftindex^{  \ell_{{\mathrm{2}}}  }{\ll}\!\! =^{  \ell_{{\mathrm{1}}}  }  \ottnt{a}   )    )   \\ & [\text{By Eqn. } (\ref{eq:cassoc})] \\
\equiv &    j^{ \ell_{{\mathrm{1}}} }_{   S_{  \ell_{{\mathrm{1}}}  } \:   S_{  \ell_{{\mathrm{2}}}  } \:  \ottnt{A}    }    \:   (    (   \mathbf{lift}^{  \ell_{{\mathrm{1}}}  }   (   \mathbf{lift}^{  \ell_{{\mathrm{1}}}  }   j^{ \ell_{{\mathrm{2}}} }_{  S_{  \ell_{{\mathrm{2}}}  } \:  \ottnt{A}  }    )    )   \:   (    (    \lambda  \ottmv{z}  .   (   \lambda  \ottmv{y}  .   \mathbf{fork}^{  \ell_{{\mathrm{1}}}  ,  \ell_{{\mathrm{2}}}  }   (   \mathbf{join}^{  \ell_{{\mathrm{2}}}  ,  \ell_{{\mathrm{1}}}  }  \ottmv{y}   )     )    \:   (   \mathbf{fork}^{  \ell_{{\mathrm{2}}}  ,  \ell_{{\mathrm{1}}}  }   (   \mathbf{up}^{  \ell_{{\mathrm{2}}}  ,   \ell_{{\mathrm{2}}}  \vee  \ell_{{\mathrm{1}}}   }   \mathbf{fork}^{  \ell_{{\mathrm{2}}}  ,  \ell_{{\mathrm{2}}}  }  \ottmv{z}    )    )    )   \:  \leftindex^{  \ell_{{\mathrm{2}}}  }{\ll}\!\! =^{  \ell_{{\mathrm{1}}}  }  \ottnt{a}   )    )   \\
& [\because    (    \lambda  \ottmv{x}  .  \ottnt{f}   \:   (   \mathbf{extr} \:  \ottmv{x}   )    )   \:  \leftindex^{   \bot   }{\ll}\!\! =^{  \ell_{{\mathrm{2}}}  }  \ottmv{z}   \equiv   \mathbf{fork}^{  \ell_{{\mathrm{2}}}  ,  \ell_{{\mathrm{1}}}  }   (   \mathbf{up}^{  \ell_{{\mathrm{2}}}  ,   \ell_{{\mathrm{2}}}  \vee  \ell_{{\mathrm{1}}}   }   \mathbf{fork}^{  \ell_{{\mathrm{2}}}  ,  \ell_{{\mathrm{2}}}  }  \ottmv{z}    )    \; (\text{See below}) ] \\
\equiv &    j^{ \ell_{{\mathrm{1}}} }_{   S_{  \ell_{{\mathrm{1}}}  } \:   S_{  \ell_{{\mathrm{2}}}  } \:  \ottnt{A}    }    \:   (    (   \mathbf{lift}^{  \ell_{{\mathrm{1}}}  }   (   \mathbf{lift}^{  \ell_{{\mathrm{1}}}  }   j^{ \ell_{{\mathrm{2}}} }_{  S_{  \ell_{{\mathrm{2}}}  } \:  \ottnt{A}  }    )    )   \:   (    (   \lambda  \ottmv{z}  .   \mathbf{fork}^{  \ell_{{\mathrm{1}}}  ,  \ell_{{\mathrm{2}}}  }   \mathbf{join}^{  \ell_{{\mathrm{2}}}  ,  \ell_{{\mathrm{1}}}  }   \mathbf{fork}^{  \ell_{{\mathrm{2}}}  ,  \ell_{{\mathrm{1}}}  }   \mathbf{up}^{  \ell_{{\mathrm{2}}}  ,   \ell_{{\mathrm{2}}}  \vee  \ell_{{\mathrm{1}}}   }   \mathbf{fork}^{  \ell_{{\mathrm{2}}}  ,  \ell_{{\mathrm{2}}}  }  \ottmv{z}        )   \:  \leftindex^{  \ell_{{\mathrm{2}}}  }{\ll}\!\! =^{  \ell_{{\mathrm{1}}}  }  \ottnt{a}   )    )   \\
\equiv &    j^{ \ell_{{\mathrm{1}}} }_{   S_{  \ell_{{\mathrm{1}}}  } \:   S_{  \ell_{{\mathrm{2}}}  } \:  \ottnt{A}    }    \:   (    (   \mathbf{lift}^{  \ell_{{\mathrm{1}}}  }   (   \mathbf{lift}^{  \ell_{{\mathrm{1}}}  }   j^{ \ell_{{\mathrm{2}}} }_{  S_{  \ell_{{\mathrm{2}}}  } \:  \ottnt{A}  }    )    )   \:   (    (   \lambda  \ottmv{z}  .   \mathbf{fork}^{  \ell_{{\mathrm{1}}}  ,  \ell_{{\mathrm{2}}}  }   \mathbf{up}^{  \ell_{{\mathrm{2}}}  ,   \ell_{{\mathrm{2}}}  \vee  \ell_{{\mathrm{1}}}   }   \mathbf{fork}^{  \ell_{{\mathrm{2}}}  ,  \ell_{{\mathrm{2}}}  }  \ottmv{z}      )   \:  \leftindex^{  \ell_{{\mathrm{2}}}  }{\ll}\!\! =^{  \ell_{{\mathrm{1}}}  }  \ottnt{a}   )    )   \\
= &    j^{ \ell_{{\mathrm{1}}} }_{   S_{  \ell_{{\mathrm{1}}}  } \:   S_{  \ell_{{\mathrm{2}}}  } \:  \ottnt{A}    }    \:   (    (   \mathbf{lift}^{  \ell_{{\mathrm{1}}}  }   (   \mathbf{lift}^{  \ell_{{\mathrm{1}}}  }   j^{ \ell_{{\mathrm{2}}} }_{  S_{  \ell_{{\mathrm{2}}}  } \:  \ottnt{A}  }    )    )   \:   (    (   \mathbf{lift}^{  \ell_{{\mathrm{1}}}  }   (   \lambda  \ottmv{z}  .   \mathbf{fork}^{  \ell_{{\mathrm{1}}}  ,  \ell_{{\mathrm{2}}}  }   \mathbf{up}^{  \ell_{{\mathrm{2}}}  ,   \ell_{{\mathrm{2}}}  \vee  \ell_{{\mathrm{1}}}   }   \mathbf{fork}^{  \ell_{{\mathrm{2}}}  ,  \ell_{{\mathrm{2}}}  }  \ottmv{z}      )    )   \:   (   \mathbf{fork}^{  \ell_{{\mathrm{1}}}  ,  \ell_{{\mathrm{2}}}  }  \ottnt{a}   )    )    )   \\
\equiv &    j^{ \ell_{{\mathrm{1}}} }_{   S_{  \ell_{{\mathrm{1}}}  } \:   S_{  \ell_{{\mathrm{2}}}  } \:  \ottnt{A}    }    \:   (    (   \mathbf{lift}^{  \ell_{{\mathrm{1}}}  }   (    \lambda  \ottmv{z}  .   (   \mathbf{lift}^{  \ell_{{\mathrm{1}}}  }   j^{ \ell_{{\mathrm{2}}} }_{  S_{  \ell_{{\mathrm{2}}}  } \:  \ottnt{A}  }    )    \:   (   \mathbf{fork}^{  \ell_{{\mathrm{1}}}  ,  \ell_{{\mathrm{2}}}  }   \mathbf{up}^{  \ell_{{\mathrm{2}}}  ,   \ell_{{\mathrm{2}}}  \vee  \ell_{{\mathrm{1}}}   }   \mathbf{fork}^{  \ell_{{\mathrm{2}}}  ,  \ell_{{\mathrm{2}}}  }  \ottmv{z}     )    )    )   \:   (   \mathbf{fork}^{  \ell_{{\mathrm{1}}}  ,  \ell_{{\mathrm{2}}}  }  \ottnt{a}   )    )   \hspace*{3pt} [\text{By Eqn. } (\ref{eq:ccomp})]\\
= &    j^{ \ell_{{\mathrm{1}}} }_{   S_{  \ell_{{\mathrm{1}}}  } \:   S_{  \ell_{{\mathrm{2}}}  } \:  \ottnt{A}    }    \:   (    (   \mathbf{lift}^{  \ell_{{\mathrm{1}}}  }   (    \lambda  \ottmv{z}  .   (   \mathbf{lift}^{  \ell_{{\mathrm{1}}}  }   (   \lambda  \ottmv{x}  .   \mathbf{join}^{  \ell_{{\mathrm{2}}}  ,  \ell_{{\mathrm{2}}}  }  \ottmv{x}    )    )    \:   (   \mathbf{fork}^{  \ell_{{\mathrm{1}}}  ,  \ell_{{\mathrm{2}}}  }   \mathbf{up}^{  \ell_{{\mathrm{2}}}  ,   \ell_{{\mathrm{2}}}  \vee  \ell_{{\mathrm{1}}}   }   \mathbf{fork}^{  \ell_{{\mathrm{2}}}  ,  \ell_{{\mathrm{2}}}  }  \ottmv{z}     )    )    )   \:   (   \mathbf{fork}^{  \ell_{{\mathrm{1}}}  ,  \ell_{{\mathrm{2}}}  }  \ottnt{a}   )    )   \\
= &    j^{ \ell_{{\mathrm{1}}} }_{   S_{  \ell_{{\mathrm{1}}}  } \:   S_{  \ell_{{\mathrm{2}}}  } \:  \ottnt{A}    }    \:   (    (   \mathbf{lift}^{  \ell_{{\mathrm{1}}}  }   (   \lambda  \ottmv{z}  .   (    (   \lambda  \ottmv{x}  .   \mathbf{join}^{  \ell_{{\mathrm{2}}}  ,  \ell_{{\mathrm{2}}}  }  \ottmv{x}    )   \:  \leftindex^{  \ell_{{\mathrm{2}}}  }{\ll}\!\! =^{  \ell_{{\mathrm{1}}}  }   \mathbf{up}^{  \ell_{{\mathrm{2}}}  ,   \ell_{{\mathrm{2}}}  \vee  \ell_{{\mathrm{1}}}   }   \mathbf{fork}^{  \ell_{{\mathrm{2}}}  ,  \ell_{{\mathrm{2}}}  }  \ottmv{z}     )    )    )   \:   (   \mathbf{fork}^{  \ell_{{\mathrm{1}}}  ,  \ell_{{\mathrm{2}}}  }  \ottnt{a}   )    )  \\
\equiv &    j^{ \ell_{{\mathrm{1}}} }_{   S_{  \ell_{{\mathrm{1}}}  } \:   S_{  \ell_{{\mathrm{2}}}  } \:  \ottnt{A}    }    \:   (    (   \mathbf{lift}^{  \ell_{{\mathrm{1}}}  }   (   \lambda  \ottmv{z}  .   (   \mathbf{up}^{   \bot   ,  \ell_{{\mathrm{1}}}  }   (    \lambda  \ottmv{x}  .   \mathbf{join}^{  \ell_{{\mathrm{2}}}  ,  \ell_{{\mathrm{2}}}  }  \ottmv{x}    \:  \leftindex^{  \ell_{{\mathrm{2}}}  }{\ll}\!\! =^{   \bot   }   \mathbf{fork}^{  \ell_{{\mathrm{2}}}  ,  \ell_{{\mathrm{2}}}  }  \ottmv{z}    )    )    )    )   \:   (   \mathbf{fork}^{  \ell_{{\mathrm{1}}}  ,  \ell_{{\mathrm{2}}}  }  \ottnt{a}   )    )   \hspace*{3pt} [\text{By Eqn. } (\ref{eq:cnatl})]\\
\equiv &    j^{ \ell_{{\mathrm{1}}} }_{   S_{  \ell_{{\mathrm{1}}}  } \:   S_{  \ell_{{\mathrm{2}}}  } \:  \ottnt{A}    }    \:   (    (   \mathbf{lift}^{  \ell_{{\mathrm{1}}}  }   (   \lambda  \ottmv{z}  .   (   \mathbf{up}^{   \bot   ,  \ell_{{\mathrm{1}}}  }   \ottkw{ret}  \:   \mathbf{extr} \:   (    \lambda  \ottmv{x}  .   \mathbf{join}^{  \ell_{{\mathrm{2}}}  ,  \ell_{{\mathrm{2}}}  }  \ottmv{x}    \:  \leftindex^{  \ell_{{\mathrm{2}}}  }{\ll}\!\! =^{   \bot   }   \mathbf{fork}^{  \ell_{{\mathrm{2}}}  ,  \ell_{{\mathrm{2}}}  }  \ottmv{z}    )      )    )    )   \:   (   \mathbf{fork}^{  \ell_{{\mathrm{1}}}  ,  \ell_{{\mathrm{2}}}  }  \ottnt{a}   )    )   \\
\equiv &    j^{ \ell_{{\mathrm{1}}} }_{   S_{  \ell_{{\mathrm{1}}}  } \:   S_{  \ell_{{\mathrm{2}}}  } \:  \ottnt{A}    }    \:   (    (   \mathbf{lift}^{  \ell_{{\mathrm{1}}}  }   (   \lambda  \ottmv{z}  .   (   \mathbf{up}^{   \bot   ,  \ell_{{\mathrm{1}}}  }   \ottkw{ret}  \:   (   \mathbf{join}^{  \ell_{{\mathrm{2}}}  ,  \ell_{{\mathrm{2}}}  }   \mathbf{fork}^{  \ell_{{\mathrm{2}}}  ,  \ell_{{\mathrm{2}}}  }  \ottmv{z}    )     )    )    )   \:   (   \mathbf{fork}^{  \ell_{{\mathrm{1}}}  ,  \ell_{{\mathrm{2}}}  }  \ottnt{a}   )    )   \hspace*{3pt} [\text{By Eqn. } (\ref{eq:cidl})] \\
\equiv &    j^{ \ell_{{\mathrm{1}}} }_{   S_{  \ell_{{\mathrm{1}}}  } \:   S_{  \ell_{{\mathrm{2}}}  } \:  \ottnt{A}    }    \:   (    (   \mathbf{lift}^{  \ell_{{\mathrm{1}}}  }   (   \lambda  \ottmv{z}  .   (   \mathbf{up}^{   \bot   ,  \ell_{{\mathrm{1}}}  }   \ottkw{ret}  \:  \ottmv{z}    )    )    )   \:   (   \mathbf{fork}^{  \ell_{{\mathrm{1}}}  ,  \ell_{{\mathrm{2}}}  }  \ottnt{a}   )    )  \\
= &  \mathbf{join}^{  \ell_{{\mathrm{1}}}  ,  \ell_{{\mathrm{1}}}  }   (    (   \mathbf{lift}^{  \ell_{{\mathrm{1}}}  }   (   \lambda  \ottmv{z}  .   \mathbf{up}^{   \bot   ,  \ell_{{\mathrm{1}}}  }   \ottkw{ret}  \:  \ottmv{z}     )    )   \:   (   \mathbf{fork}^{  \ell_{{\mathrm{1}}}  ,  \ell_{{\mathrm{2}}}  }  \ottnt{a}   )    )   \\
= &   \mathbf{fork}^{  \ell_{{\mathrm{1}}}  ,  \ell_{{\mathrm{2}}}  }  \ottnt{a}   \:  \leftindex^{  \ell_{{\mathrm{1}}}  }{\gg}\!\! =^{  \ell_{{\mathrm{1}}}  }   \lambda  \ottmv{z}  .   \mathbf{up}^{   \bot   ,  \ell_{{\mathrm{1}}}  }   \ottkw{ret}  \:  \ottmv{z}     \\
\equiv &  \mathbf{up}^{  \ell_{{\mathrm{1}}}  ,  \ell_{{\mathrm{1}}}  }   (    \mathbf{fork}^{  \ell_{{\mathrm{1}}}  ,  \ell_{{\mathrm{2}}}  }  \ottnt{a}   \:  \leftindex^{  \ell_{{\mathrm{1}}}  }{\gg}\!\! =^{   \bot   }   \lambda  \ottmv{z}  .   \ottkw{ret}  \:  \ottmv{z}     )   \hspace*{3pt} [\text{By Eqn. } (\ref{eq:natr})] \\
\equiv &   \mathbf{fork}^{  \ell_{{\mathrm{1}}}  ,  \ell_{{\mathrm{2}}}  }  \ottnt{a}   \:  \leftindex^{  \ell_{{\mathrm{1}}}  }{\gg}\!\! =^{   \bot   }   \lambda  \ottmv{z}  .   \ottkw{ret}  \:  \ottmv{z}    \equiv  \mathbf{fork}^{  \ell_{{\mathrm{1}}}  ,  \ell_{{\mathrm{2}}}  }  \ottnt{a}  \hspace*{3pt} [\text{By Eqn. } (\ref{eq:idr}) ]
\end{align*}
Note that:
\begin{align*}
&    \lambda  \ottmv{x}  .  \ottnt{f}   \:   (   \mathbf{extr} \:  \ottmv{x}   )    \:  \leftindex^{   \bot   }{\ll}\!\! =^{  \ell_{{\mathrm{2}}}  }  \ottmv{z}  \\
\equiv &   \lambda  \ottmv{x}  .   \mathbf{up}^{   \bot   ,  \ell_{{\mathrm{1}}}  }   \ottkw{ret}  \:   \mathbf{up}^{   \bot   ,  \ell_{{\mathrm{2}}}  }  \ottmv{x}      \:  \leftindex^{   \bot   }{\ll}\!\! =^{  \ell_{{\mathrm{2}}}  }  \ottmv{z}  \\
\equiv &   \lambda  \ottmv{x}  .   \mathbf{up}^{   \bot   ,  \ell_{{\mathrm{1}}}  }   \ottkw{ret}  \:  \ottmv{x}     \:  \leftindex^{  \ell_{{\mathrm{2}}}  }{\ll}\!\! =^{  \ell_{{\mathrm{2}}}  }   \mathbf{up}^{  \ell_{{\mathrm{2}}}  ,  \ell_{{\mathrm{2}}}  }  \ottmv{z}   \hspace*{3pt} [\text{By Eqn. } (\ref{eq:cnatr})] \\
\equiv &   (   \mathbf{lift}^{  \ell_{{\mathrm{2}}}  }   (   \lambda  \ottmv{x}  .   \mathbf{up}^{   \bot   ,  \ell_{{\mathrm{1}}}  }   \ottkw{ret}  \:  \ottmv{x}     )    )   \:   (   \mathbf{fork}^{  \ell_{{\mathrm{2}}}  ,  \ell_{{\mathrm{2}}}  }  \ottmv{z}   )   \\
\equiv &   (   \lambda  \ottmv{y}  .   \mathbf{up}^{   \bot   ,  \ell_{{\mathrm{1}}}  }  \ottmv{y}    )   \:  \leftindex^{   \bot   }{\ll}\!\! =^{  \ell_{{\mathrm{2}}}  }   \mathbf{fork}^{  \ell_{{\mathrm{2}}}  ,  \ell_{{\mathrm{2}}}  }  \ottmv{z}   \hspace*{3pt} [\text{By Eqn. } (\ref{lifteq2})] \\
\equiv &   \lambda  \ottmv{y}  .  \ottmv{y}   \:  \leftindex^{  \ell_{{\mathrm{1}}}  }{\ll}\!\! =^{  \ell_{{\mathrm{2}}}  }   \mathbf{up}^{  \ell_{{\mathrm{2}}}  ,   \ell_{{\mathrm{2}}}  \vee  \ell_{{\mathrm{1}}}   }   (   \mathbf{fork}^{  \ell_{{\mathrm{2}}}  ,  \ell_{{\mathrm{2}}}  }  \ottmv{z}   )    \hspace*{3pt} [\text{By Eqn. } (\ref{eq:cnatr})] \\
\equiv &  \mathbf{fork}^{  \ell_{{\mathrm{2}}}  ,  \ell_{{\mathrm{1}}}  }   (   \mathbf{up}^{  \ell_{{\mathrm{2}}}  ,   \ell_{{\mathrm{2}}}  \vee  \ell_{{\mathrm{1}}}   }   \mathbf{fork}^{  \ell_{{\mathrm{2}}}  ,  \ell_{{\mathrm{2}}}  }  \ottmv{z}    )    \hspace*{3pt} [\text{By Eqn. } (\ref{eq:cidn})]
\end{align*}

\end{itemize}
\end{proof}


\section{Proofs of lemmas/theorems stated in Section \ref{nonint}}

\begin{lemma}\label{Ap2}
If $\ell \sqsubseteq A$ in \ED{}, then $\exists$ $ k^{ \ell }_{  \llbracket  \ottnt{A}  \rrbracket  }  \in \text{Hom}_{\Ct}( \mathbf{S}_{  \ell  }   \llbracket  \ottnt{A}  \rrbracket  , \llbracket  \ottnt{A}  \rrbracket )$ such that $   k^{ \ell }_{  \llbracket  \ottnt{A}  \rrbracket  }   \circ   \mathbf{S}^{   \bot    \sqsubseteq   \ell  }_{  \llbracket  \ottnt{A}  \rrbracket  }    \circ   \eta_{  \llbracket  \ottnt{A}  \rrbracket  }   =  \text{id}_{  \llbracket  \ottnt{A}  \rrbracket  } $.
\end{lemma}

\begin{proof}
By induction on $ \ell  \sqsubseteq  \ottnt{A} $.
\begin{itemize}

\item \Rref{Prot-Prod}. Have: $ \ell  \sqsubseteq   \ottnt{A}  \times  \ottnt{B}  $ where $ \ell  \sqsubseteq  \ottnt{A} $ and $ \ell  \sqsubseteq  \ottnt{B} $.\\
By IH, $\exists  k^{ \ell }_{  \llbracket  \ottnt{A}  \rrbracket  }  \in \text{Hom}_{\Ct} ( \mathbf{S}_{  \ell  }   \llbracket  \ottnt{A}  \rrbracket  ,  \llbracket  \ottnt{A}  \rrbracket )$ and $ k^{ \ell }_{  \llbracket  \ottnt{B}  \rrbracket  }  \in \text{Hom}_{\Ct} ( \mathbf{S}_{  \ell  }   \llbracket  \ottnt{B}  \rrbracket  ,  \llbracket  \ottnt{B}  \rrbracket )$ such that\\
$   k^{ \ell }_{  \llbracket  \ottnt{A}  \rrbracket  }   \circ   \mathbf{S}^{   \bot    \sqsubseteq   \ell  }_{  \llbracket  \ottnt{A}  \rrbracket  }    \circ   \eta_{  \llbracket  \ottnt{A}  \rrbracket  }   =  \text{id}_{  \llbracket  \ottnt{A}  \rrbracket  } $ and $   k^{ \ell }_{  \llbracket  \ottnt{B}  \rrbracket  }   \circ   \mathbf{S}^{   \bot    \sqsubseteq   \ell  }_{  \llbracket  \ottnt{B}  \rrbracket  }    \circ   \eta_{  \llbracket  \ottnt{B}  \rrbracket  }   =  \text{id}_{  \llbracket  \ottnt{B}  \rrbracket  } $.\\
Define: $ k^{ \ell }_{    \llbracket  \ottnt{A}  \rrbracket   \times   \llbracket  \ottnt{B}  \rrbracket    }  =  \mathbf{S}_{  \ell  }   (    \llbracket  \ottnt{A}  \rrbracket   \times   \llbracket  \ottnt{B}  \rrbracket    )   \xrightarrow{ \langle   \mathbf{S}_{  \ell  }   \pi_1    ,   \mathbf{S}_{  \ell  }   \pi_2    \rangle }   \mathbf{S}_{  \ell  }   \llbracket  \ottnt{A}  \rrbracket    \times   \mathbf{S}_{  \ell  }   \llbracket  \ottnt{B}  \rrbracket    \xrightarrow{  k^{ \ell }_{  \llbracket  \ottnt{A}  \rrbracket  }   \times   k^{ \ell }_{  \llbracket  \ottnt{B}  \rrbracket  }  }   \llbracket  \ottnt{A}  \rrbracket   \times   \llbracket  \ottnt{B}  \rrbracket  $.\\
Need to show: $   k^{ \ell }_{    \llbracket  \ottnt{A}  \rrbracket   \times   \llbracket  \ottnt{B}  \rrbracket    }   \circ   \mathbf{S}^{   \bot    \sqsubseteq   \ell  }_{    \llbracket  \ottnt{A}  \rrbracket   \times   \llbracket  \ottnt{B}  \rrbracket    }    \circ   \eta_{    \llbracket  \ottnt{A}  \rrbracket   \times   \llbracket  \ottnt{B}  \rrbracket    }   =  \text{id}_{    \llbracket  \ottnt{A}  \rrbracket   \times   \llbracket  \ottnt{B}  \rrbracket    } $.\\

This equation follows from the commutative diagram in Figure \ref{prfcd5}. The diagram in Figure \ref{prfcd5} commutes: the triangle commutes by naturality; the square too commutes by naturality; the circular segment commutes by IH.  

\begin{figure}[h]
\begin{tikzcd}[row sep = 3.5 em, column sep = 4 em]
  \llbracket  \ottnt{A}  \rrbracket   \times   \llbracket  \ottnt{B}  \rrbracket   \arrow{r}{ \eta_{    \llbracket  \ottnt{A}  \rrbracket   \times   \llbracket  \ottnt{B}  \rrbracket    } } \arrow{dr}{  \eta_{  \llbracket  \ottnt{A}  \rrbracket  }   \times   \eta_{  \llbracket  \ottnt{B}  \rrbracket  }  } \ar[ddrr," \text{id}_{    \llbracket  \ottnt{A}  \rrbracket   \times   \llbracket  \ottnt{B}  \rrbracket    } ", bend right=30] &  \mathbf{S}_{   \bot   }   (    \llbracket  \ottnt{A}  \rrbracket   \times   \llbracket  \ottnt{B}  \rrbracket    )   \arrow{r}{ \mathbf{S}^{   \bot    \sqsubseteq   \ell  }_{    \llbracket  \ottnt{A}  \rrbracket   \times   \llbracket  \ottnt{B}  \rrbracket    } } \arrow{d}{ \langle   \mathbf{S}_{   \bot   }   \pi_1    ,   \mathbf{S}_{   \bot   }   \pi_2    \rangle } &  \mathbf{S}_{  \ell  }   (    \llbracket  \ottnt{A}  \rrbracket   \times   \llbracket  \ottnt{B}  \rrbracket    )   \arrow{d}{ \langle   \mathbf{S}_{  \ell  }   \pi_1    ,   \mathbf{S}_{  \ell  }   \pi_2    \rangle } \\
&   \mathbf{S}_{   \bot   }   \llbracket  \ottnt{A}  \rrbracket    \times   \mathbf{S}_{   \bot   }   \llbracket  \ottnt{B}  \rrbracket    \arrow{r}{   \mathbf{S}^{   \bot    \sqsubseteq   \ell  }_{  \llbracket  \ottnt{A}  \rrbracket  }    \times   \mathbf{S}^{   \bot    \sqsubseteq   \ell  }_{  \llbracket  \ottnt{B}  \rrbracket  }  } &   \mathbf{S}_{  \ell  }   \llbracket  \ottnt{A}  \rrbracket    \times   \mathbf{S}_{  \ell  }   \llbracket  \ottnt{B}  \rrbracket    \arrow{d}{  k^{ \ell }_{  \llbracket  \ottnt{A}  \rrbracket  }   \times   k^{ \ell }_{  \llbracket  \ottnt{B}  \rrbracket  }  }\\ 
& &   \llbracket  \ottnt{A}  \rrbracket   \times   \llbracket  \ottnt{B}  \rrbracket  
\end{tikzcd}
\caption{Commutative diagram}
\label{prfcd5}
\end{figure}
 
\item \Rref{Prot-Fun}. Have: $ \ell  \sqsubseteq   \ottnt{A}  \to  \ottnt{B}  $ where $ \ell  \sqsubseteq  \ottnt{B} $.\\
By IH, $\exists  k^{ \ell }_{  \llbracket  \ottnt{B}  \rrbracket  }  \in \text{Hom}_{\Ct}( \mathbf{S}_{  \ell  }   \llbracket  \ottnt{B}  \rrbracket  ,  \llbracket  \ottnt{B}  \rrbracket )$ such that $   k^{ \ell }_{  \llbracket  \ottnt{B}  \rrbracket  }   \circ   \mathbf{S}^{   \bot    \sqsubseteq   \ell  }_{  \llbracket  \ottnt{B}  \rrbracket  }    \circ   \eta_{  \llbracket  \ottnt{B}  \rrbracket  }   =  \text{id}_{  \llbracket  \ottnt{B}  \rrbracket  } $.\\
Define: $ k^{ \ell }_{    \llbracket  \ottnt{B}  \rrbracket  ^{  \llbracket  \ottnt{A}  \rrbracket  }   }  = \Lambda \Big(   \mathbf{S}_{  \ell  }     \llbracket  \ottnt{B}  \rrbracket  ^{  \llbracket  \ottnt{A}  \rrbracket  }     \times   \llbracket  \ottnt{A}  \rrbracket   \xrightarrow{ \langle   \pi_2   ,   \pi_1   \rangle }   \llbracket  \ottnt{A}  \rrbracket   \times   \mathbf{S}_{  \ell  }     \llbracket  \ottnt{B}  \rrbracket  ^{  \llbracket  \ottnt{A}  \rrbracket  }     \xrightarrow{ t^{\mathbf{S}_{  \ell  } }_{  \llbracket  \ottnt{A}  \rrbracket  ,     \llbracket  \ottnt{B}  \rrbracket  ^{  \llbracket  \ottnt{A}  \rrbracket  }   } }  \mathbf{S}_{  \ell  }   (    \llbracket  \ottnt{A}  \rrbracket   \times     \llbracket  \ottnt{B}  \rrbracket  ^{  \llbracket  \ottnt{A}  \rrbracket  }     )   \xrightarrow{ \mathbf{S}_{  \ell  }   \langle   \pi_2   ,   \pi_1   \rangle  }  \mathbf{S}_{  \ell  }   (      \llbracket  \ottnt{B}  \rrbracket  ^{  \llbracket  \ottnt{A}  \rrbracket  }    \times   \llbracket  \ottnt{A}  \rrbracket    )   \xrightarrow{ \mathbf{S}_{  \ell  }   \text{app}  }  \mathbf{S}_{  \ell  }   \llbracket  \ottnt{B}  \rrbracket   \xrightarrow{ k^{ \ell }_{  \llbracket  \ottnt{B}  \rrbracket  } }  \llbracket  \ottnt{B}  \rrbracket  \Big)$.\\

Need to show: $   k^{ \ell }_{    \llbracket  \ottnt{B}  \rrbracket  ^{  \llbracket  \ottnt{A}  \rrbracket  }   }   \circ   \mathbf{S}^{   \bot    \sqsubseteq   \ell  }_{    \llbracket  \ottnt{B}  \rrbracket  ^{  \llbracket  \ottnt{A}  \rrbracket  }   }    \circ   \eta_{    \llbracket  \ottnt{B}  \rrbracket  ^{  \llbracket  \ottnt{A}  \rrbracket  }   }   =  \text{id}_{    \llbracket  \ottnt{B}  \rrbracket  ^{  \llbracket  \ottnt{A}  \rrbracket  }   } $.\\

Now, \begin{align*}
&    k^{ \ell }_{    \llbracket  \ottnt{B}  \rrbracket  ^{  \llbracket  \ottnt{A}  \rrbracket  }   }   \circ   \mathbf{S}^{   \bot    \sqsubseteq   \ell  }_{    \llbracket  \ottnt{B}  \rrbracket  ^{  \llbracket  \ottnt{A}  \rrbracket  }   }    \circ   \eta_{    \llbracket  \ottnt{B}  \rrbracket  ^{  \llbracket  \ottnt{A}  \rrbracket  }   }  \\
= &    \Lambda \Big(        k^{ \ell }_{  \llbracket  \ottnt{B}  \rrbracket  }   \circ   \mathbf{S}_{  \ell  }   \text{app}     \circ   \mathbf{S}_{  \ell  }   \langle   \pi_2   ,   \pi_1   \rangle     \circ   t^{\mathbf{S}_{  \ell  } }_{  \llbracket  \ottnt{A}  \rrbracket  ,     \llbracket  \ottnt{B}  \rrbracket  ^{  \llbracket  \ottnt{A}  \rrbracket  }   }    \circ   \langle   \pi_2   ,   \pi_1   \rangle     \Big)   \circ   \mathbf{S}^{   \bot    \sqsubseteq   \ell  }_{    \llbracket  \ottnt{B}  \rrbracket  ^{  \llbracket  \ottnt{A}  \rrbracket  }   }    \circ   \eta_{    \llbracket  \ottnt{B}  \rrbracket  ^{  \llbracket  \ottnt{A}  \rrbracket  }   }   \\
= &  \Lambda \Big(          k^{ \ell }_{  \llbracket  \ottnt{B}  \rrbracket  }   \circ   \mathbf{S}_{  \ell  }   \text{app}     \circ   \mathbf{S}_{  \ell  }   \langle   \pi_2   ,   \pi_1   \rangle     \circ   t^{\mathbf{S}_{  \ell  } }_{  \llbracket  \ottnt{A}  \rrbracket  ,     \llbracket  \ottnt{B}  \rrbracket  ^{  \llbracket  \ottnt{A}  \rrbracket  }   }    \circ   \langle   \pi_2   ,   \pi_1   \rangle    \circ   (    \mathbf{S}^{   \bot    \sqsubseteq   \ell  }_{    \llbracket  \ottnt{B}  \rrbracket  ^{  \llbracket  \ottnt{A}  \rrbracket  }   }   \circ   \eta_{    \llbracket  \ottnt{B}  \rrbracket  ^{  \llbracket  \ottnt{A}  \rrbracket  }   }    )    \times   \text{id}_{  \llbracket  \ottnt{A}  \rrbracket  }     \Big)  \\
= &  \Lambda \Big(         k^{ \ell }_{  \llbracket  \ottnt{B}  \rrbracket  }   \circ   \mathbf{S}^{   \bot    \sqsubseteq   \ell  }_{  \llbracket  \ottnt{B}  \rrbracket  }    \circ   \eta_{  \llbracket  \ottnt{B}  \rrbracket  }    \circ   \text{app}    \circ   \langle   \pi_2   ,   \pi_1   \rangle    \circ   \langle   \pi_2   ,   \pi_1   \rangle     \Big)  \hspace*{3pt} [\text{By Figure \ref{prfcd7}}] \\
= &  \Lambda   (   \text{app}   )   \hspace*{3pt} [\text{By IH}]\\
= &  \text{id}_{    \llbracket  \ottnt{B}  \rrbracket  ^{  \llbracket  \ottnt{A}  \rrbracket  }   } 
\end{align*} 

Note that in Figure \ref{prfcd7}, $ \overline{\eta}_{ \ottnt{X} }  :=   \mathbf{S}^{   \bot    \sqsubseteq   \ell  }_{ \ottnt{X} }   \circ   \eta_{ \ottnt{X} }  $. The diagram in this figure commutes by naturality. 

\begin{figure}
\hspace*{-0.5cm}
\begin{tikzcd}[row sep = 3 em]
    \llbracket  \ottnt{B}  \rrbracket  ^{  \llbracket  \ottnt{A}  \rrbracket  }    \times   \llbracket  \ottnt{A}  \rrbracket   \arrow{d}{  \overline{\eta}_{    \llbracket  \ottnt{B}  \rrbracket  ^{  \llbracket  \ottnt{A}  \rrbracket  }   }   \times   \text{id}  } \arrow{r}{ \langle   \pi_2   ,   \pi_1   \rangle } &   \llbracket  \ottnt{A}  \rrbracket   \times     \llbracket  \ottnt{B}  \rrbracket  ^{  \llbracket  \ottnt{A}  \rrbracket  }    \arrow{d}{  \text{id}_{  \llbracket  \ottnt{A}  \rrbracket  }   \times   \overline{\eta}_{    \llbracket  \ottnt{B}  \rrbracket  ^{  \llbracket  \ottnt{A}  \rrbracket  }   }  } \arrow{r}{ \text{id} } &   \llbracket  \ottnt{A}  \rrbracket   \times     \llbracket  \ottnt{B}  \rrbracket  ^{  \llbracket  \ottnt{A}  \rrbracket  }    \arrow{d}{ \overline{\eta}_{    \llbracket  \ottnt{A}  \rrbracket   \times     \llbracket  \ottnt{B}  \rrbracket  ^{  \llbracket  \ottnt{A}  \rrbracket  }     } } \arrow{r}{ \langle   \pi_2   ,   \pi_1   \rangle } &     \llbracket  \ottnt{B}  \rrbracket  ^{  \llbracket  \ottnt{A}  \rrbracket  }    \times   \llbracket  \ottnt{A}  \rrbracket   \arrow{d}{ \overline{\eta}_{      \llbracket  \ottnt{B}  \rrbracket  ^{  \llbracket  \ottnt{A}  \rrbracket  }    \times   \llbracket  \ottnt{A}  \rrbracket    } } \arrow{r}{ \text{app} } &  \llbracket  \ottnt{B}  \rrbracket  \arrow{d}{ \overline{\eta}_{  \llbracket  \ottnt{B}  \rrbracket  } } \\
  \mathbf{S}_{  \ell  }     \llbracket  \ottnt{B}  \rrbracket  ^{  \llbracket  \ottnt{A}  \rrbracket  }     \times   \llbracket  \ottnt{A}  \rrbracket   \arrow{r}{ \langle   \pi_2   ,   \pi_1   \rangle } &   \llbracket  \ottnt{A}  \rrbracket   \times   \mathbf{S}_{  \ell  }     \llbracket  \ottnt{B}  \rrbracket  ^{  \llbracket  \ottnt{A}  \rrbracket  }     \arrow{r}{ t^{\mathbf{S}_{  \ell  } }_{  \llbracket  \ottnt{A}  \rrbracket  ,     \llbracket  \ottnt{B}  \rrbracket  ^{  \llbracket  \ottnt{A}  \rrbracket  }   } } &  \mathbf{S}_{  \ell  }   (    \llbracket  \ottnt{A}  \rrbracket   \times     \llbracket  \ottnt{B}  \rrbracket  ^{  \llbracket  \ottnt{A}  \rrbracket  }     )   \arrow{r}{ \mathbf{S}_{  \ell  }   \langle   \pi_2   ,   \pi_1   \rangle  } &  \mathbf{S}_{  \ell  }   (      \llbracket  \ottnt{B}  \rrbracket  ^{  \llbracket  \ottnt{A}  \rrbracket  }    \times   \llbracket  \ottnt{A}  \rrbracket    )   \arrow{r}{ \mathbf{S}_{  \ell  }   \text{app}  } &  \mathbf{S}_{  \ell  }   \llbracket  \ottnt{B}  \rrbracket  
\end{tikzcd}
\caption{Commutative diagram}
\label{prfcd7}
\end{figure}

\item \Rref{Prot-Monad}. Have $ \ell_{{\mathrm{1}}}  \sqsubseteq   \mathcal{T}_{ \ell_{{\mathrm{2}}} } \:  \ottnt{A}  $ where $ \ell_{{\mathrm{1}}}  \sqsubseteq  \ell_{{\mathrm{2}}} $.\\
Define: $ k^{ \ell_{{\mathrm{1}}} }_{   \mathbf{S}_{  \ell_{{\mathrm{2}}}  }   \llbracket  \ottnt{A}  \rrbracket    }  =  \mu^{  \ell_{{\mathrm{1}}}  ,  \ell_{{\mathrm{2}}}  }_{  \llbracket  \ottnt{A}  \rrbracket  } $.\\
Need to show: $   k^{ \ell_{{\mathrm{1}}} }_{   \mathbf{S}_{  \ell_{{\mathrm{2}}}  }   \llbracket  \ottnt{A}  \rrbracket    }   \circ   \mathbf{S}^{   \bot    \sqsubseteq   \ell_{{\mathrm{1}}}  }_{   \mathbf{S}_{  \ell_{{\mathrm{2}}}  }   \llbracket  \ottnt{A}  \rrbracket    }    \circ   \eta_{   \mathbf{S}_{  \ell_{{\mathrm{2}}}  }   \llbracket  \ottnt{A}  \rrbracket    }   =  \text{id}_{   \mathbf{S}_{  \ell_{{\mathrm{2}}}  }   \llbracket  \ottnt{A}  \rrbracket    } $.\\
Now,
\begin{align*}
&    k^{ \ell_{{\mathrm{1}}} }_{   \mathbf{S}_{  \ell_{{\mathrm{2}}}  }   \llbracket  \ottnt{A}  \rrbracket    }   \circ   \mathbf{S}^{   \bot    \sqsubseteq   \ell_{{\mathrm{1}}}  }_{   \mathbf{S}_{  \ell_{{\mathrm{2}}}  }   \llbracket  \ottnt{A}  \rrbracket    }    \circ   \eta_{   \mathbf{S}_{  \ell_{{\mathrm{2}}}  }   \llbracket  \ottnt{A}  \rrbracket    }   \\
= &    \mu^{  \ell_{{\mathrm{1}}}  ,  \ell_{{\mathrm{2}}}  }_{  \llbracket  \ottnt{A}  \rrbracket  }   \circ   \mathbf{S}^{   \bot    \sqsubseteq   \ell_{{\mathrm{1}}}  }_{   \mathbf{S}_{  \ell_{{\mathrm{2}}}  }   \llbracket  \ottnt{A}  \rrbracket    }    \circ   \eta_{   \mathbf{S}_{  \ell_{{\mathrm{2}}}  }   \llbracket  \ottnt{A}  \rrbracket    }   \\
= &      \mu^{  \ell_{{\mathrm{1}}}  ,  \ell_{{\mathrm{2}}}  }_{  \llbracket  \ottnt{A}  \rrbracket  }   \circ   \mathbf{S}^{   \bot    \sqsubseteq   \ell_{{\mathrm{1}}}  }_{   \mathbf{S}_{  \ell_{{\mathrm{2}}}  }   \llbracket  \ottnt{A}  \rrbracket    }    \circ   \delta^{   \bot   ,  \ell_{{\mathrm{2}}}  }_{  \llbracket  \ottnt{A}  \rrbracket  }    \circ   \mu^{   \bot   ,  \ell_{{\mathrm{2}}}  }_{  \llbracket  \ottnt{A}  \rrbracket  }    \circ   \eta_{   \mathbf{S}_{  \ell_{{\mathrm{2}}}  }   \llbracket  \ottnt{A}  \rrbracket    }   \\
= &     \mu^{  \ell_{{\mathrm{1}}}  ,  \ell_{{\mathrm{2}}}  }_{  \llbracket  \ottnt{A}  \rrbracket  }   \circ   \mathbf{S}^{   \bot    \sqsubseteq   \ell_{{\mathrm{1}}}  }_{   \mathbf{S}_{  \ell_{{\mathrm{2}}}  }   \llbracket  \ottnt{A}  \rrbracket    }    \circ   \delta^{   \bot   ,  \ell_{{\mathrm{2}}}  }_{  \llbracket  \ottnt{A}  \rrbracket  }    \circ   \text{id}_{   \mathbf{S}_{  \ell_{{\mathrm{2}}}  }   \llbracket  \ottnt{A}  \rrbracket    }   \hspace*{3pt} [\text{By lax monoidality}]\\
= &    \mu^{  \ell_{{\mathrm{1}}}  ,  \ell_{{\mathrm{2}}}  }_{  \llbracket  \ottnt{A}  \rrbracket  }   \circ   \delta^{  \ell_{{\mathrm{1}}}  ,  \ell_{{\mathrm{2}}}  }_{  \llbracket  \ottnt{A}  \rrbracket  }    \circ   \mathbf{S}^{  \ell_{{\mathrm{2}}}   \sqsubseteq    \ell_{{\mathrm{1}}}  \vee  \ell_{{\mathrm{2}}}   }_{  \llbracket  \ottnt{A}  \rrbracket  }   \hspace*{3pt} [\text{By naturality}] \\
= &  \mathbf{S}^{  \ell_{{\mathrm{2}}}   \sqsubseteq    \ell_{{\mathrm{1}}}  \vee  \ell_{{\mathrm{2}}}   }_{  \llbracket  \ottnt{A}  \rrbracket  }  =  \mathbf{S}^{  \ell_{{\mathrm{2}}}   \sqsubseteq   \ell_{{\mathrm{2}}}  }_{  \llbracket  \ottnt{A}  \rrbracket  }  [\because  \ell_{{\mathrm{1}}}  \sqsubseteq  \ell_{{\mathrm{2}}}  ] =  \text{id}_{   \mathbf{S}_{  \ell_{{\mathrm{2}}}  }   \llbracket  \ottnt{A}  \rrbracket    } 
\end{align*}

\item \Rref{Prot-Already}. Have: $ \ell  \sqsubseteq   \mathcal{T}_{ \ell' } \:  \ottnt{A}  $ where $ \ell  \sqsubseteq  \ottnt{A} $.\\
By IH, $\exists  k^{ \ell }_{  \llbracket  \ottnt{A}  \rrbracket  }  \in \text{Hom}_{\Ct} ( \mathbf{S}_{  \ell  }   \llbracket  \ottnt{A}  \rrbracket   ,  \llbracket  \ottnt{A}  \rrbracket )$ such that 
$   k^{ \ell }_{  \llbracket  \ottnt{A}  \rrbracket  }   \circ   \mathbf{S}^{   \bot    \sqsubseteq   \ell  }_{  \llbracket  \ottnt{A}  \rrbracket  }    \circ   \eta_{  \llbracket  \ottnt{A}  \rrbracket  }   =  \text{id}_{  \llbracket  \ottnt{A}  \rrbracket  } $. \\
Define: $ k^{ \ell }_{   \mathbf{S}_{  \ell'  }   \llbracket  \ottnt{A}  \rrbracket    }  =  \mathbf{S}_{  \ell  }   \mathbf{S}_{  \ell'  }   \llbracket  \ottnt{A}  \rrbracket    \xrightarrow{ \mu^{  \ell  ,  \ell'  }_{  \llbracket  \ottnt{A}  \rrbracket  } }  \mathbf{S}_{   \ell  \vee  \ell'   }   \llbracket  \ottnt{A}  \rrbracket   =  \mathbf{S}_{   \ell'  \vee  \ell   }   \llbracket  \ottnt{A}  \rrbracket   \xrightarrow{ \delta^{  \ell'  ,  \ell  }_{  \llbracket  \ottnt{A}  \rrbracket  } }  \mathbf{S}_{  \ell'  }   \mathbf{S}_{  \ell  }   \llbracket  \ottnt{A}  \rrbracket    \xrightarrow{ \mathbf{S}_{  \ell'  }   k^{ \ell }_{  \llbracket  \ottnt{A}  \rrbracket  }  }  \mathbf{S}_{  \ell'  }   \llbracket  \ottnt{A}  \rrbracket  $.\\
Need to show: $   k^{ \ell }_{  \mathbf{S}_{  \ell'  }   \llbracket  \ottnt{A}  \rrbracket   }   \circ   \mathbf{S}^{   \bot    \sqsubseteq   \ell  }_{   \mathbf{S}_{  \ell'  }   \llbracket  \ottnt{A}  \rrbracket    }    \circ   \eta_{  \mathbf{S}_{  \ell'  }   \llbracket  \ottnt{A}  \rrbracket   }   =  \text{id}_{  \mathbf{S}_{  \ell'  }   \llbracket  \ottnt{A}  \rrbracket   } $.\\

This equation follows from the commutative diagram in Figure \ref{prfcd4}. The diagram in Figure \ref{prfcd4} commutes: the left triangle in the top row commutes because $\mathbf{S}$ is a lax monoidal functor; the circular segment in the top row commutes because $\mathbf{S}$ is an oplax monoidal functor; the left square in the bottom row commutes because $\mu$ is natural in its first component; the right square in the bottom row commutes because $\delta$ is natural in its second component; the triangle to the right in the bottom row commutes by IH. 

\begin{figure}
\begin{tikzcd}[row sep = 3.5 em, column sep = 3.5 em]
&  \mathbf{S}_{  \ell'  }   \llbracket  \ottnt{A}  \rrbracket   \arrow{dl}[above left]{ \eta_{  \mathbf{S}_{  \ell'  }   \llbracket  \ottnt{A}  \rrbracket   } } \arrow{d}{ \text{id}_{  \mathbf{S}_{  \ell'  }   \llbracket  \ottnt{A}  \rrbracket   } } & & \\
 \mathbf{S}_{   \bot   }   \mathbf{S}_{  \ell'  }   \llbracket  \ottnt{A}  \rrbracket    \arrow{r}{ \mu^{   \bot   ,  \ell'  }_{  \llbracket  \ottnt{A}  \rrbracket  } } \arrow{d}{ \mathbf{S}^{   \bot    \sqsubseteq   \ell  }_{   \mathbf{S}_{  \ell'  }   \llbracket  \ottnt{A}  \rrbracket    } } &  \mathbf{S}_{  \ell'  }   \llbracket  \ottnt{A}  \rrbracket   \arrow{r}{ \delta^{  \ell'  ,   \bot   }_{  \llbracket  \ottnt{A}  \rrbracket  } } \arrow{d}{ \mathbf{S}^{  \ell'   \sqsubseteq    \ell  \vee  \ell'   }_{  \llbracket  \ottnt{A}  \rrbracket  } } \ar[rr," \text{id}_{  \mathbf{S}_{  \ell'  }   \llbracket  \ottnt{A}  \rrbracket   } ",bend left=35] &  \mathbf{S}_{  \ell'  }   \mathbf{S}_{   \bot   }   \llbracket  \ottnt{A}  \rrbracket    \arrow{r}{ \mathbf{S}_{  \ell'  }   \epsilon_{  \llbracket  \ottnt{A}  \rrbracket  }  } \arrow{d}{ \mathbf{S}_{  \ell'  }   \mathbf{S}^{   \bot    \sqsubseteq   \ell  }_{  \llbracket  \ottnt{A}  \rrbracket  }  } &  \mathbf{S}_{  \ell'  }   \llbracket  \ottnt{A}  \rrbracket   \\
 \mathbf{S}_{  \ell  }   \mathbf{S}_{  \ell'  }   \llbracket  \ottnt{A}  \rrbracket    \arrow{r}{ \mu^{  \ell  ,  \ell'  }_{  \llbracket  \ottnt{A}  \rrbracket  } } &  \mathbf{S}_{   \ell  \vee  \ell'   }   \llbracket  \ottnt{A}  \rrbracket   \arrow{r}{ \delta^{  \ell'  ,  \ell  }_{  \llbracket  \ottnt{A}  \rrbracket  } } &  \mathbf{S}_{  \ell'  }   \mathbf{S}_{  \ell  }   \llbracket  \ottnt{A}  \rrbracket    \arrow{ur}[below right]{ \mathbf{S}_{  \ell'  }   k^{ \ell }_{  \llbracket  \ottnt{A}  \rrbracket  }  } &
\end{tikzcd}
\caption{Commutative diagram}
\label{prfcd4}
\end{figure}
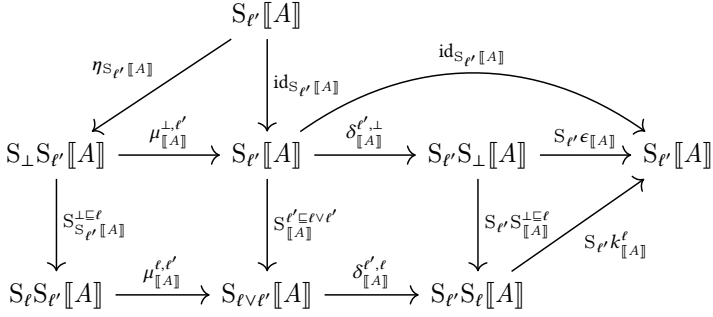

\item \Rref{Prot-Minimum}. Have: $  \bot   \sqsubseteq  \ottnt{A} $.\\
Define: $ k^{  \bot  }_{  \llbracket  \ottnt{A}  \rrbracket  }  =  \epsilon_{  \llbracket  \ottnt{A}  \rrbracket  } $. Note, $   \epsilon_{  \llbracket  \ottnt{A}  \rrbracket  }   \circ   \mathbf{S}^{   \bot    \sqsubseteq    \bot   }_{  \llbracket  \ottnt{A}  \rrbracket  }    \circ   \eta_{  \llbracket  \ottnt{A}  \rrbracket  }   =  \text{id}_{  \llbracket  \ottnt{A}  \rrbracket  } $.

\item \Rref{Prot-Combine}. Have: $ \ell  \sqsubseteq  \ottnt{A} $ where $ \ell_{{\mathrm{1}}}  \sqsubseteq  \ottnt{A} $ and $ \ell_{{\mathrm{2}}}  \sqsubseteq  \ottnt{A} $ and $ \ell  \sqsubseteq   \ell_{{\mathrm{1}}}  \vee  \ell_{{\mathrm{2}}}  $.\\
By IH, $\exists  k^{ \ell_{{\mathrm{1}}} }_{  \llbracket  \ottnt{A}  \rrbracket  }  \in \text{Hom}_{\Ct} ( \mathbf{S}_{  \ell_{{\mathrm{1}}}  }   \llbracket  \ottnt{A}  \rrbracket  ,  \llbracket  \ottnt{A}  \rrbracket )$ and $ k^{ \ell_{{\mathrm{2}}} }_{  \llbracket  \ottnt{A}  \rrbracket  }  \in \text{Hom}_{\Ct} ( \mathbf{S}_{  \ell_{{\mathrm{2}}}  }   \llbracket  \ottnt{A}  \rrbracket  ,  \llbracket  \ottnt{A}  \rrbracket )$ such that:
\begin{align*}
   k^{ \ell_{{\mathrm{1}}} }_{  \llbracket  \ottnt{A}  \rrbracket  }   \circ   \mathbf{S}^{   \bot    \sqsubseteq   \ell_{{\mathrm{1}}}  }_{  \llbracket  \ottnt{A}  \rrbracket  }    \circ   \eta_{  \llbracket  \ottnt{A}  \rrbracket  }   & =  \text{id}_{  \llbracket  \ottnt{A}  \rrbracket  }  \\
   k^{ \ell_{{\mathrm{2}}} }_{  \llbracket  \ottnt{A}  \rrbracket  }   \circ   \mathbf{S}^{   \bot    \sqsubseteq   \ell_{{\mathrm{2}}}  }_{  \llbracket  \ottnt{A}  \rrbracket  }    \circ   \eta_{  \llbracket  \ottnt{A}  \rrbracket  }   & =  \text{id}_{  \llbracket  \ottnt{A}  \rrbracket  } . 
\end{align*}
Define: $ k^{ \ell }_{  \llbracket  \ottnt{A}  \rrbracket  }  =  \mathbf{S}_{  \ell  }   \llbracket  \ottnt{A}  \rrbracket   \xrightarrow{ \mathbf{S}^{  \ell   \sqsubseteq    \ell_{{\mathrm{1}}}  \vee  \ell_{{\mathrm{2}}}   }_{  \llbracket  \ottnt{A}  \rrbracket  } }  \mathbf{S}_{   \ell_{{\mathrm{1}}}  \vee  \ell_{{\mathrm{2}}}   }   \llbracket  \ottnt{A}  \rrbracket   \xrightarrow{ \delta^{  \ell_{{\mathrm{1}}}  ,  \ell_{{\mathrm{2}}}  }_{  \llbracket  \ottnt{A}  \rrbracket  } }  \mathbf{S}_{  \ell_{{\mathrm{1}}}  }   \mathbf{S}_{  \ell_{{\mathrm{2}}}  }   \llbracket  \ottnt{A}  \rrbracket    \xrightarrow{ \mathbf{S}_{  \ell_{{\mathrm{1}}}  }   k^{ \ell_{{\mathrm{2}}} }_{  \llbracket  \ottnt{A}  \rrbracket  }  }  \mathbf{S}_{  \ell_{{\mathrm{1}}}  }   \llbracket  \ottnt{A}  \rrbracket   \xrightarrow{ k^{ \ell_{{\mathrm{1}}} }_{  \llbracket  \ottnt{A}  \rrbracket  } }  \llbracket  \ottnt{A}  \rrbracket $.\\
Now,
\begin{align*}
&    k^{ \ell }_{  \llbracket  \ottnt{A}  \rrbracket  }   \circ   \mathbf{S}^{   \bot    \sqsubseteq   \ell  }_{  \llbracket  \ottnt{A}  \rrbracket  }    \circ   \eta_{  \llbracket  \ottnt{A}  \rrbracket  }   \\
= &       k^{ \ell_{{\mathrm{1}}} }_{  \llbracket  \ottnt{A}  \rrbracket  }   \circ   \mathbf{S}_{  \ell_{{\mathrm{1}}}  }   k^{ \ell_{{\mathrm{2}}} }_{  \llbracket  \ottnt{A}  \rrbracket  }     \circ   \delta^{  \ell_{{\mathrm{1}}}  ,  \ell_{{\mathrm{2}}}  }_{  \llbracket  \ottnt{A}  \rrbracket  }    \circ   \mathbf{S}^{  \ell   \sqsubseteq    \ell_{{\mathrm{1}}}  \vee  \ell_{{\mathrm{2}}}   }_{  \llbracket  \ottnt{A}  \rrbracket  }    \circ   \mathbf{S}^{   \bot    \sqsubseteq   \ell  }_{  \llbracket  \ottnt{A}  \rrbracket  }    \circ   \eta_{  \llbracket  \ottnt{A}  \rrbracket  }   \\
= &      k^{ \ell_{{\mathrm{1}}} }_{  \llbracket  \ottnt{A}  \rrbracket  }   \circ   \mathbf{S}_{  \ell_{{\mathrm{1}}}  }   k^{ \ell_{{\mathrm{2}}} }_{  \llbracket  \ottnt{A}  \rrbracket  }     \circ   \delta^{  \ell_{{\mathrm{1}}}  ,  \ell_{{\mathrm{2}}}  }_{  \llbracket  \ottnt{A}  \rrbracket  }    \circ   \mathbf{S}^{   \bot    \sqsubseteq    \ell_{{\mathrm{1}}}  \vee  \ell_{{\mathrm{2}}}   }_{  \llbracket  \ottnt{A}  \rrbracket  }    \circ   \eta_{  \llbracket  \ottnt{A}  \rrbracket  }   \hspace*{3pt} [\because \mathbf{S} \text{ is a functor}] \\
= &       k^{ \ell_{{\mathrm{1}}} }_{  \llbracket  \ottnt{A}  \rrbracket  }   \circ   \mathbf{S}_{  \ell_{{\mathrm{1}}}  }   k^{ \ell_{{\mathrm{2}}} }_{  \llbracket  \ottnt{A}  \rrbracket  }     \circ   \delta^{  \ell_{{\mathrm{1}}}  ,  \ell_{{\mathrm{2}}}  }_{  \llbracket  \ottnt{A}  \rrbracket  }    \circ   \mathbf{S}^{  \ell_{{\mathrm{1}}}   \sqsubseteq    \ell_{{\mathrm{1}}}  \vee  \ell_{{\mathrm{2}}}   }_{  \llbracket  \ottnt{A}  \rrbracket  }    \circ   \mathbf{S}^{   \bot    \sqsubseteq   \ell_{{\mathrm{1}}}  }_{  \llbracket  \ottnt{A}  \rrbracket  }    \circ   \eta_{  \llbracket  \ottnt{A}  \rrbracket  }   \hspace*{3pt} [\because \mathbf{S} \text{ is a functor}] \\
= &     k^{ \ell_{{\mathrm{1}}} }_{  \llbracket  \ottnt{A}  \rrbracket  }   \circ   \text{id}_{  \mathbf{S}_{  \ell_{{\mathrm{1}}}  }   \llbracket  \ottnt{A}  \rrbracket   }    \circ   \mathbf{S}^{   \bot    \sqsubseteq   \ell_{{\mathrm{1}}}  }_{  \llbracket  \ottnt{A}  \rrbracket  }    \circ   \eta_{  \llbracket  \ottnt{A}  \rrbracket  }   \hspace*{3pt} [ \text{By commutative diagram in Figure \ref{prfcd6}} ]\\
= &  \text{id}_{  \llbracket  \ottnt{A}  \rrbracket  }  \hspace*{3pt} [ \text{By IH} ]
\end{align*}
 
\begin{figure}
\begin{tikzcd}[row sep = 3.5 em, column sep = 4 em]
 \mathbf{S}_{  \ell_{{\mathrm{1}}}  }   \llbracket  \ottnt{A}  \rrbracket   \arrow{r}{ \delta^{  \ell_{{\mathrm{1}}}  ,   \bot   }_{  \llbracket  \ottnt{A}  \rrbracket  } } \arrow{d}{ \mathbf{S}^{  \ell_{{\mathrm{1}}}   \sqsubseteq    \ell_{{\mathrm{1}}}  \vee  \ell_{{\mathrm{2}}}   }_{  \llbracket  \ottnt{A}  \rrbracket  } } \ar[drr," \text{id}_{  \mathbf{S}_{  \ell_{{\mathrm{1}}}  }   \llbracket  \ottnt{A}  \rrbracket   } ",bend left=120] &  \mathbf{S}_{  \ell_{{\mathrm{1}}}  }   \mathbf{S}_{   \bot   }   \llbracket  \ottnt{A}  \rrbracket    \arrow{d}{ \mathbf{S}_{  \ell_{{\mathrm{1}}}  }   \mathbf{S}^{   \bot    \sqsubseteq   \ell_{{\mathrm{2}}}  }_{  \llbracket  \ottnt{A}  \rrbracket  }  } \arrow{dr}{ \mathbf{S}_{  \ell_{{\mathrm{1}}}  }   \epsilon_{  \llbracket  \ottnt{A}  \rrbracket  }  } & \\
 \mathbf{S}_{   \ell_{{\mathrm{1}}}  \vee  \ell_{{\mathrm{2}}}   }   \llbracket  \ottnt{A}  \rrbracket   \arrow{r}{ \delta^{  \ell_{{\mathrm{1}}}  ,  \ell_{{\mathrm{2}}}  }_{  \llbracket  \ottnt{A}  \rrbracket  } }  &  \mathbf{S}_{  \ell_{{\mathrm{1}}}  }   \mathbf{S}_{  \ell_{{\mathrm{2}}}  }   \llbracket  \ottnt{A}  \rrbracket    \arrow{r}{ \mathbf{S}_{  \ell_{{\mathrm{1}}}  }   k^{ \ell_{{\mathrm{2}}} }_{  \llbracket  \ottnt{A}  \rrbracket  }  } &  \mathbf{S}_{  \ell_{{\mathrm{1}}}  }   \llbracket  \ottnt{A}  \rrbracket   
\end{tikzcd}
\caption{Commutative diagram}
\label{prfcd6}
\end{figure}
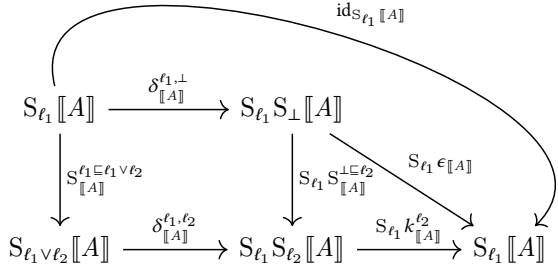

The diagram in Figure \ref{prfcd6} commutes: the circular segment commutes because $\mathbf{S}$ is an oplax monoidal functor; the square commutes because $\delta$ is natural in its second component; the triangle commutes by IH.

\end{itemize}
\end{proof}


\begin{lemma}\label{Ap3}
If $ \ell  \sqsubseteq  \ottnt{A} $ in \ED{}, then $   \mathbf{S}^{   \bot    \sqsubseteq   \ell  }_{  \llbracket  \ottnt{A}  \rrbracket  }   \circ   \eta_{  \llbracket  \ottnt{A}  \rrbracket  }    \circ   k^{ \ell }_{  \llbracket  \ottnt{A}  \rrbracket  }   =  \text{id}_{  \mathbf{S}_{  \ell  }   \llbracket  \ottnt{A}  \rrbracket   } $.
\end{lemma}

\begin{proof}
Let $ \ell  \sqsubseteq  \ottnt{A} $. By Lemma \ref{Ap2}, we know that $\exists  k^{ \ell }_{  \llbracket  \ottnt{A}  \rrbracket  } $ such that $   k^{ \ell }_{  \llbracket  \ottnt{A}  \rrbracket  }   \circ   \mathbf{S}^{   \bot    \sqsubseteq   \ell  }_{  \llbracket  \ottnt{A}  \rrbracket  }    \circ   \eta_{  \llbracket  \ottnt{A}  \rrbracket  }   =  \text{id}_{  \llbracket  \ottnt{A}  \rrbracket  } $. 

Now,
\begin{align*}
   k^{ \ell }_{  \llbracket  \ottnt{A}  \rrbracket  }   \circ   \mathbf{S}^{   \bot    \sqsubseteq   \ell  }_{  \llbracket  \ottnt{A}  \rrbracket  }    \circ   \eta_{  \llbracket  \ottnt{A}  \rrbracket  }   & =  \text{id}_{  \llbracket  \ottnt{A}  \rrbracket  }  \\
\text{or, }    \mathbf{S}_{  \ell  }   k^{ \ell }_{  \llbracket  \ottnt{A}  \rrbracket  }    \circ   \mathbf{S}_{  \ell  }   \mathbf{S}^{   \bot    \sqsubseteq   \ell  }_{  \llbracket  \ottnt{A}  \rrbracket  }     \circ   \mathbf{S}_{  \ell  }   \eta_{  \llbracket  \ottnt{A}  \rrbracket  }    & =  \mathbf{S}_{  \ell  }   \text{id}_{  \llbracket  \ottnt{A}  \rrbracket  }   \hspace*{3pt} [\because \mathbf{S}_{\ell} \text{ is a functor}]\\
\text{or, }   \mathbf{S}_{  \ell  }   k^{ \ell }_{  \llbracket  \ottnt{A}  \rrbracket  }    \circ   \mathbf{S}^{   \bot    \sqsubseteq   \ell  }_{  \mathbf{S}_{  \ell  }   \llbracket  \ottnt{A}  \rrbracket   }    \circ   \eta_{  \mathbf{S}_{  \ell  }   \llbracket  \ottnt{A}  \rrbracket   }   & =  \text{id}_{  \mathbf{S}_{  \ell  }   \llbracket  \ottnt{A}  \rrbracket   }  \hspace*{3pt} [\because (\mathbf{S}_{\ell},\mathbf{S}^{  \bot   \sqsubseteq  \ell } \circ \eta, \mu^{\ell,\ell}) \text{ is an idempotent monad}]\\
\text{or, }   \mathbf{S}^{   \bot    \sqsubseteq   \ell  }_{  \llbracket  \ottnt{A}  \rrbracket  }   \circ   \eta_{  \llbracket  \ottnt{A}  \rrbracket  }    \circ   k^{ \ell }_{  \llbracket  \ottnt{A}  \rrbracket  }   & =  \text{id}_{  \mathbf{S}_{  \ell  }   \llbracket  \ottnt{A}  \rrbracket   }  \hspace*{3pt} [\text{By naturality of } \mathbf{S}^{  \bot   \sqsubseteq  \ell } \circ \eta]
\end{align*}
This shows that $ \mathbf{S}_{  \ell  }   \llbracket  \ottnt{A}  \rrbracket   \cong  \llbracket  \ottnt{A}  \rrbracket $.\\
\end{proof}


\begin{lemma}[Lemma \ref{lemmak}]
If $ \ell  \sqsubseteq  \ottnt{A} $, then $\exists$ an isomorphism $k : \mathbf{S}_{\ell}  \llbracket  \ottnt{A}  \rrbracket  \to  \llbracket  \ottnt{A}  \rrbracket $. \\ Further, $k \circ \overline{\eta} = id_{ \llbracket  \ottnt{A}  \rrbracket }$ and $\overline{\eta} \circ k =  \text{id}_{  \mathbf{S}_{  \ell  }   \llbracket  \ottnt{A}  \rrbracket   } $ where $\overline{\eta} \triangleq \mathbf{S}^{  \bot   \sqsubseteq  \ell } \circ \eta :  \llbracket  \ottnt{A}  \rrbracket  \to \mathbf{S}_{\ell}  \llbracket  \ottnt{A}  \rrbracket $.
\end{lemma} 

\begin{proof}
Follows by lemma \ref{Ap2} and \ref{Ap3}.
\end{proof}


\begin{theorem}[Theorem \ref{dcceWD}]
If $ \Gamma  \vdash  \ottnt{a}  :  \ottnt{A} $ in \ED{}, then $ \llbracket  \ottnt{a}  \rrbracket  \in \text{Hom}_{\Ct} ( \llbracket  \Gamma  \rrbracket  ,  \llbracket  \ottnt{A}  \rrbracket )$.
\end{theorem}

\begin{proof}
By induction on $ \Gamma  \vdash  \ottnt{a}  :  \ottnt{A} $.

\begin{itemize}
\item $\lambda$-calculus. Standard.
\item \Rref{DCC-Eta}. Have: $ \Gamma  \vdash   \mathbf{eta} ^{ \ell }  \ottnt{a}   :   \mathcal{T}_{ \ell } \:  \ottnt{A}  $ where $ \Gamma  \vdash  \ottnt{a}  :  \ottnt{A} $.\\
By IH, $ \llbracket  \ottnt{a}  \rrbracket  \in \text{Hom}_{\Ct} ( \llbracket  \Gamma  \rrbracket ,  \llbracket  \ottnt{A}  \rrbracket )$.\\
Therefore, $ \llbracket    \mathbf{eta} ^{ \ell }  \ottnt{a}    \rrbracket  =    \mathbf{S}^{   \bot    \sqsubseteq   \ell  }_{  \llbracket  \ottnt{A}  \rrbracket  }   \circ   \eta_{  \llbracket  \ottnt{A}  \rrbracket  }    \circ   \llbracket  \ottnt{a}  \rrbracket   \in \text{Hom}_{\Ct} ( \llbracket  \Gamma  \rrbracket ,  \mathbf{S}_{  \ell  }   \llbracket  \ottnt{A}  \rrbracket  )$.
\item \Rref{DCC-Bind}. Have: $ \Gamma  \vdash   \mathbf{bind} ^{ \ell } \:  \ottmv{x}  =  \ottnt{a}  \: \mathbf{in} \:  \ottnt{b}   :  \ottnt{B} $ where $ \Gamma  \vdash  \ottnt{a}  :   \mathcal{T}_{ \ell } \:  \ottnt{A}  $ and $  \Gamma  ,   \ottmv{x}  :  \ottnt{A}    \vdash  \ottnt{b}  :  \ottnt{B} $ and $ \ell  \sqsubseteq  \ottnt{B} $.\\
By lemma \ref{Ap2}, $\exists  k^{ \ell }_{  \llbracket  \ottnt{B}  \rrbracket  }  \in \text{Hom}_{\Ct} ( \mathbf{S}_{  \ell  }   \llbracket  \ottnt{B}  \rrbracket  , \llbracket  \ottnt{B}  \rrbracket )$. Now,\\
$ \llbracket    \mathbf{bind} ^{ \ell } \:  \ottmv{x}  =  \ottnt{a}  \: \mathbf{in} \:  \ottnt{b}    \rrbracket  =  \llbracket  \Gamma  \rrbracket  \xrightarrow{ \langle   \text{id}_{  \llbracket  \Gamma  \rrbracket  }   ,   \llbracket  \ottnt{a}  \rrbracket   \rangle }   \llbracket  \Gamma  \rrbracket   \times   \mathbf{S}_{  \ell  }   \llbracket  \ottnt{A}  \rrbracket    \xrightarrow{ t^{\mathbf{S}_{  \ell  } }_{  \llbracket  \Gamma  \rrbracket  ,   \llbracket  \ottnt{A}  \rrbracket  } }  \mathbf{S}_{  \ell  }   (    \llbracket  \Gamma  \rrbracket   \times   \llbracket  \ottnt{A}  \rrbracket    )   \xrightarrow{ \mathbf{S}_{  \ell  }   \llbracket  \ottnt{b}  \rrbracket  }  \mathbf{S}_{  \ell  }   \llbracket  \ottnt{B}  \rrbracket   \xrightarrow{ k^{ \ell }_{  \llbracket  \ottnt{B}  \rrbracket  } }  \llbracket  \ottnt{B}  \rrbracket $.\\
Therefore, $ \llbracket    \mathbf{bind} ^{ \ell } \:  \ottmv{x}  =  \ottnt{a}  \: \mathbf{in} \:  \ottnt{b}    \rrbracket  \in \text{Hom}_{\Ct} ( \llbracket  \Gamma  \rrbracket , \llbracket  \ottnt{B}  \rrbracket )$.
\end{itemize}
\end{proof}


\begin{theorem}[Theorem \ref{dcceOE}] \label{dcceOEprf}
If $\Gamma \vdash a : A$ in \ED{} and $ \vdash  \ottnt{a}  \leadsto  \ottnt{a'} $, then $\llbracket a \rrbracket = \llbracket a’ \rrbracket$.
\end{theorem}

\begin{proof}
By induction on $ \Gamma  \vdash  \ottnt{a}  :  \ottnt{A} $.
\begin{itemize}
\item $\lambda$-calculus. Standard.
\item \Rref{DCC-Eta}. Have: $ \Gamma  \vdash   \mathbf{eta} ^{ \ell }  \ottnt{a}   :   \mathcal{T}_{ \ell } \:  \ottnt{A}  $ where $ \Gamma  \vdash  \ottnt{a}  :  \ottnt{A} $.\\ Further, $ \vdash   \mathbf{eta} ^{ \ell }  \ottnt{a}   \leadsto  \ottnt{c} $. By inversion on $ \vdash   \mathbf{eta} ^{ \ell }  \ottnt{a}   \leadsto  \ottnt{c} $.
\begin{itemize}
\item \Rref{CBV-Eta}. Have: $ \vdash   \mathbf{eta} ^{ \ell }  \ottnt{a}   \leadsto   \mathbf{eta} ^{ \ell }  \ottnt{a'}  $ where $ \vdash  \ottnt{a}  \leadsto  \ottnt{a'} $.\\
By IH, $ \llbracket  \ottnt{a}  \rrbracket  =  \llbracket  \ottnt{a'}  \rrbracket $.\\
Therefore, $ \llbracket    \mathbf{eta} ^{ \ell }  \ottnt{a}    \rrbracket  =    \mathbf{S}^{   \bot    \sqsubseteq   \ell  }_{  \llbracket  \ottnt{A}  \rrbracket  }   \circ   \eta_{  \llbracket  \ottnt{A}  \rrbracket  }    \circ   \llbracket  \ottnt{a}  \rrbracket   =    \mathbf{S}^{   \bot    \sqsubseteq   \ell  }_{  \llbracket  \ottnt{A}  \rrbracket  }   \circ   \eta_{  \llbracket  \ottnt{A}  \rrbracket  }    \circ   \llbracket  \ottnt{a'}  \rrbracket   =  \llbracket    \mathbf{eta} ^{ \ell }  \ottnt{a'}    \rrbracket $.
\end{itemize}
\item \Rref{DCC-Bind}. Have: $ \Gamma  \vdash   \mathbf{bind} ^{ \ell } \:  \ottmv{x}  =  \ottnt{a}  \: \mathbf{in} \:  \ottnt{b}   :  \ottnt{B} $ where $ \Gamma  \vdash  \ottnt{a}  :   \mathcal{T}_{ \ell } \:  \ottnt{A}  $ and $  \Gamma  ,   \ottmv{x}  :  \ottnt{A}    \vdash  \ottnt{b}  :  \ottnt{B} $ and $ \ell  \sqsubseteq  \ottnt{B} $.\\
Further, $ \vdash   \mathbf{bind} ^{ \ell } \:  \ottmv{x}  =  \ottnt{a}  \: \mathbf{in} \:  \ottnt{b}   \leadsto  \ottnt{c} $. By inversion on $ \vdash   \mathbf{bind} ^{ \ell } \:  \ottmv{x}  =  \ottnt{a}  \: \mathbf{in} \:  \ottnt{b}   \leadsto  \ottnt{c} $.
\begin{itemize}
\item \Rref{CBV-BindLeft}. Have: $ \vdash   \mathbf{bind} ^{ \ell } \:  \ottmv{x}  =  \ottnt{a}  \: \mathbf{in} \:  \ottnt{b}   \leadsto   \mathbf{bind} ^{ \ell } \:  \ottmv{x}  =  \ottnt{a'}  \: \mathbf{in} \:  \ottnt{b}  $ where $ \vdash  \ottnt{a}  \leadsto  \ottnt{a'} $.\\
By IH, $ \llbracket  \ottnt{a}  \rrbracket  =  \llbracket  \ottnt{a'}  \rrbracket $.\\
Therefore, \begin{align*}
 \llbracket    \mathbf{bind} ^{ \ell } \:  \ottmv{x}  =  \ottnt{a}  \: \mathbf{in} \:  \ottnt{b}    \rrbracket  & =     k^{ \ell }_{  \llbracket  \ottnt{B}  \rrbracket  }   \circ   \mathbf{S}_{  \ell  }   \llbracket  \ottnt{b}  \rrbracket     \circ   t^{\mathbf{S}_{  \ell  } }_{  \llbracket  \Gamma  \rrbracket  ,   \llbracket  \ottnt{A}  \rrbracket  }    \circ   \langle   \text{id}_{  \llbracket  \Gamma  \rrbracket  }   ,   \llbracket  \ottnt{a}  \rrbracket   \rangle  \\
& =     k^{ \ell }_{  \llbracket  \ottnt{B}  \rrbracket  }   \circ   \mathbf{S}_{  \ell  }   \llbracket  \ottnt{b}  \rrbracket     \circ   t^{\mathbf{S}_{  \ell  } }_{  \llbracket  \Gamma  \rrbracket  ,   \llbracket  \ottnt{A}  \rrbracket  }    \circ   \langle   \text{id}_{  \llbracket  \Gamma  \rrbracket  }   ,   \llbracket  \ottnt{a'}  \rrbracket   \rangle   =  \llbracket    \mathbf{bind} ^{ \ell } \:  \ottmv{x}  =  \ottnt{a'}  \: \mathbf{in} \:  \ottnt{b}    \rrbracket .
\end{align*}
\item \Rref{CBV-BindBeta}. Have: $ \vdash   \mathbf{bind} ^{ \ell } \:  \ottmv{x}  =   \mathbf{eta} ^{ \ell }  \ottmv{v}   \: \mathbf{in} \:  \ottnt{b}   \leadsto   \ottnt{b}  \{  \ottmv{v}  /  \ottmv{x}  \}  $.\\
Now, \begin{align*}
&  \llbracket    \mathbf{bind} ^{ \ell } \:  \ottmv{x}  =   \mathbf{eta} ^{ \ell }  \ottmv{v}   \: \mathbf{in} \:  \ottnt{b}    \rrbracket  \\
= &     k^{ \ell }_{  \llbracket  \ottnt{B}  \rrbracket  }   \circ   \mathbf{S}_{  \ell  }   \llbracket  \ottnt{b}  \rrbracket     \circ   t^{\mathbf{S}_{  \ell  } }_{  \llbracket  \Gamma  \rrbracket  ,   \llbracket  \ottnt{A}  \rrbracket  }    \circ   \langle   \text{id}_{  \llbracket  \Gamma  \rrbracket  }   ,     \mathbf{S}^{   \bot    \sqsubseteq   \ell  }_{  \llbracket  \ottnt{A}  \rrbracket  }   \circ   \eta_{  \llbracket  \ottnt{A}  \rrbracket  }    \circ   \llbracket  \ottmv{v}  \rrbracket    \rangle  \\
= &      k^{ \ell }_{  \llbracket  \ottnt{B}  \rrbracket  }   \circ   \mathbf{S}_{  \ell  }   \llbracket  \ottnt{b}  \rrbracket     \circ   t^{\mathbf{S}_{  \ell  } }_{  \llbracket  \Gamma  \rrbracket  ,   \llbracket  \ottnt{A}  \rrbracket  }    \circ   (     \text{id}_{  \llbracket  \Gamma  \rrbracket  }    \times    \mathbf{S}^{   \bot    \sqsubseteq   \ell  }_{  \llbracket  \ottnt{A}  \rrbracket  }     )    \circ   \langle   \text{id}_{  \llbracket  \Gamma  \rrbracket  }   ,    \eta_{  \llbracket  \ottnt{A}  \rrbracket  }   \circ   \llbracket  \ottmv{v}  \rrbracket    \rangle  \\
= &      k^{ \ell }_{  \llbracket  \ottnt{B}  \rrbracket  }   \circ   \mathbf{S}_{  \ell  }   \llbracket  \ottnt{b}  \rrbracket     \circ   \mathbf{S}^{   \bot    \sqsubseteq   \ell  }_{    \llbracket  \Gamma  \rrbracket   \times   \llbracket  \ottnt{A}  \rrbracket    }    \circ   t^{\mathbf{S}_{   \bot   } }_{  \llbracket  \Gamma  \rrbracket  ,   \llbracket  \ottnt{A}  \rrbracket  }    \circ   \langle   \text{id}_{  \llbracket  \Gamma  \rrbracket  }   ,    \eta_{  \llbracket  \ottnt{A}  \rrbracket  }   \circ   \llbracket  \ottmv{v}  \rrbracket    \rangle   \hspace*{3pt} [\because \mathbf{S}^{  \bot   \sqsubseteq  \ell } \text{ is strong}] \\
= &       k^{ \ell }_{  \llbracket  \ottnt{B}  \rrbracket  }   \circ   \mathbf{S}_{  \ell  }   \llbracket  \ottnt{b}  \rrbracket     \circ   \mathbf{S}^{   \bot    \sqsubseteq   \ell  }_{    \llbracket  \Gamma  \rrbracket   \times   \llbracket  \ottnt{A}  \rrbracket    }    \circ   t^{\mathbf{S}_{   \bot   } }_{  \llbracket  \Gamma  \rrbracket  ,   \llbracket  \ottnt{A}  \rrbracket  }    \circ   (    \text{id}_{  \llbracket  \Gamma  \rrbracket  }   \times   \eta_{  \llbracket  \ottnt{A}  \rrbracket  }    )    \circ   \langle   \text{id}_{  \llbracket  \Gamma  \rrbracket  }   ,   \llbracket  \ottmv{v}  \rrbracket   \rangle   \\
= &      k^{ \ell }_{  \llbracket  \ottnt{B}  \rrbracket  }   \circ   \mathbf{S}_{  \ell  }   \llbracket  \ottnt{b}  \rrbracket     \circ   \mathbf{S}^{   \bot    \sqsubseteq   \ell  }_{    \llbracket  \Gamma  \rrbracket   \times   \llbracket  \ottnt{A}  \rrbracket    }    \circ   \eta_{    \llbracket  \Gamma  \rrbracket   \times   \llbracket  \ottnt{A}  \rrbracket    }    \circ   \langle   \text{id}_{  \llbracket  \Gamma  \rrbracket  }   ,   \llbracket  \ottmv{v}  \rrbracket   \rangle   \hspace*{3pt} [ \because \eta \text{ is strong}]\\
= &      k^{ \ell }_{  \llbracket  \ottnt{B}  \rrbracket  }   \circ   \mathbf{S}^{   \bot    \sqsubseteq   \ell  }_{  \llbracket  \ottnt{B}  \rrbracket  }    \circ   \mathbf{S}_{   \bot   }   \llbracket  \ottnt{b}  \rrbracket     \circ   \eta_{    \llbracket  \Gamma  \rrbracket   \times   \llbracket  \ottnt{A}  \rrbracket    }    \circ   \langle   \text{id}_{  \llbracket  \Gamma  \rrbracket  }   ,   \llbracket  \ottmv{v}  \rrbracket   \rangle   \hspace*{3pt} [\text{By naturality}]\\
= &      k^{ \ell }_{  \llbracket  \ottnt{B}  \rrbracket  }   \circ   \mathbf{S}^{   \bot    \sqsubseteq   \ell  }_{  \llbracket  \ottnt{B}  \rrbracket  }    \circ   \eta_{  \llbracket  \ottnt{B}  \rrbracket  }    \circ   \llbracket  \ottnt{b}  \rrbracket    \circ   \langle   \text{id}_{  \llbracket  \Gamma  \rrbracket  }   ,   \llbracket  \ottmv{v}  \rrbracket   \rangle   \hspace*{3pt} [\text{By naturality}]\\
= &    \text{id}_{  \llbracket  \ottnt{B}  \rrbracket  }   \circ   \llbracket  \ottnt{b}  \rrbracket    \circ   \langle   \text{id}_{  \llbracket  \Gamma  \rrbracket  }   ,   \llbracket  \ottmv{v}  \rrbracket   \rangle   \hspace*{3pt} [\text{By Lemma \ref{Ap2}}]\\
= &  \llbracket    \ottnt{b}  \{  \ottmv{v}  /  \ottmv{x}  \}    \rrbracket 
\end{align*}
\end{itemize}
\end{itemize}

\end{proof}


\begin{theorem}[Theorem \ref{dcceCA}] \label{dcceCAprf}
Let the interpretation $\llbracket \_ \rrbracket_{(\Ct,\mathbf{S})}$ be injective for ground types.
\begin{itemize}
\item Let $ \Gamma  \vdash  \ottnt{b}  :   \mathbf{Bool}  $ and $\ottmv{v}$ be a value of type $ \mathbf{Bool} $. If $ \llbracket  \ottnt{b}  \rrbracket _{(\Ct,\mathbf{S})} =  \llbracket  \ottmv{v}  \rrbracket _{(\Ct,\mathbf{S})}$, then $ \vdash  \ottnt{b}  \leadsto^{\ast}  \ottmv{v} $. \vspace*{3pt}
\item Fix some $\ell \in L$. Let $ \Gamma  \vdash  \ottnt{b}  :   \mathcal{T}_{ \ell } \:   \mathbf{Bool}   $ and $\ottmv{v}$ be a value of type $ \mathcal{T}_{ \ell } \:   \mathbf{Bool}  $. Suppose, the morphisms $\overline{\eta}_X \triangleq   \mathbf{S}^{   \bot    \sqsubseteq   \ell  }_{ \ottnt{X} }   \circ   \eta_{ \ottnt{X} }  $ are mono for any $X \in \text{Obj} (\Ct)$.  Now, if $ \llbracket  \ottnt{b}  \rrbracket _{(\Ct,\mathbf{S})} =  \llbracket  \ottmv{v}  \rrbracket _{(\Ct,\mathbf{S})}$, then $ \vdash  \ottnt{b}  \leadsto^{\ast}  \ottmv{v} $. 
\end{itemize}
\end{theorem}

\begin{proof}
Let $ \Gamma  \vdash  \ottnt{b}  :   \mathbf{Bool}  $ and $\ottmv{v}$ be a boolean value such that $ \llbracket  \ottnt{b}  \rrbracket_{(\mathds{C},\mathbf{S})}  =  \llbracket  \ottmv{v}  \rrbracket_{(\mathds{C},\mathbf{S})} $.\\
We show that $ \vdash  \ottnt{b}  \leadsto^{\ast}  \ottmv{v} $.

First note that \ED{} is strongly normalizing with respect to the reduction relation, $\leadsto^{\ast}$. Further, \ED{} is also type sound with respect to this reduction relation.

Therefore, given $ \Gamma  \vdash  \ottnt{b}  :   \mathbf{Bool}  $, we know that there exists a value $\ottmv{v_{{\mathrm{0}}}}$ such that $ \Gamma  \vdash  \ottmv{v_{{\mathrm{0}}}}  :   \mathbf{Bool}  $ and $ \vdash  \ottnt{b}  \leadsto^{\ast}  \ottmv{v_{{\mathrm{0}}}} $.

Next, by Theorem \ref{dcceOEprf}, $ \llbracket  \ottnt{b}  \rrbracket_{(\mathds{C},\mathbf{S})}  =  \llbracket  \ottmv{v_{{\mathrm{0}}}}  \rrbracket_{(\mathds{C},\mathbf{S})} $.

Since $ \llbracket  \ottnt{b}  \rrbracket_{(\mathds{C},\mathbf{S})}  =  \llbracket  \ottmv{v}  \rrbracket_{(\mathds{C},\mathbf{S})} $ (given) and $ \llbracket  \ottnt{b}  \rrbracket_{(\mathds{C},\mathbf{S})}  =  \llbracket  \ottmv{v_{{\mathrm{0}}}}  \rrbracket_{(\mathds{C},\mathbf{S})} $ (above), therefore, $ \llbracket  \ottmv{v}  \rrbracket_{(\mathds{C},\mathbf{S})}  =  \llbracket  \ottmv{v_{{\mathrm{0}}}}  \rrbracket_{(\mathds{C},\mathbf{S})} $.

But, by injectivity, $\ottmv{v} = \ottmv{v_{{\mathrm{0}}}}$.

Thus, $ \vdash  \ottnt{b}  \leadsto^{\ast}  \ottmv{v} $.\\

For the second part, we use a similar argument.

Given $ \Gamma  \vdash  \ottnt{b}  :   \mathcal{T}_{ \ell } \:   \mathbf{Bool}   $, we know that there exists a value $\ottmv{v_{{\mathrm{0}}}}$ such that $ \Gamma  \vdash  \ottmv{v_{{\mathrm{0}}}}  :   \mathcal{T}_{ \ell } \:   \mathbf{Bool}   $ and $ \vdash  \ottnt{b}  \leadsto^{\ast}  \ottmv{v_{{\mathrm{0}}}} $.

By Theorem \ref{dcceOEprf}, $ \llbracket  \ottnt{b}  \rrbracket_{(\mathds{C},\mathbf{S})}  =  \llbracket  \ottmv{v_{{\mathrm{0}}}}  \rrbracket_{(\mathds{C},\mathbf{S})} $.

Since $ \llbracket  \ottnt{b}  \rrbracket_{(\mathds{C},\mathbf{S})}  =  \llbracket  \ottmv{v}  \rrbracket_{(\mathds{C},\mathbf{S})} $ (given) and $ \llbracket  \ottnt{b}  \rrbracket_{(\mathds{C},\mathbf{S})}  =  \llbracket  \ottmv{v_{{\mathrm{0}}}}  \rrbracket_{(\mathds{C},\mathbf{S})} $ (above), therefore, $ \llbracket  \ottmv{v}  \rrbracket_{(\mathds{C},\mathbf{S})}  =  \llbracket  \ottmv{v_{{\mathrm{0}}}}  \rrbracket_{(\mathds{C},\mathbf{S})} $.

Both $\ottmv{v}$ and $\ottmv{v_{{\mathrm{0}}}}$ are values of type $ \mathcal{T}_{ \ell } \:   \mathbf{Bool}  $. Therefore, $\ottmv{v} =  \mathbf{eta} ^{ \ell }  \ottmv{v'} $ and $\ottmv{v_{{\mathrm{0}}}} =  \mathbf{eta} ^{ \ell }  \ottmv{v'_{{\mathrm{0}}}} $, for some values $\ottmv{v'}$ and $\ottmv{v'_{{\mathrm{0}}}}$ of type $ \mathbf{Bool} $. 

Then, $ \llbracket   \mathbf{eta} ^{ \ell }  \ottmv{v'}   \rrbracket_{(\mathds{C},\mathbf{S})}  =  \llbracket   \mathbf{eta} ^{ \ell }  \ottmv{v'_{{\mathrm{0}}}}   \rrbracket_{(\mathds{C},\mathbf{S})} $, or $   \mathbf{S}^{   \bot    \sqsubseteq   \ell  }_{   \llbracket   \mathbf{Bool}   \rrbracket   }   \circ   \eta_{   \llbracket   \mathbf{Bool}   \rrbracket   }    \circ   \llbracket  \ottmv{v'}  \rrbracket_{(\mathds{C},\mathbf{S})}   =    \mathbf{S}^{   \bot    \sqsubseteq   \ell  }_{   \llbracket   \mathbf{Bool}   \rrbracket   }   \circ   \eta_{   \llbracket   \mathbf{Bool}   \rrbracket   }    \circ   \llbracket  \ottmv{v'_{{\mathrm{0}}}}  \rrbracket_{(\mathds{C},\mathbf{S})}  $.

Since for any $X \in \text{Obj} (\Ct)$, $ \overline{\eta}_{ \ottnt{X} } $ is mono, so $ \llbracket  \ottmv{v'}  \rrbracket_{(\mathds{C},\mathbf{S})}  =  \llbracket  \ottmv{v'_{{\mathrm{0}}}}  \rrbracket_{(\mathds{C},\mathbf{S})} $. 

By injectivity, $\ottmv{v'} = \ottmv{v'_{{\mathrm{0}}}}$.

Therefore, $\ottmv{v} = \ottmv{v_{{\mathrm{0}}}}$.

Hence, $ \vdash  \ottnt{b}  \leadsto^{\ast}  \ottmv{v} $.
\end{proof}


\begin{theorem}\label{equiv1}
If $ \Gamma  \vdash  \ottnt{a}  :  \ottnt{A} $ in GMCC($ \mathcal{L} $), then $ \llbracket   \overline{ \ottnt{a} }   \rrbracket  =  \llbracket  \ottnt{a}  \rrbracket $.
\end{theorem}

\begin{proof}
By induction on $ \Gamma  \vdash  \ottnt{a}  :  \ottnt{A} $. Note that $ \llbracket   \overline{ \ottnt{A} }   \rrbracket  =  \llbracket  \ottnt{A}  \rrbracket $.

\begin{itemize}

\item $\lambda$-calculus. By IH.

\item \Rref{MC-Return}. Have: $ \Gamma  \vdash   \ottkw{ret}  \:  \ottnt{a}   :   S_{ \ottsym{1} } \:  \ottnt{A}  $ where $ \Gamma  \vdash  \ottnt{a}  :  \ottnt{A} $.\\
By IH, $ \llbracket   \overline{ \ottnt{a} }   \rrbracket  =  \llbracket  \ottnt{a}  \rrbracket $.\\ Now, $ \overline{  \ottkw{ret}  \:  \ottnt{a}  }  =  \mathbf{eta} ^{  \bot  }   \overline{ \ottnt{a} }  $.\\
Then, $ \llbracket    \mathbf{eta} ^{  \bot  }   \overline{ \ottnt{a} }     \rrbracket  =    \mathbf{S}^{   \bot    \sqsubseteq    \bot   }_{  \llbracket  \ottnt{A}  \rrbracket  }   \circ   \eta_{  \llbracket  \ottnt{A}  \rrbracket  }    \circ   \llbracket   \overline{ \ottnt{a} }   \rrbracket   =   \eta_{  \llbracket  \ottnt{A}  \rrbracket  }   \circ   \llbracket  \ottnt{a}  \rrbracket   =  \llbracket    \ottkw{ret}  \:  \ottnt{a}    \rrbracket $.

\item \Rref{MC-Extract}. Have: $ \Gamma  \vdash   \mathbf{extr} \:  \ottnt{a}   :  \ottnt{A} $ where $ \Gamma  \vdash  \ottnt{a}  :   S_{ \ottsym{1} } \:  \ottnt{A}  $.\\
By IH, $ \llbracket   \overline{ \ottnt{a} }   \rrbracket  =  \llbracket  \ottnt{a}  \rrbracket $.\\ Now, $ \overline{  \mathbf{extr} \:  \ottnt{a}  }  =  \mathbf{bind} ^{  \bot  } \:  \ottmv{x}  =   \overline{ \ottnt{a} }   \: \mathbf{in} \:  \ottmv{x} $.\\
Then, \begin{align*}
&  \llbracket    \mathbf{bind} ^{  \bot  } \:  \ottmv{x}  =   \overline{ \ottnt{a} }   \: \mathbf{in} \:  \ottmv{x}    \rrbracket \\
= &     k^{  \bot  }_{  \llbracket  \ottnt{A}  \rrbracket  }   \circ   \mathbf{S}_{   \bot   }   \pi_2     \circ   t^{\mathbf{S}_{   \bot   } }_{  \llbracket  \Gamma  \rrbracket  ,   \llbracket  \ottnt{A}  \rrbracket  }    \circ   \langle   \text{id}_{  \llbracket  \Gamma  \rrbracket  }   ,   \llbracket   \overline{ \ottnt{a} }   \rrbracket   \rangle  \\
= &    k^{  \bot  }_{  \llbracket  \ottnt{A}  \rrbracket  }   \circ   \pi_2    \circ   \langle   \text{id}_{  \llbracket  \Gamma  \rrbracket  }   ,   \llbracket  \ottnt{a}  \rrbracket   \rangle   \hspace*{3pt} [\because \mathbf{S}^{ \bot } \text{ is strong}]\\
= &   \epsilon_{  \llbracket  \ottnt{A}  \rrbracket  }   \circ   \llbracket  \ottnt{a}  \rrbracket   \hspace*{3pt} [\text{By Lemma \ref{Ap2}}]\\
= &  \llbracket    \mathbf{extr} \:  \ottnt{a}    \rrbracket 
\end{align*}

\item \Rref{MC-Fmap}. Have: $ \Gamma  \vdash   \mathbf{lift}^{  \ell  }  \ottnt{f}   :    S_{  \ell  } \:  \ottnt{A}   \to   S_{  \ell  } \:  \ottnt{B}   $ where $ \Gamma  \vdash  \ottnt{f}  :   \ottnt{A}  \to  \ottnt{B}  $.\\
By IH, $ \llbracket   \overline{ \ottnt{f} }   \rrbracket  =  \llbracket  \ottnt{f}  \rrbracket $.\\
Now, $ \overline{   \mathbf{lift}^{  \ell  }  \ottnt{f}   }  =  \lambda  \ottmv{x}  :   \mathcal{T}_{ \ell } \:   \overline{ \ottnt{A} }    .   \mathbf{bind} ^{ \ell } \:  \ottmv{y}  =  \ottmv{x}  \: \mathbf{in} \:   \mathbf{eta} ^{ \ell }   (    \overline{ \ottnt{f} }   \:  \ottmv{y}   )    $.\\
Then, \begin{align*}
&  \llbracket    \lambda  \ottmv{x}  :   \mathcal{T}_{ \ell } \:   \overline{ \ottnt{A} }    .   \mathbf{bind} ^{ \ell } \:  \ottmv{y}  =  \ottmv{x}  \: \mathbf{in} \:   \mathbf{eta} ^{ \ell }   (    \overline{ \ottnt{f} }   \:  \ottmv{y}   )       \rrbracket  \\
= &  \Lambda \Big(       k^{ \ell }_{   \mathbf{S}_{  \ell  }   \llbracket  \ottnt{B}  \rrbracket    }   \circ   \mathbf{S}_{  \ell  }   \llbracket    \mathbf{eta} ^{ \ell }   (    \overline{ \ottnt{f} }   \:  \ottmv{y}   )     \rrbracket     \circ   t^{\mathbf{S}_{  \ell  } }_{    \llbracket  \Gamma  \rrbracket   \times   \mathbf{S}_{  \ell  }   \llbracket  \ottnt{A}  \rrbracket     ,   \llbracket  \ottnt{A}  \rrbracket  }    \circ   \langle   \text{id}_{    \llbracket  \Gamma  \rrbracket   \times   \mathbf{S}_{  \ell  }   \llbracket  \ottnt{A}  \rrbracket     }   ,   \pi_2   \rangle     \Big)  \\
= &  \Lambda \Big(       k^{ \ell }_{   \mathbf{S}_{  \ell  }   \llbracket  \ottnt{B}  \rrbracket    }   \circ   \mathbf{S}_{  \ell  }   (       \mathbf{S}^{   \bot    \sqsubseteq   \ell  }_{  \llbracket  \ottnt{B}  \rrbracket  }   \circ   \eta_{  \llbracket  \ottnt{B}  \rrbracket  }    \circ   \text{app}    \circ   (    \llbracket  \ottnt{f}  \rrbracket   \times   \text{id}_{  \llbracket  \ottnt{A}  \rrbracket  }    )    \circ   \langle    \pi_1   \circ   \pi_1    ,   \pi_2   \rangle    )     \circ   t^{\mathbf{S}_{  \ell  } }_{    \llbracket  \Gamma  \rrbracket   \times   \mathbf{S}_{  \ell  }   \llbracket  \ottnt{A}  \rrbracket     ,   \llbracket  \ottnt{A}  \rrbracket  }    \circ   \langle   \text{id}_{    \llbracket  \Gamma  \rrbracket   \times   \mathbf{S}_{  \ell  }   \llbracket  \ottnt{A}  \rrbracket     }   ,   \pi_2   \rangle     \Big)  \\
= &  \Lambda \Big(          k^{ \ell }_{   \mathbf{S}_{  \ell  }   \llbracket  \ottnt{B}  \rrbracket    }   \circ   \mathbf{S}_{  \ell  }   \mathbf{S}^{   \bot    \sqsubseteq   \ell  }_{  \llbracket  \ottnt{B}  \rrbracket  }     \circ   \mathbf{S}_{  \ell  }   \eta_{  \llbracket  \ottnt{B}  \rrbracket  }     \circ   \mathbf{S}_{  \ell  }   (   \Lambda^{-1}   \llbracket  \ottnt{f}  \rrbracket    )     \circ   \mathbf{S}_{  \ell  }   \langle    \pi_1   \circ   \pi_1    ,   \pi_2   \rangle     \circ   t^{\mathbf{S}_{  \ell  } }_{    \llbracket  \Gamma  \rrbracket   \times   \mathbf{S}_{  \ell  }   \llbracket  \ottnt{A}  \rrbracket     ,   \llbracket  \ottnt{A}  \rrbracket  }    \circ   \langle   \text{id}_{    \llbracket  \Gamma  \rrbracket   \times   \mathbf{S}_{  \ell  }   \llbracket  \ottnt{A}  \rrbracket     }   ,   \pi_2   \rangle     \Big)  \\
= &  \Lambda \Big(          k^{ \ell }_{   \mathbf{S}_{  \ell  }   \llbracket  \ottnt{B}  \rrbracket    }   \circ   \mathbf{S}^{   \bot    \sqsubseteq   \ell  }_{   \mathbf{S}_{  \ell  }   \llbracket  \ottnt{B}  \rrbracket    }    \circ   \eta_{   \mathbf{S}_{  \ell  }   \llbracket  \ottnt{B}  \rrbracket    }    \circ   \mathbf{S}_{  \ell  }   (   \Lambda^{-1}   \llbracket  \ottnt{f}  \rrbracket    )     \circ   \mathbf{S}_{  \ell  }   \langle    \pi_1   \circ   \pi_1    ,   \pi_2   \rangle     \circ   t^{\mathbf{S}_{  \ell  } }_{    \llbracket  \Gamma  \rrbracket   \times   \mathbf{S}_{  \ell  }   \llbracket  \ottnt{A}  \rrbracket     ,   \llbracket  \ottnt{A}  \rrbracket  }    \circ   \langle   \text{id}_{    \llbracket  \Gamma  \rrbracket   \times   \mathbf{S}_{  \ell  }   \llbracket  \ottnt{A}  \rrbracket     }   ,   \pi_2   \rangle     \Big)  \\ & \hspace*{30pt} [\text{By idempotence}] \\
= &  \Lambda \Big(        \text{id}_{  \mathbf{S}_{  \ell  }   \llbracket  \ottnt{B}  \rrbracket   }   \circ   \mathbf{S}_{  \ell  }   (   \Lambda^{-1}   \llbracket  \ottnt{f}  \rrbracket    )     \circ   \mathbf{S}_{  \ell  }   \langle    \pi_1   \circ   \pi_1    ,   \pi_2   \rangle     \circ   t^{\mathbf{S}_{  \ell  } }_{    \llbracket  \Gamma  \rrbracket   \times   \mathbf{S}_{  \ell  }   \llbracket  \ottnt{A}  \rrbracket     ,   \llbracket  \ottnt{A}  \rrbracket  }    \circ   \langle   \text{id}_{    \llbracket  \Gamma  \rrbracket   \times   \mathbf{S}_{  \ell  }   \llbracket  \ottnt{A}  \rrbracket     }   ,   \pi_2   \rangle     \Big) \\ & \hspace*{30pt} [\text{By Theorem \ref{Ap2}}] \\
= &  \Lambda \Big(     \mathbf{S}_{  \ell  }   (   \Lambda^{-1}   \llbracket  \ottnt{f}  \rrbracket    )    \circ   t^{\mathbf{S}_{  \ell  } }_{  \llbracket  \Gamma  \rrbracket  ,   \llbracket  \ottnt{A}  \rrbracket  }     \Big)  \hspace*{3pt} [\text{By commutative diagram \ref{prfcd8}}]\\
= &  \llbracket    \mathbf{lift}^{  \ell  }  \ottnt{f}    \rrbracket 
\end{align*}

The upper rectangle in Figure \ref{prfcd8} commutes by strength and naturality.  

\begin{figure}
\begin{tikzcd}[row sep = 3.5 em, column sep = 3 em]
  \llbracket  \Gamma  \rrbracket   \times   \mathbf{S}_{  \ell  }   \llbracket  \ottnt{A}  \rrbracket    \arrow{rr}{ t^{\mathbf{S}_{  \ell  } }_{  \llbracket  \Gamma  \rrbracket  ,   \llbracket  \ottnt{A}  \rrbracket  } } \arrow{d}{ \langle   \text{id}   ,   \pi_2   \rangle } & &  \mathbf{S}_{  \ell  }   (    \llbracket  \Gamma  \rrbracket   \times   \llbracket  \ottnt{A}  \rrbracket    )   \\
  (    \llbracket  \Gamma  \rrbracket   \times   \mathbf{S}_{  \ell  }   \llbracket  \ottnt{A}  \rrbracket     )   \times   \mathbf{S}_{  \ell  }   \llbracket  \ottnt{A}  \rrbracket    \arrow{rr}{ t^{\mathbf{S}_{  \ell  } }_{    \llbracket  \Gamma  \rrbracket   \times   \mathbf{S}_{  \ell  }   \llbracket  \ottnt{A}  \rrbracket     ,   \llbracket  \ottnt{A}  \rrbracket  } } \arrow{d}{ \alpha^{-1} } & &  \mathbf{S}_{  \ell  }   (    (    \llbracket  \Gamma  \rrbracket   \times   \mathbf{S}_{  \ell  }   \llbracket  \ottnt{A}  \rrbracket     )   \times   \llbracket  \ottnt{A}  \rrbracket    )   \arrow{u}{ \mathbf{S}_{  \ell  }   \langle    \pi_1   \circ   \pi_1    ,   \pi_2   \rangle  } \arrow{d}{ \mathbf{S}_{  \ell  }   \alpha^{-1}  } \\
  \llbracket  \Gamma  \rrbracket   \times   (    \mathbf{S}_{  \ell  }   \llbracket  \ottnt{A}  \rrbracket    \times   \mathbf{S}_{  \ell  }   \llbracket  \ottnt{A}  \rrbracket     )   \arrow{r}{  \text{id}   \times   t^{\mathbf{S}_{  \ell  } }_{   \mathbf{S}_{  \ell  }   \llbracket  \ottnt{A}  \rrbracket    ,   \llbracket  \ottnt{A}  \rrbracket  }  } &   \llbracket  \Gamma  \rrbracket   \times   \mathbf{S}_{  \ell  }   (    \mathbf{S}_{  \ell  }   \llbracket  \ottnt{A}  \rrbracket    \times   \llbracket  \ottnt{A}  \rrbracket    )    \arrow{r}{ t^{\mathbf{S}_{  \ell  } }_{  \llbracket  \Gamma  \rrbracket  ,     \mathbf{S}_{  \ell  }   \llbracket  \ottnt{A}  \rrbracket    \times   \llbracket  \ottnt{A}  \rrbracket    } } \arrow[crossing over, near start]{uul}{  \text{id}   \times   \mathbf{S}_{  \ell  }   \pi_2   } &  \mathbf{S}_{  \ell  }   (    \llbracket  \Gamma  \rrbracket   \times   (    \mathbf{S}_{  \ell  }   \llbracket  \ottnt{A}  \rrbracket    \times   \llbracket  \ottnt{A}  \rrbracket    )    )   \arrow[bend right=90,near start]{uu}{ \mathbf{S}_{  \ell  }   (    \text{id}   \times   \pi_2    )  }
\end{tikzcd}
\caption{Commutative diagram}
\label{prfcd8}
\end{figure}

\item \Rref{MC-Join}. Have: $ \Gamma  \vdash   \mathbf{join}^{  \ell_{{\mathrm{1}}}  ,  \ell_{{\mathrm{2}}}  }  \ottnt{a}   :   S_{   \ell_{{\mathrm{1}}}  \vee  \ell_{{\mathrm{2}}}   } \:  \ottnt{A}  $ where $ \Gamma  \vdash  \ottnt{a}  :   S_{  \ell_{{\mathrm{1}}}  } \:   S_{  \ell_{{\mathrm{2}}}  } \:  \ottnt{A}   $.\\
By IH, $ \llbracket   \overline{ \ottnt{a} }   \rrbracket  =  \llbracket  \ottnt{a}  \rrbracket $.\\ Now, $ \overline{   \mathbf{join}^{  \ell_{{\mathrm{1}}}  ,  \ell_{{\mathrm{2}}}  }  \ottnt{a}   }  =  \mathbf{bind} ^{ \ell_{{\mathrm{1}}} } \:  \ottmv{x}  =   \overline{ \ottnt{a} }   \: \mathbf{in} \:   \mathbf{bind} ^{ \ell_{{\mathrm{2}}} } \:  \ottmv{y}  =  \ottmv{x}  \: \mathbf{in} \:   \mathbf{eta} ^{  \ell_{{\mathrm{1}}}  \vee  \ell_{{\mathrm{2}}}  }  \ottmv{y}   $.\\
Then, \begin{align*}
&  \llbracket    \mathbf{bind} ^{ \ell_{{\mathrm{1}}} } \:  \ottmv{x}  =   \overline{ \ottnt{a} }   \: \mathbf{in} \:   \mathbf{bind} ^{ \ell_{{\mathrm{2}}} } \:  \ottmv{y}  =  \ottmv{x}  \: \mathbf{in} \:   \mathbf{eta} ^{  \ell_{{\mathrm{1}}}  \vee  \ell_{{\mathrm{2}}}  }  \ottmv{y}      \rrbracket  \\
= &   k^{ \ell_{{\mathrm{1}}} }_{  \mathbf{S}_{   \ell_{{\mathrm{1}}}  \vee  \ell_{{\mathrm{2}}}   }   \llbracket  \ottnt{A}  \rrbracket   }   \circ   \mathbf{S}_{  \ell_{{\mathrm{1}}}  }   \llbracket    \mathbf{bind} ^{ \ell_{{\mathrm{2}}} } \:  \ottmv{y}  =  \ottmv{x}  \: \mathbf{in} \:   \mathbf{eta} ^{  \ell_{{\mathrm{1}}}  \vee  \ell_{{\mathrm{2}}}  }  \ottmv{y}     \rrbracket    \circ   t^{\mathbf{S}_{  \ell_{{\mathrm{1}}}  } }_{  \llbracket  \Gamma  \rrbracket  ,   \mathbf{S}_{  \ell_{{\mathrm{2}}}  }   \llbracket  \ottnt{A}  \rrbracket   }   \circ   \langle   \text{id}_{  \llbracket  \Gamma  \rrbracket  }   ,   \llbracket   \overline{ \ottnt{a} }   \rrbracket   \rangle   \\
= &   k^{ \ell_{{\mathrm{1}}} }_{  \mathbf{S}_{   \ell_{{\mathrm{1}}}  \vee  \ell_{{\mathrm{2}}}   }   \llbracket  \ottnt{A}  \rrbracket   }   \circ   \mathbf{S}_{  \ell_{{\mathrm{1}}}  }   (      k^{ \ell_{{\mathrm{2}}} }_{  \mathbf{S}_{   \ell_{{\mathrm{1}}}  \vee  \ell_{{\mathrm{2}}}   }   \llbracket  \ottnt{A}  \rrbracket   }   \circ   \mathbf{S}_{  \ell_{{\mathrm{2}}}  }   \llbracket    \mathbf{eta} ^{  \ell_{{\mathrm{1}}}  \vee  \ell_{{\mathrm{2}}}  }  \ottmv{y}    \rrbracket     \circ   t^{\mathbf{S}_{  \ell_{{\mathrm{2}}}  } }_{    \llbracket  \Gamma  \rrbracket   \times   \mathbf{S}_{  \ell_{{\mathrm{2}}}  }   \llbracket  \ottnt{A}  \rrbracket     ,   \llbracket  \ottnt{A}  \rrbracket  }    \circ   \langle   \text{id}_{    \llbracket  \Gamma  \rrbracket   \times   \mathbf{S}_{  \ell_{{\mathrm{2}}}  }   \llbracket  \ottnt{A}  \rrbracket     }   ,   \pi_2   \rangle    )    \\ & \hspace*{30pt} \circ   t^{\mathbf{S}_{  \ell_{{\mathrm{1}}}  } }_{  \llbracket  \Gamma  \rrbracket  ,   \mathbf{S}_{  \ell_{{\mathrm{2}}}  }   \llbracket  \ottnt{A}  \rrbracket   }   \circ   \langle   \text{id}_{  \llbracket  \Gamma  \rrbracket  }   ,   \llbracket  \ottnt{a}  \rrbracket   \rangle   \\
= &   k^{ \ell_{{\mathrm{1}}} }_{  \mathbf{S}_{   \ell_{{\mathrm{1}}}  \vee  \ell_{{\mathrm{2}}}   }   \llbracket  \ottnt{A}  \rrbracket   }   \circ   \mathbf{S}_{  \ell_{{\mathrm{1}}}  }   (      k^{ \ell_{{\mathrm{2}}} }_{  \mathbf{S}_{   \ell_{{\mathrm{1}}}  \vee  \ell_{{\mathrm{2}}}   }   \llbracket  \ottnt{A}  \rrbracket   }   \circ   \mathbf{S}_{  \ell_{{\mathrm{2}}}  }   (     \mathbf{S}^{   \bot    \sqsubseteq    \ell_{{\mathrm{1}}}  \vee  \ell_{{\mathrm{2}}}   }_{  \llbracket  \ottnt{A}  \rrbracket  }   \circ   \eta_{  \llbracket  \ottnt{A}  \rrbracket  }    \circ   \pi_2    )     \circ   t^{\mathbf{S}_{  \ell_{{\mathrm{2}}}  } }_{    \llbracket  \Gamma  \rrbracket   \times   \mathbf{S}_{  \ell_{{\mathrm{2}}}  }   \llbracket  \ottnt{A}  \rrbracket     ,   \llbracket  \ottnt{A}  \rrbracket  }    \circ   \langle   \text{id}   ,   \pi_2   \rangle    )    \\ & \hspace*{30pt} \circ   t^{\mathbf{S}_{  \ell_{{\mathrm{1}}}  } }_{  \llbracket  \Gamma  \rrbracket  ,   \mathbf{S}_{  \ell_{{\mathrm{2}}}  }   \llbracket  \ottnt{A}  \rrbracket   }   \circ   \langle   \text{id}   ,   \llbracket  \ottnt{a}  \rrbracket   \rangle   \\
= &    \mu^{  \ell_{{\mathrm{1}}}  ,   \ell_{{\mathrm{1}}}  \vee  \ell_{{\mathrm{2}}}   }_{  \llbracket  \ottnt{A}  \rrbracket  }   \circ   \mathbf{S}_{  \ell_{{\mathrm{1}}}  }   \mu^{  \ell_{{\mathrm{2}}}  ,   \ell_{{\mathrm{1}}}  \vee  \ell_{{\mathrm{2}}}   }_{  \llbracket  \ottnt{A}  \rrbracket  }     \circ   \mathbf{S}_{  \ell_{{\mathrm{1}}}  }   (     \mathbf{S}_{  \ell_{{\mathrm{2}}}  }   (     \mathbf{S}^{   \bot    \sqsubseteq    \ell_{{\mathrm{1}}}  \vee  \ell_{{\mathrm{2}}}   }_{  \llbracket  \ottnt{A}  \rrbracket  }   \circ   \eta_{  \llbracket  \ottnt{A}  \rrbracket  }    \circ   \pi_2    )    \circ   t^{\mathbf{S}_{  \ell_{{\mathrm{2}}}  } }_{    \llbracket  \Gamma  \rrbracket   \times   \mathbf{S}_{  \ell_{{\mathrm{2}}}  }   \llbracket  \ottnt{A}  \rrbracket     ,   \llbracket  \ottnt{A}  \rrbracket  }    \circ   \langle   \text{id}   ,   \pi_2   \rangle    )    \\ & \hspace*{30pt} \circ   t^{\mathbf{S}_{  \ell_{{\mathrm{1}}}  } }_{  \llbracket  \Gamma  \rrbracket  ,   \mathbf{S}_{  \ell_{{\mathrm{2}}}  }   \llbracket  \ottnt{A}  \rrbracket   }   \circ   \langle   \text{id}   ,   \llbracket  \ottnt{a}  \rrbracket   \rangle   \hspace*{5pt} [\text{By Lemma \ref{Ap2}}]\\
= &       \mu^{  \ell_{{\mathrm{1}}}  ,   \ell_{{\mathrm{1}}}  \vee  \ell_{{\mathrm{2}}}   }_{  \llbracket  \ottnt{A}  \rrbracket  }   \circ   \mathbf{S}_{  \ell_{{\mathrm{1}}}  }   \mu^{  \ell_{{\mathrm{2}}}  ,   \ell_{{\mathrm{1}}}  \vee  \ell_{{\mathrm{2}}}   }_{  \llbracket  \ottnt{A}  \rrbracket  }     \circ   \mathbf{S}_{  \ell_{{\mathrm{1}}}  }   \mathbf{S}_{  \ell_{{\mathrm{2}}}  }   \mathbf{S}^{   \bot    \sqsubseteq    \ell_{{\mathrm{1}}}  \vee  \ell_{{\mathrm{2}}}   }_{  \llbracket  \ottnt{A}  \rrbracket  }      \circ   \mathbf{S}_{  \ell_{{\mathrm{1}}}  }   \mathbf{S}_{  \ell_{{\mathrm{2}}}  }   \eta_{  \llbracket  \ottnt{A}  \rrbracket  }      \circ   \mathbf{S}_{  \ell_{{\mathrm{1}}}  }   \pi_2     \circ   \mathbf{S}_{  \ell_{{\mathrm{1}}}  }   \langle   \text{id}   ,   \pi_2   \rangle    \\ & \hspace*{30pt} \circ   t^{\mathbf{S}_{  \ell_{{\mathrm{1}}}  } }_{  \llbracket  \Gamma  \rrbracket  ,   \mathbf{S}_{  \ell_{{\mathrm{2}}}  }   \llbracket  \ottnt{A}  \rrbracket   }   \circ   \langle   \text{id}   ,   \llbracket  \ottnt{a}  \rrbracket   \rangle   \hspace*{5pt} [\text{By strength}]\\
= &   \mu^{  \ell_{{\mathrm{1}}}  ,   \ell_{{\mathrm{1}}}  \vee  \ell_{{\mathrm{2}}}   }_{  \llbracket  \ottnt{A}  \rrbracket  }   \circ   \mathbf{S}_{  \ell_{{\mathrm{1}}}  }   \mu^{  \ell_{{\mathrm{2}}}  ,   \ell_{{\mathrm{1}}}  \vee  \ell_{{\mathrm{2}}}   }_{  \llbracket  \ottnt{A}  \rrbracket  }    \circ     \mathbf{S}_{  \ell_{{\mathrm{1}}}  }   \mathbf{S}_{  \ell_{{\mathrm{2}}}  }   \mathbf{S}^{   \bot    \sqsubseteq    \ell_{{\mathrm{1}}}  \vee  \ell_{{\mathrm{2}}}   }_{  \llbracket  \ottnt{A}  \rrbracket  }     \circ   \mathbf{S}_{  \ell_{{\mathrm{1}}}  }   \mathbf{S}_{  \ell_{{\mathrm{2}}}  }   \eta_{  \llbracket  \ottnt{A}  \rrbracket  }      \circ   \pi_2    \circ   \langle   \text{id}   ,   \llbracket  \ottnt{a}  \rrbracket   \rangle   \hspace*{3pt} [\text{By strength}]\\
= &   \mu^{  \ell_{{\mathrm{1}}}  ,   \ell_{{\mathrm{1}}}  \vee  \ell_{{\mathrm{2}}}   }_{  \llbracket  \ottnt{A}  \rrbracket  }   \circ   \mathbf{S}_{  \ell_{{\mathrm{1}}}  }   \mu^{  \ell_{{\mathrm{2}}}  ,   \ell_{{\mathrm{1}}}  \vee  \ell_{{\mathrm{2}}}   }_{  \llbracket  \ottnt{A}  \rrbracket  }    \circ    \mathbf{S}_{  \ell_{{\mathrm{1}}}  }   \mathbf{S}_{  \ell_{{\mathrm{2}}}  }   \mathbf{S}^{   \bot    \sqsubseteq    \ell_{{\mathrm{1}}}  \vee  \ell_{{\mathrm{2}}}   }_{  \llbracket  \ottnt{A}  \rrbracket  }     \circ   \mathbf{S}_{  \ell_{{\mathrm{1}}}  }   \mathbf{S}_{  \ell_{{\mathrm{2}}}  }   \eta_{  \llbracket  \ottnt{A}  \rrbracket  }      \circ   \llbracket  \ottnt{a}  \rrbracket  \\
= &   \mu^{  \ell_{{\mathrm{1}}}  ,  \ell_{{\mathrm{2}}}  }_{  \llbracket  \ottnt{A}  \rrbracket  }   \circ   \llbracket  \ottnt{a}  \rrbracket   \hspace*{3pt} [\text{By Figure \ref{prfcd11}}]\\
= &  \llbracket    \mathbf{join}^{  \ell_{{\mathrm{1}}}  ,  \ell_{{\mathrm{2}}}  }  \ottnt{a}    \rrbracket 
\end{align*}

The diagram in Figure \ref{prfcd11} commutes by lax monoidality and naturality.

\begin{figure}
\begin{tikzcd}[row sep = 3em]
 \mathbf{S}_{  \ell_{{\mathrm{1}}}  }   \mathbf{S}_{  \ell_{{\mathrm{2}}}  }   \llbracket  \ottnt{A}  \rrbracket    \arrow{r}{ \mathbf{S}_{  \ell_{{\mathrm{1}}}  }   \mathbf{S}_{  \ell_{{\mathrm{2}}}  }   \eta_{  \llbracket  \ottnt{A}  \rrbracket  }   } \arrow{dr}[left]{ \text{id} } &  \mathbf{S}_{  \ell_{{\mathrm{1}}}  }   \mathbf{S}_{  \ell_{{\mathrm{2}}}  }   \mathbf{S}_{   \bot   }   \llbracket  \ottnt{A}  \rrbracket     \arrow{r}{ \mathbf{S}_{  \ell_{{\mathrm{1}}}  }   \mathbf{S}_{  \ell_{{\mathrm{2}}}  }   \mathbf{S}^{   \bot    \sqsubseteq    \ell_{{\mathrm{1}}}  \vee  \ell_{{\mathrm{2}}}   }_{  \llbracket  \ottnt{A}  \rrbracket  }   } \arrow{d}{ \mathbf{S}_{  \ell_{{\mathrm{1}}}  }   \mu^{  \ell_{{\mathrm{2}}}  ,   \bot   }_{  \llbracket  \ottnt{A}  \rrbracket  }  } &  \mathbf{S}_{  \ell_{{\mathrm{1}}}  }   \mathbf{S}_{  \ell_{{\mathrm{2}}}  }   \mathbf{S}_{   \ell_{{\mathrm{1}}}  \vee  \ell_{{\mathrm{2}}}   }   \llbracket  \ottnt{A}  \rrbracket     \arrow{d}{ \mathbf{S}_{  \ell_{{\mathrm{1}}}  }   \mu^{  \ell_{{\mathrm{2}}}  ,   \ell_{{\mathrm{1}}}  \vee  \ell_{{\mathrm{2}}}   }_{  \llbracket  \ottnt{A}  \rrbracket  }  } \\
&  \mathbf{S}_{  \ell_{{\mathrm{1}}}  }   \mathbf{S}_{  \ell_{{\mathrm{2}}}  }   \llbracket  \ottnt{A}  \rrbracket    \arrow{r}{ \mathbf{S}_{  \ell_{{\mathrm{1}}}  }   \mathbf{S}^{  \ell_{{\mathrm{2}}}   \sqsubseteq    \ell_{{\mathrm{1}}}  \vee  \ell_{{\mathrm{2}}}   }_{  \llbracket  \ottnt{A}  \rrbracket  }  } \arrow{d}{ \mu^{  \ell_{{\mathrm{1}}}  ,  \ell_{{\mathrm{2}}}  }_{  \llbracket  \ottnt{A}  \rrbracket  } } &  \mathbf{S}_{  \ell_{{\mathrm{1}}}  }   \mathbf{S}_{   \ell_{{\mathrm{1}}}  \vee  \ell_{{\mathrm{2}}}   }   \llbracket  \ottnt{A}  \rrbracket    \arrow{d}{ \mu^{  \ell_{{\mathrm{1}}}  ,   \ell_{{\mathrm{1}}}  \vee  \ell_{{\mathrm{2}}}   }_{  \llbracket  \ottnt{A}  \rrbracket  } } \\
&  \mathbf{S}_{   \ell_{{\mathrm{1}}}  \vee  \ell_{{\mathrm{2}}}   }   \llbracket  \ottnt{A}  \rrbracket   \arrow{r}{ \text{id} } &  \mathbf{S}_{   \ell_{{\mathrm{1}}}  \vee  \ell_{{\mathrm{2}}}   }   \llbracket  \ottnt{A}  \rrbracket  
\end{tikzcd}
\caption{Commutative diagram}
\label{prfcd11}
\end{figure}

\item \Rref{MC-Fork}. Have: $ \Gamma  \vdash   \mathbf{fork}^{  \ell_{{\mathrm{1}}}  ,  \ell_{{\mathrm{2}}}  }  \ottnt{a}   :   S_{  \ell_{{\mathrm{1}}}  } \:   S_{  \ell_{{\mathrm{2}}}  } \:  \ottnt{A}   $ where $ \Gamma  \vdash  \ottnt{a}  :   S_{   \ell_{{\mathrm{1}}}  \vee  \ell_{{\mathrm{2}}}   } \:  \ottnt{A}  $.\\
By IH, $ \llbracket   \overline{ \ottnt{a} }   \rrbracket  =  \llbracket  \ottnt{a}  \rrbracket $.\\ Now $ \overline{   \mathbf{fork}^{  \ell_{{\mathrm{1}}}  ,  \ell_{{\mathrm{2}}}  }  \ottnt{a}   }  =  \mathbf{bind} ^{  \ell_{{\mathrm{1}}}  \vee  \ell_{{\mathrm{2}}}  } \:  \ottmv{x}  =   \overline{ \ottnt{a} }   \: \mathbf{in} \:   \mathbf{eta} ^{ \ell_{{\mathrm{1}}} }   \mathbf{eta} ^{ \ell_{{\mathrm{2}}} }  \ottmv{x}   $.

Then, \begin{align*}
&  \llbracket    \mathbf{bind} ^{  \ell_{{\mathrm{1}}}  \vee  \ell_{{\mathrm{2}}}  } \:  \ottmv{x}  =   \overline{ \ottnt{a} }   \: \mathbf{in} \:   \mathbf{eta} ^{ \ell_{{\mathrm{1}}} }   \mathbf{eta} ^{ \ell_{{\mathrm{2}}} }  \ottmv{x}      \rrbracket  \\
= &   k^{  \ell_{{\mathrm{1}}}  \vee  \ell_{{\mathrm{2}}}  }_{  \mathbf{S}_{  \ell_{{\mathrm{1}}}  }   \mathbf{S}_{  \ell_{{\mathrm{2}}}  }   \llbracket  \ottnt{A}  \rrbracket    }   \circ   \mathbf{S}_{   \ell_{{\mathrm{1}}}  \vee  \ell_{{\mathrm{2}}}   }   \llbracket    \mathbf{eta} ^{ \ell_{{\mathrm{1}}} }   \mathbf{eta} ^{ \ell_{{\mathrm{2}}} }  \ottmv{x}     \rrbracket    \circ   t^{\mathbf{S}_{   \ell_{{\mathrm{1}}}  \vee  \ell_{{\mathrm{2}}}   } }_{  \llbracket  \Gamma  \rrbracket  ,   \llbracket  \ottnt{A}  \rrbracket  }   \circ   \langle   \text{id}_{  \llbracket  \Gamma  \rrbracket  }   ,   \llbracket   \overline{ \ottnt{a} }   \rrbracket   \rangle   \\
= &   k^{  \ell_{{\mathrm{1}}}  \vee  \ell_{{\mathrm{2}}}  }_{  \mathbf{S}_{  \ell_{{\mathrm{1}}}  }   \mathbf{S}_{  \ell_{{\mathrm{2}}}  }   \llbracket  \ottnt{A}  \rrbracket    }   \circ   \mathbf{S}_{   \ell_{{\mathrm{1}}}  \vee  \ell_{{\mathrm{2}}}   }   (       \mathbf{S}^{   \bot    \sqsubseteq   \ell_{{\mathrm{1}}}  }_{  \mathbf{S}_{  \ell_{{\mathrm{2}}}  }   \llbracket  \ottnt{A}  \rrbracket   }   \circ   \eta_{  \mathbf{S}_{  \ell_{{\mathrm{2}}}  }   \llbracket  \ottnt{A}  \rrbracket   }    \circ   \mathbf{S}^{   \bot    \sqsubseteq   \ell_{{\mathrm{2}}}  }_{  \llbracket  \ottnt{A}  \rrbracket  }    \circ   \eta_{  \llbracket  \ottnt{A}  \rrbracket  }    \circ   \pi_2    )    \circ   t^{\mathbf{S}_{   \ell_{{\mathrm{1}}}  \vee  \ell_{{\mathrm{2}}}   } }_{  \llbracket  \Gamma  \rrbracket  ,   \llbracket  \ottnt{A}  \rrbracket  }   \circ   \langle   \text{id}_{  \llbracket  \Gamma  \rrbracket  }   ,   \llbracket  \ottnt{a}  \rrbracket   \rangle   \\
= &    k^{  \ell_{{\mathrm{1}}}  \vee  \ell_{{\mathrm{2}}}  }_{  \mathbf{S}_{  \ell_{{\mathrm{1}}}  }   \mathbf{S}_{  \ell_{{\mathrm{2}}}  }   \llbracket  \ottnt{A}  \rrbracket    }   \circ   \mathbf{S}_{   \ell_{{\mathrm{1}}}  \vee  \ell_{{\mathrm{2}}}   }   (      \mathbf{S}^{   \bot    \sqsubseteq   \ell_{{\mathrm{1}}}  }_{  \mathbf{S}_{  \ell_{{\mathrm{2}}}  }   \llbracket  \ottnt{A}  \rrbracket   }   \circ   \eta_{  \mathbf{S}_{  \ell_{{\mathrm{2}}}  }   \llbracket  \ottnt{A}  \rrbracket   }    \circ   \mathbf{S}^{   \bot    \sqsubseteq   \ell_{{\mathrm{2}}}  }_{  \llbracket  \ottnt{A}  \rrbracket  }    \circ   \eta_{  \llbracket  \ottnt{A}  \rrbracket  }    )     \circ   \llbracket  \ottnt{a}  \rrbracket   \hspace*{3pt} [\because \mathbf{S}_{ \ell_{{\mathrm{1}}}  \vee  \ell_{{\mathrm{2}}} } \text{ is strong}] \\
= &    k^{ \ell_{{\mathrm{1}}} }_{  \mathbf{S}_{  \ell_{{\mathrm{1}}}  }   \mathbf{S}_{  \ell_{{\mathrm{2}}}  }   \llbracket  \ottnt{A}  \rrbracket    }   \circ   \mathbf{S}_{  \ell_{{\mathrm{1}}}  }   k^{ \ell_{{\mathrm{2}}} }_{  \mathbf{S}_{  \ell_{{\mathrm{1}}}  }   \mathbf{S}_{  \ell_{{\mathrm{2}}}  }   \llbracket  \ottnt{A}  \rrbracket    }     \circ   \delta^{  \ell_{{\mathrm{1}}}  ,  \ell_{{\mathrm{2}}}  }_{  \mathbf{S}_{  \ell_{{\mathrm{1}}}  }   \mathbf{S}_{  \ell_{{\mathrm{2}}}  }   \llbracket  \ottnt{A}  \rrbracket    }   \circ   \mathbf{S}_{   \ell_{{\mathrm{1}}}  \vee  \ell_{{\mathrm{2}}}   }   (      \mathbf{S}^{   \bot    \sqsubseteq   \ell_{{\mathrm{1}}}  }_{  \mathbf{S}_{  \ell_{{\mathrm{2}}}  }   \llbracket  \ottnt{A}  \rrbracket   }   \circ   \eta_{  \mathbf{S}_{  \ell_{{\mathrm{2}}}  }   \llbracket  \ottnt{A}  \rrbracket   }    \circ   \mathbf{S}^{   \bot    \sqsubseteq   \ell_{{\mathrm{2}}}  }_{  \llbracket  \ottnt{A}  \rrbracket  }    \circ   \eta_{  \llbracket  \ottnt{A}  \rrbracket  }    )    \circ   \llbracket  \ottnt{a}  \rrbracket  \\ & \hspace*{30pt} [\text{By Lemma \ref{Ap2}}] \\
= &    k^{ \ell_{{\mathrm{1}}} }_{  \mathbf{S}_{  \ell_{{\mathrm{1}}}  }   \mathbf{S}_{  \ell_{{\mathrm{2}}}  }   \llbracket  \ottnt{A}  \rrbracket    }   \circ   \mathbf{S}_{  \ell_{{\mathrm{1}}}  }   (     \mathbf{S}_{  \ell_{{\mathrm{1}}}  }   k^{ \ell_{{\mathrm{2}}} }_{  \mathbf{S}_{  \ell_{{\mathrm{2}}}  }   \llbracket  \ottnt{A}  \rrbracket   }    \circ   \delta^{  \ell_{{\mathrm{1}}}  ,  \ell_{{\mathrm{2}}}  }_{  \mathbf{S}_{  \ell_{{\mathrm{2}}}  }   \llbracket  \ottnt{A}  \rrbracket   }    \circ   \mu^{  \ell_{{\mathrm{2}}}  ,  \ell_{{\mathrm{1}}}  }_{  \mathbf{S}_{  \ell_{{\mathrm{2}}}  }   \llbracket  \ottnt{A}  \rrbracket   }    )     \circ   \delta^{  \ell_{{\mathrm{1}}}  ,  \ell_{{\mathrm{2}}}  }_{  \mathbf{S}_{  \ell_{{\mathrm{1}}}  }   \mathbf{S}_{  \ell_{{\mathrm{2}}}  }   \llbracket  \ottnt{A}  \rrbracket    }  \\ & \hspace*{30pt} \circ   \mathbf{S}_{   \ell_{{\mathrm{1}}}  \vee  \ell_{{\mathrm{2}}}   }   (      \mathbf{S}^{   \bot    \sqsubseteq   \ell_{{\mathrm{1}}}  }_{  \mathbf{S}_{  \ell_{{\mathrm{2}}}  }   \llbracket  \ottnt{A}  \rrbracket   }   \circ   \eta_{  \mathbf{S}_{  \ell_{{\mathrm{2}}}  }   \llbracket  \ottnt{A}  \rrbracket   }    \circ   \mathbf{S}^{   \bot    \sqsubseteq   \ell_{{\mathrm{2}}}  }_{  \llbracket  \ottnt{A}  \rrbracket  }    \circ   \eta_{  \llbracket  \ottnt{A}  \rrbracket  }    )    \circ   \llbracket  \ottnt{a}  \rrbracket   \hspace*{3pt} [\text{By Lemma \ref{Ap2}}]\\
= &    \mu^{  \ell_{{\mathrm{1}}}  ,  \ell_{{\mathrm{1}}}  }_{  \mathbf{S}_{  \ell_{{\mathrm{2}}}  }   \llbracket  \ottnt{A}  \rrbracket   }   \circ   \mathbf{S}_{  \ell_{{\mathrm{1}}}  }   (     \mathbf{S}_{  \ell_{{\mathrm{1}}}  }   \mu^{  \ell_{{\mathrm{2}}}  ,  \ell_{{\mathrm{2}}}  }_{  \llbracket  \ottnt{A}  \rrbracket  }    \circ   \delta^{  \ell_{{\mathrm{1}}}  ,  \ell_{{\mathrm{2}}}  }_{  \mathbf{S}_{  \ell_{{\mathrm{2}}}  }   \llbracket  \ottnt{A}  \rrbracket   }    \circ   \mu^{  \ell_{{\mathrm{2}}}  ,  \ell_{{\mathrm{1}}}  }_{  \mathbf{S}_{  \ell_{{\mathrm{2}}}  }   \llbracket  \ottnt{A}  \rrbracket   }    )     \circ   \delta^{  \ell_{{\mathrm{1}}}  ,  \ell_{{\mathrm{2}}}  }_{  \mathbf{S}_{  \ell_{{\mathrm{1}}}  }   \mathbf{S}_{  \ell_{{\mathrm{2}}}  }   \llbracket  \ottnt{A}  \rrbracket    }   \\ & \hspace*{30pt} \circ   \mathbf{S}_{   \ell_{{\mathrm{1}}}  \vee  \ell_{{\mathrm{2}}}   }   (      \mathbf{S}^{   \bot    \sqsubseteq   \ell_{{\mathrm{1}}}  }_{  \mathbf{S}_{  \ell_{{\mathrm{2}}}  }   \llbracket  \ottnt{A}  \rrbracket   }   \circ   \eta_{  \mathbf{S}_{  \ell_{{\mathrm{2}}}  }   \llbracket  \ottnt{A}  \rrbracket   }    \circ   \mathbf{S}^{   \bot    \sqsubseteq   \ell_{{\mathrm{2}}}  }_{  \llbracket  \ottnt{A}  \rrbracket  }    \circ   \eta_{  \llbracket  \ottnt{A}  \rrbracket  }    )    \circ   \llbracket  \ottnt{a}  \rrbracket   \hspace*{3pt} [\text{By Lemma \ref{Ap2}}]\\
= &    \mathbf{S}_{  \ell_{{\mathrm{1}}}  }   \mu^{   \bot   ,  \ell_{{\mathrm{2}}}  }_{  \llbracket  \ottnt{A}  \rrbracket  }    \circ   \mathbf{S}_{  \ell_{{\mathrm{1}}}  }   \mathbf{S}_{   \bot   }   \mu^{  \ell_{{\mathrm{2}}}  ,  \ell_{{\mathrm{2}}}  }_{  \llbracket  \ottnt{A}  \rrbracket  }      \circ   \mathbf{S}_{  \ell_{{\mathrm{1}}}  }   \delta^{   \bot   ,  \ell_{{\mathrm{2}}}  }_{  \mathbf{S}_{  \ell_{{\mathrm{2}}}  }   \llbracket  \ottnt{A}  \rrbracket   }    \circ     \mathbf{S}_{  \ell_{{\mathrm{1}}}  }   \mathbf{S}_{  \ell_{{\mathrm{2}}}  }   \mathbf{S}^{   \bot    \sqsubseteq   \ell_{{\mathrm{2}}}  }_{  \llbracket  \ottnt{A}  \rrbracket  }     \circ   \mathbf{S}_{  \ell_{{\mathrm{1}}}  }   \mathbf{S}_{  \ell_{{\mathrm{2}}}  }   \eta_{  \llbracket  \ottnt{A}  \rrbracket  }      \circ   \delta^{  \ell_{{\mathrm{1}}}  ,  \ell_{{\mathrm{2}}}  }_{  \llbracket  \ottnt{A}  \rrbracket  }    \circ   \llbracket  \ottnt{a}  \rrbracket   \hspace*{3pt} [\text{By Figure \ref{prfcd9}}]\\
= &   \delta^{  \ell_{{\mathrm{1}}}  ,  \ell_{{\mathrm{2}}}  }_{  \llbracket  \ottnt{A}  \rrbracket  }   \circ   \llbracket  \ottnt{a}  \rrbracket   \hspace*{3pt} [\text{By Figure \ref{prfcd10}}]\\
= &  \llbracket    \mathbf{fork}^{  \ell_{{\mathrm{1}}}  ,  \ell_{{\mathrm{2}}}  }  \ottnt{a}    \rrbracket 
\end{align*}

The diagram in Figure \ref{prfcd9} commutes: all the squares commute by naturality; the ellipses commute because $ \mathbf{S}_{  \ell_{{\mathrm{1}}}  }   \mu^{   \bot   ,  \ell_{{\mathrm{2}}}  }_{  \llbracket  \ottnt{A}  \rrbracket  }   =  \mu^{  \ell_{{\mathrm{1}}}  ,   \bot   }_{  \mathbf{S}_{  \ell_{{\mathrm{2}}}  }   \llbracket  \ottnt{A}  \rrbracket   }  =  \mathbf{S}_{  \ell_{{\mathrm{1}}}  }   \epsilon_{  \mathbf{S}_{  \ell_{{\mathrm{2}}}  }   \llbracket  \ottnt{A}  \rrbracket   }  $; the circular segment commutes by lax monoidality. The diagram in Figure \ref{prfcd10} commutes because $\mathbf{S}$ is a strong monoidal functor.

\begin{figure}
\begin{tikzcd}[row sep = 3em]
 \mathbf{S}_{   \ell_{{\mathrm{1}}}  \vee  \ell_{{\mathrm{2}}}   }   \llbracket  \ottnt{A}  \rrbracket   \arrow{r}{ \mathbf{S}_{   \ell_{{\mathrm{1}}}  \vee  \ell_{{\mathrm{2}}}   }   \eta_{  \llbracket  \ottnt{A}  \rrbracket  }  } \arrow{d}{ \delta^{  \ell_{{\mathrm{1}}}  ,  \ell_{{\mathrm{2}}}  }_{  \llbracket  \ottnt{A}  \rrbracket  } } &  \mathbf{S}_{   \ell_{{\mathrm{1}}}  \vee  \ell_{{\mathrm{2}}}   }   \mathbf{S}_{   \bot   }   \llbracket  \ottnt{A}  \rrbracket    \arrow{r}{ \mathbf{S}_{   \ell_{{\mathrm{1}}}  \vee  \ell_{{\mathrm{2}}}   }   \mathbf{S}^{   \bot    \sqsubseteq   \ell_{{\mathrm{2}}}  }_{  \llbracket  \ottnt{A}  \rrbracket  }  } \arrow{d}{ \delta^{  \ell_{{\mathrm{1}}}  ,  \ell_{{\mathrm{2}}}  }_{  \mathbf{S}_{   \bot   }   \llbracket  \ottnt{A}  \rrbracket   } } &  \mathbf{S}_{   \ell_{{\mathrm{1}}}  \vee  \ell_{{\mathrm{2}}}   }   \mathbf{S}_{  \ell_{{\mathrm{2}}}  }   \llbracket  \ottnt{A}  \rrbracket    \arrow{r}{ \mathbf{S}_{   \ell_{{\mathrm{1}}}  \vee  \ell_{{\mathrm{2}}}   }   \eta_{  \mathbf{S}_{  \ell_{{\mathrm{2}}}  }   \llbracket  \ottnt{A}  \rrbracket   }  } \arrow{d}{ \delta^{  \ell_{{\mathrm{1}}}  ,  \ell_{{\mathrm{2}}}  }_{  \mathbf{S}_{  \ell_{{\mathrm{2}}}  }   \llbracket  \ottnt{A}  \rrbracket   } } &  \mathbf{S}_{   \ell_{{\mathrm{1}}}  \vee  \ell_{{\mathrm{2}}}   }   \mathbf{S}_{   \bot   }   \mathbf{S}_{  \ell_{{\mathrm{2}}}  }   \llbracket  \ottnt{A}  \rrbracket     \arrow{r}{ \mathbf{S}_{   \ell_{{\mathrm{1}}}  \vee  \ell_{{\mathrm{2}}}   }   \mathbf{S}^{   \bot    \sqsubseteq   \ell_{{\mathrm{1}}}  }_{  \mathbf{S}_{  \ell_{{\mathrm{2}}}  }   \llbracket  \ottnt{A}  \rrbracket   }  }  \arrow{d}{ \delta^{  \ell_{{\mathrm{1}}}  ,  \ell_{{\mathrm{2}}}  }_{  \mathbf{S}_{   \bot   }   \mathbf{S}_{  \ell_{{\mathrm{2}}}  }   \llbracket  \ottnt{A}  \rrbracket    } } &  \mathbf{S}_{   \ell_{{\mathrm{1}}}  \vee  \ell_{{\mathrm{2}}}   }   \mathbf{S}_{  \ell_{{\mathrm{1}}}  }   \mathbf{S}_{  \ell_{{\mathrm{2}}}  }   \llbracket  \ottnt{A}  \rrbracket     \arrow{d}{ \delta^{  \ell_{{\mathrm{1}}}  ,  \ell_{{\mathrm{2}}}  }_{  \mathbf{S}_{  \ell_{{\mathrm{1}}}  }   \mathbf{S}_{  \ell_{{\mathrm{2}}}  }   \llbracket  \ottnt{A}  \rrbracket    } }\\
 \mathbf{S}_{  \ell_{{\mathrm{1}}}  }   \mathbf{S}_{  \ell_{{\mathrm{2}}}  }   \llbracket  \ottnt{A}  \rrbracket    \arrow{r}{ \mathbf{S}_{  \ell_{{\mathrm{1}}}  }   \mathbf{S}_{  \ell_{{\mathrm{2}}}  }   \eta_{  \llbracket  \ottnt{A}  \rrbracket  }   } &  \mathbf{S}_{  \ell_{{\mathrm{1}}}  }   \mathbf{S}_{  \ell_{{\mathrm{2}}}  }   \mathbf{S}_{   \bot   }   \llbracket  \ottnt{A}  \rrbracket     \arrow{r}{ \mathbf{S}_{  \ell_{{\mathrm{1}}}  }   \mathbf{S}_{  \ell_{{\mathrm{2}}}  }   \mathbf{S}^{   \bot    \sqsubseteq   \ell_{{\mathrm{2}}}  }_{  \llbracket  \ottnt{A}  \rrbracket  }   } &  \mathbf{S}_{  \ell_{{\mathrm{1}}}  }   \mathbf{S}_{  \ell_{{\mathrm{2}}}  }   \mathbf{S}_{  \ell_{{\mathrm{2}}}  }   \llbracket  \ottnt{A}  \rrbracket     \arrow{r}{ \mathbf{S}_{  \ell_{{\mathrm{1}}}  }   \mathbf{S}_{  \ell_{{\mathrm{2}}}  }   \eta_{  \mathbf{S}_{  \ell_{{\mathrm{2}}}  }   \llbracket  \ottnt{A}  \rrbracket   }   } \arrow[bend right=90]{dr}{ \text{id} } &  \mathbf{S}_{  \ell_{{\mathrm{1}}}  }   \mathbf{S}_{  \ell_{{\mathrm{2}}}  }   \mathbf{S}_{   \bot   }   \mathbf{S}_{  \ell_{{\mathrm{2}}}  }   \llbracket  \ottnt{A}  \rrbracket      \arrow{r}{ \mathbf{S}_{  \ell_{{\mathrm{1}}}  }   \mathbf{S}_{  \ell_{{\mathrm{2}}}  }   \mathbf{S}^{   \bot    \sqsubseteq   \ell_{{\mathrm{1}}}  }_{  \mathbf{S}_{  \ell_{{\mathrm{2}}}  }   \llbracket  \ottnt{A}  \rrbracket   }   } \arrow[bend left = 30]{d}{ \mathbf{S}_{  \ell_{{\mathrm{1}}}  }   \mu^{  \ell_{{\mathrm{2}}}  ,   \bot   }_{  \mathbf{S}_{  \ell_{{\mathrm{2}}}  }   \llbracket  \ottnt{A}  \rrbracket   }  } \arrow[bend right=30]{d}[left]{ \mathbf{S}_{  \ell_{{\mathrm{1}}}  }   \mathbf{S}_{  \ell_{{\mathrm{2}}}  }   \mu^{   \bot   ,  \ell_{{\mathrm{2}}}  }_{  \llbracket  \ottnt{A}  \rrbracket  }   } &  \mathbf{S}_{  \ell_{{\mathrm{1}}}  }   \mathbf{S}_{  \ell_{{\mathrm{2}}}  }   \mathbf{S}_{  \ell_{{\mathrm{1}}}  }   \mathbf{S}_{  \ell_{{\mathrm{2}}}  }   \llbracket  \ottnt{A}  \rrbracket      \arrow{d}{ \mathbf{S}_{  \ell_{{\mathrm{1}}}  }   \mu^{  \ell_{{\mathrm{2}}}  ,  \ell_{{\mathrm{1}}}  }_{  \mathbf{S}_{  \ell_{{\mathrm{2}}}  }   \llbracket  \ottnt{A}  \rrbracket   }  } \\
& & &  \mathbf{S}_{  \ell_{{\mathrm{1}}}  }   \mathbf{S}_{  \ell_{{\mathrm{2}}}  }   \mathbf{S}_{  \ell_{{\mathrm{2}}}  }   \llbracket  \ottnt{A}  \rrbracket     \arrow{r}[below]{ \mathbf{S}_{  \ell_{{\mathrm{1}}}  }   \mathbf{S}^{  \ell_{{\mathrm{2}}}   \sqsubseteq    \ell_{{\mathrm{2}}}  \vee  \ell_{{\mathrm{1}}}   }_{  \mathbf{S}_{  \ell_{{\mathrm{2}}}  }   \llbracket  \ottnt{A}  \rrbracket   }  } \arrow{d}{ \mathbf{S}_{  \ell_{{\mathrm{1}}}  }   \delta^{   \bot   ,  \ell_{{\mathrm{2}}}  }_{  \mathbf{S}_{  \ell_{{\mathrm{2}}}  }   \llbracket  \ottnt{A}  \rrbracket   }  } &  \mathbf{S}_{  \ell_{{\mathrm{1}}}  }   \mathbf{S}_{   \ell_{{\mathrm{2}}}  \vee  \ell_{{\mathrm{1}}}   }   \mathbf{S}_{  \ell_{{\mathrm{2}}}  }   \llbracket  \ottnt{A}  \rrbracket     \arrow{d}{ \mathbf{S}_{  \ell_{{\mathrm{1}}}  }   \delta^{  \ell_{{\mathrm{1}}}  ,  \ell_{{\mathrm{2}}}  }_{  \mathbf{S}_{  \ell_{{\mathrm{2}}}  }   \llbracket  \ottnt{A}  \rrbracket   }  } \\
& & &  \mathbf{S}_{  \ell_{{\mathrm{1}}}  }   \mathbf{S}_{   \bot   }   \mathbf{S}_{  \ell_{{\mathrm{2}}}  }   \mathbf{S}_{  \ell_{{\mathrm{2}}}  }   \llbracket  \ottnt{A}  \rrbracket      \arrow{r}{ \mathbf{S}_{  \ell_{{\mathrm{1}}}  }   \mathbf{S}^{   \bot    \sqsubseteq   \ell_{{\mathrm{1}}}  }_{  \mathbf{S}_{  \ell_{{\mathrm{2}}}  }   \mathbf{S}_{  \ell_{{\mathrm{2}}}  }   \llbracket  \ottnt{A}  \rrbracket    }  } \arrow{d}{ \mathbf{S}_{  \ell_{{\mathrm{1}}}  }   \mathbf{S}_{   \bot   }   \mu^{  \ell_{{\mathrm{2}}}  ,  \ell_{{\mathrm{2}}}  }_{  \llbracket  \ottnt{A}  \rrbracket  }   } &  \mathbf{S}_{  \ell_{{\mathrm{1}}}  }   \mathbf{S}_{  \ell_{{\mathrm{1}}}  }   \mathbf{S}_{  \ell_{{\mathrm{2}}}  }   \mathbf{S}_{  \ell_{{\mathrm{2}}}  }   \llbracket  \ottnt{A}  \rrbracket      \arrow{d}{ \mathbf{S}_{  \ell_{{\mathrm{1}}}  }   \mathbf{S}_{  \ell_{{\mathrm{1}}}  }   \mu^{  \ell_{{\mathrm{2}}}  ,  \ell_{{\mathrm{2}}}  }_{  \llbracket  \ottnt{A}  \rrbracket  }   }\\
& & &  \mathbf{S}_{  \ell_{{\mathrm{1}}}  }   \mathbf{S}_{   \bot   }   \mathbf{S}_{  \ell_{{\mathrm{2}}}  }   \llbracket  \ottnt{A}  \rrbracket     \arrow{r}{ \mathbf{S}_{  \ell_{{\mathrm{1}}}  }   \mathbf{S}^{   \bot    \sqsubseteq   \ell_{{\mathrm{1}}}  }_{  \mathbf{S}_{  \ell_{{\mathrm{2}}}  }   \llbracket  \ottnt{A}  \rrbracket   }  } \arrow[bend left = 30]{d}{ \mu^{  \ell_{{\mathrm{1}}}  ,   \bot   }_{  \mathbf{S}_{  \ell_{{\mathrm{2}}}  }   \llbracket  \ottnt{A}  \rrbracket   } } \arrow[bend right=30]{d}[left]{ \mathbf{S}_{  \ell_{{\mathrm{1}}}  }   \mu^{   \bot   ,  \ell_{{\mathrm{2}}}  }_{  \llbracket  \ottnt{A}  \rrbracket  }  } &  \mathbf{S}_{  \ell_{{\mathrm{1}}}  }   \mathbf{S}_{  \ell_{{\mathrm{1}}}  }   \mathbf{S}_{  \ell_{{\mathrm{2}}}  }   \llbracket  \ottnt{A}  \rrbracket     \arrow{d}{ \mu^{  \ell_{{\mathrm{1}}}  ,  \ell_{{\mathrm{1}}}  }_{  \mathbf{S}_{  \ell_{{\mathrm{2}}}  }   \llbracket  \ottnt{A}  \rrbracket   } } \\
& & &  \mathbf{S}_{  \ell_{{\mathrm{1}}}  }   \mathbf{S}_{  \ell_{{\mathrm{2}}}  }   \llbracket  \ottnt{A}  \rrbracket    \arrow{r}{ \text{id} } &  \mathbf{S}_{  \ell_{{\mathrm{1}}}  }   \mathbf{S}_{  \ell_{{\mathrm{2}}}  }   \llbracket  \ottnt{A}  \rrbracket   
\end{tikzcd}
\caption{Commutative diagram}
\label{prfcd9}
\end{figure}

\begin{figure}
\begin{tikzcd}[row sep = 3 em]
 \mathbf{S}_{  \ell_{{\mathrm{1}}}  }   \mathbf{S}_{  \ell_{{\mathrm{2}}}  }   \llbracket  \ottnt{A}  \rrbracket    \arrow{r}{ \mathbf{S}_{  \ell_{{\mathrm{1}}}  }   \mathbf{S}_{  \ell_{{\mathrm{2}}}  }   \eta_{  \llbracket  \ottnt{A}  \rrbracket  }   } \arrow{dr}{ \text{id} } &  \mathbf{S}_{  \ell_{{\mathrm{1}}}  }   \mathbf{S}_{  \ell_{{\mathrm{2}}}  }   \mathbf{S}_{   \bot   }   \llbracket  \ottnt{A}  \rrbracket     \arrow{r}{ \mathbf{S}_{  \ell_{{\mathrm{1}}}  }   \mathbf{S}_{  \ell_{{\mathrm{2}}}  }   \mathbf{S}^{   \bot    \sqsubseteq   \ell_{{\mathrm{2}}}  }_{  \llbracket  \ottnt{A}  \rrbracket  }   } \arrow{d}{ \mathbf{S}_{  \ell_{{\mathrm{1}}}  }   \mu^{  \ell_{{\mathrm{2}}}  ,   \bot   }_{  \llbracket  \ottnt{A}  \rrbracket  }  } &  \mathbf{S}_{  \ell_{{\mathrm{1}}}  }   \mathbf{S}_{  \ell_{{\mathrm{2}}}  }   \mathbf{S}_{  \ell_{{\mathrm{2}}}  }   \llbracket  \ottnt{A}  \rrbracket     \arrow{r}{ \mathbf{S}_{  \ell_{{\mathrm{1}}}  }   \delta^{   \bot   ,  \ell_{{\mathrm{2}}}  }_{  \mathbf{S}_{  \ell_{{\mathrm{2}}}  }   \llbracket  \ottnt{A}  \rrbracket   }  } \arrow{dr}{ \text{id} } &  \mathbf{S}_{  \ell_{{\mathrm{1}}}  }   \mathbf{S}_{   \bot   }   \mathbf{S}_{  \ell_{{\mathrm{2}}}  }   \mathbf{S}_{  \ell_{{\mathrm{2}}}  }   \llbracket  \ottnt{A}  \rrbracket      \arrow{r}{ \mathbf{S}_{  \ell_{{\mathrm{1}}}  }   \mathbf{S}_{   \bot   }   \mu^{  \ell_{{\mathrm{2}}}  ,  \ell_{{\mathrm{2}}}  }_{  \llbracket  \ottnt{A}  \rrbracket  }   } \arrow{d}{ \mathbf{S}_{  \ell_{{\mathrm{1}}}  }   \mu^{   \bot   ,  \ell_{{\mathrm{2}}}  }_{  \mathbf{S}_{  \ell_{{\mathrm{2}}}  }   \llbracket  \ottnt{A}  \rrbracket   }  } &  \mathbf{S}_{  \ell_{{\mathrm{1}}}  }   \mathbf{S}_{   \bot   }   \mathbf{S}_{  \ell_{{\mathrm{2}}}  }   \llbracket  \ottnt{A}  \rrbracket     \arrow{d}{ \mathbf{S}_{  \ell_{{\mathrm{1}}}  }   \mu^{   \bot   ,  \ell_{{\mathrm{2}}}  }_{  \llbracket  \ottnt{A}  \rrbracket  }  }\\
&  \mathbf{S}_{  \ell_{{\mathrm{1}}}  }   \mathbf{S}_{  \ell_{{\mathrm{2}}}  }   \llbracket  \ottnt{A}  \rrbracket    \arrow[bend right=10]{rrr}{ \text{id} } & &  \mathbf{S}_{  \ell_{{\mathrm{1}}}  }   \mathbf{S}_{  \ell_{{\mathrm{2}}}  }   \mathbf{S}_{  \ell_{{\mathrm{2}}}  }   \llbracket  \ottnt{A}  \rrbracket     \arrow{r}{ \mathbf{S}_{  \ell_{{\mathrm{1}}}  }   \mu^{  \ell_{{\mathrm{2}}}  ,  \ell_{{\mathrm{2}}}  }_{  \llbracket  \ottnt{A}  \rrbracket  }  } &  \mathbf{S}_{  \ell_{{\mathrm{1}}}  }   \mathbf{S}_{  \ell_{{\mathrm{2}}}  }   \llbracket  \ottnt{A}  \rrbracket   
\end{tikzcd}
\caption{Commutative diagram}
\label{prfcd10}
\end{figure}

\item \Rref{MC-Up}. Have: $ \Gamma  \vdash   \mathbf{up}^{  \ell_{{\mathrm{1}}}  ,  \ell_{{\mathrm{2}}}  }  \ottnt{a}   :   S_{  \ell_{{\mathrm{2}}}  } \:  \ottnt{A}  $ where $ \Gamma  \vdash  \ottnt{a}  :   S_{  \ell_{{\mathrm{1}}}  } \:  \ottnt{A}  $ and $ \ell_{{\mathrm{1}}}  \sqsubseteq  \ell_{{\mathrm{2}}} $.\\
By IH, $ \llbracket   \overline{ \ottnt{a} }   \rrbracket  =  \llbracket  \ottnt{a}  \rrbracket $.\\ Now, $ \overline{   \mathbf{up}^{  \ell_{{\mathrm{1}}}  ,  \ell_{{\mathrm{2}}}  }  \ottnt{a}   }  =  \mathbf{bind} ^{ \ell_{{\mathrm{1}}} } \:  \ottmv{x}  =   \overline{ \ottnt{a} }   \: \mathbf{in} \:   \mathbf{eta} ^{ \ell_{{\mathrm{2}}} }  \ottmv{x}  $.\\
Then, \begin{align*}
&  \llbracket    \mathbf{bind} ^{ \ell_{{\mathrm{1}}} } \:  \ottmv{x}  =   \overline{ \ottnt{a} }   \: \mathbf{in} \:   \mathbf{eta} ^{ \ell_{{\mathrm{2}}} }  \ottmv{x}     \rrbracket  \\
= &     k^{ \ell_{{\mathrm{1}}} }_{  \mathbf{S}_{  \ell_{{\mathrm{2}}}  }   \llbracket  \ottnt{A}  \rrbracket   }   \circ   \mathbf{S}_{  \ell_{{\mathrm{1}}}  }   \llbracket    \mathbf{eta} ^{ \ell_{{\mathrm{2}}} }  \ottmv{x}    \rrbracket     \circ   t^{\mathbf{S}_{  \ell_{{\mathrm{1}}}  } }_{  \llbracket  \Gamma  \rrbracket  ,   \llbracket  \ottnt{A}  \rrbracket  }    \circ   \langle   \text{id}_{  \llbracket  \Gamma  \rrbracket  }   ,   \llbracket  \ottnt{a}  \rrbracket   \rangle   \\
= &     k^{ \ell_{{\mathrm{1}}} }_{  \mathbf{S}_{  \ell_{{\mathrm{2}}}  }   \llbracket  \ottnt{A}  \rrbracket   }   \circ   \mathbf{S}_{  \ell_{{\mathrm{1}}}  }   (     \mathbf{S}^{   \bot    \sqsubseteq   \ell_{{\mathrm{2}}}  }_{  \llbracket  \ottnt{A}  \rrbracket  }   \circ   \eta_{  \llbracket  \ottnt{A}  \rrbracket  }    \circ   \pi_2    )     \circ   t^{\mathbf{S}_{  \ell_{{\mathrm{1}}}  } }_{  \llbracket  \Gamma  \rrbracket  ,   \llbracket  \ottnt{A}  \rrbracket  }    \circ   \langle   \text{id}_{  \llbracket  \Gamma  \rrbracket  }   ,   \llbracket  \ottnt{a}  \rrbracket   \rangle   \\
= &     \mu^{  \ell_{{\mathrm{1}}}  ,  \ell_{{\mathrm{2}}}  }_{  \llbracket  \ottnt{A}  \rrbracket  }   \circ   \mathbf{S}_{  \ell_{{\mathrm{1}}}  }   (     \mathbf{S}^{   \bot    \sqsubseteq   \ell_{{\mathrm{2}}}  }_{  \llbracket  \ottnt{A}  \rrbracket  }   \circ   \eta_{  \llbracket  \ottnt{A}  \rrbracket  }    \circ   \pi_2    )     \circ   t^{\mathbf{S}_{  \ell_{{\mathrm{1}}}  } }_{  \llbracket  \Gamma  \rrbracket  ,   \llbracket  \ottnt{A}  \rrbracket  }    \circ   \langle   \text{id}_{  \llbracket  \Gamma  \rrbracket  }   ,   \llbracket  \ottnt{a}  \rrbracket   \rangle   \hspace*{3pt} [\text{By Theorem \ref{Ap2}}]\\
= &      \mu^{  \ell_{{\mathrm{1}}}  ,  \ell_{{\mathrm{2}}}  }_{  \llbracket  \ottnt{A}  \rrbracket  }   \circ   \mathbf{S}_{  \ell_{{\mathrm{1}}}  }   (     \mathbf{S}^{  \ell_{{\mathrm{1}}}   \sqsubseteq   \ell_{{\mathrm{2}}}  }_{  \llbracket  \ottnt{A}  \rrbracket  }   \circ   \mathbf{S}^{   \bot    \sqsubseteq   \ell_{{\mathrm{1}}}  }_{  \llbracket  \ottnt{A}  \rrbracket  }    \circ   \eta_{  \llbracket  \ottnt{A}  \rrbracket  }    )     \circ   \mathbf{S}_{  \ell_{{\mathrm{1}}}  }   \pi_2     \circ   t^{\mathbf{S}_{  \ell_{{\mathrm{1}}}  } }_{  \llbracket  \Gamma  \rrbracket  ,   \llbracket  \ottnt{A}  \rrbracket  }    \circ   \langle   \text{id}_{  \llbracket  \Gamma  \rrbracket  }   ,   \llbracket  \ottnt{a}  \rrbracket   \rangle   \\
= &      \mu^{  \ell_{{\mathrm{1}}}  ,  \ell_{{\mathrm{2}}}  }_{  \llbracket  \ottnt{A}  \rrbracket  }   \circ   \mathbf{S}_{  \ell_{{\mathrm{1}}}  }   \mathbf{S}^{  \ell_{{\mathrm{1}}}   \sqsubseteq   \ell_{{\mathrm{2}}}  }_{  \llbracket  \ottnt{A}  \rrbracket  }     \circ   \mathbf{S}_{  \ell_{{\mathrm{1}}}  }   (    \mathbf{S}^{   \bot    \sqsubseteq   \ell_{{\mathrm{1}}}  }_{  \llbracket  \ottnt{A}  \rrbracket  }   \circ   \eta_{  \llbracket  \ottnt{A}  \rrbracket  }    )     \circ   \pi_2    \circ   \langle   \text{id}_{  \llbracket  \Gamma  \rrbracket  }   ,   \llbracket  \ottnt{a}  \rrbracket   \rangle   \hspace*{3pt} [\because \mathbf{S}_{\ell_{{\mathrm{1}}}} \text{ is strong} ] \\
= &     \mathbf{S}^{  \ell_{{\mathrm{1}}}   \sqsubseteq   \ell_{{\mathrm{2}}}  }_{  \llbracket  \ottnt{A}  \rrbracket  }   \circ   \mu^{  \ell_{{\mathrm{1}}}  ,  \ell_{{\mathrm{1}}}  }_{  \llbracket  \ottnt{A}  \rrbracket  }    \circ   \mathbf{S}_{  \ell_{{\mathrm{1}}}  }   (    \mathbf{S}^{   \bot    \sqsubseteq   \ell_{{\mathrm{1}}}  }_{  \llbracket  \ottnt{A}  \rrbracket  }   \circ   \eta_{  \llbracket  \ottnt{A}  \rrbracket  }    )     \circ   \llbracket  \ottnt{a}  \rrbracket   \hspace*{3pt} [\text{By naturality}]\\
= &   \mathbf{S}^{  \ell_{{\mathrm{1}}}   \sqsubseteq   \ell_{{\mathrm{2}}}  }_{  \llbracket  \ottnt{A}  \rrbracket  }   \circ   \llbracket  \ottnt{a}  \rrbracket   \hspace*{3pt} [\because (\mathbf{S}_{\ell_{{\mathrm{1}}}},\mathbf{S}^{  \bot   \sqsubseteq  \ell_{{\mathrm{1}}} } \circ \eta, \mu^{\ell_{{\mathrm{1}}},\ell_{{\mathrm{1}}}}) \text{ is a monad}]\\
= &  \llbracket    \mathbf{up}^{  \ell_{{\mathrm{1}}}  ,  \ell_{{\mathrm{2}}}  }  \ottnt{a}    \rrbracket 
\end{align*}

\end{itemize}

\end{proof}


\begin{theorem}\label{equiv2}
If $ \Gamma  \vdash  \ottnt{a}  :  \ottnt{A} $ in \ED{}($ \mathcal{L} $), then $ \llbracket   \underline{ \ottnt{a} }   \rrbracket  =  \llbracket  \ottnt{a}  \rrbracket $.
\end{theorem}

\begin{proof}
By induction on $ \Gamma  \vdash  \ottnt{a}  :  \ottnt{A} $. Note that $ \llbracket   \underline{ \ottnt{A} }   \rrbracket  =  \llbracket  \ottnt{A}  \rrbracket $.

\begin{itemize}
\item $\lambda$-calculus. By IH.

\item \Rref{DCC-Eta}. Have: $ \Gamma  \vdash   \mathbf{eta} ^{ \ell }  \ottnt{a}   :   \mathcal{T}_{ \ell } \:  \ottnt{A}  $ where $ \Gamma  \vdash  \ottnt{a}  :  \ottnt{A} $.\\
By IH, $ \llbracket   \underline{ \ottnt{a} }   \rrbracket  =  \llbracket  \ottnt{a}  \rrbracket $.\\ Now, $ \underline{   \mathbf{eta} ^{ \ell }  \ottnt{a}   }  =  \mathbf{up}^{   \bot   ,  \ell  }   (   \ottkw{ret}  \:   \underline{ \ottnt{a} }    )  $.\\
Then, \begin{align*}
&  \llbracket    \mathbf{up}^{   \bot   ,  \ell  }   (   \ottkw{ret}  \:   \underline{ \ottnt{a} }    )     \rrbracket \\
= &    \mathbf{S}^{   \bot    \sqsubseteq   \ell  }_{  \llbracket  \ottnt{A}  \rrbracket  }   \circ   \eta_{  \llbracket  \ottnt{A}  \rrbracket  }    \circ   \llbracket   \underline{ \ottnt{a} }   \rrbracket   \\
= &    \mathbf{S}^{   \bot    \sqsubseteq   \ell  }_{  \llbracket  \ottnt{A}  \rrbracket  }   \circ   \eta_{  \llbracket  \ottnt{A}  \rrbracket  }    \circ   \llbracket  \ottnt{a}  \rrbracket   \hspace*{3pt} [\text{By IH}]\\
= &  \llbracket    \mathbf{eta} ^{ \ell }  \ottnt{a}    \rrbracket  
\end{align*} 

\item \Rref{DCC-Bind}. Have: $ \Gamma  \vdash   \mathbf{bind} ^{ \ell } \:  \ottmv{x}  =  \ottnt{a}  \: \mathbf{in} \:  \ottnt{b}   :  \ottnt{B} $ where $ \Gamma  \vdash  \ottnt{a}  :   \mathcal{T}_{ \ell } \:  \ottnt{A}  $ and $  \Gamma  ,   \ottmv{x}  :  \ottnt{A}    \vdash  \ottnt{b}  :  \ottnt{B} $ and $ \ell  \sqsubseteq  \ottnt{B} $.

By IH, $ \llbracket   \underline{ \ottnt{a} }   \rrbracket  =  \llbracket  \ottnt{a}  \rrbracket $ and $ \llbracket   \underline{ \ottnt{b} }   \rrbracket  =  \llbracket  \ottnt{b}  \rrbracket $. \\
Now, $ \underline{   \mathbf{bind} ^{ \ell } \:  \ottmv{x}  =  \ottnt{a}  \: \mathbf{in} \:  \ottnt{b}   }  =   j^{ \ell }_{  \underline{ \ottnt{B} }  }   \:   (    (   \mathbf{lift}^{  \ell  }   (   \lambda  \ottmv{x}  .   \underline{ \ottnt{b} }    )    )   \:   \underline{ \ottnt{a} }    )  $.\\
Then, \begin{align*}
&  \llbracket     j^{ \ell }_{  \underline{ \ottnt{B} }  }   \:   (    (   \mathbf{lift}^{  \ell  }   (   \lambda  \ottmv{x}  .   \underline{ \ottnt{b} }    )    )   \:   \underline{ \ottnt{a} }    )     \rrbracket \\
= &   \text{app}   \circ   \langle    \llbracket    j^{ \ell }_{  \underline{ \ottnt{B} }  }    \rrbracket   \circ   \langle \rangle    ,   \llbracket     (   \mathbf{lift}^{  \ell  }   (   \lambda  \ottmv{x}  .   \underline{ \ottnt{b} }    )    )   \:   \underline{ \ottnt{a} }     \rrbracket   \rangle   \\
= &   \text{app}   \circ   \langle    \llbracket    j^{ \ell }_{  \underline{ \ottnt{B} }  }    \rrbracket   \circ   \langle \rangle    ,    \text{app}   \circ   \langle   \llbracket    \mathbf{lift}^{  \ell  }   (   \lambda  \ottmv{x}  .   \underline{ \ottnt{b} }    )     \rrbracket   ,   \llbracket   \underline{ \ottnt{a} }   \rrbracket   \rangle    \rangle   \\
= &   \text{app}   \circ   \langle    \llbracket    j^{ \ell }_{  \underline{ \ottnt{B} }  }    \rrbracket   \circ   \langle \rangle    ,    \text{app}   \circ   \langle   \Lambda \Big(     \mathbf{S}_{  \ell  }   (   \Lambda^{-1}   \llbracket   \lambda  \ottmv{x}  .   \underline{ \ottnt{b} }    \rrbracket    )    \circ   t^{\mathbf{S}_{  \ell  } }_{  \llbracket  \Gamma  \rrbracket  ,   \llbracket  \ottnt{A}  \rrbracket  }     \Big)   ,   \llbracket   \underline{ \ottnt{a} }   \rrbracket   \rangle    \rangle   \\
= &   \text{app}   \circ   \langle    \llbracket    j^{ \ell }_{  \underline{ \ottnt{B} }  }    \rrbracket   \circ   \langle \rangle    ,    \text{app}   \circ   \langle   \Lambda \Big(     \mathbf{S}_{  \ell  }   \llbracket   \underline{ \ottnt{b} }   \rrbracket    \circ   t^{\mathbf{S}_{  \ell  } }_{  \llbracket  \Gamma  \rrbracket  ,   \llbracket  \ottnt{A}  \rrbracket  }     \Big)   ,   \llbracket   \underline{ \ottnt{a} }   \rrbracket   \rangle    \rangle   \\
= &   \text{app}   \circ   \langle    \llbracket    j^{ \ell }_{  \underline{ \ottnt{B} }  }    \rrbracket   \circ   \langle \rangle    ,    \text{app}   \circ   \langle   \Lambda \Big(     \mathbf{S}_{  \ell  }   \llbracket  \ottnt{b}  \rrbracket    \circ   t^{\mathbf{S}_{  \ell  } }_{  \llbracket  \Gamma  \rrbracket  ,   \llbracket  \ottnt{A}  \rrbracket  }     \Big)   ,   \llbracket  \ottnt{a}  \rrbracket   \rangle    \rangle   \hspace*{3pt} [\text{By IH}]\\
= &   \text{app}   \circ   \langle    \llbracket    j^{ \ell }_{  \underline{ \ottnt{B} }  }    \rrbracket   \circ   \langle \rangle    ,     \mathbf{S}_{  \ell  }   \llbracket  \ottnt{b}  \rrbracket    \circ   t^{\mathbf{S}_{  \ell  } }_{  \llbracket  \Gamma  \rrbracket  ,   \llbracket  \ottnt{A}  \rrbracket  }    \circ   \langle   \text{id}_{  \llbracket  \Gamma  \rrbracket  }   ,   \llbracket  \ottnt{a}  \rrbracket   \rangle    \rangle  \\
= &   \Lambda^{-1}   \llbracket    j^{ \ell }_{  \underline{ \ottnt{B} }  }    \rrbracket    \circ   \langle   \langle \rangle   ,     \mathbf{S}_{  \ell  }   \llbracket  \ottnt{b}  \rrbracket    \circ   t^{\mathbf{S}_{  \ell  } }_{  \llbracket  \Gamma  \rrbracket  ,   \llbracket  \ottnt{A}  \rrbracket  }    \circ   \langle   \text{id}_{  \llbracket  \Gamma  \rrbracket  }   ,   \llbracket  \ottnt{a}  \rrbracket   \rangle    \rangle  \\
= &      \Lambda^{-1}   \llbracket    j^{ \ell }_{  \underline{ \ottnt{B} }  }    \rrbracket    \circ   \langle   \langle \rangle   ,   \text{id}_{  \mathbf{S}_{  \ell  }   \llbracket  \ottnt{B}  \rrbracket   }   \rangle    \circ   \mathbf{S}_{  \ell  }   \llbracket  \ottnt{b}  \rrbracket     \circ   t^{\mathbf{S}_{  \ell  } }_{  \llbracket  \Gamma  \rrbracket  ,   \llbracket  \ottnt{A}  \rrbracket  }    \circ   \langle   \text{id}_{  \llbracket  \Gamma  \rrbracket  }   ,   \llbracket  \ottnt{a}  \rrbracket   \rangle  \\
= &     k^{ \ell }_{  \llbracket  \ottnt{B}  \rrbracket  }   \circ   \mathbf{S}_{  \ell  }   \llbracket  \ottnt{b}  \rrbracket     \circ   t^{\mathbf{S}_{  \ell  } }_{  \llbracket  \Gamma  \rrbracket  ,   \llbracket  \ottnt{A}  \rrbracket  }    \circ   \langle   \text{id}_{  \llbracket  \Gamma  \rrbracket  }   ,   \llbracket  \ottnt{a}  \rrbracket   \rangle   \hspace*{3pt} [\text{By Lemma \ref{jk}}]\\
= &  \llbracket    \mathbf{bind} ^{ \ell } \:  \ottmv{x}  =  \ottnt{a}  \: \mathbf{in} \:  \ottnt{b}    \rrbracket 
\end{align*} 

\end{itemize}
\end{proof}


\begin{lemma}\label{jk}
If $ \ell  \sqsubseteq  \ottnt{A} $ in \ED{}, then $ k^{ \ell }_{  \llbracket  \ottnt{A}  \rrbracket  }  =   \Lambda^{-1}   \llbracket    j^{ \ell }_{  \underline{ \ottnt{A} }  }    \rrbracket    \circ   \langle   \langle \rangle   ,   \text{id}_{  \mathbf{S}_{  \ell  }   \llbracket  \ottnt{A}  \rrbracket   }   \rangle  $. 
\end{lemma}

\begin{proof}
By induction on $ \ell  \sqsubseteq  \ottnt{A} $.

\begin{itemize}

\item \Rref{Prot-Prod}. Have $ \ell  \sqsubseteq   \ottnt{A}  \times  \ottnt{B}  $ where $ \ell  \sqsubseteq  \ottnt{A} $ and $ \ell  \sqsubseteq  \ottnt{B} $.\\
By IH, $ k^{ \ell }_{  \llbracket  \ottnt{A}  \rrbracket  }  =   \Lambda^{-1}   \llbracket    j^{ \ell }_{  \underline{ \ottnt{A} }  }    \rrbracket    \circ   \langle   \langle \rangle   ,   \text{id}_{  \mathbf{S}_{  \ell  }   \llbracket  \ottnt{A}  \rrbracket   }   \rangle  $ and $ k^{ \ell }_{  \llbracket  \ottnt{B}  \rrbracket  }  =   \Lambda^{-1}   \llbracket    j^{ \ell }_{  \underline{ \ottnt{B} }  }    \rrbracket    \circ   \langle   \langle \rangle   ,   \text{id}_{  \mathbf{S}_{  \ell  }   \llbracket  \ottnt{B}  \rrbracket   }   \rangle  $.\\
Now, $ k^{ \ell }_{    \llbracket  \ottnt{A}  \rrbracket   \times   \llbracket  \ottnt{B}  \rrbracket    }  =     k^{ \ell }_{  \llbracket  \ottnt{A}  \rrbracket  }    \times    k^{ \ell }_{  \llbracket  \ottnt{B}  \rrbracket  }     \circ   \langle   \mathbf{S}_{  \ell  }   \pi_1    ,   \mathbf{S}_{  \ell  }   \pi_2    \rangle  $.\\
Next, $ j^{ \ell }_{    \underline{ \ottnt{A} }   \times   \underline{ \ottnt{B} }    }  =  \lambda  \ottmv{z}  .   (     j^{ \ell }_{  \underline{ \ottnt{A} }  }    \:   (    (   \mathbf{lift}^{  \ell  }   (   \lambda  \ottmv{y}  .   \mathbf{proj}_1 \:  \ottmv{y}    )    )   \:  \ottmv{z}   )    ,     j^{ \ell }_{  \underline{ \ottnt{B} }  }    \:   (    (   \mathbf{lift}^{  \ell  }   (   \lambda  \ottmv{y}  .   \mathbf{proj}_2 \:  \ottmv{y}    )    )   \:  \ottmv{z}   )    )  $. 

Then, 
\begin{align*}
&   \Lambda^{-1}   \llbracket    \lambda  \ottmv{z}  .   (     j^{ \ell }_{  \underline{ \ottnt{A} }  }    \:   (    (   \mathbf{lift}^{  \ell  }   (   \lambda  \ottmv{y}  .   \mathbf{proj}_1 \:  \ottmv{y}    )    )   \:  \ottmv{z}   )    ,     j^{ \ell }_{  \underline{ \ottnt{B} }  }    \:   (    (   \mathbf{lift}^{  \ell  }   (   \lambda  \ottmv{y}  .   \mathbf{proj}_2 \:  \ottmv{y}    )    )   \:  \ottmv{z}   )    )     \rrbracket    \circ   \langle   \langle \rangle   ,   \text{id}   \rangle   \\
= &   \langle   \llbracket      j^{ \ell }_{  \underline{ \ottnt{A} }  }    \:   (    (   \mathbf{lift}^{  \ell  }   (   \lambda  \ottmv{y}  .   \mathbf{proj}_1 \:  \ottmv{y}    )    )   \:  \ottmv{z}   )     \rrbracket   ,   \llbracket      j^{ \ell }_{  \underline{ \ottnt{B} }  }    \:   (    (   \mathbf{lift}^{  \ell  }   (   \lambda  \ottmv{y}  .   \mathbf{proj}_2 \:  \ottmv{y}    )    )   \:  \ottmv{z}   )     \rrbracket   \rangle   \circ   \langle   \langle \rangle   ,   \text{id}   \rangle   \\
= &   \langle    \text{app}   \circ   \langle    \llbracket    j^{ \ell }_{  \underline{ \ottnt{A} }  }    \rrbracket   \circ   \langle \rangle    ,   \llbracket     (   \mathbf{lift}^{  \ell  }   (   \lambda  \ottmv{y}  .   \mathbf{proj}_1 \:  \ottmv{y}    )    )   \:  \ottmv{z}    \rrbracket   \rangle    ,    \text{app}   \circ   \langle    \llbracket    j^{ \ell }_{  \underline{ \ottnt{B} }  }    \rrbracket   \circ   \langle \rangle    ,   \llbracket     (   \mathbf{lift}^{  \ell  }   (   \lambda  \ottmv{y}  .   \mathbf{proj}_2 \:  \ottmv{y}    )    )   \:  \ottmv{z}    \rrbracket   \rangle    \rangle   \circ   \langle   \langle \rangle   ,   \text{id}   \rangle   \\
= &   \langle    \Lambda^{-1}   \llbracket    j^{ \ell }_{  \underline{ \ottnt{A} }  }    \rrbracket    \circ   \langle   \langle \rangle   ,   \llbracket     (   \mathbf{lift}^{  \ell  }   (   \lambda  \ottmv{y}  .   \mathbf{proj}_1 \:  \ottmv{y}    )    )   \:  \ottmv{z}    \rrbracket   \rangle    ,    \Lambda^{-1}   \llbracket    j^{ \ell }_{  \underline{ \ottnt{B} }  }    \rrbracket    \circ   \langle   \langle \rangle   ,   \llbracket     (   \mathbf{lift}^{  \ell  }   (   \lambda  \ottmv{y}  .   \mathbf{proj}_2 \:  \ottmv{y}    )    )   \:  \ottmv{z}    \rrbracket   \rangle    \rangle   \circ   \langle   \langle \rangle   ,   \text{id}   \rangle   \\
= &   \langle    k^{ \ell }_{  \llbracket  \ottnt{A}  \rrbracket  }   \circ   \llbracket     (   \mathbf{lift}^{  \ell  }   (   \lambda  \ottmv{y}  .   \mathbf{proj}_1 \:  \ottmv{y}    )    )   \:  \ottmv{z}    \rrbracket    ,    k^{ \ell }_{  \llbracket  \ottnt{B}  \rrbracket  }   \circ   \llbracket     (   \mathbf{lift}^{  \ell  }   (   \lambda  \ottmv{y}  .   \mathbf{proj}_2 \:  \ottmv{y}    )    )   \:  \ottmv{z}    \rrbracket    \rangle   \circ   \langle   \langle \rangle   ,   \text{id}   \rangle   \hspace*{3pt} [\text{By IH}]\\
= &   \langle     k^{ \ell }_{  \llbracket  \ottnt{A}  \rrbracket  }   \circ   \text{app}    \circ   \langle   \llbracket    \mathbf{lift}^{  \ell  }   (   \lambda  \ottmv{y}  .   \mathbf{proj}_1 \:  \ottmv{y}    )     \rrbracket   ,   \pi_2   \rangle    ,     k^{ \ell }_{  \llbracket  \ottnt{B}  \rrbracket  }   \circ   \text{app}    \circ   \langle   \llbracket    \mathbf{lift}^{  \ell  }   (   \lambda  \ottmv{y}  .   \mathbf{proj}_2 \:  \ottmv{y}    )     \rrbracket   ,   \pi_2   \rangle    \rangle   \circ   \langle   \langle \rangle   ,   \text{id}   \rangle   \\
= &  \langle     k^{ \ell }_{  \llbracket  \ottnt{A}  \rrbracket  }   \circ   \text{app}    \circ   \langle   \Lambda \Big(     \mathbf{S}_{  \ell  }   (   \Lambda^{-1}   \llbracket    \lambda  \ottmv{y}  .   \mathbf{proj}_1 \:  \ottmv{y}     \rrbracket    )    \circ   t^{\mathbf{S}_{  \ell  } }_{ \ottnt{X} ,  \ottnt{Y} }     \Big)   ,   \pi_2   \rangle    ,     k^{ \ell }_{  \llbracket  \ottnt{B}  \rrbracket  }   \circ   \text{app}    \circ   \langle   \Lambda \Big(     \mathbf{S}_{  \ell  }   (   \Lambda^{-1}   \llbracket    \lambda  \ottmv{y}  .   \mathbf{proj}_2 \:  \ottmv{y}     \rrbracket    )    \circ   t^{\mathbf{S}_{  \ell  } }_{ \ottnt{X} ,  \ottnt{Y} }     \Big)   ,   \pi_2   \rangle    \rangle  \\ & \hspace*{30pt} \circ  \langle   \langle \rangle   ,   \text{id}   \rangle  \hspace*{5pt} [ \ottnt{X} :=    \top   \times   \mathbf{S}_{  \ell  }   (    \llbracket  \ottnt{A}  \rrbracket   \times   \llbracket  \ottnt{B}  \rrbracket    )     \text{ and } \ottnt{Y} :=   \llbracket  \ottnt{A}  \rrbracket   \times   \llbracket  \ottnt{B}  \rrbracket   ]\\
= &    \langle     k^{ \ell }_{  \llbracket  \ottnt{A}  \rrbracket  }   \circ   \text{app}    \circ   \langle   \Lambda \Big(     \mathbf{S}_{  \ell  }   (    \pi_1   \circ   \pi_2    )    \circ   t^{\mathbf{S}_{  \ell  } }_{ \ottnt{X} ,  \ottnt{Y} }     \Big)   ,   \pi_2   \rangle    ,     k^{ \ell }_{  \llbracket  \ottnt{B}  \rrbracket  }   \circ   \text{app}    \circ   \langle   \Lambda \Big(     \mathbf{S}_{  \ell  }   (    \pi_2   \circ   \pi_2    )    \circ   t^{\mathbf{S}_{  \ell  } }_{ \ottnt{X} ,  \ottnt{Y} }     \Big)   ,   \pi_2   \rangle    \rangle   \circ   \langle   \langle \rangle   ,   \text{id}   \rangle   \\ 
= &   \langle     k^{ \ell }_{  \llbracket  \ottnt{A}  \rrbracket  }   \circ   \text{app}    \circ   \langle   \Lambda   (    \mathbf{S}_{  \ell  }   \pi_1    \circ   \pi_2    )    ,   \pi_2   \rangle    ,     k^{ \ell }_{  \llbracket  \ottnt{B}  \rrbracket  }   \circ   \text{app}    \circ   \langle   \Lambda   (    \mathbf{S}_{  \ell  }   \pi_2    \circ   \pi_2    )    ,   \pi_2   \rangle    \rangle   \circ   \langle   \langle \rangle   ,   \text{id}   \rangle   \\ 
= &   \langle       k^{ \ell }_{  \llbracket  \ottnt{A}  \rrbracket  }   \circ   \text{app}    \circ   \Lambda   (    \mathbf{S}_{  \ell  }   \pi_1    \circ   \pi_2    )     \times   \text{id}    \circ   \langle   \text{id}   ,   \pi_2   \rangle    ,       k^{ \ell }_{  \llbracket  \ottnt{B}  \rrbracket  }   \circ   \text{app}    \circ   \Lambda   (    \mathbf{S}_{  \ell  }   \pi_2    \circ   \pi_2    )     \times   \text{id}    \circ   \langle   \text{id}   ,   \pi_2   \rangle    \rangle   \circ   \langle   \langle \rangle   ,   \text{id}   \rangle   \\ 
= &   \langle      k^{ \ell }_{  \llbracket  \ottnt{A}  \rrbracket  }   \circ   \mathbf{S}_{  \ell  }   \pi_1     \circ   \pi_2    \circ   \langle   \text{id}   ,   \pi_2   \rangle    ,      k^{ \ell }_{  \llbracket  \ottnt{B}  \rrbracket  }   \circ   \mathbf{S}_{  \ell  }   \pi_2     \circ   \pi_2    \circ   \langle   \text{id}   ,   \pi_2   \rangle    \rangle   \circ   \langle   \langle \rangle   ,   \text{id}   \rangle   \\
= &     \langle    k^{ \ell }_{  \llbracket  \ottnt{A}  \rrbracket  }   \circ   \mathbf{S}_{  \ell  }   \pi_1     ,    k^{ \ell }_{  \llbracket  \ottnt{B}  \rrbracket  }   \circ   \mathbf{S}_{  \ell  }   \pi_2     \rangle   \circ   \pi_2    \circ   \langle   \text{id}   ,   \pi_2   \rangle    \circ   \langle   \langle \rangle   ,   \text{id}   \rangle   \\
= &  \langle    k^{ \ell }_{  \llbracket  \ottnt{A}  \rrbracket  }   \circ   \mathbf{S}_{  \ell  }   \pi_1     ,    k^{ \ell }_{  \llbracket  \ottnt{B}  \rrbracket  }   \circ   \mathbf{S}_{  \ell  }   \pi_2     \rangle  \\
= &  k^{ \ell }_{    \llbracket  \ottnt{A}  \rrbracket   \times   \llbracket  \ottnt{B}  \rrbracket    } 
\end{align*}

\item \Rref{Prot-Fun}. Have: $ \ell  \sqsubseteq   \ottnt{A}  \to  \ottnt{B}  $ where $ \ell  \sqsubseteq  \ottnt{B} $.\\
By IH, $ k^{ \ell }_{  \llbracket  \ottnt{B}  \rrbracket  }  =   \Lambda^{-1}   \llbracket    j^{ \ell }_{  \underline{ \ottnt{B} }  }    \rrbracket    \circ   \langle   \langle \rangle   ,   \text{id}_{  \mathbf{S}_{  \ell  }   \llbracket  \ottnt{B}  \rrbracket   }   \rangle  $.\\
Now, $ k^{ \ell }_{    \llbracket  \ottnt{B}  \rrbracket  ^{  \llbracket  \ottnt{A}  \rrbracket  }   }  =  \Lambda   (       k^{ \ell }_{  \llbracket  \ottnt{B}  \rrbracket  }   \circ   \mathbf{S}_{  \ell  }   \text{app}     \circ   \mathbf{S}_{  \ell  }   \langle   \pi_2   ,   \pi_1   \rangle     \circ   t^{\mathbf{S}_{  \ell  } }_{  \llbracket  \ottnt{A}  \rrbracket  ,     \llbracket  \ottnt{B}  \rrbracket  ^{  \llbracket  \ottnt{A}  \rrbracket  }   }    \circ   \langle   \pi_2   ,   \pi_1   \rangle    )  $.\\
And $ j^{ \ell }_{    \underline{ \ottnt{A} }   \to   \underline{ \ottnt{B} }    }  =   \lambda  \ottmv{z}  .   \lambda  \ottmv{y}  .   j^{ \ell }_{  \underline{ \ottnt{B} }  }     \:   (    (   \mathbf{lift}^{  \ell  }   (    \lambda  \ottmv{x}  .  \ottmv{x}   \:  \ottmv{y}   )    )   \:  \ottmv{z}   )  $.

Then,
\begin{align*}
&   \Lambda^{-1}   \llbracket     \lambda  \ottmv{z}  .   \lambda  \ottmv{y}  .   j^{ \ell }_{  \underline{ \ottnt{B} }  }     \:   (    (   \mathbf{lift}^{  \ell  }   (    \lambda  \ottmv{x}  .  \ottmv{x}   \:  \ottmv{y}   )    )   \:  \ottmv{z}   )     \rrbracket    \circ   \langle   \langle \rangle   ,   \text{id}   \rangle   \\
= &   \llbracket     \lambda  \ottmv{y}  .   j^{ \ell }_{  \underline{ \ottnt{B} }  }    \:   (    (   \mathbf{lift}^{  \ell  }   (    \lambda  \ottmv{x}  .  \ottmv{x}   \:  \ottmv{y}   )    )   \:  \ottmv{z}   )     \rrbracket   \circ   \langle   \langle \rangle   ,   \text{id}   \rangle   \\
= &   \Lambda \Big(     \text{app}   \circ   \langle    \llbracket    j^{ \ell }_{  \underline{ \ottnt{B} }  }    \rrbracket   \circ   \langle \rangle    ,   \llbracket     (   \mathbf{lift}^{  \ell  }   (    \lambda  \ottmv{x}  .  \ottmv{x}   \:  \ottmv{y}   )    )   \:  \ottmv{z}    \rrbracket   \rangle     \Big)   \circ   \langle   \langle \rangle   ,   \text{id}   \rangle   \\
= &   \Lambda \Big(      \text{app}   \circ   (    \llbracket    j^{ \ell }_{  \underline{ \ottnt{B} }  }    \rrbracket   \times   \text{id}    )    \circ   \langle   \langle \rangle   ,   \llbracket     (   \mathbf{lift}^{  \ell  }   (    \lambda  \ottmv{x}  .  \ottmv{x}   \:  \ottmv{y}   )    )   \:  \ottmv{z}    \rrbracket   \rangle     \Big)   \circ   \langle   \langle \rangle   ,   \text{id}   \rangle   \\
= &   \Lambda \Big(     \Lambda^{-1}   \llbracket    j^{ \ell }_{  \underline{ \ottnt{B} }  }    \rrbracket    \circ   \langle   \langle \rangle   ,    \text{app}   \circ   \langle   \llbracket    \mathbf{lift}^{  \ell  }   (    \lambda  \ottmv{x}  .  \ottmv{x}   \:  \ottmv{y}   )     \rrbracket   ,    \pi_2   \circ   \pi_1    \rangle    \rangle     \Big)   \circ   \langle   \langle \rangle   ,   \text{id}   \rangle   \\
= &   \Lambda \Big(     \Lambda^{-1}   \llbracket    j^{ \ell }_{  \underline{ \ottnt{B} }  }    \rrbracket    \circ   \langle   \langle \rangle   ,    \text{app}   \circ   \langle   \Lambda \Big(     \mathbf{S}_{  \ell  }   (   \Lambda^{-1}   \llbracket     \lambda  \ottmv{x}  .  \ottmv{x}   \:  \ottmv{y}    \rrbracket    )    \circ   t^{\mathbf{S}_{  \ell  } }_{    (    \top   \times   \mathbf{S}_{  \ell  }     \llbracket  \ottnt{B}  \rrbracket  ^{  \llbracket  \ottnt{A}  \rrbracket  }      )   \times   \llbracket  \ottnt{A}  \rrbracket    ,     \llbracket  \ottnt{B}  \rrbracket  ^{  \llbracket  \ottnt{A}  \rrbracket  }   }     \Big)   ,    \pi_2   \circ   \pi_1    \rangle    \rangle     \Big)   \circ   \langle   \langle \rangle   ,   \text{id}   \rangle   \\
= &   \Lambda \Big(     \Lambda^{-1}   \llbracket    j^{ \ell }_{  \underline{ \ottnt{B} }  }    \rrbracket    \circ   \langle   \langle \rangle   ,    \text{app}   \circ   \langle   \Lambda \Big(     \mathbf{S}_{  \ell  }   (    \text{app}   \circ   \langle   \pi_2   ,    \pi_2   \circ   \pi_1    \rangle    )    \circ   t^{\mathbf{S}_{  \ell  } }_{    (    \top   \times   \mathbf{S}_{  \ell  }     \llbracket  \ottnt{B}  \rrbracket  ^{  \llbracket  \ottnt{A}  \rrbracket  }      )   \times   \llbracket  \ottnt{A}  \rrbracket    ,     \llbracket  \ottnt{B}  \rrbracket  ^{  \llbracket  \ottnt{A}  \rrbracket  }   }     \Big)   ,    \pi_2   \circ   \pi_1    \rangle    \rangle     \Big)   \circ   \langle   \langle \rangle   ,   \text{id}   \rangle   \\
= &   \Lambda \Big(     \Lambda^{-1}   \llbracket    j^{ \ell }_{  \underline{ \ottnt{B} }  }    \rrbracket    \circ   \langle   \langle \rangle   ,       \mathbf{S}_{  \ell  }   \text{app}    \circ   \mathbf{S}_{  \ell  }   \langle   \pi_2   ,   \pi_1   \rangle     \circ   \mathbf{S}_{  \ell  }   \langle    \pi_2   \circ   \pi_1    ,   \pi_2   \rangle     \circ   t^{\mathbf{S}_{  \ell  } }_{    (    \top   \times   \mathbf{S}_{  \ell  }     \llbracket  \ottnt{B}  \rrbracket  ^{  \llbracket  \ottnt{A}  \rrbracket  }      )   \times   \llbracket  \ottnt{A}  \rrbracket    ,     \llbracket  \ottnt{B}  \rrbracket  ^{  \llbracket  \ottnt{A}  \rrbracket  }   }    \circ   \langle   \text{id}   ,    \pi_2   \circ   \pi_1    \rangle    \rangle     \Big)   \circ   \langle   \langle \rangle   ,   \text{id}   \rangle   \\
= &   \Lambda \Big(     \Lambda^{-1}   \llbracket    j^{ \ell }_{  \underline{ \ottnt{B} }  }    \rrbracket    \circ   \langle   \langle \rangle   ,       \mathbf{S}_{  \ell  }   \text{app}    \circ   \mathbf{S}_{  \ell  }   \langle   \pi_2   ,   \pi_1   \rangle     \circ   t^{\mathbf{S}_{  \ell  } }_{  \llbracket  \ottnt{A}  \rrbracket  ,     \llbracket  \ottnt{B}  \rrbracket  ^{  \llbracket  \ottnt{A}  \rrbracket  }   }    \circ   \langle    \pi_2   \circ   \pi_1    ,   \pi_2   \rangle    \circ   \langle   \text{id}   ,    \pi_2   \circ   \pi_1    \rangle    \rangle     \Big)   \circ   \langle   \langle \rangle   ,   \text{id}   \rangle   \hspace*{3pt} [\text{By naturality}]\\
= &  \Lambda \Big(      \Lambda^{-1}   \llbracket    j^{ \ell }_{  \underline{ \ottnt{B} }  }    \rrbracket    \circ   \langle   \langle \rangle   ,      \mathbf{S}_{  \ell  }   \text{app}    \circ   \mathbf{S}_{  \ell  }   \langle   \pi_2   ,   \pi_1   \rangle     \circ   t^{\mathbf{S}_{  \ell  } }_{  \llbracket  \ottnt{A}  \rrbracket  ,     \llbracket  \ottnt{B}  \rrbracket  ^{  \llbracket  \ottnt{A}  \rrbracket  }   }    \circ   \langle   \pi_2   ,    \pi_2   \circ   \pi_1    \rangle    \rangle    \circ   \langle   \langle   \langle \rangle   ,   \pi_1   \rangle   ,   \pi_2   \rangle     \Big) \\
= &  \Lambda \Big(     \Lambda^{-1}   \llbracket    j^{ \ell }_{  \underline{ \ottnt{B} }  }    \rrbracket    \circ   \langle   \langle \rangle   ,      \mathbf{S}_{  \ell  }   \text{app}    \circ   \mathbf{S}_{  \ell  }   \langle   \pi_2   ,   \pi_1   \rangle     \circ   t^{\mathbf{S}_{  \ell  } }_{  \llbracket  \ottnt{A}  \rrbracket  ,     \llbracket  \ottnt{B}  \rrbracket  ^{  \llbracket  \ottnt{A}  \rrbracket  }   }    \circ   \langle   \pi_2   ,   \pi_1   \rangle    \rangle     \Big) \\
= &  \Lambda \Big(         \Lambda^{-1}   \llbracket    j^{ \ell }_{  \underline{ \ottnt{B} }  }    \rrbracket    \circ   \langle   \langle \rangle   ,   \text{id}   \rangle    \circ   \mathbf{S}_{  \ell  }   \text{app}     \circ   \mathbf{S}_{  \ell  }   \langle   \pi_2   ,   \pi_1   \rangle     \circ   t^{\mathbf{S}_{  \ell  } }_{  \llbracket  \ottnt{A}  \rrbracket  ,     \llbracket  \ottnt{B}  \rrbracket  ^{  \llbracket  \ottnt{A}  \rrbracket  }   }    \circ   \langle   \pi_2   ,   \pi_1   \rangle     \Big) \\
= &  \Lambda \Big(        k^{ \ell }_{  \llbracket  \ottnt{B}  \rrbracket  }   \circ   \mathbf{S}_{  \ell  }   \text{app}     \circ   \mathbf{S}_{  \ell  }   \langle   \pi_2   ,   \pi_1   \rangle     \circ   t^{\mathbf{S}_{  \ell  } }_{  \llbracket  \ottnt{A}  \rrbracket  ,     \llbracket  \ottnt{B}  \rrbracket  ^{  \llbracket  \ottnt{A}  \rrbracket  }   }    \circ   \langle   \pi_2   ,   \pi_1   \rangle     \Big)  \hspace*{3pt} [\text{By IH}]\\
= &  k^{ \ell }_{    \llbracket  \ottnt{B}  \rrbracket  ^{  \llbracket  \ottnt{A}  \rrbracket  }   } 
\end{align*} 

\item \Rref{Prot-Monad}. Have: $ \ell_{{\mathrm{1}}}  \sqsubseteq   \mathcal{T}_{ \ell_{{\mathrm{2}}} } \:  \ottnt{A}  $ where $ \ell_{{\mathrm{1}}}  \sqsubseteq  \ell_{{\mathrm{2}}} $.\\
Now, $ k^{ \ell_{{\mathrm{1}}} }_{  \mathbf{S}_{  \ell_{{\mathrm{2}}}  }   \llbracket  \ottnt{A}  \rrbracket   }  =  \mu^{  \ell_{{\mathrm{1}}}  ,  \ell_{{\mathrm{2}}}  }_{  \llbracket  \ottnt{A}  \rrbracket  } $ and $ j^{ \ell }_{  \underline{ \ottnt{A} }  }  =  \lambda  \ottmv{x}  .   \mathbf{join}^{  \ell_{{\mathrm{1}}}  ,  \ell_{{\mathrm{2}}}  }  \ottmv{x}  $.\\
Then, $  \Lambda^{-1}   \llbracket    \lambda  \ottmv{x}  .   \mathbf{join}^{  \ell_{{\mathrm{1}}}  ,  \ell_{{\mathrm{2}}}  }  \ottmv{x}     \rrbracket    \circ   \langle   \langle \rangle   ,   \text{id}_{  \mathbf{S}_{  \ell_{{\mathrm{1}}}  }   \mathbf{S}_{  \ell_{{\mathrm{2}}}  }   \llbracket  \ottnt{A}  \rrbracket    }   \rangle   =    \mu^{  \ell_{{\mathrm{1}}}  ,  \ell_{{\mathrm{2}}}  }_{  \llbracket  \ottnt{A}  \rrbracket  }   \circ   \pi_2    \circ   \langle   \langle \rangle   ,   \text{id}_{  \mathbf{S}_{  \ell_{{\mathrm{1}}}  }   \mathbf{S}_{  \ell_{{\mathrm{2}}}  }   \llbracket  \ottnt{A}  \rrbracket    }   \rangle   =  \mu^{  \ell_{{\mathrm{1}}}  ,  \ell_{{\mathrm{2}}}  }_{  \llbracket  \ottnt{A}  \rrbracket  } $.

\item \Rref{Prot-Already}. Have: $ \ell  \sqsubseteq   \mathcal{T}_{ \ell' } \:  \ottnt{A}  $ where $ \ell  \sqsubseteq  \ottnt{A} $.\\
By IH, $ k^{ \ell }_{  \llbracket  \ottnt{A}  \rrbracket  }  =   \Lambda^{-1}   \llbracket    j^{ \ell }_{  \underline{ \ottnt{A} }  }    \rrbracket    \circ   \langle   \langle \rangle   ,   \text{id}_{  \mathbf{S}_{  \ell  }   \llbracket  \ottnt{A}  \rrbracket   }   \rangle  $.\\ 
Now, $ k^{ \ell }_{  \mathbf{S}_{  \ell'  }   \llbracket  \ottnt{A}  \rrbracket   }  =    \mathbf{S}_{  \ell'  }   k^{ \ell }_{  \llbracket  \ottnt{A}  \rrbracket  }    \circ   \delta^{  \ell'  ,  \ell  }_{  \llbracket  \ottnt{A}  \rrbracket  }    \circ   \mu^{  \ell  ,  \ell'  }_{  \llbracket  \ottnt{A}  \rrbracket  }  $ and $ j^{ \ell }_{  S_{  \ell'  } \:   \underline{ \ottnt{A} }   }  =   \lambda  \ottmv{x}  .   (   \mathbf{lift}^{  \ell'  }   j^{ \ell }_{  \underline{ \ottnt{A} }  }    )    \:   (   \mathbf{fork}^{  \ell'  ,  \ell  }   (   \mathbf{join}^{  \ell  ,  \ell'  }  \ottmv{x}   )    )  $.

Then,
\begin{align*}
&   \Lambda^{-1}   \llbracket     \lambda  \ottmv{x}  .   (   \mathbf{lift}^{  \ell'  }   j^{ \ell }_{  \underline{ \ottnt{A} }  }    )    \:   (   \mathbf{fork}^{  \ell'  ,  \ell  }   (   \mathbf{join}^{  \ell  ,  \ell'  }  \ottmv{x}   )    )     \rrbracket    \circ   \langle   \langle \rangle   ,   \text{id}_{  \mathbf{S}_{  \ell  }   \mathbf{S}_{  \ell'  }   \llbracket  \ottnt{A}  \rrbracket    }   \rangle   \\
= &    \text{app}   \circ   \langle   \llbracket    \mathbf{lift}^{  \ell'  }   j^{ \ell }_{  \underline{ \ottnt{A} }  }     \rrbracket   ,   \llbracket    \mathbf{fork}^{  \ell'  ,  \ell  }   (   \mathbf{join}^{  \ell  ,  \ell'  }  \ottmv{x}   )     \rrbracket   \rangle    \circ   \langle   \langle \rangle   ,   \text{id}_{  \mathbf{S}_{  \ell  }   \mathbf{S}_{  \ell'  }   \llbracket  \ottnt{A}  \rrbracket    }   \rangle   \\
= &    \text{app}   \circ   \langle    \Lambda \Big(     \mathbf{S}_{  \ell'  }   (   \Lambda^{-1}   \llbracket    j^{ \ell }_{  \underline{ \ottnt{A} }  }    \rrbracket    )    \circ   t^{\mathbf{S}_{  \ell'  } }_{  \top  ,   \mathbf{S}_{  \ell  }   \llbracket  \ottnt{A}  \rrbracket   }     \Big)   \circ   \langle \rangle    ,     \delta^{  \ell'  ,  \ell  }_{  \llbracket  \ottnt{A}  \rrbracket  }   \circ   \mu^{  \ell  ,  \ell'  }_{  \llbracket  \ottnt{A}  \rrbracket  }    \circ   \pi_2    \rangle    \circ   \langle   \langle \rangle   ,   \text{id}_{  \mathbf{S}_{  \ell  }   \mathbf{S}_{  \ell'  }   \llbracket  \ottnt{A}  \rrbracket    }   \rangle   \\
= &     \mathbf{S}_{  \ell'  }   (     \Lambda^{-1}   \llbracket    j^{ \ell }_{  \underline{ \ottnt{A} }  }    \rrbracket    \circ   \langle   \langle \rangle   ,   \text{id}_{  \mathbf{S}_{  \ell  }   \llbracket  \ottnt{A}  \rrbracket   }   \rangle    \circ   \pi_2    )    \circ   t^{\mathbf{S}_{  \ell'  } }_{  \top  ,   \mathbf{S}_{  \ell  }   \llbracket  \ottnt{A}  \rrbracket   }    \circ   \langle   \langle \rangle   ,     \delta^{  \ell'  ,  \ell  }_{  \llbracket  \ottnt{A}  \rrbracket  }   \circ   \mu^{  \ell  ,  \ell'  }_{  \llbracket  \ottnt{A}  \rrbracket  }    \circ   \pi_2    \rangle    \circ   \langle   \langle \rangle   ,   \text{id}_{  \mathbf{S}_{  \ell  }   \mathbf{S}_{  \ell'  }   \llbracket  \ottnt{A}  \rrbracket    }   \rangle  \\ & \hspace{30pt} [\because   \langle   \langle \rangle   ,   \text{id}   \rangle   \circ   \pi_2   =  \text{id}_{    \top   \times   \mathbf{S}_{  \ell  }   \llbracket  \ottnt{A}  \rrbracket     }  ] \\ 
= &      \mathbf{S}_{  \ell'  }   (    \Lambda^{-1}   \llbracket    j^{ \ell }_{  \underline{ \ottnt{A} }  }    \rrbracket    \circ   \langle   \langle \rangle   ,   \text{id}_{  \mathbf{S}_{  \ell  }   \llbracket  \ottnt{A}  \rrbracket   }   \rangle    )    \circ   \mathbf{S}_{  \ell'  }   \pi_2     \circ   t^{\mathbf{S}_{  \ell'  } }_{  \top  ,   \mathbf{S}_{  \ell  }   \llbracket  \ottnt{A}  \rrbracket   }    \circ   \langle   \langle \rangle   ,     \delta^{  \ell'  ,  \ell  }_{  \llbracket  \ottnt{A}  \rrbracket  }   \circ   \mu^{  \ell  ,  \ell'  }_{  \llbracket  \ottnt{A}  \rrbracket  }    \circ   \pi_2    \rangle    \circ   \langle   \langle \rangle   ,   \text{id}_{  \mathbf{S}_{  \ell  }   \mathbf{S}_{  \ell'  }   \llbracket  \ottnt{A}  \rrbracket    }   \rangle   \\
= &      \mathbf{S}_{  \ell'  }   k^{ \ell }_{  \llbracket  \ottnt{A}  \rrbracket  }    \circ   \mathbf{S}_{  \ell'  }   \pi_2     \circ   t^{\mathbf{S}_{  \ell'  } }_{  \top  ,   \mathbf{S}_{  \ell  }   \llbracket  \ottnt{A}  \rrbracket   }    \circ   \langle   \langle \rangle   ,     \delta^{  \ell'  ,  \ell  }_{  \llbracket  \ottnt{A}  \rrbracket  }   \circ   \mu^{  \ell  ,  \ell'  }_{  \llbracket  \ottnt{A}  \rrbracket  }    \circ   \pi_2    \rangle    \circ   \langle   \langle \rangle   ,   \text{id}_{  \mathbf{S}_{  \ell  }   \mathbf{S}_{  \ell'  }   \llbracket  \ottnt{A}  \rrbracket    }   \rangle   \hspace*{3pt} [\text{By IH}]\\
= &     \mathbf{S}_{  \ell'  }   k^{ \ell }_{  \llbracket  \ottnt{A}  \rrbracket  }    \circ   \pi_2    \circ   \langle   \langle \rangle   ,     \delta^{  \ell'  ,  \ell  }_{  \llbracket  \ottnt{A}  \rrbracket  }   \circ   \mu^{  \ell  ,  \ell'  }_{  \llbracket  \ottnt{A}  \rrbracket  }    \circ   \pi_2    \rangle    \circ   \langle   \langle \rangle   ,   \text{id}_{  \mathbf{S}_{  \ell  }   \mathbf{S}_{  \ell'  }   \llbracket  \ottnt{A}  \rrbracket    }   \rangle   \hspace*{3pt} [\text{By strength}]\\
= &    \mathbf{S}_{  \ell'  }   k^{ \ell }_{  \llbracket  \ottnt{A}  \rrbracket  }    \circ   \delta^{  \ell'  ,  \ell  }_{  \llbracket  \ottnt{A}  \rrbracket  }    \circ   \mu^{  \ell  ,  \ell'  }_{  \llbracket  \ottnt{A}  \rrbracket  }   \\
= &  k^{ \ell }_{  \mathbf{S}_{  \ell'  }   \llbracket  \ottnt{A}  \rrbracket   } 
\end{align*}

\item \Rref{Prot-Minimum}. Have: $  \bot   \sqsubseteq  \ottnt{A} $.\\
Now, $ k^{  \bot  }_{  \llbracket  \ottnt{A}  \rrbracket  }  =  \epsilon_{  \llbracket  \ottnt{A}  \rrbracket  } $ and $ j^{  \bot  }_{  \underline{ \ottnt{A} }  }  =  \lambda  \ottmv{x}  .   \mathbf{extr} \:  \ottmv{x}  $.\\
Then, $  \Lambda^{-1}   \llbracket    \lambda  \ottmv{x}  .   \mathbf{extr} \:  \ottmv{x}     \rrbracket    \circ   \langle   \langle \rangle   ,   \text{id}_{  \mathbf{S}_{   \bot   }   \llbracket  \ottnt{A}  \rrbracket   }   \rangle   =    \epsilon_{  \llbracket  \ottnt{A}  \rrbracket  }   \circ   \pi_2    \circ   \langle   \langle \rangle   ,   \text{id}_{  \mathbf{S}_{   \bot   }   \llbracket  \ottnt{A}  \rrbracket   }   \rangle   =  \epsilon_{  \llbracket  \ottnt{A}  \rrbracket  }  =  k^{  \bot  }_{  \llbracket  \ottnt{A}  \rrbracket  } $.

\item \Rref{Prot-Combine}. Have: $ \ell  \sqsubseteq  \ottnt{A} $ where $ \ell_{{\mathrm{1}}}  \sqsubseteq  \ottnt{A} $ and $ \ell_{{\mathrm{2}}}  \sqsubseteq  \ottnt{A} $ and $ \ell  \sqsubseteq   \ell_{{\mathrm{1}}}  \vee  \ell_{{\mathrm{2}}}  $.\\
By IH, $ k^{ \ell_{{\mathrm{1}}} }_{  \llbracket  \ottnt{A}  \rrbracket  }  =   \Lambda^{-1}   \llbracket    j^{ \ell_{{\mathrm{1}}} }_{  \underline{ \ottnt{A} }  }    \rrbracket    \circ   \langle   \langle \rangle   ,   \text{id}_{  \mathbf{S}_{  \ell_{{\mathrm{1}}}  }   \llbracket  \ottnt{A}  \rrbracket   }   \rangle  $ and $ k^{ \ell_{{\mathrm{2}}} }_{  \llbracket  \ottnt{A}  \rrbracket  }  =   \Lambda^{-1}   \llbracket    j^{ \ell_{{\mathrm{2}}} }_{  \underline{ \ottnt{A} }  }    \rrbracket    \circ   \langle   \langle \rangle   ,   \text{id}_{  \mathbf{S}_{  \ell_{{\mathrm{2}}}  }   \llbracket  \ottnt{A}  \rrbracket   }   \rangle  $.\\
Now, $ k^{ \ell }_{  \llbracket  \ottnt{A}  \rrbracket  }  =     k^{ \ell_{{\mathrm{1}}} }_{  \llbracket  \ottnt{A}  \rrbracket  }   \circ   \mathbf{S}_{  \ell_{{\mathrm{1}}}  }   k^{ \ell_{{\mathrm{2}}} }_{  \llbracket  \ottnt{A}  \rrbracket  }     \circ   \delta^{  \ell_{{\mathrm{1}}}  ,  \ell_{{\mathrm{2}}}  }_{  \llbracket  \ottnt{A}  \rrbracket  }    \circ   \mathbf{S}^{  \ell   \sqsubseteq    \ell_{{\mathrm{1}}}  \vee  \ell_{{\mathrm{2}}}   }_{  \llbracket  \ottnt{A}  \rrbracket  }  $ \\
and $ j^{ \ell }_{  \underline{ \ottnt{A} }  }  =   \lambda  \ottmv{x}  .   j^{ \ell_{{\mathrm{1}}} }_{  \underline{ \ottnt{A} }  }    \:   (    (   \mathbf{lift}^{  \ell_{{\mathrm{1}}}  }   j^{ \ell_{{\mathrm{2}}} }_{  \underline{ \ottnt{A} }  }    )   \:   (   \mathbf{fork}^{  \ell_{{\mathrm{1}}}  ,  \ell_{{\mathrm{2}}}  }   (   \mathbf{up}^{  \ell  ,   \ell_{{\mathrm{1}}}  \vee  \ell_{{\mathrm{2}}}   }  \ottmv{x}   )    )    )  $.\\
Then,
\begin{align*}
&   \Lambda^{-1}   \llbracket     \lambda  \ottmv{x}  .   j^{ \ell_{{\mathrm{1}}} }_{  \underline{ \ottnt{A} }  }    \:   (    (   \mathbf{lift}^{  \ell_{{\mathrm{1}}}  }   j^{ \ell_{{\mathrm{2}}} }_{  \underline{ \ottnt{A} }  }    )   \:   (   \mathbf{fork}^{  \ell_{{\mathrm{1}}}  ,  \ell_{{\mathrm{2}}}  }   (   \mathbf{up}^{  \ell  ,   \ell_{{\mathrm{1}}}  \vee  \ell_{{\mathrm{2}}}   }  \ottmv{x}   )    )    )     \rrbracket    \circ   \langle   \langle \rangle   ,   \text{id}_{  \mathbf{S}_{  \ell  }   \llbracket  \ottnt{A}  \rrbracket   }   \rangle   \\
= &   \llbracket     j^{ \ell_{{\mathrm{1}}} }_{  \underline{ \ottnt{A} }  }   \:   (    (   \mathbf{lift}^{  \ell_{{\mathrm{1}}}  }   j^{ \ell_{{\mathrm{2}}} }_{  \underline{ \ottnt{A} }  }    )   \:   (   \mathbf{fork}^{  \ell_{{\mathrm{1}}}  ,  \ell_{{\mathrm{2}}}  }   (   \mathbf{up}^{  \ell  ,   \ell_{{\mathrm{1}}}  \vee  \ell_{{\mathrm{2}}}   }  \ottmv{x}   )    )    )     \rrbracket   \circ   \langle   \langle \rangle   ,   \text{id}_{  \mathbf{S}_{  \ell  }   \llbracket  \ottnt{A}  \rrbracket   }   \rangle   \\
= &    \text{app}   \circ   \langle    \llbracket    j^{ \ell_{{\mathrm{1}}} }_{  \underline{ \ottnt{A} }  }    \rrbracket   \circ   \langle \rangle    ,    \text{app}   \circ   \langle    \llbracket    \mathbf{lift}^{  \ell_{{\mathrm{1}}}  }   j^{ \ell_{{\mathrm{2}}} }_{  \underline{ \ottnt{A} }  }     \rrbracket   \circ   \langle \rangle    ,   \llbracket    \mathbf{fork}^{  \ell_{{\mathrm{1}}}  ,  \ell_{{\mathrm{2}}}  }   (   \mathbf{up}^{  \ell  ,   \ell_{{\mathrm{1}}}  \vee  \ell_{{\mathrm{2}}}   }  \ottmv{x}   )     \rrbracket   \rangle    \rangle    \circ   \langle   \langle \rangle   ,   \text{id}_{  \mathbf{S}_{  \ell  }   \llbracket  \ottnt{A}  \rrbracket   }   \rangle   \\
= &    \text{app}   \circ   \langle    \llbracket    j^{ \ell_{{\mathrm{1}}} }_{  \underline{ \ottnt{A} }  }    \rrbracket   \circ   \langle \rangle    ,    \text{app}   \circ   \langle    \Lambda \Big(     \mathbf{S}_{  \ell_{{\mathrm{1}}}  }   (   \Lambda^{-1}   \llbracket    j^{ \ell_{{\mathrm{2}}} }_{  \underline{ \ottnt{A} }  }    \rrbracket    )    \circ   t^{\mathbf{S}_{  \ell_{{\mathrm{1}}}  } }_{  \top  ,   \mathbf{S}_{  \ell_{{\mathrm{2}}}  }   \llbracket  \ottnt{A}  \rrbracket   }     \Big)   \circ   \langle \rangle    ,     \delta^{  \ell_{{\mathrm{1}}}  ,  \ell_{{\mathrm{2}}}  }_{  \llbracket  \ottnt{A}  \rrbracket  }   \circ   \mathbf{S}^{  \ell   \sqsubseteq    \ell_{{\mathrm{1}}}  \vee  \ell_{{\mathrm{2}}}   }_{  \llbracket  \ottnt{A}  \rrbracket  }    \circ   \pi_2    \rangle    \rangle    \circ   \langle   \langle \rangle   ,   \text{id}_{  \mathbf{S}_{  \ell  }   \llbracket  \ottnt{A}  \rrbracket   }   \rangle   \\
= &   \text{app}   \circ   \langle    \llbracket    j^{ \ell_{{\mathrm{1}}} }_{  \underline{ \ottnt{A} }  }    \rrbracket   \circ   \langle \rangle    ,    \text{app}   \circ   \langle    \Lambda \Big(     \mathbf{S}_{  \ell_{{\mathrm{1}}}  }   (     \Lambda^{-1}   \llbracket    j^{ \ell_{{\mathrm{2}}} }_{  \underline{ \ottnt{A} }  }    \rrbracket    \circ   \langle   \langle \rangle   ,   \text{id}   \rangle    \circ   \pi_2    )    \circ   t^{\mathbf{S}_{  \ell_{{\mathrm{1}}}  } }_{  \top  ,   \mathbf{S}_{  \ell_{{\mathrm{2}}}  }   \llbracket  \ottnt{A}  \rrbracket   }     \Big)   \circ   \langle \rangle    ,     \delta^{  \ell_{{\mathrm{1}}}  ,  \ell_{{\mathrm{2}}}  }_{  \llbracket  \ottnt{A}  \rrbracket  }   \circ   \mathbf{S}^{  \ell   \sqsubseteq    \ell_{{\mathrm{1}}}  \vee  \ell_{{\mathrm{2}}}   }_{  \llbracket  \ottnt{A}  \rrbracket  }    \circ   \pi_2    \rangle    \rangle   \\ & \hspace*{30pt} \circ  \langle   \langle \rangle   ,   \text{id}_{  \mathbf{S}_{  \ell  }   \llbracket  \ottnt{A}  \rrbracket   }   \rangle  \\
= &    \text{app}   \circ   \langle    \llbracket    j^{ \ell_{{\mathrm{1}}} }_{  \underline{ \ottnt{A} }  }    \rrbracket   \circ   \langle \rangle    ,    \text{app}   \circ   \langle    \Lambda \Big(     \mathbf{S}_{  \ell_{{\mathrm{1}}}  }   k^{ \ell_{{\mathrm{2}}} }_{  \llbracket  \ottnt{A}  \rrbracket  }    \circ   \pi_2     \Big)   \circ   \langle \rangle    ,     \delta^{  \ell_{{\mathrm{1}}}  ,  \ell_{{\mathrm{2}}}  }_{  \llbracket  \ottnt{A}  \rrbracket  }   \circ   \mathbf{S}^{  \ell   \sqsubseteq    \ell_{{\mathrm{1}}}  \vee  \ell_{{\mathrm{2}}}   }_{  \llbracket  \ottnt{A}  \rrbracket  }    \circ   \pi_2    \rangle    \rangle    \circ   \langle   \langle \rangle   ,   \text{id}_{  \mathbf{S}_{  \ell  }   \llbracket  \ottnt{A}  \rrbracket   }   \rangle   \\ & \hspace*{30pt} [\text{By IH and strength}]\\
= &    \text{app}   \circ   \langle    \llbracket    j^{ \ell_{{\mathrm{1}}} }_{  \underline{ \ottnt{A} }  }    \rrbracket   \circ   \langle \rangle    ,      \mathbf{S}_{  \ell_{{\mathrm{1}}}  }   k^{ \ell_{{\mathrm{2}}} }_{  \llbracket  \ottnt{A}  \rrbracket  }    \circ   \pi_2    \circ   (    \text{id}   \times   (     \delta^{  \ell_{{\mathrm{1}}}  ,  \ell_{{\mathrm{2}}}  }_{  \llbracket  \ottnt{A}  \rrbracket  }   \circ   \mathbf{S}^{  \ell   \sqsubseteq    \ell_{{\mathrm{1}}}  \vee  \ell_{{\mathrm{2}}}   }_{  \llbracket  \ottnt{A}  \rrbracket  }    \circ   \pi_2    )    )    \circ   \langle   \langle \rangle   ,   \text{id}   \rangle    \rangle    \circ   \langle   \langle \rangle   ,   \text{id}   \rangle   \\
= &    \text{app}   \circ   \langle    \llbracket    j^{ \ell_{{\mathrm{1}}} }_{  \underline{ \ottnt{A} }  }    \rrbracket   \circ   \langle \rangle    ,        \mathbf{S}_{  \ell_{{\mathrm{1}}}  }   k^{ \ell_{{\mathrm{2}}} }_{  \llbracket  \ottnt{A}  \rrbracket  }    \circ   \delta^{  \ell_{{\mathrm{1}}}  ,  \ell_{{\mathrm{2}}}  }_{  \llbracket  \ottnt{A}  \rrbracket  }    \circ   \mathbf{S}^{  \ell   \sqsubseteq    \ell_{{\mathrm{1}}}  \vee  \ell_{{\mathrm{2}}}   }_{  \llbracket  \ottnt{A}  \rrbracket  }    \circ   \pi_2    \circ   \pi_2    \circ   \langle   \langle \rangle   ,   \text{id}   \rangle    \rangle    \circ   \langle   \langle \rangle   ,   \text{id}   \rangle   \\
= &    \text{app}   \circ   \langle    \llbracket    j^{ \ell_{{\mathrm{1}}} }_{  \underline{ \ottnt{A} }  }    \rrbracket   \circ   \langle \rangle    ,      \mathbf{S}_{  \ell_{{\mathrm{1}}}  }   k^{ \ell_{{\mathrm{2}}} }_{  \llbracket  \ottnt{A}  \rrbracket  }    \circ   \delta^{  \ell_{{\mathrm{1}}}  ,  \ell_{{\mathrm{2}}}  }_{  \llbracket  \ottnt{A}  \rrbracket  }    \circ   \mathbf{S}^{  \ell   \sqsubseteq    \ell_{{\mathrm{1}}}  \vee  \ell_{{\mathrm{2}}}   }_{  \llbracket  \ottnt{A}  \rrbracket  }    \circ   \pi_2    \rangle    \circ   \langle   \langle \rangle   ,   \text{id}   \rangle   \\
= &     \text{app}   \circ   (    \llbracket    j^{ \ell_{{\mathrm{1}}} }_{  \underline{ \ottnt{A} }  }    \rrbracket   \times   \text{id}_{   \mathbf{S}_{  \ell_{{\mathrm{1}}}  }   \llbracket  \ottnt{A}  \rrbracket    }    )    \circ   \langle   \langle \rangle   ,      \mathbf{S}_{  \ell_{{\mathrm{1}}}  }   k^{ \ell_{{\mathrm{2}}} }_{  \llbracket  \ottnt{A}  \rrbracket  }    \circ   \delta^{  \ell_{{\mathrm{1}}}  ,  \ell_{{\mathrm{2}}}  }_{  \llbracket  \ottnt{A}  \rrbracket  }    \circ   \mathbf{S}^{  \ell   \sqsubseteq    \ell_{{\mathrm{1}}}  \vee  \ell_{{\mathrm{2}}}   }_{  \llbracket  \ottnt{A}  \rrbracket  }    \circ   \pi_2    \rangle    \circ   \langle   \langle \rangle   ,   \text{id}   \rangle   \\
= &      \Lambda^{-1}   \llbracket    j^{ \ell_{{\mathrm{1}}} }_{  \underline{ \ottnt{A} }  }    \rrbracket    \circ   \langle   \langle \rangle   ,   \text{id}_{   \mathbf{S}_{  \ell_{{\mathrm{1}}}  }   \llbracket  \ottnt{A}  \rrbracket    }   \rangle    \circ   \pi_2    \circ   \langle   \langle \rangle   ,      \mathbf{S}_{  \ell_{{\mathrm{1}}}  }   k^{ \ell_{{\mathrm{2}}} }_{  \llbracket  \ottnt{A}  \rrbracket  }    \circ   \delta^{  \ell_{{\mathrm{1}}}  ,  \ell_{{\mathrm{2}}}  }_{  \llbracket  \ottnt{A}  \rrbracket  }    \circ   \mathbf{S}^{  \ell   \sqsubseteq    \ell_{{\mathrm{1}}}  \vee  \ell_{{\mathrm{2}}}   }_{  \llbracket  \ottnt{A}  \rrbracket  }    \circ   \pi_2    \rangle    \circ   \langle   \langle \rangle   ,   \text{id}   \rangle   \\
= &     k^{ \ell_{{\mathrm{1}}} }_{  \llbracket  \ottnt{A}  \rrbracket  }   \circ   \mathbf{S}_{  \ell_{{\mathrm{1}}}  }   k^{ \ell_{{\mathrm{2}}} }_{  \llbracket  \ottnt{A}  \rrbracket  }     \circ   \delta^{  \ell_{{\mathrm{1}}}  ,  \ell_{{\mathrm{2}}}  }_{  \llbracket  \ottnt{A}  \rrbracket  }    \circ   \mathbf{S}^{  \ell   \sqsubseteq    \ell_{{\mathrm{1}}}  \vee  \ell_{{\mathrm{2}}}   }_{  \llbracket  \ottnt{A}  \rrbracket  }   \hspace*{3pt} [\text{By IH}]\\
= &  k^{ \ell }_{  \llbracket  \ottnt{A}  \rrbracket  } 
\end{align*} 

\end{itemize}
\end{proof}


\begin{theorem}[Theorem \ref{gmccdcce}]
Let $ \mathcal{L} $ be a bounded join-semilattice.
\begin{itemize}
\item If $ \Gamma  \vdash  \ottnt{a}  :  \ottnt{A} $ in GMCC($ \mathcal{L} $), then $ \llbracket   \overline{ \ottnt{a} }   \rrbracket  =  \llbracket  \ottnt{a}  \rrbracket $.
\item If $ \Gamma  \vdash  \ottnt{a}  :  \ottnt{A} $ in \ED{}($ \mathcal{L} $), then $ \llbracket   \underline{ \ottnt{a} }   \rrbracket  =  \llbracket  \ottnt{a}  \rrbracket $.
\end{itemize}
\end{theorem}

\begin{proof}
Follows by theorems \ref{equiv1} and \ref{equiv2}.
\end{proof}


\begin{theorem}[Theorem \ref{dcceEM}]
If $ \ell  \sqsubseteq  \ottnt{A} $ in \ED{}, then $( \llbracket  \ottnt{A}  \rrbracket ,  k^{ \ell }_{  \llbracket  \ottnt{A}  \rrbracket  } )$ is an $\mathbf{S}_{\ell}$-algebra.\\
Further, if $ \ell  \sqsubseteq  \ottnt{A} $ and $ \ell  \sqsubseteq  \ottnt{B} $, then for any $f \in \text{Hom}_{\Ct} ( \llbracket  \ottnt{A}  \rrbracket  ,  \llbracket  \ottnt{B}  \rrbracket )$, $f$ is an $\mathbf{S}_{\ell}$-algebra morphism.\\
Hence, the full subcategory of $\Ct{}$ with $\text{Obj} := \{  \llbracket  \ottnt{A}  \rrbracket  \: | \:  \ell  \sqsubseteq  \ottnt{A}  \}$ is also a full subcategory of the Eilenberg-Moore category, $\Ct{}^{\mathbf{S}_{\ell}}$.
\end{theorem}

\begin{proof}
Let $ \ell  \sqsubseteq  \ottnt{A} $. We show that $( \llbracket  \ottnt{A}  \rrbracket  ,  k^{ \ell }_{  \llbracket  \ottnt{A}  \rrbracket  } )$ is an $(\mathbf{S}_{\ell} , \mathbf{S}^{  \bot   \sqsubseteq  \ell } \circ \eta , \mu^{\ell,\ell})$-algebra. 

We use the following shorthand: $ \overline{\eta}_{ \ottnt{X} }  :=   \mathbf{S}^{   \bot    \sqsubseteq   \ell  }_{ \ottnt{X} }   \circ   \eta_{ \ottnt{X} }  $ and $\overline{\mu}_{X} :=  \mu^{  \ell  ,  \ell  }_{ \ottnt{X} } $, where $X \in \text{Obj}(\Ct)$.\\

We need to show: $  k^{ \ell }_{  \llbracket  \ottnt{A}  \rrbracket  }   \circ   \overline{\eta}_{  \llbracket  \ottnt{A}  \rrbracket  }   =  \text{id}_{  \llbracket  \ottnt{A}  \rrbracket  } $. This follows by lemma \ref{Ap2}. 

Next, we need to show: $  k^{ \ell }_{  \llbracket  \ottnt{A}  \rrbracket  }   \circ   \mathbf{S}_{  \ell  }   k^{ \ell }_{  \llbracket  \ottnt{A}  \rrbracket  }    =   k^{ \ell }_{  \llbracket  \ottnt{A}  \rrbracket  }   \circ   \overline{\mu}_{  \llbracket  \ottnt{A}  \rrbracket  }  $.

Now,
\begin{align*}
  k^{ \ell }_{  \llbracket  \ottnt{A}  \rrbracket  }   \circ   \overline{\eta}_{  \llbracket  \ottnt{A}  \rrbracket  }   & =  \text{id}_{  \llbracket  \ottnt{A}  \rrbracket  }  \\
\text{or, }   \mathbf{S}_{  \ell  }   k^{ \ell }_{  \llbracket  \ottnt{A}  \rrbracket  }    \circ   \mathbf{S}_{  \ell  }   \overline{\eta}_{  \llbracket  \ottnt{A}  \rrbracket  }    & =  \text{id}_{  \mathbf{S}_{  \ell  }   \llbracket  \ottnt{A}  \rrbracket   }  \\
\text{or, }    k^{ \ell }_{  \llbracket  \ottnt{A}  \rrbracket  }   \circ   \mathbf{S}_{  \ell  }   k^{ \ell }_{  \llbracket  \ottnt{A}  \rrbracket  }     \circ   \mathbf{S}_{  \ell  }   \overline{\eta}_{  \llbracket  \ottnt{A}  \rrbracket  }    & =  k^{ \ell }_{  \llbracket  \ottnt{A}  \rrbracket  }  \\
\text{or, }     k^{ \ell }_{  \llbracket  \ottnt{A}  \rrbracket  }   \circ   \mathbf{S}_{  \ell  }   k^{ \ell }_{  \llbracket  \ottnt{A}  \rrbracket  }     \circ   \mathbf{S}_{  \ell  }   \overline{\eta}_{  \llbracket  \ottnt{A}  \rrbracket  }     \circ   \overline{\mu}_{  \llbracket  \ottnt{A}  \rrbracket  }   & =   k^{ \ell }_{  \llbracket  \ottnt{A}  \rrbracket  }   \circ   \overline{\mu}_{  \llbracket  \ottnt{A}  \rrbracket  }   \\
\text{or, }       k^{ \ell }_{  \llbracket  \ottnt{A}  \rrbracket  }   \circ   \mathbf{S}_{  \ell  }   k^{ \ell }_{  \llbracket  \ottnt{A}  \rrbracket  }     \circ   \delta^{  \ell  ,  \ell  }_{  \llbracket  \ottnt{A}  \rrbracket  }    \circ   \mu^{  \ell  ,  \ell  }_{  \llbracket  \ottnt{A}  \rrbracket  }    \circ   \mathbf{S}_{  \ell  }   \overline{\eta}_{  \llbracket  \ottnt{A}  \rrbracket  }     \circ   \overline{\mu}_{  \llbracket  \ottnt{A}  \rrbracket  }   & =   k^{ \ell }_{  \llbracket  \ottnt{A}  \rrbracket  }   \circ   \overline{\mu}_{  \llbracket  \ottnt{A}  \rrbracket  }   \\
\text{or, }       k^{ \ell }_{  \llbracket  \ottnt{A}  \rrbracket  }   \circ   \mathbf{S}_{  \ell  }   k^{ \ell }_{  \llbracket  \ottnt{A}  \rrbracket  }     \circ   \delta^{  \ell  ,  \ell  }_{  \llbracket  \ottnt{A}  \rrbracket  }    \circ   \overline{\mu}_{  \llbracket  \ottnt{A}  \rrbracket  }    \circ   \mathbf{S}_{  \ell  }   \overline{\eta}_{  \llbracket  \ottnt{A}  \rrbracket  }     \circ   \overline{\mu}_{  \llbracket  \ottnt{A}  \rrbracket  }   & =   k^{ \ell }_{  \llbracket  \ottnt{A}  \rrbracket  }   \circ   \overline{\mu}_{  \llbracket  \ottnt{A}  \rrbracket  }   \\
\text{or, }     k^{ \ell }_{  \llbracket  \ottnt{A}  \rrbracket  }   \circ   \mathbf{S}_{  \ell  }   k^{ \ell }_{  \llbracket  \ottnt{A}  \rrbracket  }     \circ   \delta^{  \ell  ,  \ell  }_{  \llbracket  \ottnt{A}  \rrbracket  }    \circ   \overline{\mu}_{  \llbracket  \ottnt{A}  \rrbracket  }   & =   k^{ \ell }_{  \llbracket  \ottnt{A}  \rrbracket  }   \circ   \overline{\mu}_{  \llbracket  \ottnt{A}  \rrbracket  }   \hspace*{3pt} [\because \: (\mathbf{S}_{\ell} , \overline{\eta}, \overline{\mu}) \text{ is a monad}]\\
\text{or, }     k^{ \ell }_{  \llbracket  \ottnt{A}  \rrbracket  }   \circ   \mathbf{S}_{  \ell  }   k^{ \ell }_{  \llbracket  \ottnt{A}  \rrbracket  }     \circ   \delta^{  \ell  ,  \ell  }_{  \llbracket  \ottnt{A}  \rrbracket  }    \circ   \mu^{  \ell  ,  \ell  }_{  \llbracket  \ottnt{A}  \rrbracket  }   & =   k^{ \ell }_{  \llbracket  \ottnt{A}  \rrbracket  }   \circ   \overline{\mu}_{  \llbracket  \ottnt{A}  \rrbracket  }  \\
\text{or, }   k^{ \ell }_{  \llbracket  \ottnt{A}  \rrbracket  }   \circ   \mathbf{S}_{  \ell  }   k^{ \ell }_{  \llbracket  \ottnt{A}  \rrbracket  }    & =   k^{ \ell }_{  \llbracket  \ottnt{A}  \rrbracket  }   \circ   \overline{\mu}_{  \llbracket  \ottnt{A}  \rrbracket  }  
\end{align*}

Hence, $( \llbracket  \ottnt{A}  \rrbracket  ,  k^{ \ell }_{  \llbracket  \ottnt{A}  \rrbracket  } )$ is an $\mathbf{S}_{\ell}$-algebra.\\

Next, let $ \ell  \sqsubseteq  \ottnt{A} $ and $ \ell  \sqsubseteq  \ottnt{B} $ such that $f \in \text{Hom}_{\Ct} ( \llbracket  \ottnt{A}  \rrbracket  ,  \llbracket  \ottnt{B}  \rrbracket )$.

We show that $f$ is an $\mathbf{S}_{\ell}$-algebra morphism.

Need to show: $ \ottnt{f}  \circ   k^{ \ell }_{  \llbracket  \ottnt{A}  \rrbracket  }   =   k^{ \ell }_{  \llbracket  \ottnt{B}  \rrbracket  }   \circ   \mathbf{S}_{  \ell  }  \ottnt{f}  $.

Since $\overline{\eta}$ is a natural transformation, we have:
\begin{align*}
  \overline{\eta}_{  \llbracket  \ottnt{B}  \rrbracket  }   \circ  \ottnt{f}  & =   \mathbf{S}_{  \ell  }  \ottnt{f}   \circ   \overline{\eta}_{  \llbracket  \ottnt{A}  \rrbracket  }   \\
\text{or, }     k^{ \ell }_{  \llbracket  \ottnt{B}  \rrbracket  }   \circ   \overline{\eta}_{  \llbracket  \ottnt{B}  \rrbracket  }    \circ  \ottnt{f}   \circ   k^{ \ell }_{  \llbracket  \ottnt{A}  \rrbracket  }   & =     k^{ \ell }_{  \llbracket  \ottnt{B}  \rrbracket  }   \circ   \mathbf{S}_{  \ell  }  \ottnt{f}    \circ   \overline{\eta}_{  \llbracket  \ottnt{A}  \rrbracket  }    \circ   k^{ \ell }_{  \llbracket  \ottnt{A}  \rrbracket  }   \\
\text{or, }    \text{id}_{  \llbracket  \ottnt{B}  \rrbracket  }   \circ  \ottnt{f}   \circ   k^{ \ell }_{  \llbracket  \ottnt{A}  \rrbracket  }   & =     k^{ \ell }_{  \llbracket  \ottnt{B}  \rrbracket  }   \circ   \mathbf{S}_{  \ell  }  \ottnt{f}    \circ   \overline{\eta}_{  \llbracket  \ottnt{A}  \rrbracket  }    \circ   k^{ \ell }_{  \llbracket  \ottnt{A}  \rrbracket  }   \hspace*{3pt} \text{[By Lemma \ref{Ap2}]} \\
\text{or, }  \ottnt{f}  \circ   k^{ \ell }_{  \llbracket  \ottnt{A}  \rrbracket  }   & =    k^{ \ell }_{  \llbracket  \ottnt{B}  \rrbracket  }   \circ   \mathbf{S}_{  \ell  }  \ottnt{f}    \circ   \text{id}_{  \mathbf{S}_{  \ell  }   \llbracket  \ottnt{A}  \rrbracket   }   \hspace*{3pt} \text{[By Lemma  \ref{Ap3}]} \\
\text{or, }  \ottnt{f}  \circ   k^{ \ell }_{  \llbracket  \ottnt{A}  \rrbracket  }   & =   k^{ \ell }_{  \llbracket  \ottnt{B}  \rrbracket  }   \circ   \mathbf{S}_{  \ell  }  \ottnt{f}  
\end{align*}

The final clause of the theorem follows. 
\end{proof}


\begin{theorem}[Theorem \ref{dcceSlA}] \label{dcceSlAprf}
$\llbracket \_ \rrbracket_{(\mathbf{Set},\mathbf{S}^{\ell})}$, for any $\ell \in L$, is a computationally adequate interpretation of \ED{}($ \mathcal{L} $).
\end{theorem}

\begin{proof}
For a bicartesian category $\Ct$, any strong monoidal functor from $\Ca( \mathcal{L} )$ to $\ES{}$ provides a computationally adequate interpretation of \ED{}($ \mathcal{L} $), given the interpretation for ground types is injective. Now, with respect to $\llbracket \_ \rrbracket_{(\mathbf{Set},\mathbf{S}^{\ell})}$, the interpretation for ground types is injective. As such, to prove that $(\mathbf{Set}, \mathbf{S}^{  \ell  } )$ is a computationally adequate interpretation of \ED{}($ \mathcal{L} $), we just need to show that $ \mathbf{S}^{  \ell  } $ is a strong monoidal functor from $\Ca( \mathcal{L} )$ to $\ES{}$.

Recall the definition of $ \mathbf{S}^{  \ell  } $:
\begin{align*}
\mathbf{S}^{\ell}(\ell') = \begin{cases} 
                            \Id , & \text{ if }  \ell'  \sqsubseteq  \ell  \hspace{10pt}\\
                            \ast , & \text{ otherwise } \hspace*{10pt}
                            \end{cases}
\mathbf{S}^{\ell}( \ell_{{\mathrm{1}}}  \sqsubseteq  \ell_{{\mathrm{2}}} ) & = \begin{cases}
                            \I , & \text{ if }  \ell_{{\mathrm{2}}}  \sqsubseteq  \ell  \\
                            \langle \rangle , & \text{ otherwise } 
                            \end{cases} 
\end{align*} 

By this definition, $ \mathbf{S}^{  \ell  }({   \bot   })  = \Id{}$. Further, $\eta = \epsilon = \I{}_{\Id{}}$. 

Now, for any $\ell_{{\mathrm{1}}} , \ell_{{\mathrm{2}}} \in L$, there are two cases to consider:
\begin{itemize}
\item $ \ell_{{\mathrm{1}}}  \sqsubseteq  \ell $ and $ \ell_{{\mathrm{2}}}  \sqsubseteq  \ell $. Then, $  \ell_{{\mathrm{1}}}  \vee  \ell_{{\mathrm{2}}}   \sqsubseteq  \ell $.\\
So, $ \mathbf{S}^{  \ell  }({  \ell_{{\mathrm{1}}}  })  =  \mathbf{S}^{  \ell  }({  \ell_{{\mathrm{2}}}  })  =  \mathbf{S}^{  \ell  }({   \ell_{{\mathrm{1}}}  \vee  \ell_{{\mathrm{2}}}   })  = \Id{}$.\\
In this case, $\mu^{\ell_{{\mathrm{1}}},\ell_{{\mathrm{2}}}} = \delta^{\ell_{{\mathrm{1}}},\ell_{{\mathrm{2}}}} = \I{}_{\Id{}}$.
\item $\neg ( \ell_{{\mathrm{1}}}  \sqsubseteq  \ell )$ or $\neg ( \ell_{{\mathrm{2}}}  \sqsubseteq  \ell )$. Then, $\neg (  \ell_{{\mathrm{1}}}  \vee  \ell_{{\mathrm{2}}}   \sqsubseteq  \ell )$. \\
So, $ \mathbf{S}^{  \ell  }({  \ell_{{\mathrm{1}}}  })  = \ast$ or $ \mathbf{S}^{  \ell  }({  \ell_{{\mathrm{2}}}  })  = \ast$. Hence, $  \mathbf{S}^{  \ell  }({  \ell_{{\mathrm{1}}}  })   \circ   \mathbf{S}^{  \ell  }({  \ell_{{\mathrm{2}}}  })   = \ast$. Further, $ \mathbf{S}^{  \ell  }({   \ell_{{\mathrm{1}}}  \vee  \ell_{{\mathrm{2}}}   })  = \ast$.\\
In this case, $\mu^{\ell_{{\mathrm{1}}},\ell_{{\mathrm{2}}}} = \delta^{\ell_{{\mathrm{1}}},\ell_{{\mathrm{2}}}} = \I{}_{\ast}$. 
\end{itemize}
Hence, $ \mathbf{S}^{  \ell  } $ is a strong (in fact a strict) monoidal functor from $\Ca( \mathcal{L} )$ to $\ES{}$.
\end{proof}


\begin{lemma}[Lemma \ref{dcceSlObs}] \label{dcceSlObsprf}
If $ \ell''  \sqsubseteq  \ottnt{A} $ and $\neg ( \ell''  \sqsubseteq  \ell )$, then $ \llbracket  \ottnt{A}  \rrbracket _{(\Ct, \mathbf{S}^{  \ell  } )} \cong \top$. 
\end{lemma}

\begin{proof}
By induction on $ \ell''  \sqsubseteq  \ottnt{A} $.
\begin{itemize}
\item \Rref{Prot-Prod}. Have: $ \ell''  \sqsubseteq   \ottnt{A}  \times  \ottnt{B}  $, where $ \ell''  \sqsubseteq  \ottnt{A} $ and $ \ell''  \sqsubseteq  \ottnt{B} $. Further, $\neg ( \ell''  \sqsubseteq  \ell )$.\\
Need to show: $ \llbracket    \ottnt{A}  \times  \ottnt{B}    \rrbracket _{(\Ct, \mathbf{S}^{  \ell  } )} \cong \top$.\\
By IH, $ \llbracket  \ottnt{A}  \rrbracket _{(\Ct, \mathbf{S}^{  \ell  } )} \cong \top$ and $ \llbracket  \ottnt{B}  \rrbracket _{(\Ct, \mathbf{S}^{  \ell  } )} \cong \top$.\\
Therefore, $ \llbracket    \ottnt{A}  \times  \ottnt{B}    \rrbracket _{(\Ct, \mathbf{S}^{  \ell  } )} =  \llbracket  \ottnt{A}  \rrbracket _{(\Ct, \mathbf{S}^{  \ell  } )} \times  \llbracket  \ottnt{B}  \rrbracket _{(\Ct, \mathbf{S}^{  \ell  } )} \cong \top \times \top \cong \top$.
\item \Rref{Prot-Fun}. Have: $ \ell''  \sqsubseteq   \ottnt{A}  \to  \ottnt{B}  $, where $ \ell''  \sqsubseteq  \ottnt{B} $. Further, $\neg ( \ell''  \sqsubseteq  \ell )$.\\
Need to show: $ \llbracket   \ottnt{A}  \to  \ottnt{B}   \rrbracket_{(\mathds{C},\mathbf{S}^{  \ell  })}  \cong \top$.\\
By IH, $ \llbracket  \ottnt{B}  \rrbracket_{(\mathds{C},\mathbf{S}^{  \ell  })}  \cong \top$.\\
Therefore, $ \llbracket   \ottnt{A}  \to  \ottnt{B}   \rrbracket_{(\mathds{C},\mathbf{S}^{  \ell  })}  =    (   \llbracket  \ottnt{B}  \rrbracket_{(\mathds{C},\mathbf{S}^{  \ell  })}   )  ^{  (   \llbracket  \ottnt{A}  \rrbracket_{(\mathds{C},\mathbf{S}^{  \ell  })}   )  }  \cong   \top  ^{  (   \llbracket  \ottnt{A}  \rrbracket_{(\mathds{C},\mathbf{S}^{  \ell  })}   )  }  \cong \top$.
\item \Rref{Prot-Monad}. Have: $ \ell''  \sqsubseteq   \mathcal{T}_{ \ell' } \:  \ottnt{A}  $, where $ \ell''  \sqsubseteq  \ell' $. Further, $\neg ( \ell''  \sqsubseteq  \ell )$.\\
Need to show: $ \llbracket   \mathcal{T}_{ \ell' } \:  \ottnt{A}   \rrbracket_{(\mathds{C},\mathbf{S}^{  \ell  })}  \cong \top$.\\
Now, since $ \ell''  \sqsubseteq  \ell' $ and $\neg( \ell''  \sqsubseteq  \ell )$, so $\neg( \ell'  \sqsubseteq  \ell )$.\\
Then, $ \llbracket   \mathcal{T}_{ \ell' } \:  \ottnt{A}   \rrbracket_{(\mathds{C},\mathbf{S}^{  \ell  })}  = ( \mathbf{S}^{  \ell  }({  \ell'  }) ) \llbracket  \ottnt{A}  \rrbracket_{(\mathds{C},\mathbf{S}^{  \ell  })}  = (\ast) \llbracket  \ottnt{A}  \rrbracket_{(\mathds{C},\mathbf{S}^{  \ell  })}  = \top$.
\item \Rref{Prot-Already}. Have: $ \ell''  \sqsubseteq   \mathcal{T}_{ \ell' } \:  \ottnt{A}  $, where $ \ell''  \sqsubseteq  \ottnt{A} $. Further $\neg( \ell''  \sqsubseteq  \ell )$.\\
Need to show: $ \llbracket   \mathcal{T}_{ \ell' } \:  \ottnt{A}   \rrbracket_{(\mathds{C},\mathbf{S}^{  \ell  })}  \cong \top$.\\
By IH, $ \llbracket  \ottnt{A}  \rrbracket_{(\mathds{C},\mathbf{S}^{  \ell  })}  \cong \top$.\\
Therefore, $ \llbracket   \mathcal{T}_{ \ell' } \:  \ottnt{A}   \rrbracket_{(\mathds{C},\mathbf{S}^{  \ell  })}  = ( \mathbf{S}^{  \ell  }({  \ell'  }) ) \llbracket  \ottnt{A}  \rrbracket_{(\mathds{C},\mathbf{S}^{  \ell  })}  \cong ( \mathbf{S}^{  \ell  }({  \ell'  }) )(\top) \cong \top$.
\item \Rref{Prot-Minimum}. Have: $  \bot   \sqsubseteq  \ottnt{A} $.\\
Further, $\neg(  \bot   \sqsubseteq  \ell )$. But this is a contradiction because for any $\ell \in L$, $  \bot   \sqsubseteq  \ell $.
\item \Rref{Prot-Combine}. Have: $ \ell''  \sqsubseteq  \ottnt{A} $, where $ \ell_{{\mathrm{1}}}  \sqsubseteq  \ottnt{A} $ and $ \ell_{{\mathrm{2}}}  \sqsubseteq  \ottnt{A} $ and $ \ell''  \sqsubseteq   \ell_{{\mathrm{1}}}  \vee  \ell_{{\mathrm{2}}}  $. Further, $\neg( \ell''  \sqsubseteq  \ell )$.\\
Now, if both $ \ell_{{\mathrm{1}}}  \sqsubseteq  \ell $ and $ \ell_{{\mathrm{2}}}  \sqsubseteq  \ell $ hold, then $ \ell''  \sqsubseteq   \ell_{{\mathrm{1}}}  \vee  \ell_{{\mathrm{2}}}   \sqsubseteq \ell$.\\
Therefore, either $\neg( \ell_{{\mathrm{1}}}  \sqsubseteq  \ell )$ or $\neg( \ell_{{\mathrm{2}}}  \sqsubseteq  \ell )$.\\
In either case, by IH, $ \llbracket  \ottnt{A}  \rrbracket_{(\mathds{C},\mathbf{S}^{  \ell  })}  \cong \top$.
\end{itemize}
\end{proof}


\begin{theorem}[Theorem \ref{dcceNI}]
Let $ \mathcal{L}  = (L,\vee,\bot)$ be the parametrizing semilattice.
\begin{itemize}
\item Suppose $\ell \in L$ such that $\neg ( \ell  \sqsubseteq   \bot  )$. Let $ \ell  \sqsubseteq  \ottnt{A} $. Let $  \emptyset   \vdash  \ottnt{f}  :   \ottnt{A}  \to   \mathbf{Bool}   $ and $  \emptyset   \vdash  \ottnt{a_{{\mathrm{1}}}}  :  \ottnt{A} $ and $  \emptyset   \vdash  \ottnt{a_{{\mathrm{2}}}}  :  \ottnt{A} $. Then, $ \vdash   \ottnt{f}  \:  \ottnt{a_{{\mathrm{1}}}}   \leadsto^{\ast}  \ottmv{v} $ if and only if $ \vdash   \ottnt{f}  \:  \ottnt{a_{{\mathrm{2}}}}   \leadsto^{\ast}  \ottmv{v} $, where $\ottmv{v}$ is a value of type $ \mathbf{Bool} $.

\vspace*{3pt}

\item Suppose $\ell , \ell' \in L$ such that $\neg ( \ell  \sqsubseteq  \ell' )$. Let $ \ell  \sqsubseteq  \ottnt{A} $. Let $  \emptyset   \vdash  \ottnt{f}  :   \ottnt{A}  \to   \mathcal{T}_{ \ell' } \:   \mathbf{Bool}    $ and $  \emptyset   \vdash  \ottnt{a_{{\mathrm{1}}}}  :  \ottnt{A} $ and $  \emptyset   \vdash  \ottnt{a_{{\mathrm{2}}}}  :  \ottnt{A} $. Then, $ \vdash   \ottnt{f}  \:  \ottnt{a_{{\mathrm{1}}}}   \leadsto^{\ast}  \ottmv{v} $ if and only if $ \vdash   \ottnt{f}  \:  \ottnt{a_{{\mathrm{2}}}}   \leadsto^{\ast}  \ottmv{v} $, where $\ottmv{v}$ is a value of type $ \mathcal{T}_{ \ell' } \:   \mathbf{Bool}  $.
\end{itemize}
\end{theorem}

\begin{proof}
For the first part, note that $(\mathbf{Set}, \mathbf{S}^{   \bot   } )$ gives us a computationally adequate model of \ED{}($ \mathcal{L} $) (by Theorem \ref{dcceSlAprf}).\\
Next, since $ \ell  \sqsubseteq  \ottnt{A} $ and $\neg( \ell  \sqsubseteq   \bot  )$, by Lemma \ref{dcceSlObsprf}, $ \llbracket  \ottnt{A}  \rrbracket_{(\mathds{C},\mathbf{S}^{   \bot   })}  \cong \top$.\\
Therefore, $ \llbracket    \ottnt{f}  \:  \ottnt{a_{{\mathrm{1}}}}    \rrbracket_{(\mathbf{Set},\mathbf{S}^{   \bot   })}  =  \llbracket    \ottnt{f}  \:  \ottnt{a_{{\mathrm{2}}}}    \rrbracket_{(\mathbf{Set},\mathbf{S}^{   \bot   })}  =   \text{app}   \circ   \langle   \llbracket  \ottnt{f}  \rrbracket_{(\mathbf{Set},\mathbf{S}^{   \bot   })}   ,   \langle \rangle   \rangle  $.\\
Then, by Theorem \ref{dcceCAprf}, $ \vdash   \ottnt{f}  \:  \ottnt{a_{{\mathrm{1}}}}   \leadsto^{\ast}  \ottmv{v} $ if and only if $ \vdash   \ottnt{f}  \:  \ottnt{a_{{\mathrm{2}}}}   \leadsto^{\ast}  \ottmv{v} $.\\

For the second part, note that $(\mathbf{Set}, \mathbf{S}^{  \ell'  } )$ gives us a computationally adequate model of \ED{}($ \mathcal{L} $) (by Theorem \ref{dcceSlAprf}). Further, note that $\overline{\eta} \triangleq ( \mathbf{S}^{  \ell'  }  (  \bot   \sqsubseteq  \ell' )) \circ \eta = \I{}_{\Id{}} \circ \I{}_{\Id{}} = \I{}_{\Id{}}$. So, the morphisms $ \overline{\eta}_{ \ottnt{X} } $ are mono for any $X \in \text{Obj}(\Ct)$. \\
Next, since $ \ell  \sqsubseteq  \ottnt{A} $ and $\neg( \ell  \sqsubseteq  \ell' )$, by Lemma \ref{dcceSlObsprf}, $ \llbracket  \ottnt{A}  \rrbracket_{(\mathds{C},\mathbf{S}^{  \ell'  })}  \cong \top$.\\
Therefore, $ \llbracket    \ottnt{f}  \:  \ottnt{a_{{\mathrm{1}}}}    \rrbracket_{(\mathbf{Set},\mathbf{S}^{  \ell'  })}  =  \llbracket    \ottnt{f}  \:  \ottnt{a_{{\mathrm{2}}}}    \rrbracket_{(\mathbf{Set},\mathbf{S}^{  \ell'  })}  =   \text{app}   \circ   \langle   \llbracket  \ottnt{f}  \rrbracket_{(\mathbf{Set},\mathbf{S}^{  \ell'  })}   ,   \langle \rangle   \rangle  $.\\
Then, by Theorem \ref{dcceCAprf}, $ \vdash   \ottnt{f}  \:  \ottnt{a_{{\mathrm{1}}}}   \leadsto^{\ast}  \ottmv{v} $ if and only if $ \vdash   \ottnt{f}  \:  \ottnt{a_{{\mathrm{2}}}}   \leadsto^{\ast}  \ottmv{v} $.
\end{proof}


\section{Proofs of lemmas/theorems stated in Section \ref{lcirc}}

\begin{theorem}[Theorem \ref{lcCat}] \label{lcCatprf}
If $ \Gamma  \vdash  \ottnt{a}  :^{ n }  \ottnt{A} $ in \lc{}, then $ \llbracket  \ottnt{a}  \rrbracket  \in \text{Hom}_{\Ct} ( \llbracket  \Gamma  \rrbracket , \mathbf{S}_{n}  \llbracket  \ottnt{A}  \rrbracket )$. Further, if $ \Gamma  \vdash  \ottnt{a_{{\mathrm{1}}}}  :^{ n }  \ottnt{A} $ and $ \Gamma  \vdash  \ottnt{a_{{\mathrm{2}}}}  :^{ n }  \ottnt{A} $ such that $ \ottnt{a_{{\mathrm{1}}}}  \equiv  \ottnt{a_{{\mathrm{2}}}} $ in \lc{}, then $ \llbracket  \ottnt{a_{{\mathrm{1}}}}  \rrbracket  =  \llbracket  \ottnt{a_{{\mathrm{2}}}}  \rrbracket  \in \text{Hom}_{\Ct} ( \llbracket  \Gamma  \rrbracket , \mathbf{S}_{n}  \llbracket  \ottnt{A}  \rrbracket )$. 
\end{theorem}

\begin{proof}
The first part follows by induction on the typing derivation. Most of the cases are straightforward, given the detailed interpretation in Section \ref{seclcCat}. However, we elaborate on the case of $\mathbf{case}$-expressions.
\begin{itemize}
\item \Rref{LC-Case}. Have: $ \Gamma  \vdash   \mathbf{case} \:  \ottnt{a}  \: \mathbf{of} \:  \ottnt{b_{{\mathrm{1}}}}  \: ; \:  \ottnt{b_{{\mathrm{2}}}}   :^{ n }  \ottnt{B} $ where $ \Gamma  \vdash  \ottnt{a}  :^{ n }   \ottnt{A_{{\mathrm{1}}}}  +  \ottnt{A_{{\mathrm{2}}}}  $ and $ \Gamma  \vdash  \ottnt{b_{{\mathrm{1}}}}  :^{ n }   \ottnt{A_{{\mathrm{1}}}}  \to  \ottnt{B}  $ and $ \Gamma  \vdash  \ottnt{b_{{\mathrm{2}}}}  :^{ n }   \ottnt{A_{{\mathrm{2}}}}  \to  \ottnt{B}  $.\\
Need to show: $ \llbracket   \mathbf{case} \:  \ottnt{a}  \: \mathbf{of} \:  \ottnt{b_{{\mathrm{1}}}}  \: ; \:  \ottnt{b_{{\mathrm{2}}}}   \rrbracket  \in \text{Hom}_{\Ct}( \llbracket  \Gamma  \rrbracket ,  \mathbf{S}_{  n  }   \llbracket  \ottnt{B}  \rrbracket  )$.\\
By IH, $ \llbracket  \ottnt{a}  \rrbracket  \in \text{Hom}_{\Ct}( \llbracket  \Gamma  \rrbracket ,  \mathbf{S}_{  n  }   (    \llbracket  \ottnt{A_{{\mathrm{1}}}}  \rrbracket   +   \llbracket  \ottnt{A_{{\mathrm{2}}}}  \rrbracket    )  )$ and $ \llbracket  \ottnt{b_{{\mathrm{1}}}}  \rrbracket  \in \text{Hom}_{\Ct}( \llbracket  \Gamma  \rrbracket , \mathbf{S}_{  n  }   (    \llbracket  \ottnt{B}  \rrbracket  ^{  \llbracket  \ottnt{A_{{\mathrm{1}}}}  \rrbracket  }   )  )$ and $ \llbracket  \ottnt{b_{{\mathrm{2}}}}  \rrbracket  \in \text{Hom}_{\Ct}( \llbracket  \Gamma  \rrbracket , \mathbf{S}_{  n  }   (    \llbracket  \ottnt{B}  \rrbracket  ^{  \llbracket  \ottnt{A_{{\mathrm{2}}}}  \rrbracket  }   )  )$.\\
Now,
\begin{align*}
 \llbracket    \mathbf{case} \:  \ottnt{a}  \: \mathbf{of} \:  \ottnt{b_{{\mathrm{1}}}}  \: ; \:  \ottnt{b_{{\mathrm{2}}}}    \rrbracket  =  \llbracket  \Gamma  \rrbracket  & \xrightarrow{ \langle   \langle   \llbracket  \ottnt{b_{{\mathrm{1}}}}  \rrbracket   ,   \llbracket  \ottnt{b_{{\mathrm{2}}}}  \rrbracket   \rangle   ,   \llbracket  \ottnt{a}  \rrbracket   \rangle }   (    \mathbf{S}_{  n  }   (    \llbracket  \ottnt{B}  \rrbracket  ^{  \llbracket  \ottnt{A_{{\mathrm{1}}}}  \rrbracket  }   )    \times   \mathbf{S}_{  n  }   (    \llbracket  \ottnt{B}  \rrbracket  ^{  \llbracket  \ottnt{A_{{\mathrm{2}}}}  \rrbracket  }   )     )   \times   \mathbf{S}_{  n  }   (    \llbracket  \ottnt{A_{{\mathrm{1}}}}  \rrbracket   +   \llbracket  \ottnt{A_{{\mathrm{2}}}}  \rrbracket    )    \hspace{3pt} [\text{By IH}] \\ & \xrightarrow{   p^{-1}_{    \llbracket  \ottnt{B}  \rrbracket  ^{  \llbracket  \ottnt{A_{{\mathrm{1}}}}  \rrbracket  }   ,     \llbracket  \ottnt{B}  \rrbracket  ^{  \llbracket  \ottnt{A_{{\mathrm{2}}}}  \rrbracket  }   }    \times   \text{id}  }   \mathbf{S}_{  n  }   (      \llbracket  \ottnt{B}  \rrbracket  ^{  \llbracket  \ottnt{A_{{\mathrm{1}}}}  \rrbracket  }    \times     \llbracket  \ottnt{B}  \rrbracket  ^{  \llbracket  \ottnt{A_{{\mathrm{2}}}}  \rrbracket  }     )    \times   \mathbf{S}_{  n  }   (    \llbracket  \ottnt{A_{{\mathrm{1}}}}  \rrbracket   +   \llbracket  \ottnt{A_{{\mathrm{2}}}}  \rrbracket    )    \\ & \xrightarrow{ p^{-1}_{     \llbracket  \ottnt{B}  \rrbracket  ^{  \llbracket  \ottnt{A_{{\mathrm{1}}}}  \rrbracket  }    \times     \llbracket  \ottnt{B}  \rrbracket  ^{  \llbracket  \ottnt{A_{{\mathrm{2}}}}  \rrbracket  }    ,     \llbracket  \ottnt{A_{{\mathrm{1}}}}  \rrbracket   +   \llbracket  \ottnt{A_{{\mathrm{2}}}}  \rrbracket    } }  \mathbf{S}_{  n  }   (    (      \llbracket  \ottnt{B}  \rrbracket  ^{  \llbracket  \ottnt{A_{{\mathrm{1}}}}  \rrbracket  }    \times     \llbracket  \ottnt{B}  \rrbracket  ^{  \llbracket  \ottnt{A_{{\mathrm{2}}}}  \rrbracket  }     )   \times   (    \llbracket  \ottnt{A_{{\mathrm{1}}}}  \rrbracket   +   \llbracket  \ottnt{A_{{\mathrm{2}}}}  \rrbracket    )    )   \\ & \xrightarrow{ \mathbf{S}_{  n  }   (   \ottnt{h}  \times   \text{id}    )  }  \mathbf{S}_{  n  }   (      \llbracket  \ottnt{B}  \rrbracket  ^{    \llbracket  \ottnt{A_{{\mathrm{1}}}}  \rrbracket   +   \llbracket  \ottnt{A_{{\mathrm{2}}}}  \rrbracket    }    \times   (    \llbracket  \ottnt{A_{{\mathrm{1}}}}  \rrbracket   +   \llbracket  \ottnt{A_{{\mathrm{2}}}}  \rrbracket    )    )   \xrightarrow{ \mathbf{S}_{  n  }   \text{app}  }  \mathbf{S}_{  n  }   \llbracket  \ottnt{B}  \rrbracket  
\end{align*}
where $\ottnt{h} =  \Lambda \Big(      [    \Lambda^{-1}   \pi_1    \circ   \langle   \pi_2   ,   \pi_1   \rangle    ,    \Lambda^{-1}   \pi_2    \circ   \langle   \pi_2   ,   \pi_1   \rangle    ]   \circ   \Lambda^{-1}   [   \Lambda   i_1    ,   \Lambda   i_2    ]     \circ   \langle   \pi_2   ,   \pi_1   \rangle     \Big) $.\\

Note that $\ottnt{h} :     \llbracket  \ottnt{B}  \rrbracket  ^{  \llbracket  \ottnt{A_{{\mathrm{1}}}}  \rrbracket  }    \times     \llbracket  \ottnt{B}  \rrbracket  ^{  \llbracket  \ottnt{A_{{\mathrm{2}}}}  \rrbracket  }    \to   \llbracket  \ottnt{B}  \rrbracket  ^{    \llbracket  \ottnt{A_{{\mathrm{1}}}}  \rrbracket   +   \llbracket  \ottnt{A_{{\mathrm{2}}}}  \rrbracket    } $ is an isomorphism, where \[ h^{-1} =  \langle    \llbracket  \ottnt{B}  \rrbracket  ^{  i_1  }   ,    \llbracket  \ottnt{B}  \rrbracket  ^{  i_2  }   \rangle  :   \llbracket  \ottnt{B}  \rrbracket  ^{    \llbracket  \ottnt{A_{{\mathrm{1}}}}  \rrbracket   +   \llbracket  \ottnt{A_{{\mathrm{2}}}}  \rrbracket    }  \to     \llbracket  \ottnt{B}  \rrbracket  ^{  \llbracket  \ottnt{A_{{\mathrm{1}}}}  \rrbracket  }    \times     \llbracket  \ottnt{B}  \rrbracket  ^{  \llbracket  \ottnt{A_{{\mathrm{2}}}}  \rrbracket  }    \] Recall that for $X , Y, Z \in \text{Obj}(\Ct)$ and $f \in \text{Hom}_{\Ct} (Y, Z)$, \[ X^f \triangleq \Lambda \Big( X^Z \times Y \xrightarrow{  \text{id}   \times  \ottnt{f} } X^Z \times Z \xrightarrow{ \text{app} } X \Big) \in \text{Hom}_{\Ct} (X^Z,X^Y) \]
\end{itemize}

For the second part, we invert the equality judgement.
\begin{itemize}
\item $   (   \lambda  \ottmv{x}  :  \ottnt{A}  .  \ottnt{b}   )   \:  \ottnt{a}   \equiv   \ottnt{b}  \{  \ottnt{a}  /  \ottmv{x}  \}  $. \\
Now, \begin{align*}
&  \llbracket     (   \lambda  \ottmv{x}  :  \ottnt{A}  .  \ottnt{b}   )   \:  \ottnt{a}    \rrbracket  \\
= &   \text{app}   \circ   \langle     q_{  \llbracket  \ottnt{A}  \rrbracket  ,   \llbracket  \ottnt{B}  \rrbracket  }    \circ   \llbracket    \lambda  \ottmv{x}  :  \ottnt{A}  .  \ottnt{b}    \rrbracket    ,   \llbracket  \ottnt{a}  \rrbracket   \rangle   \\
= &   \text{app}   \circ   \langle      q_{  \llbracket  \ottnt{A}  \rrbracket  ,   \llbracket  \ottnt{B}  \rrbracket  }    \circ    q^{-1}_{  \llbracket  \ottnt{A}  \rrbracket  ,   \llbracket  \ottnt{B}  \rrbracket  }     \circ   \Lambda   \llbracket  \ottnt{b}  \rrbracket     ,   \llbracket  \ottnt{a}  \rrbracket   \rangle   \\
= &   \text{app}   \circ   \langle   \Lambda   \llbracket  \ottnt{b}  \rrbracket    ,   \llbracket  \ottnt{a}  \rrbracket   \rangle   =   \llbracket  \ottnt{b}  \rrbracket   \circ   \langle   \text{id}   ,   \llbracket  \ottnt{a}  \rrbracket   \rangle   =  \llbracket    \ottnt{b}  \{  \ottnt{a}  /  \ottmv{x}  \}    \rrbracket .
\end{align*}
\item $ \ottnt{b}  \equiv    \lambda  \ottmv{x}  :  \ottnt{A}  .  \ottnt{b}   \:  \ottmv{x}  $.\\
Now, \begin{align*}
&  \llbracket     \lambda  \ottmv{x}  :  \ottnt{A}  .  \ottnt{b}   \:  \ottmv{x}    \rrbracket  \\
= &    q^{-1}_{  \llbracket  \ottnt{A}  \rrbracket  ,   \llbracket  \ottnt{B}  \rrbracket  }    \circ   \Lambda   \llbracket    \ottnt{b}  \:  \ottmv{x}    \rrbracket    \\
= &    q^{-1}_{  \llbracket  \ottnt{A}  \rrbracket  ,   \llbracket  \ottnt{B}  \rrbracket  }    \circ   \Lambda \Big(     \text{app}   \circ   \langle      q_{  \llbracket  \ottnt{A}  \rrbracket  ,   \llbracket  \ottnt{B}  \rrbracket  }    \circ   \llbracket  \ottnt{b}  \rrbracket    \circ   \pi_1    ,   \pi_2   \rangle     \Big)   \\
= &     q^{-1}_{  \llbracket  \ottnt{A}  \rrbracket  ,   \llbracket  \ottnt{B}  \rrbracket  }    \circ    q_{  \llbracket  \ottnt{A}  \rrbracket  ,   \llbracket  \ottnt{B}  \rrbracket  }     \circ   \llbracket  \ottnt{b}  \rrbracket   =  \llbracket  \ottnt{b}  \rrbracket .
\end{align*}
\item $  \mathbf{prev} \:   (   \mathbf{next} \:  \ottnt{a}   )    \equiv  \ottnt{a} $.\\
Now, $ \llbracket    \mathbf{prev} \:   (   \mathbf{next} \:  \ottnt{a}   )     \rrbracket  =    \mu^{ n ,  0'  }_{  \llbracket  \ottnt{A}  \rrbracket  }   \circ   \delta^{ n ,  0'  }_{  \llbracket  \ottnt{A}  \rrbracket  }    \circ   \llbracket  \ottnt{a}  \rrbracket   =  \llbracket  \ottnt{a}  \rrbracket $.
\item $  \mathbf{next} \:   (   \mathbf{prev} \:  \ottnt{a}   )    \equiv  \ottnt{a} $.\\
$ \llbracket    \mathbf{next} \:   (   \mathbf{prev} \:  \ottnt{a}   )     \rrbracket  =    \delta^{ n ,  0'  }_{  \llbracket  \ottnt{A}  \rrbracket  }   \circ   \mu^{ n ,  0'  }_{  \llbracket  \ottnt{A}  \rrbracket  }    \circ   \llbracket  \ottnt{a}  \rrbracket   =  \llbracket  \ottnt{a}  \rrbracket $.
\item $  \mathbf{case} \:   \mathbf{inj}_1 \:  \ottnt{a_{{\mathrm{1}}}}   \: \mathbf{of} \:  \ottnt{b_{{\mathrm{1}}}}  \: ; \:  \ottnt{b_{{\mathrm{2}}}}   \equiv   \ottnt{b_{{\mathrm{1}}}}  \:  \ottnt{a_{{\mathrm{1}}}}  $.\\
Now, \begin{align*}
&  \llbracket    \mathbf{case} \:   \mathbf{inj}_1 \:  \ottnt{a_{{\mathrm{1}}}}   \: \mathbf{of} \:  \ottnt{b_{{\mathrm{1}}}}  \: ; \:  \ottnt{b_{{\mathrm{2}}}}    \rrbracket  \\
= &      \mathbf{S}_{  n  }   \text{app}    \circ   \mathbf{S}_{  n  }   (   \ottnt{h}  \times   \text{id}    )     \circ   p^{-1}_{     \llbracket  \ottnt{B}  \rrbracket  ^{  \llbracket  \ottnt{A_{{\mathrm{1}}}}  \rrbracket  }    \times     \llbracket  \ottnt{B}  \rrbracket  ^{  \llbracket  \ottnt{A_{{\mathrm{2}}}}  \rrbracket  }    ,     \llbracket  \ottnt{A_{{\mathrm{1}}}}  \rrbracket   +   \llbracket  \ottnt{A_{{\mathrm{2}}}}  \rrbracket    }    \circ     p^{-1}_{    \llbracket  \ottnt{B}  \rrbracket  ^{  \llbracket  \ottnt{A_{{\mathrm{1}}}}  \rrbracket  }   ,     \llbracket  \ottnt{B}  \rrbracket  ^{  \llbracket  \ottnt{A_{{\mathrm{2}}}}  \rrbracket  }   }   \times   \text{id}      \circ   \langle   \langle   \llbracket  \ottnt{b_{{\mathrm{1}}}}  \rrbracket   ,   \llbracket  \ottnt{b_{{\mathrm{2}}}}  \rrbracket   \rangle   ,    \mathbf{S}_{  n  }   i_1    \circ   \llbracket  \ottnt{a_{{\mathrm{1}}}}  \rrbracket    \rangle   \\
= &     \mathbf{S}_{  n  }   (     [    \Lambda^{-1}   \pi_1    \circ   \langle   \pi_2   ,   \pi_1   \rangle    ,    \Lambda^{-1}   \pi_2    \circ   \langle   \pi_2   ,   \pi_1   \rangle    ]   \circ   \Lambda^{-1}   [   \Lambda   i_1    ,   \Lambda   i_2    ]     \circ   \langle   \pi_2   ,   \pi_1   \rangle    )    \circ   p^{-1}    \circ   p^{-1}    \times   \text{id}   \\ & \hspace*{30pt} \circ  \langle   \langle   \llbracket  \ottnt{b_{{\mathrm{1}}}}  \rrbracket   ,   \llbracket  \ottnt{b_{{\mathrm{2}}}}  \rrbracket   \rangle   ,    \mathbf{S}_{  n  }   i_1    \circ   \llbracket  \ottnt{a_{{\mathrm{1}}}}  \rrbracket    \rangle  \\
= &      \mathbf{S}_{  n  }   (    [    \Lambda^{-1}   \pi_1    \circ   \langle   \pi_2   ,   \pi_1   \rangle    ,    \Lambda^{-1}   \pi_2    \circ   \langle   \pi_2   ,   \pi_1   \rangle    ]   \circ   \Lambda^{-1}   [   \Lambda   i_1    ,   \Lambda   i_2    ]     )    \circ   p^{-1}    \circ   \langle   \pi_2   ,   \pi_1   \rangle    \circ   p^{-1}    \times   \text{id}   \\ & \hspace*{30pt} \circ  \langle   \langle   \llbracket  \ottnt{b_{{\mathrm{1}}}}  \rrbracket   ,   \llbracket  \ottnt{b_{{\mathrm{2}}}}  \rrbracket   \rangle   ,    \mathbf{S}_{  n  }   i_1    \circ   \llbracket  \ottnt{a_{{\mathrm{1}}}}  \rrbracket    \rangle  \\
= &    \mathbf{S}_{  n  }   (    [    \Lambda^{-1}   \pi_1    \circ   \langle   \pi_2   ,   \pi_1   \rangle    ,    \Lambda^{-1}   \pi_2    \circ   \langle   \pi_2   ,   \pi_1   \rangle    ]   \circ   \Lambda^{-1}   [   \Lambda   i_1    ,   \Lambda   i_2    ]     )    \circ   p^{-1}    \circ   \langle    \mathbf{S}_{  n  }   i_1    \circ   \llbracket  \ottnt{a_{{\mathrm{1}}}}  \rrbracket    ,    p^{-1}   \circ   \langle   \llbracket  \ottnt{b_{{\mathrm{1}}}}  \rrbracket   ,   \llbracket  \ottnt{b_{{\mathrm{2}}}}  \rrbracket   \rangle    \rangle   \\
= &     \mathbf{S}_{  n  }   (    [    \Lambda^{-1}   \pi_1    \circ   \langle   \pi_2   ,   \pi_1   \rangle    ,    \Lambda^{-1}   \pi_2    \circ   \langle   \pi_2   ,   \pi_1   \rangle    ]   \circ   \Lambda^{-1}   [   \Lambda   i_1    ,   \Lambda   i_2    ]     )    \circ   p^{-1}    \circ     \mathbf{S}_{  n  }   i_1    \times   \text{id}      \circ   \langle   \llbracket  \ottnt{a_{{\mathrm{1}}}}  \rrbracket   ,    p^{-1}   \circ   \langle   \llbracket  \ottnt{b_{{\mathrm{1}}}}  \rrbracket   ,   \llbracket  \ottnt{b_{{\mathrm{2}}}}  \rrbracket   \rangle    \rangle   \\
= &     \mathbf{S}_{  n  }   (    [    \Lambda^{-1}   \pi_1    \circ   \langle   \pi_2   ,   \pi_1   \rangle    ,    \Lambda^{-1}   \pi_2    \circ   \langle   \pi_2   ,   \pi_1   \rangle    ]   \circ   \Lambda^{-1}   [   \Lambda   i_1    ,   \Lambda   i_2    ]     )    \circ   \mathbf{S}_{  n  }   (    i_1   \times   \text{id}    )     \circ   p^{-1}    \circ   \langle   \llbracket  \ottnt{a_{{\mathrm{1}}}}  \rrbracket   ,    p^{-1}   \circ   \langle   \llbracket  \ottnt{b_{{\mathrm{1}}}}  \rrbracket   ,   \llbracket  \ottnt{b_{{\mathrm{2}}}}  \rrbracket   \rangle    \rangle   \\
= &      \mathbf{S}_{  n  }   [    \Lambda^{-1}   \pi_1    \circ   \langle   \pi_2   ,   \pi_1   \rangle    ,    \Lambda^{-1}   \pi_2    \circ   \langle   \pi_2   ,   \pi_1   \rangle    ]    \circ   \mathbf{S}_{  n  }   \Lambda^{-1}   [   \Lambda   i_1    ,   \Lambda   i_2    ]      \circ   \mathbf{S}_{  n  }   (    i_1   \times   \text{id}    )     \circ   p^{-1}    \circ   \langle   \llbracket  \ottnt{a_{{\mathrm{1}}}}  \rrbracket   ,    p^{-1}   \circ   \langle   \llbracket  \ottnt{b_{{\mathrm{1}}}}  \rrbracket   ,   \llbracket  \ottnt{b_{{\mathrm{2}}}}  \rrbracket   \rangle    \rangle   \\
= &     \mathbf{S}_{  n  }   [    \Lambda^{-1}   \pi_1    \circ   \langle   \pi_2   ,   \pi_1   \rangle    ,    \Lambda^{-1}   \pi_2    \circ   \langle   \pi_2   ,   \pi_1   \rangle    ]    \circ   \mathbf{S}_{  n  }   i_1     \circ   p^{-1}    \circ   \langle   \llbracket  \ottnt{a_{{\mathrm{1}}}}  \rrbracket   ,    p^{-1}   \circ   \langle   \llbracket  \ottnt{b_{{\mathrm{1}}}}  \rrbracket   ,   \llbracket  \ottnt{b_{{\mathrm{2}}}}  \rrbracket   \rangle    \rangle   \\
= &     \mathbf{S}_{  n  }   \Lambda^{-1}   \pi_1     \circ   \mathbf{S}_{  n  }   \langle   \pi_2   ,   \pi_1   \rangle     \circ   p^{-1}    \circ   \langle   \llbracket  \ottnt{a_{{\mathrm{1}}}}  \rrbracket   ,    p^{-1}   \circ   \langle   \llbracket  \ottnt{b_{{\mathrm{1}}}}  \rrbracket   ,   \llbracket  \ottnt{b_{{\mathrm{2}}}}  \rrbracket   \rangle    \rangle   \\
= &     \mathbf{S}_{  n  }   \Lambda^{-1}   \pi_1     \circ   p^{-1}    \circ   \langle   \pi_2   ,   \pi_1   \rangle    \circ   \langle   \llbracket  \ottnt{a_{{\mathrm{1}}}}  \rrbracket   ,    p^{-1}   \circ   \langle   \llbracket  \ottnt{b_{{\mathrm{1}}}}  \rrbracket   ,   \llbracket  \ottnt{b_{{\mathrm{2}}}}  \rrbracket   \rangle    \rangle   \\
= &     \mathbf{S}_{  n  }   \text{app}    \circ   \mathbf{S}_{  n  }   (    \pi_1   \times   \text{id}    )     \circ   p^{-1}    \circ   \langle    p^{-1}   \circ   \langle   \llbracket  \ottnt{b_{{\mathrm{1}}}}  \rrbracket   ,   \llbracket  \ottnt{b_{{\mathrm{2}}}}  \rrbracket   \rangle    ,   \llbracket  \ottnt{a_{{\mathrm{1}}}}  \rrbracket   \rangle   \\
= &      \mathbf{S}_{  n  }   \text{app}    \circ   p^{-1}    \circ    \mathbf{S}_{  n  }   \pi_1      \times   \text{id}    \circ   \langle    p^{-1}   \circ   \langle   \llbracket  \ottnt{b_{{\mathrm{1}}}}  \rrbracket   ,   \llbracket  \ottnt{b_{{\mathrm{2}}}}  \rrbracket   \rangle    ,   \llbracket  \ottnt{a_{{\mathrm{1}}}}  \rrbracket   \rangle   \\
= &        \mathbf{S}_{  n  }   \text{app}    \circ   p^{-1}    \circ    \mathbf{S}_{  n  }   \pi_1      \times   \text{id}    \circ   p^{-1}    \times   \text{id}    \circ   \langle   \langle   \llbracket  \ottnt{b_{{\mathrm{1}}}}  \rrbracket   ,   \llbracket  \ottnt{b_{{\mathrm{2}}}}  \rrbracket   \rangle   ,   \llbracket  \ottnt{a_{{\mathrm{1}}}}  \rrbracket   \rangle   \\
= &          \mathbf{S}_{  n  }   \text{app}    \circ   p^{-1}    \circ   \pi_1    \times   \text{id}    \circ   p    \times   \text{id}    \circ   p^{-1}    \times   \text{id}    \circ   \langle   \langle   \llbracket  \ottnt{b_{{\mathrm{1}}}}  \rrbracket   ,   \llbracket  \ottnt{b_{{\mathrm{2}}}}  \rrbracket   \rangle   ,   \llbracket  \ottnt{a_{{\mathrm{1}}}}  \rrbracket   \rangle   \\
= &    \mathbf{S}_{  n  }   \text{app}    \circ   p^{-1}    \circ   \langle   \llbracket  \ottnt{b_{{\mathrm{1}}}}  \rrbracket   ,   \llbracket  \ottnt{a_{{\mathrm{1}}}}  \rrbracket   \rangle   \\
= &     \text{app}   \circ   \Lambda   (    \mathbf{S}_{  n  }   \text{app}    \circ   p^{-1}    )     \times   \text{id}    \circ   \langle   \llbracket  \ottnt{b_{{\mathrm{1}}}}  \rrbracket   ,   \llbracket  \ottnt{a_{{\mathrm{1}}}}  \rrbracket   \rangle   =   \text{app}   \circ   \langle    q   \circ   \llbracket  \ottnt{b_{{\mathrm{1}}}}  \rrbracket    ,   \llbracket  \ottnt{a_{{\mathrm{1}}}}  \rrbracket   \rangle   =  \llbracket    \ottnt{b_{{\mathrm{1}}}}  \:  \ottnt{a_{{\mathrm{1}}}}    \rrbracket 
\end{align*}
\item $  \ottnt{b}  \:  \ottnt{a}   \equiv    \mathbf{case} \:  \ottnt{a}  \: \mathbf{of} \:    \lambda  \ottmv{x_{{\mathrm{1}}}}  .  \ottnt{b}   \:   (   \mathbf{inj}_1 \:  \ottmv{x_{{\mathrm{1}}}}   )    \: ; \:   \lambda  \ottmv{x_{{\mathrm{2}}}}  .  \ottnt{b}    \:   (   \mathbf{inj}_2 \:  \ottmv{x_{{\mathrm{2}}}}   )    $.\\
Note that \begin{align*}
&  \llbracket     \lambda  \ottmv{x_{{\mathrm{1}}}}  .  \ottnt{b}   \:   (   \mathbf{inj}_1 \:  \ottmv{x_{{\mathrm{1}}}}   )     \rrbracket  \\
= &   q^{-1}_{  \llbracket  \ottnt{A_{{\mathrm{1}}}}  \rrbracket  ,   \llbracket  \ottnt{B}  \rrbracket  }   \circ   \Lambda \Big(     \text{app}   \circ   \langle     q_{    \llbracket  \ottnt{A_{{\mathrm{1}}}}  \rrbracket   +   \llbracket  \ottnt{A_{{\mathrm{2}}}}  \rrbracket    ,   \llbracket  \ottnt{B}  \rrbracket  }   \circ   \llbracket  \ottnt{b}  \rrbracket    \circ   \pi_1    ,    \mathbf{S}_{  n  }   i_1    \circ   \pi_2    \rangle     \Big)  \\
= &   q^{-1}_{  \llbracket  \ottnt{A_{{\mathrm{1}}}}  \rrbracket  ,   \llbracket  \ottnt{B}  \rrbracket  }   \circ   \Lambda \Big(       \text{app}   \circ   (    \text{id}   \times   \mathbf{S}_{  n  }   i_1     )    \circ   (    q_{    \llbracket  \ottnt{A_{{\mathrm{1}}}}  \rrbracket   +   \llbracket  \ottnt{A_{{\mathrm{2}}}}  \rrbracket    ,   \llbracket  \ottnt{B}  \rrbracket  }   \circ   \llbracket  \ottnt{b}  \rrbracket    )    \times   \text{id}     \Big)   \\
= &     q^{-1}_{  \llbracket  \ottnt{A_{{\mathrm{1}}}}  \rrbracket  ,   \llbracket  \ottnt{B}  \rrbracket  }   \circ   \Lambda   (    \text{app}   \circ   (    \text{id}   \times   \mathbf{S}_{  n  }   i_1     )    )     \circ   q_{    \llbracket  \ottnt{A_{{\mathrm{1}}}}  \rrbracket   +   \llbracket  \ottnt{A_{{\mathrm{2}}}}  \rrbracket    ,   \llbracket  \ottnt{B}  \rrbracket  }    \circ   \llbracket  \ottnt{b}  \rrbracket   \\
= &      q^{-1}_{  \llbracket  \ottnt{A_{{\mathrm{1}}}}  \rrbracket  ,   \llbracket  \ottnt{B}  \rrbracket  }    \circ     (   \mathbf{S}_{  n  }   \llbracket  \ottnt{B}  \rrbracket    )  ^{   \mathbf{S}_{  n  }   i_1    }     \circ   q_{    \llbracket  \ottnt{A_{{\mathrm{1}}}}  \rrbracket   +   \llbracket  \ottnt{A_{{\mathrm{2}}}}  \rrbracket    ,   \llbracket  \ottnt{B}  \rrbracket  }    \circ   \llbracket  \ottnt{b}  \rrbracket   \\
= &     \mathbf{S}_{  n  }     \llbracket  \ottnt{B}  \rrbracket  ^{  i_1  }     \circ   q^{-1}_{    \llbracket  \ottnt{A_{{\mathrm{1}}}}  \rrbracket   +   \llbracket  \ottnt{A_{{\mathrm{2}}}}  \rrbracket    ,   \llbracket  \ottnt{B}  \rrbracket  }    \circ   q_{    \llbracket  \ottnt{A_{{\mathrm{1}}}}  \rrbracket   +   \llbracket  \ottnt{A_{{\mathrm{2}}}}  \rrbracket    ,   \llbracket  \ottnt{B}  \rrbracket  }    \circ   \llbracket  \ottnt{b}  \rrbracket   =   \mathbf{S}_{  n  }     \llbracket  \ottnt{B}  \rrbracket  ^{  i_1  }     \circ   \llbracket  \ottnt{b}  \rrbracket  
\end{align*}
Similarly, $ \llbracket     \lambda  \ottmv{x_{{\mathrm{2}}}}  .  \ottnt{b}   \:   (   \mathbf{inj}_2 \:  \ottmv{x_{{\mathrm{2}}}}   )     \rrbracket  =   \mathbf{S}_{  n  }     \llbracket  \ottnt{B}  \rrbracket  ^{  i_2  }     \circ   \llbracket  \ottnt{b}  \rrbracket  $.\\
Now,\begin{align*}
&  \llbracket     \mathbf{case} \:  \ottnt{a}  \: \mathbf{of} \:    \lambda  \ottmv{x_{{\mathrm{1}}}}  .  \ottnt{b}   \:   (   \mathbf{inj}_1 \:  \ottmv{x_{{\mathrm{1}}}}   )    \: ; \:   \lambda  \ottmv{x_{{\mathrm{2}}}}  .  \ottnt{b}    \:   (   \mathbf{inj}_2 \:  \ottmv{x_{{\mathrm{2}}}}   )     \rrbracket  \\
= &      \mathbf{S}_{  n  }   \text{app}    \circ   \mathbf{S}_{  n  }   (   \ottnt{h}  \times   \text{id}    )     \circ   p^{-1}_{     \llbracket  \ottnt{B}  \rrbracket  ^{  \llbracket  \ottnt{A_{{\mathrm{1}}}}  \rrbracket  }    \times     \llbracket  \ottnt{B}  \rrbracket  ^{  \llbracket  \ottnt{A_{{\mathrm{2}}}}  \rrbracket  }    ,     \llbracket  \ottnt{A_{{\mathrm{1}}}}  \rrbracket   +   \llbracket  \ottnt{A_{{\mathrm{2}}}}  \rrbracket    }    \circ     p^{-1}_{    \llbracket  \ottnt{B}  \rrbracket  ^{  \llbracket  \ottnt{A_{{\mathrm{1}}}}  \rrbracket  }   ,     \llbracket  \ottnt{B}  \rrbracket  ^{  \llbracket  \ottnt{A_{{\mathrm{2}}}}  \rrbracket  }   }   \times   \text{id}      \circ   \langle    \langle   \mathbf{S}_{  n  }     \llbracket  \ottnt{B}  \rrbracket  ^{  i_1  }     ,   \mathbf{S}_{  n  }     \llbracket  \ottnt{B}  \rrbracket  ^{  i_1  }     \rangle   \circ   \llbracket  \ottnt{b}  \rrbracket    ,   \llbracket  \ottnt{a}  \rrbracket   \rangle   \\
= &        \mathbf{S}_{  n  }   \text{app}    \circ   \mathbf{S}_{  n  }   (   \ottnt{h}  \times   \text{id}    )     \circ   p^{-1}    \circ     p^{-1}   \times   \text{id}      \circ   \langle   \mathbf{S}_{  n  }     \llbracket  \ottnt{B}  \rrbracket  ^{  i_1  }     ,   \mathbf{S}_{  n  }     \llbracket  \ottnt{B}  \rrbracket  ^{  i_1  }     \rangle    \times   \text{id}    \circ   \langle   \llbracket  \ottnt{b}  \rrbracket   ,   \llbracket  \ottnt{a}  \rrbracket   \rangle   \\
= &       \mathbf{S}_{  n  }   \text{app}    \circ   \mathbf{S}_{  n  }   (   \ottnt{h}  \times   \text{id}    )     \circ   p^{-1}    \circ   \mathbf{S}_{  n  }   \langle    \llbracket  \ottnt{B}  \rrbracket  ^{  i_1  }   ,    \llbracket  \ottnt{B}  \rrbracket  ^{  i_1  }   \rangle     \times   \text{id}    \circ   \langle   \llbracket  \ottnt{b}  \rrbracket   ,   \llbracket  \ottnt{a}  \rrbracket   \rangle   \\
= &      \mathbf{S}_{  n  }   \text{app}    \circ   \mathbf{S}_{  n  }   (   \ottnt{h}  \times   \text{id}    )     \circ   \mathbf{S}_{  n  }   (    \langle    \llbracket  \ottnt{B}  \rrbracket  ^{  i_1  }   ,    \llbracket  \ottnt{B}  \rrbracket  ^{  i_1  }   \rangle   \times   \text{id}    )     \circ   p^{-1}    \circ   \langle   \llbracket  \ottnt{b}  \rrbracket   ,   \llbracket  \ottnt{a}  \rrbracket   \rangle   \\
= &    \mathbf{S}_{  n  }   (     \text{app}   \circ   (   \ottnt{h}  \times   \text{id}    )    \circ   (    h^{-1}   \times   \text{id}    )    )    \circ   p^{-1}    \circ   \langle   \llbracket  \ottnt{b}  \rrbracket   ,   \llbracket  \ottnt{a}  \rrbracket   \rangle   \\
= &    \mathbf{S}_{  n  }   \text{app}    \circ   p^{-1}    \circ   \langle   \llbracket  \ottnt{b}  \rrbracket   ,   \llbracket  \ottnt{a}  \rrbracket   \rangle   =     \text{app}   \circ   \Lambda   (    \mathbf{S}_{  n  }   \text{app}    \circ   p^{-1}    )     \times   \text{id}    \circ   \langle   \llbracket  \ottnt{b}  \rrbracket   ,   \llbracket  \ottnt{a}  \rrbracket   \rangle   =  \llbracket    \ottnt{b}  \:  \ottnt{a}    \rrbracket . 
\end{align*}
\end{itemize}
\end{proof}


\begin{theorem}[Theorem \ref{lcCA}] \label{lcCAprf}
Let $ \Gamma  \vdash  \ottnt{b}  :^{ \ottsym{0} }   \mathbf{Bool}  $ and $\ottmv{v}$ be a value of type $ \mathbf{Bool} $. If $ \llbracket  \ottnt{b}  \rrbracket  =  \llbracket  \ottmv{v}  \rrbracket $ then $ \vdash  \ottnt{b}  \longmapsto^{\ast}  \ottmv{v} $. 
\end{theorem}

\begin{proof}
Let $ \Gamma  \vdash  \ottnt{b}  :^{ \ottsym{0} }   \mathbf{Bool}  $ and $\ottmv{v}$ be a boolean value such that $ \llbracket  \ottnt{b}  \rrbracket_{(\mathds{C},\mathbf{S})}  =  \llbracket  \ottmv{v}  \rrbracket_{(\mathds{C},\mathbf{S})} $.\\
We show that $ \vdash  \ottnt{b}  \longmapsto^{\ast}  \ottmv{v} $.

First note that \lc{} is strongly normalizing with respect to the reduction relation, $\longmapsto^{\ast}$. Further, \lc{} is also type sound with respect to this reduction relation.

Therefore, given $ \Gamma  \vdash  \ottnt{b}  :^{ \ottsym{0} }   \mathbf{Bool}  $, we know that there exists a value $\ottmv{v_{{\mathrm{0}}}}$ such that $ \Gamma  \vdash  \ottmv{v_{{\mathrm{0}}}}  :^{ \ottsym{0} }   \mathbf{Bool}  $ and $ \vdash  \ottnt{b}  \longmapsto^{\ast}  \ottmv{v_{{\mathrm{0}}}} $.

Next, by Theorem \ref{lcCatprf}, $ \llbracket  \ottnt{b}  \rrbracket_{(\mathds{C},\mathbf{S})}  =  \llbracket  \ottmv{v_{{\mathrm{0}}}}  \rrbracket_{(\mathds{C},\mathbf{S})} $.

Since $ \llbracket  \ottnt{b}  \rrbracket_{(\mathds{C},\mathbf{S})}  =  \llbracket  \ottmv{v}  \rrbracket_{(\mathds{C},\mathbf{S})} $ (given) and $ \llbracket  \ottnt{b}  \rrbracket_{(\mathds{C},\mathbf{S})}  =  \llbracket  \ottmv{v_{{\mathrm{0}}}}  \rrbracket_{(\mathds{C},\mathbf{S})} $ (above), therefore, $ \llbracket  \ottmv{v}  \rrbracket_{(\mathds{C},\mathbf{S})}  =  \llbracket  \ottmv{v_{{\mathrm{0}}}}  \rrbracket_{(\mathds{C},\mathbf{S})} $.

By injectivity, $\ottmv{v} = \ottmv{v_{{\mathrm{0}}}}$.

Thus, $ \vdash  \ottnt{b}  \longmapsto^{\ast}  \ottmv{v} $.
\end{proof}


\begin{theorem}[Theorem \ref{lcS0}] \label{lcS0prf}
$\llbracket \_ \rrbracket_{(\mathbf{Set},\mathbf{S}^{\ottsym{0}})}$ is a computationally adequate interpretation of \lc{}.
\end{theorem} 

\begin{proof}
For a bicartesian category $\Ct$, any strong monoidal functor from $\Ca(\mathcal{N})$ to $\EC{}$ provides a computationally adequate interpretation of \lc{}, given the interpretation for ground types is injective. Now, with respect to $\llbracket \_ \rrbracket_{(\mathbf{Set},\mathbf{S}^{\ottsym{0}})}$, the interpretation for ground types is injective. As such, to prove that $(\mathbf{Set}, \mathbf{S}^{  \ottsym{0}  } )$ is a computationally adequate interpretation of \lc{}, we just need to show that $ \mathbf{S}^{  \ottsym{0}  } $ is a strong monoidal functor from $\Ca(\mathcal{N})$ to $\EC{}$.

Recall the definition of $ \mathbf{S}^{  \ottsym{0}  } $:
\begin{align*}
\mathbf{S}^0(n) = \begin{cases}
                  \Id & \text{ if } n = 0 \\
                  \ast & \text{ otherwise}
                  \end{cases}
\end{align*}

By this definition, $ \mathbf{S}^{  \ottsym{0}  }({  \ottsym{0}  })  = \Id{}$. Further, $\eta = \epsilon = \I{}_{\Id{}}$. Now for any $n_{{\mathrm{1}}} , n_{{\mathrm{2}}} \in \mathbb{N}$, there are two cases to consider:
\begin{itemize}
\item $n_{{\mathrm{1}}} = n_{{\mathrm{2}}} = 0$. Then $ n_{{\mathrm{1}}}  +  n_{{\mathrm{2}}}  = 0$.\\
So, $ \mathbf{S}^{  \ottsym{0}  }({  n_{{\mathrm{1}}}  })  =  \mathbf{S}^{  \ottsym{0}  }({  n_{{\mathrm{2}}}  })  =  \mathbf{S}^{  \ottsym{0}  }({   n_{{\mathrm{1}}}  +  n_{{\mathrm{2}}}   })  = \Id{}$.\\
In this case, $\mu^{n_{{\mathrm{1}}},n_{{\mathrm{2}}}} = \delta^{n_{{\mathrm{1}}},n_{{\mathrm{2}}}} = \I{}_{\Id{}}$.
\item $n_{{\mathrm{1}}} \neq 0$ or $n_{{\mathrm{2}}} \neq 0$. Then, $ n_{{\mathrm{1}}}  +  n_{{\mathrm{2}}}  \neq 0$.\\
So, $ \mathbf{S}^{  \ottsym{0}  }({  n_{{\mathrm{1}}}  })  = \ast$ or $ \mathbf{S}^{  \ottsym{0}  }({  n_{{\mathrm{2}}}  })  = \ast$. Hence, $  \mathbf{S}^{  \ottsym{0}  }({  n_{{\mathrm{1}}}  })   \circ   \mathbf{S}^{  \ottsym{0}  }({  n_{{\mathrm{2}}}  })   = \ast$.\\
Further, $ \mathbf{S}^{  \ottsym{0}  }({   n_{{\mathrm{1}}}  +  n_{{\mathrm{2}}}   })  = \ast$.\\
In this case, $\mu^{n_{{\mathrm{1}}},n_{{\mathrm{2}}}} = \delta^{n_{{\mathrm{1}}},n_{{\mathrm{2}}}} = \I{}_{\ast}$.
\end{itemize}
Hence, $ \mathbf{S}^{  \ottsym{0}  } $ is a strong (in fact a strict) monoidal functor from $\Ca(\mathcal{N})$ to $\EC{}$.
\end{proof}


\begin{theorem}[Theorem \ref{lcNI}]
Let $  \emptyset   \vdash  \ottnt{f}  :^{ \ottsym{0} }    \bigcirc  \ottnt{A}   \to   \mathbf{Bool}   $ and $  \emptyset   \vdash  \ottnt{b_{{\mathrm{1}}}}  :^{ \ottsym{0} }   \bigcirc  \ottnt{A}  $ and $  \emptyset   \vdash  \ottnt{b_{{\mathrm{2}}}}  :^{ \ottsym{0} }   \bigcirc  \ottnt{A}  $. Then, $ \vdash   \ottnt{f}  \:  \ottnt{b_{{\mathrm{1}}}}   \longmapsto^{\ast}  \ottmv{v} $ if and only if $ \vdash   \ottnt{f}  \:  \ottnt{b_{{\mathrm{2}}}}   \longmapsto^{\ast}  \ottmv{v} $, where $\ottmv{v}$ is a value of type $ \mathbf{Bool} $.
\end{theorem}

\begin{proof}
First note that $(\mathbf{Set}, \mathbf{S}^{  \ottsym{0}  } )$ gives us a computationally adequate model of \lc{} (by Theorem \ref{lcS0prf}).\\
Next, $ \llbracket   \bigcirc  \ottnt{A}   \rrbracket_{(\mathds{C},\mathbf{S}^{  \ottsym{0}  })}  = ( \mathbf{S}^{  \ottsym{0}  }({   0'   }) ) \llbracket  \ottnt{A}  \rrbracket_{(\mathds{C},\mathbf{S}^{  \ottsym{0}  })}  = (\ast) \llbracket  \ottnt{A}  \rrbracket_{(\mathds{C},\mathbf{S}^{  \ottsym{0}  })}  = \top$.\\
Therefore, $ \llbracket    \ottnt{f}  \:  \ottnt{a_{{\mathrm{1}}}}    \rrbracket_{(\mathbf{Set},\mathbf{S}^{  \ottsym{0}  })}  =  \llbracket    \ottnt{f}  \:  \ottnt{a_{{\mathrm{2}}}}    \rrbracket_{(\mathbf{Set},\mathbf{S}^{  \ottsym{0}  })}  =   \text{app}   \circ   \langle    q   \circ   \llbracket  \ottnt{f}  \rrbracket_{(\mathbf{Set},\mathbf{S}^{  \ottsym{0}  })}    ,   \langle \rangle   \rangle  $.\\
Then, by Theorem \ref{lcCAprf}, $ \vdash   \ottnt{f}  \:  \ottnt{a_{{\mathrm{1}}}}   \longmapsto^{\ast}  \ottmv{v} $ if and only if $ \vdash   \ottnt{f}  \:  \ottnt{a_{{\mathrm{2}}}}   \longmapsto^{\ast}  \ottmv{v} $.
\end{proof}


\section{Proofs of propositions/theorems stated in Section \ref{secgmcce}}

\begin{prop}[Proposition \ref{propiso}]
Let $ \mathcal{M} $ be any preordered monoid. Then, in \Ge{}($ \mathcal{M} $),
\begin{itemize}
\item The types $ S_{ m } \:   \ottkw{Unit}  $ and $ \ottkw{Unit} $ are isomorphic.
\item The types $ S_{ m } \:   (   \ottnt{A_{{\mathrm{1}}}}  \times  \ottnt{A_{{\mathrm{2}}}}   )  $ and $  S_{ m } \:  \ottnt{A_{{\mathrm{1}}}}   \times   S_{ m } \:  \ottnt{A_{{\mathrm{2}}}}  $, for all types $\ottnt{A_{{\mathrm{1}}}}$ and $\ottnt{A_{{\mathrm{2}}}}$, are isomorphic.
\item The types $ S_{ m } \:   (   \ottnt{A}  \to  \ottnt{B}   )  $ and $  S_{ m } \:  \ottnt{A}   \to   S_{ m } \:  \ottnt{B}  $, for all types $\ottnt{A}$ and $\ottnt{B}$, are isomorphic. 
\end{itemize}
\end{prop}

\begin{proof}
For any two types, $\ottnt{A}$ and $\ottnt{B}$, to show that $\ottnt{A} \cong \ottnt{B}$, we need to provide terms  $  \emptyset   \vdash  \ottnt{f}  :^{ m' }   \ottnt{A}  \to  \ottnt{B}  $ and $  \emptyset   \vdash  \ottnt{g}  :^{ m' }   \ottnt{B}  \to  \ottnt{A}  $ such that $  \ottmv{x}  :^{ m' }  \ottnt{A}   \vdash   \ottnt{g}  \:   (   \ottnt{f}  \:  \ottmv{x}   )    \equiv  \ottmv{x}  :^{ m' }  \ottnt{A} $ and $  \ottmv{y}  :^{ m' }  \ottnt{B}   \vdash   \ottnt{f}  \:   (   \ottnt{g}  \:  \ottmv{y}   )    \equiv  \ottmv{y}  :^{ m' }  \ottnt{B} $, for any $m' \in M$. Note that the judgement $ \Gamma  \vdash  \ottnt{a_{{\mathrm{1}}}}  \equiv  \ottnt{a_{{\mathrm{2}}}}  :^{ m }  \ottnt{A} $ is a shorthand for the judgements $ \Gamma  \vdash  \ottnt{a_{{\mathrm{1}}}}  :^{ m }  \ottnt{A} $ and $ \Gamma  \vdash  \ottnt{a_{{\mathrm{2}}}}  :^{ m }  \ottnt{A} $ and $ \ottnt{a_{{\mathrm{1}}}}  \equiv  \ottnt{a_{{\mathrm{2}}}} $.
\begin{itemize}
\item Need to show: $ S_{ m } \:   \ottkw{Unit}   \cong  \ottkw{Unit} $.\\
We have: 
\[ \infer[]{  \emptyset   \vdash   \lambda  \ottmv{x}  .   \ottkw{unit}    :^{ m' }    S_{ m } \:   \ottkw{Unit}    \to   \ottkw{Unit}   }{  \ottmv{x}  :^{ m' }   S_{ m } \:   \ottkw{Unit}     \vdash   \ottkw{unit}   :^{ m' }   \ottkw{Unit}  } \hspace*{10pt}
 \infer[]{  \emptyset   \vdash   \lambda  \ottmv{y}  .   \mathbf{split}^{ m }   \ottkw{unit}     :^{ m' }    \ottkw{Unit}   \to   S_{ m } \:   \ottkw{Unit}    }
            {\infer[]{  \ottmv{y}  :^{ m' }   \ottkw{Unit}    \vdash   \mathbf{split}^{ m }   \ottkw{unit}    :^{ m' }   S_{ m } \:   \ottkw{Unit}   }{  \ottmv{y}  :^{ m' }   \ottkw{Unit}    \vdash   \ottkw{unit}   :^{  m'  \cdot  m  }   \ottkw{Unit}  }}\]
Say, $\ottnt{f_{{\mathrm{1}}}} \triangleq  \lambda  \ottmv{x}  .   \ottkw{unit}  $ and $\ottnt{g_{{\mathrm{1}}}} \triangleq  \lambda  \ottmv{y}  .   \mathbf{split}^{ m }   \ottkw{unit}   $.\\
Next, \[ \infer[]{  \ottmv{x}  :^{ m' }   S_{ m } \:   \ottkw{Unit}     \vdash   \ottkw{unit}   \equiv   \mathbf{merge}^{ m }  \ottmv{x}   :^{  m'  \cdot  m  }   \ottkw{Unit}  }
                  {\infer[]{  \ottmv{x}  :^{ m' }   S_{ m } \:   \ottkw{Unit}     \vdash   \mathbf{merge}^{ m }  \ottmv{x}   :^{  m'  \cdot  m  }   \ottkw{Unit}  }{  \ottmv{x}  :^{ m' }   S_{ m } \:   \ottkw{Unit}     \vdash  \ottmv{x}  :^{ m' }   S_{ m } \:   \ottkw{Unit}   }} \]
So, given $ \ottmv{x}  :^{ m' }   S_{ m } \:   \ottkw{Unit}   $, we have, $ \ottnt{g_{{\mathrm{1}}}}  \:   (   \ottnt{f_{{\mathrm{1}}}}  \:  \ottmv{x}   )   \equiv   \mathbf{split}^{ m }   \ottkw{unit}    \equiv   \mathbf{split}^{ m }   (   \mathbf{merge}^{ m }  \ottmv{x}   )    \equiv x$.\\
And, given $ \ottmv{y}  :^{ m' }   \ottkw{Unit}  $, we have, $ \ottnt{f_{{\mathrm{1}}}}  \:   (   \ottnt{g_{{\mathrm{1}}}}  \:  \ottmv{y}   )   \equiv   \ottkw{unit}   \equiv  \ottmv{y} $.

\item Need to show: $ S_{ m } \:   (   \ottnt{A_{{\mathrm{1}}}}  \times  \ottnt{A_{{\mathrm{2}}}}   )   \cong   S_{ m } \:  \ottnt{A_{{\mathrm{1}}}}   \times   S_{ m } \:  \ottnt{A_{{\mathrm{2}}}}  $.\\
We have: \[ \infer[]{  \emptyset   \vdash   \lambda  \ottmv{x}  .   (   \mathbf{split}^{ m }   (   \mathbf{proj}_1 \:   (   \mathbf{merge}^{ m }  \ottmv{x}   )    )    ,   \mathbf{split}^{ m }   (   \mathbf{proj}_2 \:   (   \mathbf{merge}^{ m }  \ottmv{x}   )    )    )    :^{ m' }     S_{ m } \:   (   \ottnt{A_{{\mathrm{1}}}}  \times  \ottnt{A_{{\mathrm{2}}}}   )    \to   S_{ m } \:  \ottnt{A_{{\mathrm{1}}}}    \times   S_{ m } \:  \ottnt{A_{{\mathrm{2}}}}    }
        {\infer[]{  \ottmv{x}  :^{ m' }   S_{ m } \:   (   \ottnt{A_{{\mathrm{1}}}}  \times  \ottnt{A_{{\mathrm{2}}}}   )     \vdash   (   \mathbf{split}^{ m }   (   \mathbf{proj}_1 \:   (   \mathbf{merge}^{ m }  \ottmv{x}   )    )    ,   \mathbf{split}^{ m }   (   \mathbf{proj}_2 \:   (   \mathbf{merge}^{ m }  \ottmv{x}   )    )    )   :^{ m' }    S_{ m } \:  \ottnt{A_{{\mathrm{1}}}}   \times   S_{ m } \:  \ottnt{A_{{\mathrm{2}}}}   }
        {\infer[]{  \ottmv{x}  :^{ m' }   S_{ m } \:   (   \ottnt{A_{{\mathrm{1}}}}  \times  \ottnt{A_{{\mathrm{2}}}}   )     \vdash   \mathbf{split}^{ m }   (   \mathbf{proj}_1 \:   (   \mathbf{merge}^{ m }  \ottmv{x}   )    )    :^{ m' }   S_{ m } \:  \ottnt{A_{{\mathrm{1}}}}  }
        {\infer[]{  \ottmv{x}  :^{ m' }   S_{ m } \:   (   \ottnt{A_{{\mathrm{1}}}}  \times  \ottnt{A_{{\mathrm{2}}}}   )     \vdash   \mathbf{proj}_1 \:   (   \mathbf{merge}^{ m }  \ottmv{x}   )    :^{  m'  \cdot  m  }  \ottnt{A_{{\mathrm{1}}}} }
        {\infer[]{  \ottmv{x}  :^{ m' }   S_{ m } \:   (   \ottnt{A_{{\mathrm{1}}}}  \times  \ottnt{A_{{\mathrm{2}}}}   )     \vdash   \mathbf{merge}^{ m }  \ottmv{x}   :^{  m'  \cdot  m  }   \ottnt{A_{{\mathrm{1}}}}  \times  \ottnt{A_{{\mathrm{2}}}}  }
        {  \ottmv{x}  :^{ m' }   S_{ m } \:   (   \ottnt{A_{{\mathrm{1}}}}  \times  \ottnt{A_{{\mathrm{2}}}}   )     \vdash  \ottmv{x}  :^{ m' }   S_{ m } \:   (   \ottnt{A_{{\mathrm{1}}}}  \times  \ottnt{A_{{\mathrm{2}}}}   )   }}}
        } } \] 
and \[ \infer[]{  \emptyset   \vdash   \lambda  \ottmv{y}  .   \mathbf{split}^{ m }   (   \mathbf{merge}^{ m }   (   \mathbf{proj}_1 \:  \ottmv{y}   )    ,   \mathbf{merge}^{ m }   (   \mathbf{proj}_2 \:  \ottmv{y}   )    )     :^{ m' }     S_{ m } \:  \ottnt{A_{{\mathrm{1}}}}   \times   S_{ m } \:  \ottnt{A_{{\mathrm{2}}}}    \to   S_{ m } \:   (   \ottnt{A_{{\mathrm{1}}}}  \times  \ottnt{A_{{\mathrm{2}}}}   )    }
       {\infer[]{  \ottmv{y}  :^{ m' }    S_{ m } \:  \ottnt{A_{{\mathrm{1}}}}   \times   S_{ m } \:  \ottnt{A_{{\mathrm{2}}}}     \vdash   \mathbf{split}^{ m }   (   \mathbf{merge}^{ m }   (   \mathbf{proj}_1 \:  \ottmv{y}   )    ,   \mathbf{merge}^{ m }   (   \mathbf{proj}_2 \:  \ottmv{y}   )    )    :^{ m' }   S_{ m } \:   (   \ottnt{A_{{\mathrm{1}}}}  \times  \ottnt{A_{{\mathrm{2}}}}   )   }
       {\infer[]{  \ottmv{y}  :^{ m' }    S_{ m } \:  \ottnt{A_{{\mathrm{1}}}}   \times   S_{ m } \:  \ottnt{A_{{\mathrm{2}}}}     \vdash   (   \mathbf{merge}^{ m }   (   \mathbf{proj}_1 \:  \ottmv{y}   )    ,   \mathbf{merge}^{ m }   (   \mathbf{proj}_2 \:  \ottmv{y}   )    )   :^{  m'  \cdot  m  }   \ottnt{A_{{\mathrm{1}}}}  \times  \ottnt{A_{{\mathrm{2}}}}  }
       {\infer[]{  \ottmv{y}  :^{ m' }    S_{ m } \:  \ottnt{A_{{\mathrm{1}}}}   \times   S_{ m } \:  \ottnt{A_{{\mathrm{2}}}}     \vdash   \mathbf{merge}^{ m }   (   \mathbf{proj}_1 \:  \ottmv{y}   )    :^{  m'  \cdot  m  }  \ottnt{A_{{\mathrm{1}}}} }
       {\infer[]{  \ottmv{y}  :^{ m' }    S_{ m } \:  \ottnt{A_{{\mathrm{1}}}}   \times   S_{ m } \:  \ottnt{A_{{\mathrm{2}}}}     \vdash   \mathbf{proj}_1 \:  \ottmv{y}   :^{ m' }   S_{ m } \:  \ottnt{A_{{\mathrm{1}}}}  }
       {  \ottmv{y}  :^{ m' }    S_{ m } \:  \ottnt{A_{{\mathrm{1}}}}   \times   S_{ m } \:  \ottnt{A_{{\mathrm{2}}}}     \vdash  \ottmv{y}  :^{ m' }    S_{ m } \:  \ottnt{A_{{\mathrm{1}}}}   \times   S_{ m } \:  \ottnt{A_{{\mathrm{2}}}}   }}}}} \]
Say, \begin{align*}
\ottnt{f_{{\mathrm{2}}}} & \triangleq  \lambda  \ottmv{x}  .   (   \mathbf{split}^{ m }   (   \mathbf{proj}_1 \:   (   \mathbf{merge}^{ m }  \ottmv{x}   )    )    ,   \mathbf{split}^{ m }   (   \mathbf{proj}_2 \:   (   \mathbf{merge}^{ m }  \ottmv{x}   )    )    )   \\      
\ottnt{g_{{\mathrm{2}}}} & \triangleq  \lambda  \ottmv{y}  .   \mathbf{split}^{ m }   (   \mathbf{merge}^{ m }   (   \mathbf{proj}_1 \:  \ottmv{y}   )    ,   \mathbf{merge}^{ m }   (   \mathbf{proj}_2 \:  \ottmv{y}   )    )   
\end{align*}
Next, given $ \ottmv{x}  :^{ m' }   S_{ m } \:   (   \ottnt{A_{{\mathrm{1}}}}  \times  \ottnt{A_{{\mathrm{2}}}}   )   $, we have,
\begin{align*}
&  \ottnt{g_{{\mathrm{2}}}}  \:   (   \ottnt{f_{{\mathrm{2}}}}  \:  \ottmv{x}   )   \\
\equiv &  (   \lambda  \ottmv{y}  .   \mathbf{split}^{ m }   (   \mathbf{merge}^{ m }   (   \mathbf{proj}_1 \:  \ottmv{y}   )    ,   \mathbf{merge}^{ m }   (   \mathbf{proj}_2 \:  \ottmv{y}   )    )     ) \\ & \hspace*{30pt}  (   \mathbf{split}^{ m }   (   \mathbf{proj}_1 \:   (   \mathbf{merge}^{ m }  \ottmv{x}   )    )    ,   \mathbf{split}^{ m }   (   \mathbf{proj}_2 \:   (   \mathbf{merge}^{ m }  \ottmv{x}   )    )    )  \\
\equiv &  \mathbf{split}^{ m }   (   \mathbf{merge}^{ m }   (   \mathbf{split}^{ m }   (   \mathbf{proj}_1 \:   (   \mathbf{merge}^{ m }  \ottmv{x}   )    )    )    ,   \mathbf{merge}^{ m }   (   \mathbf{split}^{ m }   (   \mathbf{proj}_2 \:   (   \mathbf{merge}^{ m }  \ottmv{x}   )    )    )    )   \\
\equiv &  \mathbf{split}^{ m }   (   \mathbf{proj}_1 \:   (   \mathbf{merge}^{ m }  \ottmv{x}   )    ,   \mathbf{proj}_2 \:   (   \mathbf{merge}^{ m }  \ottmv{x}   )    )   \equiv  \mathbf{split}^{ m }   (   \mathbf{merge}^{ m }  \ottmv{x}   )   \equiv x
\end{align*}             
And, given $ \ottmv{y}  :^{ m' }    S_{ m } \:  \ottnt{A_{{\mathrm{1}}}}   \times   S_{ m } \:  \ottnt{A_{{\mathrm{2}}}}   $, we have,
\begin{align*}
&  \ottnt{f_{{\mathrm{2}}}}  \:   (   \ottnt{g_{{\mathrm{2}}}}  \:  \ottmv{y}   )   \\
\equiv &  (   \lambda  \ottmv{x}  .   (   \mathbf{split}^{ m }   (   \mathbf{proj}_1 \:   (   \mathbf{merge}^{ m }  \ottmv{x}   )    )    ,   \mathbf{split}^{ m }   (   \mathbf{proj}_2 \:   (   \mathbf{merge}^{ m }  \ottmv{x}   )    )    )    )  \\ & \hspace*{30pt}  (   \mathbf{split}^{ m }   (   \mathbf{merge}^{ m }   (   \mathbf{proj}_1 \:  \ottmv{y}   )    ,   \mathbf{merge}^{ m }   (   \mathbf{proj}_2 \:  \ottmv{y}   )    )    )  \\
\equiv & ( \mathbf{split}^{ m }   (   \mathbf{proj}_1 \:   (   \mathbf{merge}^{ m }   (   \mathbf{split}^{ m }   (   \mathbf{merge}^{ m }   (   \mathbf{proj}_1 \:  \ottmv{y}   )    ,   \mathbf{merge}^{ m }   (   \mathbf{proj}_2 \:  \ottmv{y}   )    )    )    )    )  , \\
& \hspace*{30pt}  \mathbf{split}^{ m }   (   \mathbf{proj}_2 \:   (   \mathbf{merge}^{ m }   (   \mathbf{split}^{ m }   (   \mathbf{merge}^{ m }   (   \mathbf{proj}_1 \:  \ottmv{y}   )    ,   \mathbf{merge}^{ m }   (   \mathbf{proj}_2 \:  \ottmv{y}   )    )    )    )    )  ) \\
\equiv & ( \mathbf{split}^{ m }   (   \mathbf{proj}_1 \:   (   \mathbf{merge}^{ m }   (   \mathbf{proj}_1 \:  \ottmv{y}   )    ,   \mathbf{merge}^{ m }   (   \mathbf{proj}_2 \:  \ottmv{y}   )    )    )  , \\ & \hspace*{30pt}  \mathbf{split}^{ m }   (   \mathbf{proj}_2 \:   (   \mathbf{merge}^{ m }   (   \mathbf{proj}_1 \:  \ottmv{y}   )    ,   \mathbf{merge}^{ m }   (   \mathbf{proj}_2 \:  \ottmv{y}   )    )    )  ) \\
\equiv &  (   \mathbf{split}^{ m }   (   \mathbf{merge}^{ m }   (   \mathbf{proj}_1 \:  \ottmv{y}   )    )    ,   \mathbf{split}^{ m }   (   \mathbf{merge}^{ m }   (   \mathbf{proj}_2 \:  \ottmv{y}   )    )    )  \equiv  (   \mathbf{proj}_1 \:  \ottmv{y}   ,   \mathbf{proj}_2 \:  \ottmv{y}   )  \equiv y
\end{align*}

\item Need to show: $ S_{ m } \:   (   \ottnt{A}  \to  \ottnt{B}   )   \cong   S_{ m } \:  \ottnt{A}   \to   S_{ m } \:  \ottnt{B}  $.\\
We have: \[ \mkern-25mu \infer[]{  \emptyset   \vdash   \lambda  \ottmv{x}  .   \lambda  \ottmv{z}  .   \mathbf{split}^{ m }   (    (   \mathbf{merge}^{ m }  \ottmv{x}   )   \:   (   \mathbf{merge}^{ m }  \ottmv{z}   )    )      :^{ m' }    S_{ m } \:   (   \ottnt{A}  \to  \ottnt{B}   )    \to   (    S_{ m } \:  \ottnt{A}   \to   S_{ m } \:  \ottnt{B}    )   }
             {\infer[]{  \ottmv{x}  :^{ m' }   S_{ m } \:   (   \ottnt{A}  \to  \ottnt{B}   )     \vdash   \lambda  \ottmv{z}  .   \mathbf{split}^{ m }   (    (   \mathbf{merge}^{ m }  \ottmv{x}   )   \:   (   \mathbf{merge}^{ m }  \ottmv{z}   )    )     :^{ m' }    S_{ m } \:  \ottnt{A}   \to   S_{ m } \:  \ottnt{B}   }
             {\infer[]{   \ottmv{x}  :^{ m' }   S_{ m } \:   (   \ottnt{A}  \to  \ottnt{B}   )     ,   \ottmv{z}  :^{ m' }   S_{ m } \:  \ottnt{A}     \vdash   \mathbf{split}^{ m }   (    (   \mathbf{merge}^{ m }  \ottmv{x}   )   \:   (   \mathbf{merge}^{ m }  \ottmv{z}   )    )    :^{ m' }   S_{ m } \:  \ottnt{B}  }
             {\infer[]{   \ottmv{x}  :^{ m' }   S_{ m } \:   (   \ottnt{A}  \to  \ottnt{B}   )     ,   \ottmv{z}  :^{ m' }   S_{ m } \:  \ottnt{A}     \vdash    (   \mathbf{merge}^{ m }  \ottmv{x}   )   \:   (   \mathbf{merge}^{ m }  \ottmv{z}   )    :^{  m'  \cdot  m  }  \ottnt{B} }
             {\infer[]{   \ottmv{x}  :^{ m' }   S_{ m } \:   (   \ottnt{A}  \to  \ottnt{B}   )     ,   \ottmv{z}  :^{ m' }   S_{ m } \:  \ottnt{A}     \vdash   \mathbf{merge}^{ m }  \ottmv{x}   :^{  m'  \cdot  m  }   \ottnt{A}  \to  \ottnt{B}  }
             {   \ottmv{x}  :^{ m' }   S_{ m } \:   (   \ottnt{A}  \to  \ottnt{B}   )     ,   \ottmv{z}  :^{ m' }   S_{ m } \:  \ottnt{A}     \vdash  \ottmv{x}  :^{ m' }   S_{ m } \:   (   \ottnt{A}  \to  \ottnt{B}   )   }
             &
             \infer[]{   \ottmv{x}  :^{ m' }   S_{ m } \:   (   \ottnt{A}  \to  \ottnt{B}   )     ,   \ottmv{z}  :^{ m' }   S_{ m } \:  \ottnt{A}     \vdash   \mathbf{merge}^{ m }  \ottmv{z}   :^{  m'  \cdot  m  }  \ottnt{A} }
             {   \ottmv{x}  :^{ m' }   S_{ m } \:   (   \ottnt{A}  \to  \ottnt{B}   )     ,   \ottmv{z}  :^{ m' }   S_{ m } \:  \ottnt{A}     \vdash  \ottmv{z}  :^{ m' }   S_{ m } \:  \ottnt{A}  }
             }}}} \] 
and \[ \mkern-25mu \infer[]{  \emptyset   \vdash   \lambda  \ottmv{y}  .   \mathbf{split}^{ m }   (   \lambda  \ottmv{w}  .   \mathbf{merge}^{ m }   (   \ottmv{y}  \:   (   \mathbf{split}^{ m }  \ottmv{w}   )    )     )     :^{ m' }    (    S_{ m } \:  \ottnt{A}   \to   S_{ m } \:  \ottnt{B}    )   \to   S_{ m } \:   (   \ottnt{A}  \to  \ottnt{B}   )    }
          {\infer[]{  \ottmv{y}  :^{ m' }    S_{ m } \:  \ottnt{A}   \to   S_{ m } \:  \ottnt{B}     \vdash   \mathbf{split}^{ m }   (   \lambda  \ottmv{w}  .   \mathbf{merge}^{ m }   (   \ottmv{y}  \:   (   \mathbf{split}^{ m }  \ottmv{w}   )    )     )    :^{ m' }   S_{ m } \:   (   \ottnt{A}  \to  \ottnt{B}   )   }
          {\infer[]{  \ottmv{y}  :^{ m' }    S_{ m } \:  \ottnt{A}   \to   S_{ m } \:  \ottnt{B}     \vdash   \lambda  \ottmv{w}  .   \mathbf{merge}^{ m }   (   \ottmv{y}  \:   (   \mathbf{split}^{ m }  \ottmv{w}   )    )     :^{  m'  \cdot  m  }   \ottnt{A}  \to  \ottnt{B}  }
          {\infer[]{   \ottmv{y}  :^{ m' }    S_{ m } \:  \ottnt{A}   \to   S_{ m } \:  \ottnt{B}     ,   \ottmv{w}  :^{  m'  \cdot  m  }  \ottnt{A}    \vdash   \mathbf{merge}^{ m }   (   \ottmv{y}  \:   (   \mathbf{split}^{ m }  \ottmv{w}   )    )    :^{  m'  \cdot  m  }  \ottnt{B} }
          {\infer[]{   \ottmv{y}  :^{ m' }    S_{ m } \:  \ottnt{A}   \to   S_{ m } \:  \ottnt{B}     ,   \ottmv{w}  :^{  m'  \cdot  m  }  \ottnt{A}    \vdash   \ottmv{y}  \:   (   \mathbf{split}^{ m }  \ottmv{w}   )    :^{ m' }   S_{ m } \:  \ottnt{B}  }
          { {   \ottmv{y}  :^{ m' }    S_{ m } \:  \ottnt{A}   \to   S_{ m } \:  \ottnt{B}     ,   \ottmv{w}  :^{  m'  \cdot  m  }  \ottnt{A}    \vdash  \ottmv{y}  :^{ m' }    S_{ m } \:  \ottnt{A}   \to   S_{ m } \:  \ottnt{B}   }
           &
          {\infer[]{   \ottmv{y}  :^{ m' }    S_{ m } \:  \ottnt{A}   \to   S_{ m } \:  \ottnt{B}     ,   \ottmv{w}  :^{  m'  \cdot  m  }  \ottnt{A}    \vdash   \mathbf{split}^{ m }  \ottmv{w}   :^{ m' }   S_{ m } \:  \ottnt{A}  }
          {   \ottmv{y}  :^{ m' }    S_{ m } \:  \ottnt{A}   \to   S_{ m } \:  \ottnt{B}     ,   \ottmv{w}  :^{  m'  \cdot  m  }  \ottnt{A}    \vdash  \ottmv{w}  :^{  m'  \cdot  m  }  \ottnt{A} }}}}}}} \]
Say, \begin{align*}
\ottnt{f_{{\mathrm{3}}}} & \triangleq  \lambda  \ottmv{x}  .   \lambda  \ottmv{z}  .   \mathbf{split}^{ m }   (    (   \mathbf{merge}^{ m }  \ottmv{x}   )   \:   (   \mathbf{merge}^{ m }  \ottmv{z}   )    )     \\
\ottnt{g_{{\mathrm{3}}}} & \triangleq  \lambda  \ottmv{y}  .   \mathbf{split}^{ m }   (   \lambda  \ottmv{w}  .   \mathbf{merge}^{ m }   (   \ottmv{y}  \:   (   \mathbf{split}^{ m }  \ottmv{w}   )    )     )   
\end{align*}
Next, given $ \ottmv{x}  :^{ m' }   S_{ m } \:   (   \ottnt{A}  \to  \ottnt{B}   )   $, we have,
\begin{align*}
&  \ottnt{g_{{\mathrm{3}}}}  \:   (   \ottnt{f_{{\mathrm{3}}}}  \:  \ottmv{x}   )   \\
\equiv &   (   \lambda  \ottmv{y}  .   \mathbf{split}^{ m }   (   \lambda  \ottmv{w}  .   \mathbf{merge}^{ m }   (   \ottmv{y}  \:   (   \mathbf{split}^{ m }  \ottmv{w}   )    )     )     )   \:   (   \lambda  \ottmv{z}  .   \mathbf{split}^{ m }   (    (   \mathbf{merge}^{ m }  \ottmv{x}   )   \:   (   \mathbf{merge}^{ m }  \ottmv{z}   )    )     )   \\
\equiv &  \mathbf{split}^{ m }   (   \lambda  \ottmv{w}  .   \mathbf{merge}^{ m }   (    (   \lambda  \ottmv{z}  .   \mathbf{split}^{ m }   (    (   \mathbf{merge}^{ m }  \ottmv{x}   )   \:   (   \mathbf{merge}^{ m }  \ottmv{z}   )    )     )   \:   (   \mathbf{split}^{ m }  \ottmv{w}   )    )     )  \\
\equiv &  \mathbf{split}^{ m }   (   \lambda  \ottmv{w}  .   \mathbf{merge}^{ m }   (   \mathbf{split}^{ m }   (    (   \mathbf{merge}^{ m }  \ottmv{x}   )   \:   (   \mathbf{merge}^{ m }   (   \mathbf{split}^{ m }  \ottmv{w}   )    )    )    )     )  \\
\equiv &  \mathbf{split}^{ m }   (   \lambda  \ottmv{w}  .   \mathbf{merge}^{ m }   (   \mathbf{split}^{ m }   (    (   \mathbf{merge}^{ m }  \ottmv{x}   )   \:  \ottmv{w}   )    )     )  \\
\equiv &  \mathbf{split}^{ m }   (    \lambda  \ottmv{w}  .   (   \mathbf{merge}^{ m }  \ottmv{x}   )    \:  \ottmv{w}   )   \equiv  \mathbf{split}^{ m }   (   \mathbf{merge}^{ m }  \ottmv{x}   )   \equiv x
\end{align*}
And, given $ \ottmv{y}  :^{ m' }    S_{ m } \:  \ottnt{A}   \to   S_{ m } \:  \ottnt{B}   $, we have,
\begin{align*}
&  \ottnt{f_{{\mathrm{3}}}}  \:   (   \ottnt{g_{{\mathrm{3}}}}  \:  \ottmv{y}   )   \\
\equiv &   (   \lambda  \ottmv{x}  .   \lambda  \ottmv{z}  .   \mathbf{split}^{ m }   (    (   \mathbf{merge}^{ m }  \ottmv{x}   )   \:   (   \mathbf{merge}^{ m }  \ottmv{z}   )    )      )   \:   (   \mathbf{split}^{ m }   (   \lambda  \ottmv{w}  .   \mathbf{merge}^{ m }   (   \ottmv{y}  \:   (   \mathbf{split}^{ m }  \ottmv{w}   )    )     )    )   \\
\equiv &  \lambda  \ottmv{z}  .   \mathbf{split}^{ m }   (    (   \mathbf{merge}^{ m }   (   \mathbf{split}^{ m }   (   \lambda  \ottmv{w}  .   \mathbf{merge}^{ m }   (   \ottmv{y}  \:   (   \mathbf{split}^{ m }  \ottmv{w}   )    )     )    )    )   \:   (   \mathbf{merge}^{ m }  \ottmv{z}   )    )    \\
\equiv &  \lambda  \ottmv{z}  .   \mathbf{split}^{ m }   (    (   \lambda  \ottmv{w}  .   \mathbf{merge}^{ m }   (   \ottmv{y}  \:   (   \mathbf{split}^{ m }  \ottmv{w}   )    )     )   \:   (   \mathbf{merge}^{ m }  \ottmv{z}   )    )    \\
\equiv &  \lambda  \ottmv{z}  .   \mathbf{split}^{ m }   (   \mathbf{merge}^{ m }   (   \ottmv{y}  \:   (   \mathbf{split}^{ m }   (   \mathbf{merge}^{ m }  \ottmv{z}   )    )    )    )    \\
\equiv &   \lambda  \ottmv{z}  .  \ottmv{y}   \:   (   \mathbf{split}^{ m }   (   \mathbf{merge}^{ m }  \ottmv{z}   )    )   \equiv   \lambda  \ottmv{z}  .  \ottmv{y}   \:  \ottmv{z}  \equiv y
\end{align*}
\end{itemize}
\end{proof}


\begin{theorem}[Theorem \ref{gmcceSound}]
If $ \Gamma  \vdash  \ottnt{a}  :^{ m }  \ottnt{A} $ in \Ge{}, then $ \llbracket  \ottnt{a}  \rrbracket  \in \text{Hom}_{\Ct} ( \llbracket  \Gamma  \rrbracket , \mathbf{S}_{m}  \llbracket  \ottnt{A}  \rrbracket )$. Further, if $ \Gamma  \vdash  \ottnt{a_{{\mathrm{1}}}}  :^{ m }  \ottnt{A} $ and $ \Gamma  \vdash  \ottnt{a_{{\mathrm{2}}}}  :^{ m }  \ottnt{A} $ such that $ \ottnt{a_{{\mathrm{1}}}}  \equiv  \ottnt{a_{{\mathrm{2}}}} $ in \Ge{}, then $ \llbracket  \ottnt{a_{{\mathrm{1}}}}  \rrbracket  =  \llbracket  \ottnt{a_{{\mathrm{2}}}}  \rrbracket  \in \text{Hom}_{\Ct} ( \llbracket  \Gamma  \rrbracket , \mathbf{S}_{m}  \llbracket  \ottnt{A}  \rrbracket )$. 
\end{theorem}

\begin{proof}
The first part follows by induction on the typing derivation. Most of the cases are similar to those in Theorem \ref{lcCatprf}. We present the differing ones below.
\begin{itemize}
\item \Rref{E-Pair}. Have: $ \Gamma  \vdash   (  \ottnt{a_{{\mathrm{1}}}}  ,  \ottnt{a_{{\mathrm{2}}}}  )   :^{ m }   \ottnt{A_{{\mathrm{1}}}}  \times  \ottnt{A_{{\mathrm{2}}}}  $ where $ \Gamma  \vdash  \ottnt{a_{{\mathrm{1}}}}  :^{ m }  \ottnt{A_{{\mathrm{1}}}} $ and $ \Gamma  \vdash  \ottnt{a_{{\mathrm{2}}}}  :^{ m }  \ottnt{A_{{\mathrm{2}}}} $.\\
Need to show: $ \llbracket   (  \ottnt{a_{{\mathrm{1}}}}  ,  \ottnt{a_{{\mathrm{2}}}}  )   \rrbracket  \in \text{Hom}_{\Ct} ( \llbracket  \Gamma  \rrbracket  ,  \mathbf{S}_{ m }   (    \llbracket  \ottnt{A_{{\mathrm{1}}}}  \rrbracket   \times   \llbracket  \ottnt{A_{{\mathrm{2}}}}  \rrbracket    )  )$.\\
By IH, $ \llbracket  \ottnt{a_{{\mathrm{1}}}}  \rrbracket  \in \text{Hom}_{\Ct} ( \llbracket  \Gamma  \rrbracket  ,  \mathbf{S}_{ m }   \llbracket  \ottnt{A_{{\mathrm{1}}}}  \rrbracket  )$ and $ \llbracket  \ottnt{a_{{\mathrm{2}}}}  \rrbracket  \in \text{Hom}_{\Ct} ( \llbracket  \Gamma  \rrbracket ,  \mathbf{S}_{ m }   \llbracket  \ottnt{A_{{\mathrm{2}}}}  \rrbracket  )$.\\
Now, $ \llbracket   (  \ottnt{a_{{\mathrm{1}}}}  ,  \ottnt{a_{{\mathrm{2}}}}  )   \rrbracket  =  \llbracket  \Gamma  \rrbracket  \xrightarrow{ \langle   \llbracket  \ottnt{a_{{\mathrm{1}}}}  \rrbracket   ,   \llbracket  \ottnt{a_{{\mathrm{2}}}}  \rrbracket   \rangle }   \mathbf{S}_{ m }   \llbracket  \ottnt{A_{{\mathrm{1}}}}  \rrbracket    \times   \mathbf{S}_{ m }   \llbracket  \ottnt{A_{{\mathrm{2}}}}  \rrbracket    \xrightarrow{ p^{-1}_{  \llbracket  \ottnt{A_{{\mathrm{1}}}}  \rrbracket  ,   \llbracket  \ottnt{A_{{\mathrm{2}}}}  \rrbracket  } }  \mathbf{S}_{ m }   (    \llbracket  \ottnt{A_{{\mathrm{1}}}}  \rrbracket   \times   \llbracket  \ottnt{A_{{\mathrm{2}}}}  \rrbracket    )  $.
\item \Rref{E-Proj}. Have: $ \Gamma  \vdash   \mathbf{proj}_i \:  \ottnt{a}   :^{ m }   A_i  $ where $ \Gamma  \vdash  \ottnt{a}  :^{ m }   \ottnt{A_{{\mathrm{1}}}}  \times  \ottnt{A_{{\mathrm{2}}}}  $.\\
Need to show: $ \llbracket   \mathbf{proj}_i \:  \ottnt{a}   \rrbracket  \in \text{Hom}_{\Ct} ( \llbracket  \Gamma  \rrbracket ,  \mathbf{S}_{ m }   \llbracket   A_i   \rrbracket  )$.\\
By IH, $ \llbracket  \ottnt{a}  \rrbracket  \in \text{Hom}_{\Ct} ( \llbracket  \Gamma  \rrbracket  ,  \mathbf{S}_{ m }   (    \llbracket  \ottnt{A_{{\mathrm{1}}}}  \rrbracket   \times   \llbracket  \ottnt{A_{{\mathrm{2}}}}  \rrbracket    )  )$.\\
Now, $ \llbracket   \mathbf{proj}_i \:  \ottnt{a}   \rrbracket  =  \llbracket  \Gamma  \rrbracket  \xrightarrow{ \llbracket  \ottnt{a}  \rrbracket }  \mathbf{S}_{ m }   (    \llbracket  \ottnt{A_{{\mathrm{1}}}}  \rrbracket   \times   \llbracket  \ottnt{A_{{\mathrm{2}}}}  \rrbracket    )   \xrightarrow{ \mathbf{S}_{ m }   \pi_i  }  \mathbf{S}_{ m }   \llbracket   A_i   \rrbracket  $.
\item \Rref{E-Split}. Have: $ \Gamma  \vdash   \mathbf{split}^{ m_{{\mathrm{2}}} }  \ottnt{a}   :^{ m_{{\mathrm{1}}} }   S_{ m_{{\mathrm{2}}} } \:  \ottnt{A}  $ where $ \Gamma  \vdash  \ottnt{a}  :^{  m_{{\mathrm{1}}}  \cdot  m_{{\mathrm{2}}}  }  \ottnt{A} $.\\
Need to show: $ \llbracket   \mathbf{split}^{ m_{{\mathrm{2}}} }  \ottnt{a}   \rrbracket  \in \text{Hom}_{\Ct} ( \llbracket  \Gamma  \rrbracket ,  \mathbf{S}_{ m_{{\mathrm{1}}} }   \mathbf{S}_{ m_{{\mathrm{2}}} }   \llbracket  \ottnt{A}  \rrbracket   )$.\\
By IH, $ \llbracket  \ottnt{a}  \rrbracket  \in \text{Hom}_{\Ct}( \llbracket  \Gamma  \rrbracket  ,  \mathbf{S}_{  m_{{\mathrm{1}}}  \cdot  m_{{\mathrm{2}}}  }   \llbracket  \ottnt{A}  \rrbracket  )$.\\
Now, $ \llbracket   \mathbf{split}^{ m_{{\mathrm{2}}} }  \ottnt{a}   \rrbracket  =  \llbracket  \Gamma  \rrbracket  \xrightarrow{ \llbracket  \ottnt{a}  \rrbracket }  \mathbf{S}_{  m_{{\mathrm{1}}}  \cdot  m_{{\mathrm{2}}}  }   \llbracket  \ottnt{A}  \rrbracket   \xrightarrow{ \delta^{ m_{{\mathrm{1}}} , m_{{\mathrm{2}}} }_{  \llbracket  \ottnt{A}  \rrbracket  } }  \mathbf{S}_{ m_{{\mathrm{1}}} }   \mathbf{S}_{ m_{{\mathrm{2}}} }   \llbracket  \ottnt{A}  \rrbracket   $.
\item \Rref{E-Merge}. Have: $ \Gamma  \vdash   \mathbf{merge}^{ m_{{\mathrm{2}}} }  \ottnt{a}   :^{  m_{{\mathrm{1}}}  \cdot  m_{{\mathrm{2}}}  }  \ottnt{A} $ where $ \Gamma  \vdash  \ottnt{a}  :^{ m_{{\mathrm{1}}} }   S_{ m_{{\mathrm{2}}} } \:  \ottnt{A}  $.\\
Need to show: $ \llbracket   \mathbf{merge}^{ m_{{\mathrm{2}}} }  \ottnt{a}   \rrbracket  \in \text{Hom}_{\Ct} ( \llbracket  \Gamma  \rrbracket ,  \mathbf{S}_{  m_{{\mathrm{1}}}  \cdot  m_{{\mathrm{2}}}  }   \llbracket  \ottnt{A}  \rrbracket  )$.\\
By IH, $ \llbracket  \ottnt{a}  \rrbracket  \in \text{Hom}_{\Ct} ( \llbracket  \Gamma  \rrbracket  ,  \mathbf{S}_{ m_{{\mathrm{1}}} }   \mathbf{S}_{ m_{{\mathrm{2}}} }   \llbracket  \ottnt{A}  \rrbracket   )$.\\
Now, $ \llbracket   \mathbf{merge}^{ m_{{\mathrm{2}}} }  \ottnt{a}   \rrbracket  =  \llbracket  \Gamma  \rrbracket  \xrightarrow{ \llbracket  \ottnt{a}  \rrbracket }  \mathbf{S}_{ m_{{\mathrm{1}}} }   \mathbf{S}_{ m_{{\mathrm{2}}} }   \llbracket  \ottnt{A}  \rrbracket    \xrightarrow{ \mu^{ m_{{\mathrm{1}}} , m_{{\mathrm{2}}} }_{  \llbracket  \ottnt{A}  \rrbracket  } }  \mathbf{S}_{  m_{{\mathrm{1}}}  \cdot  m_{{\mathrm{2}}}  }   \llbracket  \ottnt{A}  \rrbracket  $.
\item \Rref{E-Up}. Have: $ \Gamma  \vdash  \ottnt{a}  :^{ m_{{\mathrm{2}}} }  \ottnt{A} $ where $ \Gamma  \vdash  \ottnt{a}  :^{ m_{{\mathrm{1}}} }  \ottnt{A} $ and $ m_{{\mathrm{1}}}   \leq   m_{{\mathrm{2}}} $.\\
Need to show: $\llbracket  \Gamma  \vdash  \ottnt{a}  :^{ m_{{\mathrm{2}}} }  \ottnt{A}  \rrbracket \in \text{Hom}_{\Ct} ( \llbracket  \Gamma  \rrbracket  ,  \mathbf{S}_{ m_{{\mathrm{2}}} }   \llbracket  \ottnt{A}  \rrbracket  )$.\\
By IH, $\llbracket  \Gamma  \vdash  \ottnt{a}  :^{ m_{{\mathrm{1}}} }  \ottnt{A}  \rrbracket \in \text{Hom}_{\Ct} ( \llbracket  \Gamma  \rrbracket  ,  \mathbf{S}_{ m_{{\mathrm{1}}} }   \llbracket  \ottnt{A}  \rrbracket  )$.\\
Now, $\llbracket  \Gamma  \vdash  \ottnt{a}  :^{ m_{{\mathrm{2}}} }  \ottnt{A}  \rrbracket =  \llbracket  \Gamma  \rrbracket  \xrightarrow{\llbracket  \Gamma  \vdash  \ottnt{a}  :^{ m_{{\mathrm{1}}} }  \ottnt{A}  \rrbracket}  \mathbf{S}_{ m_{{\mathrm{1}}} }   \llbracket  \ottnt{A}  \rrbracket   \xrightarrow{ \mathbf{S}^{ m_{{\mathrm{1}}}  \leq  m_{{\mathrm{2}}} }_{  \llbracket  \ottnt{A}  \rrbracket  } }  \mathbf{S}_{ m_{{\mathrm{2}}} }   \llbracket  \ottnt{A}  \rrbracket  $.
\end{itemize}

For the second part, invert the equality judgement, $ \ottnt{a_{{\mathrm{1}}}}  \equiv  \ottnt{a_{{\mathrm{2}}}} $. Most of the cases are similar to those in Theorem \ref{lcCatprf}. We present the differing ones below.

\begin{itemize}
\item $  \mathbf{proj}_i \:   (  \ottnt{a_{{\mathrm{1}}}}  ,  \ottnt{a_{{\mathrm{2}}}}  )    \equiv   a_i  $.\\
Now, \begin{align*}
&  \llbracket   \mathbf{proj}_i \:   (  \ottnt{a_{{\mathrm{1}}}}  ,  \ottnt{a_{{\mathrm{2}}}}  )    \rrbracket  \\
= &    \mathbf{S}_{ m }   \pi_i    \circ   p^{-1}_{  \llbracket  \ottnt{A_{{\mathrm{1}}}}  \rrbracket  ,   \llbracket  \ottnt{A_{{\mathrm{2}}}}  \rrbracket  }    \circ   \langle   \llbracket  \ottnt{a_{{\mathrm{1}}}}  \rrbracket   ,   \llbracket  \ottnt{a_{{\mathrm{2}}}}  \rrbracket   \rangle   \\
= &     \pi_i   \circ   \langle   \mathbf{S}_{ m }   \pi_1    ,   \mathbf{S}_{ m }   \pi_2    \rangle    \circ   p^{-1}_{  \llbracket  \ottnt{A_{{\mathrm{1}}}}  \rrbracket  ,   \llbracket  \ottnt{A_{{\mathrm{2}}}}  \rrbracket  }    \circ   \langle   \llbracket  \ottnt{a_{{\mathrm{1}}}}  \rrbracket   ,   \llbracket  \ottnt{a_{{\mathrm{2}}}}  \rrbracket   \rangle   \\
= &     \pi_i   \circ   p_{  \llbracket  \ottnt{A_{{\mathrm{1}}}}  \rrbracket  ,   \llbracket  \ottnt{A_{{\mathrm{2}}}}  \rrbracket  }    \circ   p^{-1}_{  \llbracket  \ottnt{A_{{\mathrm{1}}}}  \rrbracket  ,   \llbracket  \ottnt{A_{{\mathrm{2}}}}  \rrbracket  }    \circ   \langle   \llbracket  \ottnt{a_{{\mathrm{1}}}}  \rrbracket   ,   \llbracket  \ottnt{a_{{\mathrm{2}}}}  \rrbracket   \rangle   =   \pi_i   \circ   \langle   \llbracket  \ottnt{a_{{\mathrm{1}}}}  \rrbracket   ,   \llbracket  \ottnt{a_{{\mathrm{2}}}}  \rrbracket   \rangle   =  \llbracket   a_i   \rrbracket .
\end{align*}
\item $ \ottnt{a}  \equiv   (   \mathbf{proj}_1 \:  \ottnt{a}   ,   \mathbf{proj}_2 \:  \ottnt{a}   )  $.\\
Now, \begin{align*}
&  \llbracket   (   \mathbf{proj}_1 \:  \ottnt{a}   ,   \mathbf{proj}_2 \:  \ottnt{a}   )   \rrbracket  \\
= &   p^{-1}_{  \llbracket  \ottnt{A_{{\mathrm{1}}}}  \rrbracket  ,   \llbracket  \ottnt{A_{{\mathrm{2}}}}  \rrbracket  }   \circ   \langle   \llbracket   \mathbf{proj}_1 \:  \ottnt{a}   \rrbracket   ,   \llbracket   \mathbf{proj}_2 \:  \ottnt{a}   \rrbracket   \rangle   \\
= &   p^{-1}_{  \llbracket  \ottnt{A_{{\mathrm{1}}}}  \rrbracket  ,   \llbracket  \ottnt{A_{{\mathrm{2}}}}  \rrbracket  }   \circ   \langle    \mathbf{S}_{ m }   \pi_1    \circ   \llbracket  \ottnt{a}  \rrbracket    ,    \mathbf{S}_{ m }   \pi_2    \circ   \llbracket  \ottnt{a}  \rrbracket    \rangle   \\
= &    p^{-1}_{  \llbracket  \ottnt{A_{{\mathrm{1}}}}  \rrbracket  ,   \llbracket  \ottnt{A_{{\mathrm{2}}}}  \rrbracket  }   \circ   \langle   \mathbf{S}_{ m }   \pi_1    ,   \mathbf{S}_{ m }   \pi_2    \rangle    \circ   \llbracket  \ottnt{a}  \rrbracket   =    p^{-1}_{  \llbracket  \ottnt{A_{{\mathrm{1}}}}  \rrbracket  ,   \llbracket  \ottnt{A_{{\mathrm{2}}}}  \rrbracket  }   \circ   p_{  \llbracket  \ottnt{A_{{\mathrm{1}}}}  \rrbracket  ,   \llbracket  \ottnt{A_{{\mathrm{2}}}}  \rrbracket  }    \circ   \llbracket  \ottnt{a}  \rrbracket   =  \llbracket  \ottnt{a}  \rrbracket 
\end{align*}
\item $ \Gamma  \vdash   \mathbf{merge}^{ m }   (   \mathbf{split}^{ m }  \ottnt{a}   )    \equiv  \ottnt{a}  :^{  m'  \cdot  m  }  \ottnt{A} $.\\
Now, \begin{align*}
& \llbracket  \Gamma  \vdash   \mathbf{merge}^{ m }   (   \mathbf{split}^{ m }  \ottnt{a}   )    :^{  m'  \cdot  m  }  \ottnt{A}  \rrbracket\\
= &  \mu^{ m' , m }_{  \llbracket  \ottnt{A}  \rrbracket  }  \circ \llbracket  \Gamma  \vdash   \mathbf{split}^{ m }  \ottnt{a}   :^{ m' }   S_{ m } \:  \ottnt{A}   \rrbracket \\
= &   \mu^{ m' , m }_{  \llbracket  \ottnt{A}  \rrbracket  }   \circ   \delta^{ m' , m }_{  \llbracket  \ottnt{A}  \rrbracket  }   \circ \llbracket  \Gamma  \vdash  \ottnt{a}  :^{  m'  \cdot  m  }  \ottnt{A}  \rrbracket = \llbracket  \Gamma  \vdash  \ottnt{a}  :^{  m'  \cdot  m  }  \ottnt{A}  \rrbracket
\end{align*} 
\item $ \Gamma  \vdash  \ottnt{a}  \equiv   \mathbf{split}^{ m }   (   \mathbf{merge}^{ m }  \ottnt{a}   )    :^{ m' }   S_{ m } \:  \ottnt{A}  $.\\
Now, \begin{align*}
& \llbracket  \Gamma  \vdash   \mathbf{split}^{ m }   (   \mathbf{merge}^{ m }  \ottnt{a}   )    :^{ m' }   S_{ m } \:  \ottnt{A}   \rrbracket \\
= &  \delta^{ m' , m }_{  \llbracket  \ottnt{A}  \rrbracket  }  \circ \llbracket  \Gamma  \vdash   \mathbf{merge}^{ m }  \ottnt{a}   :^{  m'  \cdot  m  }  \ottnt{A}  \rrbracket \\
= &   \delta^{ m' , m }_{  \llbracket  \ottnt{A}  \rrbracket  }   \circ   \mu^{ m' , m }_{  \llbracket  \ottnt{A}  \rrbracket  }   \circ \llbracket  \Gamma  \vdash  \ottnt{a}  :^{ m' }   S_{ m } \:  \ottnt{A}   \rrbracket = \llbracket  \Gamma  \vdash  \ottnt{a}  :^{ m' }   S_{ m } \:  \ottnt{A}   \rrbracket
\end{align*}
\end{itemize}
\end{proof}


\begin{theorem}[Theorem \ref{gmcc2gmcce}]
If $ \Gamma  \vdash  \ottnt{a}  :  \ottnt{A} $ in GMCC($ \mathcal{L} $), then $  \Gamma ^{   \bot   }   \vdash   \widetilde{ \ottnt{a} }   :^{   \bot   }  \ottnt{A} $ in \Ge{}($ \mathcal{L} $). Further, if $ \Gamma  \vdash  \ottnt{a_{{\mathrm{1}}}}  :  \ottnt{A} $ and $ \Gamma  \vdash  \ottnt{a_{{\mathrm{2}}}}  :  \ottnt{A} $ such that $ \ottnt{a_{{\mathrm{1}}}}  \equiv  \ottnt{a_{{\mathrm{2}}}} $ in GMCC($ \mathcal{L} $), then $  \widetilde{ \ottnt{a_{{\mathrm{1}}}} }   \equiv   \widetilde{ \ottnt{a_{{\mathrm{2}}}} }  $ in \Ge{}($ \mathcal{L} $).
\end{theorem}

\begin{proof}
Let $ \Gamma  \vdash  \ottnt{a}  :  \ottnt{A} $ in GMCC($ \mathcal{L} $). We show $  \Gamma ^{   \bot   }   \vdash   \widetilde{ \ottnt{a} }   :^{   \bot   }  \ottnt{A} $ in \Ge{}($ \mathcal{L} $) by induction on the typing derivation.
\begin{itemize}
\item $\lambda$-calculus. By IH.
\item \Rref{MC-Return}. Have: $ \Gamma  \vdash   \ottkw{ret}  \:  \ottnt{a}   :   S_{   \bot   } \:  \ottnt{A}  $ where $ \Gamma  \vdash  \ottnt{a}  :  \ottnt{A} $.\\
Need to show: $  \Gamma ^{   \bot   }   \vdash   \mathbf{split}^{   \bot   }   \widetilde{ \ottnt{a} }    :^{   \bot   }   S_{   \bot   } \:  \ottnt{A}  $.\\
By IH, $  \Gamma ^{   \bot   }   \vdash   \widetilde{ \ottnt{a} }   :^{   \bot   }  \ottnt{A} $.\\
This case follows by \rref{E-Split}.
\item \Rref{MC-Extract}. Have: $ \Gamma  \vdash   \mathbf{extr} \:  \ottnt{a}   :  \ottnt{A} $ where $ \Gamma  \vdash  \ottnt{a}  :   S_{   \bot   } \:  \ottnt{A}  $.\\
Need to show: $  \Gamma ^{   \bot   }   \vdash   \mathbf{merge}^{   \bot   }   \widetilde{ \ottnt{a} }    :^{   \bot   }  \ottnt{A} $.\\
By IH, $  \Gamma ^{   \bot   }   \vdash   \widetilde{ \ottnt{a} }   :^{   \bot   }   S_{   \bot   } \:  \ottnt{A}  $.\\
This case follows by \rref{E-Merge}.
\item \Rref{MC-Join}. Have: $ \Gamma  \vdash   \mathbf{join}^{  \ell_{{\mathrm{1}}}  ,  \ell_{{\mathrm{2}}}  }  \ottnt{a}   :   S_{   \ell_{{\mathrm{1}}}  \vee  \ell_{{\mathrm{2}}}   } \:  \ottnt{A}  $ where $ \Gamma  \vdash  \ottnt{a}  :   S_{  \ell_{{\mathrm{1}}}  } \:   S_{  \ell_{{\mathrm{2}}}  } \:  \ottnt{A}   $.\\
Need to show: $  \Gamma ^{   \bot   }   \vdash   \mathbf{split}^{   \ell_{{\mathrm{1}}}  \vee  \ell_{{\mathrm{2}}}   }   (   \mathbf{merge}^{  \ell_{{\mathrm{2}}}  }   (   \mathbf{merge}^{  \ell_{{\mathrm{1}}}  }   \widetilde{ \ottnt{a} }    )    )    :^{   \bot   }   S_{   \ell_{{\mathrm{1}}}  \vee  \ell_{{\mathrm{2}}}   } \:  \ottnt{A}  $.\\
By IH, $  \Gamma ^{   \bot   }   \vdash   \widetilde{ \ottnt{a} }   :^{   \bot   }   S_{  \ell_{{\mathrm{1}}}  } \:   S_{  \ell_{{\mathrm{2}}}  } \:  \ottnt{A}   $.\\
By \rref{E-Merge}, $  \Gamma ^{   \bot   }   \vdash   \mathbf{merge}^{  \ell_{{\mathrm{1}}}  }   \widetilde{ \ottnt{a} }    :^{  \ell_{{\mathrm{1}}}  }   S_{  \ell_{{\mathrm{2}}}  } \:  \ottnt{A}  $.\\
Applying \rref{E-Merge} again, $  \Gamma ^{   \bot   }   \vdash   \mathbf{merge}^{  \ell_{{\mathrm{2}}}  }   (   \mathbf{merge}^{  \ell_{{\mathrm{1}}}  }   \widetilde{ \ottnt{a} }    )    :^{   \ell_{{\mathrm{1}}}  \vee  \ell_{{\mathrm{2}}}   }  \ottnt{A} $.\\
This case, then, follows by \rref{E-Split}.
\item \Rref{MC-Fork}. Have: $ \Gamma  \vdash   \mathbf{fork}^{  \ell_{{\mathrm{1}}}  ,  \ell_{{\mathrm{2}}}  }  \ottnt{a}   :   S_{  \ell_{{\mathrm{1}}}  } \:   S_{  \ell_{{\mathrm{2}}}  } \:  \ottnt{A}   $ where $ \Gamma  \vdash  \ottnt{a}  :   S_{   \ell_{{\mathrm{1}}}  \vee  \ell_{{\mathrm{2}}}   } \:  \ottnt{A}  $.\\
Need to show: $  \Gamma ^{   \bot   }   \vdash   \mathbf{split}^{  \ell_{{\mathrm{1}}}  }   (   \mathbf{split}^{  \ell_{{\mathrm{2}}}  }   (   \mathbf{merge}^{   \ell_{{\mathrm{1}}}  \vee  \ell_{{\mathrm{2}}}   }   \widetilde{ \ottnt{a} }    )    )    :^{   \bot   }   S_{  \ell_{{\mathrm{1}}}  } \:   S_{  \ell_{{\mathrm{2}}}  } \:  \ottnt{A}   $.\\
By IH, $  \Gamma ^{   \bot   }   \vdash   \widetilde{ \ottnt{a} }   :^{   \bot   }   S_{   \ell_{{\mathrm{1}}}  \vee  \ell_{{\mathrm{2}}}   } \:  \ottnt{A}  $.\\
By \rref{E-Merge}, $  \Gamma ^{   \bot   }   \vdash   \mathbf{merge}^{   \ell_{{\mathrm{1}}}  \vee  \ell_{{\mathrm{2}}}   }   \widetilde{ \ottnt{a} }    :^{   \ell_{{\mathrm{1}}}  \vee  \ell_{{\mathrm{2}}}   }  \ottnt{A} $.\\
By \rref{E-Split}, $  \Gamma ^{   \bot   }   \vdash   \mathbf{split}^{  \ell_{{\mathrm{2}}}  }   (   \mathbf{merge}^{   \ell_{{\mathrm{1}}}  \vee  \ell_{{\mathrm{2}}}   }   \widetilde{ \ottnt{a} }    )    :^{  \ell_{{\mathrm{1}}}  }   S_{  \ell_{{\mathrm{2}}}  } \:  \ottnt{A}  $.\\
Applying \rref{E-Split} again, $  \Gamma ^{   \bot   }   \vdash   \mathbf{split}^{  \ell_{{\mathrm{1}}}  }   (   \mathbf{split}^{  \ell_{{\mathrm{2}}}  }   (   \mathbf{merge}^{   \ell_{{\mathrm{1}}}  \vee  \ell_{{\mathrm{2}}}   }   \widetilde{ \ottnt{a} }    )    )    :^{   \bot   }   S_{  \ell_{{\mathrm{1}}}  } \:   S_{  \ell_{{\mathrm{2}}}  } \:  \ottnt{A}   $.
\item \Rref{MC-Fmap}. Have: $ \Gamma  \vdash   \mathbf{lift}^{  \ell  }  \ottnt{f}   :    S_{  \ell  } \:  \ottnt{A}   \to   S_{  \ell  } \:  \ottnt{B}   $ where $ \Gamma  \vdash  \ottnt{f}  :   \ottnt{A}  \to  \ottnt{B}  $.\\
Need to show: $  \Gamma ^{   \bot   }   \vdash   \lambda  \ottmv{x}  .   \mathbf{split}^{  \ell  }   (     \widetilde{ \ottnt{f} }    \:   (   \mathbf{merge}^{  \ell  }  \ottmv{x}   )    )     :^{   \bot   }    S_{  \ell  } \:  \ottnt{A}   \to   S_{  \ell  } \:  \ottnt{B}   $.\\
By IH, $  \Gamma ^{   \bot   }   \vdash   \widetilde{ \ottnt{f} }   :^{   \bot   }   \ottnt{A}  \to  \ottnt{B}  $.\\
Now,\[ \infer[]{  \Gamma ^{   \bot   }   \vdash   \lambda  \ottmv{x}  .   \mathbf{split}^{  \ell  }   (     \widetilde{ \ottnt{f} }    \:   (   \mathbf{merge}^{  \ell  }  \ottmv{x}   )    )     :^{   \bot   }    S_{  \ell  } \:  \ottnt{A}   \to   S_{  \ell  } \:  \ottnt{B}   }
             {\infer[]{   \Gamma ^{   \bot   }   ,   \ottmv{x}  :^{   \bot   }   S_{  \ell  } \:  \ottnt{A}     \vdash   \mathbf{split}^{  \ell  }   (     \widetilde{ \ottnt{f} }    \:   (   \mathbf{merge}^{  \ell  }  \ottmv{x}   )    )    :^{   \bot   }   S_{  \ell  } \:  \ottnt{B}  }
             {\infer[]{   \Gamma ^{   \bot   }   ,   \ottmv{x}  :^{   \bot   }   S_{  \ell  } \:  \ottnt{A}     \vdash     \widetilde{ \ottnt{f} }    \:   (   \mathbf{merge}^{  \ell  }  \ottmv{x}   )    :^{  \ell  }  \ottnt{B} }
             {\infer[\text{(E-Up)}]{   \Gamma ^{   \bot   }   ,   \ottmv{x}  :^{   \bot   }   S_{  \ell  } \:  \ottnt{A}     \vdash   \widetilde{ \ottnt{f} }   :^{  \ell  }   \ottnt{A}  \to  \ottnt{B}  }
             {   \Gamma ^{   \bot   }   ,   \ottmv{x}  :^{   \bot   }   S_{  \ell  } \:  \ottnt{A}     \vdash   \widetilde{ \ottnt{f} }   :^{   \bot   }   \ottnt{A}  \to  \ottnt{B}  }
             & 
             \infer[]{   \Gamma ^{   \bot   }   ,   \ottmv{x}  :^{   \bot   }   S_{  \ell  } \:  \ottnt{A}     \vdash   \mathbf{merge}^{  \ell  }  \ottmv{x}   :^{  \ell  }  \ottnt{A} }
             {   \Gamma ^{   \bot   }   ,   \ottmv{x}  :^{   \bot   }   S_{  \ell  } \:  \ottnt{A}     \vdash  \ottmv{x}  :^{   \bot   }   S_{  \ell  } \:  \ottnt{A}  }}}} \]
\item \Rref{MC-Up}. Have: $ \Gamma  \vdash   \mathbf{up}^{  \ell_{{\mathrm{1}}}  ,  \ell_{{\mathrm{2}}}  }  \ottnt{a}   :   S_{  \ell_{{\mathrm{2}}}  } \:  \ottnt{A}  $ where $ \Gamma  \vdash  \ottnt{a}  :   S_{  \ell_{{\mathrm{1}}}  } \:  \ottnt{A}  $ and $ \ell_{{\mathrm{1}}}  \sqsubseteq  \ell_{{\mathrm{2}}} $.\\
Need to show: $  \Gamma ^{   \bot   }   \vdash   \mathbf{split}^{  \ell_{{\mathrm{2}}}  }   (   \mathbf{merge}^{  \ell_{{\mathrm{1}}}  }   \widetilde{ \ottnt{a} }    )    :^{   \bot   }   S_{  \ell_{{\mathrm{2}}}  } \:  \ottnt{A}  $.\\
By IH, $  \Gamma ^{   \bot   }   \vdash   \widetilde{ \ottnt{a} }   :^{   \bot   }   S_{  \ell_{{\mathrm{1}}}  } \:  \ottnt{A}  $.\\
By \rref{E-Merge}, $  \Gamma ^{   \bot   }   \vdash   \mathbf{merge}^{  \ell_{{\mathrm{1}}}  }   \widetilde{ \ottnt{a} }    :^{  \ell_{{\mathrm{1}}}  }  \ottnt{A} $.\\
By \rref{E-Up}, $  \Gamma ^{   \bot   }   \vdash   \mathbf{merge}^{  \ell_{{\mathrm{1}}}  }   \widetilde{ \ottnt{a} }    :^{  \ell_{{\mathrm{2}}}  }  \ottnt{A} $.\\
This case, then, follows by \rref{E-Split}.
\end{itemize}

Next, we show that if $ \Gamma  \vdash  \ottnt{a_{{\mathrm{1}}}}  :  \ottnt{A} $ and $ \Gamma  \vdash  \ottnt{a_{{\mathrm{2}}}}  :  \ottnt{A} $ such that $ \ottnt{a_{{\mathrm{1}}}}  \equiv  \ottnt{a_{{\mathrm{2}}}} $ in GMCC($ \mathcal{L} $), then $  \widetilde{ \ottnt{a_{{\mathrm{1}}}} }   \equiv   \widetilde{ \ottnt{a_{{\mathrm{2}}}} }  $ in \Ge{}($ \mathcal{L} $).\\
By inversion on $ \ottnt{a_{{\mathrm{1}}}}  \equiv  \ottnt{a_{{\mathrm{2}}}} $.
\begin{itemize}
\item $\lambda$-calculus. By IH.
\item $  \mathbf{lift}^{  \ell  }   (   \lambda  \ottmv{x}  .  \ottmv{x}   )    \equiv   \lambda  \ottmv{x}  .  \ottmv{x}  $.\\
Now, \begin{align*}
& \wdtilde{ \mathbf{lift}^{  \ell  }   (   \lambda  \ottmv{x}  .  \ottmv{x}   )  } \\
= &  \lambda  \ottmv{y}  .   \mathbf{split}^{  \ell  }   (    (   \lambda  \ottmv{x}  .  \ottmv{x}   )   \:   (   \mathbf{merge}^{  \ell  }  \ottmv{y}   )    )    \\
\equiv &  \lambda  \ottmv{y}  .   \mathbf{split}^{  \ell  }   (   \mathbf{merge}^{  \ell  }  \ottmv{y}   )    \equiv  \lambda  \ottmv{y}  .  \ottmv{y} 
\end{align*}
\item $  \mathbf{lift}^{  \ell  }   (    \lambda  \ottmv{x}  .  \ottnt{g}   \:   (   \ottnt{f}  \:  \ottmv{x}   )    )    \equiv    \lambda  \ottmv{x}  .   (   \mathbf{lift}^{  \ell  }  \ottnt{g}   )    \:   (    (   \mathbf{lift}^{  \ell  }  \ottnt{f}   )   \:  \ottmv{x}   )   $.\\
Now, \begin{align*}
& \wdtilde{ \mathbf{lift}^{  \ell  }   (    \lambda  \ottmv{x}  .  \ottnt{g}   \:   (   \ottnt{f}  \:  \ottmv{x}   )    )  } \\
= &  \lambda  \ottmv{y}  .   \mathbf{split}^{  \ell  }   (    (    \lambda  \ottmv{x}  .    \widetilde{ \ottnt{g} }     \:   (     \widetilde{ \ottnt{f} }    \:  \ottmv{x}   )    )   \:   (   \mathbf{merge}^{  \ell  }  \ottmv{y}   )    )    \\
\equiv &  \lambda  \ottmv{y}  .   \mathbf{split}^{  \ell  }   (     \widetilde{ \ottnt{g} }    \:   (     \widetilde{ \ottnt{f} }    \:   (   \mathbf{merge}^{  \ell  }  \ottmv{y}   )    )    )    
\end{align*}
and \begin{align*}
& \wdtilde{  \lambda  \ottmv{x}  .   (   \mathbf{lift}^{  \ell  }  \ottnt{g}   )    \:   (    (   \mathbf{lift}^{  \ell  }  \ottnt{f}   )   \:  \ottmv{x}   )  } \\
\equiv &   \lambda  \ottmv{x}  .   (   \lambda  \ottmv{y}  .   \mathbf{split}^{  \ell  }   (     \widetilde{ \ottnt{g} }    \:   (   \mathbf{merge}^{  \ell  }  \ottmv{y}   )    )     )    \:   (   \mathbf{split}^{  \ell  }   (     \widetilde{ \ottnt{f} }    \:   (   \mathbf{merge}^{  \ell  }  \ottmv{x}   )    )    )   \\
\equiv &  \lambda  \ottmv{x}  .   \mathbf{split}^{  \ell  }   (     \widetilde{ \ottnt{g} }    \:   (   \mathbf{merge}^{  \ell  }   (   \mathbf{split}^{  \ell  }   (     \widetilde{ \ottnt{f} }    \:   (   \mathbf{merge}^{  \ell  }  \ottmv{x}   )    )    )    )    )    \\
\equiv &  \lambda  \ottmv{x}  .   \mathbf{split}^{  \ell  }   (     \widetilde{ \ottnt{g} }    \:   (     \widetilde{ \ottnt{f} }    \:   (   \mathbf{merge}^{  \ell  }  \ottmv{x}   )    )    )   
\end{align*}
\item $  \mathbf{up}^{  \ell_{{\mathrm{1}}}  ,  \ell_{{\mathrm{1}}}  }  \ottnt{a}   \equiv  \ottnt{a} $.\\
Now, $ \widetilde{   \mathbf{up}^{  \ell_{{\mathrm{1}}}  ,  \ell_{{\mathrm{1}}}  }  \ottnt{a}   }  =  \mathbf{split}^{  \ell_{{\mathrm{1}}}  }   (   \mathbf{merge}^{  \ell_{{\mathrm{1}}}  }   \widetilde{ \ottnt{a} }    )   \equiv  \widetilde{ \ottnt{a} } $.
\item $  \mathbf{up}^{  \ell_{{\mathrm{2}}}  ,  \ell_{{\mathrm{3}}}  }   (   \mathbf{up}^{  \ell_{{\mathrm{1}}}  ,  \ell_{{\mathrm{2}}}  }  \ottnt{a}   )    \equiv   \mathbf{up}^{  \ell_{{\mathrm{1}}}  ,  \ell_{{\mathrm{3}}}  }  \ottnt{a}  $.\\
Now,\begin{align*}
& \wdtilde{ \mathbf{up}^{  \ell_{{\mathrm{2}}}  ,  \ell_{{\mathrm{3}}}  }   (   \mathbf{up}^{  \ell_{{\mathrm{1}}}  ,  \ell_{{\mathrm{2}}}  }  \ottnt{a}   )  } \\
= &  \mathbf{split}^{  \ell_{{\mathrm{3}}}  }   (   \mathbf{merge}^{  \ell_{{\mathrm{2}}}  }   (   \mathbf{split}^{  \ell_{{\mathrm{2}}}  }   (   \mathbf{merge}^{  \ell_{{\mathrm{1}}}  }   \widetilde{ \ottnt{a} }    )    )    )   \\
\equiv &  \mathbf{split}^{  \ell_{{\mathrm{3}}}  }   (   \mathbf{merge}^{  \ell_{{\mathrm{1}}}  }   \widetilde{ \ottnt{a} }    )   =  \widetilde{  \mathbf{up}^{  \ell_{{\mathrm{1}}}  ,  \ell_{{\mathrm{3}}}  }  \ottnt{a}  } 
\end{align*} 
\item $   (   \mathbf{up}^{  \ell_{{\mathrm{1}}}  ,  \ell'_{{\mathrm{1}}}  }  \ottnt{a}   )   \:  \leftindex^{  \ell'_{{\mathrm{1}}}  }{\gg}\!\! =^{  \ell_{{\mathrm{2}}}  }  \ottnt{f}   \equiv   \mathbf{up}^{    \ell_{{\mathrm{1}}}  \vee  \ell_{{\mathrm{2}}}    ,    \ell'_{{\mathrm{1}}}  \vee  \ell_{{\mathrm{2}}}    }   (   \ottnt{a}  \:  \leftindex^{  \ell_{{\mathrm{1}}}  }{\gg}\!\! =^{  \ell_{{\mathrm{2}}}  }  \ottnt{f}   )   $.\\
First, note that: \begin{align*}
& \wdtilde{ \ottnt{a}  \:  \leftindex^{  \ell_{{\mathrm{1}}}  }{\gg}\!\! =^{  \ell_{{\mathrm{2}}}  }  \ottnt{f} } \\
= & \wdtilde{ \mathbf{join}^{  \ell_{{\mathrm{1}}}  ,  \ell_{{\mathrm{2}}}  }   (    (   \mathbf{lift}^{  \ell_{{\mathrm{1}}}  }  \ottnt{f}   )   \:  \ottnt{a}   )  }\\
= &  \mathbf{split}^{   \ell_{{\mathrm{1}}}  \vee  \ell_{{\mathrm{2}}}   }   (   \mathbf{merge}^{  \ell_{{\mathrm{2}}}  }   (   \mathbf{merge}^{  \ell_{{\mathrm{1}}}  }   (    \widetilde{   \mathbf{lift}^{  \ell_{{\mathrm{1}}}  }  \ottnt{f}   }   \:   \widetilde{ \ottnt{a} }    )    )    )   \\
\equiv &  \mathbf{split}^{   \ell_{{\mathrm{1}}}  \vee  \ell_{{\mathrm{2}}}   }   (   \mathbf{merge}^{  \ell_{{\mathrm{2}}}  }   (   \mathbf{merge}^{  \ell_{{\mathrm{1}}}  }   (   \mathbf{split}^{  \ell_{{\mathrm{1}}}  }   (     \widetilde{ \ottnt{f} }    \:   (   \mathbf{merge}^{  \ell_{{\mathrm{1}}}  }    \widetilde{ \ottnt{a} }     )    )    )    )    )   \\
\equiv &  \mathbf{split}^{   \ell_{{\mathrm{1}}}  \vee  \ell_{{\mathrm{2}}}   }   (   \mathbf{merge}^{  \ell_{{\mathrm{2}}}  }   (     \widetilde{ \ottnt{f} }    \:   (   \mathbf{merge}^{  \ell_{{\mathrm{1}}}  }    \widetilde{ \ottnt{a} }     )    )    )  
\end{align*}
Next, \begin{align*}
& \wdtilde{  (   \mathbf{up}^{  \ell_{{\mathrm{1}}}  ,  \ell'_{{\mathrm{1}}}  }  \ottnt{a}   )   \:  \leftindex^{  \ell'_{{\mathrm{1}}}  }{\gg}\!\! =^{  \ell_{{\mathrm{2}}}  }  \ottnt{f} } \\
\equiv &   \mathbf{split}^{   \ell'_{{\mathrm{1}}}  \vee  \ell_{{\mathrm{2}}}   }   (   \mathbf{merge}^{  \ell_{{\mathrm{2}}}  }   (     \widetilde{ \ottnt{f} }    \:   (   \mathbf{merge}^{  \ell'_{{\mathrm{1}}}  }    \widetilde{  \mathbf{up}^{  \ell_{{\mathrm{1}}}  ,  \ell'_{{\mathrm{1}}}  }  \ottnt{a}  }     )    )    )   \\
\equiv &   \mathbf{split}^{   \ell'_{{\mathrm{1}}}  \vee  \ell_{{\mathrm{2}}}   }   (   \mathbf{merge}^{  \ell_{{\mathrm{2}}}  }   (     \widetilde{ \ottnt{f} }    \:   (   \mathbf{merge}^{  \ell'_{{\mathrm{1}}}  }   (   \mathbf{split}^{  \ell'_{{\mathrm{1}}}  }   (   \mathbf{merge}^{  \ell_{{\mathrm{1}}}  }   \widetilde{ \ottnt{a} }    )    )    )    )    )   \\
\equiv &   \mathbf{split}^{   \ell'_{{\mathrm{1}}}  \vee  \ell_{{\mathrm{2}}}   }   (   \mathbf{merge}^{  \ell_{{\mathrm{2}}}  }   (     \widetilde{ \ottnt{f} }    \:   (   \mathbf{merge}^{  \ell_{{\mathrm{1}}}  }   \widetilde{ \ottnt{a} }    )    )    )   
\end{align*}
and \begin{align*}
& \wdtilde{ \mathbf{up}^{    \ell_{{\mathrm{1}}}  \vee  \ell_{{\mathrm{2}}}    ,    \ell'_{{\mathrm{1}}}  \vee  \ell_{{\mathrm{2}}}    }   (   \ottnt{a}  \:  \leftindex^{  \ell_{{\mathrm{1}}}  }{\gg}\!\! =^{  \ell_{{\mathrm{2}}}  }  \ottnt{f}   )  } \\
= &  \mathbf{split}^{   \ell'_{{\mathrm{1}}}  \vee  \ell_{{\mathrm{2}}}   }   (   \mathbf{merge}^{   \ell_{{\mathrm{1}}}  \vee  \ell_{{\mathrm{2}}}   }   \wdtilde{   \ottnt{a}  \:  \leftindex^{  \ell_{{\mathrm{1}}}  }{\gg}\!\! =^{  \ell_{{\mathrm{2}}}  }  \ottnt{f}   }    )   \\
\equiv &  \mathbf{split}^{   \ell'_{{\mathrm{1}}}  \vee  \ell_{{\mathrm{2}}}   }   (   \mathbf{merge}^{   \ell_{{\mathrm{1}}}  \vee  \ell_{{\mathrm{2}}}   }   (   \mathbf{split}^{   \ell_{{\mathrm{1}}}  \vee  \ell_{{\mathrm{2}}}   }   (   \mathbf{merge}^{  \ell_{{\mathrm{2}}}  }   (     \widetilde{ \ottnt{f} }    \:   (   \mathbf{merge}^{  \ell_{{\mathrm{1}}}  }    \widetilde{ \ottnt{a} }     )    )    )    )    )   \\
\equiv &  \mathbf{split}^{   \ell'_{{\mathrm{1}}}  \vee  \ell_{{\mathrm{2}}}   }   (   \mathbf{merge}^{  \ell_{{\mathrm{2}}}  }   (     \widetilde{ \ottnt{f} }    \:   (   \mathbf{merge}^{  \ell_{{\mathrm{1}}}  }    \widetilde{ \ottnt{a} }     )    )    )  
\end{align*}
\item $  \ottnt{a}  \:  \leftindex^{  \ell_{{\mathrm{1}}}  }{\gg}\!\! =^{  \ell'_{{\mathrm{2}}}  }   (   \lambda  \ottmv{x}  .   \mathbf{up}^{  \ell_{{\mathrm{2}}}  ,  \ell'_{{\mathrm{2}}}  }  \ottnt{b}    )    \equiv   \mathbf{up}^{    \ell_{{\mathrm{1}}}  \vee  \ell_{{\mathrm{2}}}    ,    \ell_{{\mathrm{1}}}  \vee  \ell'_{{\mathrm{2}}}    }   (   \ottnt{a}  \:  \leftindex^{  \ell_{{\mathrm{1}}}  }{\gg}\!\! =^{  \ell_{{\mathrm{2}}}  }   \lambda  \ottmv{x}  .  \ottnt{b}    )   $.\\
Now, \begin{align*}
& \wdtilde{ \ottnt{a}  \:  \leftindex^{  \ell_{{\mathrm{1}}}  }{\gg}\!\! =^{  \ell'_{{\mathrm{2}}}  }   (   \lambda  \ottmv{x}  .   \mathbf{up}^{  \ell_{{\mathrm{2}}}  ,  \ell'_{{\mathrm{2}}}  }  \ottnt{b}    )  } \\
\equiv &  \mathbf{split}^{   \ell_{{\mathrm{1}}}  \vee  \ell'_{{\mathrm{2}}}   }   (   \mathbf{merge}^{  \ell'_{{\mathrm{2}}}  }   (    \wdtilde{   \lambda  \ottmv{x}  .   \mathbf{up}^{  \ell_{{\mathrm{2}}}  ,  \ell'_{{\mathrm{2}}}  }  \ottnt{b}    }   \:   (   \mathbf{merge}^{  \ell_{{\mathrm{1}}}  }    \widetilde{ \ottnt{a} }     )    )    )   \\
= &  \mathbf{split}^{   \ell_{{\mathrm{1}}}  \vee  \ell'_{{\mathrm{2}}}   }   (   \mathbf{merge}^{  \ell'_{{\mathrm{2}}}  }   (    (   \lambda  \ottmv{x}  .   \mathbf{split}^{  \ell'_{{\mathrm{2}}}  }   (   \mathbf{merge}^{  \ell_{{\mathrm{2}}}  }   \widetilde{ \ottnt{b} }    )     )   \:   (   \mathbf{merge}^{  \ell_{{\mathrm{1}}}  }    \widetilde{ \ottnt{a} }     )    )    )   \\
\equiv &  \mathbf{split}^{   \ell_{{\mathrm{1}}}  \vee  \ell'_{{\mathrm{2}}}   }   (   \mathbf{merge}^{  \ell'_{{\mathrm{2}}}  }   (   \mathbf{split}^{  \ell'_{{\mathrm{2}}}  }   (    \mathbf{merge}^{  \ell_{{\mathrm{2}}}  }   \widetilde{ \ottnt{b} }    \{   \mathbf{merge}^{  \ell_{{\mathrm{1}}}  }    \widetilde{ \ottnt{a} }     /  \ottmv{x}  \}   )    )    )   \\
\equiv &  \mathbf{split}^{   \ell_{{\mathrm{1}}}  \vee  \ell'_{{\mathrm{2}}}   }   (    \mathbf{merge}^{  \ell_{{\mathrm{2}}}  }   \widetilde{ \ottnt{b} }    \{   \mathbf{merge}^{  \ell_{{\mathrm{1}}}  }    \widetilde{ \ottnt{a} }     /  \ottmv{x}  \}   )  
\end{align*}
and \begin{align*}
& \wdtilde{ \mathbf{up}^{    \ell_{{\mathrm{1}}}  \vee  \ell_{{\mathrm{2}}}    ,    \ell_{{\mathrm{1}}}  \vee  \ell'_{{\mathrm{2}}}    }   (   \ottnt{a}  \:  \leftindex^{  \ell_{{\mathrm{1}}}  }{\gg}\!\! =^{  \ell_{{\mathrm{2}}}  }   \lambda  \ottmv{x}  .  \ottnt{b}    )  } \\
= &  \mathbf{split}^{   \ell_{{\mathrm{1}}}  \vee  \ell'_{{\mathrm{2}}}   }   (   \mathbf{merge}^{   \ell_{{\mathrm{1}}}  \vee  \ell_{{\mathrm{2}}}   }   \wdtilde{   \ottnt{a}  \:  \leftindex^{  \ell_{{\mathrm{1}}}  }{\gg}\!\! =^{  \ell_{{\mathrm{2}}}  }   \lambda  \ottmv{x}  .  \ottnt{b}    }    )   \\
\equiv &  \mathbf{split}^{   \ell_{{\mathrm{1}}}  \vee  \ell'_{{\mathrm{2}}}   }   (   \mathbf{merge}^{   \ell_{{\mathrm{1}}}  \vee  \ell_{{\mathrm{2}}}   }   (   \mathbf{split}^{   \ell_{{\mathrm{1}}}  \vee  \ell_{{\mathrm{2}}}   }   (   \mathbf{merge}^{  \ell_{{\mathrm{2}}}  }   (    (   \lambda  \ottmv{x}  .   \widetilde{ \ottnt{b} }    )   \:   (   \mathbf{merge}^{  \ell_{{\mathrm{1}}}  }    \widetilde{ \ottnt{a} }     )    )    )    )    )   \\ 
\equiv &  \mathbf{split}^{   \ell_{{\mathrm{1}}}  \vee  \ell'_{{\mathrm{2}}}   }   (   \mathbf{merge}^{  \ell_{{\mathrm{2}}}  }   (    (   \lambda  \ottmv{x}  .   \widetilde{ \ottnt{b} }    )   \:   (   \mathbf{merge}^{  \ell_{{\mathrm{1}}}  }    \widetilde{ \ottnt{a} }     )    )    )   \\ 
\equiv &  \mathbf{split}^{   \ell_{{\mathrm{1}}}  \vee  \ell'_{{\mathrm{2}}}   }   (    \mathbf{merge}^{  \ell_{{\mathrm{2}}}  }   \widetilde{ \ottnt{b} }    \{   \mathbf{merge}^{  \ell_{{\mathrm{1}}}  }   \widetilde{ \ottnt{a} }    /  \ottmv{x}  \}   )  
\end{align*}

\item $   (   \ottkw{ret}  \:  \ottnt{a}   )   \:  \leftindex^{   \bot   }{\gg}\!\! =^{  \ell  }  \ottnt{f}   \equiv   \ottnt{f}  \:  \ottnt{a}  $.\\
Now, \begin{align*}
&  \wdtilde{    (   \ottkw{ret}  \:  \ottnt{a}   )   \:  \leftindex^{   \bot   }{\gg}\!\! =^{  \ell  }  \ottnt{f}   }  \\
\equiv &  \mathbf{split}^{  \ell  }   (   \mathbf{merge}^{  \ell  }   (     \widetilde{ \ottnt{f} }    \:   (   \mathbf{merge}^{   \bot   }    \widetilde{  \ottkw{ret}  \:  \ottnt{a}  }     )    )    )   \\
= &  \mathbf{split}^{  \ell  }   (   \mathbf{merge}^{  \ell  }   (     \widetilde{ \ottnt{f} }    \:   (   \mathbf{merge}^{   \bot   }   (   \mathbf{split}^{   \bot   }   \widetilde{ \ottnt{a} }    )    )    )    )   \\
\equiv &  \mathbf{split}^{  \ell  }   (   \mathbf{merge}^{  \ell  }   (     \widetilde{ \ottnt{f} }    \:   \widetilde{ \ottnt{a} }    )    )   \equiv    \widetilde{ \ottnt{f} }    \:   \widetilde{ \ottnt{a} }  
\end{align*}

\item $  \ottnt{a}  \:  \leftindex^{  \ell_{{\mathrm{1}}}  }{\gg}\!\! =^{   \bot   }   (   \lambda  \ottmv{x}  .   \ottkw{ret}  \:  \ottmv{x}    )    \equiv  \ottnt{a} $.\\
Now, \begin{align*}
&  \wdtilde{   \ottnt{a}  \:  \leftindex^{  \ell_{{\mathrm{1}}}  }{\gg}\!\! =^{   \bot   }   (   \lambda  \ottmv{x}  .   \ottkw{ret}  \:  \ottmv{x}    )    }  \\
\equiv &  \mathbf{split}^{  \ell_{{\mathrm{1}}}  }   (   \mathbf{merge}^{   \bot   }   (     \wdtilde{   \lambda  \ottmv{x}  .   \ottkw{ret}  \:  \ottmv{x}    }    \:   (   \mathbf{merge}^{  \ell_{{\mathrm{1}}}  }    \widetilde{ \ottnt{a} }     )    )    )   \\
= &  \mathbf{split}^{  \ell_{{\mathrm{1}}}  }   (   \mathbf{merge}^{   \bot   }   (    (   \lambda  \ottmv{x}  .   \mathbf{split}^{   \bot   }  \ottmv{x}    )   \:   (   \mathbf{merge}^{  \ell_{{\mathrm{1}}}  }    \widetilde{ \ottnt{a} }     )    )    )   \\
\equiv &  \mathbf{split}^{  \ell_{{\mathrm{1}}}  }   (   \mathbf{merge}^{   \bot   }   (   \mathbf{split}^{   \bot   }   (   \mathbf{merge}^{  \ell_{{\mathrm{1}}}  }    \widetilde{ \ottnt{a} }     )    )    )   \\
\equiv &  \mathbf{split}^{  \ell_{{\mathrm{1}}}  }   (   \mathbf{merge}^{  \ell_{{\mathrm{1}}}  }    \widetilde{ \ottnt{a} }     )   \equiv  \widetilde{ \ottnt{a} } 
\end{align*}

\item $   (   \ottnt{a}  \:  \leftindex^{  \ell_{{\mathrm{1}}}  }{\gg}\!\! =^{  \ell_{{\mathrm{2}}}  }  \ottnt{f}   )   \:  \leftindex^{    \ell_{{\mathrm{1}}}  \vee  \ell_{{\mathrm{2}}}    }{\gg}\!\! =^{  \ell_{{\mathrm{3}}}  }  \ottnt{g}   \equiv   \ottnt{a}  \:  \leftindex^{  \ell_{{\mathrm{1}}}  }{\gg}\!\! =^{    \ell_{{\mathrm{2}}}  \vee  \ell_{{\mathrm{3}}}    }   (   \lambda  \ottmv{x}  .   (    \ottnt{f}  \:  \ottmv{x}   \:  \leftindex^{  \ell_{{\mathrm{2}}}  }{\gg}\!\! =^{  \ell_{{\mathrm{3}}}  }  \ottnt{g}   )    )   $.\\
Now, \begin{align*}
&  \wdtilde{    (   \ottnt{a}  \:  \leftindex^{  \ell_{{\mathrm{1}}}  }{\gg}\!\! =^{  \ell_{{\mathrm{2}}}  }  \ottnt{f}   )   \:  \leftindex^{    \ell_{{\mathrm{1}}}  \vee  \ell_{{\mathrm{2}}}    }{\gg}\!\! =^{  \ell_{{\mathrm{3}}}  }  \ottnt{g}   }  \\
\equiv &  \mathbf{split}^{    (   \ell_{{\mathrm{1}}}  \vee  \ell_{{\mathrm{2}}}   )   \vee  \ell_{{\mathrm{3}}}   }   (   \mathbf{merge}^{  \ell_{{\mathrm{3}}}  }   (     \widetilde{ \ottnt{g} }    \:   (   \mathbf{merge}^{   \ell_{{\mathrm{1}}}  \vee  \ell_{{\mathrm{2}}}   }    \wdtilde{   \ottnt{a}  \:  \leftindex^{  \ell_{{\mathrm{1}}}  }{\gg}\!\! =^{  \ell_{{\mathrm{2}}}  }  \ottnt{f}   }     )    )    )   \\
\equiv &  \mathbf{split}^{    (   \ell_{{\mathrm{1}}}  \vee  \ell_{{\mathrm{2}}}   )   \vee  \ell_{{\mathrm{3}}}   }   (   \mathbf{merge}^{  \ell_{{\mathrm{3}}}  }   (     \widetilde{ \ottnt{g} }    \:   (   \mathbf{merge}^{   \ell_{{\mathrm{1}}}  \vee  \ell_{{\mathrm{2}}}   }   (   \mathbf{split}^{   \ell_{{\mathrm{1}}}  \vee  \ell_{{\mathrm{2}}}   }   (   \mathbf{merge}^{  \ell_{{\mathrm{2}}}  }   (     \widetilde{ \ottnt{f} }    \:   (   \mathbf{merge}^{  \ell_{{\mathrm{1}}}  }    \widetilde{ \ottnt{a} }     )    )    )    )    )    )    )   \\
\equiv &  \mathbf{split}^{     \ell_{{\mathrm{1}}}  \vee  \ell_{{\mathrm{2}}}    \vee  \ell_{{\mathrm{3}}}   }   (   \mathbf{merge}^{  \ell_{{\mathrm{3}}}  }   (     \widetilde{ \ottnt{g} }    \:   (   \mathbf{merge}^{  \ell_{{\mathrm{2}}}  }   (     \widetilde{ \ottnt{f} }    \:   (   \mathbf{merge}^{  \ell_{{\mathrm{1}}}  }    \widetilde{ \ottnt{a} }     )    )    )    )    )   
\end{align*}
and \begin{align*}
&  \wdtilde{   \ottnt{a}  \:  \leftindex^{  \ell_{{\mathrm{1}}}  }{\gg}\!\! =^{    \ell_{{\mathrm{2}}}  \vee  \ell_{{\mathrm{3}}}    }   (   \lambda  \ottmv{x}  .   (    \ottnt{f}  \:  \ottmv{x}   \:  \leftindex^{  \ell_{{\mathrm{2}}}  }{\gg}\!\! =^{  \ell_{{\mathrm{3}}}  }  \ottnt{g}   )    )    }  \\
\equiv &  \mathbf{split}^{   \ell_{{\mathrm{1}}}  \vee   (   \ell_{{\mathrm{2}}}  \vee  \ell_{{\mathrm{3}}}   )    }   (   \mathbf{merge}^{   \ell_{{\mathrm{2}}}  \vee  \ell_{{\mathrm{3}}}   }   (     \wdtilde{  (   \lambda  \ottmv{x}  .   (    \ottnt{f}  \:  \ottmv{x}   \:  \leftindex^{  \ell_{{\mathrm{2}}}  }{\gg}\!\! =^{  \ell_{{\mathrm{3}}}  }  \ottnt{g}   )    )  }    \:   (   \mathbf{merge}^{  \ell_{{\mathrm{1}}}  }    \widetilde{ \ottnt{a} }     )    )    )   \\
\equiv &  \mathbf{split}^{   \ell_{{\mathrm{1}}}  \vee   (   \ell_{{\mathrm{2}}}  \vee  \ell_{{\mathrm{3}}}   )    }   (   \mathbf{merge}^{   \ell_{{\mathrm{2}}}  \vee  \ell_{{\mathrm{3}}}   }   (    (   \lambda  \ottmv{x}  .   (   \mathbf{split}^{   \ell_{{\mathrm{2}}}  \vee  \ell_{{\mathrm{3}}}   }   (   \mathbf{merge}^{  \ell_{{\mathrm{3}}}  }   (     \widetilde{ \ottnt{g} }    \:   (   \mathbf{merge}^{  \ell_{{\mathrm{2}}}  }   (     \widetilde{ \ottnt{f} }    \:  \ottmv{x}   )    )    )    )    )    )   \:   (   \mathbf{merge}^{  \ell_{{\mathrm{1}}}  }    \widetilde{ \ottnt{a} }     )    )    )   \\
\equiv &  \mathbf{split}^{   \ell_{{\mathrm{1}}}  \vee   (   \ell_{{\mathrm{2}}}  \vee  \ell_{{\mathrm{3}}}   )    }   (   \mathbf{merge}^{   \ell_{{\mathrm{2}}}  \vee  \ell_{{\mathrm{3}}}   }   (   \mathbf{split}^{   \ell_{{\mathrm{2}}}  \vee  \ell_{{\mathrm{3}}}   }   (   \mathbf{merge}^{  \ell_{{\mathrm{3}}}  }   (     \widetilde{ \ottnt{g} }    \:   (   \mathbf{merge}^{  \ell_{{\mathrm{2}}}  }   (     \widetilde{ \ottnt{f} }    \:   (   \mathbf{merge}^{  \ell_{{\mathrm{1}}}  }    \widetilde{ \ottnt{a} }     )    )    )    )    )    )    )   \\
\equiv &  \mathbf{split}^{   \ell_{{\mathrm{1}}}  \vee    \ell_{{\mathrm{2}}}  \vee  \ell_{{\mathrm{3}}}     }   (   \mathbf{merge}^{  \ell_{{\mathrm{3}}}  }   (     \widetilde{ \ottnt{g} }    \:   (   \mathbf{merge}^{  \ell_{{\mathrm{2}}}  }   (     \widetilde{ \ottnt{f} }    \:   (   \mathbf{merge}^{  \ell_{{\mathrm{1}}}  }    \widetilde{ \ottnt{a} }     )    )    )    )    )   
\end{align*}

\item $  \ottkw{ret}  \:   (   \mathbf{extr} \:  \ottnt{a}   )    \equiv  \ottnt{a} $.\\
Now, $ \wdtilde{   \ottkw{ret}  \:   (   \mathbf{extr} \:  \ottnt{a}   )    }  =  \mathbf{split}^{   \bot   }   (   \mathbf{merge}^{   \bot   }   \widetilde{ \ottnt{a} }    )   \equiv  \widetilde{ \ottnt{a} } $.
\item $  \mathbf{extr} \:   (   \ottkw{ret}  \:  \ottnt{a}   )    \equiv  \ottnt{a} $.\\
Now, $ \wdtilde{   \mathbf{extr} \:   (   \ottkw{ret}  \:  \ottnt{a}   )    }  =  \mathbf{merge}^{   \bot   }   (   \mathbf{split}^{   \bot   }   \widetilde{ \ottnt{a} }    )   \equiv  \widetilde{ \ottnt{a} } $.
\item $  \mathbf{join}^{  \ell_{{\mathrm{1}}}  ,  \ell_{{\mathrm{2}}}  }   (   \mathbf{fork}^{  \ell_{{\mathrm{1}}}  ,  \ell_{{\mathrm{2}}}  }  \ottnt{a}   )    \equiv  \ottnt{a} $.\\
Now, \begin{align*}
&  \wdtilde{   \mathbf{join}^{  \ell_{{\mathrm{1}}}  ,  \ell_{{\mathrm{2}}}  }   (   \mathbf{fork}^{  \ell_{{\mathrm{1}}}  ,  \ell_{{\mathrm{2}}}  }  \ottnt{a}   )    }  \\
= &  \mathbf{split}^{   \ell_{{\mathrm{1}}}  \vee  \ell_{{\mathrm{2}}}   }   (   \mathbf{merge}^{  \ell_{{\mathrm{2}}}  }   (   \mathbf{merge}^{  \ell_{{\mathrm{1}}}  }   \widetilde{  \mathbf{fork}^{  \ell_{{\mathrm{1}}}  ,  \ell_{{\mathrm{2}}}  }  \ottnt{a}  }    )    )   \\
= &  \mathbf{split}^{   \ell_{{\mathrm{1}}}  \vee  \ell_{{\mathrm{2}}}   }   (   \mathbf{merge}^{  \ell_{{\mathrm{2}}}  }   (   \mathbf{merge}^{  \ell_{{\mathrm{1}}}  }   (   \mathbf{split}^{  \ell_{{\mathrm{1}}}  }   (   \mathbf{split}^{  \ell_{{\mathrm{2}}}  }   (   \mathbf{merge}^{   \ell_{{\mathrm{1}}}  \vee  \ell_{{\mathrm{2}}}   }   \widetilde{ \ottnt{a} }    )    )    )    )    )   \\
\equiv &  \mathbf{split}^{   \ell_{{\mathrm{1}}}  \vee  \ell_{{\mathrm{2}}}   }   (   \mathbf{merge}^{  \ell_{{\mathrm{2}}}  }   (   \mathbf{split}^{  \ell_{{\mathrm{2}}}  }   (   \mathbf{merge}^{   \ell_{{\mathrm{1}}}  \vee  \ell_{{\mathrm{2}}}   }   \widetilde{ \ottnt{a} }    )    )    )   \\
\equiv &  \mathbf{split}^{   \ell_{{\mathrm{1}}}  \vee  \ell_{{\mathrm{2}}}   }   (   \mathbf{merge}^{   \ell_{{\mathrm{1}}}  \vee  \ell_{{\mathrm{2}}}   }   \widetilde{ \ottnt{a} }    )   \equiv  \widetilde{ \ottnt{a} } 
\end{align*}
\item $  \mathbf{fork}^{  \ell_{{\mathrm{1}}}  ,  \ell_{{\mathrm{2}}}  }   (   \mathbf{join}^{  \ell_{{\mathrm{1}}}  ,  \ell_{{\mathrm{2}}}  }  \ottnt{a}   )    \equiv  \ottnt{a} $.\\
Now, \begin{align*}
&  \wdtilde{   \mathbf{fork}^{  \ell_{{\mathrm{1}}}  ,  \ell_{{\mathrm{2}}}  }   (   \mathbf{join}^{  \ell_{{\mathrm{1}}}  ,  \ell_{{\mathrm{2}}}  }  \ottnt{a}   )    }  \\
= &  \mathbf{split}^{  \ell_{{\mathrm{1}}}  }   (   \mathbf{split}^{  \ell_{{\mathrm{2}}}  }   (   \mathbf{merge}^{   \ell_{{\mathrm{1}}}  \vee  \ell_{{\mathrm{2}}}   }   \widetilde{   \mathbf{join}^{  \ell_{{\mathrm{1}}}  ,  \ell_{{\mathrm{2}}}  }  \ottnt{a}   }    )    )   \\
= &  \mathbf{split}^{  \ell_{{\mathrm{1}}}  }   (   \mathbf{split}^{  \ell_{{\mathrm{2}}}  }   (   \mathbf{merge}^{   \ell_{{\mathrm{1}}}  \vee  \ell_{{\mathrm{2}}}   }   (   \mathbf{split}^{   \ell_{{\mathrm{1}}}  \vee  \ell_{{\mathrm{2}}}   }   (   \mathbf{merge}^{  \ell_{{\mathrm{2}}}  }   (   \mathbf{merge}^{  \ell_{{\mathrm{1}}}  }   \widetilde{ \ottnt{a} }    )    )    )    )    )   \\
\equiv &  \mathbf{split}^{  \ell_{{\mathrm{1}}}  }   (   \mathbf{split}^{  \ell_{{\mathrm{2}}}  }   (   \mathbf{merge}^{  \ell_{{\mathrm{2}}}  }   (   \mathbf{merge}^{  \ell_{{\mathrm{1}}}  }   \widetilde{ \ottnt{a} }    )    )    )   \\
\equiv &  \mathbf{split}^{  \ell_{{\mathrm{1}}}  }   (   \mathbf{merge}^{  \ell_{{\mathrm{1}}}  }   \widetilde{ \ottnt{a} }    )   \equiv  \widetilde{ \ottnt{a} } 
\end{align*}
\end{itemize}
\end{proof}


\begin{theorem}[Theorem \ref{lc2gmcce}]
If $ \Gamma  \vdash  \ottnt{a}  :^{ n }  \ottnt{A} $ in \lc{}, then $  \widehat{  \Gamma  }   \vdash   \widehat{ \ottnt{a} }   :^{  n  }   \widehat{ \ottnt{A} }  $ in \Ge{}($ \mathcal{N} $). Further, if $ \Gamma  \vdash  \ottnt{a_{{\mathrm{1}}}}  :^{ n }  \ottnt{A} $ and $ \Gamma  \vdash  \ottnt{a_{{\mathrm{2}}}}  :^{ n }  \ottnt{A} $ such that $ \ottnt{a_{{\mathrm{1}}}}  \equiv  \ottnt{a_{{\mathrm{2}}}} $ in \lc{}, then $  \widehat{ \ottnt{a_{{\mathrm{1}}}} }   \equiv   \widehat{ \ottnt{a_{{\mathrm{2}}}} }  $ in \Ge{}($ \mathcal{N} $).
\end{theorem}

\begin{proof}
The first part follows by induction on the typing derivation. The second part follows by inversion on the equality judgement.
\end{proof}


\section{Proof of proposition mentioned in Section \ref{secdiscuss}}

\begin{prop}
The exponential object in $\mathcal{DC}$ does not satisfy the universal property, unless the relations in the definition of $\text{Obj}(\mathcal{DC})$ are restricted to reflexive ones only.
\end{prop}

\begin{proof}
An object of $\mathcal{DC}$ is a pair, $X := (|X|,R_{X,\ell})$, where $|X|$ is a set and $R_{X,\ell}$ is a family of binary relations on $|X|$, indexed by elements, $\ell$, of the parametrizing lattice, $\mathcal{L}$. (Strictly speaking, an object is a pair consisting of a cpo and a family of directed-complete relations. However, since we are not considering nonterminating computations, we can simplify it to a set and a family of relations.) A morphism from $X$ to $Y$ is any function from $|X|$ to $|Y|$ that respects the binary relations, i.e. a function $h : | X | \to | Y |$ such that if $(x_1,x_2) \in R_{X,\ell}$, then $(h \: x_1, h \: x_2) \in R_{Y,\ell}$.\\ A product object is defined in $\mathcal{DC}$ as follows:
\[  \ottnt{X}  \times  \ottnt{Y}  := (| X | \times | Y |, \; \{ ((x_1,y_1),(x_2,y_2)) \: | \: (x_1,x_2) \in R_{X,\ell} \: \wedge \: (y_1,y_2) \in R_{Y,\ell}\} \]
An exponential object is defined as:
\[   \ottnt{X}  \Rightarrow  \ottnt{Y}  := (\text{Hom}_{\mathcal{DC}} (X,Y), \; \{ (f,g) \: | \: \forall (x_1,x_2) \in R_{X,\ell},(f\: x_1,g \: x_2) \in R_{Y,\ell} \}) \]

Now we present our counter-example. Let the parametrizing lattice $\mathcal{L}$ be $ \mathbf{Public}  \sqsubset  \mathbf{Secret} $. For a set $A$, let $\textit{id}$ and $\textit{true}$ denote the identity relation on $A$ and the total relation on $A$ respectively. Now define the following objects in $\mathcal{DC}$:
\begin{align*}
X & := ( \{ x_1, x_2 \}, \: R_{X, \mathbf{Public} } = R_{X, \mathbf{Secret} } = \{ (x_1 , x_1) \}) \\
Y & := ( \{ y_1, y_2 \}, \: R_{Y, \mathbf{Public} } = \textit{true} \: \wedge \: R_{Y, \mathbf{Secret} } = \textit{id}) \\
Z & := ( \{ z_1, z_2 \}, \: R_{Z, \mathbf{Public} } =  R_{Z, \mathbf{Secret} } = \textit{id})
\end{align*} 
Then, $\text{Hom}_{\mathcal{DC}} ( \ottnt{X}  \times  \ottnt{Y}  , Z)$ has $8$ elements but $\text{Hom}_{\mathcal{DC}} (X,  \ottnt{Y}  \Rightarrow  \ottnt{Z} )$ has only $4$ elements. Therefore, $ \ottnt{Y}  \Rightarrow  \ottnt{Z} $ does not satisfy the universal property. Hence, category $\mathcal{DC}$, as presented by \citet{dcc}, is not cartesian closed.

An analysis of the above counter-example shows that the problem stems from the fact that $R_{X,\ell}$ does not relate $x_2$ to itself. In fact, when we set $R_{X, \mathbf{Public} } = R_{X, \mathbf{Secret} } = \textit{id}$, both $\text{Hom}_{\mathcal{DC}} ( \ottnt{X}  \times  \ottnt{Y}  , Z)$ and $\text{Hom}_{\mathcal{DC}} (X,  \ottnt{Y}  \Rightarrow  \ottnt{Z} )$ have $4$ elements. Below, we show that if for every object $X$, the relation $R_{X,\ell}$ is required to be reflexive, then the exponential object indeed satisfies the universal property. 

Let $X , Y , Z \in \text{Obj}(\mathcal{DC})$. Define the exponential object, $ \ottnt{X}  \Rightarrow  \ottnt{Y} $, as above. Need to check that for any $h \in \text{Hom}_{\mathcal{DC}} (X,Y)$, we have, $(h , h) \in R_{ \ottnt{X}  \Rightarrow  \ottnt{Y} ,\ell}$.\\
Suppose, $ (  \ottmv{x_{{\mathrm{1}}}}  ,  \ottmv{x_{{\mathrm{2}}}}  )  \in R_{X,\ell}$. Need to show: $ (   \ottnt{h}  \:  \ottmv{x_{{\mathrm{1}}}}   ,   \ottnt{h}  \:  \ottmv{x_{{\mathrm{2}}}}   )  \in R_{Y,\ell}$.\\
But this is true because $h \in \text{Hom}_{\mathcal{DC}} (X,Y)$.\\

Next, we define $ \text{app} $ as:
\[   (   \ottnt{X}  \Rightarrow  \ottnt{Y}   )   \times  \ottnt{X}  \xrightarrow{\mathbf{\lambda} w .   (    \pi_1   \:  \ottmv{w}   )   \:   (    \pi_2   \:  \ottmv{w}   )  } Y \]
Need to check that $ \text{app} $ is a $\mathcal{DC}$-morphism.\\
Suppose, $ (  \ottmv{w}  ,  \ottmv{w'}  )  \in R_{  (   \ottnt{X}  \Rightarrow  \ottnt{Y}   )   \times  \ottnt{X} ,\ell}$. Then $ (    \pi_1   \:  \ottmv{w}   ,    \pi_1   \:  \ottmv{w'}   )  \in R_{ \ottnt{X}  \Rightarrow  \ottnt{Y} ,\ell}$ and $ (    \pi_2   \:  \ottmv{w}   ,    \pi_2   \:  \ottmv{w'}   )  \in R_{X,\ell}$.\\
Therefore, $ (    (    \pi_1   \:  \ottmv{w}   )   \:   (    \pi_2   \:  \ottmv{w}   )    ,    (    \pi_1   \:  \ottmv{w'}   )   \:   (    \pi_2   \:  \ottmv{w'}   )    )  \in R_{Y,\ell}$.\\

Now, say $h \in \text{Hom}_{\mathcal{DC}} ( \ottnt{Z}  \times  \ottnt{X} ,Y)$. We define $ \Lambda  \ottnt{h}  \in \text{Hom}_{\mathcal{DC}} (Z ,  \ottnt{X}  \Rightarrow  \ottnt{Y} )$ as:
\[  \Lambda  \ottnt{h}  \triangleq \mathbf{\lambda} z . \mathbf{\lambda} x .  \ottnt{h}  \:   (  \ottmv{z}  ,  \ottmv{x}  )    \]
Need to check the following:
\begin{itemize}
\item First that $ \Lambda  \ottnt{h} $ is well-defined.\\
For any $\ottmv{z_{{\mathrm{0}}}} \in |Z|$, need to show: $  (   \Lambda  \ottnt{h}   )   \:  \ottmv{z_{{\mathrm{0}}}}  \in |  \ottnt{X}  \Rightarrow  \ottnt{Y}  |$, i.e. $\mathbf{\lambda} x .  \ottnt{h}  \:   (  \ottmv{z_{{\mathrm{0}}}}  ,  \ottmv{x}  )   \in \text{Hom}_{\mathcal{DC}} (X,Y)$.\\
Suppose $ (  \ottmv{x_{{\mathrm{1}}}}  ,  \ottmv{x_{{\mathrm{2}}}}  )  \in R_{X,\ell}$. Need to show: $ (   \ottnt{h}  \:   (  \ottmv{z_{{\mathrm{0}}}}  ,  \ottmv{x_{{\mathrm{1}}}}  )    ,   \ottnt{h}  \:   (  \ottmv{z_{{\mathrm{0}}}}  ,  \ottmv{x_{{\mathrm{2}}}}  )    )  \in R_{Y,\ell}$.\\
\textbf{Since $R_{Z,\ell}$ is reflexive}, $ (  \ottmv{z_{{\mathrm{0}}}}  ,  \ottmv{z_{{\mathrm{0}}}}  )  \in R_{Z,\ell}$. Therefore, $ (   (  \ottmv{z_{{\mathrm{0}}}}  ,  \ottmv{x_{{\mathrm{1}}}}  )   ,   (  \ottmv{z_{{\mathrm{0}}}}  ,  \ottmv{x_{{\mathrm{2}}}}  )   )  \in R_{ \ottnt{Z}  \times  \ottnt{X} ,\ell}$.\\
This implies that $ (   \ottnt{h}  \:   (  \ottmv{z_{{\mathrm{0}}}}  ,  \ottmv{x_{{\mathrm{1}}}}  )    ,   \ottnt{h}  \:   (  \ottmv{z_{{\mathrm{0}}}}  ,  \ottmv{x_{{\mathrm{2}}}}  )    )  \in R_{Y,\ell}$ [ $ \because h \in \text{Hom}_{\mathcal{DC}} ( \ottnt{Z}  \times  \ottnt{X} ,Y)$ ].
\item Next, that $ \Lambda  \ottnt{h} $ is a $\mathcal{DC}$-morphism.\\
For  $ (  \ottmv{z_{{\mathrm{1}}}}  ,  \ottmv{z_{{\mathrm{2}}}}  )  \in R_{Z,\ell}$, need to show $(\mathbf{\lambda} x .  \ottnt{h}  \:   (  \ottmv{z_{{\mathrm{1}}}}  ,  \ottmv{x}  )   , \mathbf{\lambda} x .  \ottnt{h}  \:   (  \ottmv{z_{{\mathrm{2}}}}  ,  \ottmv{x}  )  ) \in R_{ \ottnt{X}  \Rightarrow  \ottnt{Y} ,\ell}$. \\
Suppose $ (  \ottmv{x_{{\mathrm{1}}}}  ,  \ottmv{x_{{\mathrm{2}}}}  )  \in R_{X,\ell}$. Need to show: $ (   \ottnt{h}  \:   (  \ottmv{z_{{\mathrm{1}}}}  ,  \ottmv{x_{{\mathrm{1}}}}  )    ,   \ottnt{h}  \:   (  \ottmv{z_{{\mathrm{2}}}}  ,  \ottmv{x_{{\mathrm{2}}}}  )    )  \in R_{Y,\ell}$.\\
From what we are given, $ (   (  \ottmv{z_{{\mathrm{1}}}}  ,  \ottmv{x_{{\mathrm{1}}}}  )   ,   (  \ottmv{z_{{\mathrm{2}}}}  ,  \ottmv{x_{{\mathrm{2}}}}  )   )  \in R_{ \ottnt{Z}  \times  \ottnt{X} , \ell}$.\\
As such, $ (   \ottnt{h}  \:   (  \ottmv{z_{{\mathrm{1}}}}  ,  \ottmv{x_{{\mathrm{1}}}}  )    ,   \ottnt{h}  \:   (  \ottmv{z_{{\mathrm{2}}}}  ,  \ottmv{x_{{\mathrm{2}}}}  )    )  \in R_{Y,\ell}$ [ $ \because h \in \text{Hom}_{\mathcal{DC}} ( \ottnt{Z}  \times  \ottnt{X} ,Y)$ ].
\end{itemize}

Now we check that $ \Lambda  \ottnt{h} $ satisfies the standard existence and uniqueness properties:

\begin{itemize}
\item Existence.
\begin{align*}
&   \text{app}   \circ   (    \Lambda  \ottnt{h}   \times   \text{id}    )   \\
= &   \mathbf{\lambda}  \ottmv{v}  .   \text{app}    \:   (     \Lambda  \ottnt{h}    \:   (    \pi_1   \:  \ottmv{v}   )    ,    \pi_2   \:  \ottmv{v}   )   \\
= &   \mathbf{\lambda}  \ottmv{v}  .   (    \mathbf{\lambda}  \ottmv{w}  .   (    \pi_1   \:  \ottmv{w}   )    \:   (    \pi_2   \:  \ottmv{w}   )    )    \:   (    (    \mathbf{\lambda}  \ottmv{z}  .   \mathbf{\lambda}  \ottmv{x}  .  \ottnt{h}    \:   (  \ottmv{z}  ,  \ottmv{x}  )    )   \:   (    \pi_1   \:  \ottmv{v}   )    ,    \pi_2   \:  \ottmv{v}   )   \\
= &    \mathbf{\lambda}  \ottmv{v}  .   (    \mathbf{\lambda}  \ottmv{z}  .   \mathbf{\lambda}  \ottmv{x}  .  \ottnt{h}    \:   (  \ottmv{z}  ,  \ottmv{x}  )    )    \:   (    \pi_1   \:  \ottmv{v}   )    \:   (    \pi_2   \:  \ottmv{v}   )   \\
= &   \mathbf{\lambda}  \ottmv{v}  .  \ottnt{h}   \:   (    \pi_1   \:  \ottmv{v}   ,    \pi_2   \:  \ottmv{v}   )   =   \mathbf{\lambda}  \ottmv{v}  .  \ottnt{h}   \:  \ottmv{v}  = \ottnt{h}
\end{align*} 
\item Uniqueness.\\
Suppose, $h' \in \text{Hom}_{\mathcal{DC}} (Z ,  \ottnt{X}  \Rightarrow  \ottnt{Y} )$ such that $  \text{app}   \circ   (   \ottnt{h'}  \times   \text{id}    )   = h$.\\
Then, \[  \ottnt{h}  \:   (  \ottmv{z}  ,  \ottmv{x}  )   =   (     \text{app}   \circ  \ottnt{h'}   \times   \text{id}    )   \:   (  \ottmv{z}  ,  \ottmv{x}  )   =   (    \mathbf{\lambda}  \ottmv{w}  .   (    \pi_1   \:  \ottmv{w}   )    \:   (    \pi_2   \:  \ottmv{w}   )    )   \:   (   \ottnt{h'}  \:  \ottmv{z}   ,  \ottmv{x}  )   =    \ottnt{h'}  \:  \ottmv{z}   \:  \ottmv{x}  \]
But, $ \ottnt{h}  \:   (  \ottmv{z}  ,  \ottmv{x}  )   =     \Lambda  \ottnt{h}    \:  \ottmv{z}   \:  \ottmv{x} $.\\
So, by function extensionality, $\ottnt{h'} =  \Lambda  \ottnt{h} $.
\end{itemize}
\end{proof}

%% file: paper.bbl

\begin{thebibliography}{40}


\ifx \showCODEN    \undefined \def \showCODEN     #1{\unskip}     \fi
\ifx \showDOI      \undefined \def \showDOI       #1{#1}\fi
\ifx \showISBNx    \undefined \def \showISBNx     #1{\unskip}     \fi
\ifx \showISBNxiii \undefined \def \showISBNxiii  #1{\unskip}     \fi
\ifx \showISSN     \undefined \def \showISSN      #1{\unskip}     \fi
\ifx \showLCCN     \undefined \def \showLCCN      #1{\unskip}     \fi
\ifx \shownote     \undefined \def \shownote      #1{#1}          \fi
\ifx \showarticletitle \undefined \def \showarticletitle #1{#1}   \fi
\ifx \showURL      \undefined \def \showURL       {\relax}        \fi
\providecommand\bibfield[2]{#2}
\providecommand\bibinfo[2]{#2}
\providecommand\natexlab[1]{#1}
\providecommand\showeprint[2][]{arXiv:#2}

\bibitem[\protect\citeauthoryear{Abadi, Banerjee, Heintze, and Riecke}{Abadi
  et~al\mbox{.}}{1999}]%
        {dcc}
\bibfield{author}{\bibinfo{person}{Mart\'{\i}n Abadi}, \bibinfo{person}{Anindya
  Banerjee}, \bibinfo{person}{Nevin Heintze}, {and} \bibinfo{person}{Jon~G.
  Riecke}.} \bibinfo{year}{1999}\natexlab{}.
\newblock \showarticletitle{A Core Calculus of Dependency}. In
  \bibinfo{booktitle}{\emph{Proceedings of the 26th ACM SIGPLAN-SIGACT
  Symposium on Principles of Programming Languages}} (San Antonio, Texas, USA)
  \emph{(\bibinfo{series}{POPL '99})}. \bibinfo{publisher}{Association for
  Computing Machinery}, \bibinfo{address}{New York, NY, USA},
  \bibinfo{pages}{147–160}.
\newblock
\showISBNx{1581130953}
\urldef\tempurl%
\url{https://doi.org/10.1145/292540.292555}
\showDOI{\tempurl}


\bibitem[\protect\citeauthoryear{Algehed}{Algehed}{2018}]%
        {persdcc}
\bibfield{author}{\bibinfo{person}{Maximilian Algehed}.}
  \bibinfo{year}{2018}\natexlab{}.
\newblock \showarticletitle{A Perspective on the Dependency Core Calculus}. In
  \bibinfo{booktitle}{\emph{Proceedings of the 13th Workshop on Programming
  Languages and Analysis for Security}} (Toronto, Canada)
  \emph{(\bibinfo{series}{PLAS '18})}. \bibinfo{publisher}{Association for
  Computing Machinery}, \bibinfo{address}{New York, NY, USA},
  \bibinfo{pages}{24–28}.
\newblock
\showISBNx{9781450359931}
\urldef\tempurl%
\url{https://doi.org/10.1145/3264820.3264823}
\showDOI{\tempurl}


\bibitem[\protect\citeauthoryear{Algehed and Bernardy}{Algehed and
  Bernardy}{2019}]%
        {algehed}
\bibfield{author}{\bibinfo{person}{Maximilian Algehed} {and}
  \bibinfo{person}{Jean-Philippe Bernardy}.} \bibinfo{year}{2019}\natexlab{}.
\newblock \showarticletitle{Simple Noninterference from Parametricity}.
\newblock \bibinfo{journal}{\emph{Proc. ACM Program. Lang.}}
  \bibinfo{volume}{3}, \bibinfo{number}{ICFP}, Article \bibinfo{articleno}{89}
  (\bibinfo{date}{July} \bibinfo{year}{2019}), \bibinfo{numpages}{22}~pages.
\newblock
\urldef\tempurl%
\url{https://doi.org/10.1145/3341693}
\showDOI{\tempurl}


\bibitem[\protect\citeauthoryear{Bowman and Ahmed}{Bowman and Ahmed}{2015}]%
        {ahmed}
\bibfield{author}{\bibinfo{person}{William~J. Bowman} {and}
  \bibinfo{person}{Amal Ahmed}.} \bibinfo{year}{2015}\natexlab{}.
\newblock \showarticletitle{Noninterference for Free}.
\newblock \bibinfo{journal}{\emph{SIGPLAN Not.}} \bibinfo{volume}{50},
  \bibinfo{number}{9} (\bibinfo{date}{Aug.} \bibinfo{year}{2015}),
  \bibinfo{pages}{101–113}.
\newblock
\showISSN{0362-1340}
\urldef\tempurl%
\url{https://doi.org/10.1145/2858949.2784733}
\showDOI{\tempurl}


\bibitem[\protect\citeauthoryear{Brookes and Geva}{Brookes and Geva}{1991}]%
        {brookes}
\bibfield{author}{\bibinfo{person}{Stephen Brookes} {and} \bibinfo{person}{Shai
  Geva}.} \bibinfo{year}{1991}\natexlab{}.
\newblock \showarticletitle{Computational Comonads and Intensional Semantics}.
  \bibinfo{publisher}{Cambridge Univ. Press}, \bibinfo{pages}{1--44}.
\newblock


\bibitem[\protect\citeauthoryear{Brunel, Gaboardi, Mazza, and Zdancewic}{Brunel
  et~al\mbox{.}}{2014}]%
        {brunel}
\bibfield{author}{\bibinfo{person}{Alo\"{\i}s Brunel}, \bibinfo{person}{Marco
  Gaboardi}, \bibinfo{person}{Damiano Mazza}, {and} \bibinfo{person}{Steve
  Zdancewic}.} \bibinfo{year}{2014}\natexlab{}.
\newblock \showarticletitle{A Core Quantitative Coeffect Calculus}. In
  \bibinfo{booktitle}{\emph{Proceedings of the 23rd European Symposium on
  Programming Languages and Systems - Volume 8410}}.
  \bibinfo{publisher}{Springer-Verlag}, \bibinfo{address}{Berlin, Heidelberg},
  \bibinfo{pages}{351–370}.
\newblock
\showISBNx{9783642548321}
\urldef\tempurl%
\url{https://doi.org/10.1007/978-3-642-54833-8_19}
\showDOI{\tempurl}


\bibitem[\protect\citeauthoryear{Calcagno, Taha, Huang, and Leroy}{Calcagno
  et~al\mbox{.}}{2003}]%
        {meta}
\bibfield{author}{\bibinfo{person}{Cristiano Calcagno}, \bibinfo{person}{Walid
  Taha}, \bibinfo{person}{Liwen Huang}, {and} \bibinfo{person}{Xavier Leroy}.}
  \bibinfo{year}{2003}\natexlab{}.
\newblock \showarticletitle{Implementing Multi-stage Languages Using ASTs,
  Gensym, and Reflection}. In \bibinfo{booktitle}{\emph{Generative Programming
  and Component Engineering}}, \bibfield{editor}{\bibinfo{person}{Frank
  Pfenning} {and} \bibinfo{person}{Yannis Smaragdakis}} (Eds.).
  \bibinfo{publisher}{Springer Berlin Heidelberg}, \bibinfo{address}{Berlin,
  Heidelberg}, \bibinfo{pages}{57--76}.
\newblock
\showISBNx{978-3-540-39815-8}


\bibitem[\protect\citeauthoryear{Choudhury}{Choudhury}{2022}]%
        {gmcc}
\bibfield{author}{\bibinfo{person}{Pritam Choudhury}.}
  \bibinfo{year}{2022}\natexlab{}.
\newblock \showarticletitle{Monadic and Comonadic Aspects of Dependency
  Analysis}.
\newblock \bibinfo{journal}{\emph{Proc. ACM Program. Lang.}}
  \bibinfo{volume}{6}, \bibinfo{number}{OOPSLA2}, Article
  \bibinfo{articleno}{172} (\bibinfo{date}{Oct.} \bibinfo{year}{2022}),
  \bibinfo{numpages}{29}~pages.
\newblock
\urldef\tempurl%
\url{https://doi.org/10.1145/3563335}
\showDOI{\tempurl}


\bibitem[\protect\citeauthoryear{Choudhury, Eades, and Weirich}{Choudhury
  et~al\mbox{.}}{2022}]%
        {ddc}
\bibfield{author}{\bibinfo{person}{Pritam Choudhury}, \bibinfo{person}{Harley
  Eades}, {and} \bibinfo{person}{Stephanie Weirich}.}
  \bibinfo{year}{2022}\natexlab{}.
\newblock \showarticletitle{A Dependent Dependency Calculus}. In
  \bibinfo{booktitle}{\emph{Programming Languages and Systems}},
  \bibfield{editor}{\bibinfo{person}{Ilya Sergey}} (Ed.).
  \bibinfo{publisher}{Springer International Publishing},
  \bibinfo{address}{Cham}, \bibinfo{pages}{403--430}.
\newblock
\showISBNx{978-3-030-99336-8}


\bibitem[\protect\citeauthoryear{Davies}{Davies}{2017}]%
        {lambdacirc}
\bibfield{author}{\bibinfo{person}{Rowan Davies}.}
  \bibinfo{year}{2017}\natexlab{}.
\newblock \showarticletitle{A Temporal Logic Approach to Binding-Time
  Analysis}.
\newblock \bibinfo{journal}{\emph{J. ACM}} \bibinfo{volume}{64},
  \bibinfo{number}{1}, Article \bibinfo{articleno}{1} (\bibinfo{date}{mar}
  \bibinfo{year}{2017}), \bibinfo{numpages}{45}~pages.
\newblock
\showISSN{0004-5411}
\urldef\tempurl%
\url{https://doi.org/10.1145/3011069}
\showDOI{\tempurl}


\bibitem[\protect\citeauthoryear{Denning}{Denning}{1976}]%
        {denning1}
\bibfield{author}{\bibinfo{person}{Dorothy~E. Denning}.}
  \bibinfo{year}{1976}\natexlab{}.
\newblock \showarticletitle{A Lattice Model of Secure Information Flow}.
\newblock \bibinfo{journal}{\emph{Commun. ACM}} \bibinfo{volume}{19},
  \bibinfo{number}{5} (\bibinfo{date}{May} \bibinfo{year}{1976}),
  \bibinfo{pages}{236–243}.
\newblock
\showISSN{0001-0782}
\urldef\tempurl%
\url{https://doi.org/10.1145/360051.360056}
\showDOI{\tempurl}


\bibitem[\protect\citeauthoryear{Denning and Denning}{Denning and
  Denning}{1977}]%
        {denning2}
\bibfield{author}{\bibinfo{person}{Dorothy~E. Denning} {and}
  \bibinfo{person}{Peter~J. Denning}.} \bibinfo{year}{1977}\natexlab{}.
\newblock \showarticletitle{Certification of Programs for Secure Information
  Flow}.
\newblock \bibinfo{journal}{\emph{Commun. ACM}} \bibinfo{volume}{20},
  \bibinfo{number}{7} (\bibinfo{date}{July} \bibinfo{year}{1977}),
  \bibinfo{pages}{504–513}.
\newblock
\showISSN{0001-0782}
\urldef\tempurl%
\url{https://doi.org/10.1145/359636.359712}
\showDOI{\tempurl}


\bibitem[\protect\citeauthoryear{Eilenberg and Kelly}{Eilenberg and
  Kelly}{1966}]%
        {kelly}
\bibfield{author}{\bibinfo{person}{Samuel Eilenberg} {and}
  \bibinfo{person}{G.~Max Kelly}.} \bibinfo{year}{1966}\natexlab{}.
\newblock \showarticletitle{Closed Categories}. In
  \bibinfo{booktitle}{\emph{Proceedings of the Conference on Categorical
  Algebra}}, \bibfield{editor}{\bibinfo{person}{S.~Eilenberg},
  \bibinfo{person}{D.~K. Harrison}, \bibinfo{person}{S.~MacLane}, {and}
  \bibinfo{person}{H.~R{\"o}hrl}} (Eds.). \bibinfo{publisher}{Springer Berlin
  Heidelberg}, \bibinfo{address}{Berlin, Heidelberg},
  \bibinfo{pages}{421--562}.
\newblock
\showISBNx{978-3-642-99902-4}


\bibitem[\protect\citeauthoryear{Fujii}{Fujii}{2019}]%
        {fujii}
\bibfield{author}{\bibinfo{person}{Soichiro Fujii}.}
  \bibinfo{year}{2019}\natexlab{}.
\newblock \bibinfo{title}{A 2-Categorical Study of Graded and Indexed Monads}.
\newblock
\newblock
\showeprint[arxiv]{1904.08083}~[math.CT]


\bibitem[\protect\citeauthoryear{Ghica and Smith}{Ghica and Smith}{2014}]%
        {ghica}
\bibfield{author}{\bibinfo{person}{Dan~R. Ghica} {and} \bibinfo{person}{Alex~I.
  Smith}.} \bibinfo{year}{2014}\natexlab{}.
\newblock \showarticletitle{Bounded Linear Types in a Resource Semiring}. In
  \bibinfo{booktitle}{\emph{Programming Languages and Systems}},
  \bibfield{editor}{\bibinfo{person}{Zhong Shao}} (Ed.).
  \bibinfo{publisher}{Springer Berlin Heidelberg}, \bibinfo{address}{Berlin,
  Heidelberg}, \bibinfo{pages}{331--350}.
\newblock
\showISBNx{978-3-642-54833-8}


\bibitem[\protect\citeauthoryear{Gifford and Lucassen}{Gifford and
  Lucassen}{1986}]%
        {gifford}
\bibfield{author}{\bibinfo{person}{David~K. Gifford} {and}
  \bibinfo{person}{John~M. Lucassen}.} \bibinfo{year}{1986}\natexlab{}.
\newblock \showarticletitle{Integrating Functional and Imperative Programming}.
  In \bibinfo{booktitle}{\emph{Proceedings of the 1986 ACM Conference on LISP
  and Functional Programming}} (Cambridge, Massachusetts, USA)
  \emph{(\bibinfo{series}{LFP '86})}. \bibinfo{publisher}{Association for
  Computing Machinery}, \bibinfo{address}{New York, NY, USA},
  \bibinfo{pages}{28–38}.
\newblock
\showISBNx{0897912004}
\urldef\tempurl%
\url{https://doi.org/10.1145/319838.319848}
\showDOI{\tempurl}


\bibitem[\protect\citeauthoryear{Gl\"{u}ck and J\o{}rgensen}{Gl\"{u}ck and
  J\o{}rgensen}{1995}]%
        {gluck}
\bibfield{author}{\bibinfo{person}{Robert Gl\"{u}ck} {and}
  \bibinfo{person}{Jesper J\o{}rgensen}.} \bibinfo{year}{1995}\natexlab{}.
\newblock \showarticletitle{Efficient Multi-Level Generating Extensions for
  Program Specialization}. In \bibinfo{booktitle}{\emph{Proceedings of the 7th
  International Symposium on Programming Languages: Implementations, Logics and
  Programs}} \emph{(\bibinfo{series}{PLILPS '95})}.
  \bibinfo{publisher}{Springer-Verlag}, \bibinfo{address}{Berlin, Heidelberg},
  \bibinfo{pages}{259–278}.
\newblock
\showISBNx{354060359X}


\bibitem[\protect\citeauthoryear{Goguen and Meseguer}{Goguen and
  Meseguer}{1982}]%
        {goguen}
\bibfield{author}{\bibinfo{person}{J.~A. Goguen} {and} \bibinfo{person}{J.
  Meseguer}.} \bibinfo{year}{1982}\natexlab{}.
\newblock \showarticletitle{Security Policies and Security Models}. In
  \bibinfo{booktitle}{\emph{1982 IEEE Symposium on Security and Privacy}}.
  \bibinfo{pages}{11--11}.
\newblock


\bibitem[\protect\citeauthoryear{Gomard and Jones}{Gomard and Jones}{1991}]%
        {gomard}
\bibfield{author}{\bibinfo{person}{Carsten~K. Gomard} {and}
  \bibinfo{person}{Neil~D. Jones}.} \bibinfo{year}{1991}\natexlab{}.
\newblock \showarticletitle{A partial evaluator for the untyped
  lambda-calculus}.
\newblock \bibinfo{journal}{\emph{Journal of Functional Programming}}
  \bibinfo{volume}{1}, \bibinfo{number}{1} (\bibinfo{year}{1991}),
  \bibinfo{pages}{21–69}.
\newblock
\urldef\tempurl%
\url{https://doi.org/10.1017/S0956796800000058}
\showDOI{\tempurl}


\bibitem[\protect\citeauthoryear{Hatcliff and Danvy}{Hatcliff and
  Danvy}{1997}]%
        {hatcliff}
\bibfield{author}{\bibinfo{person}{John Hatcliff} {and}
  \bibinfo{person}{Olivier Danvy}.} \bibinfo{year}{1997}\natexlab{}.
\newblock \showarticletitle{A computational formalization for partial
  evaluation}.
\newblock \bibinfo{journal}{\emph{Mathematical Structures in Computer Science}}
  \bibinfo{volume}{7}, \bibinfo{number}{5} (\bibinfo{year}{1997}),
  \bibinfo{pages}{507–541}.
\newblock
\urldef\tempurl%
\url{https://doi.org/10.1017/S0960129597002405}
\showDOI{\tempurl}


\bibitem[\protect\citeauthoryear{Heintze and Riecke}{Heintze and
  Riecke}{1998}]%
        {slam}
\bibfield{author}{\bibinfo{person}{Nevin Heintze} {and} \bibinfo{person}{Jon~G.
  Riecke}.} \bibinfo{year}{1998}\natexlab{}.
\newblock \showarticletitle{The SLam Calculus: Programming with Secrecy and
  Integrity}. In \bibinfo{booktitle}{\emph{Proceedings of the 25th ACM
  SIGPLAN-SIGACT Symposium on Principles of Programming Languages}} (San Diego,
  California, USA) \emph{(\bibinfo{series}{POPL '98})}.
  \bibinfo{publisher}{Association for Computing Machinery},
  \bibinfo{address}{New York, NY, USA}, \bibinfo{pages}{365–377}.
\newblock
\showISBNx{0897919793}
\urldef\tempurl%
\url{https://doi.org/10.1145/268946.268976}
\showDOI{\tempurl}


\bibitem[\protect\citeauthoryear{Jacobs}{Jacobs}{1999}]%
        {jacob}
\bibfield{author}{\bibinfo{person}{Bart Jacobs}.}
  \bibinfo{year}{1999}\natexlab{}.
\newblock \bibinfo{booktitle}{\emph{Categorical Logic and Type Theory}}.
\newblock \bibinfo{publisher}{Elsevier}, \bibinfo{address}{Amsterdam, The
  Netherlands}.
\newblock


\bibitem[\protect\citeauthoryear{Katsumata}{Katsumata}{2014}]%
        {katsumata}
\bibfield{author}{\bibinfo{person}{Shin-ya Katsumata}.}
  \bibinfo{year}{2014}\natexlab{}.
\newblock \showarticletitle{Parametric Effect Monads and Semantics of Effect
  Systems}.
\newblock \bibinfo{journal}{\emph{SIGPLAN Not.}} \bibinfo{volume}{49},
  \bibinfo{number}{1} (\bibinfo{date}{jan} \bibinfo{year}{2014}),
  \bibinfo{pages}{633–645}.
\newblock
\showISSN{0362-1340}
\urldef\tempurl%
\url{https://doi.org/10.1145/2578855.2535846}
\showDOI{\tempurl}


\bibitem[\protect\citeauthoryear{Kavvos}{Kavvos}{2019}]%
        {kavvos}
\bibfield{author}{\bibinfo{person}{G.~A. Kavvos}.}
  \bibinfo{year}{2019}\natexlab{}.
\newblock \showarticletitle{Modalities, Cohesion, and Information Flow}.
\newblock \bibinfo{journal}{\emph{Proc. ACM Program. Lang.}}
  \bibinfo{volume}{3}, \bibinfo{number}{POPL}, Article \bibinfo{articleno}{20}
  (\bibinfo{date}{jan} \bibinfo{year}{2019}), \bibinfo{numpages}{29}~pages.
\newblock
\urldef\tempurl%
\url{https://doi.org/10.1145/3290333}
\showDOI{\tempurl}


\bibitem[\protect\citeauthoryear{Kock}{Kock}{1970}]%
        {kock1}
\bibfield{author}{\bibinfo{person}{Anders Kock}.}
  \bibinfo{year}{1970}\natexlab{}.
\newblock \showarticletitle{Monads on symmetric monoidal closed categories}.
\newblock   \bibinfo{volume}{21} (\bibinfo{year}{1970}),
  \bibinfo{pages}{1--10}.
\newblock
\urldef\tempurl%
\url{https://doi.org/10.1007/BF01220868}
\showURL{%
\tempurl}


\bibitem[\protect\citeauthoryear{Kock}{Kock}{1972}]%
        {kock2}
\bibfield{author}{\bibinfo{person}{Anders Kock}.}
  \bibinfo{year}{1972}\natexlab{}.
\newblock \showarticletitle{Strong Functors and Monoidal Monads}.
\newblock   \bibinfo{volume}{23} (\bibinfo{year}{1972}),
  \bibinfo{pages}{113--120}.
\newblock
\urldef\tempurl%
\url{https://doi.org/10.1007/BF01304852}
\showURL{%
\tempurl}


\bibitem[\protect\citeauthoryear{MacLane}{MacLane}{1971}]%
        {maclane}
\bibfield{author}{\bibinfo{person}{Saunders MacLane}.}
  \bibinfo{year}{1971}\natexlab{}.
\newblock \bibinfo{booktitle}{\emph{Categories for the Working Mathematician}}.
\newblock \bibinfo{publisher}{Springer-Verlag}, \bibinfo{address}{New York}.
  ix+262 pages.
\newblock
\newblock
\shownote{Graduate Texts in Mathematics, Vol. 5.}


\bibitem[\protect\citeauthoryear{Moggi}{Moggi}{1991}]%
        {moggi}
\bibfield{author}{\bibinfo{person}{Eugenio Moggi}.}
  \bibinfo{year}{1991}\natexlab{}.
\newblock \showarticletitle{Notions of computation and monads}.
\newblock \bibinfo{journal}{\emph{Information and Computation}}
  \bibinfo{volume}{93}, \bibinfo{number}{1} (\bibinfo{year}{1991}),
  \bibinfo{pages}{55--92}.
\newblock
\showISSN{0890-5401}
\urldef\tempurl%
\url{https://www.sciencedirect.com/science/article/pii/0890540191900524}
\showURL{%
\tempurl}
\newblock
\shownote{Selections from 1989 IEEE Symposium on Logic in Computer Science.}


\bibitem[\protect\citeauthoryear{Palsberg and \O{}rb\ae{}k}{Palsberg and
  \O{}rb\ae{}k}{1995}]%
        {trust}
\bibfield{author}{\bibinfo{person}{Jens Palsberg} {and} \bibinfo{person}{Peter
  \O{}rb\ae{}k}.} \bibinfo{year}{1995}\natexlab{}.
\newblock \showarticletitle{Trust in the Lambda-Calculus}. In
  \bibinfo{booktitle}{\emph{Proceedings of the Second International Symposium
  on Static Analysis}} \emph{(\bibinfo{series}{SAS '95})}.
  \bibinfo{publisher}{Springer-Verlag}, \bibinfo{address}{Berlin, Heidelberg},
  \bibinfo{pages}{314–329}.
\newblock
\showISBNx{3540603603}


\bibitem[\protect\citeauthoryear{Petricek, Orchard, and Mycroft}{Petricek
  et~al\mbox{.}}{2013}]%
        {petricek1}
\bibfield{author}{\bibinfo{person}{Tomas Petricek}, \bibinfo{person}{Dominic
  Orchard}, {and} \bibinfo{person}{Alan Mycroft}.}
  \bibinfo{year}{2013}\natexlab{}.
\newblock \showarticletitle{Coeffects: Unified Static Analysis of
  Context-Dependence}. In \bibinfo{booktitle}{\emph{Automata, Languages, and
  Programming}}. \bibinfo{publisher}{Springer Berlin Heidelberg},
  \bibinfo{address}{Berlin, Heidelberg}, \bibinfo{pages}{385--397}.
\newblock
\showISBNx{978-3-642-39212-2}


\bibitem[\protect\citeauthoryear{Petricek, Orchard, and Mycroft}{Petricek
  et~al\mbox{.}}{2014}]%
        {petricek}
\bibfield{author}{\bibinfo{person}{Tomas Petricek}, \bibinfo{person}{Dominic
  Orchard}, {and} \bibinfo{person}{Alan Mycroft}.}
  \bibinfo{year}{2014}\natexlab{}.
\newblock \showarticletitle{Coeffects: A calculus of context-dependent
  computation}. In \bibinfo{booktitle}{\emph{Proceedings of International
  Conference on Functional Programming}} (Gothenburg, Sweden)
  \emph{(\bibinfo{series}{ICFP 2014})}.
\newblock


\bibitem[\protect\citeauthoryear{Shikuma and Igarashi}{Shikuma and
  Igarashi}{2006}]%
        {igarashi}
\bibfield{author}{\bibinfo{person}{Naokata Shikuma} {and}
  \bibinfo{person}{Atsushi Igarashi}.} \bibinfo{year}{2006}\natexlab{}.
\newblock \showarticletitle{Proving Noninterference by a Fully Complete
  Translation to the Simply Typed $\lambda$-Calculus}. In
  \bibinfo{booktitle}{\emph{Proceedings of the 11th Asian Computing Science
  Conference on Advances in Computer Science: Secure Software and Related
  Issues}} (Tokyo, Japan) \emph{(\bibinfo{series}{ASIAN'06})}.
  \bibinfo{publisher}{Springer-Verlag}, \bibinfo{address}{Berlin, Heidelberg},
  \bibinfo{pages}{301–315}.
\newblock
\showISBNx{3540775048}


\bibitem[\protect\citeauthoryear{Smith}{Smith}{2007}]%
        {smith}
\bibfield{author}{\bibinfo{person}{Geoffrey Smith}.}
  \bibinfo{year}{2007}\natexlab{}.
\newblock \showarticletitle{Principles of Secure Information Flow Analysis}. In
  \bibinfo{booktitle}{\emph{Malware Detection}},
  \bibfield{editor}{\bibinfo{person}{Mihai Christodorescu},
  \bibinfo{person}{Somesh Jha}, \bibinfo{person}{Douglas Maughan},
  \bibinfo{person}{Dawn Song}, {and} \bibinfo{person}{Cliff Wang}} (Eds.).
  \bibinfo{publisher}{Springer US}, \bibinfo{address}{Boston, MA},
  \bibinfo{pages}{291--307}.
\newblock
\showISBNx{978-0-387-44599-1}


\bibitem[\protect\citeauthoryear{Tang and Jouvelot}{Tang and Jouvelot}{1995}]%
        {tang}
\bibfield{author}{\bibinfo{person}{Yan~Mei Tang} {and} \bibinfo{person}{Pierre
  Jouvelot}.} \bibinfo{year}{1995}\natexlab{}.
\newblock \showarticletitle{Effect Systems with Subtyping}. In
  \bibinfo{booktitle}{\emph{In ACM Conference on Partial Evaluation and Program
  Manipulation}}. \bibinfo{publisher}{ACM Press}, \bibinfo{pages}{45--53}.
\newblock


\bibitem[\protect\citeauthoryear{Tip}{Tip}{1995}]%
        {tip}
\bibfield{author}{\bibinfo{person}{Frank Tip}.}
  \bibinfo{year}{1995}\natexlab{}.
\newblock \showarticletitle{A Survey of Program Slicing Techniques}.
\newblock \bibinfo{journal}{\emph{Journal of Programming Languages}}
  \bibinfo{volume}{3} (\bibinfo{year}{1995}), \bibinfo{pages}{121--189}.
\newblock


\bibitem[\protect\citeauthoryear{Tofte and Talpin}{Tofte and Talpin}{1997}]%
        {memory}
\bibfield{author}{\bibinfo{person}{Mads Tofte} {and}
  \bibinfo{person}{Jean-Pierre Talpin}.} \bibinfo{year}{1997}\natexlab{}.
\newblock \showarticletitle{Region-Based Memory Management}.
\newblock \bibinfo{journal}{\emph{Inf. Comput.}} \bibinfo{volume}{132},
  \bibinfo{number}{2} (\bibinfo{date}{feb} \bibinfo{year}{1997}),
  \bibinfo{pages}{109–176}.
\newblock
\showISSN{0890-5401}
\urldef\tempurl%
\url{https://doi.org/10.1006/inco.1996.2613}
\showDOI{\tempurl}


\bibitem[\protect\citeauthoryear{Tse and Zdancewic}{Tse and Zdancewic}{2004}]%
        {tse-zdancewic}
\bibfield{author}{\bibinfo{person}{Stephen Tse} {and} \bibinfo{person}{Steve
  Zdancewic}.} \bibinfo{year}{2004}\natexlab{}.
\newblock \showarticletitle{Translating Dependency into Parametricity}. In
  \bibinfo{booktitle}{\emph{Proceedings of the Ninth ACM SIGPLAN International
  Conference on Functional Programming}} (Snow Bird, UT, USA)
  \emph{(\bibinfo{series}{ICFP '04})}. \bibinfo{publisher}{Association for
  Computing Machinery}, \bibinfo{address}{New York, NY, USA},
  \bibinfo{pages}{115–125}.
\newblock
\showISBNx{1581139055}
\urldef\tempurl%
\url{https://doi.org/10.1145/1016850.1016868}
\showDOI{\tempurl}


\bibitem[\protect\citeauthoryear{Uustalu and Vene}{Uustalu and Vene}{2008}]%
        {uustalu}
\bibfield{author}{\bibinfo{person}{Tarmo Uustalu} {and} \bibinfo{person}{Varmo
  Vene}.} \bibinfo{year}{2008}\natexlab{}.
\newblock \showarticletitle{Comonadic Notions of Computation}.
\newblock \bibinfo{journal}{\emph{Electronic Notes in Theoretical Computer
  Science}} \bibinfo{volume}{203}, \bibinfo{number}{5} (\bibinfo{year}{2008}),
  \bibinfo{pages}{263--284}.
\newblock
\showISSN{1571-0661}
\urldef\tempurl%
\url{https://doi.org/10.1016/j.entcs.2008.05.029}
\showDOI{\tempurl}
\newblock
\shownote{Proceedings of the Ninth Workshop on Coalgebraic Methods in Computer
  Science (CMCS 2008).}


\bibitem[\protect\citeauthoryear{Volpano, Irvine, and Smith}{Volpano
  et~al\mbox{.}}{1996}]%
        {volpano}
\bibfield{author}{\bibinfo{person}{Dennis Volpano}, \bibinfo{person}{Cynthia
  Irvine}, {and} \bibinfo{person}{Geoffrey Smith}.}
  \bibinfo{year}{1996}\natexlab{}.
\newblock \showarticletitle{A Sound Type System for Secure Flow Analysis}.
\newblock \bibinfo{journal}{\emph{J. Comput. Secur.}} \bibinfo{volume}{4},
  \bibinfo{number}{2–3} (\bibinfo{date}{jan} \bibinfo{year}{1996}),
  \bibinfo{pages}{167–187}.
\newblock
\showISSN{0926-227X}


\bibitem[\protect\citeauthoryear{Wadler and Thiemann}{Wadler and
  Thiemann}{2003}]%
        {wadler}
\bibfield{author}{\bibinfo{person}{Philip Wadler} {and} \bibinfo{person}{Peter
  Thiemann}.} \bibinfo{year}{2003}\natexlab{}.
\newblock \showarticletitle{The Marriage of Effects and Monads}.
\newblock \bibinfo{journal}{\emph{ACM Trans. Comput. Logic}}
  \bibinfo{volume}{4}, \bibinfo{number}{1} (\bibinfo{date}{jan}
  \bibinfo{year}{2003}), \bibinfo{pages}{1–32}.
\newblock
\showISSN{1529-3785}
\urldef\tempurl%
\url{https://doi.org/10.1145/601775.601776}
\showDOI{\tempurl}


\end{thebibliography}
